\documentclass[11pt,english]{article}
\usepackage{lmodern}

\usepackage[T1]{fontenc}
\usepackage[latin9]{inputenc}
\usepackage{colortbl}
\usepackage{float}
\usepackage{amsmath}
\usepackage{amsthm}
\usepackage{amssymb}
\usepackage{stmaryrd}
\usepackage{graphicx}
\usepackage{geometry}
\usepackage{setspace}
\usepackage[authoryear]{natbib}
\makeatletter

\providecommand{\tabularnewline}{\\}

\setcitestyle{round}
\usepackage{lscape}
\usepackage{longtable}

\usepackage{graphicx}
\usepackage{tabularx}
\usepackage{siunitx}
\newcolumntype{P}[1]{>{\centering\arraybackslash}p{#1}}
\newcolumntype{Y}{>{\centering\arraybackslash}X}

\usepackage{array,booktabs,ragged2e}
\newcolumntype{C}[1]{>{\centering\arraybackslash}p{#1}}
\newcolumntype{J}[1]{>{\justify\arraybackslash}p{#1}}
\newcolumntype{R}[1]{>{\RaggedLeft\arraybackslash}p{#1}}
\newcolumntype{Q}[1]{>{\columncolor{Gray}\RaggedLeft\arraybackslash}p{#1}}
\newcolumntype{L}[1]{>{\RaggedRight\arraybackslash}p{#1}}
\newcolumntype{G}{@{\extracolsep{0.5cm}}l@{\extracolsep{0pt}}}%
\usepackage{multirow}

\ifdefined\showcaptionsetup
 \PassOptionsToPackage{caption=false}{subfig}
\fi
\usepackage{subfig}
\AtBeginDocument{

}

\makeatother

\theoremstyle{plain}
\newtheorem{assumption}{\protect\assumptionname}
\newtheorem{thm}{\protect\theoremname}
\newtheorem{prop}{\protect\propositionname}
\newtheorem{cor}{\protect\corollaryname}
\newtheorem{lem}{\protect\lemmaname}
\theoremstyle{remark}
\newtheorem{rem}{\protect\remarkname}
\usepackage{babel}
\providecommand{\remarkname}{Remark}
\providecommand{\assumptionname}{Assumption}
\providecommand{\corollaryname}{Corollary}
\providecommand{\lemmaname}{Lemma}
\providecommand{\propositionname}{Proposition}
\providecommand{\theoremname}{Theorem}

\begin{document}
\title{Robust Estimation of Realized Correlation:\linebreak{}
New Insights about Intraday Fluctuations in Market Betas{\normalsize\emph{\medskip{}
}}}
\author{\textbf{Peter Reinhard Hansen}{\normalsize$^{a}$}\textbf{ and Yiyao
Luo$^{b}$}\bigskip{}
\\
{\normalsize$^{a}$}{\normalsize\emph{University of North Carolina
at Chapel Hill}}\textbf{ }\\
{\normalsize$^{b}$}{\normalsize\emph{Goldman Sachs}}\textbf{\thanks{We are grateful for helpful comments from Ron Gallant and seminar
participants at UNC, Duke University, and the 2023 Annual Meeting
of SoFiE in Seoul, Korea, and the 2024 NBER-NSF Time-Series Conference
in Philadelphia. }}}
\date{{\normalsize\emph{\today}}}
\maketitle
\begin{abstract}
Time-varying volatility is an inherent feature of economic time series
and complicates correlation measurement at high frequencies. While
the quadrant correlation estimator is robust to volatility dynamics,
it has low efficiency under Gaussian sampling. We introduce a subsampled quadrant
correlation estimator that improves efficiency while retaining robustness
and computational simplicity. The estimator is based solely on sign
information and involves no smoothing parameter beyond the sampling
horizon, which makes it attractive for high-frequency financial data
and large cross sections. An empirical
application to a broad panel of U.S. stocks reveals substantial intraday
variation in market betas. Decomposing beta dynamics shows that relative volatility
declines over the day by a similar amount for all stocks, so that the
differences across stocks in how their betas move within the day are
predominantly driven by differences in intraday correlation changes.

\bigskip{}
\end{abstract}
{\small\textit{\noindent JEL:}}{\small{} C13,C14,C32,C58}{\small\par}

\noindent{\small\newpage}{\small\par}

\section{Introduction}

Correlation, a measure of association between two variables, plays
a central role in empirical methods in economics and finance. It is
most commonly estimated using the sample correlation, also known
as the Pearson correlation. Alternative classical measures
include the quadrant estimator, Kendall\textquoteright s tau, Spearman\textquoteright s
rank correlation, and the Gaussian rank correlation; see \citet{Kruskal:1958}
for an overview of their properties and historical development. The
choice among these estimators involves a trade-off between robustness
and efficiency. In particular, time-varying volatility, which is pervasive
in economic time series, generally makes conventional correlation
estimators inconsistent for the constant instantaneous correlation. A recent
addition to this literature is the coefficient of \citet{Chatterjee:2021},
which measures dependence rather than correlation: it is zero if and
only if the two variables are independent and one if and only if
one is a measurable function of the other. It is therefore unsigned,
and in stochastic volatility models with a common volatility factor
it can be positive even when $\rho=0$, because the common factor
makes the increments dependent while leaving them uncorrelated.

This paper develops a subsampled quadrant correlation estimator for
settings with time-varying volatility, such as high-frequency financial
data. Among classical correlation measures, the quadrant estimator
is distinctive in that it remains consistent under heteroskedasticity,
but its efficiency is low: under bivariate normality its asymptotic
relative efficiency at $\rho=0$ is $4/\pi^{2}\approx0.405$. Subsampling
preserves the robustness and substantially improves the efficiency.
The estimator is consistent under time-varying volatility without
any estimate of volatility, it is trivial to compute, and it has no
smoothing parameter beyond the sampling horizon; the only remaining
choices, the sampling frequency and the estimation window, are common
to all estimators in this setting. We study its properties theoretically
and in simulations calibrated to high-frequency financial data.

The theory in Section 3 is developed in a bivariate setting in which
volatilities are time-varying while the correlation is constant over
the estimation window. This isolates the effect that heteroskedasticity
has on correlation estimators. In the empirical analysis we apply
the estimator over short windows, within which constant correlation
is a reasonable approximation, and assess time variation in correlation
by comparing localized estimates across windows. Consistent estimators
of spot and integrated correlations and betas have been developed
for general It\^{o} semimartingale models, including localized and
kernel-based spot estimators, see \citet{KalninaXiu:2017}, \citet{BibingerHautschMalecReiss:2019},
and \citet{LiuLiu:2020}, high-frequency regression and integrated
beta estimators, see \citet{AitSahaliaKalninaXiu:2020} and \citet{OhKimWang:2024},
and the general theory for nonlinear functionals of volatility matrices
in \citet{Chen:2018}. These methods are more efficient in the models
for which they are designed, but they require tuning parameters such
as block sizes, bandwidths, or kernels, and they are not designed
for the problem addressed here, which is not to improve efficiency
for integrated or spot correlation estimation but to obtain a robust
measure when volatility is difficult to estimate. When fully parametric
or kernel-based estimators apply, $Q_{S}$ can serve as a tool for
exploratory analysis and intraday diagnostics. While our analysis
is bivariate, higher-dimensional correlation matrices can be constructed
from pairwise estimates, albeit without a guarantee of positive semidefiniteness,
which can be restored by standard projection methods, such as the
nearest correlation matrix of \citet{Higham:2002}.

To illustrate the practical advantages of our proposed estimator,
we apply it to high-frequency data
for a large panel of U.S. stocks and an exchange-traded fund (SPY).
Decomposing market betas into correlations and relative volatilities
explains the intraday variation in market betas documented in
\citet{AndersenThyrsgaardTodorov:2021}. We find that the correlations
of all stocks in our sample
with the market increase over the trading day, whereas their
volatilities relative to market volatility decrease over the
trading day. The decline in relative volatility is quite similar for
all stocks. The relative volatility during the last hour of active
trading is typically 50\%--75\% of relative volatility during
the first hour of trading. The intraday variation in correlations is far more heterogeneous
across stocks.
About half of the stocks in our sample see their correlation with
the market increase by 50\% or more over the trading day, and for
some the increase is by a factor between two and seven. The intraday
betas for these stocks increase nearly linearly during the trading day. Another set of stocks is characterized
by relatively high correlations with the market during the first hour
of trading. Their correlations also increase during the trading day,
but to a lesser extent such that their market betas decrease
over the trading day. Thus, we find that variations in both intraday
correlation and relative volatility contribute to market beta fluctuations,
but the variation across stocks is primarily driven by their time variation
in correlations with the market. Our empirical results also shed light
on the diurnal patterns in volatility and explain why the market's
$U$-shaped pattern differs from the reverse $J$-shaped patterns of
individual stocks, see \citet{ChristensenHounyoPodolskij:2018}
for the latter.

Time-varying volatility in high-frequency financial data is well documented,
see, e.g., \citet{AndersenBollserslev:1998b}. Similarly, it is well
documented that market microstructure noise can harm realized measures
of volatility, see \citet[1998]{Zhou:1996}\nocite{Zhou:1998}, \citet{zhang-mykland-aitsahalia:05},
\citet{BandiRussell:2006}, and \citet{HansenLunde:JBES2006}. Market
microstructure noise is defined as the difference between the observed
price and the true price. The latter is characterized by certain martingale
properties, whereas the former typically entails some degree of predictability.
Market microstructure noise arises from many intricate aspects of
high-frequency data. For instance, noise can arise from artifacts of
imputation methods and recording and rounding errors. These issues
are all important for correlation estimation, see, e.g., \citet{Reno2003},
\citet{PrecupIori:2007}, \citet{MunnixSchaeferGuhr:2011}, and
\citet[2009]{TothKertesz:2007}.\nocite{TothKertesz:2009} The lack
of synchronicity in observation times induces a type of noise that
is particularly important for covariance and correlation estimation.
This will often manifest as the Epps effect, where the sample correlation
decreases as the sampling frequency increases, see \citet{Epps:1979}.
\citet{HayashiYoshida:2005} proposed an estimator that adjusts for
asynchronicity and \citet{VoevLunde:2007} and \citet{GriffinOomen:2011}
proposed related estimators that are robust to additional forms of
noise. Jumps in prices adversely affect empirical measures, including
realized variances, covariances, and correlations. However, these
effects can be alleviated by truncation methods, see \citet{Mancini:2009}
and \citet{RaymaekersRousseeuw:2021}. Heavy-tailed return distributions
pose related challenges for volatility and covariance estimation,
and a growing literature develops robust methods for this setting,
see \citet{FanKim:2018} and \citet{ShinKimFan:2023}. Because the
estimator proposed in this paper is based solely on sign information,
it is by construction insensitive to heavy tails in the return
distribution.

A standard remedy for market microstructure noise in high-frequency
data is sparse sampling. \citet{AndersenBollerslev:1998a} estimated
realized variances using $5$-minute intraday returns and this sampling
frequency appears to offer a reasonably good compromise between bias
and variance in many applications, see, e.g., \citet{HansenLunde:JBES2006},
\citet{BandiRussell:2008}, and \citet{LiuPattonSheppard:2015}. Realized
measures that utilize more information include the subsampled realized
variance by \citet{zhang-mykland-aitsahalia:05}, the realized kernel
estimator by \citet{BNHLS-RK:2008}, and the pre-averaging estimators
by \citet{JacodLiMyklandPodolskijVetterPreaverage}. These three approaches,
subsampling, realized kernels, and pre-averaging, lead to the same
class of estimators, aside from minor differences caused by end-effects,
see \citet{BNHLS-SubRK:2011}.

Realized correlations are often computed from multivariate estimators,
such as those proposed in \citet{MalliavinMancino:2002}, \citet{BNS:2004},
\citet{ChristensenKinnebrockPodolskij:2010}, \citet{AitSahaliaFanXiu:2010},
\citet{BNHLS-MRK:2011}, and \citet{ChristensenPodolskijVetter:2013},
among others. If volatility varies over the estimation period, the resulting
correlation estimator will generally be inconsistent, aside from
the special cases detailed in Section 3.

The remainder of the paper is organized as follows. We review correlation
estimators and introduce the new subsampled quadrant estimator in
Section 2. We establish theoretical properties for the estimators
in Section 3, including results on consistency, efficiency, and robustness.
We report results from an extensive simulation study in Section 4.
Our empirical analysis is presented in Section 5 followed by our concluding
remarks in Section 6. All proofs are presented in Appendix A. Supplementary
Material, which accompanies this paper, contains additional empirical,
theoretical, and simulation-based results.

\section{Population and Empirical Measures of Correlation}

We begin by reviewing classical correlation measures, starting with
the Pearson correlation.

\subsection{Population Measures}

For two random variables, $X$ and $Y$, with finite and non-zero
variances, the
Pearson correlation coefficient is defined by
\[
\rho=\frac{\sigma_{XY}}{\sqrt{\sigma^{2}_{X}\sigma^{2}_{Y}}},\qquad\text{where}\quad\sigma_{XY}=\mathrm{cov}(X,Y)=\mathbb{E}[(X-\mu_{X})(Y-\mu_{Y})],
\]
$\mu_{X}=\mathbb{E}(X)$, $\mu_{Y}=\mathbb{E}(Y)$, $\sigma^{2}_{X}=\mathrm{var}(X)$,
and $\sigma^{2}_{Y}=\mathrm{var}(Y)$. 

Nowadays, the correlation is commonly understood to mean the Pearson correlation,
but $\rho$ is just one of several classical population measures of
correlation. Another measure is based on sign-concordances, 
\[
\tau=\mathbb{E}[\mathrm{sgn}\{(X-\xi_{X})(Y-\xi_{Y})\}],
\]
where $\mathrm{sgn}(x)$ denotes the sign of $x$, and $\xi_{X}$
and $\xi_{Y}$ are the medians of $X$ and $Y$, respectively. The
parameter $\tau$ is given by the quadrant probabilities of the recentered
variables, $\tilde{X}=X-\xi_{X}$ and $\tilde{Y}=Y-\xi_{Y}$, since
$\tau=\Pr[Z>0]-\Pr[Z<0]$, where $Z=\tilde{X}\tilde{Y}$, and $\tau$
is the population quantity that is estimated by the quadrant estimators
we use below. For variables that are spherical and continuously distributed,
we have $q\equiv\Pr[Z>0]=1-\Pr[Z<0]=(\tau+1)/2$, such that $\tau=2q-1$.
A closely related population measure is Kendall's tau, which is given
by
\[
\tau_{K}=\mathbb{E}[\mathrm{sgn}\{(X_{1}-X_{2})(Y_{1}-Y_{2})\}],
\]
where $(X_{1},Y_{1})$ and $(X_{2},Y_{2})$ are independent and distributed
as $(X,Y)$. For a continuous bivariate distribution with cdf $F$,
it can be shown that $\tau_{K}=\mathbb{E}[4\{F(X,Y)-\tfrac{1}{4}\}]$
whereas $\tau=4[F(\xi_{x},\xi_{y})-\tfrac{1}{4}]$. The two quantities,
$\tau$ and $\tau_{K}$, are identical for elliptical distributions. 

Other classical correlation measures include the Gaussian rank correlation
and Spearman's rank correlation, where the latter estimates $\eta=12\{\mathbb{E}[F(X)G(Y)]-\tfrac{1}{4}\},$
where $F$ and $G$ are the cumulative distribution functions for
$X$ and $Y$, respectively. We do not include these estimators, because they are generally inconsistent
for the constant instantaneous correlation under time-varying volatility,
see the proof of Theorem \ref{theo:Consistency},
and, in the case of the Gaussian rank correlation, not robust.

The population measures, $\rho$, $\tau$, $\tau_{K}$, and $\eta$
are closely related and all have values ranging between $-1$ and
$1$. The exact relation between these quantities depends on the bivariate
distribution of $(X,Y)$. For elliptical distributions we have $\tau(\rho)=\tfrac{2}{\pi}\arcsin\rho$,
such that the inverse mapping is:
\begin{equation}
\rho=\sin(\tfrac{\pi}{2}\tau).\label{eq:Greiner}
\end{equation}
This link function was derived in \citet[p.~236]{Greiner:1909}; the
equivalent statement in the parametrization $\rho=\cos D$ is already
in \citet[\S27]{Sheppard:1899}. Greiner derived
the result under the assumption that $(X,Y)$ are normally distributed,
but (\ref{eq:Greiner}) is valid for all elliptical distributions
without an atom at their center, with $\rho$ defined from the dispersion
matrix, see Proposition \ref{prop:Greiner}. This $\rho$ is the ordinary
correlation coefficient whenever second moments exist, as for the multivariate
$t$-distribution with more than two degrees of freedom, and the identity
continues to hold when they do not.
The link function is also unaffected by skewness as defined by the
moments and cumulants of odd order, see \citet{Kendall:1949}. The
link function in (\ref{eq:Greiner}) makes it possible to translate
an estimate of $\tau$ into an estimate of $\rho$. In general, the
link function between $\tau$ and $\rho$ depends on the underlying bivariate
distribution and we have shown three examples of this link function
in Figure \ref{fig:Greiner's-link-function}. 
\begin{figure}[!h]
\centering{}\includegraphics[width=0.54\textwidth]{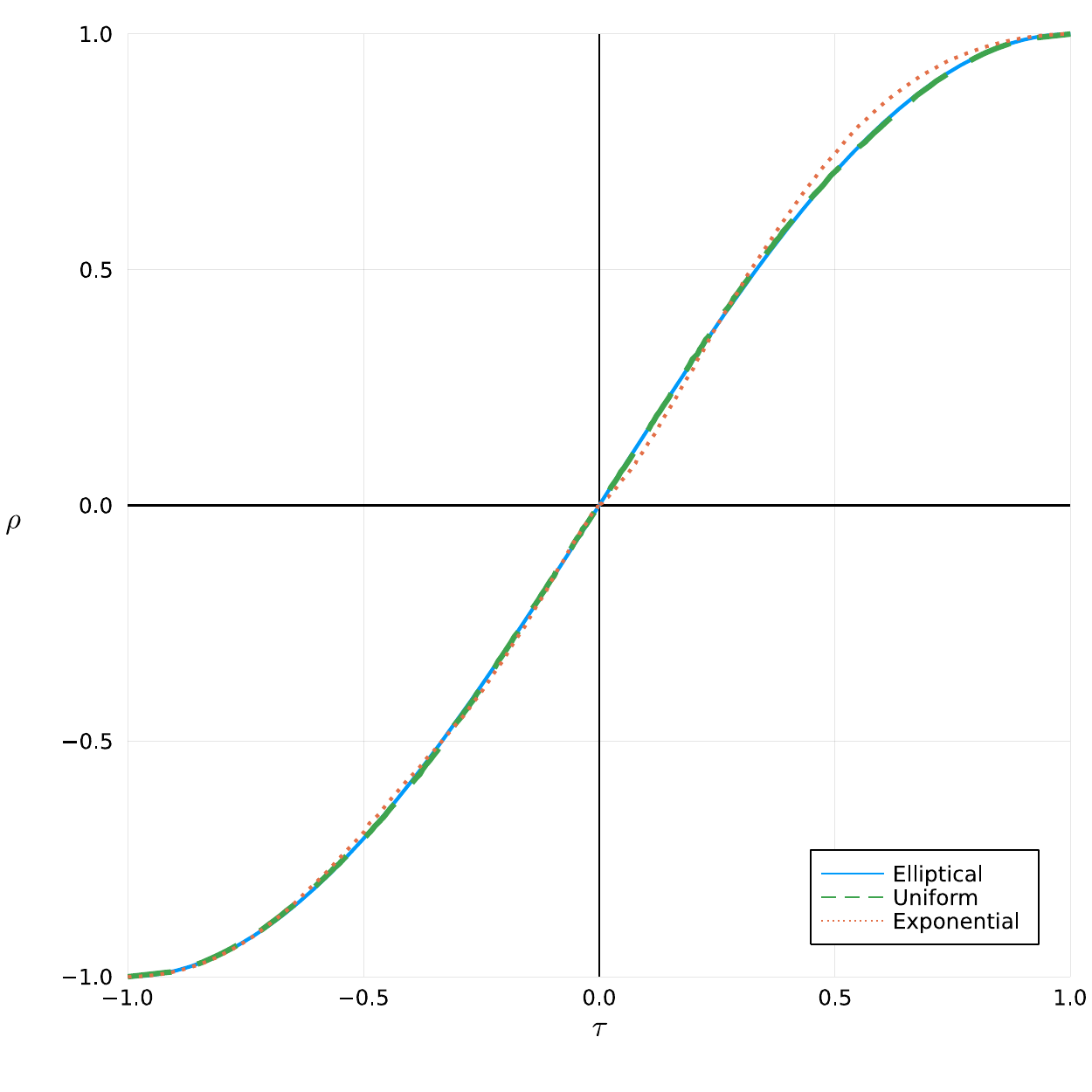}\includegraphics[width=0.18\textwidth]{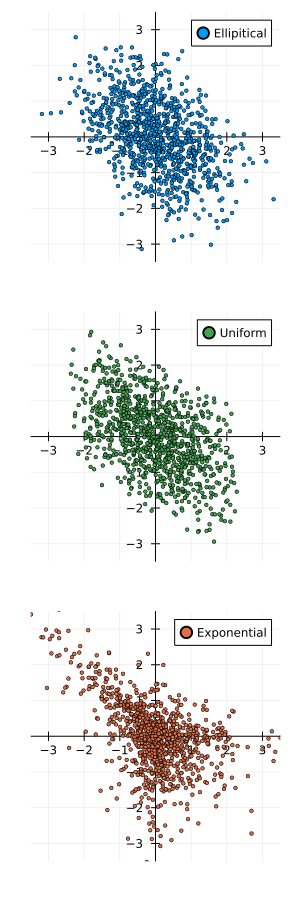}\caption{The mapping from $\tau=\mathbb{E}[\mathrm{sgn}(X-\xi_{x})(Y-\xi_{y})]$
to $\rho=\mathrm{corr}(X,Y)$ for three bivariate distributions.\label{fig:Greiner's-link-function}}
\end{figure}

The blue line represents Greiner's link function, (\ref{eq:Greiner}),
whereas the green dashed and the red dotted lines represent link functions
for two non-elliptical distributions. The green dashed link function,
labeled ``Uniform,'' is based on the following bivariate distribution
\[
X=U_{1}-U_{2}\quad Y=\rho X+\sqrt{1-\rho^{2}}(U_{3}-U_{4}),\qquad\text{\ensuremath{\rho\in[-1,1]}},
\]
where $U_{1}$,..., $U_{4}$ are independent and uniformly distributed
on $[0,1]$, and the red dotted link function, labeled ``Exponential,''
is based on 
\[
X=aZ_{0}-(1-|a|)Z_{1}\quad Y=|a|Z_{0}+(1-|a|)Z_{2},\qquad a\in[-1,1],
\]
where $Z_{0},Z_{1},Z_{2}$ are independent and standard exponentially
distributed. The correlation for this distribution is $\rho(a)=a|a|/[a^{2}+(1-|a|)^{2}]$.
Realizations sampled from these distributions are shown in the right
panels using scatterplots with 1,000 observations. The upper right panel is the
bivariate normal distribution with correlation $-0.5$. The middle
right panel is based on four uniformly distributed random variables
with $\rho=-0.5$, and the lower right panel is based on the exponential
random variables with $a=-0.5$, which happens to translate to the
same correlation, $\rho(-0.5)=-0.5$. Thus Greiner's link function
is a good approximation to the two other link functions in Figure
\ref{fig:Greiner's-link-function}, and may offer a good approximation
to a broader class of distributions beyond elliptical distributions.
However, it is also possible to construct a bivariate distribution
whose link function can differ substantially from that given in (\ref{eq:Greiner}).\footnote{For instance, pathological examples can be created by shifting probability
mass near zero, affecting the binary distribution of signs, with negligible
impact on $\rho$. }

\subsection{Classical Correlation Estimators: Pearson, Quadrant, and Kendall\label{subsec:classical}}

Next, we introduce classical correlation estimators. To simplify
the exposition we use $(x_{i},y_{i})$, $i=1,\ldots,n$, to denote
recentered variables whose sample means (or sample medians)
are zero.\footnote{The sample mean is subtracted before applying the Pearson estimator
and the sample median is subtracted if the Quadrant or Kendall estimators
are used. }

The Pearson correlation estimator is the well-known sample correlation,
which takes the form 

\[
P=\frac{\sum^{n}_{i=1}x_{i}y_{i}}{\sqrt{\sum^{n}_{i=1}x^{2}_{i}\sum^{n}_{i=1}y^{2}_{i}}}.
\]
This estimator is asymptotically efficient if the data are iid Gaussian.
A drawback of the Pearson estimator is that it is sensitive to outliers,
as we discuss below. More robust estimators of $\rho$ can be constructed
from estimators of $\tau$, such as the quadrant estimator

\[
\hat{\tau}_{Q}=\frac{1}{n}\sum^{n}_{i=1}\mathrm{sgn}(x_{i}y_{i}),
\]
and Kendall's tau coefficient
\[
\hat{\tau}_{K}=\frac{2}{n(n-1)}\sum_{i<j}\mathrm{sgn}([x_{i}-x_{j}][y_{i}-y_{j}]).
\]
Quadrant-based estimation of the correlation was introduced in \citet{Sheppard:1899},
with the relation between $\rho$ and $\tau$ spelled out in \citet{Greiner:1909}.
The sampling theory of the quadrant estimator was established in modern
form by \citet{Blomqvist:1950}, although the leading-order comparison
with the moment estimator is already in \citet[\S30]{Sheppard:1899}.
\citet{Esscher:1924} introduced the $\hat{\tau}_{K}$
estimator and cited \citet{Greiner:1909} for the link function. This
estimator was rediscovered in \citet{Kendall:1938} and is commonly
known as Kendall's tau coefficient or the Kendall rank correlation
coefficient. Note that $\hat{\tau}_{K}$ is the quadrant estimator
applied to $\{(X_{i}-X_{j},Y_{i}-Y_{j})\}_{i<j}$, and that $\rho=\mathrm{corr}(X_{1},Y_{1})=\mathrm{corr}(X_{1}-X_{2},Y_{1}-Y_{2})$
if $(X_{1},Y_{1})$ and $(X_{2},Y_{2})$ are independent and identically
distributed, because $\mathrm{cov}(X_{1}-X_{2},Y_{1}-Y_{2})=2\sigma_{XY}$
and $\mathrm{var}(X_{1}-X_{2})=2\sigma^{2}_{X}$, $\mathrm{var}(Y_{1}-Y_{2})=2\sigma^{2}_{Y}$,
so that the factor $2$ cancels in the ratio.

In this paper, we employ Greiner's link function to map the estimators
of $\tau$ to estimators of $\rho$. The estimators of $\rho$ are
therefore defined by
\[
Q=\sin(\tfrac{\pi}{2}\hat{\tau}_{Q})\qquad\text{and}\qquad K=\sin(\tfrac{\pi}{2}\hat{\tau}_{K}),
\]
respectively. Neither estimator requires estimates of the variances,
or even that the variances exist. For non-elliptical distributions, (\ref{eq:Greiner}) may
not be the appropriate link function, and this type of misspecification
can therefore induce a bias in these estimators. Figure \ref{fig:Greiner's-link-function} shows that Greiner's link
function is a good approximation for two non-elliptical distributions.
In our empirical application we use sparsely sampled financial returns,
a setting in which a Gaussian assumption has some theoretical
justification. 

\subsection{A New Correlation Estimator}

Market microstructure noise, jumps, and time-varying volatility all
affect correlation estimators applied to high-frequency data. $P$
is the ideal estimator for i.i.d. Gaussian data, but under time-varying
volatility it is generally inconsistent for the constant instantaneous
correlation. The same is true of $K$, which is more robust, while $Q$
is very inefficient. This
motivates the subsampled variant of $Q$ introduced below. 

\subsubsection{Notation with High-Frequency Data}

Let $X(t)$ and $Y(t)$ denote the observed logarithmically transformed
price processes over some period, such as a trading day. We denote
the intraday returns over a time interval of length $\delta$ by
\[
\Delta_{\delta}X_{t}=X(t)-X(t-\delta),
\]
and similarly for $\Delta_{\delta}Y_{t}$. In the context of high-frequency
data it is common to sample sparsely to mitigate the effects
of market microstructure noise, and a popular choice is to set $\delta$
equal to five minutes. If we normalize the time interval to be $[0,1]$
and set $\delta=\frac{1}{n}$, then the correlation estimators given
above may be applied to $(x_{i},y_{i})=(\Delta_{\frac{1}{n}}X_{\frac{i}{n}},\Delta_{\frac{1}{n}}Y_{\frac{i}{n}})$
for $i=1,\ldots,n$. 

Let $N$ denote the number of intraday returns at the highest possible
sampling frequency and suppose, for simplicity, that $N$ is divisible
by $n$, such that $S=N/n\in\mathbb{N}$. Then we create $S$ distinct
grids by shifting the initial observation time as $t_{s}=s/(Sn)$
for $s=0,\ldots,S-1$. Each grid will generate sparsely sampled and
non-overlapping returns, and we have $N-S+1$ pairs of sparsely
sampled returns, $(\Delta_{\delta}X_{\frac{j}{N}},\Delta_{\delta}Y_{\frac{j}{N}})$,
for $j=S,\ldots,N$.\footnote{For instance, a 6.5-hour trading day has $n=78$ intraday returns
when partitioned into $5$-minute intervals. By shifting the starting
time we obtain partitions with distinct 5-minute returns, each having
just 77 returns. By shifting the starting time in one-minute increments
we obtain $S=5$ different partitions and a total of $386$ 5-minute
returns.}

\subsubsection{Subsampled Quadrant Estimator\label{subsubsec:subsampled}}

We are now ready to introduce the subsampled variant of the quadrant
estimator, defined by

\[
Q_{S}=\sin(\tfrac{\pi}{2}\hat{\tau}_{S})\qquad\text{with}\qquad\hat{\tau}_{S}=\frac{1}{N-S+1}\sum^{N}_{j=S}\mathrm{sgn}(\Delta_{\frac{S}{N}}X_{\frac{j}{N}}\Delta_{\frac{S}{N}}Y_{\frac{j}{N}}).
\]
The estimator does not require $N$ to be divisible by $S$, but if
$N=nS$, then $\hat{\tau}_{S}$ can be expressed in terms of the $S$
$\tau$-estimators based on the different grids. The grid that starts
at the first observation contains $n$ sparse returns and the other
$S-1$ grids contain $n-1$, so with $\hat{\tau}^{(s)}$ denoting the
$\tau$-estimator on grid $s$,
\[
\hat{\tau}_{S}=\frac{n\hat{\tau}^{(0)}+(n-1)\sum^{S-1}_{s=1}\hat{\tau}^{(s)}}{Sn-S+1},
\]
a weighted average of the grid-specific estimators that pools the
sign concordances before the sine transformation is applied, and which
is asymptotically equivalent to their equally weighted average for
fixed $S$. This construction is analogous to many robust estimators of
the long-run variance. \citet{PolitisRomanoWolf99} noted that the
subsampled sample variance is identical to the moving-blocks estimator
and the jackknife variance estimator, and it is almost identical to
the Bartlett estimator \nocite{Bartlett:1950}\citep[1950]{Bartlett:1946},
which is often referred to as the Newey-West estimator in the econometrics
literature.\footnote{\citet[p.98]{PolitisRomanoWolf99}: ``{[}...{]} the variance estimator
$\hat{\sigma}^{2}_{\text{\textsc{sub}}}$ is actually asymptotically
equivalent to the Bartlett kernel estimator{[}...{]}'' and \citet[p.98]{PolitisRomanoWolf99}:
``In addition, $\hat{\sigma}^{2}_{\text{\textsc{sub}}}$ is identical
to the moving blocks bootstrap and/or jackknife variance estimator
of the variance of the sample mean proposed by \citet{Kunsch:89}
and \citet{LiuSingh:92} {[}...{]}''.} In the context of volatility estimation with high-frequency data,
the subsampling idea was first used in \citet{Zhou:1996}. The theoretical
foundation for subsampled realized variances was established in \citet{zhang-mykland-aitsahalia:05}
and \citet{Zhang:2006}, and the close connection between subsampled
estimators and kernel estimators is detailed in \citet{BNHLS-SubRK:2011}.
In this analogy $S$ plays the role of a Bartlett bandwidth, but it
is not a free parameter: with all offsets used, $S=\delta N$ is determined
by the sampling horizon $\delta$.

The subsampled quadrant correlation estimator has several appealing
properties. First, it inherits the robustness of the quadrant estimator
while being more precise than $Q$ (Section \ref{subsec:influence}).
Second, $Q_{S}$ is consistent under time-varying volatility. Third, $Q_{S}$ relies only on binary variables, making it computationally
attractive and easy to scale to large data sets.

The classical quadrant estimator is applied to data recentered at their
sample medians. We do not recenter. For the high-frequency increments
to which $Q_{S}$ is applied, zero is the natural center: the price
processes of Section \ref{sec:properties} are martingales, so the
increments have mean zero, and over intervals of a few minutes the
median return is negligible relative to the standard deviation of
returns. We therefore compute $\hat{\tau}_{S}$ from the raw increments
in the theoretical analysis, in the simulations, and in the empirical
analysis. Theorems \ref{thm:CLTBSM} and \ref{thm:CLTleverage} apply
to this statistic, up to the treatment of zero increments discussed
in Section \ref{subsubsec:implementation} and Corollary \ref{cor:nonzero}.

Table \ref{tab:recentering} compares the two versions. Both were computed
from the same synchronized prices, for all $33$ asset-pair series of
the Small Universe of Section \ref{sec:empirical}, at five sampling
frequencies. At one-second sampling the two are
indistinguishable: the increments are predominantly zero, the sample
median is exactly zero, and recentering has no effect at all. As
sampling coarsens the recentered statistic is lower, in every series, by
$0.016$ on average at the three-minute frequency used in the empirical
analysis.
\begin{table}[H]
\caption{Uncentered versus median-recentered $Q_{S}$\label{tab:recentering}}
\begin{centering}
\begin{tabular*}{\textwidth}{@{\extracolsep{\fill}}lccccc}
\hline
Sampling interval & 1s & 30s & 1min & 3min & 5min\\
\hline
Mean $Q_{S}$, raw increments        & $0.591$ & $0.493$ & $0.488$ & $0.480$ & $0.477$\\
Mean $Q_{S}$, recentered            & $0.591$ & $0.490$ & $0.481$ & $0.464$ & $0.460$\\
Difference                          & $0.000$ & $0.002$ & $0.007$ & $0.016$ & $0.017$\\
Series affected                     & none    & all     & all     & all     & all\\
\hline
\end{tabular*}
\par\end{centering}
Note: Both statistics are computed from the same synchronized prices for
all $33$ asset-pair series considered in the Small Universe of Section \ref{sec:empirical},
over $1{,}763$ trading days, and averaged across series. The difference
is the mean of the raw statistic minus the mean of the recentered one. The
last row reports whether the recentered statistic is lower for none or
for all of the $33$ series.
\end{table}
The table shows that the choice of recentering is not innocuous in
finite samples. Both versions of the estimator are consistent: the
sample median differs from zero by $O_{p}(n^{-1/2})$ relative to the
scale of the increments, so recentering changes the sign of an $O_{p}(n^{-1/2})$
fraction of the products and alters the estimator by a term of the
same order as its sampling error. If that term is $o_{p}(n^{-1/2})$, the two versions share the limit
distributions of Section \ref{sec:properties}; we have not established
this. The sign
of the difference depends on the data. The sample median can fall on
either side of zero, and even a fixed displacement of the center can
raise concordance as well as lower it: for independent standard normals
and a common displacement $a$ of the center, $\mathbb{E}[\mathrm{sgn}(X-a)\mathrm{sgn}(Y-a)]=[2\Phi(a)-1]^{2}>0$
although the correlation is zero. Under leverage, increments with
zero mean need not have zero median, and the displacement of the
conditional median is of relative order $n^{-1/2}$, so it too can
matter at the $\sqrt{n}$ scale. We therefore use the raw increments
throughout and report the recentered estimates only in Table \ref{tab:recentering}.

\subsubsection{Implementation at Ultra High Frequencies\label{subsubsec:implementation}}

One challenge with ultra-high-frequency data is that price increments
can be zero over short time intervals, resulting in $\Delta_{\delta}X\Delta_{\delta}Y=0$.
This issue may be caused by stale prices and rounding to a grid defined
by the minimum tick size. The incidence of zero increments declines
with sparse sampling, but it does not disappear at the frequencies
used in practice and it remains substantial for less liquid assets.
For instance, nearly 20\% of the 3-minute returns for AMD are zero
in our sample period, see Table \ref{tab:StatisticsSmallUniverse},
and $\delta=3$ minutes is the frequency underlying most of our empirical
results. We will, moreover, explore the properties of the estimators
at higher sampling frequencies, where zero returns are far more prevalent,
to gain insight into their behavior and the effects of market microstructure noise. An explicit
convention for zero returns is therefore needed, and we adopt a
convention under which they are discarded,
\[
\hat{\tau}^{\mathrm{nz}}_{S}=\frac{1}{N_{1}}\sum^{N}_{j=S}\mathrm{sgn}(\Delta_{\delta}X_{\frac{j}{N}}\Delta_{\delta}Y_{\frac{j}{N}}),\qquad\text{with}\quad N_{1}=\sum^{N}_{j=S}1_{\{\Delta_{\delta}X_{\frac{j}{N}}\Delta_{\delta}Y_{\frac{j}{N}}\neq0\}},
\]
and $Q^{\mathrm{nz}}_{S}=\sin(\tfrac{\pi}{2}\hat{\tau}^{\mathrm{nz}}_{S})$,
such that only the non-zero terms are counted. Zero terms contribute
nothing to the numerator, so the two statistics differ through the
normalization alone, $\hat{\tau}^{\mathrm{nz}}_{S}=\tfrac{N-S+1}{N_{1}}\hat{\tau}_{S}$.
\begin{rem}
\label{rem:ZeroIncrements}Under Assumption \ref{assu:Brownian}, and
under Assumption \ref{assu:BSM} together with Assumption \ref{assu:BSMCLT},
the two statistics coincide almost surely, such that the convention
is an implementation detail rather than a modification of the estimator.
The volatilities are bounded away from zero in both cases (constant
with $\Sigma$ non-degenerate in the first, bounded below by assumption
in the second) and $\sigma$ is independent of $W$ given $\mathcal{F}_{0}$,
so that the increments $\Delta_{\delta}X$ and $\Delta_{\delta}Y$
are atomless for every $\delta>0$ and $\Pr(\Delta_{\delta}X\Delta_{\delta}Y=0)=0$.\footnote{Conditionally
on the volatility path, $\Delta_{\delta}X$ is Gaussian with variance
$\int\sigma^{2}_{x}(u)\mathrm{d}u\geq c^{2}\delta>0$, hence atomless,
and mixing over the path preserves the absence of atoms. The argument
requires that conditioning on the path leaves $W$ a Wiener process,
which is what the independence of $\sigma$ and $W$ delivers; it is
immediate under Assumption \ref{assu:Brownian} and holds under Assumption
\ref{assu:BSMCLT}, where the volatility is driven by $(W^{\prime},\mu)$
independent of $W$. The positivity of $\sigma$ is essential: a volatility
that is allowed to vanish on an interval would place an atom at zero.} Zero increments in the data are consequently artifacts of market
frictions, primarily price discreteness and stale prices, rather
than a feature of the price process; $N_{1}=N-S+1$ almost surely,
$\hat{\tau}^{\mathrm{nz}}_{S}=\hat{\tau}_{S}$, and Theorems \ref{thm:efficiency}
and \ref{thm:CLTBSM} apply unchanged to $Q^{\mathrm{nz}}_{S}$.

Under leverage, that is, under Assumption \ref{assu:BSM} in general
and under Assumption \ref{assu:Leverage}, the increments are not
conditionally Gaussian and we do not know whether they are atomless.
The asymptotic results do not require this. Lemma \ref{lem:SmallBall}
gives $\Pr(\Delta_{\delta}X\Delta_{\delta}Y=0)\leq C\delta^{2/3}$,
so $\sqrt{n}(\hat{\tau}^{\mathrm{nz}}_{S}-\hat{\tau}_{S})\overset{p}{\rightarrow}0$
by Corollary \ref{cor:nonzero} and Theorem \ref{thm:CLTleverage}
holds for $Q^{\mathrm{nz}}_{S}$. For Theorem \ref{theo:Consistency}
it suffices that the fraction of zero increments vanishes in probability,
which follows from the argument of Lemma \ref{lem:Flip} with $\mathbb{E}[R^{2}_{x,j}]=o(\delta)$. In
the data, where zero increments are common, the convention matters:
$\hat{\tau}^{\mathrm{nz}}_{S}$ estimates the conditional sign concordance,
$\mathbb{E}[\mathrm{sgn}(\Delta_{\delta}X\Delta_{\delta}Y)\mid\Delta_{\delta}X\Delta_{\delta}Y\neq0]$,
i.e., concordance among the increments that actually move. The alternative
convention, which divides by $N-S+1$ and thereby treats every zero
as a non-concordant observation, shrinks $\hat{\tau}$ toward zero
by the factor $N_{1}/(N-S+1)$, and $Q_{S}=\sin(\tfrac{\pi}{2}\hat{\tau})$
accordingly. The difference between the two is
immaterial for liquid assets at moderate sampling frequencies, but
it is not negligible for illiquid assets or at the highest frequencies,
where the shrinkage factor can be far from one. We use $Q^{\mathrm{nz}}_{S}$
in the simulations and in the empirical analysis, and we suppress
the superscript in the sections that follow, where zero increments
have probability zero.
\end{rem}

\section{Properties of the Estimators\label{sec:properties}}

This section establishes the properties of the estimators under
time-varying volatility and constant correlation. We do not derive
asymptotic distributions under time-varying correlation, market microstructure
noise, or non-synchronous sampling. We first consider the ideal case of Brownian
motions with constant volatility and no measurement error, where the
increments are i.i.d. normal, then turn to time-varying volatility, and
finally discuss robustness by means of the influence functions of the
estimators. Market microstructure noise is analyzed in Section \ref{sec:Simulation-Study}.
\begin{rem}
\label{rem:efficiency}The subsampled quadrant estimator is not an
efficient estimator of the correlation. Under ideal circumstances (Assumption
\ref{assu:Brownian} below), the Pearson estimator is the maximum
likelihood estimator, and in general It\^{o} semimartingale models,
more efficient estimators of spot and integrated correlation are
available, see, e.g., \citet{KalninaXiu:2017}, \citet{BibingerHautschMalecReiss:2019},
\citet{AitSahaliaKalninaXiu:2020}, and \citet{Chen:2018}. The case
for $Q_{S}$ is based on robustness to heteroskedasticity, computational
simplicity, and the absence of any smoothing parameter beyond the
sampling horizon, rather than on efficiency.
\end{rem}

\subsection{Limit Distributions under Ideal Circumstances\label{subsec:ideal}}

We begin with the simplest possible situation, where logarithmic price
processes follow Brownian motions with constant volatilities and constant
correlation.
\begin{assumption}
\label{assu:Brownian}Suppose that $(X,Y)$ is given by a bivariate
Brownian motion, such that $(X_{t},Y_{t})\sim N_{2}(0,t\Sigma)$.
\end{assumption}
In this situation, intraday returns, $(\Delta_{\delta}X_{i\delta},\Delta_{\delta}Y_{i\delta})$,
$i=1,\dots,n$ are iid and normally distributed. This is the ideal
situation for the Pearson estimator, because $P$ is the maximum likelihood
estimator of $\rho$ for the sample with the $n$ pairs of observations.
We should therefore expect $P$ to compare favorably to $Q$ and $K$.
It is less obvious how $P$ compares with $Q_{S}$ because $Q_{S}$
utilizes additional data from shifted grids. The asymptotic distribution
of the new estimator is given next, where $\mathrm{asin}$ denotes
the inverse sine function, $\mathrm{asin}:[-1,1]\rightarrow[-\tfrac{\pi}{2},\tfrac{\pi}{2}]$.
\begin{thm}
\label{thm:efficiency}Suppose that Assumption \ref{assu:Brownian}
holds and let $S\in\mathbb{N}$ be fixed. Then the subsampled quadrant
correlation estimator is asymptotically normally distributed 
\[
\sqrt{n}(Q_{S}-\rho)\overset{d}{\rightarrow}N(0,V_{S}(\rho)),
\]
as $n\rightarrow\infty$ , where

\[
V_{S}(\rho)=(1-\rho^{2})\frac{1}{S}\sum^{S}_{s=-S}\left[\mathrm{asin}^{2}(w_{s})-\mathrm{asin}^{2}(w_{s}\rho)\right],\qquad w_{s}=\tfrac{S-|s|}{S},
\]
which is decreasing in $S$ and bounded from below by 
\[
V_{\infty}(\rho)=\lim_{S\rightarrow\infty}V_{S}(\rho)=(1-\rho^{2})2\left[\mathrm{asin}^{2}(1)-2\tfrac{\sqrt{1-\rho^{2}}}{\rho}\mathrm{asin}(\rho)-\mathrm{asin}^{2}(\rho)\right],
\]
where the expression is interpreted by continuity at $\rho=0$, with
$V_{\infty}(0)=\tfrac{\pi^{2}}{2}-4\approx0.9348$.
\end{thm}
The corresponding asymptotic distributions for the estimators, $Q$,
$K$, and $P$, are well known, see, e.g., \citet{CrouxDehon:2010}.
For the sake of comparison, these are included below.
\begin{prop}
If Assumption \ref{assu:Brownian} holds, then as $n\rightarrow\infty$
we have 
\begin{eqnarray*}
\sqrt{n}(P-\rho) & \stackrel{d}{\rightarrow} & N(0,V_{P}),\qquad\text{with}\quad V_{P}=(1-\rho^{2})^{2},\\
\sqrt{n}(K-\rho) & \stackrel{d}{\rightarrow} & N(0,V_{K}),\qquad\text{with}\quad V_{K}=(1-\rho^{2})[(\tfrac{\pi}{3})^{2}-4\mathrm{asin}^{2}(\tfrac{\rho}{2})],\\
\sqrt{n}(Q-\rho) & \stackrel{d}{\rightarrow} & N(0,V_{Q}),\qquad\text{with}\quad V_{Q}=(1-\rho^{2})[(\tfrac{\pi}{2})^{2}-\mathrm{asin}^{2}(\rho)].
\end{eqnarray*}
\end{prop}
We can compare the asymptotic variance of $Q_{S}$ with those of the
other estimators. The asymptotic variances depend on $\rho$ and those
for $P$, $K$, and $Q_{S}$ are shown in Figure \ref{fig:avar_asym}.
For $S=1$, we obviously have $V_{Q_{S}}=V_{Q}$.\footnote{With $S=1$ we have $\sum^{S}_{s=-S}\arcsin^{2}(\tfrac{S-|s|}{S})-\arcsin^{2}(\tfrac{S-|s|}{S}\rho)=\arcsin^{2}(1)-\arcsin^{2}(\rho)=\tfrac{\pi^{2}}{4}-\arcsin^{2}\rho$
such that $V_{Q_{S}}=V_{Q}$, as expected.} For a sufficiently large $S$, the subsampled quadrant estimator
is more accurate than the Kendall estimator. For small absolute values of $\rho$,
$Q_{S}$ is more accurate than $K$ when $S\geq5$, whereas a larger
value of $S$ is required for larger absolute values of $\rho$. The $Q_{S}$
estimator is similar to $P$ for large values of $S$, with $Q_{S}$
having the edge for moderate levels of correlation, whereas $P$
has the edge as $\rho$ approaches $\pm1$. 

Realized measures are commonly computed from sparsely sampled returns
(e.g., 5-minute intervals) to mitigate market microstructure noise.
There will typically be a large number of observations within each
5-minute interval, and this makes it possible to use a relatively
large value for $S$.

\begin{figure}[H]
\begin{centering}
\includegraphics[width=0.9\textwidth]{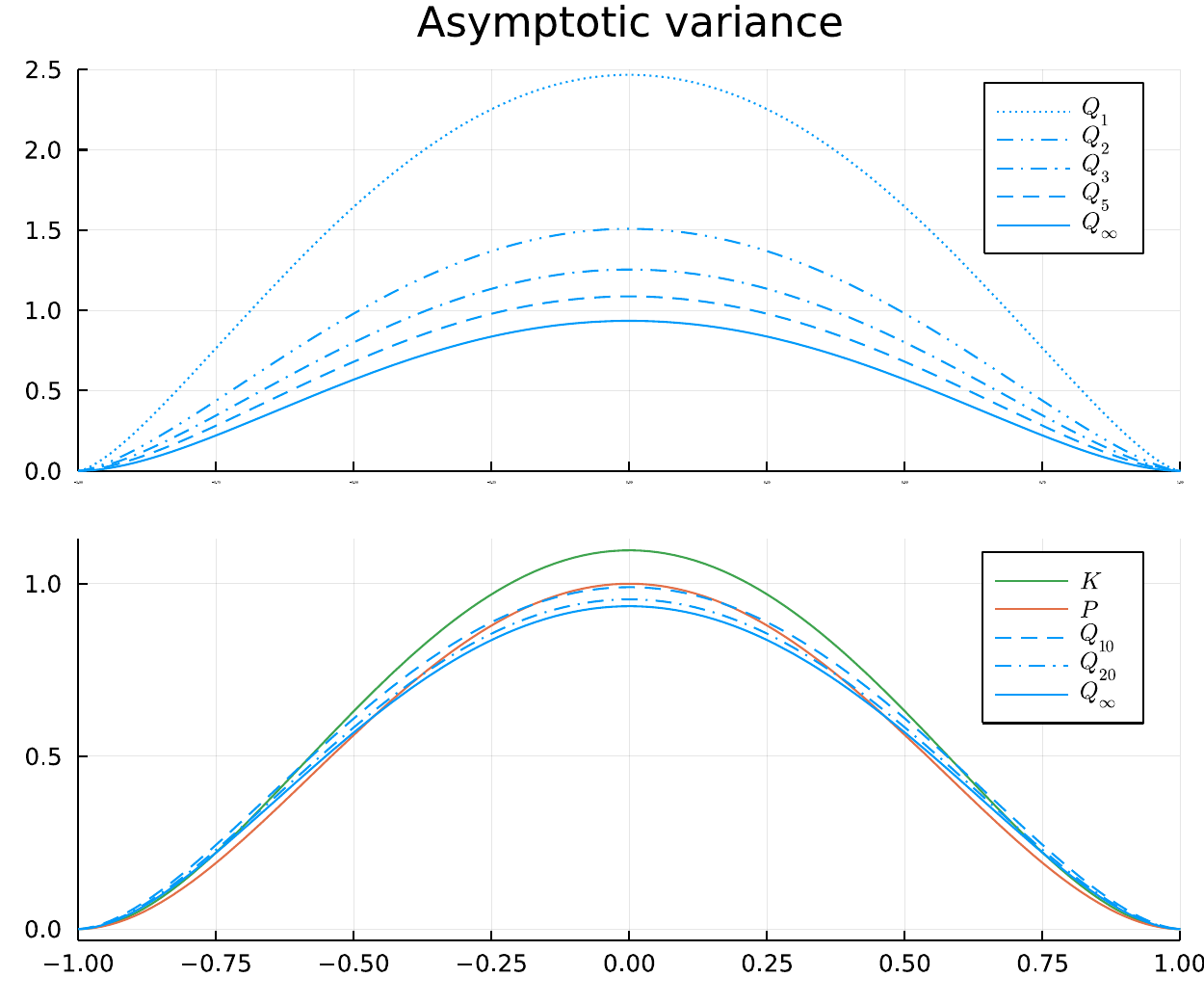}
\par\end{centering}
\caption{Asymptotic variances as a function of the true correlation, $\rho$,
for quadrant and subsampled quadrant estimators in the upper panel
and subsampled Quadrant, Pearson, and Kendall estimators in the lower
panel.}

\label{fig:avar_asym}
\end{figure}

\subsection{Properties with Time-Varying Volatility}

Volatility is found to vary substantially in high-frequency financial
data, even within a trading day. Next, we
relax the assumption of constant volatility and evaluate its effect
on the correlation estimators. The asymptotic properties
of correlation estimators stated above may not apply in this context,
because they were derived under constant volatility. 

We can illustrate the issues that arise from time-varying volatility
with a simple bivariate Brownian semimartingale. 
\begin{assumption}
\label{assu:BSM}Suppose that
\begin{equation}
\left(\begin{array}{c}
X_{t}\\
Y_{t}
\end{array}\right)=\left(\begin{array}{c}
X_{0}\\
Y_{0}
\end{array}\right)+\int^{t}_{0}\sigma(u)\mathrm{d}W(u),\label{eq:ito}
\end{equation}
where $W(u)$ is a bivariate Wiener process with $\mathrm{cor}(\mathrm{d}W_{x},\mathrm{d}W_{y})=\rho$
and 
\[
\sigma(u)=\begin{pmatrix}\sigma_{x}(u) & 0\\
0 & \sigma_{y}(u)
\end{pmatrix},
\]
is a CADLAG process that satisfies $\int^{1}_{0}[\sigma^{2}_{x}(u)+\sigma^{2}_{y}(u)]\mathrm{d}u<\infty$
and
\[
\inf_{u\in[0,1]}\min\{\sigma_{x}(u),\sigma_{y}(u)\}>0,\qquad\text{almost surely}.
\]
\end{assumption}
The requirement that the volatilities are bounded away from zero is
used in the proof of Theorem \ref{theo:Consistency} below, as well
as in the proof of Theorem \ref{thm:CLTBSM}, and, when combined with
the independence of $\sigma$ and $W$, it is what makes the increments
atomless, cf. Remark \ref{rem:ZeroIncrements}. It rules out intervals
over which one of the price processes is locally constant, in which
case the sign of the increment product would be undefined on a set
of positive probability. The assumption can be
generalized in many ways, such as having a random drift term,\footnote{Formally, we can let the logarithmic price process be defined on the
filtered probability space $(\Omega,\mathcal{F},(\mathcal{F})_{t\in[0,1]},\mathbb{P})$
with a locally bounded predictable drift function, $a(u)$, where
$a$, $\sigma$, and $W$ are adapted to a $\mathcal{F}_{t}$.} but the simple setup presented here suffices to show that traditional
correlation estimators become biased in the presence of time-varying
volatility and to establish that quadrant-based estimators are robust
to time-varying volatility. We initially take the correlation
coefficient to be constant over time. The case with time-varying correlation
is discussed below in Section \ref{subsec:tvcorr}.
\begin{thm}
\label{theo:Consistency}If Assumption \ref{assu:BSM} holds, then
$Q_{S}\overset{p}{\rightarrow}\rho$ as $\delta\rightarrow0$, whereas
\[
P\overset{p}{\rightarrow}\lambda\rho,\qquad\text{where}\quad\lambda=\tfrac{\int^{1}_{0}\sigma_{x}(u)\sigma_{y}(u)\mathrm{d}u}{\sqrt{\int^{1}_{0}\sigma^{2}_{x}(u)\mathrm{d}u\int^{1}_{0}\sigma^{2}_{y}(u)\mathrm{d}u}}.
\]
If, in addition, $\sigma$ is independent of $W$, then
\[
K\overset{p}{\rightarrow}\sin\left(\int^{1}_{0}\int^{1}_{0}\arcsin\left(h(u,v)\rho\right)\mathrm{d}u\mathrm{d}v\right),
\]
where $h(u,v)=\tfrac{\sigma_{x}(u)\sigma_{y}(u)+\sigma_{x}(v)\sigma_{y}(v)}{\sqrt{\sigma^{2}_{x}(u)+\sigma^{2}_{x}(v)}\sqrt{\sigma^{2}_{y}(u)+\sigma^{2}_{y}(v)}}$.
\end{thm}
The consistency of $Q_{S}$ and the limit of $P$ require nothing beyond
Assumption \ref{assu:BSM}; in particular, $\sigma$ may depend on $W$.
The limit of $K$ is derived for the case where volatility is independent
of the price innovations, because it uses a Gaussian calculation
for pairs of increments at different times. The probability limit
of $K$ under leverage is not established here.
The important consequence of Theorem \ref{theo:Consistency} is that
$Q_{S}$ is the only one of the three estimators that is consistent
under time-varying volatility in general. Both $P$ and $K$ are inconsistent
except in two special cases: when $\rho=0$, in which case both limits equal
zero because $\lambda\rho=0$ and $\arcsin(0)=0$, and under perfect
volatility collinearity, as follows.
\begin{figure}[tbh]
\begin{centering}
\includegraphics[scale=0.8]{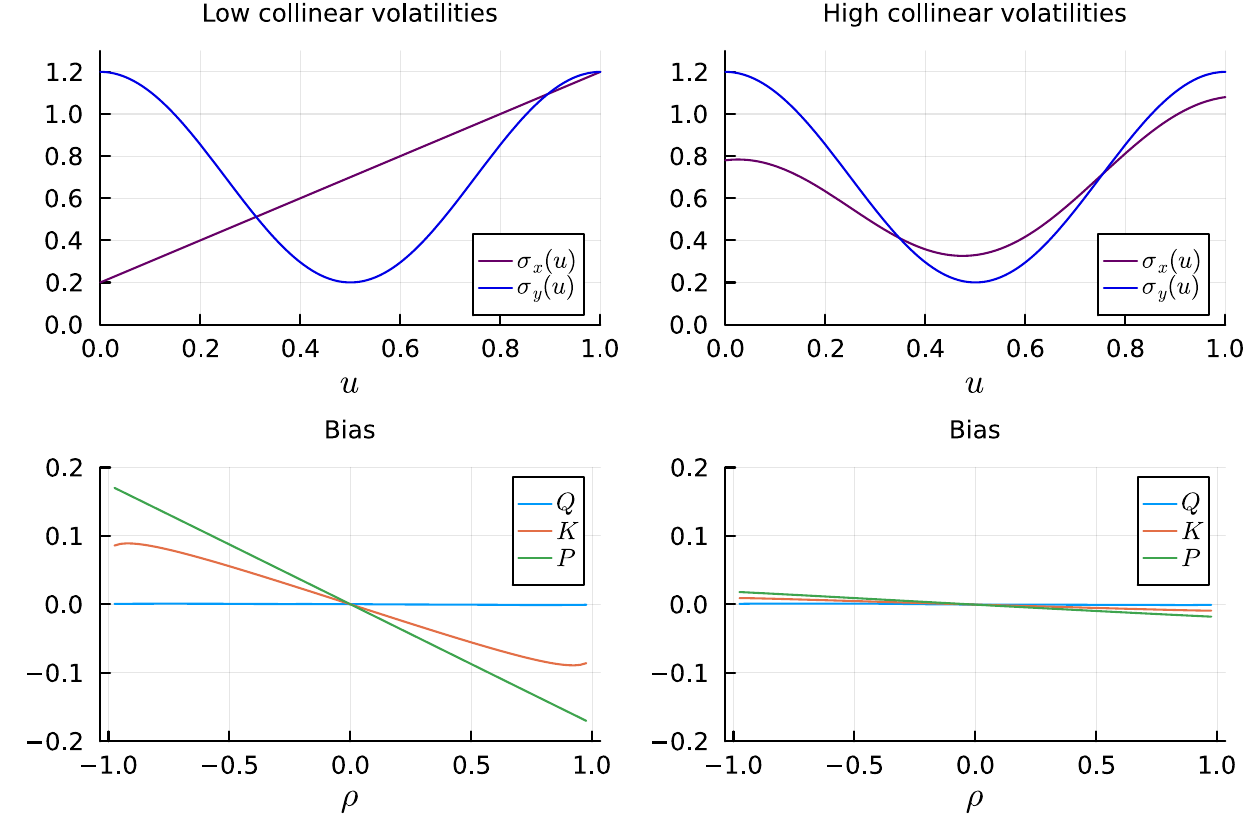}
\par\end{centering}
\caption{Bias in correlation estimators, $P$, $K$, and $Q_{S}$, with two
different degrees of collinearity between the volatilities.\label{fig:BiasTVvol}}
\end{figure}

\begin{cor}
Suppose that Assumption \ref{assu:BSM} holds and that $\sigma_{y}(u)=c\sigma_{x}(u)$
for some $c>0$. Then $P\overset{p}{\rightarrow}\rho$, and if $\sigma$
is independent of $W$, also $K\overset{p}{\rightarrow}\rho$.
\end{cor}
The results for $P$ and $K$ in this special case are easy to verify,
because perfect collinearity implies $\lambda=1$ and that

\[
h(u,v)=\frac{c\sigma^{2}_{x}(u)+c\sigma^{2}_{x}(v)}{\sqrt{\sigma^{2}_{x}(u)+\sigma^{2}_{x}(v)}\sqrt{c^{2}(\sigma^{2}_{x}(u)+\sigma^{2}_{x}(v))}}=1,
\]
for all $u,v$.

Thus, for $\rho\neq0$, it requires perfect volatility collinearity,
$\sigma_{y}(u)=c\sigma_{x}(u)$, for $P$ to be consistent, and likewise
for $K$ in the case where $\sigma$ is independent of $W$. We illustrate the inconsistency
of estimators $P$ and $K$ with a simple example. Consider the functions,
$g(u)=\tfrac{1}{5}+\tfrac{4}{5}u$ and $h(u)=\frac{1}{2}(\tfrac{6}{5}+\cos(2\pi u))$,
which we will use to construct volatility paths with varying degrees
of volatility collinearity. The upper left panel in Figure \ref{fig:BiasTVvol}
represents a case with low volatility collinearity, where $\sigma_{x}(u)=g(u)$
and $\sigma_{y}(u)=h(u)$, and the upper right panel corresponds to
a case with high collinearity, where $\sigma_{x}(u)=\frac{3}{10}g(u)+\frac{6}{10}h(u)$
and $\sigma_{y}(u)=h(u)$. The lower panels show the resulting bias
of the correlation coefficients, $P$, $K$, and $Q_{S}$, as a function
of the true correlation coefficient, $\rho$. With low collinearity,
the sample correlation, $P$, has a large bias unless $\rho$ is close
to zero, while the bias in $K$ is roughly half that of $P$. These
estimators, $P$ and $K$, are also biased in the example with high
collinearity, but the bias is substantially smaller. The biases of
these estimators are also pronounced in a standard simulation design
with the Heston model, as we document in Section \ref{sec:Simulation-Study}.

Theorem \ref{theo:Consistency} shows that the probability limit of
$Q_{S}$ is unaffected by time-varying volatility. The next two results
show that the same holds for its asymptotic distribution under a regularity
condition on the volatility paths. Theorem \ref{thm:CLTBSM} assumes
that the volatility innovations are independent of the price innovations
(no leverage); Theorem \ref{thm:CLTleverage} allows the volatility
to load on the price innovations, under slightly stronger regularity.
\begin{assumption}
\label{assu:BSMCLT}Suppose that $\sigma_{x}$ and $\sigma_{y}$ are
It\^{o} semimartingales of the form,
\begin{equation}
\sigma_{z}(t)=\sigma_{z}(0)+\int^{t}_{0}\tilde{a}_{z}(u)\mathrm{d}u+\int^{t}_{0}\nu^{\prime}_{z}(u)\mathrm{d}W^{\prime}(u)+\int^{t}_{0}\int_{E}\gamma_{z}(u,e)\tilde{\mu}(\mathrm{d}u,\mathrm{d}e),\label{eq:Grigelionis}
\end{equation}
for $z\in\{x,y\}$ and $t\in[0,1]$, where
\begin{enumerate}
\item [(i)] $W^{\prime}$ is a Wiener process, possibly multivariate, that
is independent of $W$;
\item [(ii)] $\mu$ is a Poisson random measure on $[0,1]\times E$ with
compensator $\bar{\mu}(\mathrm{d}u,\mathrm{d}e)=\mathrm{d}u\otimes\lambda(\mathrm{d}e)$,
where $\lambda$ is a $\sigma$-finite measure on a Polish space $E$,
$\tilde{\mu}=\mu-\bar{\mu}$, and $\mu$ is independent of $(W,W^{\prime})$;
\item [(iii)] the coefficients $\tilde{a}_{z}$ and $\nu^{\prime}_{z}$
are CADLAG, locally bounded, and adapted to the filtration generated
by $\mathcal{F}_{0}$ and $(W^{\prime},\mu)$; and
\item [(iv)] the jump coefficient $\gamma_{z}$ is predictable with respect
to that same filtration and satisfies the envelope condition $|\gamma_{z}(u,e)|\leq\Gamma(e)$
for a deterministic function $\Gamma$ with $\int_{E}\Gamma^{2}(e)\lambda(\mathrm{d}e)<\infty$.
\end{enumerate}
\end{assumption}
Equation (\ref{eq:Grigelionis}) is the Grigelionis representation
of an It\^{o} semimartingale, see \citet[assumption (H)]{JacodProtter:2012},
specialized to the case where the volatility is driven by innovations
$(W^{\prime},\mu)$ that are independent of the price innovations $W$.
Volatility jumps are permitted, and the envelope condition in (iv)
is a square-integrability requirement rather than a restriction on
jump activity: jumps of infinite activity are allowed, provided their
compensated sum is square integrable. What (iii) and (iv) exclude
is the \emph{leverage} channel, since together they make $\sigma$
independent of $W$ given $\mathcal{F}_{0}$. This is what Theorem
\ref{thm:CLTleverage} below relaxes. The positivity of the volatilities
is not repeated here, because it is part of Assumption \ref{assu:BSM},
which is maintained throughout.
\begin{thm}
\label{thm:CLTBSM}Suppose that Assumptions \ref{assu:BSM} and \ref{assu:BSMCLT}
hold, let $S\in\mathbb{N}$ be fixed, and let $|\rho|<1$. Then, as
$\delta=S/N\rightarrow0$,
\[
\sqrt{n}\left(Q_{S}-\rho\right)\overset{d}{\rightarrow}N\left(0,V_{S}(\rho)\right),
\]
where $n=N/S$ and $V_{S}(\rho)$ is the asymptotic variance of Theorem
\ref{thm:efficiency}.
\end{thm}
Thus, time-varying volatility affects neither the probability limit
nor the asymptotic distribution of $Q_{S}$.

The proof is short for the following reason. Conditionally
on the volatility path over a window, the pair of increments is exactly
bivariate Gaussian, so the conditional expectation of the sign product
is exactly $\tfrac{2}{\pi}\mathrm{asin}(\rho_{j})$, where $\rho_{j}=\rho h_{j}$
and $h_{j}\leq1$ is the Cauchy-Schwarz ratio of the window volatilities,
cf. Theorem \ref{theo:Consistency}. The only bias is $\rho-\rho_{j}$,
which is second order in the oscillation of $\sigma$.

Leverage breaks the conditional Gaussianity, because the volatility
over the window is then a function of the same increments whose signs
are being taken. The conclusion nevertheless survives, for a reason
we explain after the theorem. It requires the standard second-order
regularity of \citet[assumption (K)]{JacodProtter:2012}, under which
the coefficients of $\sigma$ are themselves It\^{o} semimartingales.
\setcounter{assumption}{2}
\renewcommand{\theassumption}{\arabic{assumption}'}
\begin{assumption}
\label{assu:Leverage}Suppose that $\sigma_{x}$ and $\sigma_{y}$
are It\^{o} semimartingales of the form (\ref{eq:Grigelionis}),
augmented with the leverage terms $\sum_{b\in\{x,y\}}\int^{t}_{0}\nu_{zb}(u)\mathrm{d}W_{b}(u)$,
where (i) and (ii) of Assumption \ref{assu:BSMCLT} hold, and where
\begin{enumerate}
\item [(iii')] the coefficients $\tilde{a}_{z}$, $\nu_{zb}$,
and $\nu^{\prime}_{z}$ are adapted to $(\mathcal{F}_{t})$, locally
bounded, and are themselves It\^{o} semimartingales with locally bounded
characteristics; and
\item [(iv')] the jump coefficient $\gamma_{z}$ is predictable
with respect to $(\mathcal{F}_{t})$ and satisfies the envelope condition
of Assumption \ref{assu:BSMCLT}(iv), and $u\mapsto\gamma_{z}(u,\cdot)$
is an It\^{o} semimartingale in $L^{2}(\lambda)$ with locally bounded
characteristics.
\end{enumerate}
\end{assumption}
\renewcommand{\theassumption}{\arabic{assumption}}
\setcounter{assumption}{3}
The coefficients $\nu_{zb}$ are the leverage loadings: $\sigma$ now
loads on $W$ itself, and the volatility over a window is correlated
with the increments observed over that window. The jump \emph{sizes}
may also load on $W$, because $\gamma_{z}$ is predictable with respect
to the full filtration and can therefore depend on the path of $W$
up to the jump time, whereas the jump \emph{times} remain independent
of $(W,W^{\prime})$ by (ii). The strengthening relative to Assumption
\ref{assu:BSMCLT} is that the coefficients are required to be It\^{o}
semimartingales rather than merely CADLAG. This corresponds to assumption (K) in
\citet{JacodProtter:2012}, and it is what leads to the rate $\mathbb{E}[(\sigma(u)-\sigma(t))^{2}]\leq C(u-t)$
for the coefficients as well as for $\sigma$ itself. It is needed
because the leverage argument proceeds by freezing the coefficients
at the start of each window, and the error introduced by this approximation
must be $o(\delta^{1/2})$.
\begin{thm}
\label{thm:CLTleverage}Suppose that Assumptions \ref{assu:BSM} and
\ref{assu:Leverage} hold, let $S\in\mathbb{N}$ be fixed, and let
$|\rho|<1$. Then the conclusion of Theorem \ref{thm:CLTBSM} continues
to hold.
\end{thm}
In the proof, we use the fact that $\mathrm{sgn}(xy)$
is an \emph{even} function of $(x,y)$, which is the same feature that
makes limit theorems for even functionals of high-frequency increments
robust to leverage, see \citet{BNGJPS:2006} and \citet[chapter 5]{JacodProtter:2012}.
The leverage contribution to an increment is, to first order, a quadratic
functional of the underlying Brownian motion, hence invariant under the
reflection $(W_{x},W_{y})\mapsto(-W_{x},-W_{y})$, whereas the increment
itself changes sign. Averaging over the reflection therefore cancels
the leading leverage contribution, and the remaining terms are negligible
at the central-limit-theorem scale.
\begin{rem}
\label{rem:pivotal}The asymptotic variance in Theorem \ref{thm:CLTBSM}
is deterministic and only depends on $(\rho,S)$. This contrasts with
typical high-frequency limit theory, where the asymptotic variance
is a random functional of the volatility path, such that feasible
inference requires auxiliary estimates of spot volatility and stable
convergence arguments. Because $\mathrm{sgn}(\Delta X\Delta Y)$
is scale invariant, the volatility path is asymptotically eliminated
from the estimator as well as from its sampling variation.
Moreover, it follows that the statistic can be studentized with
the plug-in variance $V_{S}(Q_{S})$. Since $\rho\mapsto V_{S}(\rho)$
is continuous and strictly positive for $|\rho|<1$, and $Q_{S}\overset{p}{\rightarrow}\rho$
by Theorem \ref{theo:Consistency}, we have $V_{S}(Q_{S})\overset{p}{\rightarrow}V_{S}(\rho)$,
and Slutsky's theorem yields
\[
\sqrt{n}(Q_{S}-\rho)/\sqrt{V_{S}(Q_{S})}\overset{d}{\rightarrow}N(0,1).
\]
A confidence interval with asymptotic coverage $1-\alpha$ is therefore
given by $Q_{S}\pm z_{1-\alpha/2}\sqrt{V_{S}(Q_{S})/n}$, which requires
no estimate of any volatility functional and no smoothing parameter
beyond the sampling horizon already used by $Q_{S}$.
\end{rem}

\begin{rem}
\label{rem:Roughness}Our central limit theorems are not established
under CADLAG volatility alone. Heterogeneity
of $\sigma$ within a window attenuates the effective correlation of
the sparsely sampled increments by a Cauchy-Schwarz factor of order
$\mathrm{osc}^{2}_{j}$, where
\[
\mathrm{osc}_{j}(f)=\sup_{s,u\in[t_{j-S},t_{j}]}|f(s)-f(u)|
\]
is the oscillation of a function $f$ over the $j$-th window, $[t_{j-S},t_{j}]$
with $t_{j}=j/N$, of length $\delta$; the bound is derived in the
proof of Lemma \ref{lem:CondBias}. The no-leverage bias argument in Lemma \ref{lem:CondBias}(a) requires
only the averaged $L^{2}$ modulus bound
\[
\frac{1}{N}\sum^{N}_{j=S}\mathbb{E}\big[\mathrm{osc}^{2}_{j}(\sigma_{x})+\mathrm{osc}^{2}_{j}(\sigma_{y})\big]=o(\delta^{1/2}).
\]
Assumption \ref{assu:BSMCLT} delivers this average at the rate $O(\delta)$, and so does Assumption \ref{assu:Leverage}; under
leverage the proof additionally uses the regularity of the coefficients
and the cancellation of the leading quadratic term in Lemma \ref{lem:CondBias}(b),
which the modulus bound alone does not supply. Being
an average, the condition is compatible with jumps in $\sigma$, including
jumps of infinite activity. Under infinite activity every window
contains infinitely many jumps, but the compensated jump component
of $\sigma$ is a square-integrable martingale, so that the Burkholder-Davis-Gundy and Doob inequalities
bound its expected squared oscillation over a window of length $\delta$
by $C\delta\int_{E}\Gamma^{2}(e)\lambda(\mathrm{d}e)$, which is $O(\delta)$
under the envelope condition of Assumption \ref{assu:BSMCLT}(iv)
irrespective of the activity index; the same bound appears as
(\ref{eq:Aprob}) in the proof of Lemma \ref{lem:CondBias}. A pathwise H\"{o}lder condition would exclude paths with jumps; the
averaged condition does not.

For continuous volatility paths that are H\"{o}lder
continuous of order $\alpha$, the average is of order $\delta^{2\alpha}$,
so that $\alpha>1/4$ suffices, whereas the condition can fail for
rougher trajectories, such as models calibrated with a Hurst index
below $1/4$, see \citet{GatheralJaissonRosenbaum:2018}. The condition is sufficient, but not necessary: the attenuation
vanishes when the two volatility paths are proportional, however rough they are, so roughness alone need not
induce a bias. When the condition fails and the deficiency does not
cancel, the attenuation is always toward zero in the absence of
leverage, so that it cannot average out across windows; when $\sigma$
loads on $W$ the sign of the residual bias is not determined by this
argument. In either case $Q_{S}$ remains
consistent, since Theorem \ref{theo:Consistency} only requires CADLAG
paths, but $\sqrt{n}(Q_{S}-\rho)$ need not be asymptotically centered
at zero. Under rough volatility with non-proportional paths, the argument
suggests an attenuation of order $\delta^{2H}$, where $H$ denotes
the Hurst index, for any correlation estimator subject to within-window
heterogeneity; we have not established this rate, and it would require
additional conditions. Such an attenuation would be a discretization
effect, akin to the Epps effect.
\end{rem}

\subsection{Time-Varying Correlations\label{subsec:tvcorr}}

The correlation may be time-varying, as is the case for volatility.
To accommodate this situation we could modify Assumption \ref{assu:BSM}
and let $\rho(u)$ be a CADLAG process. In this situation, the integrated
correlation, $\rho_{\bullet}=\int^{1}_{0}\rho(u)\mathrm{d}u$, is
a natural object of interest. Unfortunately, none of the correlation
estimators are consistent for $\rho_{\bullet}$. For instance, $Q_{S}$
will estimate $\tilde{\rho}_{\bullet}=\sin[\int^{1}_{0}\mathrm{asin}\{\rho(u)\}\mathrm{d}u]$,
and since $\mathrm{asin}(t)$ is strictly convex for $t>0$ and concave
for $t<0$, it follows that $\tilde{\rho}_{\bullet}\geq\rho_{\bullet}$
if $\rho(u)\geq0$ for all $u$ and $\tilde{\rho}_{\bullet}\leq\rho_{\bullet}$
if $\rho(u)\leq0$ for all $u$.

One way to account for time-varying correlations is to estimate them
over short intervals and aggregate these estimates to approximate
$\int^{1}_{0}\rho(u)\mathrm{d}u$. This approach was used in jump-robust
estimation of the integrated covariance in \citet{BoudtCornelissenCroux:2012jump}.
In our empirical analysis we will also use local estimates of $\rho$
to assess time-variation in $\rho(u)$.

A localized version of $Q_{S}$, computed over a moving window or
over local blocks, is a spot-type estimator of $\rho(u)$, and averaging
such local estimates yields an estimator of the integrated correlation,
$\rho_{\bullet}$, up to a quadrature error that vanishes as the blocks
shrink. Spot correlations and related functionals can be estimated consistently
with localized and kernel-based methods, see \citet{KalninaXiu:2017}
and \citet{BibingerHautschMalecReiss:2019}, also in the presence of
infinite-variation jumps, see \citet{LiuLiu:2020}. Spot estimates
of regression coefficients, including time-varying betas, can similarly
be averaged to recover integrated quantities, see \citet{AitSahaliaKalninaXiu:2020}
and \citet{OhKimWang:2024}, and \citet{Chen:2018} provides a general
theory for nonlinear functionals of volatility matrices, which includes
integrated correlation as a special case. The asymptotic properties
of an averaged local $Q_{S}$ estimator could presumably be established
with the methods in \citet{Chen:2018}, and such an estimator would
be consistent for $\rho_{\bullet}$ under time-varying correlation.
In our empirical analysis, we use rolling-window estimates
of $Q_{S}$ as localized measures of correlation, which is in the
spirit of a local-block approach.

\subsubsection{A Simulation Design with Time-Varying Correlation}

Next, we verify the following two claims: that the full-day $Q_{S}$
converges to $\tilde{\rho}_{\bullet}$ rather than to $\rho_{\bullet}$,
and that the block-averaged $Q_{S}$ recovers $\rho_{\bullet}$ up to
a quadrature error, and we also quantify what the localized estimator gains,
in a design where the two limits can be computed in closed form.
The volatilities follow the Heston model calibrated in Section \ref{sec:Simulation-Study},
and the correlation varies deterministically over the day with $\mathrm{asin}\,\rho(u)$
linear in $u$, from $\rho(0)=0.25$ to $\rho(1)=0.60$. This mimics
the intraday increase in market correlation that we document in the
empirical analysis below. Writing $a_{0}=\mathrm{asin}(0.25)$ and $a_{1}=\mathrm{asin}(0.60)$,
the two objects of interest are
\[
\rho_{\bullet}=\int^{1}_{0}\rho(u)\mathrm{d}u=\frac{\cos a_{0}-\cos a_{1}}{a_{1}-a_{0}}=0.43049,\qquad\tilde{\rho}_{\bullet}=\sin\Big(\int^{1}_{0}\mathrm{asin}\,\rho(u)\mathrm{d}u\Big)=\sin\tfrac{a_{0}+a_{1}}{2}=0.43325,
\]
and the convexity inequality $\tilde{\rho}_{\bullet}\geq\rho_{\bullet}$
holds with a gap of $0.0028$.

The block-averaged estimator approximates $\rho_{\bullet}$ rather than
targeting it exactly. For a fixed partition of the day into blocks
$I_{b}$ with length weights $w_{b}$, the probability limit is
\begin{equation}
\sum_{b}w_{b}\sin\Big\{\tfrac{1}{|I_{b}|}\int_{I_{b}}\mathrm{asin}\,\rho(u)\mathrm{d}u\Big\},\label{eq:blocklimit}
\end{equation}
which converges to $\rho_{\bullet}$ only as the blocks shrink. When
$\mathrm{asin}\,\rho$ is linear, as it is in Design A, each block term
equals $\rho$ at the block midpoint, so (\ref{eq:blocklimit}) is the
midpoint rule for $\int^{1}_{0}\rho(u)\mathrm{d}u$. With six $60$-minute
blocks and a final $30$-minute block it evaluates to $0.43055$, above
$\rho_{\bullet}$ by $0.00006$ because $\rho$ is concave for $u\in[0,1]$.
The discrepancy is a quadrature effect, and we use (\ref{eq:blocklimit})
as the benchmark.

Figure \ref{fig:dynrhoA} reports the mean of each estimator over
$20{,}000$ replications. At the finest sampling frequency the full-day
$Q_{S}$ has mean $0.43331$ against the predicted $\tilde{\rho}_{\bullet}=0.43325$,
while the block-averaged $Q_{S}$, formed from non-overlapping $60$-minute
blocks, has mean $0.43049$ against the blockwise limit $0.43055$.
The Monte Carlo standard error is $0.00006$, so each estimator is within
about one standard error of its own limit, and their limits are separated
as predicted. The same figure
shows that $P$ converges to its own, more distant, limit, and that
$K$ tracks the corresponding benchmark: $0.39767$ against a predicted
$0.39770$ for $P$, and $0.41438$ against $0.41458$ for $K$, where
these predictions are the path-wise expressions of Theorem \ref{theo:Consistency}
extended to time-varying $\rho$.\footnote{With both volatility and correlation time-varying, the
extensions are $\int\sigma_{x}\sigma_{y}\rho\,\mathrm{d}u/\sqrt{\int\sigma^{2}_{x}\mathrm{d}u\int\sigma^{2}_{y}\mathrm{d}u}$
for $P$ and $\sin(\int\int\mathrm{asin}\,h(u,v)\mathrm{d}u\mathrm{d}v)$
for $K$, with $h(u,v)=[\sigma_{x}(u)\sigma_{y}(u)\rho(u)+\sigma_{x}(v)\sigma_{y}(v)\rho(v)]/[\sqrt{\sigma^{2}_{x}(u)+\sigma^{2}_{x}(v)}\sqrt{\sigma^{2}_{y}(u)+\sigma^{2}_{y}(v)}]$.
Both reduce to the expressions in Theorem \ref{theo:Consistency}
when $\rho$ is constant. The expression for $P$ is the probability
limit under the Heston design. The expression for $K$ is a benchmark
rather than a proved limit: Theorem \ref{theo:Consistency} establishes
it only when $\sigma$ is independent of $W$, whereas the Heston design
has leverage. The simulated mean of $K$ shows close numerical agreement with this
benchmark, within $0.0002$. We evaluate both expressions along each simulated
path and average.}

The block-averaged $Q_{S}$ has
the smallest root mean squared error for $\rho_{\bullet}$ at frequencies
up to about $30$ seconds, after which the full-day $Q_{S}$ overtakes
it: localization removes the convexity bias but does so at the cost of
using fewer observations within each block, and at coarse frequencies the blocks become too short. Moreover,
the gap between the two $Q_{S}$ limits is an order of magnitude smaller
than the gap between either of them and the limit of $P$ or the benchmark
for $K$.
Under time-varying correlation the choice between the full-day and
the localized $Q_{S}$ matters far less than the choice between $Q_{S}$
and the conventional estimators.

We also report the accuracy of the rolling-window $Q_{S}$ as an estimator
of the spot correlation $\rho(u)$ at three times of day, in Figure
\ref{fig:dynrhoAspot} and Table \ref{tab:dynrhoAspot}. The window
is centered on the evaluation time; a trailing window ending at $t$
would estimate $\rho$ at $t-L/2$ rather than at $t$, see Section
\ref{subsec:WindowLength}. At one-second sampling the centered rolling $Q_{S}$ has a bias
of less than $0.0003$ in absolute value at all three times, whereas
$K$ and $P$ carry biases of $-0.005$ to $-0.008$ and $-0.009$
to $-0.015$ respectively, growing with the level of $\rho$. At the
three-minute frequency used in the empirical analysis the rolling
$Q_{S}$ also attains the smallest root mean squared error of the
three. At the finest frequencies $K$ and $P$ are more precise than
$Q_{S}$ despite being biased, cf. Remark \ref{rem:efficiency}; recall
that this design contains no market microstructure noise.

\begin{figure}[H]
\begin{centering}
\includegraphics[width=\textwidth]{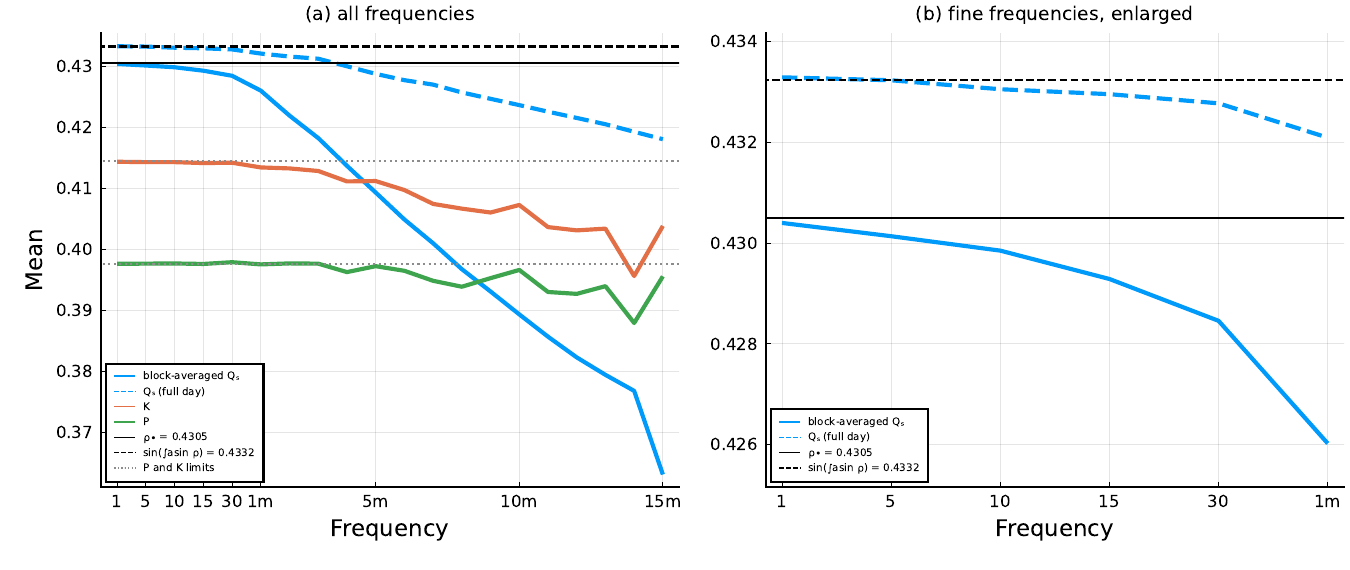}
\par\end{centering}
\caption{Mean correlation estimates under time-varying correlation (Design
A), as a function of sampling frequency, over $20{,}000$ replications.
Panel (b) enlarges the fine-frequency region, from one second to one
minute. The solid horizontal line is the integrated correlation
$\rho_{\bullet}=0.4305$, which the block-averaged $Q_{S}$ approximates
(its exact probability limit (\ref{eq:blocklimit}) is $0.4306$ for
this block partition); the dashed horizontal line is $\tilde{\rho}_{\bullet}=\sin(\int\mathrm{asin}\,\rho)=0.4332$,
the target of the full-day $Q_{S}$; the dotted lines in panel (a)
are the probability limit of $P$ and the benchmark for $K$ described
in the text.\label{fig:dynrhoA}}
\end{figure}

\begin{figure}[H]
\begin{centering}
\includegraphics[width=\textwidth]{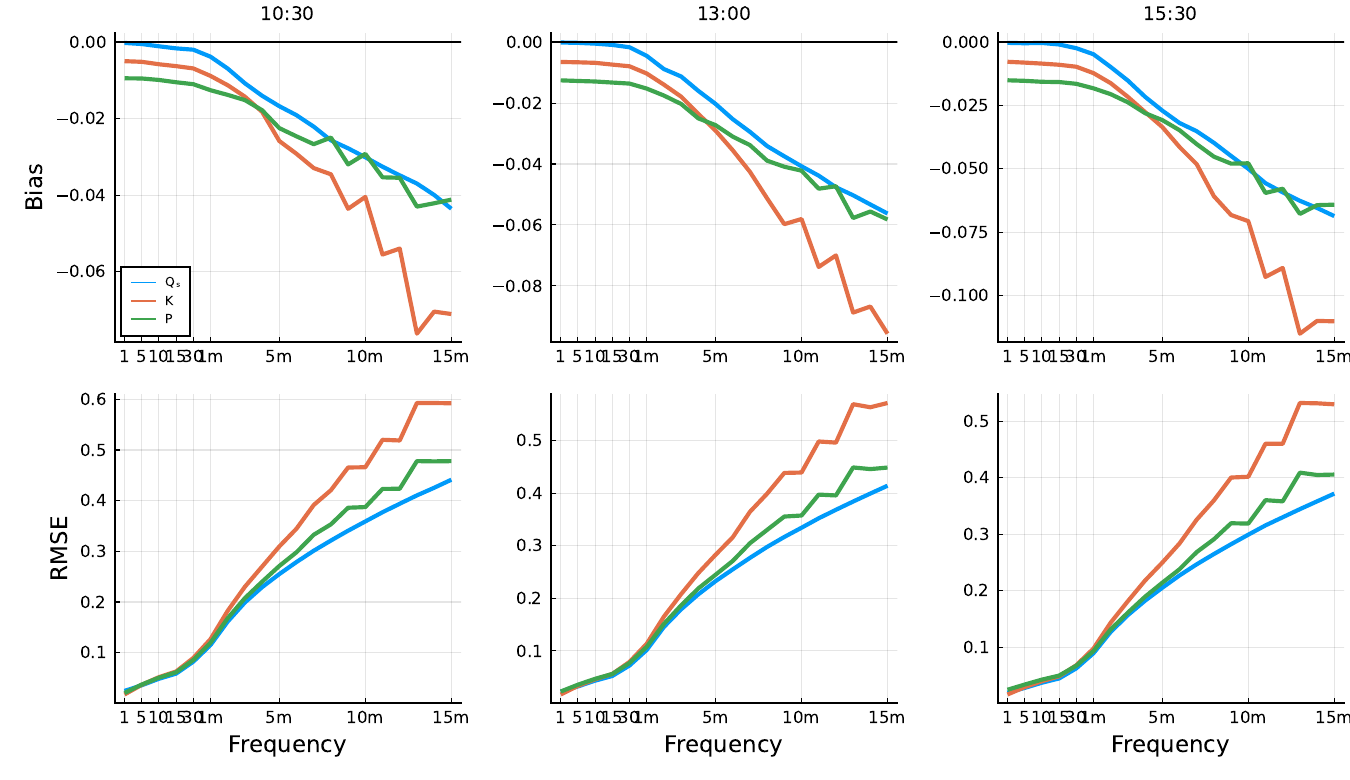}
\par\end{centering}
\caption{Bias (upper panels) and root mean squared error (lower panels) of
rolling-window estimates of the spot correlation $\rho(u)$ at 10:30,
13:00 and 15:30, as a function of sampling frequency. The window is
$60$ minutes and is centered on the evaluation time.\label{fig:dynrhoAspot}}
\end{figure}

\subsubsection{Choice of the Window Length\label{subsec:WindowLength}}

The localized estimator requires a choice of window length. We compute
the centered rolling
$Q_{S}$ with windows of $30$, $60$, $90$, and $120$ minutes, and
report the bias and root mean squared error for the spot correlation
in Figure \ref{fig:spotwindows}. In addition to Design A we use a
second design, Design B, in which $\rho(u)=\sin\{\tfrac{\pi}{2}\tanh(Z_{u})\}$
for an Ornstein-Uhlenbeck process $Z$ that is independent of the price
and volatility innovations, calibrated so that the stationary mean
of $\rho(u)$ is about $0.64$ with a stationary standard deviation of
about $0.14$.

Under Design A the centered window has no smoothing bias for any
window length. A window of length $L$ centered at $u$ has probability
limit
\[
\sin\Big\{\frac{1}{L}\int^{u+L/2}_{u-L/2}\mathrm{asin}\,\rho(v)\,\mathrm{d}v\Big\},
\]
and $\mathrm{asin}\,\rho(v)$ is linear in $v$ by construction, so the
average over the symmetric interval equals its value at the midpoint,
and the limit is exactly $\rho(u)$ regardless of $L$. Design A therefore
isolates the sampling-error side of the trade-off, and its bias at
one-second sampling is below $0.0005$ in absolute value at every window
length. So, we consider a second design, Design B, in which the trade-off
between smoothing bias and sampling error becomes visible. This design
has curvature in $\mathrm{asin}\,\rho(u)$, which is required for
a smoothing bias to arise; there the bias at one-second sampling grows
from $0.0016$ at $30$ minutes to $0.0052$ at $120$ minutes.

The direction of the trade-off depends
on the sampling frequency, because sampling error and smoothing error
scale differently with the sampling frequency.
At one-second sampling the sampling error is negligible and the smoothing
error dominates, so shorter windows are preferable: the root mean squared
error under Design B rises from $0.052$ at $30$ minutes to $0.087$
at $120$. Under Design B the target $\rho(u)$ is a random path, so
the smoothing error is itself random: the window average of $\mathrm{asin}\,\rho$
differs from its value at the midpoint by an amount that depends on
the realized path and has mean close to zero. The mean bias is only
$0.0052$ at $120$ minutes, against a root mean squared error of $0.087$
(Table \ref{tab:spotwindows}). At three-minute sampling the ranking reverses, since a
$30$-minute window then contains few sparsely sampled returns, and
the error falls from $0.206$ to $0.128$ as the window lengthens.
Between these lies the intermediate case: at $30$-second sampling
under Design B the root mean squared error is $0.091$, $0.084$, $0.089$,
and $0.094$ at the four window lengths, minimized at $60$ minutes,
where the bias is also within one Monte Carlo standard error of zero,
having changed sign between the $30$- and $60$-minute windows. No
single window length is optimal at all frequencies. Sixty minutes
minimizes the error at $30$-second sampling. At the three-minute frequency
of the empirical analysis the root mean squared error falls as the
window lengthens from $30$ to $120$ minutes, by $51\%$ under Design
A and by $38\%$ under Design B (Table \ref{tab:spotwindows}), so
the $60$-minute window trades some precision for temporal resolution:
a $120$-minute window averages over nearly a third of the trading
day, which is too coarse for the intraday patterns we study. In the
Supplementary Material we recompute the empirical intraday estimates
with $30$- and $90$-minute windows and find the patterns qualitatively
unchanged.

The placement of the window matters more than its length. A trailing
window ending at $t$ averages the correlation over $[t-L,t]$ and therefore
estimates $\rho$ at the midpoint $t-L/2$, not at $t$. Plotting such
estimates against $t$ shifts the entire intraday profile backward by
half a window. In Design A at one-second sampling this produces apparent
biases of $-0.014$, $-0.027$, $-0.040$, and $-0.053$ at the four window
lengths, proportional to $L$, compared with $0.0005$ when the window
is centered. We therefore center all rolling windows on the evaluation
time. The shift is easily mistaken for an intraday pattern, and it
grows with the window length.

\begin{figure}[H]
\begin{centering}
\includegraphics[width=\textwidth]{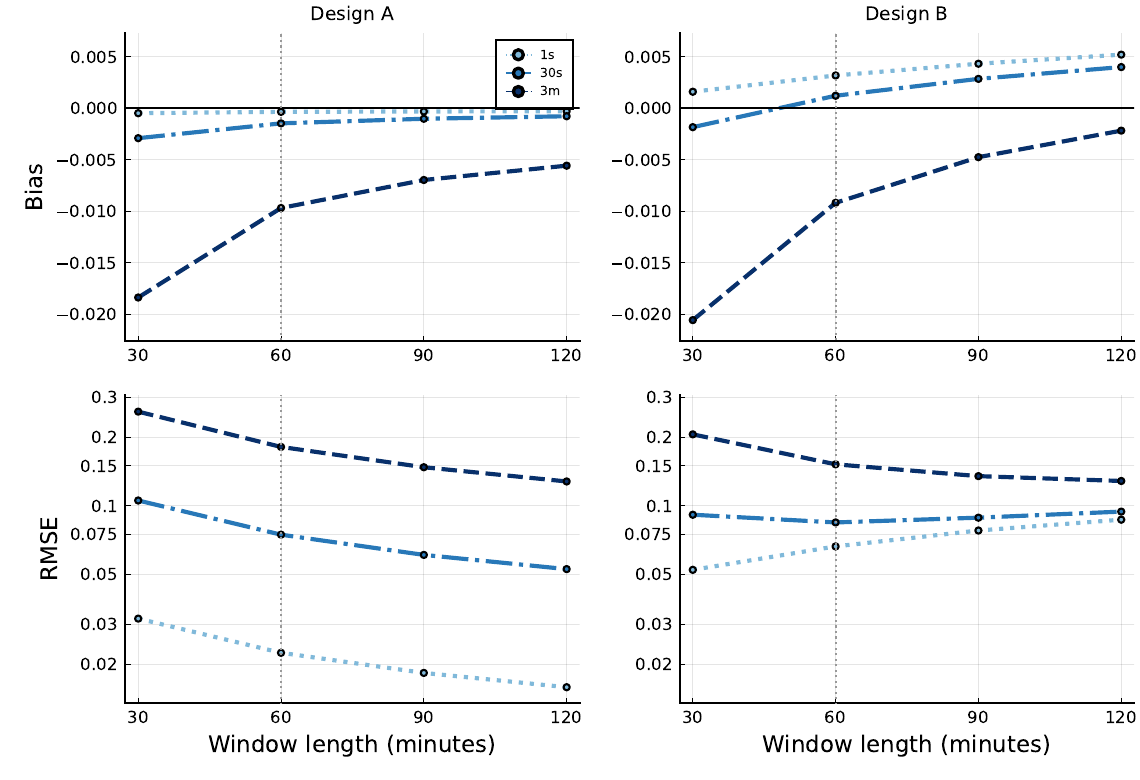}
\par\end{centering}
\caption{Bias (upper panels) and root mean squared error (lower panels, logarithmic
scale) of the centered rolling-window $Q_{S}$ as an estimator of the
spot correlation $\rho(u)$, as a function of the window length, for
three sampling frequencies. Design A has a deterministic correlation
path and Design B a stochastic one. Both are averaged over the window
midpoints for which every window length fits within the trading day.
The dotted vertical line marks the $60$-minute window used in the
empirical analysis.\label{fig:spotwindows}}
\end{figure}

\begin{table}[H]
\caption{Spot correlation estimates under time-varying correlation\label{tab:dynrhoAspot}}

\begin{centering}
\begin{tabular*}{\textwidth}{@{\extracolsep{\fill}}llrrrrrr}
\hline\hline
 & & \multicolumn{2}{c}{$Q_S$} & \multicolumn{2}{c}{$K$} & \multicolumn{2}{c}{$P$} \\
\cmidrule(lr){3-4}\cmidrule(lr){5-6}\cmidrule(lr){7-8}
Time & $m$ & Bias & RMSE & Bias & RMSE & Bias & RMSE \\
\hline
10:30 & 1s & -0.0002 &  0.0245 & -0.0050 &  0.0174 & -0.0094 &  0.0203 \\
10:30 & 30s & -0.0020 &  0.0822 & -0.0069 &  0.0895 & -0.0110 &  0.0855 \\
10:30 & 3m & -0.0108 &  0.1991 & -0.0143 &  0.2301 & -0.0151 &  0.2077 \\
10:30 & 10m & -0.0301 &  0.3590 & -0.0405 &  0.4660 & -0.0292 &  0.3873 \\
\hline
13:00 & 1s & -0.0001 &  0.0225 & -0.0065 &  0.0171 & -0.0125 &  0.0228 \\
13:00 & 30s & -0.0017 &  0.0723 & -0.0079 &  0.0791 & -0.0136 &  0.0770 \\
13:00 & 3m & -0.0113 &  0.1791 & -0.0179 &  0.2080 & -0.0203 &  0.1865 \\
13:00 & 10m & -0.0408 &  0.3342 & -0.0581 &  0.4392 & -0.0422 &  0.3575 \\
\hline
15:30 & 1s & -0.0003 &  0.0196 & -0.0078 &  0.0164 & -0.0151 &  0.0249 \\
15:30 & 30s & -0.0024 &  0.0627 & -0.0097 &  0.0682 & -0.0164 &  0.0674 \\
15:30 & 3m & -0.0152 &  0.1569 & -0.0216 &  0.1816 & -0.0238 &  0.1613 \\
15:30 & 10m & -0.0500 &  0.2992 & -0.0707 &  0.4017 & -0.0478 &  0.3189 \\
\hline\hline
\end{tabular*}
\par\end{centering}
Note: Bias and root mean squared error of rolling-window estimates
of the spot correlation $\rho(u)$ in Design A, over $20{,}000$ replications.
The window is $60$ minutes, centered on the evaluation time, and
$m$ is the sampling frequency of the underlying returns. The targets
are $\rho(u)=0.3077$, $0.4467$, and $0.5757$ at 10:30, 13:00, and
15:30, respectively.
\end{table}
Table \ref{tab:dynrhoAspot} shows two things. First, the
root mean squared error of every estimator rises steeply with the sampling
interval, from about $0.02$ at one second to about $0.35$ at ten minutes.
A $60$-minute window contains only six $10$-minute returns, too few
for any correlation estimator to be precise. Spot estimation at
coarse frequencies therefore requires a longer window, see Figure
\ref{fig:spotwindows}. Second, $Q_{S}$ deteriorates
least as the sampling interval grows. From thirty-second sampling onward it has the smallest
root mean squared error in every row, by $3$--$7\%$ relative to $P$ and by
$8$--$26\%$ relative to $K$, and up to three-minute sampling it also has
the smallest bias; at ten minutes its bias and that of $P$ are within
$0.002$ of each other. Only at one-second sampling does $K$ have the lower
root mean squared error, by $16$--$29\%$, which is the efficiency cost of
the sign statistic discussed in Section \ref{subsec:ideal}, and there
$K$ carries a bias more than twenty times that of $Q_{S}$. The
one-second results are for a design without market microstructure
noise. With actual data, the unadjusted estimators considered here are
substantially distorted by market microstructure noise at one-second
sampling, see Section \ref{subsec:microstructure}.
At the three-minute frequency used in the empirical analysis the root
mean squared error is below $0.2$, and the sampling is coarse enough
to mitigate the microstructure effects.

\subsection{Estimating Integrated Covariance\label{subsec:intcov}}

In the absence of noise and with synchronous observations, the realized
covariance is the efficient estimator of the integrated covariance,
and consistent estimators are available for settings with noise, jumps,
non-synchronous, and endogenous observation times, see \citet{BNHLS-MRK:2011},
\citet{Koike:2016}, and \citet{Koike:2017}, among others. This section
cautions against a seemingly natural use of robust correlation estimators
for this purpose. Combining robust correlation estimators with volatility
estimators is not advisable for estimating integrated covariance,
$\mathrm{IC}=\int\sigma_{xy}(u)\mathrm{d}u$, which simplifies to
$\rho\int\sigma_{x}(u)\sigma_{y}(u)\mathrm{d}u$ when $\rho(u)=\rho$
for all $u$. Now, if we multiply $Q_{S}$ by consistent estimates
of $\sqrt{\int\sigma^{2}_{x}(u)\mathrm{d}u}$ and $\sqrt{\int\sigma^{2}_{y}(u)\mathrm{d}u}$,
we will be estimating $\frac{1}{\lambda}\mathrm{IC}$, instead of
$\mathrm{IC}$. Using $K$ is not advisable either, because it leads
to another incorrect limit. For this problem, localized estimators
of spot volatility and spot correlation can be used, as proposed in
\citet{BoudtCornelissenCroux:2012jump}. In order to be robust to
jumps, they combine the MedRV estimator by \citet{AndersenDobrevSchaumburg2012}
and the Gaussian rank correlation. This is further explored in \citet{VanderVeredas:2017}
who employ additional robust correlation estimators as components
in estimating $\mathrm{IC}$. They also combine non-localized estimates
of $\sqrt{\int\sigma^{2}_{x}(u)\mathrm{d}u}$ and $\sqrt{\int\sigma^{2}_{y}(u)\mathrm{d}u}$
with a range of correlation estimators. Some of these combinations
will be inconsistent for the reason stated earlier. This may explain
why \citet{VanderVeredas:2017} find the bivariate realized kernel
estimator by \citet{BNHLS-MRK:2011} to be the most accurate estimator
of $\mathrm{IC}$ in the absence of jumps.

\subsection{Influence Function\label{subsec:influence}}

The \emph{influence function} measures an estimator's
sensitivity to data contamination. For a statistical functional $R$
and baseline distribution $F$, it is defined by
\[
\mathrm{IF}((x_{0},y_{0}),R,F)=\lim_{\eta\searrow0}\frac{R((1-\eta)F+\eta\mathbf{\Delta}_{(x_{0},y_{0})})-R(F)}{\eta},
\]
where $\mathbf{\Delta}_{(x_{0},y_{0})}$ is the Dirac measure at $(x_{0},y_{0})$.
We have $R_{P}(F)=R_{Q}(F)=R_{K}(F)=\rho$ for $F=\Phi_{\rho}$, which
denotes the standard bivariate normal distribution with correlation
equal to $\rho$. From \citet{DevlinGnanadesikanKettenring:1975}
and \citet{CrouxDehon:2010} we have their influence functions.
\begin{prop}
The influence functions of correlation estimators at $\Phi_{\rho}$
are given by
\[
\begin{aligned}\mathrm{IF}((x_{0},y_{0}),R_{Q},\Phi_{\rho}) & =\tfrac{\pi}{2}\sqrt{1-\rho^{2}}(\mathrm{sgn}(x_{0}y_{0})-\tau)\\
\mathrm{IF}((x_{0},y_{0}),R_{K},\Phi_{\rho}) & =2\pi\sqrt{1-\rho^{2}}(2\Phi(x_{0},y_{0})+1-\Phi(x_{0})-\Phi(y_{0})-\tfrac{1}{2}(\tau+1))\\
\mathrm{IF}((x_{0},y_{0}),R_{P},\Phi_{\rho}) & =x_{0}y_{0}-(\frac{x^{2}_{0}+y^{2}_{0}}{2})\rho,
\end{aligned}
\]
where $\Phi(\bullet,\bullet)$ denotes the joint CDF for $\Phi_{\rho}$
and $\Phi(\bullet)$ denotes the marginal CDF for a standard normal
distribution. 
\end{prop}
$Q$ and $K$ have bounded influence functions, whereas that of $P$
is unbounded; this is why we call the former robust. Because $Q$
and $K$ are computed from signs alone, the contribution of any single
observation is bounded regardless of its magnitude. The analogous results for $Q_{S}$ are qualitatively
very similar, and are presented in the Supplementary Material. One
way to alleviate the sensitivity that $P$ has to outliers is to use
truncation estimators, which are commonly used for estimating realized
variances, see, e.g., \citet{Mancini:2009}.

The influence function for the Spearman estimator is also bounded
but can be shown to have a larger bound than $Q$ and $K$, whereas
the Gaussian rank estimator has an unbounded influence function, see
\citet{Rousseeuw:1984}, \citet{BoudtCornelissenCroux:2012gauss},
and \citet{RaymaekersRousseeuw:2021} for details and additional results
on influence functions. 

\section{Simulation Study\label{sec:Simulation-Study}}

We compare the estimators in simulations that mimic our empirical
setting. We generate the two logarithmic price processes,
$X^{\ast}_{t}$ and $Y^{\ast}_{t}$, using the Heston model:
\begin{equation}
\begin{aligned}\mathrm{d}X^{\ast}_{t} & =\mu_{x}\mathrm{d}t+\sigma_{xt}\mathrm{d}W_{xt},\\
\mathrm{d}\sigma^{2}_{x,t} & =\kappa_{x}(\bar{\sigma}^{2}_{x}-\sigma^{2}_{x,t})\mathrm{d}t+s_{x}\sigma_{x,t}dB_{x,t},
\end{aligned}
\label{eq:heston}
\end{equation}
where $W_{x,t}$ and $B_{x,t}$ are standard Brownian motions with
$\mathrm{cov}(\mathrm{d}W_{x,t},\mathrm{d}B_{x,t})=\varrho_{x}\mathrm{d}t$,
and $Y^{\ast}_{t}$ is generated similarly with $\mathrm{cov}(\mathrm{d}W_{x,t},\mathrm{d}W_{y,t})=\rho\mathrm{d}t$.
The four Brownian motions are constructed as follows: $(W_{x},W_{y})$
is a Brownian motion with correlation $\rho$, and $B_{x}=\varrho_{x}W_{x}+\sqrt{1-\varrho^{2}_{x}}U_{x}$
and $B_{y}=\varrho_{y}W_{y}+\sqrt{1-\varrho^{2}_{y}}U_{y}$, where
$U_{x}$ and $U_{y}$ are standard Brownian motions independent of each
other and of $(W_{x},W_{y})$. This implies $\mathrm{cov}(\mathrm{d}W_{x,t},\mathrm{d}B_{y,t})=\varrho_{y}\rho\mathrm{d}t$,
$\mathrm{cov}(\mathrm{d}W_{y,t},\mathrm{d}B_{x,t})=\varrho_{x}\rho\mathrm{d}t$,
and $\mathrm{cov}(\mathrm{d}B_{x,t},\mathrm{d}B_{y,t})=\varrho_{x}\varrho_{y}\rho\mathrm{d}t$,
so that the covariance matrix of $(\mathrm{d}W_{x},\mathrm{d}W_{y},\mathrm{d}B_{x},\mathrm{d}B_{y})$
is positive semidefinite for every $\rho\in[-1,1]$ and every
$\varrho_{x},\varrho_{y}\in[-1,1]$; in particular the values $\rho=2/3$
and $\rho=3/4$ used below are admissible with the leverage parameters
of Table \ref{tab:Parameters-calibration}. The model is calibrated using the simulation design in Table \ref{tab:Parameters-calibration},
which was previously used in \citet{AitSahaliaFanXiu:2010}. The initial
values for volatility $\sigma^{2}_{x,0}$ and $\sigma^{2}_{y,0}$
are drawn from Gamma distributions, $\Gamma(2\kappa_{x}\bar{\sigma}^{2}_{x}/s^{2}_{x},s^{2}_{x}/2\kappa_{x})$
and $\Gamma(2\kappa_{y}\bar{\sigma}^{2}_{y}/s^{2}_{y},s^{2}_{y}/2\kappa_{y})$,
and the price processes are initialized with $X^{\ast}_{0}=\log(100)$
and $Y^{\ast}_{0}=\log(40)$. The simulated model is a discretized
version with $23,400$ increments, which translates to one-second observations
over a 6.5-hour period, the length of a typical trading day.
\begin{table}[H]
\centering{}\caption{\label{tab:Parameters-calibration}Parameters calibration for the
Heston model}
\begin{tabular*}{\textwidth}{@{\extracolsep{\fill}}cccccc}
\hline 
 & $\mu$ & $\bar{\sigma}^{2}$ & $\kappa$ & $s$ & $\varrho$\tabularnewline
\hline 
$X$ & 0.05 & 0.16 & 3 & 0.8 & -0.60\tabularnewline
$Y$ & 0.03 & 0.09 & 2 & 0.5 & -0.75\tabularnewline
\hline 
\end{tabular*}
\end{table}

We present results for two values of the true correlation, $\rho=1/4$
and $\rho=2/3$, which are typical levels of the correlation in our
empirical analysis. In the Supplementary Material we present the corresponding
results for $\rho=1/2$ and $\rho=3/4$.

\subsection{Case without Noise}

We first consider the case where prices are observed without measurement
error. Without noise there is no need to sample sparsely, but the
results at lower frequencies are informative about the estimators.

The Heston model generates price processes with time-varying volatilities,
so $P$ is inconsistent by Theorem \ref{theo:Consistency}. The theorem
covers the limit of $K$ only without leverage, but the simulations
show that $K$ is biased in this design as well.
While $Q$ and $Q_{S}$ are consistent, they may have a bias in finite
samples, because sampling error in $\hat{\tau}$ and the non-linear
transformation, $\rho=\sin(\frac{\pi}{2}\tau)$, will induce a finite-sample
bias in $Q$ and $Q_{S}$.
\begin{figure}[tbh]
\begin{centering}
\includegraphics[width=0.45\textwidth]{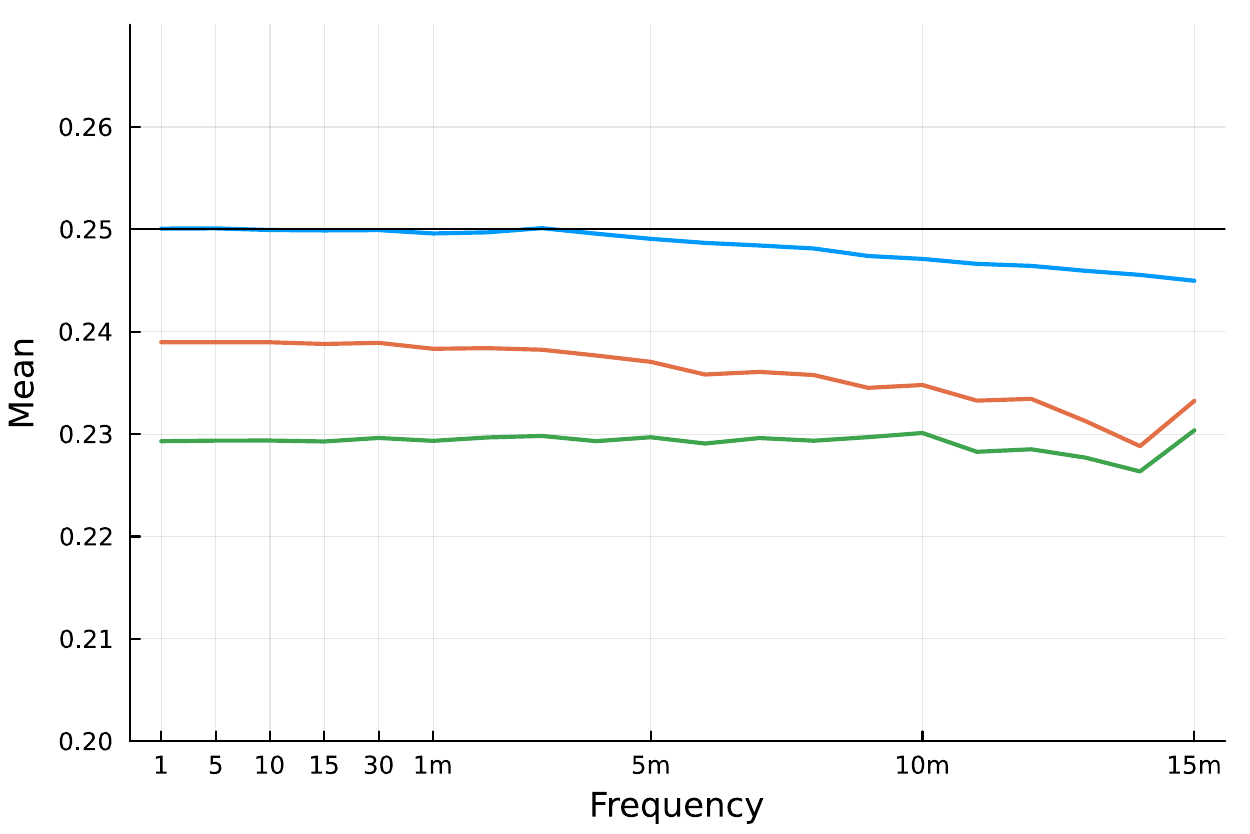}\includegraphics[width=0.45\textwidth]{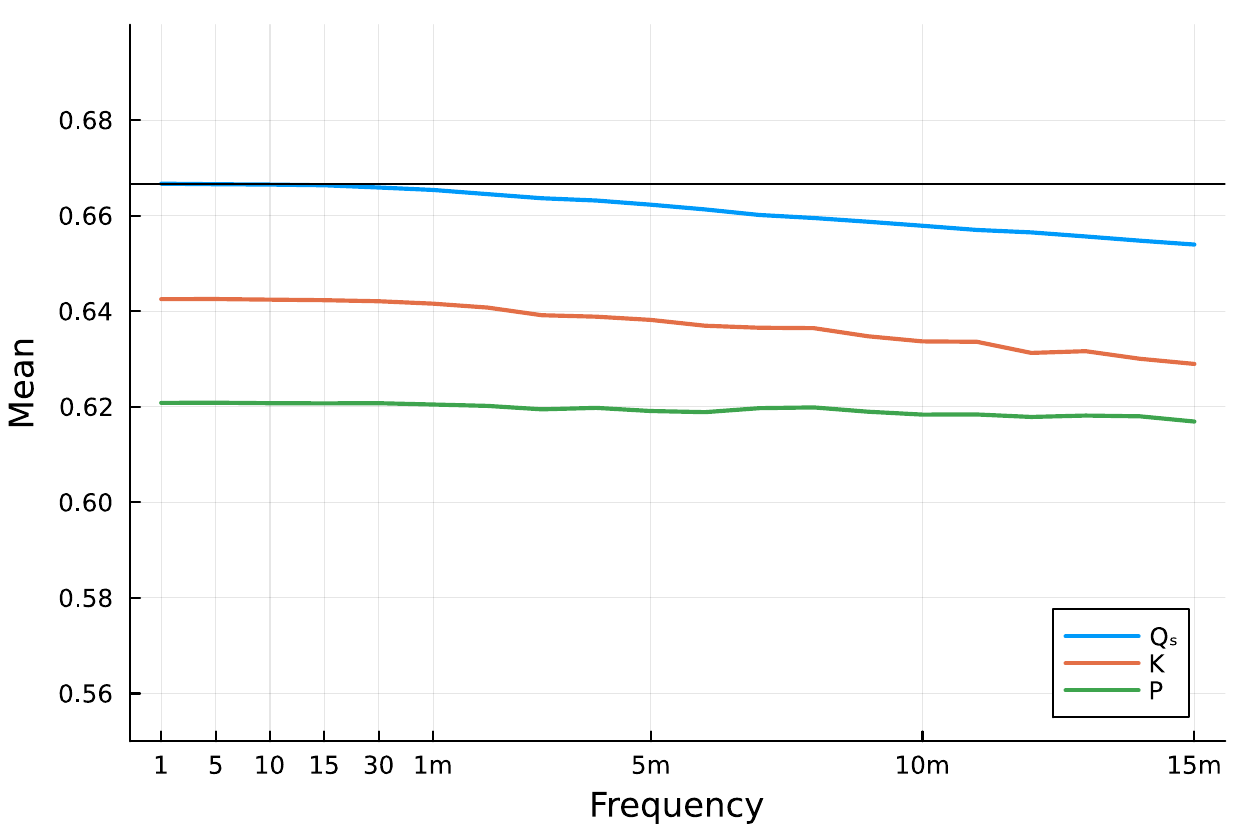}
\par\end{centering}
\centering{}\includegraphics[width=0.45\textwidth]{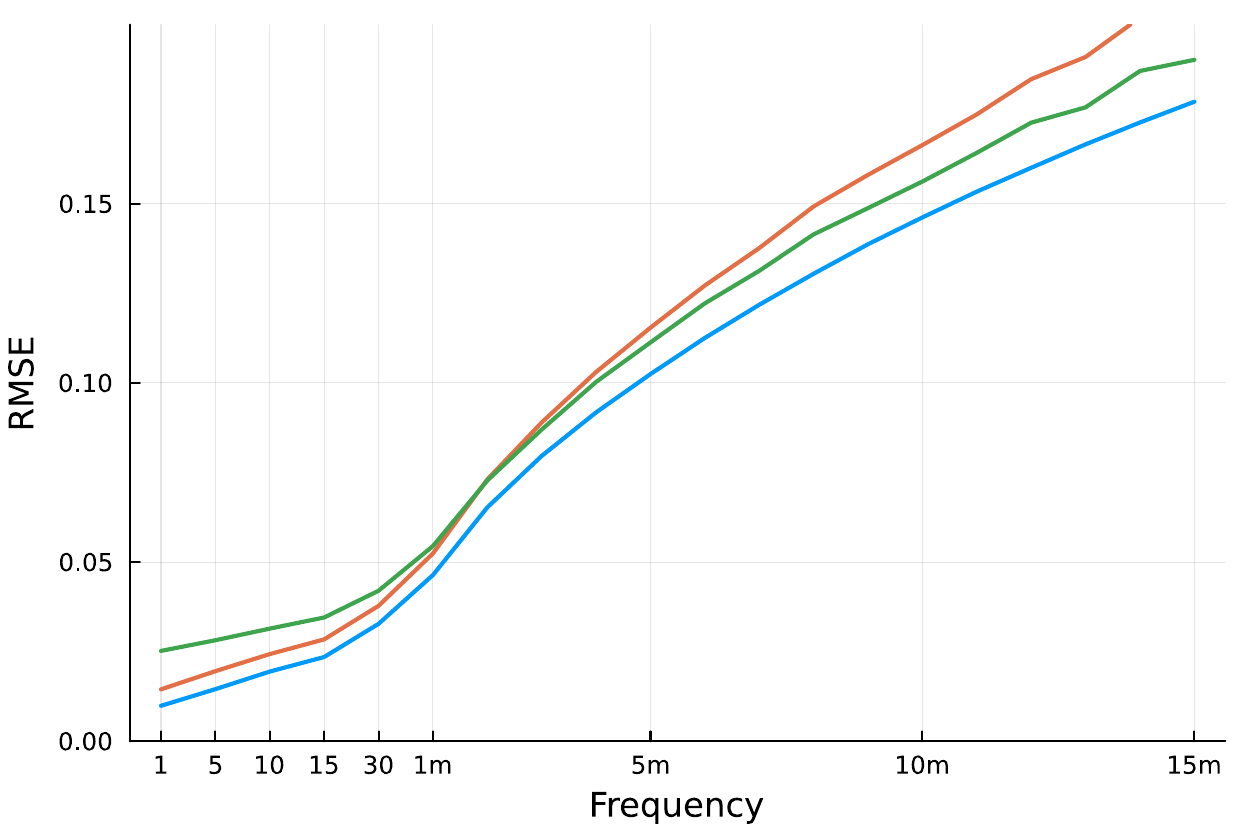}\includegraphics[width=0.45\textwidth]{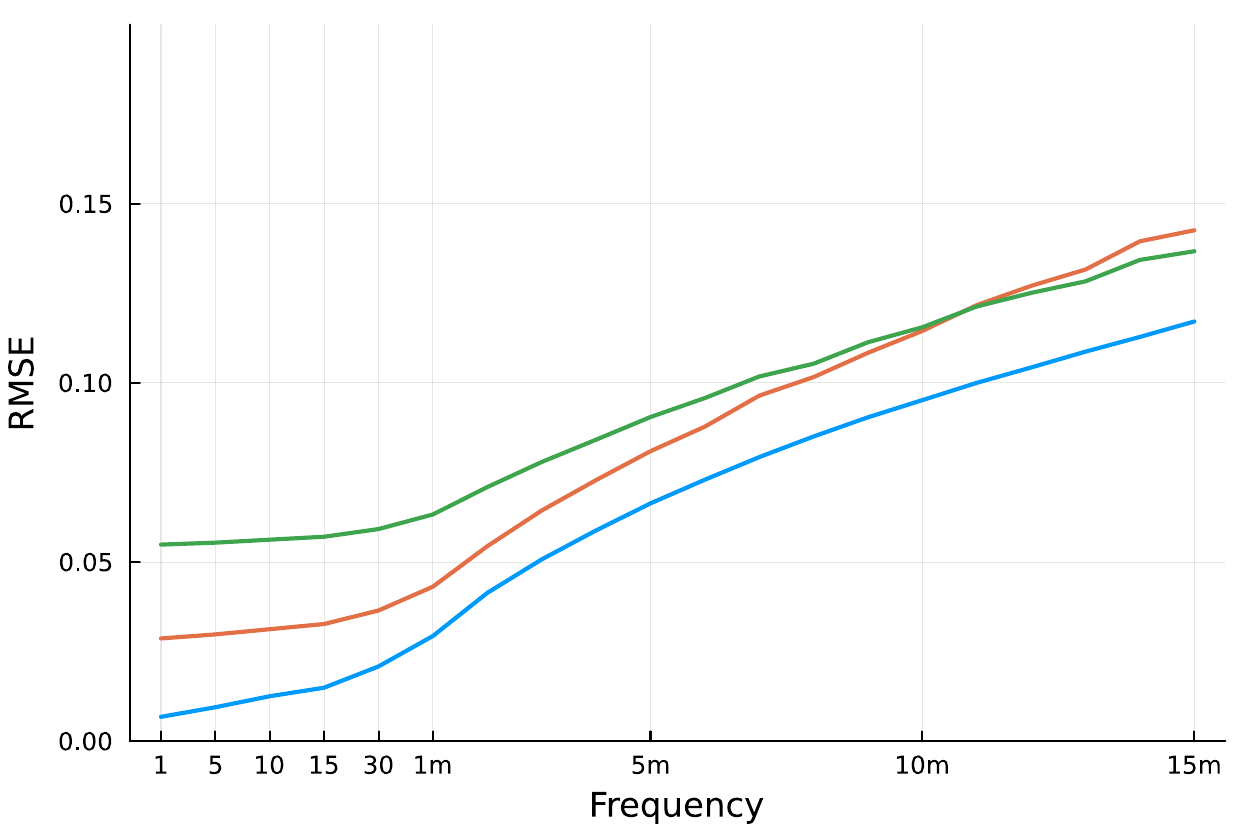}\caption{Means and RMSEs for the estimators as a function of sampling frequency.
Prices are generated by the Heston model with the true correlation
being $\rho=0.25$ (left panels) and $\rho=0.6\overline{6}$ (right
panels).\label{fig:Plot-the-averages1}}
\end{figure}

The average values of the estimators are shown in the upper panels
of Figure \ref{fig:Plot-the-averages1} for the case where $\rho=0.25$
and $\rho=0.6\overline{6}$. As expected, $P$ and $K$ are biased,
since volatility is time-varying in the Heston model,
with $P$ being substantially more biased than $K$. At the highest
sampling frequency, $P$ becomes a very accurate estimate of its probability
limit in Theorem \ref{theo:Consistency}, and $K$ of the corresponding
benchmark, both of which differ from $\rho$.
The $Q_{S}$ estimator is largely unbiased when returns are sampled
more frequently than every minute. At slower sampling frequencies
a bias begins to emerge in $Q_{S}$, which is a consequence of Jensen's
inequality. The variance of $\hat{\tau}$ increases with the sampling
frequency, and the concavity of $\tau\mapsto\sin(\frac{\pi}{2}\tau)$
for $\tau>0$ explains the downward bias that becomes evident at lower
sampling frequencies. However the bias of $Q_{S}$ is substantially
smaller than those of $K$ and $P$.

The corresponding root mean squared errors (RMSEs) are shown in the
two lower panels. The new estimator has the smallest RMSE, because its bias is smaller.

\subsection{Microstructure Issues\label{subsec:microstructure}}

Next, we add different forms of market microstructure noise. Noise
that alters the sign of a small fraction of returns, such as rare
outliers, has little effect on the robust estimators but can have
a large effect on $P$, as the influence functions show.

\subsubsection{Independent Noise}

Independent noise in the price processes can induce the Epps effect.
The independent noise reduces the correlation in returns and this
downward bias is increasing in the sampling frequency. We simulate
independent noise as follows:
\[
X_{t}=X^{\ast}_{t}+\epsilon_{xt},
\]
where $\epsilon_{xt}\sim iidN(0,\omega^{2}_{x})$ and analogously for
$Y_{t}$, with $\epsilon_{xt}$ independent of $\epsilon_{yt}$. Following
similar simulation designs in this literature, see, e.g., \citet{BandiRussell:2006}
and \citet{BNHLS-RK:2008}, we set $\omega^{2}_{x}=\xi^{2}\sqrt{T^{-1}\sum^{T}_{i=1}\sigma^{4}_{x,i/T}}$
with $\xi^{2}=0.001$, so that the variance of the noise is proportional
to square root of the integrated quarticity.

\subsubsection{Prices with tick-size increments}

In practice, high-frequency financial prices are restricted to a grid
defined by their tick-size. This induces a particular type of market
microstructure noise, as analyzed in \citet{DelattreJacod:1997},
\citet{Horel2007}, \citet{Rosenbaum:2009}, \citet{ManciniGobbi:2012},
\citet{Hansen:MC2015}, \citet{LiMykland:2015}, \citet{HansenHorelLundeArchakov},
and \citet{LiZhangLi:2018}. We will study this phenomenon by letting
observed prices be given by 
\[
X_{t}=\alpha\lfloor X^{\ast}_{t}/\alpha\rfloor\qquad\text{and}\qquad Y_{t}=\alpha\lfloor Y^{\ast}_{t}/\alpha\rfloor,
\]
where $\alpha$ defines the coarseness of the grid.\footnote{In reality it is nominal prices, $\exp X_{t}$ and $\exp Y_{t}$ that
are confined to a grid, but it makes no practical difference over
trading day.} In our simulations we let the coarseness be proportional to the level
of volatility, $\alpha=c\sigma$ in order to control the average number
of price changes within a given period of time. The true price processes
are, as before, defined by (\ref{eq:heston}).

\subsubsection{Tick size with Noise}

Next we add noise to the grid of observed prices. Specifically we
now observe 
\[
X_{t}=\alpha\lfloor X^{\ast}_{t}/\alpha\rfloor+\epsilon_{xt},\qquad\text{with}\quad\epsilon_{xt}=\begin{cases}
0 & \text{with probability }1-p,\\
\pm\alpha & \text{with probability }p/2,
\end{cases}
\]
and similarly for $Y_{t}$ with $\epsilon_{xt}$ and $\epsilon_{yt}$
independent. 

\subsubsection{Stale Prices}

We introduce stale pricing using

\begin{equation}
X_{\frac{j}{N}}=\begin{cases}
\alpha\lfloor X^{\ast}_{\frac{j}{N}}/\alpha\rfloor & \text{with probability }1-q,\\
X_{\frac{j-1}{N}} & \text{with probability }q.
\end{cases}\label{eq:StalePrices}
\end{equation}
This generates ``flat pricing,'' and the expected duration between
price updates is $(1-q)^{-1}/N$.

\subsubsection{Jumps}

Jumps are prevalent in high-frequency prices, and we can introduce
jumps via
\[
X_{t}=X^{\ast}_{t}+\sum_{s\leq t}J^{x}_{s},\qquad Y_{t}=Y^{\ast}_{t}+\sum_{s\leq t}J^{y}_{s},
\]
where $J^{x}_{t}$ and $J^{y}_{t}$ denote jump processes. The impact
that jumps have on the estimators is characterized by their influence
functions. The robust estimators, $Q_{S}$ and $K$, are essentially
unaffected by jumps, whereas $P$ is highly sensitive. Independent
jumps will cause $P$ to be biased toward zero, whereas a co-jump
(a simultaneous jump in both series) will bias $P$ toward $-1$
or $1$, depending on the sign of $J^{x}_{s}J^{y}_{s}$. Co-jumps
in the same direction will cause $P$ to be biased toward one, whereas
co-jumps in the opposite direction will cause $P$ to be biased toward
$-1$. Jumps can be alleviated by truncation methods, see \citet[2009]{Mancini:2001}
and \citet{AndersenDobrevSchaumburg2012}.\nocite{Mancini:2009} Simulation
results with jumps are presented in the Supplementary Material.

\subsection{Simulation Results}

The biases induced by different types of noise on the correlation estimators
are shown in Figure \ref{fig:Correlation-signature-plots}. The true
correlation is $\rho=1/4$ in the left panels and $\rho=2/3$ in the
right panels. Results for additional levels of correlation and types
of noise are presented in the Supplementary Material. Panel (a) in
Figure \ref{fig:Correlation-signature-plots} presents the results
when the efficient prices are contaminated with independent Gaussian
noise with a variance that is about $10^{-3}$ times the square root
of integrated quarticity of the two series. Independent noise is one
of several mechanisms that can generate the Epps effect. The independent
noise reduces the correlation between returns, which induces a downward
bias that increases with sampling frequency, to an extent that all
estimators essentially become noisy estimates of zero when computed
with 1-second intraday returns. The independent-noise assumption is contradicted by empirical evidence,
see \citet{HansenLunde:JBES2006}, and by our own data.

The correlation signature plots in our empirical analysis resemble
those in panel (b) of Figure \ref{fig:Correlation-signature-plots},
where prices are rounded to a grid ($\alpha=10^{-3}\sigma$),
mimicking tick-size discreteness. Interestingly, the rounding error
causes $Q_{S}$ to be upward biased at higher sampling frequencies.
This is also true for $K$, but to a much lesser extent, whereas $P$
is largely unaffected, but maintains the downward bias caused by
time-varying volatilities. In Figure \ref{fig:Correlation-signature-plots}
(c) we consider the same level of rounding error ($\alpha=10^{-3}\sigma$)
and add additional noise by shifting the price up or down by one tick
size with probability $p/2$ in each direction, where $p=0.75$. This induces a
downward bias, which is most pronounced at high sampling frequencies.
Finally, in Figure \ref{fig:Correlation-signature-plots} (d) we add
additional staleness to prices on the grid, as defined by (\ref{eq:StalePrices}),
where one price series remains stale with probability $q_{x}=0.5$
and the other series remains stale with probability $q_{y}=0.8$.
The combined impact of rounding and staleness is a sizable downward
bias.

\subsubsection{Can Signature Plots Separate Noise from Non-Synchronicity?}

The Epps effect in panel (a) of Figure \ref{fig:Correlation-signature-plots}
is produced by cross-sectionally independent noise. In data the effect
is more commonly attributed to non-synchronous trading and previous-tick
interpolation. These are distinct mechanisms. We use four designs with a common
efficient price path to see whether the shape of a signature plot
can tell them apart: independent noise, as in panel (a); non-synchronous observation
with no noise, where the two assets are observed at Poisson times
with mean inter-arrival of $3$ and $10$ seconds and prices are carried
forward by previous-tick interpolation; the two combined; and noise
that is cross-sectionally correlated, with $\mathrm{cor}(\varepsilon_{x},\varepsilon_{y})=1/2$.
The resulting signature plots are in Figure \ref{fig:NoiseVsNonsync}.
Two features of the plots distinguish the mechanisms.

The first is the level at which the signature plot bottoms out. At
the highest frequencies the observed returns are dominated by the
noise increments, so the correlation of observed returns approaches
the correlation of the noise. With independent noise the one-second
means of $Q_{S}$ are $0.005$ and $0.014$ for $\rho=1/4$ and $\rho=2/3$;
with correlated noise they are $0.495$ and $0.504$. The correlated-noise
plot descends to $1/2$ from above when $\rho=2/3$ and ascends to $1/2$
from below when $\rho=1/4$. In these designs the high-frequency limit
of the signature plot reveals the noise correlation, and a collapse
toward zero indicates noise that is cross-sectionally independent,
rather than noise in general. This holds for additive, serially independent
noise with a constant cross-sectional correlation; with serially dependent
or endogenous noise the high-frequency limit depends on the joint sign
structure of the noise increments and not on a single correlation
parameter.

The second is the gap between $Q_{S}$ and $P$. Under pure non-synchronicity
the one-second means are $0.381$ for $Q_{S}$ and $0.052$ for $P$
when $\rho=2/3$, a gap of $0.33$; under either noise design the two
estimators agree to within $0.04$. The difference is due to zero returns.
Previous-tick interpolation leaves $97\%$ of one-second increment
products equal to zero in this design. $Q_{S}$ discards them and estimates
the concordance among the ticks at which both assets moved, whereas
$P$ includes them and is pulled toward zero. The fraction of zero
increment products is observable, and Table \ref{tab:ZeroProductsByTime}
in the Supplementary Material reports it for our assets.

In summary, a signature plot that collapses to zero for both $Q_{S}$
and $P$ indicates cross-sectionally independent noise; one in which
$Q_{S}$ stays well above $P$, together with a high incidence of zero
returns, indicates non-synchronicity; and one that flattens out at
a non-zero level indicates correlated noise, with the level estimating
the noise correlation. These are diagnostics based on four stylized
designs, not identification results; we provide no theory under noise
or asynchronous observation.

\begin{figure}[H]
\begin{centering}
\includegraphics[width=\textwidth]{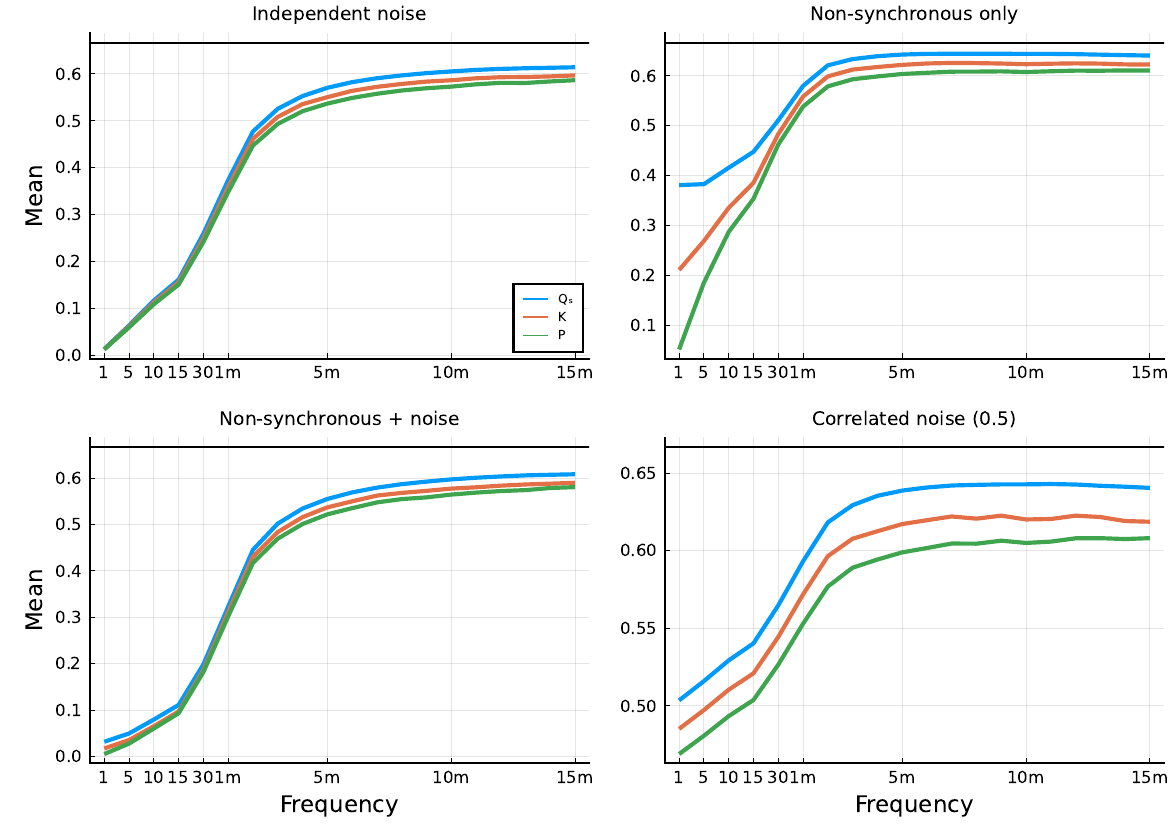}
\par\end{centering}
\caption{Correlation signature plots under four contamination schemes, with
$\rho=2/3$ and a common efficient price path within each replication.
The horizontal line is the true correlation. The corresponding figure
for $\rho=1/4$, and the incidence of zero increment products by frequency,
are in the Supplementary Material, Figures \ref{fig:NoiseVsNonsyncRho25}
and \ref{fig:NoiseVsNonsyncZeros}.\label{fig:NoiseVsNonsync}}
\end{figure}
\begin{figure}[H]
\begin{centering}
\subfloat[Heston with independent noise]{
\centering{}\includegraphics[width=0.4\textwidth]{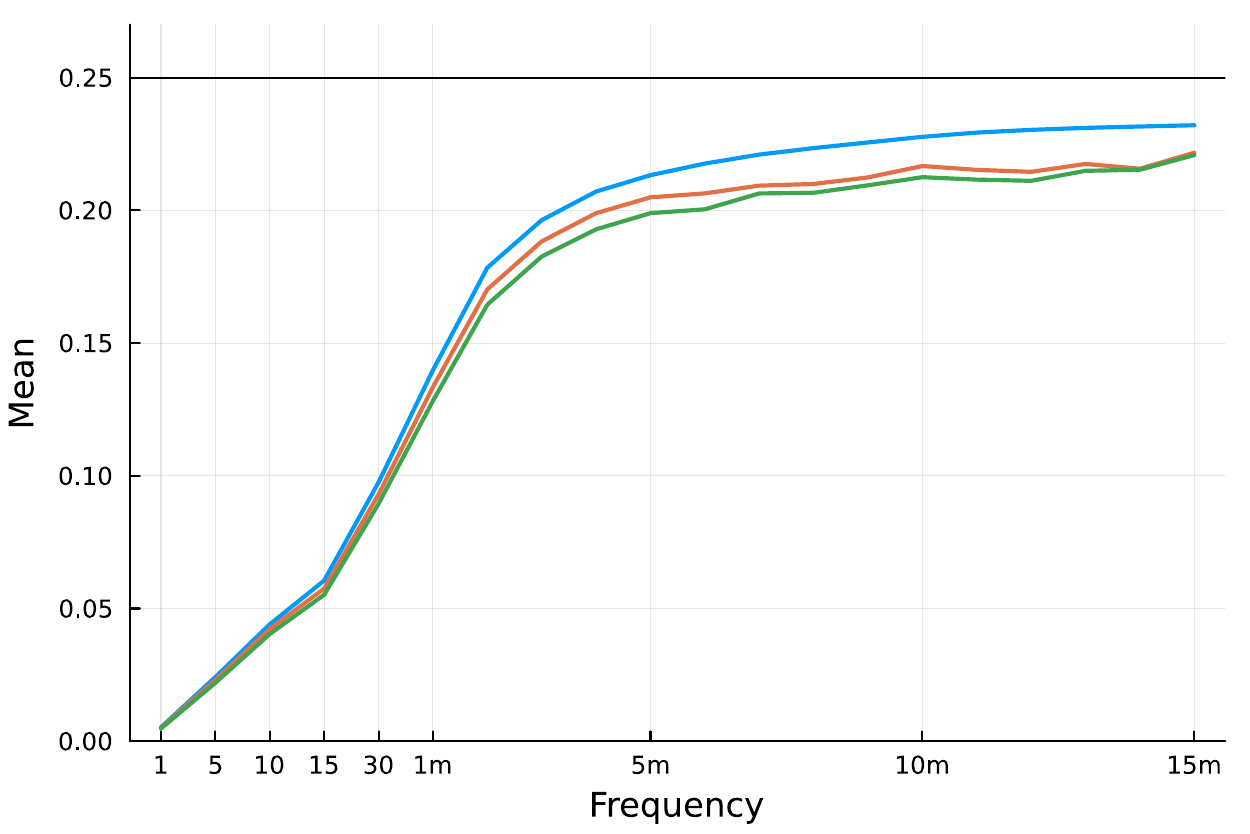}\includegraphics[width=0.4\textwidth]{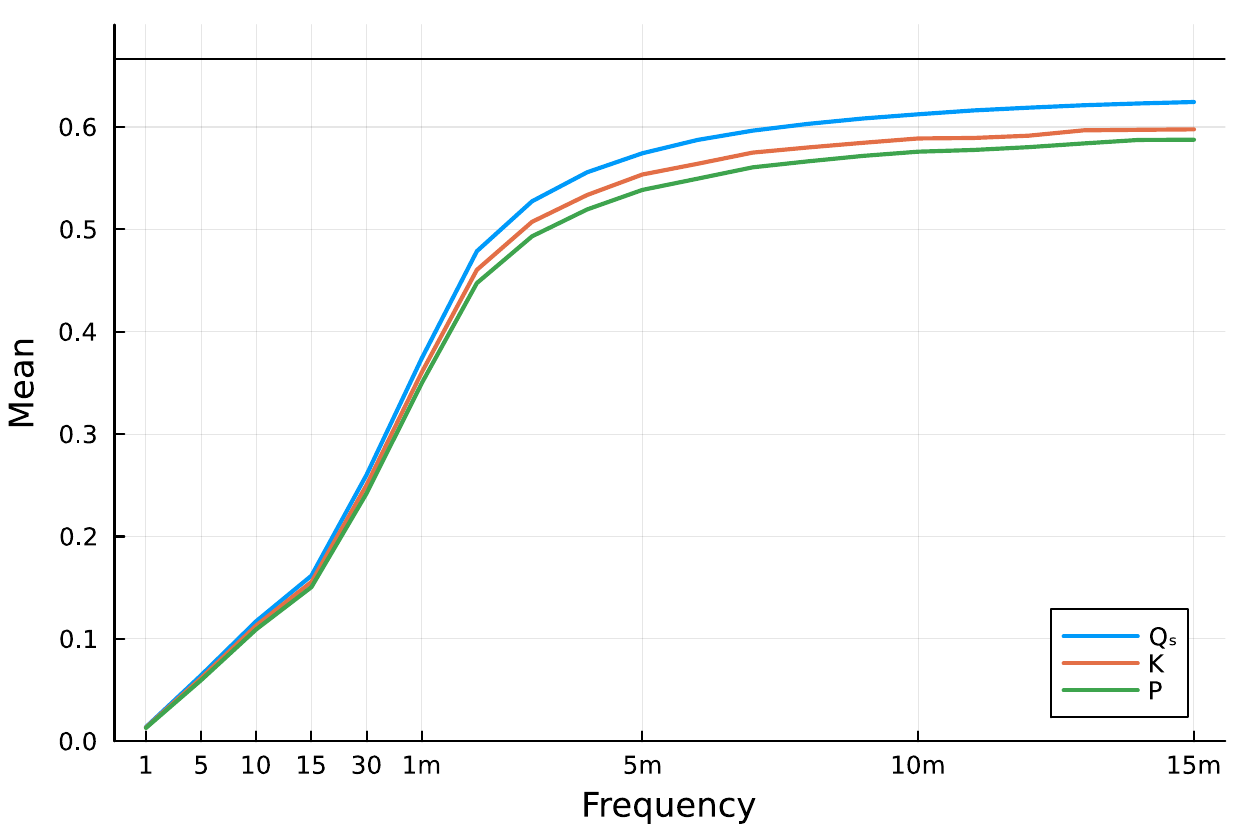}}
\par\end{centering}
\begin{centering}
\subfloat[Heston with rounding to a grid]{\begin{centering}
\includegraphics[width=0.4\textwidth]{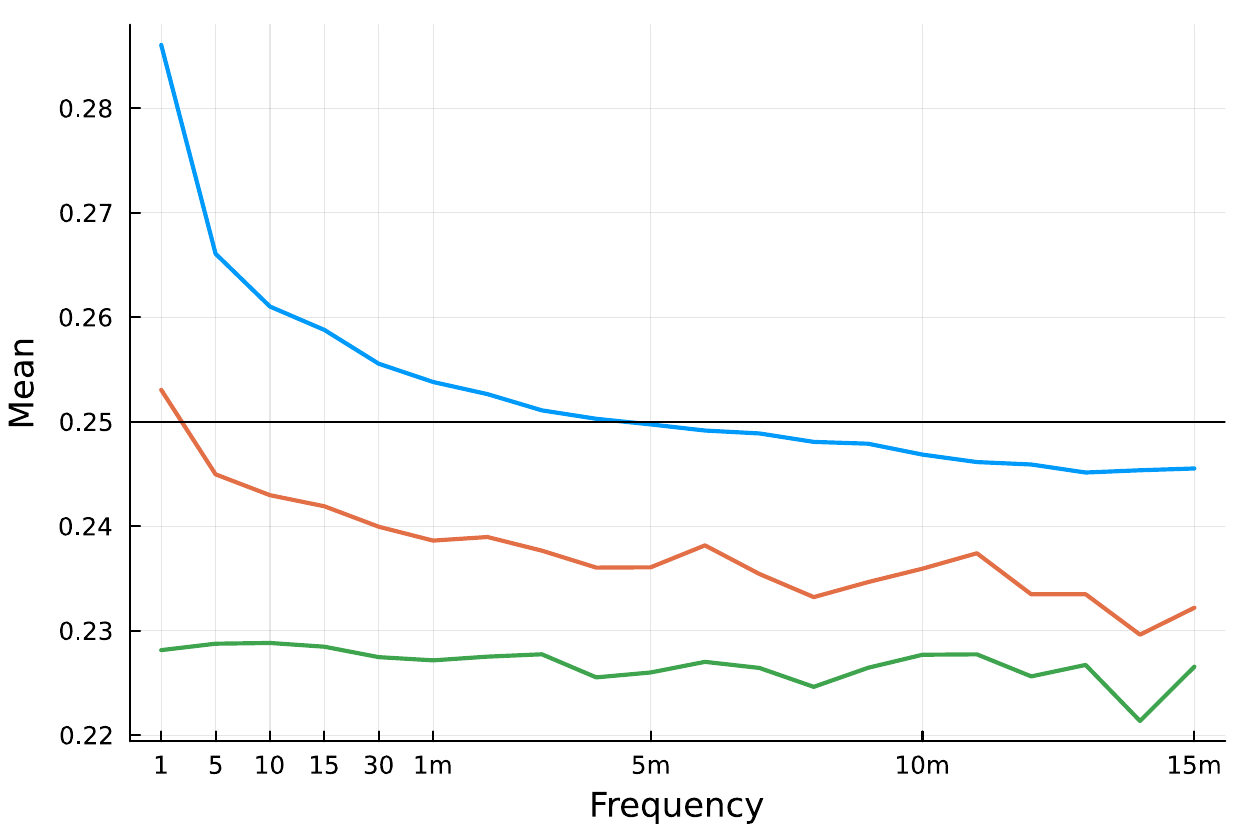}\includegraphics[width=0.4\textwidth]{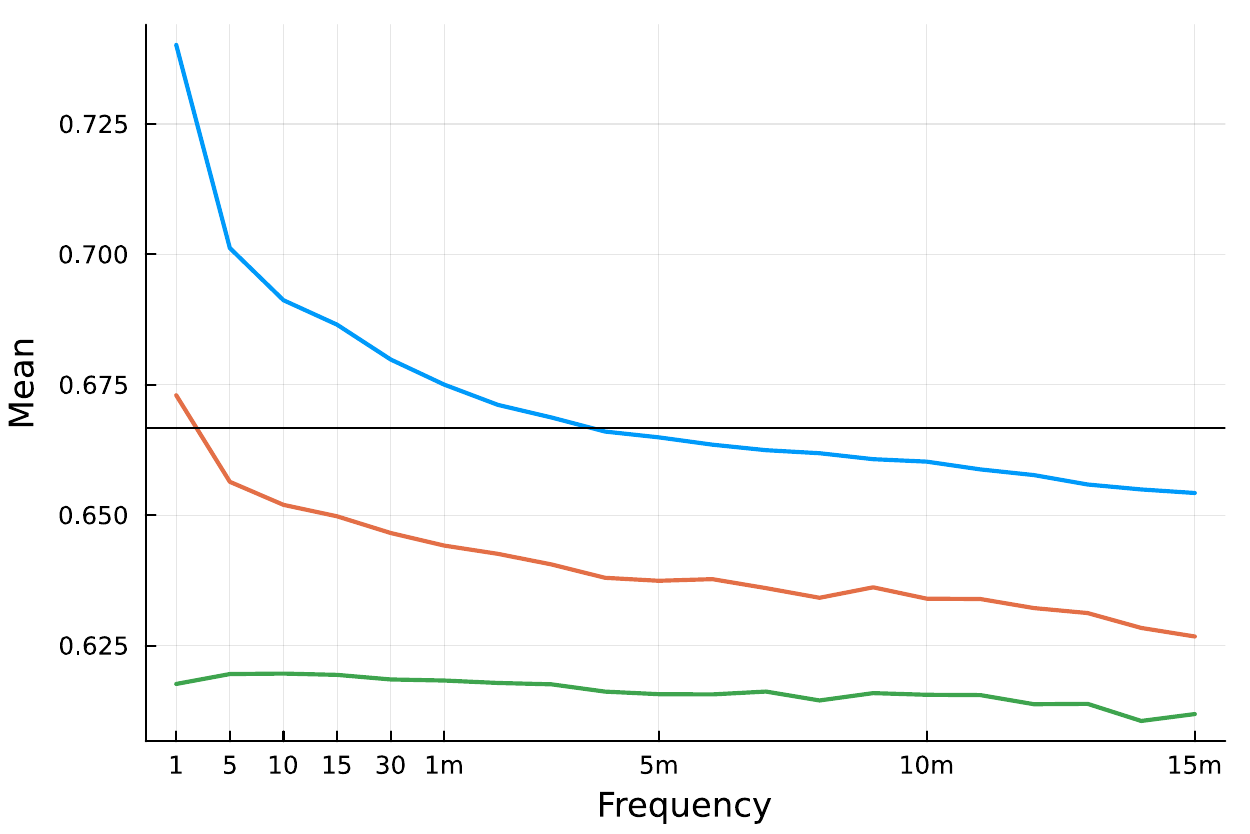}
\par\end{centering}
}
\par\end{centering}
\begin{centering}
\subfloat[Heston with rounding and grid-noise]{\begin{centering}
\includegraphics[width=0.4\textwidth]{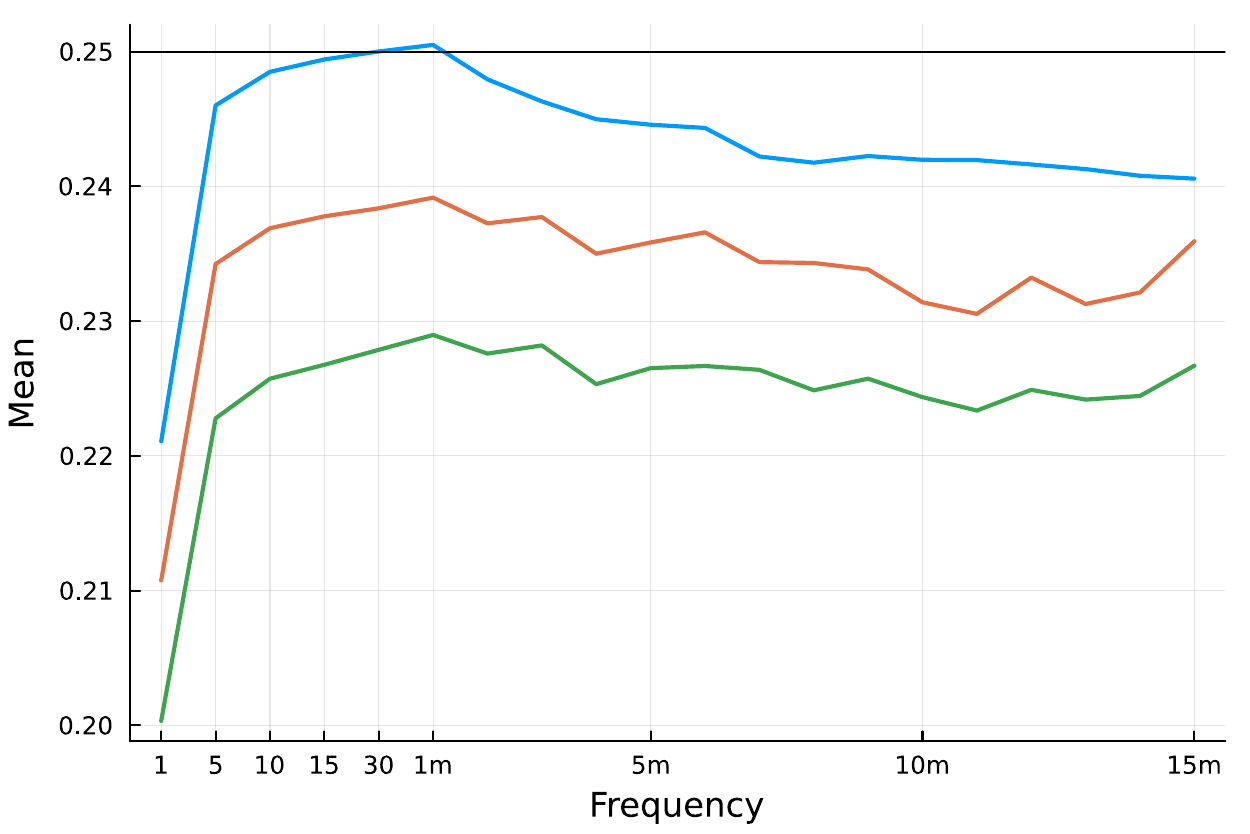}\includegraphics[width=0.4\textwidth]{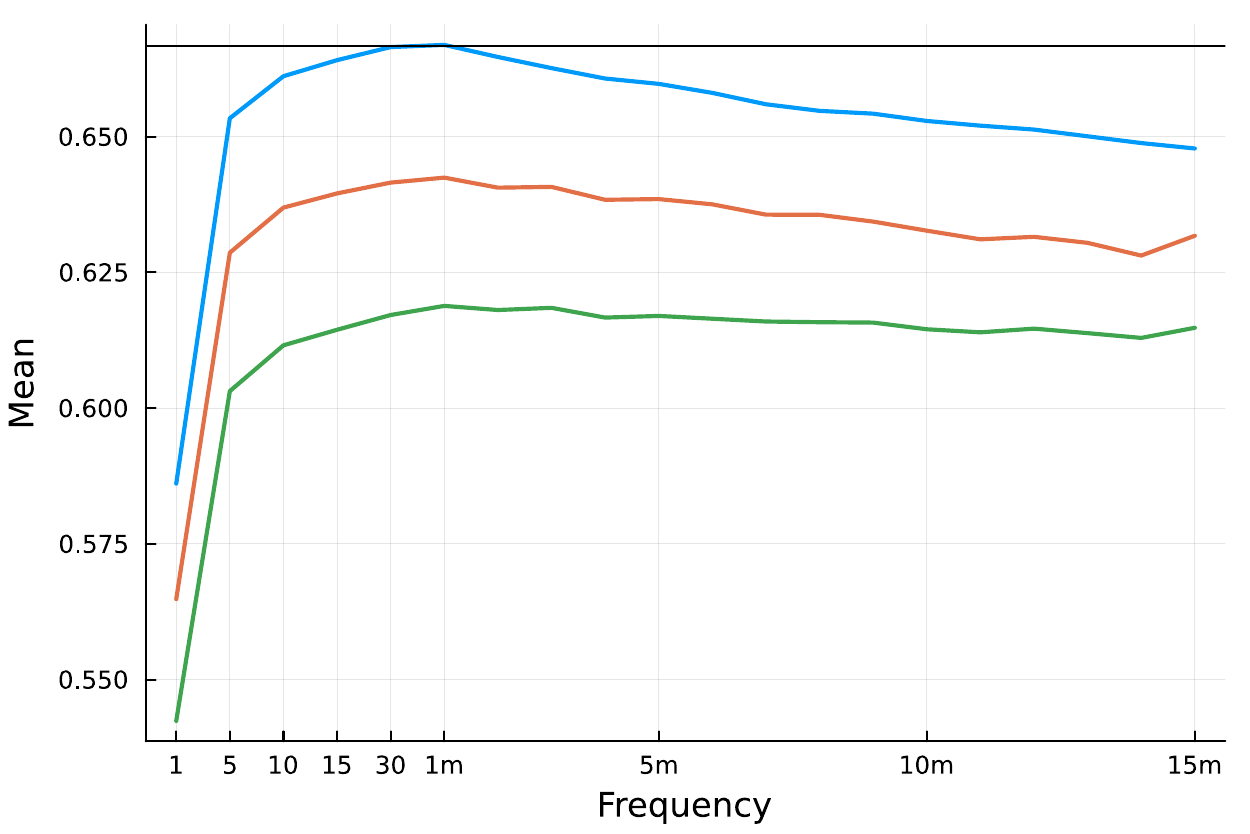}
\par\end{centering}
}
\par\end{centering}
\begin{centering}
\subfloat[Heston with rounding and staleness]{\begin{centering}
\includegraphics[width=0.4\textwidth]{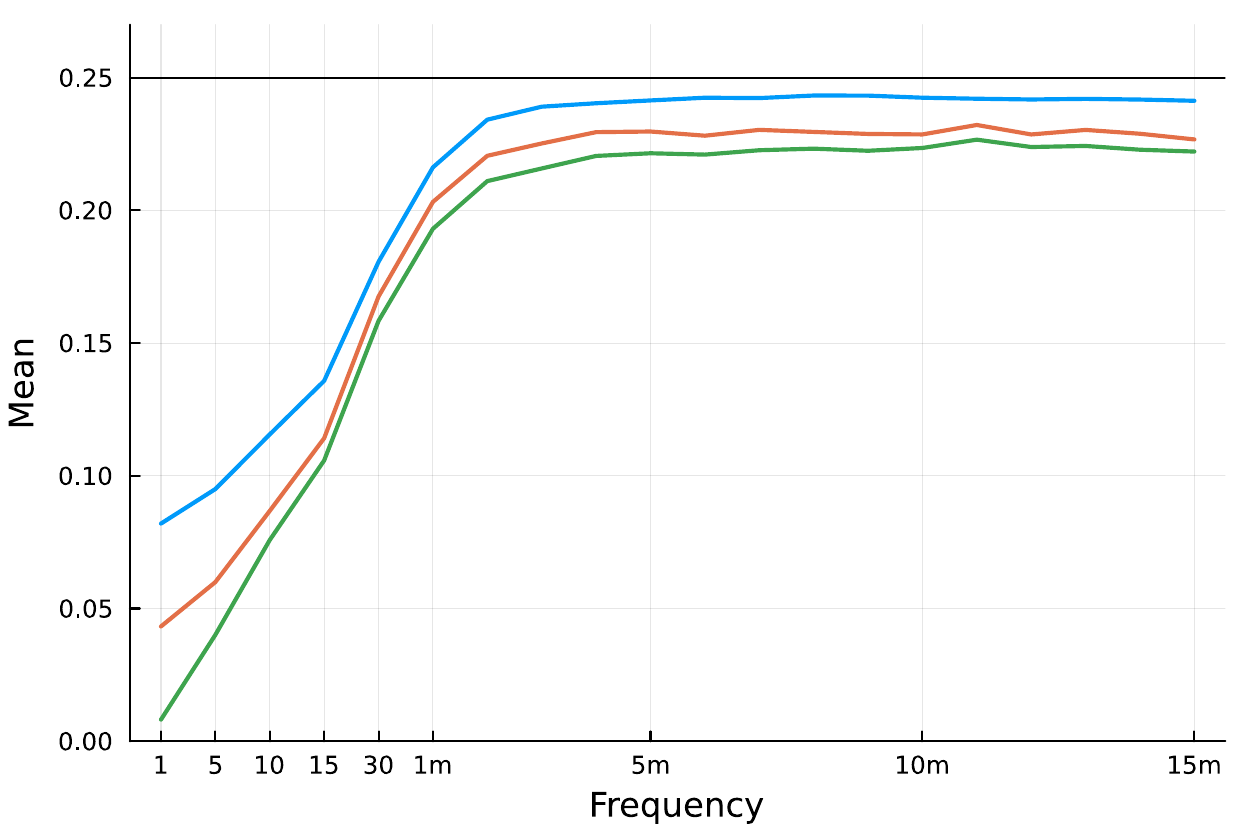}\includegraphics[width=0.4\textwidth]{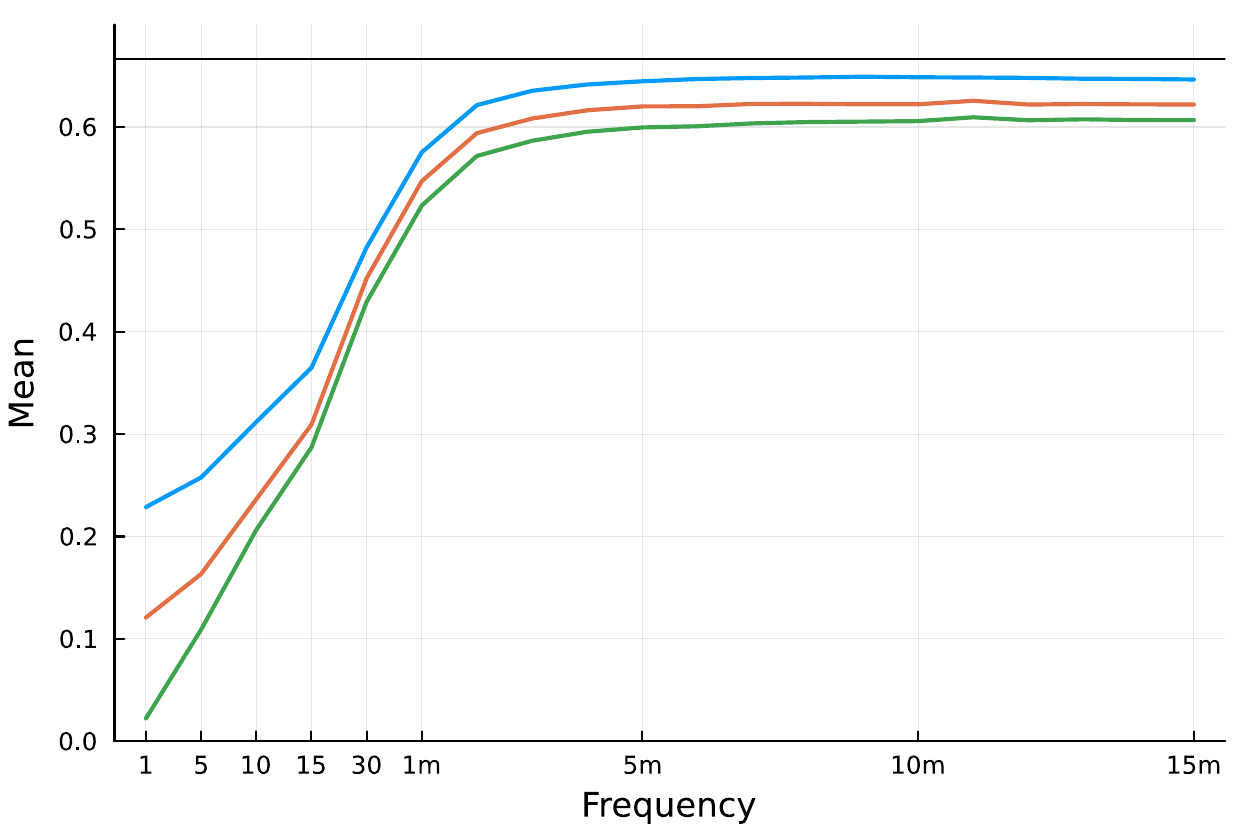}
\par\end{centering}
}
\par\end{centering}
\caption{Correlation signature plots, with sampling frequency ranging from
1 second to 15 minutes. Observed prices are generated with Heston
models with a layer of noise added. The true correlation is $\rho=1/4$
in the left panels and $\rho=2/3$ in the right panels. The four types of
noise are: (a) Independent noise ($\xi^{2}=10^{-3}$); (b) rounding
to a grid ($\alpha=10^{-3}\sigma$); (c) rounding with grid-noise
($\alpha=10^{-3}\sigma$ and $p=0.75$); and (d) rounding with stale
prices ($\alpha=10^{-3}\sigma$ and $q_{x}=0.50$ and $q_{y}=0.80$).\label{fig:Correlation-signature-plots}}
\end{figure}
\begin{figure}[H]
\begin{centering}
\subfloat[Heston with independent noise]{
\centering{}\includegraphics[width=0.4\textwidth]{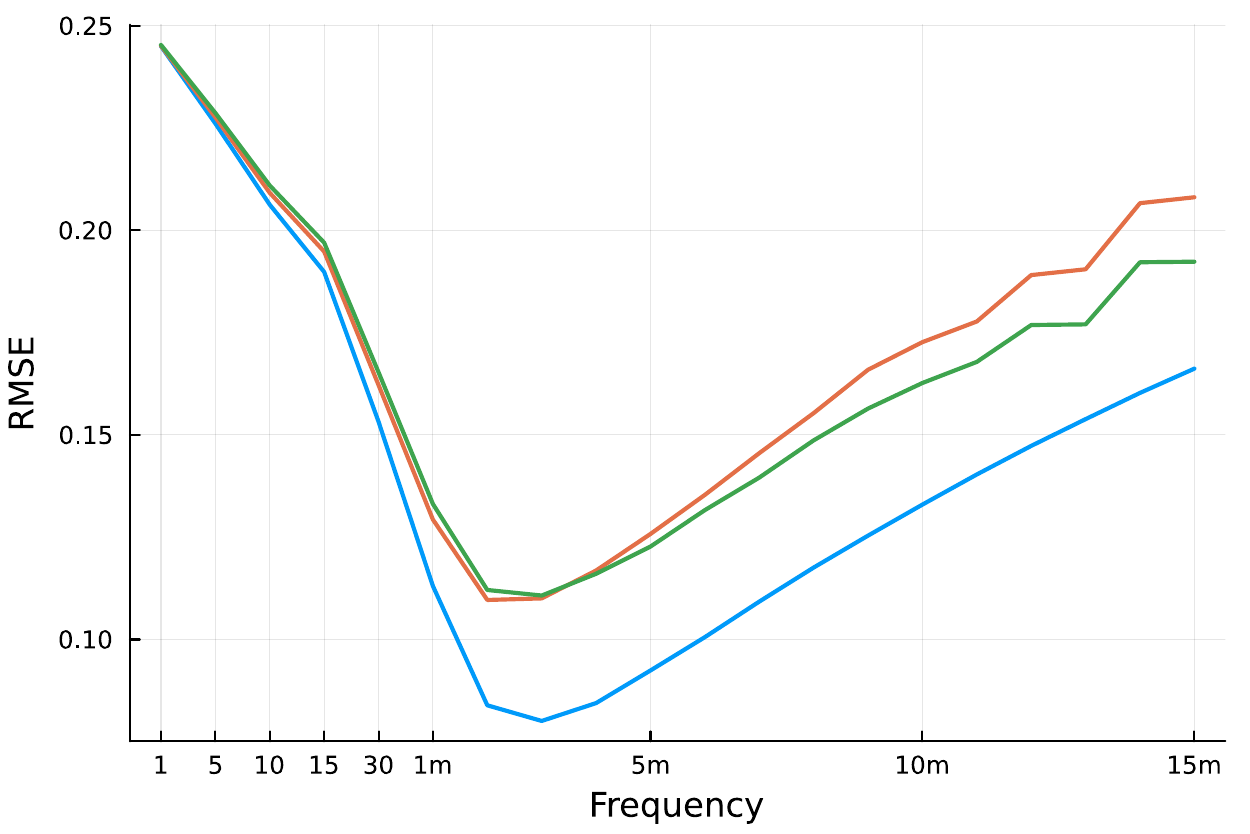}\includegraphics[width=0.4\textwidth]{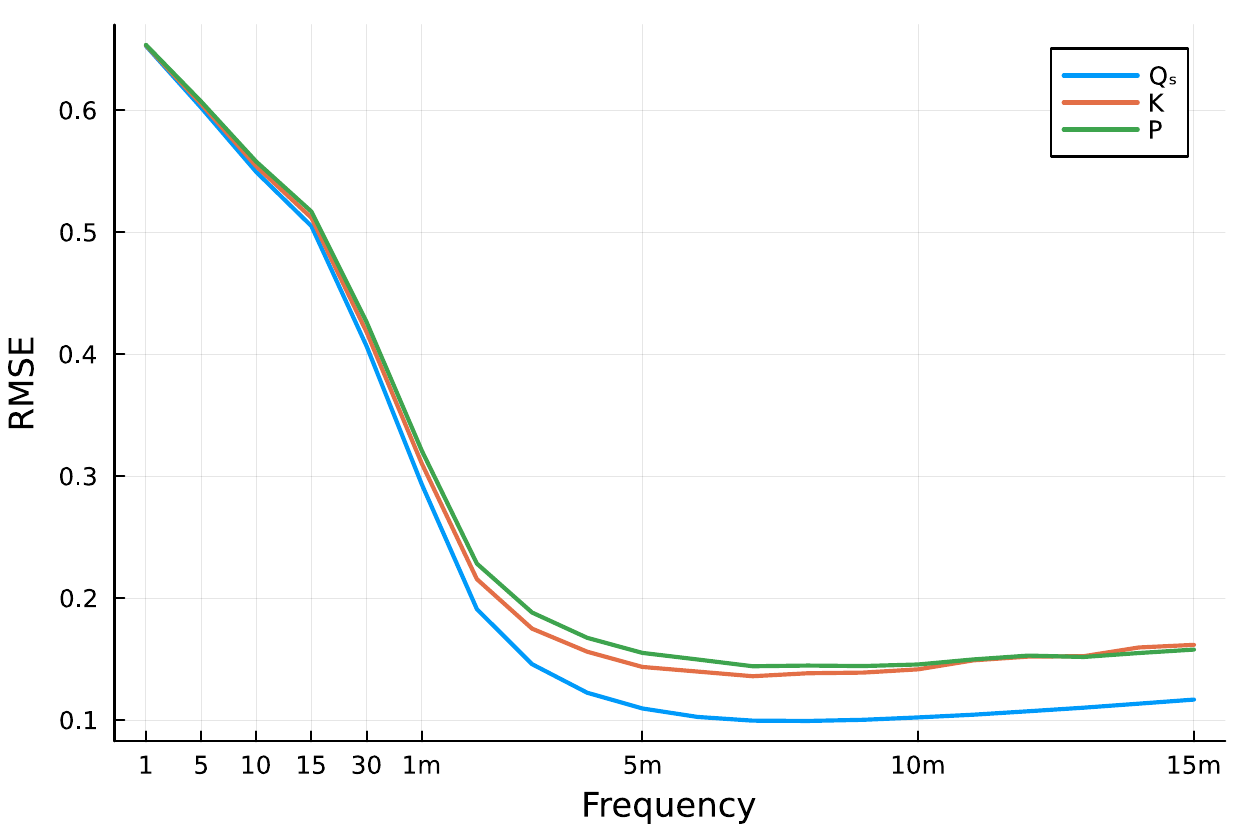}}
\par\end{centering}
\begin{centering}
\subfloat[Heston with rounding to a grid]{\begin{centering}
\includegraphics[width=0.4\textwidth]{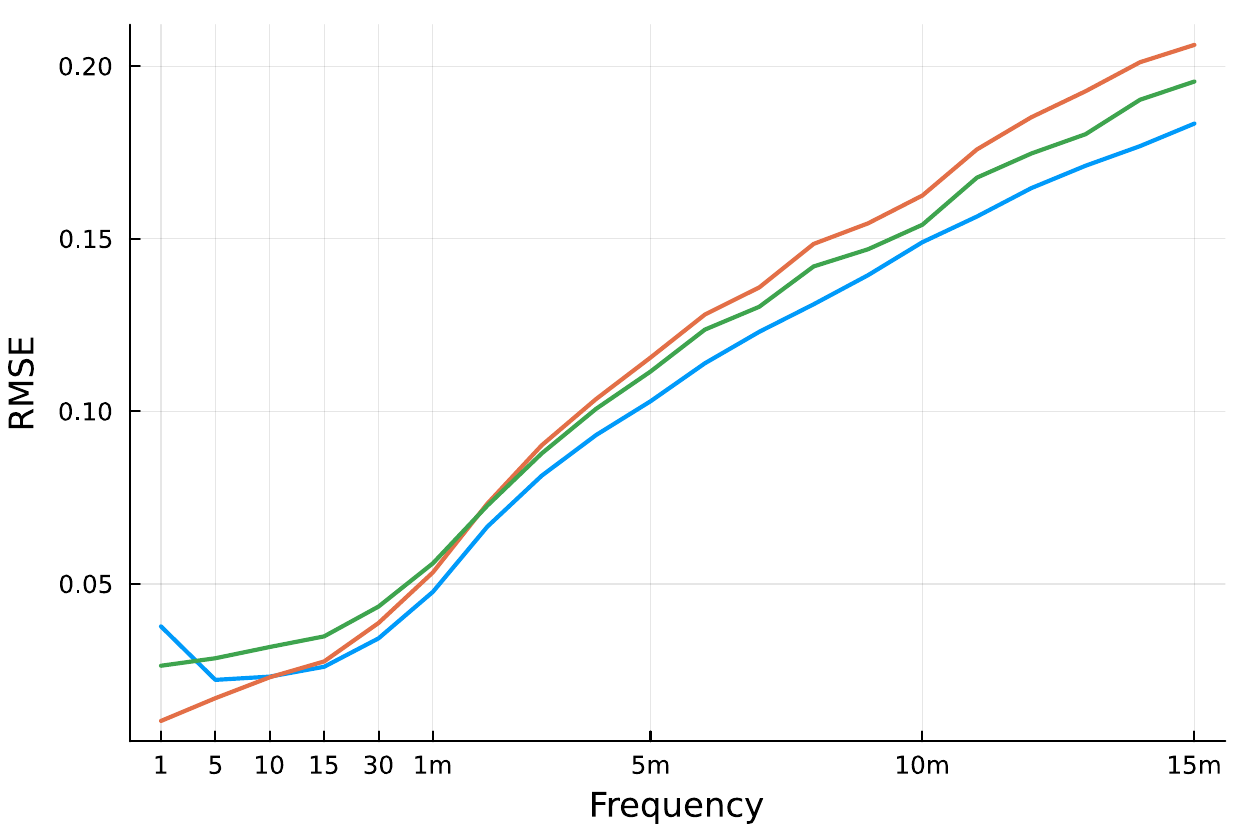}\includegraphics[width=0.4\textwidth]{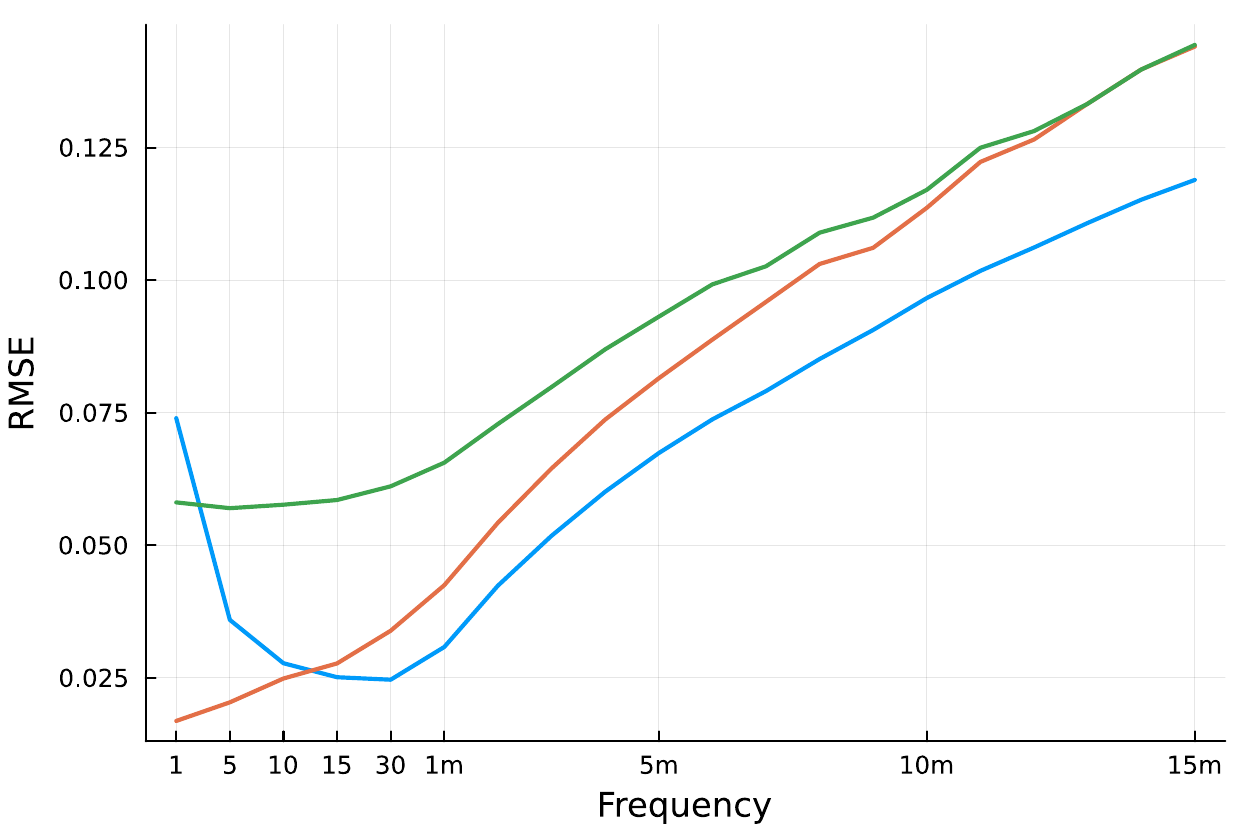}
\par\end{centering}
}
\par\end{centering}
\begin{centering}
\subfloat[Heston with rounding and grid-noise]{\begin{centering}
\includegraphics[width=0.4\textwidth]{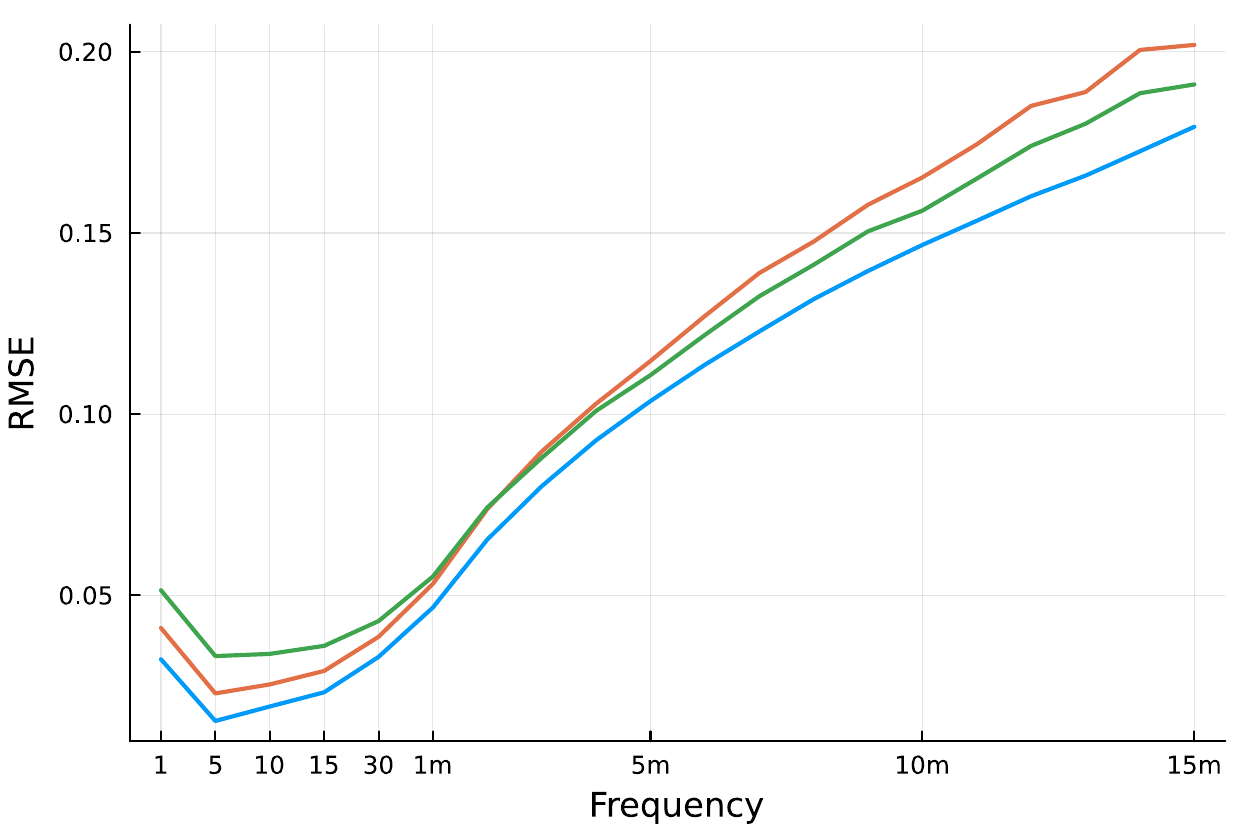}\includegraphics[width=0.4\textwidth]{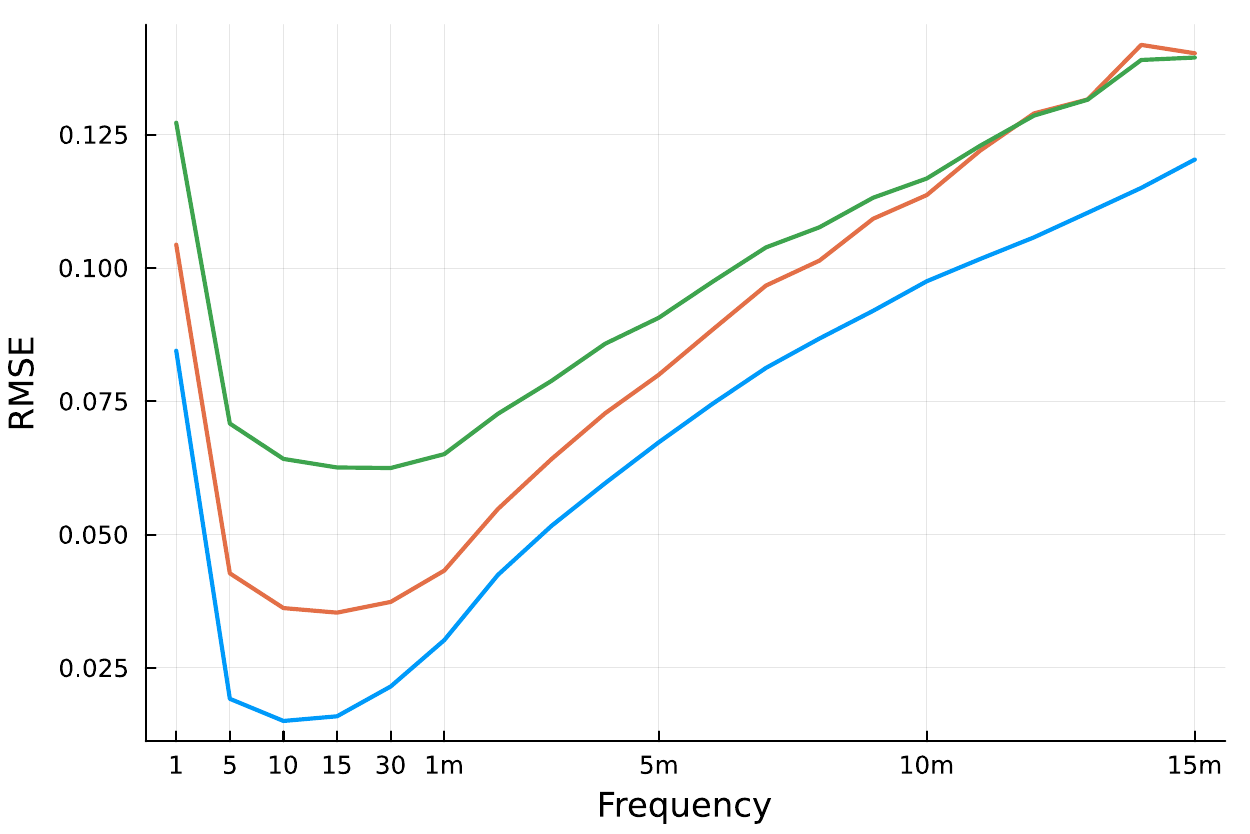}
\par\end{centering}
}
\par\end{centering}
\begin{centering}
\subfloat[Heston with rounding and staleness]{\begin{centering}
\includegraphics[width=0.4\textwidth]{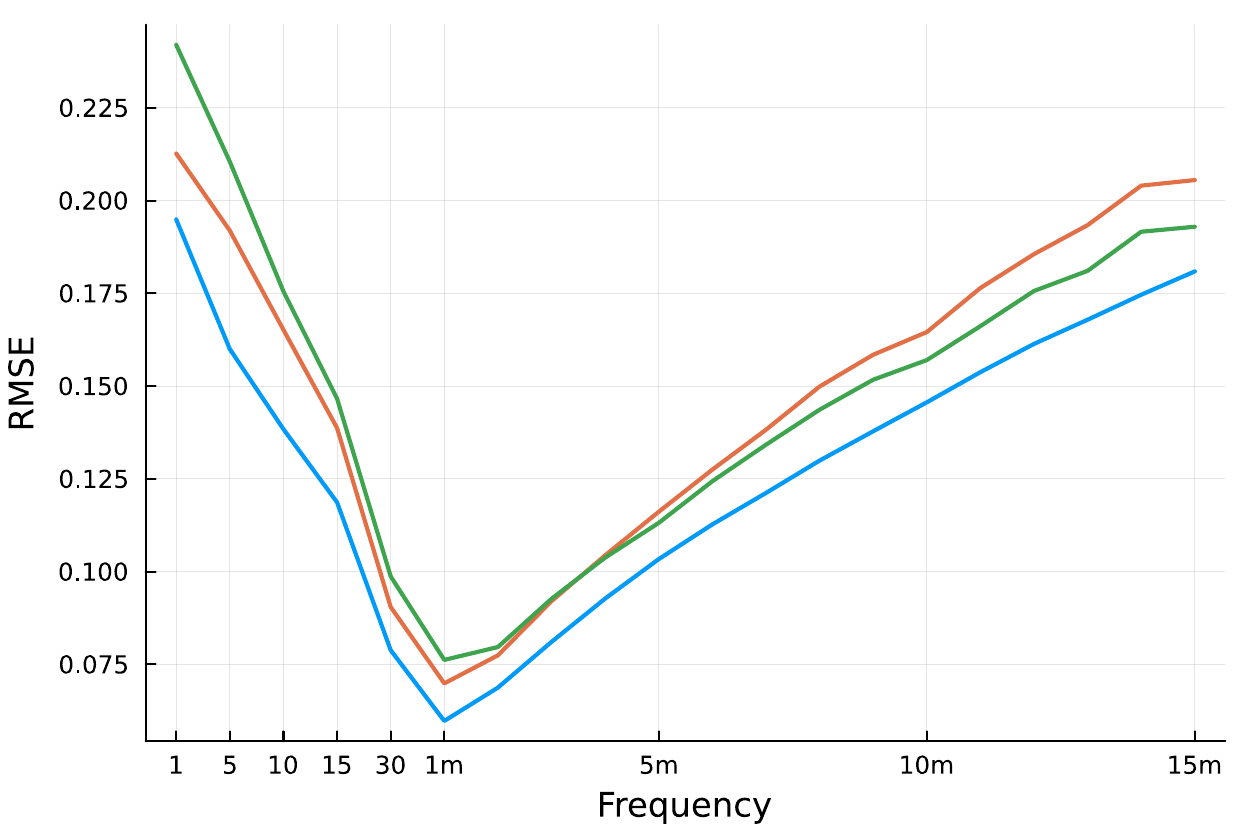}\includegraphics[width=0.4\textwidth]{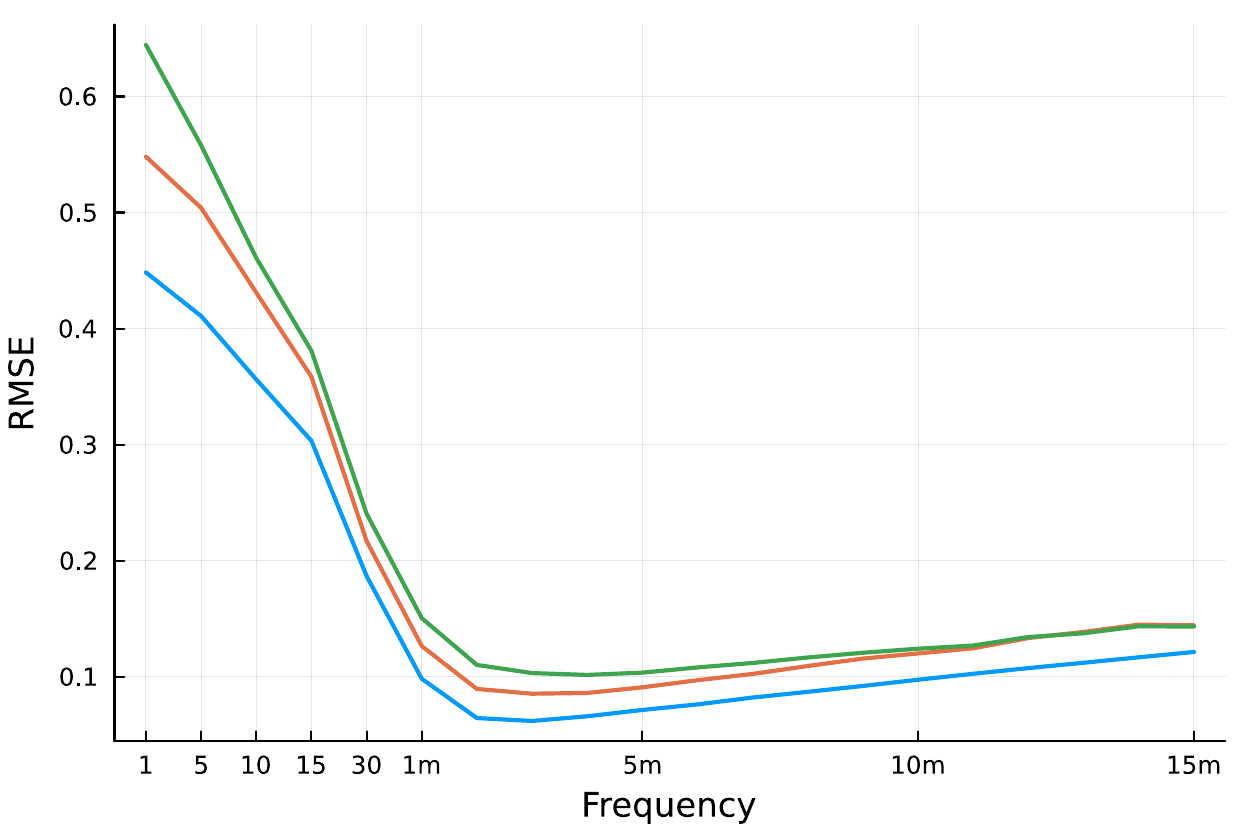}
\par\end{centering}
}
\par\end{centering}
\caption{RMSE signature plots. Observed prices are generated with Heston models
with a layer of noise added. The true correlation is $\rho=1/4$ in
the left panels and $\rho=2/3$ in the right panels. The four types of noise
are: (a) Independent noise; (b) rounding to a grid ($\alpha=10^{-3}\sigma$);
(c) rounding with grid-noise ($\alpha=10^{-3}\sigma$ and $p=0.75$);
and (d) rounding with stale prices ($\alpha=10^{-3}\sigma$ and $q_{x}=0.50$
and $q_{y}=0.80$).\label{fig:RMSE-signature-plots}}
\end{figure}

The corresponding root mean squared errors (RMSEs) are reported in
Figure \ref{fig:RMSE-signature-plots}. The new correlation estimator,
$Q_{S}$, tends to have the smallest RMSE, which is also true for
the additional simulation experiments presented in the Supplementary
Material.

\section{Empirical Application\label{sec:empirical}}

We apply the correlation estimators to high-frequency data for about
100 assets. We first compare daily correlation estimates, using signature plots
to study market microstructure noise and the sensitivity of the estimators
to the sampling frequency, $\delta$.
Then we turn to estimation of intraday correlations, which we find
to vary substantially over active trading hours. Correlations
between stocks and the market are, on average, increasing for all
assets in our sample. We obtain estimates of intraday betas, by combining
the correlation estimates with estimates of relative volatility. We
then relate intraday variation in correlations and betas
with asset characteristics, such as the low-frequency market beta,
market capitalization, and book-to-market valuations. This part of
the analysis is done with an expanded set of assets detailed below.

Our sample runs from January 1, 2015, through December
31, 2021, and contains 1,763 trading days. We use NYSE and NASDAQ transaction
prices from the TAQ database that were accessed through the Wharton
Research Data Services (WRDS) system. The data were cleaned following
the guidelines of \citet{BNHLS-RKpractice:2009}, and prices (when
unavailable) were interpolated by the previous-tick methods. We will
analyze 22 stocks and SPY, an exchange-traded fund that tracks the
S\&P 500 index, in great detail. We label this data set ``Small Universe''.
The 22 stocks were selected to be the two stocks with the largest
market capitalization from each of the eleven GICS\footnote{Global Industry Classification Standard.}
sectors. A larger panel of 98 stocks is used to identify stock characteristics
that are associated with different patterns in intraday market betas.
We refer to this data set as the ``Large Universe.'' It consists of
the stocks in the S\&P 100 index as of September 19, 2022, excluding
PYPL and RTX. PYPL was excluded because it began trading only in 2015
after being spun off from eBay, and RTX (formerly Raytheon Tech) was
excluded because its merger with United Technologies was completed
in April 2020.
\begin{table}[H]
\caption{Summary Statistics Small Universe\label{tab:StatisticsSmallUniverse}}

\begin{centering}
\begin{tabular*}{1\textwidth}{@{\extracolsep{\fill}}>{\raggedright}m{4.5cm}lS[table-format=4.2]S[table-format=2.2]S[table-format=2.2]S[table-format=2.2]}
\hline
\ \\[-6pt]
Sector (GICS code) & Ticker & \multicolumn{1}{c}{Average} & \multicolumn{1}{c}{Duration} & \multicolumn{2}{c}{Zero-returns (\%)} \\
 &  & \multicolumn{1}{c}{Price} & \multicolumn{1}{c}{(seconds)} & \multicolumn{1}{c}{$\delta=1$s} & \multicolumn{1}{c}{$\delta=3$m}\\[5pt]
\hline
\ \\[-6pt]
                        & SPY  & 281.89 & 2.28 & 63.39 & 3.73\\[5pt]
Energy                  & HAL  & 33.38  & 5.56 & 89.25 & 8.37\\
(10)                    & XOM  & 72.44  & 3.83 & 82.73 & 6.71\\[3pt]
Materials               & LYB  & 89.55  & 9.23 & 90.85 & 4.94\\
(15)                    & NEM  & 40.10  & 6.67 & 89.76 & 8.73\\[3pt]
Industrials             & AAL  & 34.32  & 6.83 & 89.76 & 8.54\\
(20)                    & UNP  & 144.23 & 6.86 & 87.82 & 4.29\\[3pt]
Consumer Discretionary  & TSLA & 415.90 & 5.12 & 70.38 & 1.14\\
(25)                    & AMZN & 1658.92& 5.03 & 76.60 & 0.75\\[3pt]
Consumer Staples        & PG   & 101.81 & 4.86 & 85.85 & 6.92\\
(30)                    & WMT  & 99.02  & 4.46 & 84.87 & 6.75\\[3pt]
Health Care             & JNJ  & 131.77 & 4.57 & 83.74 & 5.89\\
(35)                    & MRK  & 68.76  & 5.06 & 87.05 & 7.50\\[3pt]
Financials              & JPM  & 101.17 & 3.22 & 78.24 & 5.53\\
(40)                    & WFC  & 48.11  & 4.50 & 86.83 & 7.84\\[3pt]
Information Technology  & AAPL & 167.52 & 2.27 & 63.05 & 3.57\\
(45)                    & AMD  & 33.45  & 14.87& 82.87 & 19.17\\[3pt]
Communication Services  & DIS  & 121.95 & 4.19 & 80.95 & 5.18\\
(50)                    & FB   & 181.59 & 2.80 & 68.63 & 3.18\\[3pt]
Utilities               & D    & 75.08  & 8.04 & 91.35 & 7.92\\
(55)                    & DUK  & 84.94  & 7.47 & 90.61 & 7.26\\[3pt]
Real Estate             & AMT  & 168.88 & 10.18& 91.43 & 4.15\\
(60)                    & PLD  & 72.88  & 9.77 & 92.62 & 8.70\\[5pt]
\hline 
\end{tabular*}
\par\end{centering}
Note: Summary statistics for SPY and 22 assets (two from each of the
11 sectors) for the sampling period from January 1, 2015 to December
31, 2021. Average price and average duration between consecutive
transactions are reported along with the percentage of zero returns
when returns are sampled at 1 second and 3 minutes, respectively.
\end{table}

Table \ref{tab:StatisticsSmallUniverse} presents summary statistics
for the Small Universe with 22 assets. The exchange-traded fund, SPY,
is the most frequently traded asset, followed by AAPL and FB. On average,
their average duration between transactions is just over two seconds.
The price level is an interesting statistic, because the tick-size
is more likely to induce rounding errors and price staleness for assets
trading at low prices. This appears to be relevant for AMD that traded
for less than \$3 in all of 2015 and below \$10 during most of the
first three years in our sample period. This helps explain why nearly
20\% of 3-minute returns for AMD are zero in this sample period.

The assets in the Large Universe are listed by sector
in Table \ref{tab:StatisticsLargeUniverse}. 
\begin{table}[H]
\caption{Large Universe\label{tab:StatisticsLargeUniverse}}

\begin{centering}
\vspace{0.2cm}
\begin{small}
\begin{tabularx}{\textwidth}{YYYYYYYYYYYYY}
\toprule 
    Energy & Materials & Industrials & \multicolumn{1}{c}{Consumer} & \multicolumn{1}{c}{Consumer }& Healthcare  \\
          &        &       & \multicolumn{1}{c}{Discretionary} & \multicolumn{1}{c}{Staples} &\\
    \midrule
         &        &       &       &      &         \\[-3pt]
    COP  & DOW    & BA    & AMZN  & CL   & ABBV    \\
    CVX  & LIN    & CAT   & BKNG  & COST & ABT     \\
    XOM  &        & EMR   & F     & KHC  & AMGN    \\
         &        & FDX   & GM    & KO   & BMY     \\
         &        & GD    & HD    & MDLZ & CVS     \\
         &        & GE    & LOW   & MO   & DHR     \\
         &        & HON   & MCD   & PEP  & GILD    \\
         &        & LMT   & NKE   & PG   & JNJ     \\
         &        & MMM   & SBUX  & PM   & LLY     \\
         &        & UNP   & TGT   & WBA  & MDT     \\
         &        & UPS   & TSLA  & WMT  & MRK     \\
         &        &       &       &      & PFE     \\
         &        &       &       &      & TMO     \\
         &        &       &       &      & UNH     \\
         &        &       &       &      &         \\
    \midrule
     Financials & Information & Telecom. & Utilities & Real Estate &\\
                & Technology  & Services &           &             & \\
    \midrule
           &       &       &     &     & \\[-3pt]
     AIG   & AAPL  & CHTR  & DUK & AMT & \\
     AXP   & ACN   & CMCSA & EXC & SPG & \\
     BAC   & ADBE  & DIS   & NEE &     & \\
     BK    & AMD   & GOOGL & SO  &     & \\
     BLK   & AVGO  & FB    &     &     & \\
     BRKB  & CRM   & NFLX  &     &     & \\
     C     & CSCO  & T     &     &     & \\
     COF   & IBM   & TMUS  &     &     & \\
     GS    & INTC  & VZ    &     &     & \\
     JPM   & MA    &       &     &     & \\
     MET   & MSFT  &       &     &     & \\
     MS    & NVDA  &       &     &     & \\
     SCHW  & ORCL  &       &     &     & \\
     USB   & QCOM  &       &     &     & \\
     WFC   & TXN   &       &     &     & \\
           & V     &       &     &     & \\
\bottomrule
\end{tabularx}
\end{small}
\par\end{centering}
Note: List of assets in ``Large Universe'', organized by sectors.
\end{table}

\subsection{Estimates of Daily Correlations}

We apply the correlation estimators to daily high-frequency data using
calendar-time sampling with frequencies ranging from 1 second to 15
minutes. The resulting correlation signature plots are shown in Figure
\ref{fig:CorrSignaturePlots} for a subset of the assets. These are
the two most actively traded securities, SPY and AAPL, the stock with
the most zero returns, AMD, and the two stocks from the Materials sector,
LYB and NEM, whose liquidity and zero-return frequencies are more
typical of assets in the Small Universe. Signature plots were introduced
in \citet{andersen-bollerslev-diebold-labys:00a} who plotted the
average realized variance against the sampling frequency used to compute
the underlying intraday returns. Signature plots help identify bias
in the estimators, which tend to be most pronounced at high sampling
frequencies. If the estimator is unbiased over a range of sampling
frequencies, then the signature plot will be roughly flat over those
sampling frequencies.
\begin{figure}[H]
\begin{centering}
\subfloat[$\mathrm{corr}$(SPY,AAPL)]{\begin{centering}
\includegraphics[width=0.32\textwidth]{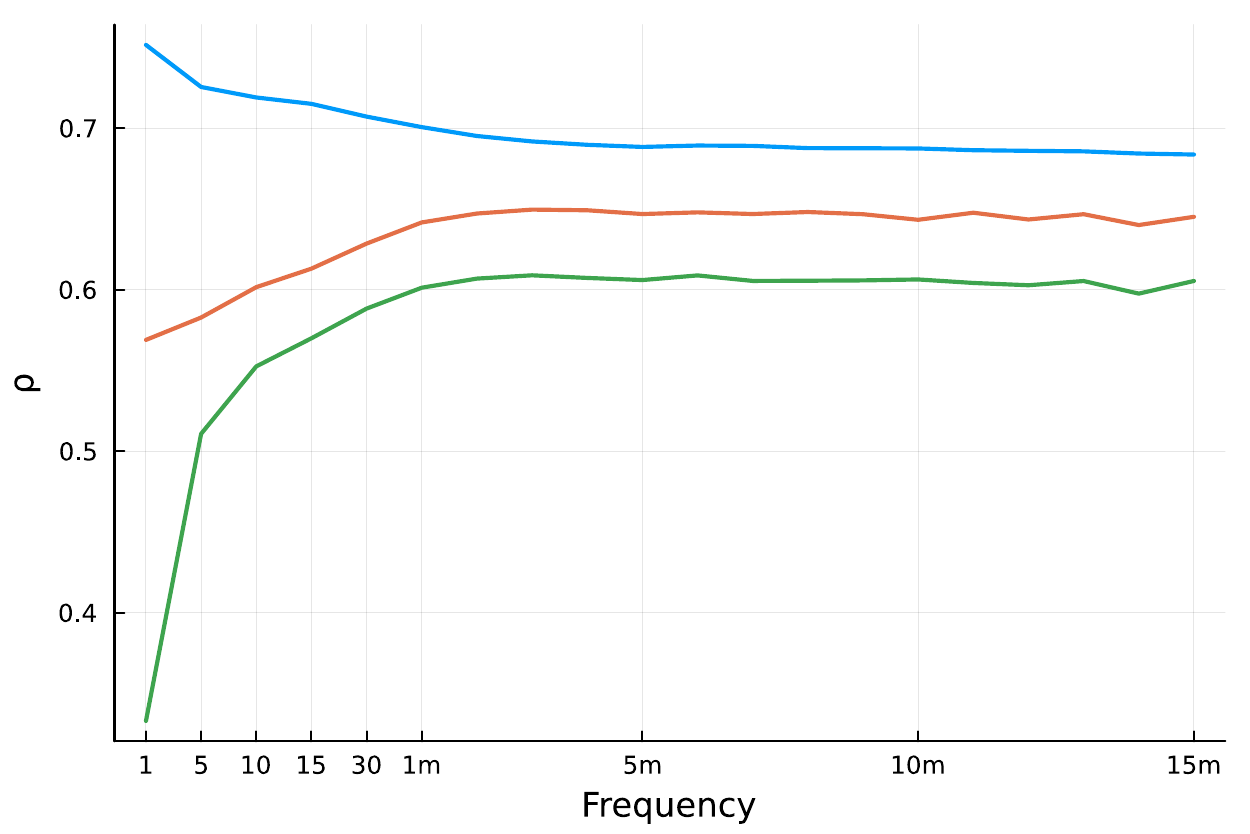}
\par\end{centering}
}\subfloat[$\mathrm{corr}$(SPY,AMD)]{\centering{}\includegraphics[width=0.32\textwidth]{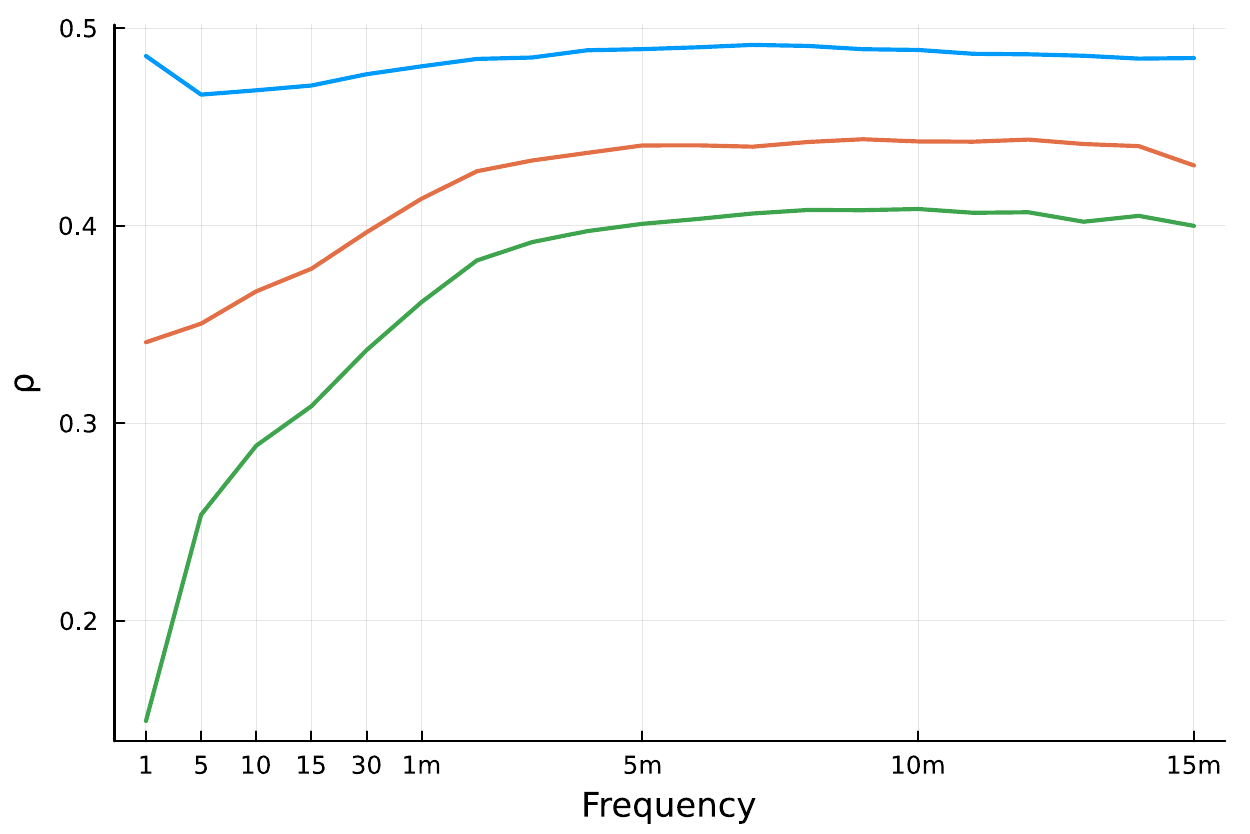}}\subfloat[$\mathrm{corr}$(AAPL,AMD)]{\begin{centering}
\includegraphics[width=0.32\textwidth]{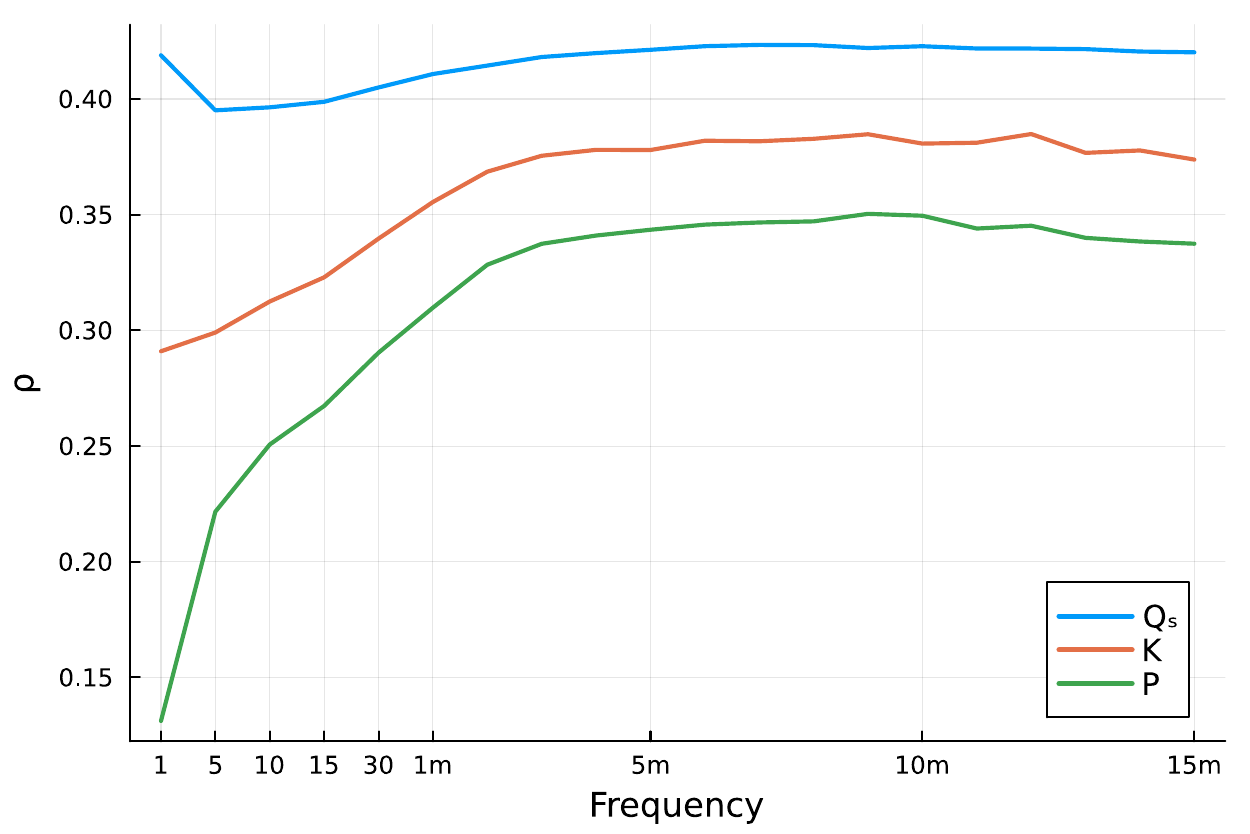}
\par\end{centering}
}
\par\end{centering}
\begin{centering}
\subfloat[$\mathrm{corr}$(SPY,LYB)]{\begin{centering}
\includegraphics[width=0.32\textwidth]{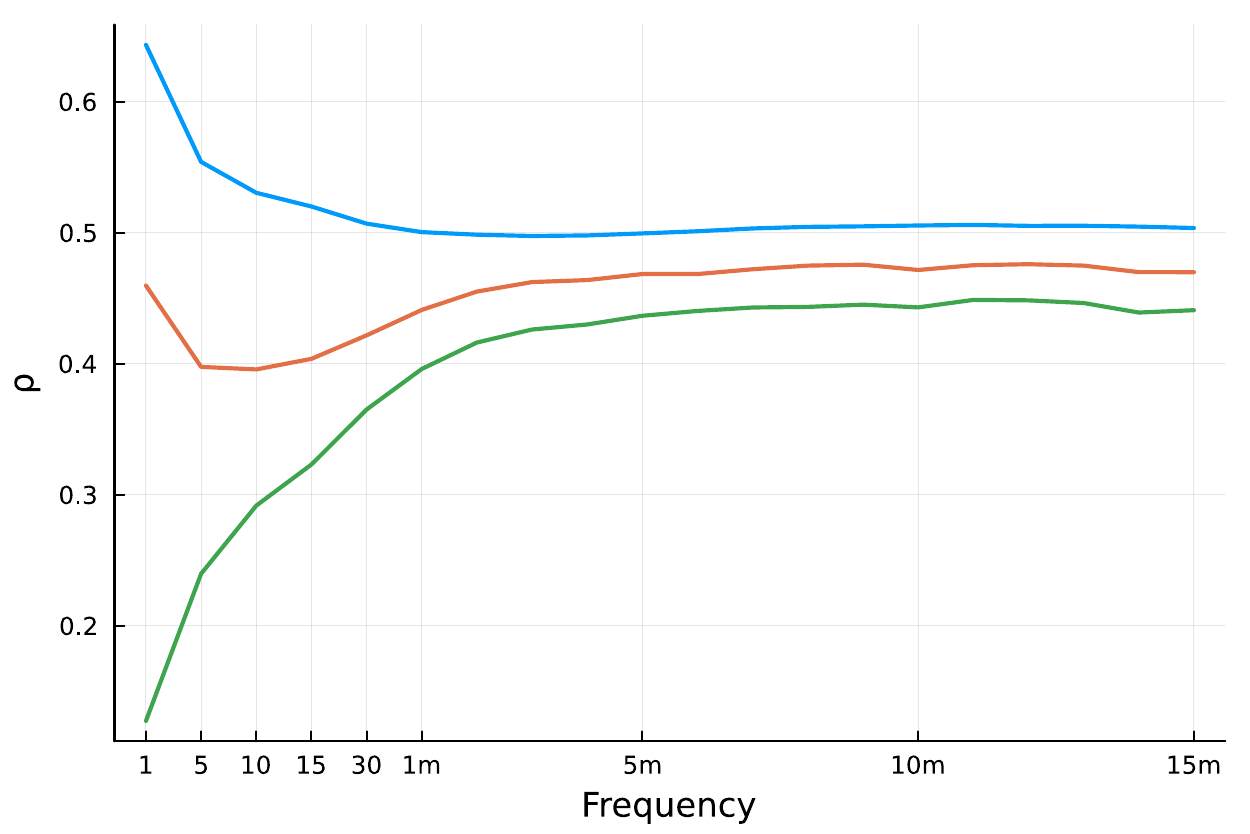}
\par\end{centering}
}\subfloat[$\mathrm{corr}$(SPY,NEM)]{\centering{}\includegraphics[width=0.32\textwidth]{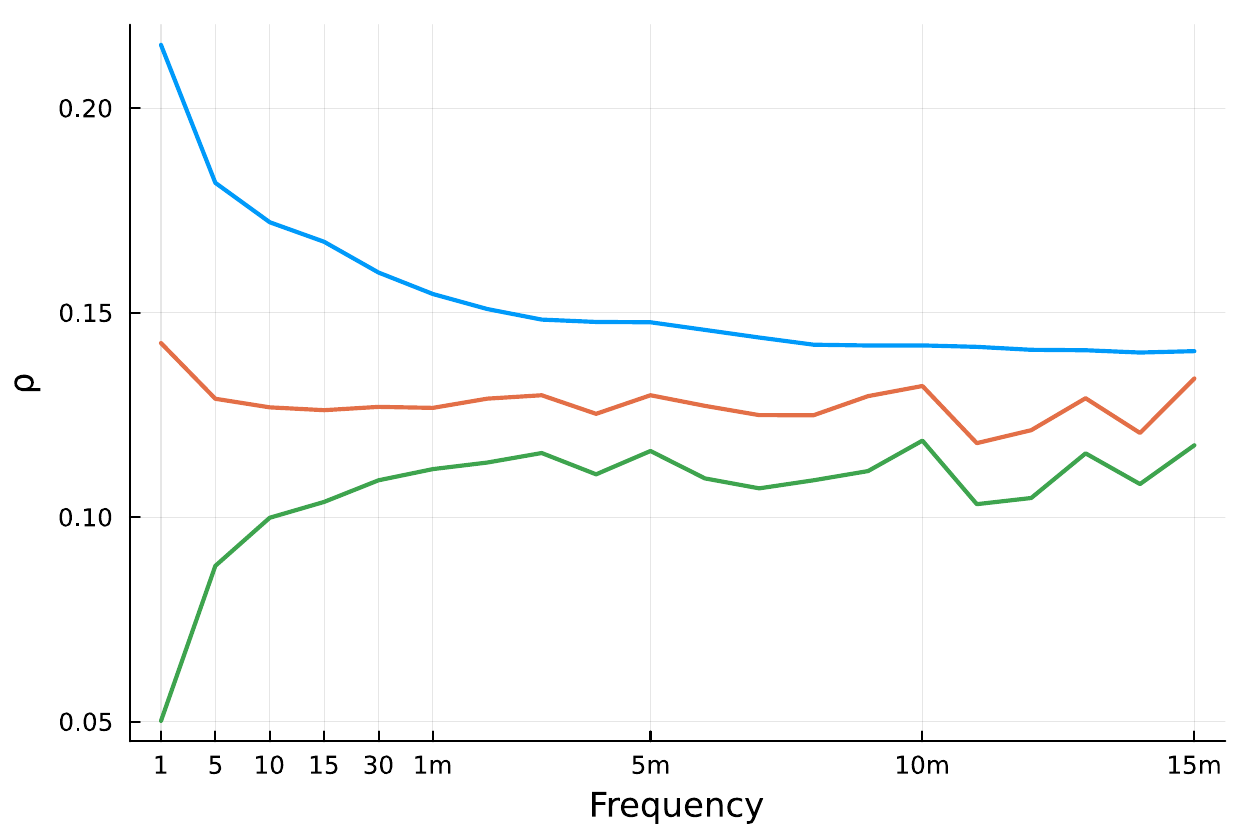}}\subfloat[$\mathrm{corr}$(LYB,NEM)]{\begin{centering}
\includegraphics[width=0.32\textwidth]{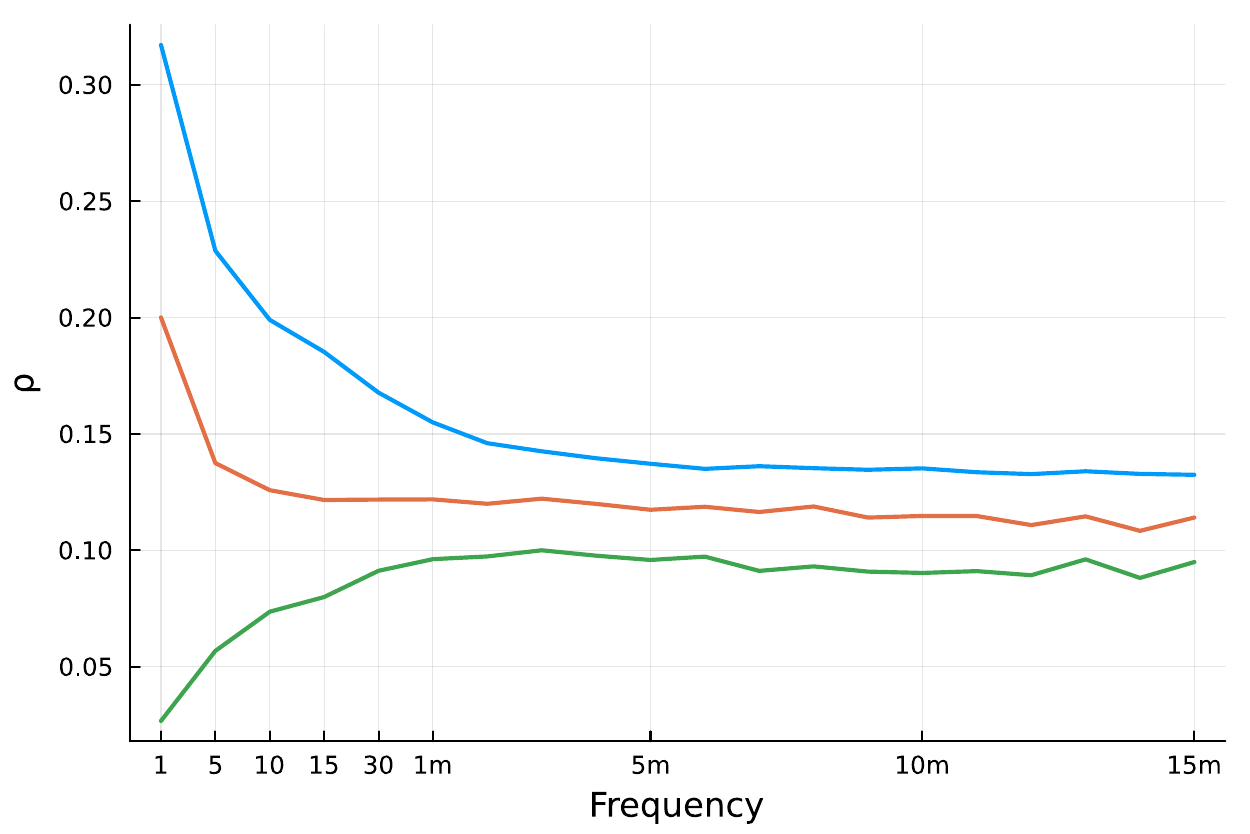}
\par\end{centering}
}
\par\end{centering}
\caption{Correlation signature plots, where the average correlation estimate
is plotted against sampling frequency for the three estimators, $Q_{S}$,
$P$, and $K$.\label{fig:CorrSignaturePlots}}
\end{figure}

The plots in Figure \ref{fig:CorrSignaturePlots} are correlation
signature plots, which can be used to visualize biases, such
as the Epps effect. Here we observe that many of the plots have patterns
that resemble the effect for rounding to a grid, because $Q_{S}$
often has an upward bias at high sampling frequencies, while $P$
has a downward bias. Additional signature plots are presented in
the Supplementary Material, see Figure \ref{fig:SignaturePlotsSmallUniverse}.
We adopt three minutes as a common sampling interval for all estimators.
This choice is partly because signature plots tend to be flat for
$\delta\geq$3 minutes and partly because it makes our results more
comparable to those in ATT. 
\begin{table}[tbh]
\caption{Correlation estimates.\label{tab:SectorCorr}}

\begin{centering}
\setlength{\tabcolsep}{5.3pt}
\begin{tabular*}{\textwidth}{@{\extracolsep{\fill}}l ccc c ccc c ccc c ccc}
\hline
\\[-5pt]
 \multicolumn{1}{c}{Sector code}  & \multicolumn{3}{c}{Average} &  & \multicolumn{3}{c}{25th quantile} &  & \multicolumn{3}{c}{Median} &  & \multicolumn{3}{c}{75th quantile}\\
 \multicolumn{1}{c}{(Asset pair)} & $Q_S$ & $K$   & $P$   &  & $Q_S$ & $K$   & $P$    &  & $Q_S$ & $K$   & $P$   &  & $Q_S$ & $K$   & $P$ \\[4pt]
 \hline
 & & & \\[-5pt]
10 {\footnotesize{(HAL,XOM)}}  &  0.59 &  0.55 &  0.53 &  &  0.49 &  0.45 &  0.42  &  &  0.60 &  0.56 &  0.54 &  &  0.70 &  0.66 &  0.65\\[2pt]
15 {\footnotesize{(LYB,NEM)}}  &  0.14 &  0.11 &  0.09 &  &  0.02 &  0.00 &  -0.03 &  &  0.14 &  0.12 &  0.10 &  &  0.25 &  0.23 &  0.22\\[2pt]
20 {\footnotesize{(AAL,UNP)}}  &  0.31 &  0.28 &  0.25 &  &  0.19 &  0.17 &  0.13  &  &  0.29 &  0.27 &  0.24 &  &  0.41 &  0.38 &  0.37\\[2pt]
25 {\footnotesize{(AMZN,TSLA)}}&  0.37 &  0.35 &  0.33 &  &  0.24 &  0.23 &  0.21  &  &  0.35 &  0.34 &  0.33 &  &  0.48 &  0.48 &  0.45\\[2pt]
30 {\footnotesize{(PG,WMT)}}   &  0.41 &  0.37 &  0.33 &  &  0.29 &  0.25 &  0.20  &  &  0.40 &  0.36 &  0.33 &  &  0.52 &  0.49 &  0.46\\[2pt]
35 {\footnotesize{(JNJ,MRK)}}  &  0.56 &  0.52 &  0.49 &  &  0.46 &  0.41 &  0.38  &  &  0.57 &  0.52 &  0.50 &  &  0.66 &  0.63 &  0.61\\[2pt]
40 {\footnotesize{(JPM,WFC)}}  &  0.74 &  0.72 &  0.68 &  &  0.67 &  0.64 &  0.61  &  &  0.76 &  0.72 &  0.70 &  &  0.83 &  0.79 &  0.78\\[2pt]
45 {\footnotesize{(AAPL,AMD)}} &  0.40 &  0.37 &  0.33 &  &  0.26 &  0.21 &  0.16  &  &  0.40 &  0.36 &  0.33 &  &  0.55 &  0.52 &  0.49\\[2pt]
50 {\footnotesize{(DIS,FB)}}   &  0.34 &  0.31 &  0.27 &  &  0.20 &  0.17 &  0.14  &  &  0.33 &  0.30 &  0.26 &  &  0.48 &  0.45 &  0.42\\[2pt]
55 {\footnotesize{(D,DUK)}}    &  0.76 &  0.72 &  0.70 &  &  0.70 &  0.66 &  0.63  &  &  0.78 &  0.73 &  0.71 &  &  0.83 &  0.79 &  0.78\\[2pt]
60 {\footnotesize{(AMT,PLD)}}  &  0.52 &  0.48 &  0.44 &  &  0.43 &  0.38 &  0.32  &  &  0.53 &  0.48 &  0.46 &  &  0.63 &  0.59 &  0.57\\[4pt]
\hline 
\end{tabular*}\medskip{}
\par\end{centering}
Note: Summary statistics for pairs of assets, within each of the 11
sectors: Energy (10), Materials (15), Industrials (20), Consumer Discretionary
(25), Consumer Staples (30), Health Care (35), Financials (40), Information
Technology (45), Communication Services (50), Utilities (55), Real
Estate (60).
\end{table}
 The optimal sampling frequency likely varies over time and across
assets. We leave its estimation for future research.

Table \ref{tab:SectorCorr} presents summary statistics for correlations
between stocks in the same sector. For each of the three estimators,
we compute the average, median, and interquartile range across the
1,763 daily estimates. For all pairs, these quantities are similar
for the three estimators. The interquartile range across days is predominantly
driven by time variation in daily correlations. So, a similar width
for the interquartile range should not be interpreted as indicating
that the estimators have similar precision. The simulations in Section \ref{sec:Simulation-Study} and the results
in the next subsection indicate that $Q_{S}$ is more precise than
$K$ and $P$. While the summary statistics are similar for the three
estimators, we always have $Q_{S}>K>P$. The same ordering is observed
in the simulation designs of Section \ref{sec:Simulation-Study}.
It is consistent with Theorem \ref{theo:Consistency}, under which
time-varying volatility attenuates both $P$ and $K$, although the
theorem does not rank the two attenuations.

Next, we estimate daily correlations between each of the 22 stocks
and SPY. The average, median, and interquartile range (over the 1,763)
estimates are shown in Figure \ref{fig:Daily-Correlations-with-SPY}.
Once again we see that the quantities are similar for the three estimators,
as was the case in Table \ref{tab:SectorCorr}, and once again
we find $Q_{S}>K>P$ uniformly across all assets and across all measurements.
\begin{figure}[H]
\begin{centering}
\includegraphics[width=1\textwidth]{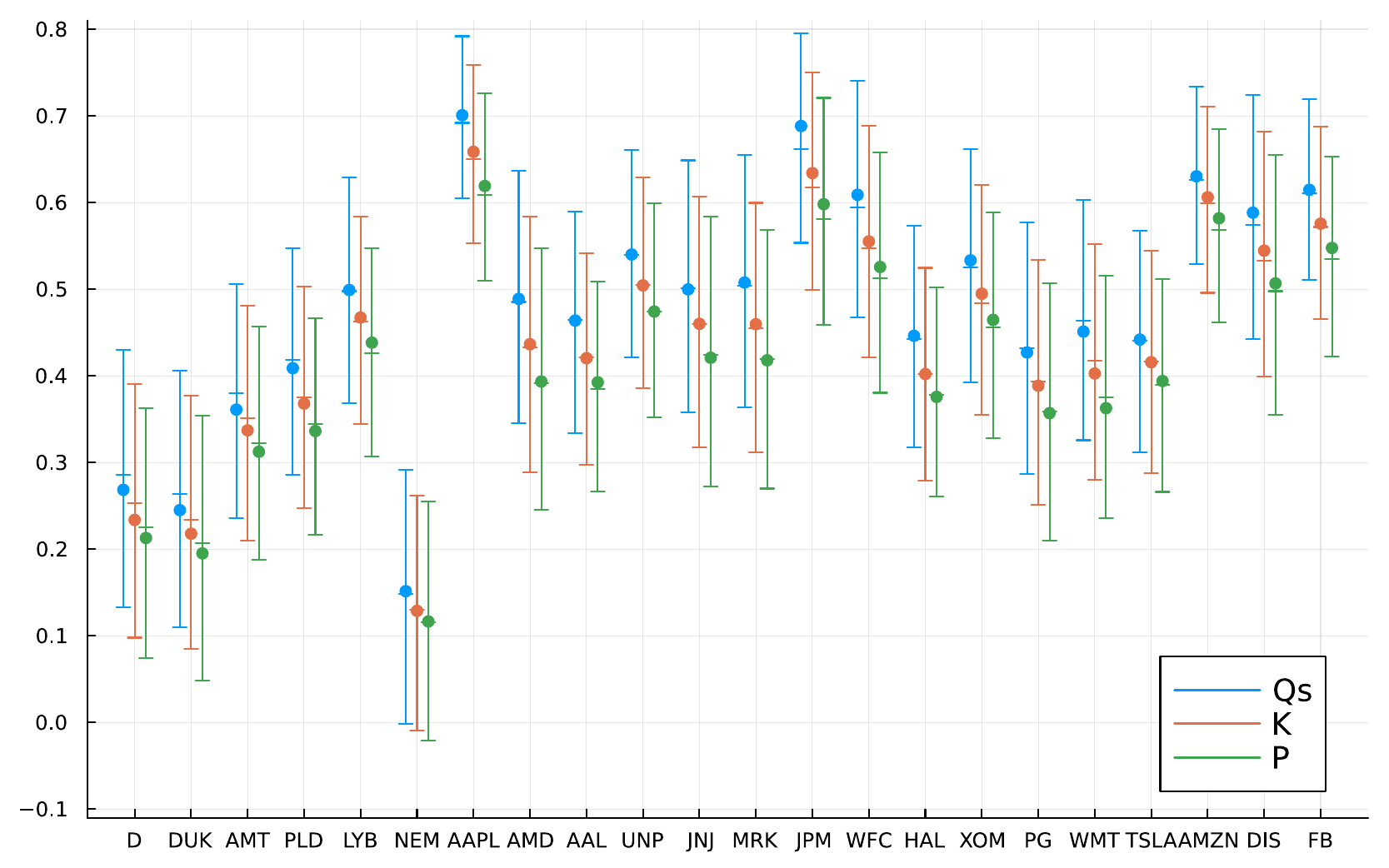}
\par\end{centering}
\caption{Daily correlations between assets and market returns. The three estimators
were applied to 1,763 trading days. The average estimate (dash) and
median estimate (bullet) are shown. The vertical lines present the
interquartile range over the 1,763 daily estimates.\label{fig:Daily-Correlations-with-SPY}}

\end{figure}

\subsection{Intraday Correlations and Market Beta Estimation\label{sec:Empirical-intraday}}

Estimating betas from high-frequency data is an active research area,
see, e.g., \citet[2006]{AndersenBollerslevDieboldWu:2005}\nocite{ABDW:2006},
\citet{TodorovBollerslev:2010}, \citet{DovononGoncalvesMeddahi:2013},
\citet{HansenLundeVoev:2014}, and \citet{ReiBTodorovTauchen:2015}.
This literature has also documented substantial time variation in
betas. Recently, \citet{AndersenThyrsgaardTodorov:2021}
(ATT) documented systematic time-variation in betas within the trading
day. Specifically, they estimated betas for rolling windows (spanning
two hours) using 3-minute intraday returns. Their local estimate of
beta is simply a local estimator of the covariance between asset and
market returns divided by a local estimate of the quadratic variation
of market returns. The local (time-of-day) estimates are averaged
over the 2,243 days in their sample period (2010-2018). ATT found a great deal of variation in the betas within the day. Some
stocks had increasing betas over the trading day, whereas others
had decreasing betas. Related empirical evidence is reported
in \citet{ShenShi:2024}, who study intraday variation in pairwise
correlations and betas for U.S. stocks and link the cross-sectional
comovement patterns to index-based trading strategies. Their findings
are broadly consistent with ours, including the intraday increase
in correlations, although there are also some differences in the details,
which may reflect differences in estimation methods and sample periods.

We decompose the market beta into the product of correlation and
relative volatility,
\[
\beta_{i,t}=\rho_{i,t}\times\lambda_{i,t},\qquad\lambda_{i,t}=\frac{\sigma_{i,t}}{\sigma_{0,t}},
\]
where $\rho_{i,t}$ is the correlation between the $i$-th asset and
the market and $\sigma_{i,t}$ and $\sigma_{0,t}$ are the volatilities
for the $i$-th asset and the market, respectively. We will estimate
local market betas using local correlation estimators combined with
estimators of relative volatility. Specifically, we compute $P$,
$K$, and $Q_{S}$ and $\lambda$ using a rolling window with 60 minutes
of high-frequency data. We will then investigate how much of the intraday
variation in betas is explained by intraday variation in correlations
and how much can be ascribed to intraday variation in relative volatility. 

We estimate $\rho_{i,t}$ with each of the correlation estimators,
$P,$ $K$, and $Q_{S}$, using a rolling window that spans 60 minutes.
Similarly, we estimate the relative volatility, $\lambda_{i,t}$,
with a subsampled truncated absolute-return estimator, defined
by 
\[
\hat{\lambda}_{i,t}=\frac{\sum_{j\in I_{t}}\left\llbracket \Delta_{\frac{S}{N}}Y_{\frac{j}{N}}\right\rrbracket _{\nu_{y}}}{\sum_{j\in I_{t}}\left\llbracket \Delta_{\frac{S}{N}}X_{\frac{j}{N}}\right\rrbracket _{\nu_{x}}},\qquad\left\llbracket x\right\rrbracket _{\nu}=\begin{cases}
|x| & \text{if }|x|<\nu\\
0 & \text{otherwise}
\end{cases}\qquad I_{t}=[\lfloor tN\rfloor-\tfrac{W}{2}+S,\lfloor tN\rfloor+\tfrac{W}{2}],
\]
where $\nu_{x}$ and $\nu_{y}$ are adaptive thresholds for jump truncation.
These are defined by $\nu_{x}=4\sqrt{\mathrm{BV}_{x}}/n^{0.49}$ and
$\nu_{y}=4\sqrt{\mathrm{BV}_{y}}/n^{0.49}$, where $\mathrm{BV}_{x}=\frac{\pi}{2}\sum^{n}_{j=2}|\Delta_{\delta}X_{j\delta}||\Delta_{\delta}X_{(j-1)\delta}|$
is the jump-robust bipower variation estimator of daily integrated
volatility, see \citet{barndorff-shephard:2004BiPower}, and $\mathrm{BV}_{y}$
is defined analogously. We also estimated relative volatility
using bipower variation with the same thresholds. These
estimates were virtually identical to $\hat{\lambda}_{i,t}$.

In our implementation, we have $N=23,400$, $W=3,600$, and $S=180$.
This results in 3,421 overlapping 3-minute returns within each hour
we use to compute the subsampled quantities.

The product $Q_{S}\hat{\lambda}_{i,t}$ is a local beta under a within-window
approximation. Over the window $I_{t}$, the sum of truncated absolute
returns estimates $\int_{I_{t}}\sigma_{i}(u)\,\mathrm{d}u$ up to a
constant that cancels in the ratio, so $\hat{\lambda}_{i,t}$ estimates
$\int_{I_{t}}\sigma_{i}(u)\,\mathrm{d}u/\int_{I_{t}}\sigma_{0}(u)\,\mathrm{d}u$,
where $0$ indexes the market. Its product with the window's $Q_{S}$
equals the regression beta $\int_{I_{t}}\rho_{i}(u)\sigma_{i}(u)\sigma_{0}(u)\,\mathrm{d}u/\int_{I_{t}}\sigma_{0}^{2}(u)\,\mathrm{d}u$,
and the average of the spot beta $\rho_{i}(u)\sigma_{i}(u)/\sigma_{0}(u)$
over the window, when $\rho_{i}$, $\sigma_{i}$, and $\sigma_{0}$ are
constant within the window, and it is in this sense that we interpret
the product as a local beta. Section \ref{subsec:intcov} makes the
same point for the full trading day, where the approximation is less
accurate than over sixty minutes.

The choice of a 60-minute window reflects the usual trade-off between
localization and sampling error: a shorter window tracks intraday
variation more closely, but is based on fewer observations, whereas
a longer window smooths over the intraday variation we seek to
measure. We explored a range of window lengths and settled on 60 minutes
as the shortest window for which the estimates were sufficiently stable.
This is half the length of the two-hour rolling window used by ATT.
The subsampling underlying $Q_{S}$ utilizes all high-frequency
observations within the window, which makes a one-hour window feasible. The simulation evidence
on this trade-off is in Section \ref{subsec:WindowLength}, and we
report the corresponding empirical estimates with $30$- and $90$-minute
windows in the Supplementary Material, see Figure \ref{fig:WindowSensitivityEmpirical}.

Each window is centered on the time at which the estimate is reported:
the index set $I_{t}$ above collects the right endpoints of the $3$-minute
returns that lie entirely within $[\lfloor tN\rfloor-\tfrac{W}{2},\lfloor tN\rfloor+\tfrac{W}{2}]$.
The endpoints themselves are offset by $S/2$; the midpoints of the
returns are symmetric about $\lfloor tN\rfloor$, and the estimate
labeled $t$ uses the $30$ minutes on either side of $t$. Labeling
a window by its right endpoint instead would shift the intraday profile
backward by half a window length, which in the simulation design of
Section \ref{subsec:WindowLength} amounts to about $-0.03$ in correlation
units at $W=3{,}600$, comparable to the intraday variation we measure.
The centered convention loses half a window at each end of the day,
so the reported estimates run from 10:00 to 15:30.
\begin{figure}[H]
\begin{centering}
\subfloat[Correlation (AAPL,SPY)]{\begin{centering}
\includegraphics[width=0.32\textwidth]{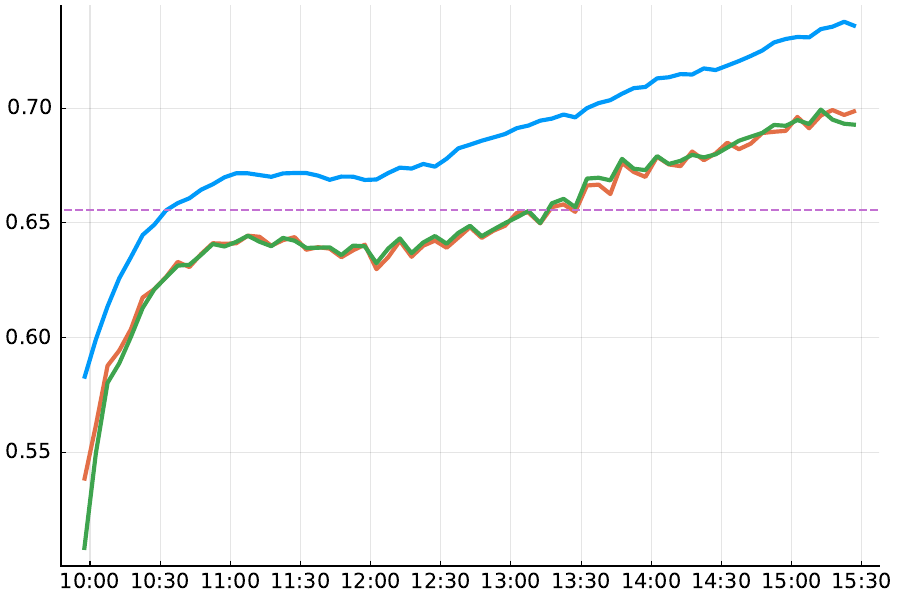}
\par\end{centering}
}\negthinspace{}\subfloat[Volatility ratio (AAPL,SPY)]{\begin{centering}
\includegraphics[width=0.32\textwidth]{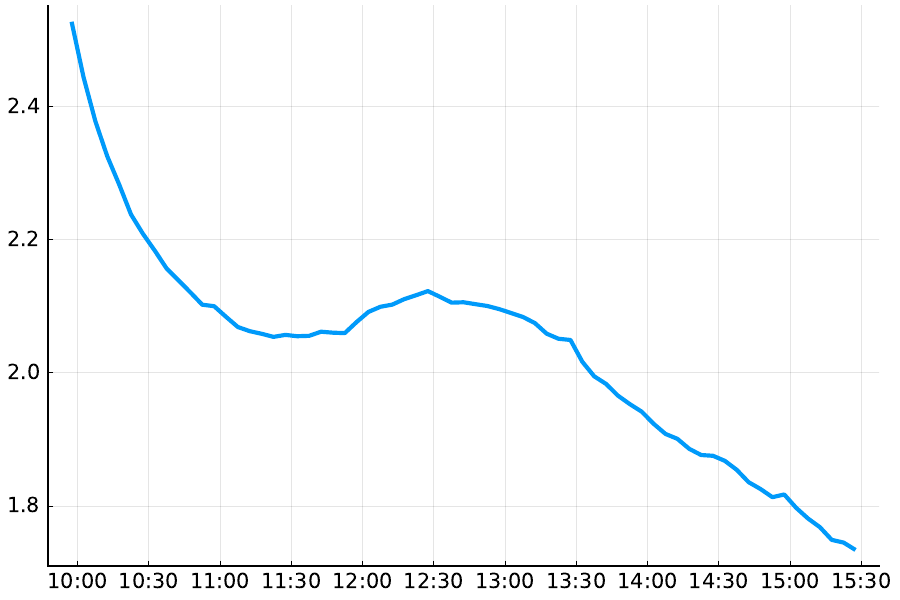}
\par\end{centering}
}\negthinspace{}\subfloat[Market $\beta$ (AAPL,SPY)]{\begin{centering}
\includegraphics[width=0.32\textwidth]{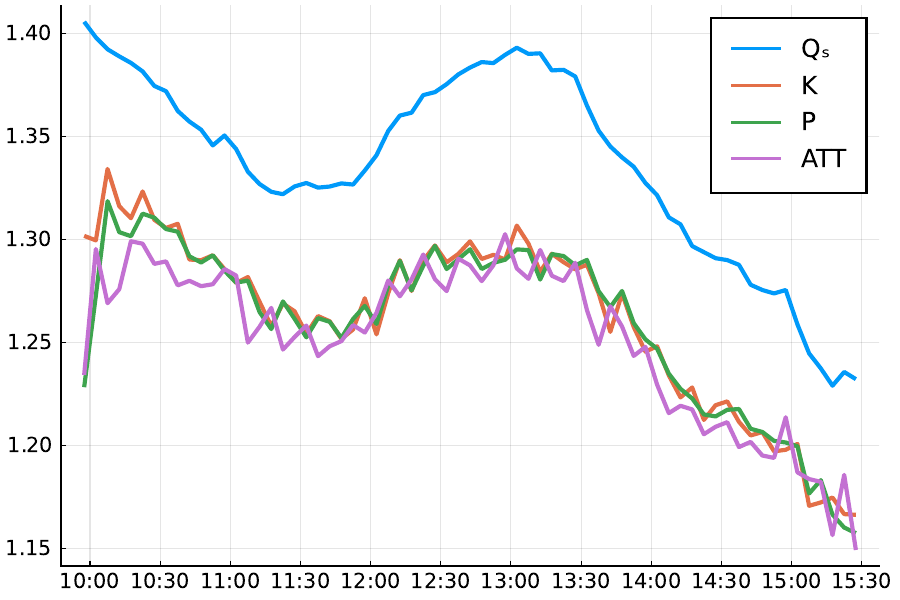}
\par\end{centering}
}
\par\end{centering}
\begin{centering}
\subfloat[Correlation (AMD,SPY)]{\begin{centering}
\includegraphics[width=0.32\textwidth]{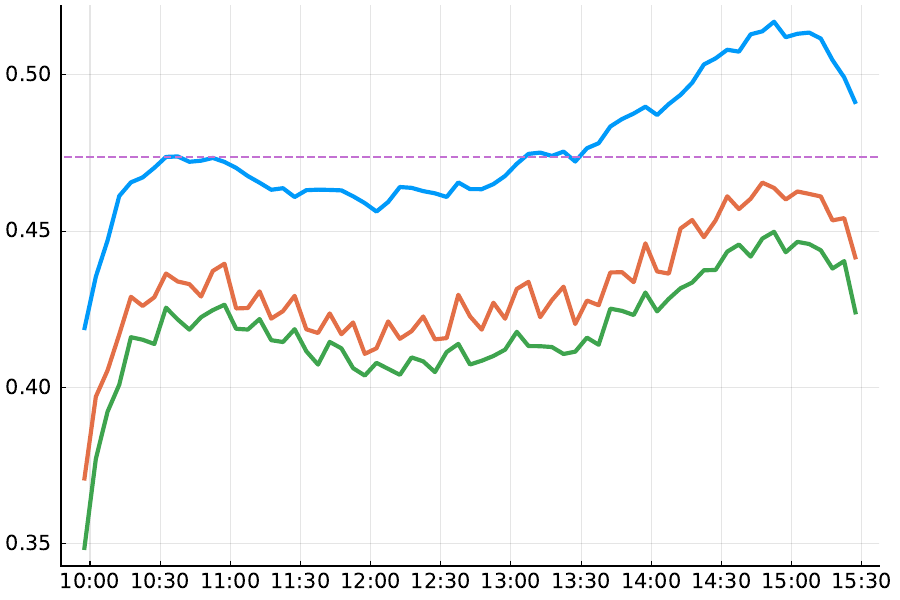}
\par\end{centering}
}\negthinspace{}\subfloat[Volatility ratio (AMD,SPY)]{\begin{centering}
\includegraphics[width=0.32\textwidth]{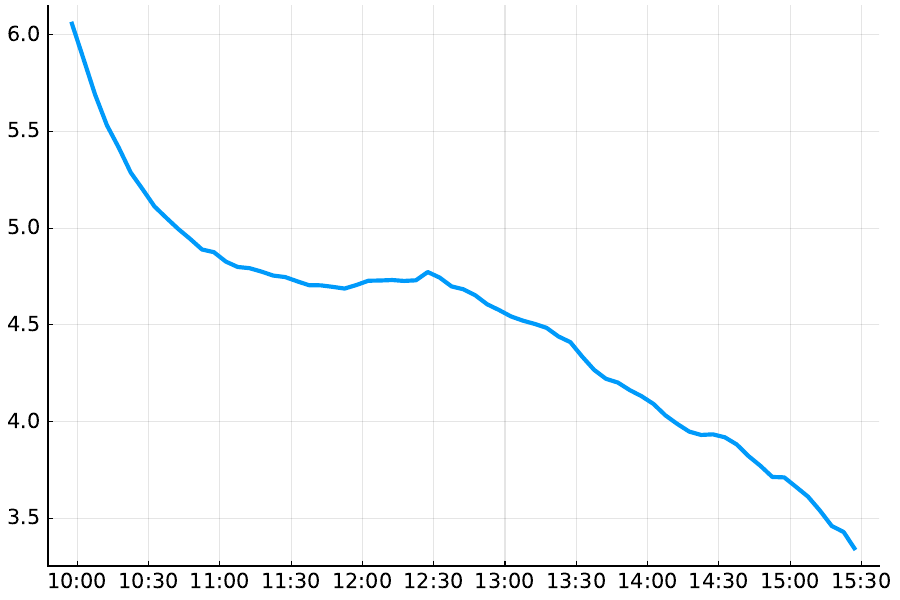}
\par\end{centering}
}\negthinspace{}\subfloat[Market $\beta$ (AMD,SPY)]{\begin{centering}
\includegraphics[width=0.32\textwidth]{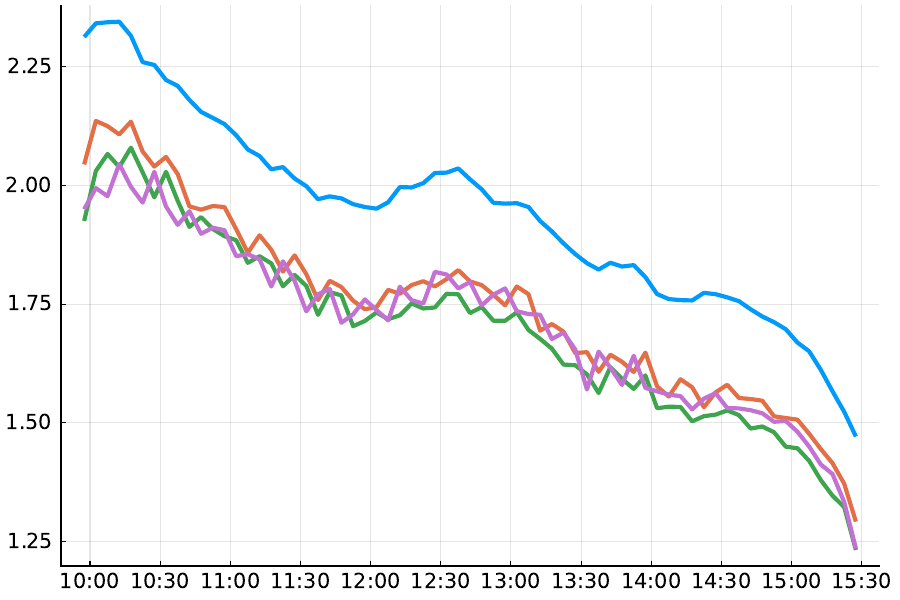}
\par\end{centering}
}
\par\end{centering}
\begin{centering}
\subfloat[Correlation (LYB,SPY)]{\begin{centering}
\includegraphics[width=0.33\textwidth]{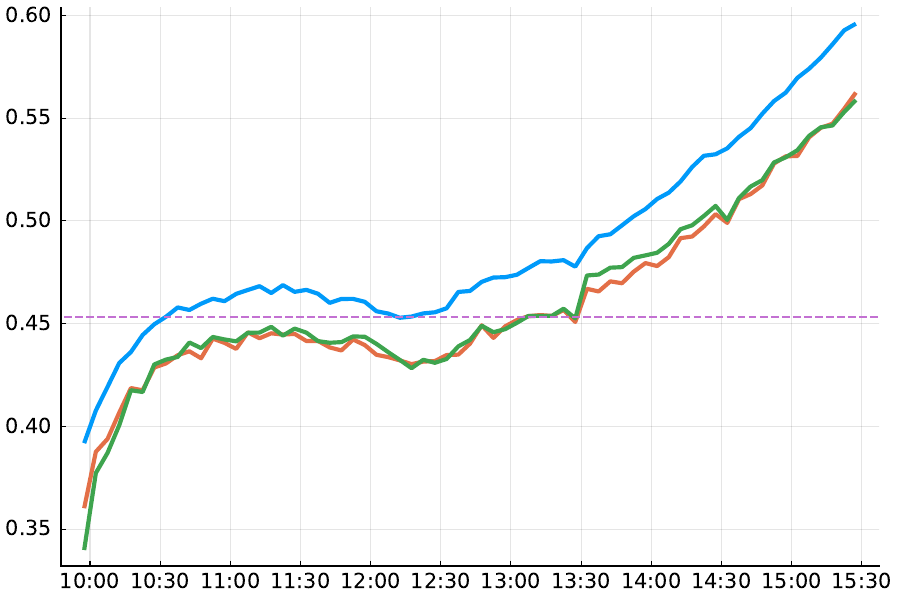}
\par\end{centering}
}\negthinspace{}\subfloat[Volatility ratio (LYB,SPY)]{\begin{centering}
\includegraphics[width=0.32\textwidth]{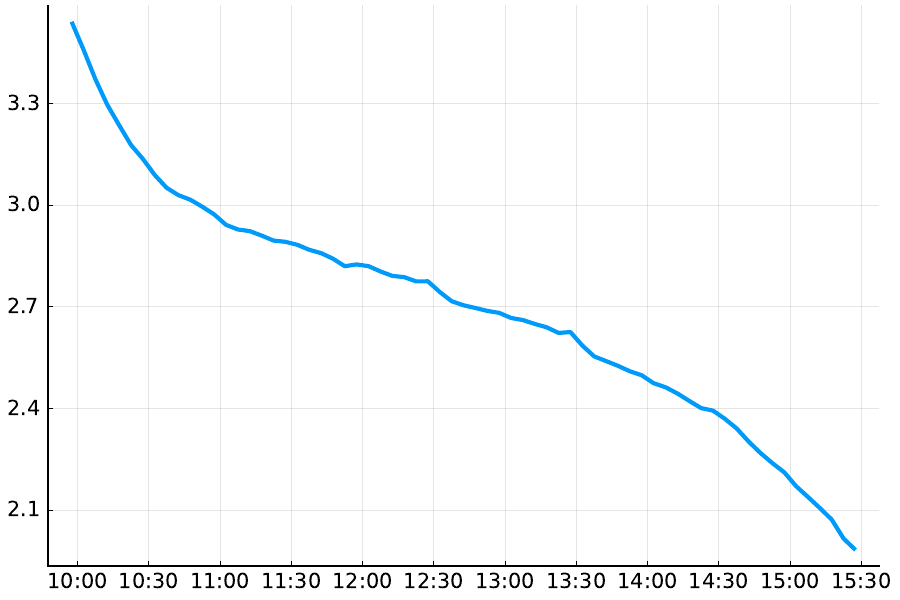}
\par\end{centering}
}\negthinspace{}\subfloat[Market $\beta$ (LYB,SPY)]{\begin{centering}
\includegraphics[width=0.32\textwidth]{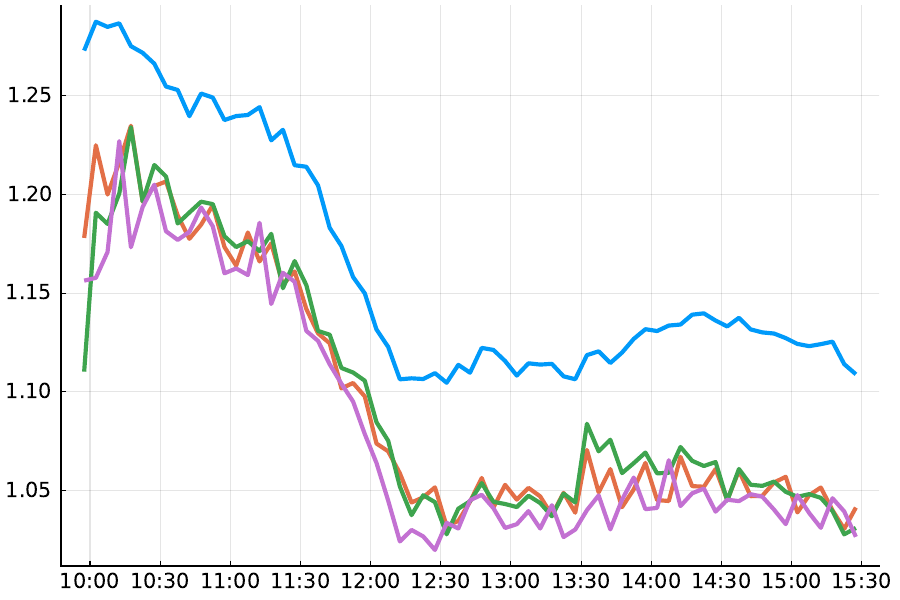}
\par\end{centering}
}
\par\end{centering}
\begin{centering}
\subfloat[Correlation (NEM,SPY)]{\begin{centering}
\includegraphics[width=0.32\textwidth]{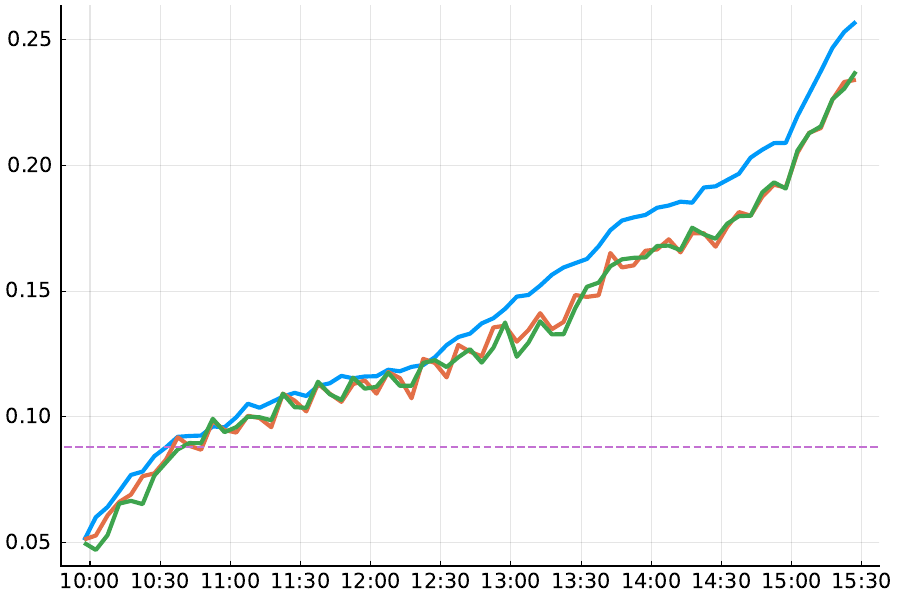}
\par\end{centering}
}\negthinspace{}\subfloat[Volatility ratio (NEM,SPY)]{\begin{centering}
\includegraphics[width=0.32\textwidth]{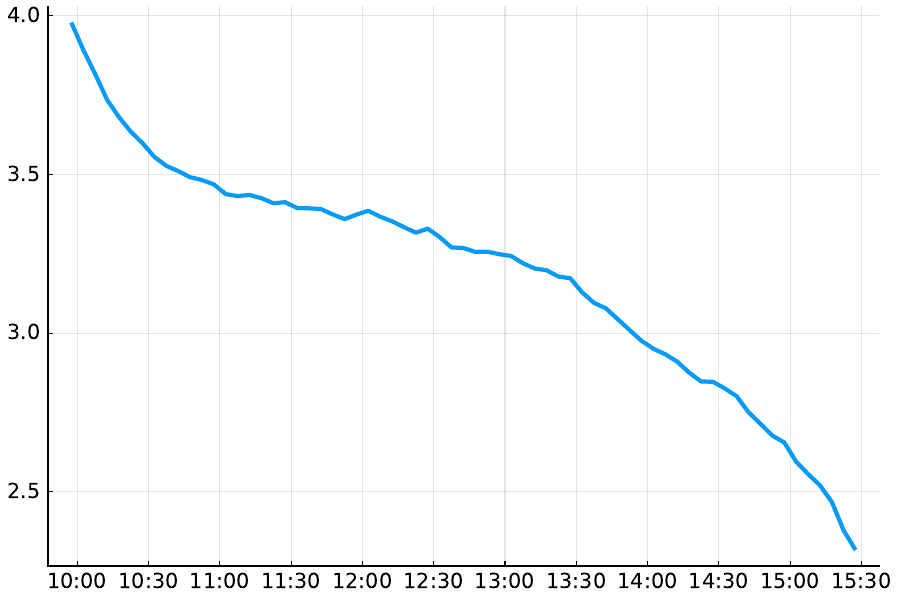}
\par\end{centering}
}\negthinspace{}\subfloat[Market $\beta$ (NEM,SPY)]{\begin{centering}
\includegraphics[width=0.32\textwidth]{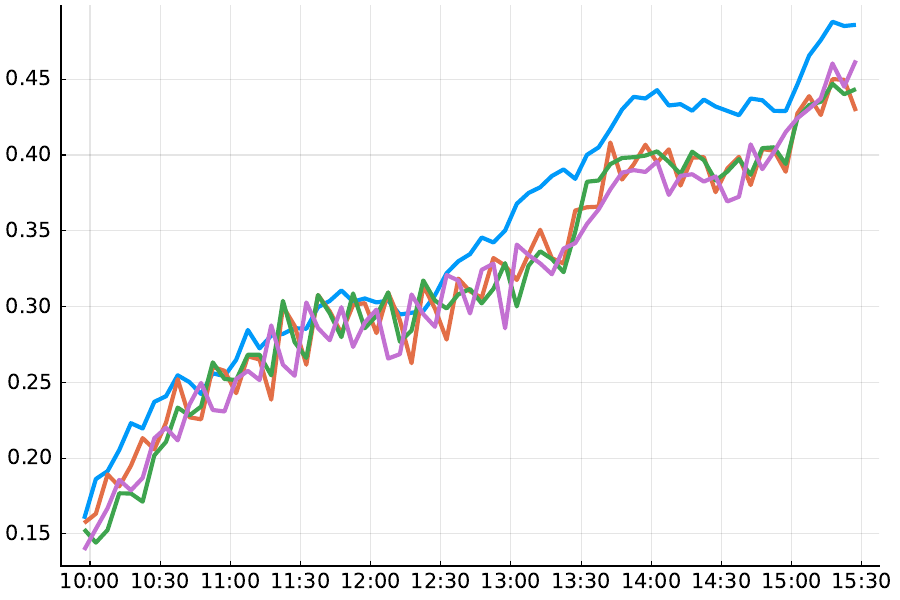}
\par\end{centering}
}
\par\end{centering}
\caption{Intraday correlations (left panels), relative volatilities (middle
panels), and market $\beta$ (right panels) for four assets, AAPL,
AMD, LYB, and NEM, relative to SPY that tracks the S\&P 500 index.
Quantities are estimated with a 60-minute rolling window
centered on the indicated timestamp, i.e., the 30 minutes
on either side of it. Estimates are averaged over the days in the
sample period. The estimates of market
betas using the methodology in ATT are included in the right panels.\label{fig:RhoLambdaBetaIntraday}}
\end{figure}

Intraday estimates of correlations, relative volatilities, and betas
are shown for four assets in Figure \ref{fig:RhoLambdaBetaIntraday}.
All quantities are estimated using a 60-minute rolling window and
averaged over the 1,763 trading days in the sample. The left panels report the intraday correlation
estimates for $P$, $K$, and $Q_{S}$. A horizontal dashed line indicates
the average correlations over the trading day. The middle panels report
the relative volatility as defined by $\hat{\lambda}_{i,t}$ above,
which is multiplied by the three correlation estimates to obtain estimates
of market betas. The multiplication is done day by day: the beta
reported at time $t$ is the average over days of the daily products
$\hat{\rho}_{i,t,d}\hat{\lambda}_{i,t,d}$, not the product of the
averaged correlation in the left panel and the averaged relative volatility
in the middle panel. The two differ by the covariance across days between
correlation and relative volatility, which is negative and not negligible
here; we return to it in Section \ref{subsec:decomposition}. These three
intraday market betas are shown in the right panels along with the
regression-based estimate, based on the same methodology as ATT. The corresponding results for all assets
in the Small Universe are presented in the Supplementary Material,
Figure \ref{fig:RhoLambdaBetasSmallUniverse}. 

The estimates of $Q_{S}$ vary more smoothly over the day than those
of $K$ and $P$. The simulations in Section \ref{sec:Simulation-Study}
show that $Q_{S}$ has a smaller bias and, at the sampling frequencies
used here, a smaller root mean squared error under time-varying volatility,
and the smoother paths are consistent with this. $Q_{S}$ also tends
to be slightly larger than $K$ and $P$, which is the direction of
the bias that Theorem \ref{theo:Consistency} predicts for $K$ and
$P$ when the volatilities are not collinear. These smoother lines for $Q_{S}$ and slightly larger
values carry over to the intraday estimates of market betas. The
lines for the regression-based estimates of intraday betas are also
less smooth than those based on $Q_{S}$. 

We find correlations generally increasing over the day, while relative
volatilities are decreasing. Whether their product, the market beta,
is increasing or decreasing will depend on which component changes
more. Unlike correlations and relative volatility, the paths for
intraday market betas take many different shapes. Some assets have
clearly increasing market betas over the day (e.g., NEM), others have
decreasing market betas (e.g., AMD), and a third group of assets have
market betas that go both up and down (e.g., AAPL), or stay relatively
flat for a large part of the day (e.g., LYB). It is interesting to
compare the market betas for LYB and NEM, which are both Materials
sector stocks, with low trading intensities relative to other stocks
in the Small Universe. Despite these commonalities, the intraday beta
patterns for LYB and NEM are very different. The reason can be found
in their intraday correlations. For LYB the correlation only increases
by about 50\% over the trading day (from about 0.39 to 0.60), whereas
the correlation for NEM increases by about 400\% (from about 0.05
to 0.26). A great variety of shapes for time-varying betas are shown
in Figure \ref{fig:RhoLambdaBetasSmallUniverse} in the Supplementary
Material.

Because $Q_{S}$ discards increment products that are exactly zero,
the intraday profiles could in principle reflect the incidence of
zeros rather than the correlation. Table
\ref{tab:StatisticsSmallUniverse} reports the marginal frequency of
zero returns over the whole day; Table \ref{tab:ZeroProductsByTime}
in the Supplementary Material reports the frequency of zero increment
products with SPY, at the three-minute sampling used here, by hour
of the trading day. The incidence is modest for all but the lowest-priced
stock: over the full day it lies between $4\%$ and $12\%$ for 21 of
the 22 stocks, and is $22\%$ for AMD. Its intraday pattern is hump-shaped,
not monotone:
for every stock it is lowest in the first hour, roughly doubles by
early afternoon, and falls back in the final hour, in line with the
$U$-shape in trading activity. The estimated correlations in Figure
\ref{fig:RhoLambdaBetaIntraday}, by contrast, generally rise over
the day. The profiles therefore do not simply track the frequency
of discarded zeros. The comparison does not rule out a selection effect
operating through the sign composition of the discarded products.
To see what such an effect would involve, treat each observed
zero product as concealing a latent sign of the efficient-price increments,
and write $\tau(t)=(1-p(t))\tau_{\mathrm{keep}}(t)+p(t)\tau_{\mathrm{omit}}(t)$,
where $p(t)$ is the omission frequency and $\tau_{\mathrm{keep}}$
and $\tau_{\mathrm{omit}}$ are the latent concordances among the
retained and the omitted products. The distortion from discarding
the zeros depends on the difference between the two latent concordances
as well as on $p(t)$, and neither latent quantity is observable. Corollary
\ref{cor:nonzero} applies only to the frictionless case.

\subsection{Decomposing Intraday Variation in Market Betas\label{subsec:decomposition}}

The identity $\beta_{i,t}=\rho_{i,t}\lambda_{i,t}$ holds day by day,
but the quantities in Figure \ref{fig:RhoLambdaBetaIntraday} are averages
over days, and averaging does not commute with multiplication. Let
$\bar{\rho}_{i,t}$, $\bar{\lambda}_{i,t}$ and $\bar{\beta}_{i,t}$
denote the averages over the $D$ days of the daily estimates $\hat{\rho}_{i,t,d}$,
$\hat{\lambda}_{i,t,d}$ and $\hat{\beta}_{i,t,d}=\hat{\rho}_{i,t,d}\hat{\lambda}_{i,t,d}$.
Then
\[
\bar{\beta}_{i,t}=\bar{\rho}_{i,t}\bar{\lambda}_{i,t}+\mathrm{cov}_{d}(\hat{\rho}_{i,t,d},\hat{\lambda}_{i,t,d})=\bar{\rho}_{i,t}\bar{\lambda}_{i,t}\bar{\gamma}_{i,t},\qquad\bar{\gamma}_{i,t}=1+\frac{\mathrm{cov}_{d}(\hat{\rho}_{i,t,d},\hat{\lambda}_{i,t,d})}{\bar{\rho}_{i,t}\bar{\lambda}_{i,t}},
\]
where $\mathrm{cov}_{d}$ is the sample covariance across days (see
Section \ref{subsec:decomp-supp} in the Supplementary Material). The
averages $\bar{\beta}_{i,t}$, $\bar{\rho}_{i,t}$ and $\bar{\lambda}_{i,t}$
are positive for every stock and time of day in our sample, so $\bar{\gamma}_{i,t}>0$
and we can factorize the logarithm of the average intraday beta as
\[
\log\bar{\beta}_{i,t}=\log\bar{\rho}_{i,t}+\log\bar{\lambda}_{i,t}+\log\bar{\gamma}_{i,t}.
\]
The third term is not negligible. Correlation and relative volatility
are negatively correlated across days for every stock in the Large
Universe, with cross-day correlations between $-0.10$ and $-0.57$,
so that $\bar{\gamma}_{i,t}$ is below one: it averages $0.89$ across
stocks in the first hour of trading and $0.93$ in the last, and is as
low as $0.74$ for NEM. Days on which a stock is highly correlated with
the market tend to be days on which its volatility relative to the
market is low, consistent with market-wide volatility
episodes raising correlations and market volatility by more than
they raise individual-stock volatility. Multiplying the
averaged correlation by the averaged relative volatility would therefore
overstate the average beta by about $10\%$, which is why the betas in
Figure \ref{fig:RhoLambdaBetaIntraday} are computed from the daily
products.

We use the three-term decomposition to investigate how much of the
intraday variation in market betas can be ascribed to changes in correlations,
changes in relative volatilities, and changes in the covariance term.
For this purpose, we expand the analysis to include all
assets in the Large Universe. For each asset we define the intraday
change as the logarithmic difference between the estimate for the last
hour of trading and that for the first hour, $\Delta\beta_{i}=\log\bar{\beta}_{i,\mathrm{last}}-\log\bar{\beta}_{i,\mathrm{first}}$,
and similarly $\Delta\rho_{i}$, $\Delta\lambda_{i}$ and $\Delta\gamma_{i}$,
so that $\Delta\beta_{i}=\Delta\rho_{i}+\Delta\lambda_{i}+\Delta\gamma_{i}$
exactly.
\begin{figure}[H]
\begin{centering}
\subfloat[]{\begin{centering}
\includegraphics[width=0.49\textwidth]{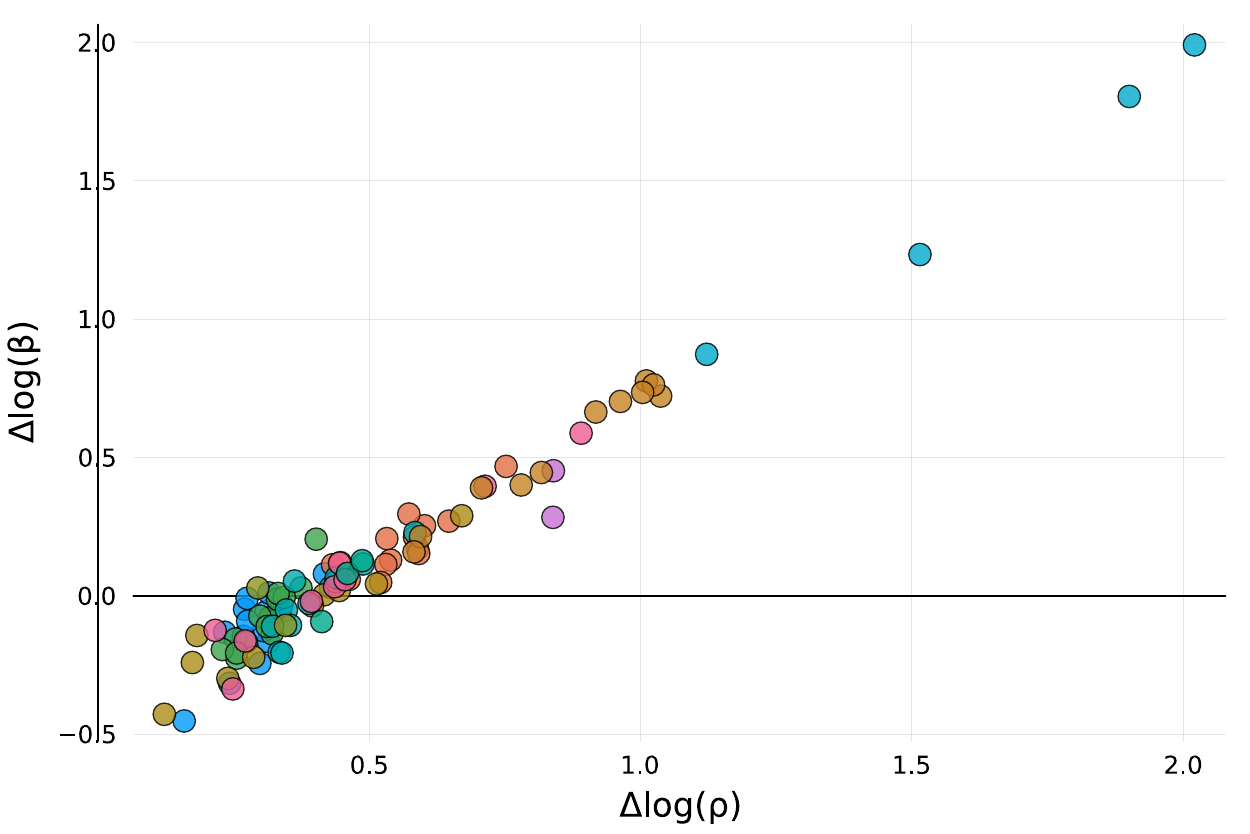}
\par\end{centering}
}\negthinspace{}\subfloat[]{\begin{centering}
\includegraphics[width=0.49\textwidth]{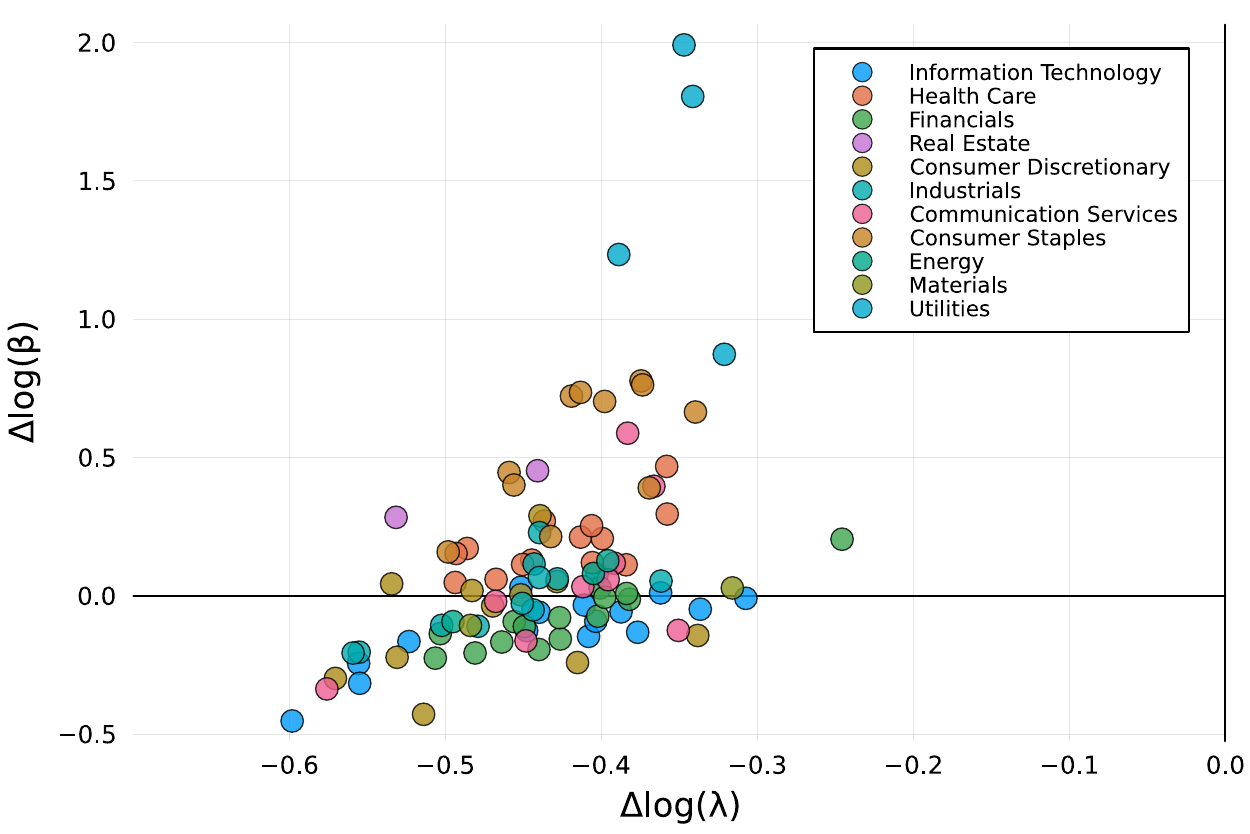}
\par\end{centering}
}
\par\end{centering}
\begin{centering}
\subfloat[]{\begin{centering}
\includegraphics[width=0.49\textwidth]{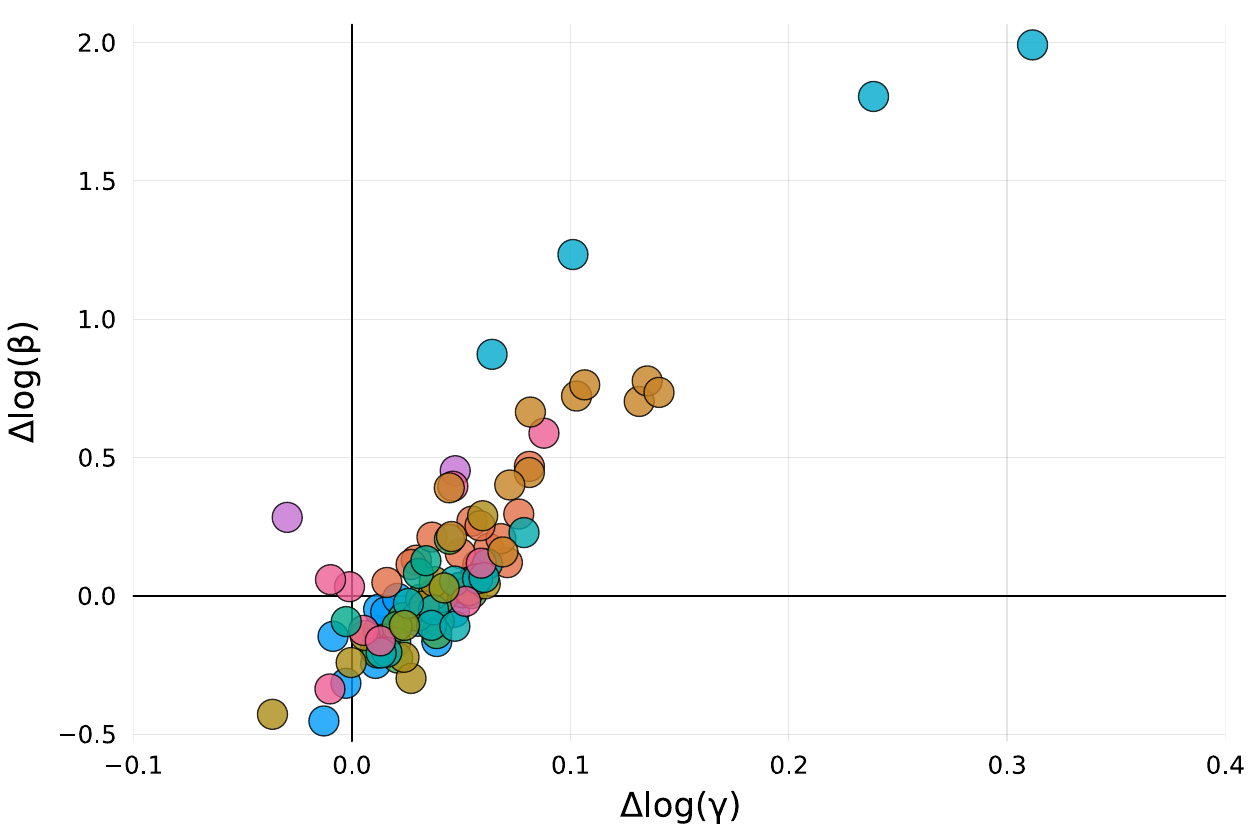}
\par\end{centering}
}
\par\end{centering}
\caption{Scatterplot of changes in intraday betas plotted against changes in
correlations (a), changes in relative volatility (b), and changes in
the covariance term $\bar{\gamma}$ (c), for the Large Universe. Changes
are logarithmic differences between the last-hour and first-hour estimates
of the trading day, so that the three horizontal coordinates sum to
the vertical coordinate.\label{fig:DeltaRhoDeltaLambda-vs-DeltaBeta}}
\end{figure}

In Figure \ref{fig:DeltaRhoDeltaLambda-vs-DeltaBeta} we have plotted
the changes in intraday market betas against the changes in intraday
correlations, in relative volatility, and in the covariance term, for
all assets in the Large Universe. We use color-coding to indicate the
GICS industry sectors for each of the assets.

The averages across the 98 stocks are $\Delta\rho=0.50$, $\Delta\lambda=-0.43$,
$\Delta\gamma=0.05$ and $\Delta\beta=0.12$: the typical stock sees
its correlation with the market rise by about $65\%$ from the first
hour to the last, and its relative volatility fall by about $35\%$,
and the two nearly cancel in the beta. Both terms therefore matter
for the average intraday change. The cross-sectional dispersion is
dominated by the correlation term. The standard deviation of $\Delta\rho$
across stocks is $0.33$, against $0.07$ for $\Delta\lambda$ and $0.05$
for $\Delta\gamma$, and this is visible in the ranges of the horizontal
axes in Figure \ref{fig:DeltaRhoDeltaLambda-vs-DeltaBeta}. Since
$\Delta\beta_{i}=\Delta\rho_{i}+\Delta\lambda_{i}+\Delta\gamma_{i}$,
the cross-sectional variance of $\Delta\beta$ is the sum of its covariances
with the three terms, and the corresponding shares are $81\%$ for
correlation, $8\%$ for relative volatility and $11\%$ for the covariance
term. Panel (a) accordingly shows a nearly linear relationship, with
a correlation of $0.98$ between $\Delta\beta$ and $\Delta\rho$,
whereas panel (b) shows a weak one. Panel (c) shows that $\Delta\gamma$,
though small, is not orthogonal to $\Delta\beta$: the stocks whose
correlation with the market rises most are also the stocks for which
the cross-day covariance term rises most, which is why its share exceeds
that of relative volatility despite its smaller dispersion. The differences
across stocks in the intraday
movement of their market betas are almost entirely differences in
the intraday movement of their correlations with the market, because
the decline in relative volatility is common to all stocks. 

In Figure \ref{fig:DeltaRhoDeltaLambda-vs-DeltaBeta} the changes
in correlation cluster by sector, whereas the changes in relative
volatility vary within every sector.

In the Supplementary Material, Figure \ref{fig:DeltaRhoLambdaBeta-vs-BetaSizeBTM},
we have explored the intraday variation in greater detail. For instance,
we plot changes in intraday market correlations, relative volatility,
and market betas against Fama--French-type variables. We do not detect
a strong association with key characteristics such as market capitalization
and book-to-market ratios. 

In Figure \ref{fig:MarketBetaScatterPlots}, we present scatterplots
of the intraday market betas against conventional market betas, which
are computed from daily returns. The left panel plots the first-hour market beta
and the right panel the last-hour market beta, each against the low-frequency
market beta. Not surprisingly, we find a strong relationship between
the low-frequency market beta, estimated from daily returns, and the
intraday market betas. The scatterplots in Figure \ref{fig:MarketBetaScatterPlots}
corroborate findings in ATT, who found market betas to be less dispersed
at the end of the day than at the beginning of the day.
\begin{figure}[H]
\begin{centering}
\subfloat[First hour]{\begin{centering}
\includegraphics[width=0.49\textwidth]{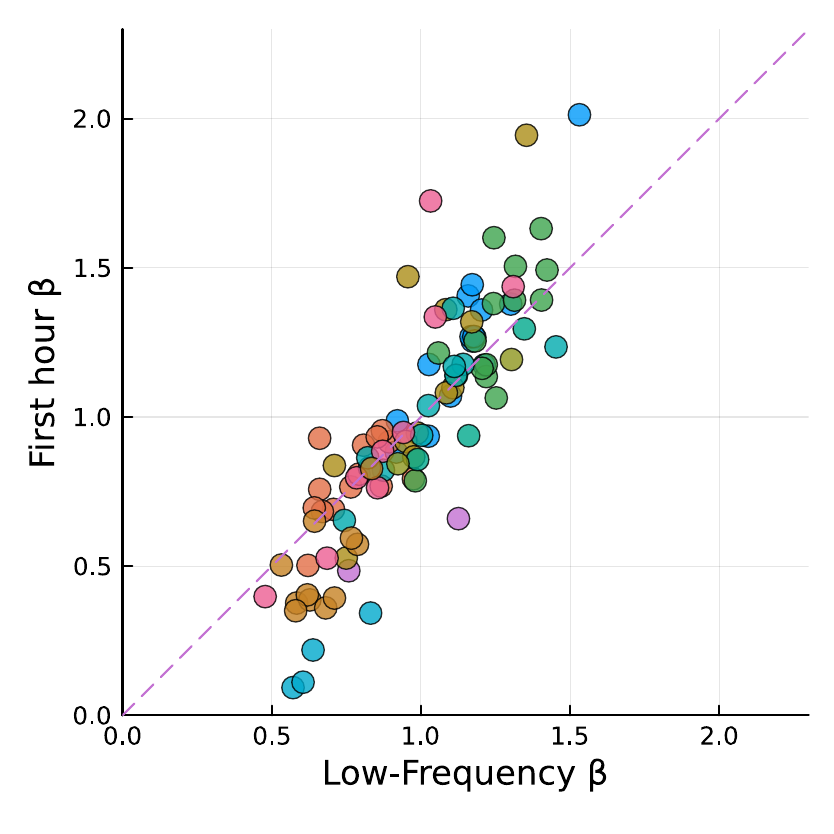}
\par\end{centering}
}\negthinspace{}\subfloat[Last hour]{\begin{centering}
\includegraphics[width=0.49\textwidth]{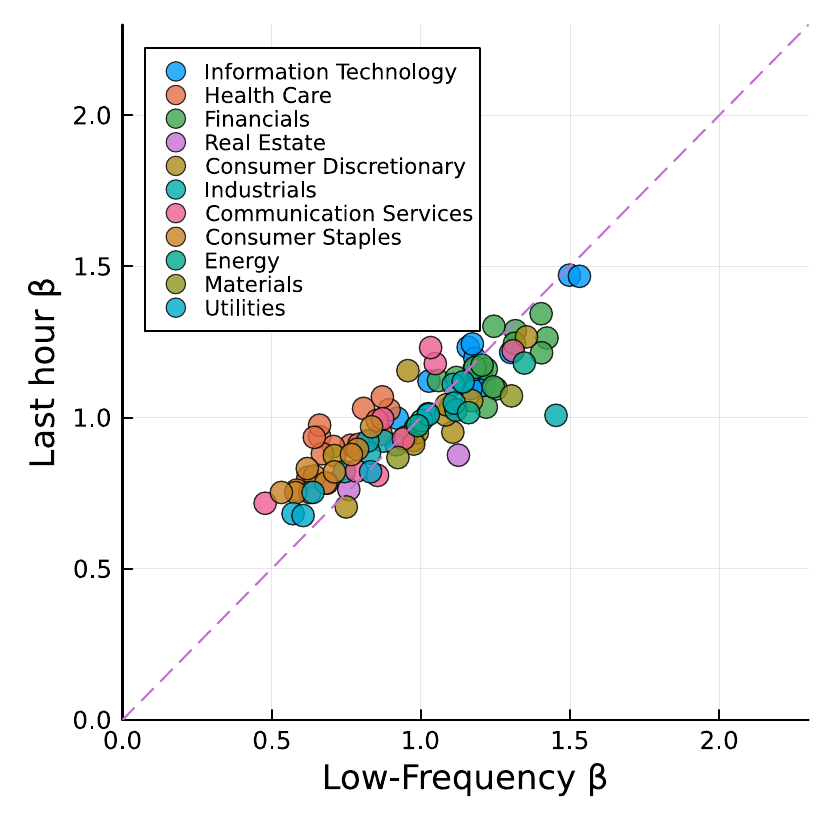}
\par\end{centering}
}
\par\end{centering}
\caption{Market beta for the first and last hour plotted against the conventional
market beta, which is computed from daily returns.\label{fig:MarketBetaScatterPlots}}
\end{figure}

\subsection{Diurnal Volatility Patterns: The Market Smiles and Individual Stocks
Smirk}

The intraday variation in correlation and relative volatility has implications
for the distribution of volatility over the trading day. \citet{Andersen1997}
estimated diurnal patterns in volatility using absolute five-minute
returns and found the S\&P 500 index to have a $U$-shape. For individual
stocks, \citet{ChristensenHounyoPodolskij:2018} reported an asymmetric pattern,
with higher volatility at the beginning of the trading day. They referred to this as a reverse $J$-shape in
volatility. 

The difference between the diurnal pattern of the market and those
of individual stocks is consistent with the declining relative volatilities:
with the market's pattern approximately $U$-shaped, a declining ratio
of stock to market volatility lowers the afternoon part of a stock's
profile relative to the morning, which is the asymmetry of the reverse
$J$-shape.
\begin{figure}[tbh]
\begin{centering}
\subfloat[All days (1,763)]{\begin{centering}
\includegraphics[width=0.48\textwidth]{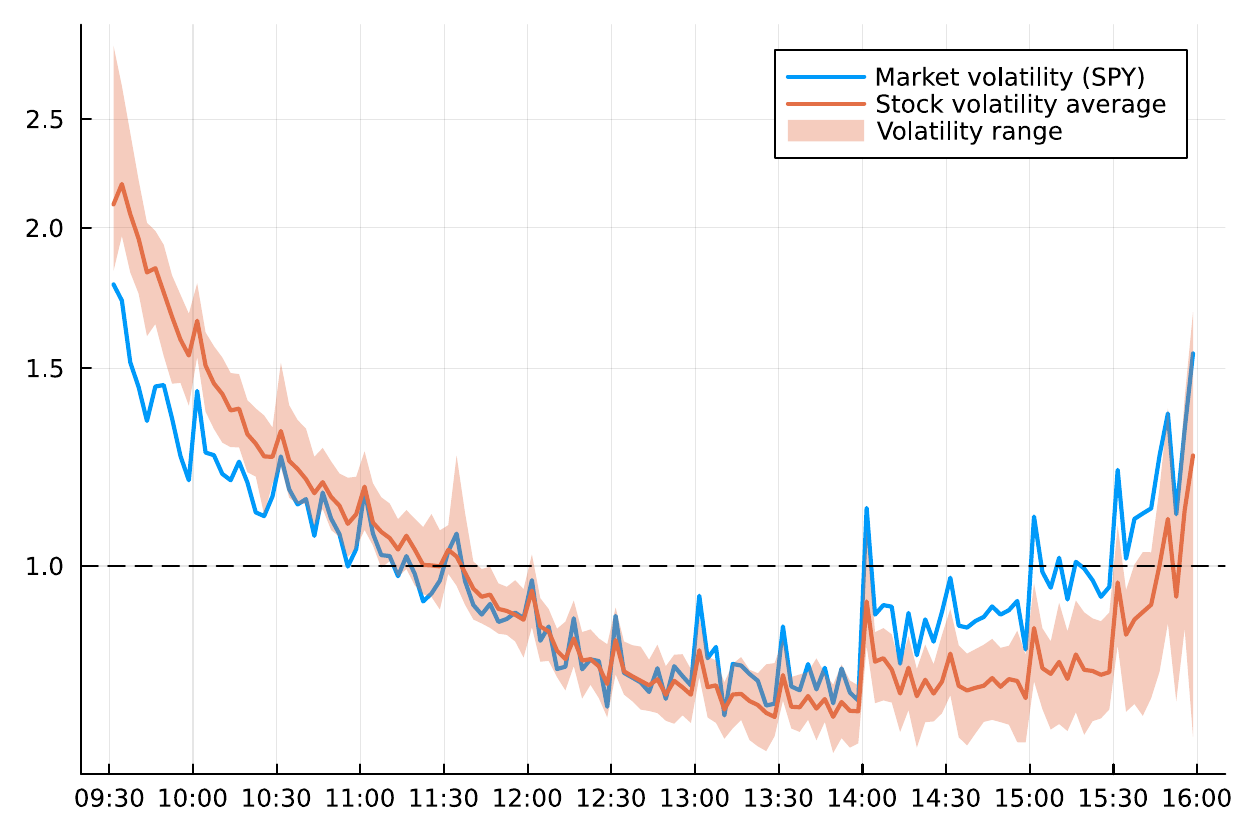}
\par\end{centering}

}\subfloat[Non-FOMC days (1,651 days)]{\begin{centering}
\includegraphics[width=0.48\textwidth]{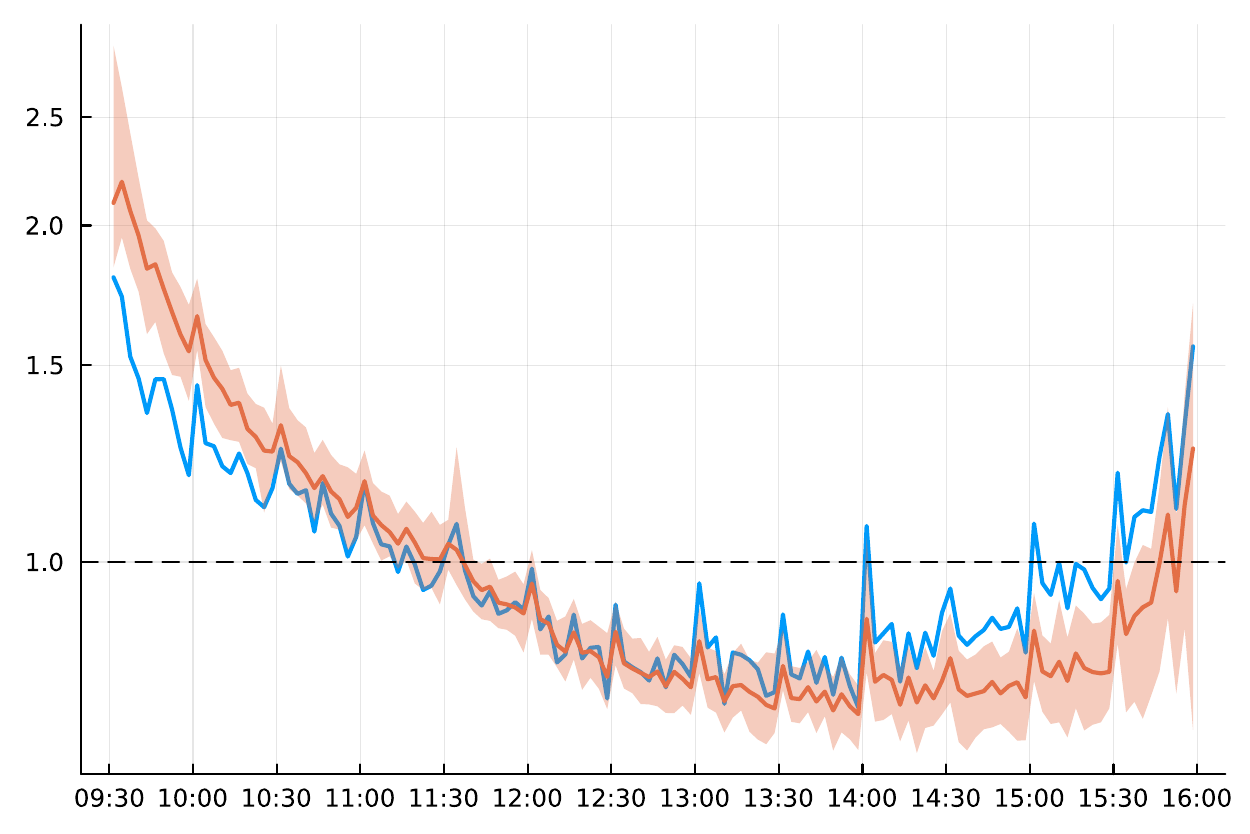}
\par\end{centering}
}
\par\end{centering}
\begin{centering}
\subfloat[FOMC Statements (56 days)]{\begin{centering}
\includegraphics[width=0.48\textwidth]{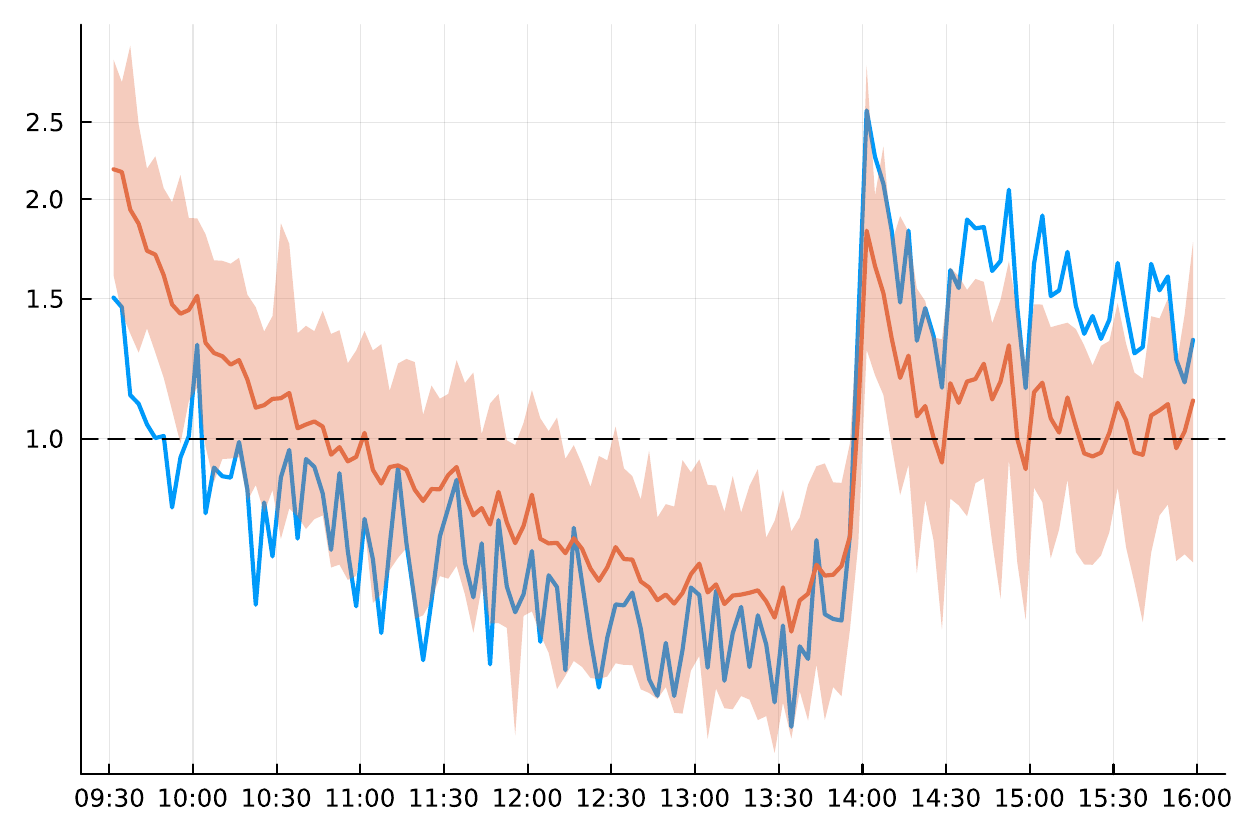}
\par\end{centering}
}\subfloat[FOMC Minutes (56 days)]{\begin{centering}
\includegraphics[width=0.48\textwidth]{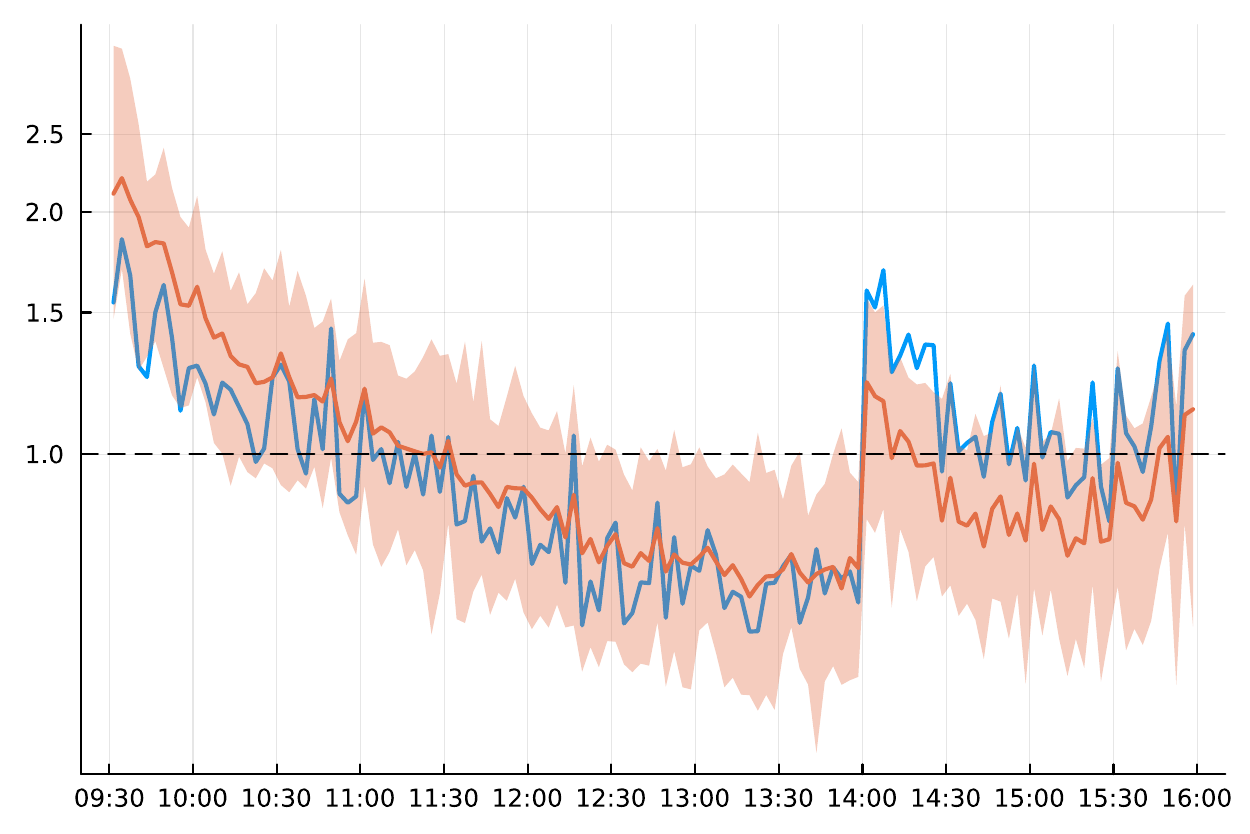}
\par\end{centering}
}
\par\end{centering}
\caption{Distribution of volatility for the market and 98 individual stocks
(Large Universe) over the trading day for all days during the years
2015-2021, and partitioned by FOMC and non-FOMC days in the sample
period.\label{fig:VolatilityDistribution}}
\end{figure}

We divide the period from 9:30 AM to 4:00 PM into 130 non-overlapping
3-minute intervals. For each asset, and each 3-minute interval, we
take the average absolute return over all trading days, using the
same truncation method to estimate relative volatility. For
each asset, we normalize the 130 estimates by dividing by their sample
average, such that one corresponds to average volatility. The estimated
distributions of volatility over the trading day are shown in panel (a)
of Figure \ref{fig:VolatilityDistribution}. The blue line shows
the distribution for the market (SPY) and the red line represents the average
distribution for the 98 stocks in the Large Universe. The shaded region
shows the full range of distributions for all 98 stocks, which is
very narrow. This shows that the diurnal patterns in intraday
volatility are very similar across all stocks. 

The intraday spikes in volatility at the top of each hour and, to
a lesser extent, every 15 minutes, are also very pronounced in Figure
\ref{fig:VolatilityDistribution}.\footnote{The recurrent spikes in volatility during the minutes of the hour
may also, in part, be driven by trading patterns. Even more pronounced
volatility patterns are observed in cryptocurrency markets, see \citet{HansenKimKimbrough:2024},
and it seems implausible that news announcements are the main drivers
of patterns in these markets. This seems implausible in part because
the patterns recur within every hour of the week.} Some of these spikes are driven by news announcements with the largest
spike at 14:00 coinciding with Federal Open Market Committee (FOMC)
announcements. We partitioned the days into days with FOMC statements
and days on which the minutes of FOMC meetings are released (typically
with a three-week delay) and non-FOMC days. These diurnal patterns for
each of these types of days are presented in panels (b), (c), and
(d). The spike at 14:00 is substantially larger on FOMC days than
on non-FOMC days; however, the pattern for non-FOMC days in panel (b)
is quite similar to the ``all days'' pattern in panel (a), which
shows that FOMC days cannot explain the large spike at 14:00. A related
comparison of FOMC and non-FOMC days was made in \citet{BibingerHautschMalecReiss:2019},
who compared intraday spot correlations, rather than volatilities,
on days with and without FOMC announcements, see their Figures 7 and
8, and documented pronounced changes in spot covariation around the
announcement time.

The irregular diurnal volatility patterns are a further reason to
use an estimator that is consistent under time-varying volatility.

\section{Concluding Remarks}

In this paper, we have introduced a robust correlation estimator
for high-frequency financial data. We have shown that the sample correlation, $P$, and Kendall's
tau, $K$, are generally inconsistent for the constant instantaneous
correlation under time-varying volatility, while the quadrant estimator
is robust to time-varying volatility. The subsampled
quadrant estimator, $Q_{S}$, inherits the consistency of the  quadrant
and is more efficient, because it uses all high-frequency observations. The simulations document the finite-sample bias and root mean squared
error of the estimators under time-varying volatility and market microstructure
noise, and the empirical analysis, based on seven years of high-frequency
data for about 100 securities, documents the intraday patterns in
correlations and market betas.

The empirical analysis documents how market betas vary within the
trading day. Market betas can be expressed
as the product of its correlation with market returns and its relative
volatility. We have documented that relative volatility declines over the day
by a similar amount for all stocks, so that the differences across
stocks in the intraday movement of their market betas are mainly driven
by differences in the intraday movement of their correlations with the market.

$Q_{S}$ is designed for high-frequency financial data, where volatility
varies within the day, but it may also be useful for other time series
with time-varying volatility or pronounced outliers and noise. The $Q_{S}$ estimator might also be useful for nonparametric
estimation of the leverage effect, as analyzed in \citet{KalninaXiu:2017}.
There are several ways the $Q_{S}$ estimator could be extended and
possibly improved. For instance, there might be more efficient ways
to handle zero returns, such as distinguishing between cases where
both returns are zero and cases where only one of the returns is zero.
A multivariate version of the $Q_{S}$ estimator would be interesting
to explore. Constructing a correlation matrix from pairwise correlation
estimates can result in an invalid correlation matrix, because it
need not be positive semidefinite. A subsequent projection onto
the set of correlation matrices, which restores positive semidefiniteness,
might be needed, using, for example, the projection methods of \citet{Higham:2002}
and \citet{QiSun:2006}.

\appendix
\setcounter{equation}{0}\renewcommand{\theequation}{A.\arabic{equation}}
\setcounter{lem}{0}\renewcommand{\thelem}{A.\arabic{lem}}
\setcounter{prop}{0}\renewcommand{\theprop}{A.\arabic{prop}}
\setcounter{cor}{0}\renewcommand{\thecor}{A.\arabic{cor}}

\section{Appendix of Proofs}
\begin{prop}[Greiner]
\label{prop:Greiner}Suppose that $(X,Y)$ is elliptically distributed
with location parameter $\mu=(\mu_{1},\mu_{2})^{\prime}$ and dispersion
matrix $\Sigma$ with $\Sigma_{11}>0$ and $\Sigma_{22}>0$, and that
$\Pr[(X,Y)=\mu]=0$. Let $\rho=\Sigma_{12}/\sqrt{\Sigma_{11}\Sigma_{22}}$
and let $\tau=\Pr[Z>0]-\Pr[Z<0]$ with $Z=(X-\mu_{1})(Y-\mu_{2})$.
Then $\tau=\tfrac{2}{\pi}\arcsin\rho$, so that Greiner's identity
(\ref{eq:Greiner}) holds.
\end{prop}
\begin{rem}
The marginal distributions of an elliptical distribution are symmetric
about $\mu_{1}$ and $\mu_{2}$, so $\mu_{i}$ is a median of the $i$-th
coordinate and $Z$ is the product of the recentered variables in the
definition of $\tau$ in Section \ref{subsec:classical}. No moment
condition is needed: $\rho$ is defined from the dispersion matrix, and
it coincides with the Pearson correlation whenever second moments exist.
The condition $\Pr[(X,Y)=\mu]=0$ excludes an atom at the center, which
would put positive mass on $Z=0$. The result is anticipated, without
proof, by a footnote to \citet[\S27]{Sheppard:1899}, who observes that
his quadrant formula ``obviously applies in any case in which the
correlation-solid is a solid of revolution.''
\end{rem}
\begin{proof}
An elliptical distribution has the stochastic representation
\[
\left(\begin{array}{c}
X\\
Y
\end{array}\right)=\mu+VAU,
\]
where $A$ is a $2\times2$ matrix with $AA^{\prime}=\Sigma$, $U=(\cos\theta,\sin\theta)^{\prime}$
with $\theta\sim\mathrm{Uniform}[-\pi,\pi]$, and $V\geq0$ is a scalar
random variable independent of $\theta$, see \citet{CambanisHuangSimons:1981}.
The condition $\Pr[(X,Y)=\mu]=0$ is equivalent to $\Pr[V=0]=0$. Let
$a_{1}^{\prime}$ and $a_{2}^{\prime}$ denote the rows of $A$, so that
$\Sigma_{11}=\Vert a_{1}\Vert^{2}$, $\Sigma_{22}=\Vert a_{2}\Vert^{2}$
and $\Sigma_{12}=a_{1}^{\prime}a_{2}$. Both rows are non-zero, and
we write them in polar form, $a_{i}=\Vert a_{i}\Vert(\cos\alpha_{i},\sin\alpha_{i})^{\prime}$,
$i=1,2$. Then $\rho=a_{1}^{\prime}a_{2}/(\Vert a_{1}\Vert\Vert a_{2}\Vert)=\cos(\alpha_{1}-\alpha_{2})=\cos\delta$,
where $\delta\in[0,\pi]$ is the angle between the two rows, and
\[
X-\mu_{1}=V\Vert a_{1}\Vert\cos(\theta-\alpha_{1}),\qquad Y-\mu_{2}=V\Vert a_{2}\Vert\cos(\theta-\alpha_{2}).
\]
Since $V>0$ almost surely and $\theta$ has a continuous distribution,
$Z\neq0$ almost surely and $\mathrm{sgn}(Z)=\mathrm{sgn}\{\cos(\theta-\alpha_{1})\cos(\theta-\alpha_{2})\}$.
Now $\cos(\theta-\alpha_{i})>0$ if and only if $\theta$ lies in the
open arc of length $\pi$ centered at $\alpha_{i}$ (modulo $2\pi$). Two
such arcs, centered at $\alpha_{1}$ and $\alpha_{2}$, intersect in an
arc of length $\pi-\delta$, and the same is true of the two complementary
arcs on which the cosines are negative. Since $\theta$ is uniform on
the circle,
\[
\Pr[Z>0]=\frac{2(\pi-\delta)}{2\pi}=1-\frac{\delta}{\pi},\qquad\Pr[Z<0]=\frac{\delta}{\pi},
\]
so that $\tau=1-2\delta/\pi$. Finally, $\rho=\cos\delta=\sin(\tfrac{\pi}{2}-\delta)$
with $\tfrac{\pi}{2}-\delta\in[-\tfrac{\pi}{2},\tfrac{\pi}{2}]$, hence
$\arcsin\rho=\tfrac{\pi}{2}-\delta$ and $\tau=\tfrac{2}{\pi}\arcsin\rho$.
\end{proof}

\begin{lem}
\label{lem:SignOf4}Let $(X,Y,\tilde{X},\tilde{Y})$ be normally distributed
with mean zero and correlation matrix
\[
C(\rho,\omega)=\left[\begin{array}{cccc}
1 & \rho & \omega & \omega\rho\\
\rho & 1 & \omega\rho & \omega\\
\omega & \omega\rho & 1 & \rho\\
\omega\rho & \omega & \rho & 1
\end{array}\right],
\]
and define $Z=1_{\{XY>0\}}$ and $\tilde{Z}=1_{\{\tilde{X}\tilde{Y}>0\}}$.
Then
\[
\mathrm{cov}(Z,\tilde{Z})=\Gamma(\rho,\omega)=\frac{\arcsin^{2}(\omega)-\arcsin^{2}(\omega\rho)}{\pi^{2}}.
\]
 
\end{lem}
\begin{proof}
Let $(\oplus\oplus\ominus\ominus)$ denote the orthant event, $(X>0,Y>0,-\tilde{X}>0,-\tilde{Y}>0)$,
and define events with other combinations of the signs similarly.
For a quadrivariate variable with covariance $C(\rho,\omega)$, we
define the orthant probability 
\[
G(\rho,\omega)=\Pr(\oplus\oplus\oplus\oplus),
\]
and note that $G(\rho,\omega)=\Pr(\ominus\ominus\ominus\ominus)$,
and since $\mathrm{corr}((X,Y,-\tilde{X},-\tilde{Y})^{\prime})=C(\rho,-\omega)$
we have $\Pr(\oplus\oplus\ominus\ominus)=\Pr(\ominus\ominus\oplus\oplus)=G(\rho,-\omega)$. 

Consider the two terms of $\mathrm{cov}(Z,\tilde{Z})=\mathbb{E}(Z\tilde{Z})-\mathbb{E}(Z)\mathbb{E}(\tilde{Z})$.
For the first term we have
\begin{eqnarray*}
\mathbb{E}(Z\tilde{Z})=\Pr(XY>0,\tilde{X}\tilde{Y}>0) & = & \Pr((\oplus\oplus\oplus\oplus)\cup(\oplus\oplus\ominus\ominus)\cup(\ominus\ominus\oplus\oplus)\cup(\ominus\ominus\ominus\ominus))\\
 & = & 2\left[G(\rho,\omega)+G(\rho,-\omega)\right],
\end{eqnarray*}
and for the second term we use that $\mathbb{E}(Z)=\mathbb{E}(\tilde{Z})$
can be obtained by setting $\omega=1$, such that
\[
\mathrm{cov}(Z,\tilde{Z})=2\left[G(\rho,\omega)+G(\rho,-\omega)\right]-4\left[G(\rho,1)+G(\rho,-1)\right]^{2}.
\]
From \citet{Cheng:1969}, we have the expression
\[
G(\rho,\omega)=\tfrac{1}{16}+\tfrac{\arcsin\rho+\arcsin\omega+\arcsin(\omega\rho)}{4\pi}+\tfrac{\arcsin^{2}\rho+\arcsin^{2}\omega-\arcsin^{2}(\omega\rho)}{4\pi^{2}},
\]
such that 
\[
2\left[G(\rho,\omega)+G(\rho,-\omega)\right]=\tfrac{1}{4}+\tfrac{\arcsin\rho}{\pi}+\tfrac{\arcsin^{2}\rho+\arcsin^{2}\omega-\arcsin^{2}(\omega\rho)}{\pi^{2}},
\]
and
\begin{eqnarray*}
\left(2[G(\rho,1)+G(\rho,-1)]\right)^{2} & = & \left(\tfrac{1}{4}+\tfrac{\arcsin\rho}{\pi}+\tfrac{\arcsin^{2}(1)}{\pi^{2}}\right)^{2}=\left(\tfrac{1}{4}+\tfrac{\arcsin\rho}{\pi}+\tfrac{\pi^{2}/4}{\pi^{2}}\right)^{2}\\
 & = & \left(\tfrac{1}{2}+\tfrac{\arcsin\rho}{\pi}\right)^{2}=\tfrac{1}{4}+\tfrac{\arcsin\rho}{\pi}+\tfrac{\arcsin^{2}\rho}{\pi^{2}}.
\end{eqnarray*}
Finally, we arrive at
\begin{eqnarray*}
\mathrm{cov}(Z,\tilde{Z}) & = & \frac{1}{4}+\frac{\arcsin(\rho)}{\pi}+\frac{\arcsin^{2}(\rho)+\arcsin^{2}(\omega)-\arcsin^{2}(\omega\rho)}{\pi^{2}}-\frac{1}{4}-\frac{\arcsin(\rho)}{\pi}-\frac{\arcsin^{2}(\rho)}{\pi^{2}}\\
 & = & \frac{\arcsin^{2}(\omega)-\arcsin^{2}(\omega\rho)}{\pi^{2}}.
\end{eqnarray*}
\end{proof}

\noindent\textbf{Proof of Theorem \ref{thm:efficiency}.} Define
$Z_{S,j}=1_{\{\Delta_{\frac{S}{N}}X_{\frac{j}{N}}\Delta_{\frac{S}{N}}Y_{\frac{j}{N}}>0\}}$
and $\hat{q}_{S}=\frac{1}{N-S+1}\sum^{N}_{j=S}Z_{S,j}$ and note that
$\hat{\tau}_{S}=2\hat{q}_{S}-1$. Thus, the asymptotic variance of
$Q_{S}=\sin(\tfrac{\pi}{2}\hat{\tau}_{S})$ is given by that of $\hat{q}_{S}$
and the delta method. For the latter we need to multiply the asymptotic
variance of $\hat{q}_{S}$ with the square of:
\[
\frac{\partial\sin(\tfrac{\pi}{2}(2q-1))}{\partial q}=\pi\cos(\tfrac{\pi}{2}(2q-1))=\pi\cos(\arcsin\rho)=\pi\sqrt{1-\rho^{2}}.
\]
We have
\begin{eqnarray*}
\mathrm{var}\left(\sum^{N}_{j=S}Z_{S,j}\right) & = & \sum^{N}_{i=S}\sum^{N}_{j=S}\mathrm{cov}\left(Z_{S,i},Z_{S,j}\right)\\
 & = & \sum^{N-1}_{h=-N+1}(N-S+1-|h|)\mathrm{cov}\left(Z_{S,i},Z_{S,i+h}\right)\\
 & = & \sum^{S-1}_{h=-S+1}(N-S+1-|h|)\tfrac{\arcsin^{2}(\omega_{h})-\arcsin^{2}(\omega_{h}\rho)}{\pi^{2}},
\end{eqnarray*}
where we have used Lemma \ref{lem:SignOf4} with $\omega=\tfrac{S-|h|}{S}$
and $\omega=0$ for $|h|\geq S$. Next, as $N\rightarrow\infty$ we
have $\frac{n}{N-S+1}\rightarrow\frac{1}{S}$ and $\tfrac{N-S+1-|h|}{N-S+1}\rightarrow1$
for all $|h|\leq S$, such that 
\[
\mathrm{var}\left(\sqrt{n}\hat{q}_{S}\right)=\mathrm{var}\left(\sqrt{n}\frac{1}{N-S+1}\sum^{N}_{j=S}Z_{S,j}\right)\rightarrow\frac{1}{S}\sum^{S-1}_{h=-S+1}\tfrac{\arcsin^{2}(\omega_{h})-\arcsin^{2}(\omega_{h}\rho)}{\pi^{2}}.
\]
Asymptotic normality of $\hat{q}_{S}$ follows from the central limit
theorem for strictly stationary $m$-dependent sequences, see
\citet{HoeffdingRobbins:1948}, because $\{Z_{S,j}\}$ is strictly
stationary and $(S-1)$-dependent under Assumption \ref{assu:Brownian}.
The asymptotic behavior of the sum, as $S\rightarrow\infty$, can
be inferred from
\begin{eqnarray*}
I(\rho)=\int^{1}_{0}\arcsin^{2}(x\rho)\mathrm{d}x & = & \left[\tfrac{2\sqrt{1-\rho^{2}x^{2}}\arcsin(\rho x)}{\rho}+x\arcsin^{2}(\rho x)-2x\right]^{1}_{0}\\
 & = & \left(2\tfrac{\sqrt{1-\rho^{2}}}{\rho}\arcsin(\rho)+\arcsin^{2}(\rho)-2\right)-0,
\end{eqnarray*}
such that $\mathrm{avar}\left(\hat{q}_{S}\right)$ is $2/\pi^{2}$
times
\begin{eqnarray*}
\int^{1}_{0}\arcsin^{2}(x)-\arcsin^{2}(x\rho)\mathrm{d}x & = & I(1)-I(\rho)\\
 & = & 0\arcsin(1)+\arcsin^{2}(1)-2-\tfrac{2\sqrt{1-\rho^{2}}\arcsin(\rho)}{\rho}-\arcsin^{2}(\rho)+2\\
 & = & (\tfrac{\pi}{2})^{2}-2\tfrac{\sqrt{1-\rho^{2}}}{\rho}\arcsin(\rho)-\arcsin^{2}(\rho).
\end{eqnarray*}
Thus, multiplying by $2/\pi^{2}$ and applying the delta method we
have 
\begin{eqnarray*}
\mathrm{var}(\sqrt{n}\hat{Q}_{S}) & \rightarrow & \pi^{2}(1-\rho^{2})\frac{2}{\pi^{2}}\left[(\tfrac{\pi}{2})^{2}-2\tfrac{\sqrt{1-\rho^{2}}}{\rho}\arcsin(\rho)-\arcsin^{2}(\rho)\right]\\
 & = & (1-\rho^{2})2\left[\mathrm{arcsin}^{2}(1)-2\tfrac{\sqrt{1-\rho^{2}}}{\rho}\arcsin(\rho)-\arcsin^{2}(\rho)\right].
\end{eqnarray*}
Since $\sqrt{1-\rho^{2}}\arcsin(\rho)/\rho\rightarrow1$ as $\rho\rightarrow0$,
the expression extends continuously to $\rho=0$ with $V_{\infty}(0)=2(\tfrac{\pi^{2}}{4}-2)=\tfrac{\pi^{2}}{2}-4$.

For the monotonicity in $S$ and the lower bound, let $f_{\rho}(x)=\arcsin^{2}(x)-\arcsin^{2}(\rho x)$
for $x\in[0,1]$, so that $f_{\rho}(0)=0$ and $V_{S}(\rho)=(1-\rho^{2})\frac{1}{S}[f_{\rho}(1)+2\sum^{S-1}_{k=1}f_{\rho}(k/S)]=2(1-\rho^{2})T_{S}(f_{\rho})$,
where $T_{S}(f)=\frac{1}{S}[\tfrac{1}{2}f(0)+\sum^{S-1}_{k=1}f(k/S)+\tfrac{1}{2}f(1)]$
is the trapezoidal approximation with $S$ panels to $\int^{1}_{0}f(x)\mathrm{d}x$,
and $V_{\infty}(\rho)=2(1-\rho^{2})\int^{1}_{0}f_{\rho}(x)\mathrm{d}x$
by the calculation above. The function $f_{\rho}$ is convex on $[0,1]$:
with $g(x)=\arcsin^{2}(x)$ we have $g^{\prime\prime}(x)=\tfrac{2}{1-x^{2}}+\tfrac{2x\arcsin(x)}{(1-x^{2})^{3/2}}$,
which is positive and increasing on $[0,1)$, so $f^{\prime\prime}_{\rho}(x)=g^{\prime\prime}(x)-\rho^{2}g^{\prime\prime}(\rho x)\geq(1-\rho^{2})g^{\prime\prime}(\rho x)\geq0$
for $|\rho|<1$. For a convex function the trapezoidal sum overestimates
the integral, $T_{S}(f)\geq\int^{1}_{0}f$, which gives the lower bound.
Convexity alone does not imply that $T_{S}(f)$ is decreasing in $S$
(for $f(x)=|x-\tfrac{1}{2}|$ one has $T_{2}(f)=\tfrac{1}{4}<T_{3}(f)=\tfrac{5}{18}$);
the monotonicity result of \citet{BennettJameson:2000} requires the
derivative of $f$ to be convex, which holds for each power $x^{2k}$,
$k\geq1$, on $[0,1]$. We therefore use the Maclaurin series
\[
\arcsin^{2}(x)=\sum^{\infty}_{k=1}a_{k}x^{2k},\qquad a_{k}=\frac{4^{k}}{2k^{2}\binom{2k}{k}}>0,
\]
which converges on $[-1,1]$, so that $f_{\rho}(x)=\sum^{\infty}_{k=1}a_{k}(1-\rho^{2k})x^{2k}$
with non-negative coefficients. Since $T_{S}$ is a finite linear combination
of point evaluations on $[0,1]$, $T_{S}(f_{\rho})=\sum^{\infty}_{k=1}a_{k}(1-\rho^{2k})T_{S}(x^{2k})$,
and each $T_{S}(x^{2k})$ is decreasing in $S$ by \citet{BennettJameson:2000},
so $T_{S+1}(f_{\rho})\leq T_{S}(f_{\rho})$. The decrease is strict
for $|\rho|<1$, because the $k=1$ term has a positive coefficient and
$T_{S}(x^{2})=\tfrac{1}{3}+\tfrac{1}{6S^{2}}$ is strictly decreasing.

\hfill{}$\square$

\noindent\textbf{Proof of Theorem \ref{theo:Consistency}.} Write
$t_{j}=j/N$, $\delta=S/N$, $M=N-S+1$, and abbreviate $\Delta X_{j}=\Delta_{\frac{S}{N}}X_{\frac{j}{N}}$
and $\Delta W_{x,j}=W_{x}(t_{j})-W_{x}(t_{j-S})$, with the analogous
definitions for $Y$. The limit $\delta\rightarrow0$ is taken with
$N\rightarrow\infty$ and $S/N\rightarrow0$, and we assume $N\geq2S$
throughout. Since convergence in probability is preserved under localization,
we may assume, as in the proofs of Theorems \ref{thm:CLTBSM} and \ref{thm:CLTleverage}
below, that $c\leq\sigma_{x}(u),\sigma_{y}(u)\leq C$ for all $u\in[0,1]$,
for constants $0<c\leq C<\infty$; the lower bound is Assumption \ref{assu:BSM}
and the upper bound follows from $\sigma$ being CADLAG. $C$ denotes
a finite constant that may change from one expression to the next.

\emph{The estimator $Q_{S}$.} Let $\zeta_{j}=\mathrm{sgn}(\Delta X_{j}\Delta Y_{j})$
and $\tilde{\zeta}_{j}=\mathrm{sgn}(\Delta W_{x,j}\Delta W_{y,j})$,
so that $\hat{\tau}_{S}=M^{-1}\sum^{N}_{j=S}\zeta_{j}$. The argument
freezes the volatility at the left endpoint of each window. For $j\geq S$
and $t=t_{j-S}$,
\[
\Delta X_{j}=\sigma_{x}(t)\Delta W_{x,j}+R_{x,j},\qquad R_{x,j}=\int^{t_{j}}_{t}[\sigma_{x}(u)-\sigma_{x}(t)]\mathrm{d}W_{x}(u),
\]
and similarly for $Y$. Because $\sigma_{x}(t)\sigma_{y}(t)>0$, the
sign of $\sigma_{x}(t)\Delta W_{x,j}\sigma_{y}(t)\Delta W_{y,j}$ is
$\tilde{\zeta}_{j}$, and the signs of $\Delta X_{j}$ and $\sigma_{x}(t)\Delta W_{x,j}$
can differ only if $|R_{x,j}|\geq|\sigma_{x}(t)\Delta W_{x,j}|$. Hence
\[
\{\zeta_{j}\neq\tilde{\zeta}_{j}\}\subseteq\{|R_{x,j}|\geq|\sigma_{x}(t)\Delta W_{x,j}|\}\cup\{|R_{y,j}|\geq|\sigma_{y}(t)\Delta W_{y,j}|\},
\]
which also covers the event $\Delta X_{j}\Delta Y_{j}=0$. For any $\eta>0$,
\[
\Pr\big(|R_{x,j}|\geq|\sigma_{x}(t)\Delta W_{x,j}|\big)\leq\Pr\big(|R_{x,j}|\geq\eta\sqrt{\delta}\big)+\Pr\big(|\Delta W_{x,j}|\leq\eta\sqrt{\delta}/c\big)\leq\frac{\mathbb{E}[R^{2}_{x,j}]}{\eta^{2}\delta}+\frac{\eta}{c},
\]
where the second term uses that $\sigma_{x}(t)\geq c$ is $\mathcal{F}_{t}$-measurable
while $\Delta W_{x,j}\sim N(0,\delta)$ is independent of $\mathcal{F}_{t}$,
so that $\Pr(|\Delta W_{x,j}|\leq a\sqrt{\delta})\leq a\sqrt{2/\pi}\leq a$.
This is the only place where the relation between $\sigma$ and $W$
enters, and it requires only that $\sigma$ be adapted. By the Itô
isometry, $\mathbb{E}[R^{2}_{x,j}]=\mathbb{E}\int^{t_{j}}_{t_{j-S}}[\sigma_{x}(u)-\sigma_{x}(t_{j-S})]^{2}\mathrm{d}u$.
Each $u\in[0,1]$ belongs to at most $S+1$ of the windows $[t_{j-S},t_{j}]$,
and for each of them the left endpoint lies in $[u-\delta,u]$, so
\[
\sum^{N}_{j=S}\mathbb{E}[R^{2}_{x,j}]\leq(S+1)\,\varepsilon_{x}(\delta),\qquad\varepsilon_{x}(\delta)=\mathbb{E}\int^{1}_{0}\sup_{s\in[u-\delta,u]}[\sigma_{x}(u)-\sigma_{x}(s)]^{2}\mathrm{d}u.
\]
For a CADLAG path, $\sup_{s\in[u-\delta,u]}|\sigma_{x}(u)-\sigma_{x}(s)|\rightarrow|\sigma_{x}(u)-\sigma_{x}(u-)|$
as $\delta\downarrow0$, which is zero except at the jump times of $\sigma_{x}$,
of which there are at most countably many. The integrand is bounded
by $4C^{2}$, so $\varepsilon_{x}(\delta)\rightarrow0$ by dominated
convergence, first with respect to $\mathrm{d}u$ and then with respect
to $\mathbb{P}$. The same holds for $\varepsilon_{y}(\delta)$. Since
$M\delta\geq S/2$ for $N\geq2S$, so that $(S+1)/(M\delta)\leq4$, we obtain
\[
\frac{1}{M}\sum^{N}_{j=S}\Pr(\zeta_{j}\neq\tilde{\zeta}_{j})\leq\frac{4}{\eta^{2}}\big[\varepsilon_{x}(\delta)+\varepsilon_{y}(\delta)\big]+\frac{2\eta}{c},
\]
and letting first $\delta\rightarrow0$ and then $\eta\rightarrow0$
shows that the left hand side vanishes. Consequently, with $\tilde{\tau}_{S}=M^{-1}\sum^{N}_{j=S}\tilde{\zeta}_{j}$,
\[
\mathbb{E}|\hat{\tau}_{S}-\tilde{\tau}_{S}|\leq\frac{2}{M}\sum^{N}_{j=S}\Pr(\zeta_{j}\neq\tilde{\zeta}_{j})\rightarrow0.
\]
The sequence $\{\tilde{\zeta}_{j}\}$ is a functional of $W$ alone.
It is strictly stationary and $(S-1)$-dependent, with $\mathbb{E}\tilde{\zeta}_{j}=\tau=\tfrac{2}{\pi}\arcsin\rho$
by Greiner's identity, because $(\Delta W_{x,j},\Delta W_{y,j})$ is
bivariate normal with correlation $\rho$. Hence $\mathrm{var}(\tilde{\tau}_{S})\leq(2S-1)/M\leq4\delta\rightarrow0$,
so $\tilde{\tau}_{S}\overset{p}{\rightarrow}\tau$, and therefore
$\hat{\tau}_{S}\overset{p}{\rightarrow}\tau$ and $Q_{S}=\sin(\tfrac{\pi}{2}\hat{\tau}_{S})\overset{p}{\rightarrow}\sin(\tfrac{\pi}{2}\tau)=\rho$
by the continuous mapping theorem.

\emph{The estimator $P$.} Under Assumption \ref{assu:BSM}, $(X,Y)$
is a continuous Itô semimartingale with CADLAG adapted volatility,
so the realized covariance and the realized variances computed from
the $\lfloor\delta^{-1}\rfloor$ non-overlapping increments converge
in probability to $\rho\int^{1}_{0}\sigma_{x}(u)\sigma_{y}(u)\mathrm{d}u$,
$\int^{1}_{0}\sigma^{2}_{x}(u)\mathrm{d}u$ and $\int^{1}_{0}\sigma^{2}_{y}(u)\mathrm{d}u$,
respectively, see \citet[theorem 3.3.1]{JacodProtter:2012}. The
centering by sample means is asymptotically negligible, since the
product of the sample means of the increments is $(X_{1}-X_{0})(Y_{1}-Y_{0})/\lfloor\delta^{-1}\rfloor^{2}$,
and the integrated variances are positive, so $P\overset{p}{\rightarrow}\lambda\rho$
by the continuous mapping theorem. Note that no condition on the
dependence between $\sigma$ and $W$ is required.

\emph{The estimator $K$.} Now suppose that $\sigma$ is independent
of $W$, and condition on the path of $\sigma$. Let $m=\lfloor\delta^{-1}\rfloor$
and let $I_{i}=((i-1)\delta,i\delta]$, $i=1,\ldots,m$, be the non-overlapping
sampling intervals. Conditionally on $\sigma$, the increments $(\Delta_{\delta}X_{i\delta},\Delta_{\delta}Y_{i\delta})$,
$i=1,\ldots,m$, are independent bivariate normal vectors with variances
$\int_{I_{i}}\sigma^{2}_{x}$ and $\int_{I_{i}}\sigma^{2}_{y}$ and covariance
$\rho\int_{I_{i}}\sigma_{x}\sigma_{y}$. For $i\neq j$, the pair of
differences $(\Delta_{\delta}X_{i\delta}-\Delta_{\delta}X_{j\delta},\Delta_{\delta}Y_{i\delta}-\Delta_{\delta}Y_{j\delta})$
is therefore bivariate normal with positive variances and correlation
$\rho h_{\delta}(i,j)$, where
\[
h_{\delta}(i,j)=\frac{\int_{I_{i}}\sigma_{x}\sigma_{y}+\int_{I_{j}}\sigma_{x}\sigma_{y}}{\sqrt{\int_{I_{i}}\sigma^{2}_{x}+\int_{I_{j}}\sigma^{2}_{x}}\sqrt{\int_{I_{i}}\sigma^{2}_{y}+\int_{I_{j}}\sigma^{2}_{y}}}\in[0,1],
\]
and Greiner's identity gives $\mathbb{E}\big[\mathrm{sgn}([\Delta_{\delta}X_{i\delta}-\Delta_{\delta}X_{j\delta}][\Delta_{\delta}Y_{i\delta}-\Delta_{\delta}Y_{j\delta}])\mid\sigma\big]=\tfrac{2}{\pi}\arcsin(\rho h_{\delta}(i,j))$.
The Kendall statistic $\hat{\tau}_{K}=\tfrac{2}{m(m-1)}\sum_{i<j}\mathrm{sgn}(\cdot)$
is a $U$-statistic of order two in conditionally independent (though
not identically distributed) variables with a bounded kernel, so its
conditional variance is bounded by $C/m$ by the Hoeffding decomposition,
and
\[
\hat{\tau}_{K}-\frac{2}{m(m-1)}\sum_{i<j}\tfrac{2}{\pi}\arcsin(\rho h_{\delta}(i,j))\overset{p}{\rightarrow}0.
\]
It remains to identify the limit of the conditional mean, for which
we write it as an integral. Define the step functions $\bar{h}_{\delta}(u,v)=h_{\delta}(\lceil u/\delta\rceil,\lceil v/\delta\rceil)$,
so that
\[
\frac{2}{m(m-1)}\sum_{i<j}\tfrac{2}{\pi}\arcsin(\rho h_{\delta}(i,j))=\frac{m}{m-1}\,\frac{1}{(m\delta)^{2}}\int^{m\delta}_{0}\int^{m\delta}_{0}1_{\{\lceil u/\delta\rceil\neq\lceil v/\delta\rceil\}}\tfrac{2}{\pi}\arcsin(\rho\bar{h}_{\delta}(u,v))\,\mathrm{d}u\,\mathrm{d}v.
\]
The window averages $\delta^{-1}\int_{I_{\lceil u/\delta\rceil}}\sigma^{2}_{x}$,
$\delta^{-1}\int_{I_{\lceil u/\delta\rceil}}\sigma^{2}_{y}$ and $\delta^{-1}\int_{I_{\lceil u/\delta\rceil}}\sigma_{x}\sigma_{y}$
converge to $\sigma^{2}_{x}(u)$, $\sigma^{2}_{y}(u)$ and $\sigma_{x}(u)\sigma_{y}(u)$
at every point $u$ at which both $\sigma_{x}$ and $\sigma_{y}$ are
continuous, because the window containing $u$ shrinks to $u$. The
exceptional set is at most countable, so $\bar{h}_{\delta}(u,v)\rightarrow h(u,v)$
for Lebesgue-almost every $(u,v)\in[0,1]^{2}$; the diagonal strip
$\{\lceil u/\delta\rceil=\lceil v/\delta\rceil\}$ has measure at most
$\delta$; and the integrand is bounded by one. By dominated convergence
the conditional mean converges, for every path of $\sigma$, to
$\tfrac{2}{\pi}\int^{1}_{0}\int^{1}_{0}\arcsin[\rho h(u,v)]\mathrm{d}u\mathrm{d}v$.
Combining, $\hat{\tau}_{K}$ converges in probability to this limit,
and $K=\sin(\tfrac{\pi}{2}\hat{\tau}_{K})$ converges to the expression
in the theorem. The uniform convergence of $\bar{h}_{\delta}$ that
a continuous $\sigma$ would deliver is not needed; almost-everywhere
convergence and boundedness suffice.

For the Spearman and Gaussian rank correlation estimators, inconsistency
follows from a simple example. First note that under constant volatility
these estimators will be consistent. Now suppose $\rho$ is large,
$\rho=0.8$ say, but the volatility processes are time-varying such
that $(\sigma_{x}(u),\sigma_{y}(u))=(10,1)$ for $u<\frac{1}{2}$
and $(\sigma_{x}(u),\sigma_{y}(u))=(1,10)$ for $u\geq0.5$. Then
the extreme ranks for $\Delta X$ will be concentrated in the first
half of the data whereas those $\Delta Y$ will be concentrated on
the second half of the observations. This implies a low rank correlation,
such that the estimators will be inconsistent and biased toward zero
in this example.

\hfill{}$\square$

\medskip{}

We now turn to the proofs of Theorems \ref{thm:CLTBSM} and \ref{thm:CLTleverage}.
Throughout,
$C$ denotes a finite constant that may change from one expression
to the next, $t_{j}=j/N$, $\delta=S/N$, and we abbreviate $\Delta X_{j}=\Delta_{\frac{S}{N}}X_{\frac{j}{N}}$
and $\Delta W_{x,j}=W_{x}(t_{j})-W_{x}(t_{j-S})$, with analogous
definitions for $Y$. Define the sign variables
\[
\zeta_{j}=\mathrm{sgn}(\Delta X_{j}\Delta Y_{j}),\qquad\tilde{\zeta}_{j}=\mathrm{sgn}(\Delta W_{x,j}\Delta W_{y,j}),\qquad D_{j}=\zeta_{j}-\tilde{\zeta}_{j}\in\{-2,-1,0,1,2\},
\]
such that $\hat{\tau}_{S}=\frac{1}{N-S+1}\sum^{N}_{j=S}\zeta_{j}$.
By a standard localization argument, see \citet[lemma 4.4.9]{JacodProtter:2012},
we may assume, without loss of generality, that the characteristics
of $\sigma=(\sigma_{x},\sigma_{y})$ are uniformly bounded and that
$c\leq\sigma_{x}(u),\sigma_{y}(u)\leq C$ for all $u\in[0,1]$, for
constants $0<c\leq C<\infty$. An implication, which follows from
the Burkholder-Davis-Gundy inequality and the boundedness of the
characteristics, and will be used repeatedly, is the modulus of continuity
\begin{equation}
\mathbb{E}\big[\sup{}_{t\leq u\leq t+\delta}|\sigma_{x}(u)-\sigma_{x}(t)|^{2}\big]\leq C\delta,\label{eq:modulus}
\end{equation}
with the corresponding bound for $\sigma_{y}$.

Note that $\{\tilde{\zeta}_{j}\}$ is a functional
of $W$ alone. Because $W$ is a bivariate Wiener process with constant
correlation $\rho$, the sequence $\{\tilde{\zeta}_{j}\}$ is strictly
stationary and $(S-1)$-dependent, regardless of the properties of
$\sigma$, with $\mathbb{E}\tilde{\zeta}_{j}=\tau=\frac{2}{\pi}\mathrm{asin}(\rho)$
and covariances given by Lemma \ref{lem:SignOf4}, exactly as in the
proof of Theorem \ref{thm:efficiency}. Both theorems therefore amount
to coupling $\zeta_{j}$ to $\tilde{\zeta}_{j}$ and showing that the
substitution error is $o_{p}(n^{-1/2})$. This requires two lemmas:
a bound on the probability that the two signs disagree (Lemma \ref{lem:Flip}),
and a bound on the conditional bias of the disagreement (Lemma \ref{lem:CondBias}).
Lemma \ref{lem:CondBias} is stated in two parts. Part (a) is the no-leverage
case of Assumption \ref{assu:BSMCLT}, which is all that Theorem \ref{thm:CLTBSM} needs; part
(b) covers Assumption \ref{assu:Leverage}, where $\sigma$ may load
on $W$, and is used only for Theorem \ref{thm:CLTleverage}.
\begin{lem}
\label{lem:Flip}Suppose that Assumption \ref{assu:BSM} holds together
with Assumption \ref{assu:BSMCLT} or Assumption \ref{assu:Leverage},
with the localization above in force. Then $\Pr(D_{j}\neq0)\leq C\delta^{1/3}$,
and consequently $\mathbb{E}[D_{j}^{2}]\leq4\Pr(D_{j}\neq0)\leq C\delta^{1/3}$, uniformly
in $j$.
\end{lem}
\begin{proof}
Write $\Delta X_{j}=\sigma_{x}(t_{j-S})\Delta W_{x,j}+R_{x,j}$ where
$R_{x,j}=\int^{t_{j}}_{t_{j-S}}[\sigma_{x}(u)-\sigma_{x}(t_{j-S})]\mathrm{d}W_{x}(u)$.
By the It\^{o} isometry, which only requires that $\sigma_{x}$ be
adapted, and (\ref{eq:modulus}),
\[
\mathbb{E}[R^{2}_{x,j}]=\mathbb{E}\int^{t_{j}}_{t_{j-S}}[\sigma_{x}(u)-\sigma_{x}(t_{j-S})]^{2}\mathrm{d}u\leq C\int^{t_{j}}_{t_{j-S}}(u-t_{j-S})\mathrm{d}u\leq C\delta^{2}.
\]
Since $\sigma_{x}(t_{j-S})\geq c>0$ we have $\mathrm{sgn}(\sigma_{x}(t_{j-S})\Delta W_{x,j})=\mathrm{sgn}(\Delta W_{x,j})$,
such that $\mathrm{sgn}(\Delta X_{j})\neq\mathrm{sgn}(\Delta W_{x,j})$
requires $|\Delta W_{x,j}|\leq|R_{x,j}|/c$. For any $a>0$,
\[
\Pr\big(|\Delta W_{x,j}|\leq\tfrac{|R_{x,j}|}{c}\big)\leq\Pr\big(|\Delta W_{x,j}|\leq a\sqrt{\delta}\big)+\Pr\big(|R_{x,j}|\geq ca\sqrt{\delta}\big)\leq Ca+\tfrac{\mathbb{E}[R^{2}_{x,j}]}{c^{2}a^{2}\delta}\leq C\big(a+\tfrac{\delta}{a^{2}}\big),
\]
where we used $\Delta W_{x,j}\sim N(0,\delta)$ and Markov's inequality.
Setting $a=\delta^{1/3}$ yields $\Pr(\mathrm{sgn}(\Delta X_{j})\neq\mathrm{sgn}(\Delta W_{x,j}))\leq C\delta^{1/3}$.
The event $\{D_{j}\neq0\}$ requires a sign disagreement for $X$
or for $Y$, and the union bound and $|D_{j}|\leq2$ complete the
proof.
\end{proof}
\begin{lem}
\label{lem:CondBias}Suppose that Assumption \ref{assu:BSM} holds,
with the localization above in force, and that $|\rho|<1$. Suppose
in addition that either (a) Assumption \ref{assu:BSMCLT} holds, or
(b) Assumption \ref{assu:Leverage} holds. Then, as $\delta\rightarrow0$,
\[
\frac{1}{N}\sum^{N}_{j=S}\mathbb{E}\big|\mathbb{E}[D_{j}\mid\mathcal{F}_{t_{j-S}}]\big|=o(\delta^{1/2}).
\]
\end{lem}
\begin{proof}
Set $t=t_{j-S}$. The increments of $W$ over $[t,t_{j}]$ are independent
of $\mathcal{F}_{t}$ and bivariate Gaussian with correlation $\rho$,
such that $\mathbb{E}[\tilde{\zeta}_{j}\mid\mathcal{F}_{t}]=\tau$
exactly, and the lemma is equivalent to the corresponding statement
for $\mathbb{E}[\zeta_{j}\mid\mathcal{F}_{t}]-\tau$.

The case $\rho=0$ must be separated, because the Cauchy-Schwarz step
below is expressed through the ratio $\rho_{j}/\rho$, which is $0/0$
when $\rho=0$. It is also the trivial case. The window correlations
$\rho_{j}$ and $\hat{\rho}_{j}$ defined below are both of the form
$\rho$ times a ratio of inner products, so $\rho=0$ gives $\rho_{j}=\hat{\rho}_{j}=0$
exactly, and since $\tau=\tfrac{2}{\pi}\mathrm{asin}(0)=0$ the Cauchy-Schwarz
bound holds with a left-hand side that is identically zero. Nothing
else in the proof divides by $\rho$, so we assume $\rho\neq0$ whenever
that bound is derived.

\textit{Part (a): the case without leverage.} Under Assumption \ref{assu:BSMCLT}
the volatility is a functional of $(\mathcal{F}_{0},W^{\prime},\mu)$
and is therefore independent of $W$. Conditionally on $\mathcal{F}_{t}$
and the
path of $\sigma$ over the window $[t,t_{j}]$, the pair $(\Delta X_{j},\Delta Y_{j})$
is bivariate Gaussian with mean zero and correlation
\[
\rho_{j}=\rho\frac{\int^{t_{j}}_{t}\sigma_{x}(u)\sigma_{y}(u)\mathrm{d}u}{\sqrt{\int^{t_{j}}_{t}\sigma^{2}_{x}(u)\mathrm{d}u\int^{t_{j}}_{t}\sigma^{2}_{y}(u)\mathrm{d}u}},
\]
such that $\mathbb{E}[\zeta_{j}\mid\mathcal{F}_{t},\sigma]=\tfrac{2}{\pi}\mathrm{asin}(\rho_{j})$
exactly, by Greiner's identity (Proposition \ref{prop:Greiner}).
The Cauchy-Schwarz deficiency $1-\rho_{j}/\rho$ is second order in
the oscillation of $\sigma$ over the window: let $f=\sigma_{x}$
and $g=\sigma_{y}$ on $[t,t_{j}]$, let $u=f/\Vert f\Vert$ and $v=g/\Vert g\Vert$
be the associated unit vectors in $L^{2}[t,t_{j}]$, and let $w$
denote the constant unit vector. Then $\rho_{j}/\rho=\langle u,v\rangle=1-\tfrac{1}{2}\Vert u-v\Vert^{2}$,
and since $\Vert f\Vert/\sqrt{\delta}\in[\min f,\max f]$, we have
$|f(s)-\Vert f\Vert/\sqrt{\delta}|\leq\mathrm{osc}_{j}(f)$ pointwise,
where $\mathrm{osc}_{j}(f)=\sup_{s,u\in[t,t_{j}]}|f(s)-f(u)|$, such
that $\Vert u-w\Vert\leq\mathrm{osc}_{j}(f)/c$, and similarly for
$v$. Hence
\[
1-\rho_{j}/\rho=\tfrac{1}{2}\Vert u-v\Vert^{2}\leq C\big[\mathrm{osc}^{2}_{j}(\sigma_{x})+\mathrm{osc}^{2}_{j}(\sigma_{y})\big].
\]
Because $|\rho_{j}|\leq|\rho|<1$ and $\mathrm{asin}$ is Lipschitz
on $[-|\rho|,|\rho|]$ with constant $(1-\rho^{2})^{-1/2}$, it follows
from (\ref{eq:modulus}) that
\[
\mathbb{E}\big|\mathbb{E}[\zeta_{j}\mid\mathcal{F}_{t}]-\tau\big|\leq C\,\mathbb{E}\big[\mathrm{osc}^{2}_{j}(\sigma_{x})+\mathrm{osc}^{2}_{j}(\sigma_{y})\big]\leq C\delta,
\]
which is stronger than required. Note that this step is where the regularity
of the volatility paths matters: if $\sigma$ is merely $\alpha$-H\"{o}lder
continuous, then $\mathbb{E}[\mathrm{osc}^{2}_{j}]=O(\delta^{2\alpha})$,
which is $o(\delta^{1/2})$ if and only if $\alpha>1/4$, cf.\ Remark
\ref{rem:Roughness}.

\textit{Part (b): the leverage case.} Now let $\sigma$ satisfy Assumption
\ref{assu:Leverage}. Writing the Grigelionis representation (\ref{eq:Grigelionis}),
augmented with the leverage terms, from the start of the window gives $\ensuremath{\sigma_{x}(u)=\sigma_{x}(t)+\int^{u}_{t}\tilde{a}_{x}(s)\mathrm{d}s+\sum_{b\in\{x,y\}}\int^{u}_{t}\nu_{xb}(s)\mathrm{d}W_{b}(s)+\int^{u}_{t}\nu^{\prime}_{x}(s)\mathrm{d}W^{\prime}(s)+M^{J}_{x}(u)-M^{J}_{x}(t)}$
for $u\in[t,t_{j}]$, where $M^{J}_{x}(u)=\int^{u}_{0}\int_{E}\gamma_{x}(s,e)\tilde{\mu}(\mathrm{d}s,\mathrm{d}e)$
is the compensated jump martingale of $\sigma_{x}$, which is square
integrable by the envelope condition, and all coefficients are adapted and
bounded (by localization) and are themselves It\^{o} semimartingales
with bounded characteristics by Assumption \ref{assu:Leverage}(iii')--(iv').
Decompose
\[
\Delta X_{j}=\mathbb{X}_{j}+H_{x,j}+\varepsilon_{x,j},
\]
where, with the coefficients frozen at the start of the window,
\[
\hat{\sigma}_{x}(u)=\sigma_{x}(t)+\tilde{a}_{x}(t)(u-t)+\int^{u}_{t}\nu^{\prime}_{x}(t)\mathrm{d}W^{\prime}(s)+\int^{u}_{t}\int_{E}\gamma_{x}(t,e)\tilde{\mu}(\mathrm{d}s,\mathrm{d}e),
\]
we set
\begin{eqnarray*}
\mathbb{X}_{j} & = & \int^{t_{j}}_{t}\hat{\sigma}_{x}(u)\mathrm{d}W_{x}(u),\\
H_{x,j} & = & \sum{}_{b\in\{x,y\}}\nu_{xb}(t)\int^{t_{j}}_{t}[W_{b}(u)-W_{b}(t)]\mathrm{d}W_{x}(u),
\end{eqnarray*}
and $\varepsilon_{x,j}$ collects the coefficient-freezing errors,
i.e., stochastic integrals (with respect to $\mathrm{d}W_{x}$) of
terms such as $\int^{u}_{t}[\nu_{xb}(s)-\nu_{xb}(t)]\mathrm{d}W_{b}(s)$,
$\int^{u}_{t}[\tilde{a}_{x}(s)-\tilde{a}_{x}(t)]\mathrm{d}s$, and
$\int^{u}_{t}\int_{E}[\gamma_{x}(s,e)-\gamma_{x}(t,e)]\tilde{\mu}(\mathrm{d}s,\mathrm{d}e)$.
Freezing the jump coefficient at $t$ is what permits jump sizes to
load on $W$: $\gamma_{x}(t,\cdot)$ is $\mathcal{F}_{t}$-measurable,
so the frozen jump martingale in $\hat{\sigma}_{x}$ is a functional
of $\mu$ and $\mathcal{F}_{t}$ alone. Let
$\mathcal{G}=\mathcal{F}_{t}\vee\sigma(\{W^{\prime}(u),\mu((t,u]\times\cdot)\}_{u\in[t,t_{j}]})$,
where $\mu$ is the jump measure, such that $\hat{\sigma}_{x}$ and
$\hat{\sigma}_{y}$ are $\mathcal{G}$-measurable processes, while
the increments of $W$ over the window remain a bivariate Wiener process
with correlation $\rho$, independent of $\mathcal{G}$, by parts (i)
and (ii) of Assumption \ref{assu:BSMCLT}. Three properties are used
below. First, conditionally
on $\mathcal{G}$, the pair $(\mathbb{X}_{j},\mathbb{Y}_{j})$ is
bivariate Gaussian with mean zero and correlation $\hat{\rho}_{j}=\rho\langle\hat{\sigma}_{x},\hat{\sigma}_{y}\rangle/(\Vert\hat{\sigma}_{x}\Vert\Vert\hat{\sigma}_{y}\Vert)$,
because $(\hat{\sigma}_{x},\hat{\sigma}_{y})$ are $\mathcal{G}$-measurable
integrands. Second, $H_{x,j}$ is a second-order Wiener functional
with $\mathbb{E}[H^{2}_{x,j}]\leq C\delta^{2}$, and it is \textit{even}
under the reflection $(W_{x},W_{y})\mapsto(-W_{x},-W_{y})$, whereas
$(\mathbb{X}_{j},\mathbb{Y}_{j})$ is odd; the conditional law of
the increments of $W$ given $\mathcal{G}$ is invariant to this reflection.
Third, $\mathbb{E}[\varepsilon^{2}_{x,j}]\leq C\delta^{3}$ uniformly
in $j$. This is where Assumption \ref{assu:Leverage}(iii')--(iv')
is used: because the coefficients are It\^{o} semimartingales with
bounded characteristics, $\mathbb{E}[(\nu_{xb}(s)-\nu_{xb}(t))^{2}]\leq C(s-t)$,
and the It\^{o} isometry gives
\[
\mathbb{E}\Big(\int^{u}_{t}[\nu_{xb}(s)-\nu_{xb}(t)]\mathrm{d}W_{b}(s)\Big)^{2}=\int^{u}_{t}\mathbb{E}[(\nu_{xb}(s)-\nu_{xb}(t))^{2}]\mathrm{d}s\leq C(u-t)^{2},
\]
so a further integration in $u$ over the window yields $C\delta^{3}$.
The drift term is bounded in the same way, and so is the jump-freezing
term, since the compensated-jump isometry gives
$\mathbb{E}\big(\int^{u}_{t}\int_{E}[\gamma_{x}(s,e)-\gamma_{x}(t,e)]\tilde{\mu}(\mathrm{d}s,\mathrm{d}e)\big)^{2}=\mathbb{E}\int^{u}_{t}\Vert\gamma_{x}(s,\cdot)-\gamma_{x}(t,\cdot)\Vert^{2}_{L^{2}(\lambda)}\mathrm{d}s\leq C(u-t)^{2}$
under (iv'). Merely CADLAG coefficients would deliver $o(\delta^{2})$
in place of $C\delta^{3}$, which is not enough for step (iv) below.

Fix $\kappa\in(0,1/4)$ and let $\eta=\delta^{1/2-\kappa}$. Let $\mathcal{A}_{j}$
denote the event $\sup_{u}|\hat{\sigma}_{x}(u)-\sigma_{x}(t)|\leq c/2$
(and similarly for $y$), on which $c/2\leq\hat{\sigma}_{x}\leq C$
and $\Vert\hat{\sigma}_{x}\Vert^{2}\geq(c^{2}/4)\delta$. We claim
that
\begin{equation}
\Pr(\mathcal{A}^{c}_{j})\leq C\delta,\qquad\text{uniformly in }j.\label{eq:Aprob}
\end{equation}
Split $\hat{\sigma}_{x}-\sigma_{x}(t)$ into its continuous part and
its frozen jump martingale $\hat{M}^{J}_{x}(u)=\int^{u}_{t}\int_{E}\gamma_{x}(t,e)\tilde{\mu}(\mathrm{d}s,\mathrm{d}e)$.
The continuous part is a $\mathcal{G}$-measurable It\^{o} process
with bounded characteristics over an interval of length $\delta$,
so its supremum exceeds $c/4$ with probability at most $C_{p}\delta^{p}$
for every $p\geq1$. For the jump part, Doob's inequality and the
compensated-jump isometry give
\[
\Pr\big(\sup{}_{u\in[t,t_{j}]}|\hat{M}^{J}_{x}(u)|>c/4\big)\leq\tfrac{16}{c^{2}}\,4\,\mathbb{E}[\hat{M}^{J}_{x}(t_{j})^{2}]=\tfrac{64}{c^{2}}\,\delta\int_{E}\gamma^{2}_{x}(t,e)\lambda(\mathrm{d}e)\leq C\delta,
\]
by the envelope condition. This is the correct order and it cannot
be improved: a window that contains a volatility jump of non-negligible
size is not an event of arbitrarily high order in $\delta$, and the
fraction of such windows is exactly $O(\delta)$. It is nevertheless
harmless, because $|D_{j}|\leq2$, so that
$\mathbb{E}|\mathbb{E}[D_{j}1_{\mathcal{A}^{c}_{j}}\mid\mathcal{F}_{t}]|\leq2\Pr(\mathcal{A}^{c}_{j})\leq C\delta=o(\delta^{1/2})$,
uniformly in $j$. We may therefore work on $\mathcal{A}_{j}$ for
the remainder of the proof.

On $\mathcal{A}_{j}$, normalize $\hat{U}_{x}=\mathbb{X}_{j}/\Vert\hat{\sigma}_{x}\Vert$,
$\hat{H}_{x}=H_{x,j}/\Vert\hat{\sigma}_{x}\Vert$, and $\hat{e}_{x}=\varepsilon_{x,j}/\Vert\hat{\sigma}_{x}\Vert$
(and similarly for $y$), such that
$\zeta_{j}=\mathrm{sgn}\{(\hat{U}_{x}+\hat{H}_{x}+\hat{e}_{x})(\hat{U}_{y}+\hat{H}_{y}+\hat{e}_{y})\}$,
where, conditionally on $\mathcal{G}$, $(\hat{U}_{x},\hat{U}_{y})$
is standard bivariate Gaussian with correlation $\hat{\rho}_{j}$
and $|\hat{\rho}_{j}|\leq|\rho|<1$. The normalizations satisfy
$\mathbb{E}[\hat{H}^{2}_{x}]\leq C\delta$ and $\mathbb{E}[\hat{e}^{2}_{x}]\leq C\delta^{2}$,
using $\Vert\hat{\sigma}_{x}\Vert^{2}\geq(c^{2}/4)\delta$ together
with the second and third properties above. The same argument as in
part (a), applied to $\hat{\sigma}$ in place of $\sigma$, gives
\begin{equation}
\big|\tfrac{2}{\pi}\mathrm{asin}(\hat{\rho}_{j})-\tau\big|\leq C\big[\mathrm{osc}^{2}_{j}(\hat{\sigma}_{x})+\mathrm{osc}^{2}_{j}(\hat{\sigma}_{y})\big],\qquad\mathbb{E}\big[\mathrm{osc}^{2}_{j}(\hat{\sigma}_{x})\big]\leq C\delta.\label{eq:CSbias}
\end{equation}
It therefore remains to bound $\mathbb{E}[\zeta_{j}\mid\mathcal{G}]-\tfrac{2}{\pi}\mathrm{asin}(\hat{\rho}_{j})=\mathbb{E}[\mathrm{sgn}\{(\hat{U}_{x}+\hat{H}_{x}+\hat{e}_{x})(\hat{U}_{y}+\hat{H}_{y}+\hat{e}_{y})\}-\mathrm{sgn}\{\hat{U}_{x}\hat{U}_{y}\}\mid\mathcal{G}]$.
We split the sample space into four regions and treat the freezing
error $\hat{e}$ separately from the even term $\hat{H}$ at the end.
Write $\hat{h}_{x}=\hat{H}_{x}+\hat{e}_{x}$.

Before the four regions, we record the representation of $\hat{H}$
that the argument uses. Conditionally on $\mathcal{G}$, the increments
of $W$ over the window span a Gaussian Hilbert space $\mathcal{H}$,
and we decompose $\mathcal{H}=\mathrm{span}\{\hat{U}_{x},\hat{U}_{y}\}\oplus\mathcal{H}^{\perp}$.
Let $\beta$ denote the projection of the window path onto $\mathcal{H}^{\perp}$,
i.e., its Brownian bridge components. Being jointly Gaussian and uncorrelated
with $\hat{U}=(\hat{U}_{x},\hat{U}_{y})$ given $\mathcal{G}$, the
process $\beta$ is conditionally \emph{independent} of $\hat{U}$.
Since $H_{x,j}$ belongs to the second Wiener chaos over $\mathcal{H}$
with $\mathcal{G}$-measurable bounded coefficients, this decomposition
yields
\begin{equation}
\hat{H}_{x}=\sqrt{\delta}\,\Psi_{x}(\hat{U}_{x},\hat{U}_{y};\beta),\label{eq:chaos}
\end{equation}
where $\Psi_{x}$ is a polynomial of degree at most two in $(\hat{U}_{x},\hat{U}_{y})$
whose coefficients are measurable with respect to $\mathcal{G}\vee\sigma(\beta)$
and belong to the first two chaoses over $\beta$. By hypercontractivity
on a fixed chaos, together with $\mathbb{E}[\hat{H}^{2}_{x}]\leq C\delta$
and the boundedness of $\nu(t)$ and $\hat{\sigma}$, all conditional
moments of these coefficients are bounded by constants that do not
depend on $j$ or $\delta$. The analogous representation holds for
$\hat{H}_{y}$. The scale factor $\sqrt{\delta}$ in (\ref{eq:chaos})
is the reason the leverage term is not negligible by a crude bound:
$\hat{H}_{x}$ is of the same order as the width of the strip in which
$\hat{U}_{x}$ can change sign, and only the cancellation exploited
in step (iii) makes its contribution $o(\delta^{1/2})$.

(i) \textit{Off-strip region}, $|\hat{U}_{x}|>\eta$ and $|\hat{U}_{y}|>\eta$:
a sign change requires $|\hat{h}_{x}|>\eta$ or $|\hat{h}_{y}|>\eta$.
By (\ref{eq:chaos}) and the bounded moments of $\Psi_{x}$,
$\Pr(|\hat{H}_{x}|>\eta/2\mid\mathcal{G})\leq C_{p}\delta^{p}\eta^{-2p}=C_{p}\delta^{2p\kappa}$
for every $p$, which is $o(\delta)$ for $p$ large, while
$\Pr(|\hat{e}_{x}|>\eta/2)\leq4\mathbb{E}[\hat{e}^{2}_{x}]\eta^{-2}\leq C\delta^{2}\delta^{-1+2\kappa}=C\delta^{1+2\kappa}=o(\delta)$.
The same holds with $y$ in place of $x$.

(ii) \textit{Corner region}, $|\hat{U}_{x}|\leq\eta$ and $|\hat{U}_{y}|\leq\eta$:
since the joint density of $(\hat{U}_{x},\hat{U}_{y})$ is bounded
(uniformly, because $|\hat{\rho}_{j}|\leq|\rho|<1$), this event has
probability at most $C\eta^{2}=C\delta^{1-2\kappa}=o(\delta^{1/2})$
for $\kappa<1/4$.

(iii) \textit{The strips}, e.g. $|\hat{U}_{x}|\leq\eta<|\hat{U}_{y}|$:
here $\mathrm{sgn}(\hat{U}_{y}+\hat{h}_{y})=\mathrm{sgn}(\hat{U}_{y})$
up to an event of probability $o(\delta)$ by the argument in (i),
such that the contribution of this strip is, up to $o(\delta)$ terms,
\[
A_{x}=\mathbb{E}\big[\big(\mathrm{sgn}(\hat{U}_{x}+\hat{H}_{x})-\mathrm{sgn}(\hat{U}_{x})\big)\,\mathrm{sgn}(\hat{U}_{y})1_{\{|\hat{U}_{x}|\leq\eta<|\hat{U}_{y}|\}}\mid\mathcal{G}\big].
\]
This is the term for which the reflection symmetry is decisive. The
conditional law of the post-$t$ increments of $W$ given $\mathcal{G}$
is invariant under $(W_{x},W_{y})\mapsto(-W_{x},-W_{y})$; this reflection
maps $(\hat{U},\hat{H})\mapsto(-\hat{U},\hat{H})$, because $\hat{U}$
is linear and $\hat{H}=(\hat{H}_{x},\hat{H}_{y})$ is bilinear in
the increments of $W$ (the $\mathcal{G}$-measurable integrands $\hat{\sigma}$
and the frozen coefficients $\nu(t)$ are unchanged). Applying the
reflection inside $A_{x}$ and averaging the two representations,
\[
A_{x}=\tfrac{1}{2}\mathbb{E}\big[\big(\mathrm{sgn}(\hat{U}_{x}+\hat{H}_{x})+\mathrm{sgn}(\hat{U}_{x}-\hat{H}_{x})-2\,\mathrm{sgn}(\hat{U}_{x})\big)\,\mathrm{sgn}(\hat{U}_{y})1_{\{|\hat{U}_{x}|\leq\eta<|\hat{U}_{y}|\}}\mid\mathcal{G}\big].
\]
The integrand vanishes identically on $\{|\hat{U}_{x}|>|\hat{H}_{x}|\}$,
because $\mathrm{sgn}(a+h)+\mathrm{sgn}(a-h)=2\,\mathrm{sgn}(a)$
when $|a|>|h|$, while on $\{|\hat{U}_{x}|<|\hat{H}_{x}|\}$ it equals
$-2\,\mathrm{sgn}(\hat{U}_{x})\,\mathrm{sgn}(\hat{U}_{y})1_{\{\cdot\}}$,
because $\mathrm{sgn}(a+h)+\mathrm{sgn}(a-h)=0$ when $|a|<|h|$.
Hence
\begin{equation}
|A_{x}|\leq\mathbb{E}\big[\big|\mathbb{E}\big[\mathrm{sgn}(\hat{U}_{x})\,1_{\{|\hat{U}_{x}|<|\hat{H}_{x}|\}}1_{\{|\hat{U}_{x}|\leq\eta<|\hat{U}_{y}|\}}\mid\mathcal{H}_{j}\big]\big|\ \big|\ \mathcal{G}\big],\label{eq:Astrip}
\end{equation}
where $\mathcal{H}_{j}=\mathcal{G}\vee\sigma(\beta,\hat{U}_{y})$
and we used that $\mathrm{sgn}(\hat{U}_{y})$ is $\mathcal{H}_{j}$-measurable.

In (\ref{eq:Astrip}) we condition on the bridge components $\beta$
and $\hat{U}_{y}$ rather than on $(\hat{H}_{x},\hat{U}_{y})$, because
one cannot condition on $\hat{H}_{x}$ and then treat $\hat{U}_{x}$ as
a variable with a smooth conditional density: $\hat{H}_{x}$ is built
from the same increments as $\hat{U}_{x}$, and it contains the term
$\nu_{xx}(t)\int^{t_{j}}_{t}[W_{x}(u)-W_{x}(t)]\mathrm{d}W_{x}(u)=\tfrac{1}{2}\nu_{xx}(t)[(\Delta W_{x,j})^{2}-\delta]$,
which is a deterministic function of $|\hat{U}_{x}|$ when $\hat{\sigma}_{x}$
is constant over the window. Conditioning on $\hat{H}_{x}$ can therefore
pin down $|\hat{U}_{x}|$, and the conditional law of $\hat{U}_{x}$
is then degenerate rather than smooth. Conditioning on $(\beta,\hat{U}_{y})$
avoids this. Given $\mathcal{H}_{j}$, the only remaining randomness
is $\hat{U}_{x}$, whose conditional law is $N(\hat{\rho}_{j}\hat{U}_{y},1-\hat{\rho}^{2}_{j})$
with density $q$; because $|\hat{\rho}_{j}|\leq|\rho|<1$, both $\Vert q\Vert_{\infty}$
and $\Vert q^{\prime}\Vert_{\infty}$ are bounded uniformly in $j$
and in $\mathcal{H}_{j}$. At the same time, by (\ref{eq:chaos}),
$\hat{H}_{x}$ becomes the explicit quadratic polynomial $v\mapsto\sqrt{\delta}\Psi_{x}(v)$,
where the $\mathcal{H}_{j}$-measurable arguments are suppressed.

Because (\ref{eq:Astrip}) already carries the strip indicator
$1_{\{|\hat{U}_{x}|\leq\eta\}}$, the set over which the inner expectation
is taken is
\[
B_{\eta}=\{v:|v|<\sqrt{\delta}|\Psi_{x}(v)|,\ |v|\leq\eta\},
\]
which is \emph{not} the symmetric interval $(-|\hat{H}_{x}|,|\hat{H}_{x}|)$,
since the threshold depends on $v$. The restriction $|v|\leq\eta$ is
not cosmetic: $\Psi_{x}$ is quadratic, so the unrestricted set
$\{v:|v|<\sqrt{\delta}|\Psi_{x}(v)|\}$ also has remote components
where $|v|$ is of order $\delta^{-1/2}$, and it is not an interval.
Those components lie outside the strip and are therefore excluded at
the outset. On $B_{\eta}$ the set is, however, symmetric to leading
order. Let $R=\sup_{|v|\leq1}\{|\Psi_{x}(v)|+|\Psi^{\prime}_{x}(v)|\}$,
which has moments of all orders bounded uniformly, and restrict attention
to the event $\{R\leq\delta^{-\kappa/2}\}$; its complement contributes
at most $\Pr(R>\delta^{-\kappa/2})\leq\mathbb{E}[R^{p}]\delta^{p\kappa/2}=o(\delta)$
for $p$ large. Take $\delta$ small enough that $\eta=\delta^{1/2-\kappa}<1$,
so that every $v\in B_{\eta}$ has $|v|\leq\eta<1$ and the bound $R$
applies to it. On the event $\{R\leq\delta^{-\kappa/2}\}$ any such $v$
then satisfies $|v|<\sqrt{\delta}R\leq\delta^{1/2-\kappa/2}$, and
$\delta^{1/2-\kappa/2}<\delta^{1/2-\kappa}=\eta$, so that
$B_{\eta}\subseteq(-\delta^{1/2-\kappa/2},\delta^{1/2-\kappa/2})$ lies
strictly inside the strip and the constraint $|v|\leq\eta$ is not
binding there. On that interval $\sqrt{\delta}|\Psi^{\prime}_{x}(v)|\leq\sqrt{\delta}R\leq\delta^{1/2-\kappa/2}<1$,
so that $v\mapsto v\mp\sqrt{\delta}|\Psi_{x}(v)|$ is strictly
increasing. Consequently $B_{\eta}=(v_{-},v_{+})$ is an interval containing
the origin, with $v_{\pm}$ solving $v=\pm\sqrt{\delta}|\Psi_{x}(v)|$,
and a one-step expansion around $v=0$ gives
\[
v_{\pm}=\pm a+r_{\pm},\qquad a=\sqrt{\delta}|\Psi_{x}(0)|\leq\delta^{1/2-\kappa/2},\qquad|r_{\pm}|\leq\sqrt{\delta}\,|\Psi^{\prime}_{x}|_{\infty}\max|v_{\pm}|\leq C\delta^{1-\kappa}.
\]
The inner expectation in (\ref{eq:Astrip}) is thus bounded by
\[
\Big|\int^{v_{+}}_{v_{-}}\mathrm{sgn}(v)q(v)\mathrm{d}v\Big|\leq\Big|\int^{a}_{0}[q(v)-q(-v)]\mathrm{d}v\Big|+\Vert q\Vert_{\infty}(|r_{+}|+|r_{-}|)\leq\Vert q^{\prime}\Vert_{\infty}a^{2}+C\delta^{1-\kappa}\leq C\delta^{1-\kappa}.
\]
The first term is the near-symmetry of the Gaussian density about
the origin, which is what produces the square $a^{2}=O(\delta^{1-\kappa})$
in place of the crude bound $a=O(\delta^{1/2-\kappa/2})$; the second
is the asymmetry of the set $B_{\eta}$ itself, which is of the same order.
Hence $|A_{x}|\leq C\delta^{1-\kappa}+o(\delta)$. The strip with the
roles of $x$ and $y$ interchanged is handled identically, conditioning
on $\mathcal{G}\vee\sigma(\beta,\hat{U}_{x})$ instead.

(iv) \textit{The freezing errors} $\hat{e}$: replacing $\mathrm{sgn}(\hat{U}_{x}+\hat{H}_{x}+\hat{e}_{x})$
by $\mathrm{sgn}(\hat{U}_{x}+\hat{H}_{x})$ changes the conditional
expectation by at most $2\Pr(|\hat{U}_{x}+\hat{H}_{x}|\leq|\hat{e}_{x}|\mid\mathcal{G})$.
Here too the threshold is random and dependent on the variable being
thresholded, so one cannot integrate the density of $\hat{U}_{x}+\hat{H}_{x}$
against the law of $\hat{e}_{x}$. We therefore use a deterministic
threshold: for any $0<b\leq\eta/4$,
\[
\Pr\big(|\hat{U}_{x}+\hat{H}_{x}|\leq|\hat{e}_{x}|\mid\mathcal{G}\big)\leq\Pr\big(|\hat{U}_{x}+\hat{H}_{x}|\leq b\mid\mathcal{G}\big)+\Pr\big(|\hat{e}_{x}|>b\mid\mathcal{G}\big)\leq Cb+o(\delta)+b^{-2}\mathbb{E}[\hat{e}^{2}_{x}\mid\mathcal{G}],
\]
where the $o(\delta)$ term, uniform in $j$, collects the exceptional
events excluded in the argument that follows.
The first bound is an anti-concentration statement for $\hat{U}_{x}+\hat{H}_{x}$,
because $\Psi_{x}$ is quadratic and a
quadratic map can carry values of $v$ far from the origin back into
a neighbourhood of it; the local monotonicity of step (iii) does not
by itself exclude such a distant preimage. Write $\Psi_{x}(v)=c_{0}+c_{1}v+c_{2}v^{2}$,
with $\mathcal{H}_{j}$-measurable coefficients, and let $\bar{R}=|c_{0}|+|c_{1}|+|c_{2}|$,
which has moments of all orders bounded uniformly by the same hypercontractivity
argument, and which satisfies $R\leq3\bar{R}$ and $|\Psi_{x}(v)|\leq\bar{R}(1+v^{2})$
for \emph{every} $v$. As in step (iii) we may restrict to $\{\bar{R}\leq\delta^{-\kappa/2}\}$,
whose complement contributes $o(\delta)$. Take $b\leq\eta/4$ and split
$\{|v+\sqrt{\delta}\Psi_{x}(v)|\leq b\}$ according to the size of $v$.

For $|v|\leq\eta$ we have $\sqrt{\delta}|\Psi^{\prime}_{x}(v)|\leq\sqrt{\delta}\,3\bar{R}(1+\eta)\leq\tfrac{1}{2}$
for $\delta$ small, so the map is strictly increasing there with derivative
in $[\tfrac{1}{2},\tfrac{3}{2}]$; its inverse is Lipschitz with constant
$2$, so the set of such $v$ is contained in an interval of length
at most $4b$, and it carries conditional probability at most $4b\Vert q\Vert_{\infty}$.
For $|v|>\eta$, the requirement $|v+\sqrt{\delta}\Psi_{x}(v)|\leq b$
forces $\sqrt{\delta}\bar{R}(1+v^{2})\geq|v|-b\geq\tfrac{3}{4}\eta=\tfrac{3}{4}\delta^{1/2-\kappa}$,
which on $\{\bar{R}\leq\delta^{-\kappa/2}\}$ is impossible unless
$|v|\geq\tfrac{1}{2}\delta^{-\kappa/4}$. Such $v$ are the distant
preimages, and they are negligible by a Gaussian tail bound: $\hat{U}_{x}$
is standard normal, so
$\Pr(|\hat{U}_{x}|\geq\tfrac{1}{2}\delta^{-\kappa/4})\leq2\exp(-\tfrac{1}{8}\delta^{-\kappa/2})$,
which is $o(\delta^{p})$ for every $p$. Taking expectations,
$\Pr(|\hat{U}_{x}+\hat{H}_{x}|\leq b\mid\mathcal{G})\leq Cb+o(\delta)$,
which is the first bound. The relevant $b$ is of order $\delta^{2/3}$,
which does satisfy $b\leq\eta/4$ for small $\delta$ since $2/3>1/2-\kappa$.
Optimizing over $b$ gives the bound $C(\mathbb{E}[\hat{e}^{2}_{x}])^{1/3}$
after taking expectations, and $\mathbb{E}[\hat{e}^{2}_{x}]\leq C\delta^{2}$
yields
\[
\Pr\big(|\hat{U}_{x}+\hat{H}_{x}|\leq|\hat{e}_{x}|\big)\leq C\delta^{2/3}=o(\delta^{1/2}),\qquad\text{uniformly in }j.
\]
This is the step that requires Assumption \ref{assu:Leverage}(iii')--(iv'):
with merely CADLAG coefficients one has $\mathbb{E}[\hat{e}^{2}_{x}]=o(\delta)$
only, and the same argument would give $o(\delta^{1/3})$, which is
not $o(\delta^{1/2})$.

Collecting (i)-(iv), (\ref{eq:CSbias}) and (\ref{eq:Aprob}), and
averaging over $j$,
\[
\frac{1}{N}\sum^{N}_{j=S}\mathbb{E}\big|\mathbb{E}[\zeta_{j}\mid\mathcal{F}_{t_{j-S}}]-\tau\big|\leq C\delta+C\delta^{1-2\kappa}+C\delta^{1-\kappa}+C\delta^{2/3}=o(\delta^{1/2}),
\]
for $\kappa<1/4$, which completes the proof.
\end{proof}
\noindent\textbf{Proof of Theorems \ref{thm:CLTBSM} and \ref{thm:CLTleverage}.}
The three steps below are common to both theorems. They invoke Lemma
\ref{lem:CondBias}, and this is the only place where the two differ:
under Assumptions \ref{assu:BSM} and \ref{assu:BSMCLT} the lemma
holds by its part (a), which is exact and elementary, and under Assumptions
\ref{assu:BSM} and \ref{assu:Leverage} it holds by its part (b).
Write
\[
\sqrt{n}(\hat{\tau}_{S}-\tau)=\sqrt{n}\Big(\frac{1}{N-S+1}\sum^{N}_{j=S}\tilde{\zeta}_{j}-\tau\Big)+\frac{\sqrt{n}}{N-S+1}\sum^{N}_{j=S}D_{j}.
\]

\textit{Step 1.} The sequence $\{\tilde{\zeta}_{j}\}$ is strictly
stationary, bounded, and $(S-1)$-dependent, such that the central
limit theorem of \citet{HoeffdingRobbins:1948} applies (see also
\citet{Berk:1973} for a triangular-array version), and
\[
\sqrt{n}\Big(\frac{1}{N-S+1}\sum^{N}_{j=S}\tilde{\zeta}_{j}-\tau\Big)\overset{d}{\rightarrow}N\big(0,\Sigma_{S}(\rho)\big),\qquad\Sigma_{S}(\rho)=\frac{4}{S}\sum^{S-1}_{h=-S+1}\Gamma(\rho,\omega_{h}),
\]
with $\omega_{h}=\frac{S-|h|}{S}$ and $\Gamma$ from Lemma \ref{lem:SignOf4};
the variance is computed exactly as in the proof of Theorem \ref{thm:efficiency},
the factor $4$ arising from $\tilde{\zeta}_{j}=2\cdot1_{\{\Delta W_{x,j}\Delta W_{y,j}>0\}}-1$.

\textit{Step 2.} We show that $\sum^{N}_{j=S}D_{j}=o_{p}(\sqrt{NS})$,
which suffices because $\sqrt{n}/(N-S+1)\leq C(NS)^{-1/2}$. Decompose
$D_{j}=b_{j}+m_{j}$ with $b_{j}=\mathbb{E}[D_{j}\mid\mathcal{F}_{t_{j-S}}]$.
For the bias part, Lemma \ref{lem:CondBias} gives
\[
\mathbb{E}\Big|\sum{}_{j}b_{j}\Big|\leq\sum{}_{j}\mathbb{E}|b_{j}|=N\cdot o(\delta^{1/2})=o(\sqrt{NS}).
\]
For the martingale part, split the index set by phase: for each $s\in\{0,\ldots,S-1\}$,
the subsequence $\{m_{j}:j\equiv s\ (\mathrm{mod}\ S)\}$ is a martingale
difference sequence with respect to $\{\mathcal{F}_{t_{j}}\}$, because
consecutive indices in a phase correspond to non-overlapping windows,
such that $m_{j^{\prime}}$ is $\mathcal{F}_{t_{j-S}}$-measurable
for $j^{\prime}<j$ in the same phase, while $\mathbb{E}[m_{j}\mid\mathcal{F}_{t_{j-S}}]=0$
by construction. Hence, using orthogonality within phases, Lemma \ref{lem:Flip},
and the Cauchy-Schwarz inequality across the $S$ phases,
\[
\mathbb{E}\Big[\Big(\sum{}_{j}m_{j}\Big)^{2}\Big]\leq S\sum^{S-1}_{s=0}\sum_{j\equiv s}\mathbb{E}[m^{2}_{j}]\leq S\cdot S\cdot\frac{N}{S}\cdot C\delta^{1/3}=CNS\delta^{1/3}=o(NS).
\]
Thus $\sum_{j}D_{j}=o_{p}(\sqrt{NS})$, and $\sqrt{n}(\hat{\tau}_{S}-\tau)\overset{d}{\rightarrow}N(0,\Sigma_{S}(\rho))$.

\textit{Step 3.} Since $Q_{S}=\sin(\tfrac{\pi}{2}\hat{\tau}_{S})$
and $\tfrac{\pi}{2}\cos(\tfrac{\pi}{2}\tau)=\tfrac{\pi}{2}\sqrt{1-\rho^{2}}\neq0$
for $|\rho|<1$, the delta method yields
\[
\sqrt{n}(Q_{S}-\rho)\overset{d}{\rightarrow}N\Big(0,\tfrac{\pi^{2}}{4}(1-\rho^{2})\Sigma_{S}(\rho)\Big),\qquad\tfrac{\pi^{2}}{4}(1-\rho^{2})\Sigma_{S}(\rho)=(1-\rho^{2})\frac{1}{S}\sum^{S-1}_{h=-S+1}\big[\mathrm{asin}^{2}(\omega_{h})-\mathrm{asin}^{2}(\omega_{h}\rho)\big],
\]
which is $V_{S}(\rho)$ of Theorem \ref{thm:efficiency}. \hfill{}$\square$

\medskip{}

\subsection*{Zero increments}

The estimator used in the simulations and in the empirical analysis
discards zero increments, giving $\hat{\tau}^{\mathrm{nz}}_{S}$ of
Section \ref{subsubsec:implementation}. It is otherwise $\hat{\tau}_{S}$: the increments are
not recentered, so that the statistic to which the theorems apply is
the one that is computed. We show that discarding zeros does not
affect the limit distribution. The argument uses a small-ball
estimate for a single increment, which is available from the machinery
already assembled for Lemma \ref{lem:CondBias}(b).
\begin{lem}
\label{lem:SmallBall}Suppose that Assumption \ref{assu:BSM} holds,
with the localization above in force, that $|\rho|<1$, and that
either Assumption \ref{assu:BSMCLT} or Assumption \ref{assu:Leverage}
holds. Let $\kappa\in(0,1/4)$ be the constant fixed in the proof of
Lemma \ref{lem:CondBias}(b). Then there is a constant $C<\infty$,
not depending on $j$ or $\delta$, such that for every $b\in(0,\tfrac{c}{16}\delta^{1/2-\kappa}]$
and all sufficiently small $\delta$,
\begin{equation}
\Pr\big(|\Delta X_{j}|\leq b\sqrt{\delta}\big)\leq C\big(b+\delta^{2}b^{-2}+\delta\big),\label{eq:smallballup}
\end{equation}
and the same with $Y$ in place of $X$. Under Assumption \ref{assu:BSMCLT}
the bound holds for every $b>0$. In particular, taking $b=\delta^{2/3}$,
which balances the first two terms and lies in the admissible range
for small $\delta$ because $2/3>1/2-\kappa$,
\begin{equation}
\Pr(\Delta X_{j}\Delta Y_{j}=0)\leq C\delta^{2/3},\qquad\text{uniformly in }j.\label{eq:zeroprob}
\end{equation}
\end{lem}
\begin{proof}
Under Assumption \ref{assu:BSMCLT} we have $\nu_{zb}\equiv0$, hence
$\hat{H}\equiv0$, and (\ref{eq:smallballup}) follows from the
conditional Gaussianity of $\Delta X_{j}$ established in part (a) of
the proof of Lemma \ref{lem:CondBias}. We therefore argue under
Assumption \ref{assu:Leverage} and adopt the notation of part (b): on
$\mathcal{A}_{j}$, with $\Pr(\mathcal{A}^{c}_{j})\leq C\delta$ by
(\ref{eq:Aprob}),
\[
\Delta X_{j}=\Vert\hat{\sigma}_{x}\Vert(\hat{U}_{x}+\hat{H}_{x}+\hat{e}_{x}),\qquad\tfrac{c}{2}\sqrt{\delta}\leq\Vert\hat{\sigma}_{x}\Vert\leq C\sqrt{\delta},
\]
where $(\hat{U}_{x},\hat{U}_{y})$ is standard bivariate Gaussian given
$\mathcal{G}$ with correlation $\hat{\rho}_{j}$, $|\hat{\rho}_{j}|\leq|\rho|<1$,
and $\mathbb{E}[\hat{e}^{2}_{x}]\leq C\delta^{2}$. On $\mathcal{A}_{j}$
the event $\{|\Delta X_{j}|\leq b\sqrt{\delta}\}$ implies
$|\hat{U}_{x}+\hat{H}_{x}+\hat{e}_{x}|\leq2b/c$, and
\[
\Pr\big(|\hat{U}_{x}+\hat{H}_{x}+\hat{e}_{x}|\leq\tfrac{2b}{c}\big)\leq\Pr\big(|\hat{U}_{x}+\hat{H}_{x}|\leq\tfrac{4b}{c}\big)+\Pr\big(|\hat{e}_{x}|\geq\tfrac{2b}{c}\big)\leq Cb+o(\delta)+C\delta^{2}b^{-2},
\]
the first term by the anti-concentration bound of step (iv) of the
proof of Lemma \ref{lem:CondBias}(b), applied at the threshold $4b/c$,
which is admissible there when $4b/c\leq\eta/4$, that is, $b\leq\tfrac{c}{16}\delta^{1/2-\kappa}$;
this is where the restriction on $b$ enters. The $o(\delta)$ term
is the contribution of the exceptional events excluded in that step.
The second term is Chebyshev's inequality. Adding $\Pr(\mathcal{A}^{c}_{j})\leq C\delta$
and absorbing the $o(\delta)$ term into $C\delta$ gives (\ref{eq:smallballup}); applying it to $X$ and to $Y$ with
$b=\delta^{2/3}$ and using the union bound gives (\ref{eq:zeroprob}).
\end{proof}
Lemma \ref{lem:SmallBall} does not require the increments to be atomless,
which we cannot establish when $\sigma$ loads on $W$; it only requires
that zeros be rare enough not to matter at the $\sqrt{n}$ scale.
\begin{cor}
\label{cor:nonzero}Under the assumptions of Lemma \ref{lem:SmallBall},
$\sqrt{n}(\hat{\tau}^{\mathrm{nz}}_{S}-\hat{\tau}_{S})\overset{p}{\rightarrow}0$.
Consequently, Theorems \ref{thm:CLTBSM} and \ref{thm:CLTleverage}
hold unchanged with $Q^{\mathrm{nz}}_{S}$ in place of $Q_{S}$.
\end{cor}
\begin{proof}
Write $M=N-S+1$. By (\ref{eq:zeroprob}),
$\mathbb{E}[M-N_{1}]=\sum_{j}\Pr(\Delta X_{j}\Delta Y_{j}=0)\leq CM\delta^{2/3}$,
so $(M-N_{1})/M=O_{p}(\delta^{2/3})$ by Markov's inequality and, in
particular, $N_{1}/M\overset{p}{\rightarrow}1$. Since
$\hat{\tau}^{\mathrm{nz}}_{S}=(M/N_{1})\hat{\tau}_{S}$ and
$|\hat{\tau}_{S}|\leq1$,
\[
|\hat{\tau}^{\mathrm{nz}}_{S}-\hat{\tau}_{S}|=|\hat{\tau}_{S}|\frac{M-N_{1}}{N_{1}}=O_{p}(\delta^{2/3}),
\]
whence $\sqrt{n}|\hat{\tau}^{\mathrm{nz}}_{S}-\hat{\tau}_{S}|=O_{p}(n^{1/2}n^{-2/3})=O_{p}(n^{-1/6})\rightarrow0$,
using $\delta=1/n$. The conclusion follows from Slutsky's theorem and
the delta method of Step 3.
\end{proof}

{\footnotesize\bibliographystyle{agsm}
\bibliography{prh}
}{\footnotesize\par}

\clearpage\setcounter{equation}{0}\renewcommand{\theequation}{S.\arabic{equation}}
\setcounter{prop}{0}\renewcommand{\theprop}{S.\arabic{prop}}
\setcounter{section}{0}\renewcommand{\thesection}{S.\arabic{section}}
\setcounter{figure}{0}\renewcommand{\thefigure}{S.\arabic{figure}}
\setcounter{table}{0}\renewcommand{\thetable}{S.\arabic{table}}
\setcounter{page}{1}\renewcommand{\thepage}{S.\arabic{page}}

\part*{Supplementary Material}

\section{Supplementary Empirical Results}

\subsection{Sensitivity to the Rolling-Window Length}

The intraday estimates in Section \ref{sec:Empirical-intraday} use
a $60$-minute rolling window. Section \ref{subsec:WindowLength} examines
the trade-off in simulations; here we verify that the empirical patterns
do not depend on the choice. Figure \ref{fig:WindowSensitivityEmpirical}
reproduces the intraday correlation, relative volatility and market
beta estimates for the four representative pairs of Figure \ref{fig:RhoLambdaBetaIntraday},
namely SPY paired with AAPL, AMD, LYB, and NEM, using windows of $30$,
$60$, and $90$ minutes.

The three curves are not estimates of the same quantity, so we do
not expect them to coincide. The features interpreted in the main
text, the level and slope of the correlation profile, the shape of
the relative-volatility pattern, and the direction of the intraday
movement in beta, are the same for all three windows. The $30$-minute
estimates are noisier and the $90$-minute estimates smooth the opening
and closing periods. All windows are centered on the reporting time,
see Section \ref{subsec:WindowLength}.

\begin{figure}[H]
\begin{centering}
\subfloat[Correlation]{\begin{centering}
\includegraphics[width=0.24\textwidth]{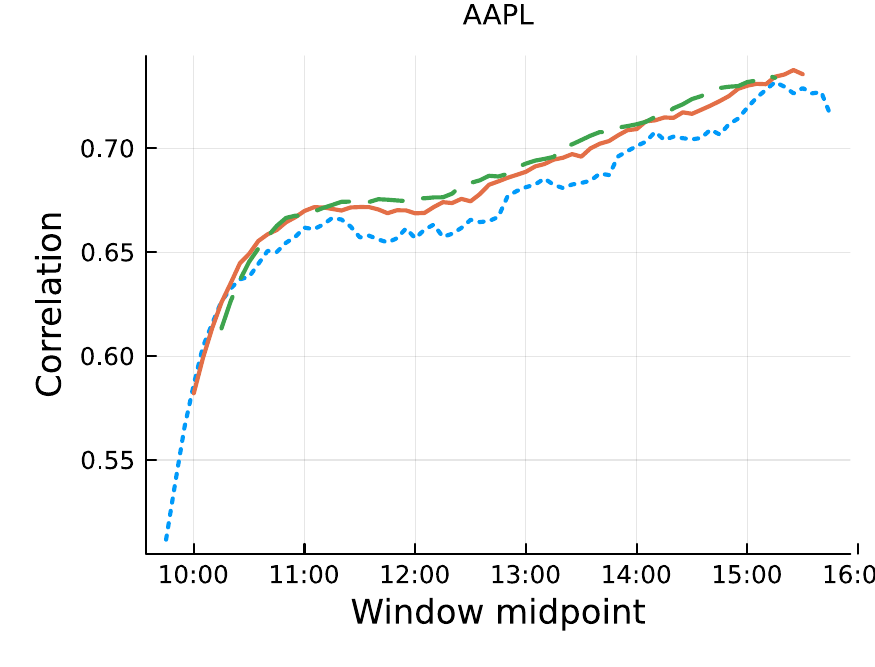}\includegraphics[width=0.24\textwidth]{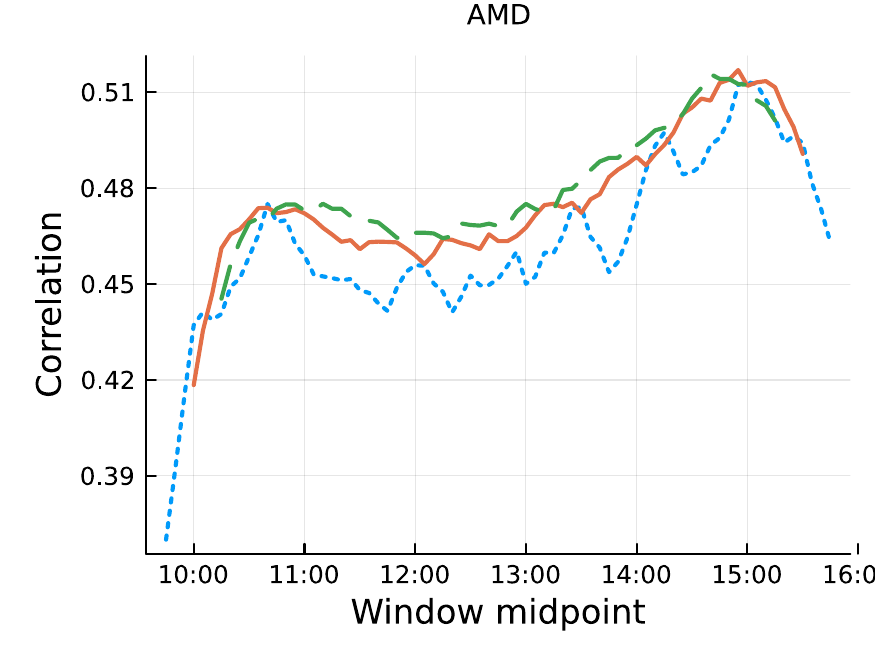}\includegraphics[width=0.24\textwidth]{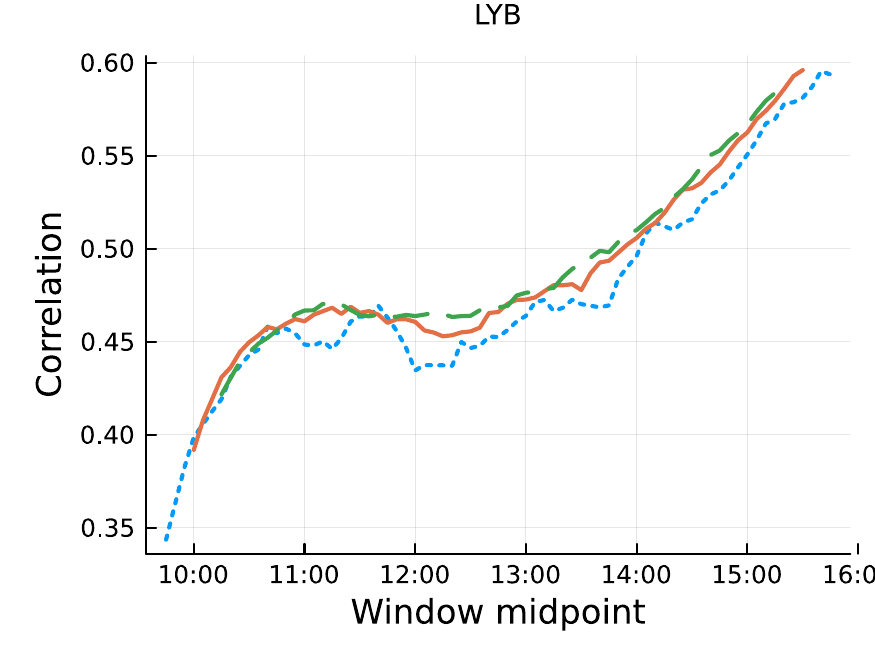}\includegraphics[width=0.24\textwidth]{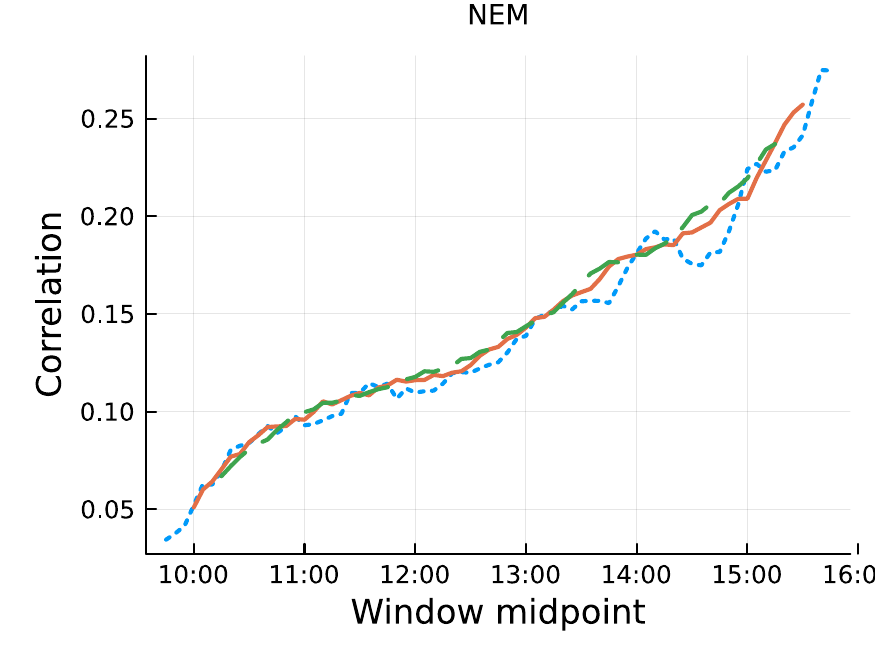}
\par\end{centering}
}
\par\end{centering}
\begin{centering}
\subfloat[Relative volatility]{\begin{centering}
\includegraphics[width=0.24\textwidth]{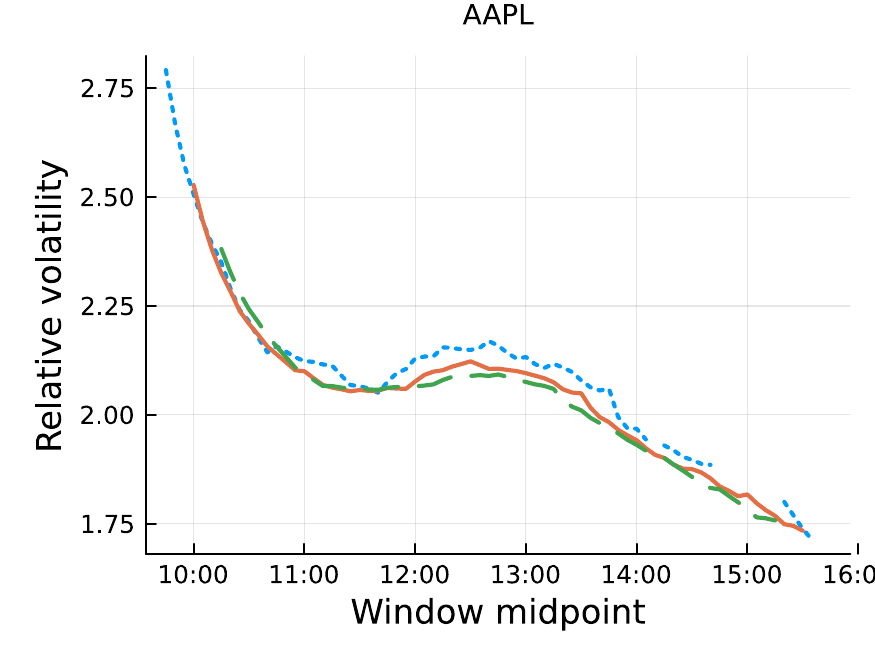}\includegraphics[width=0.24\textwidth]{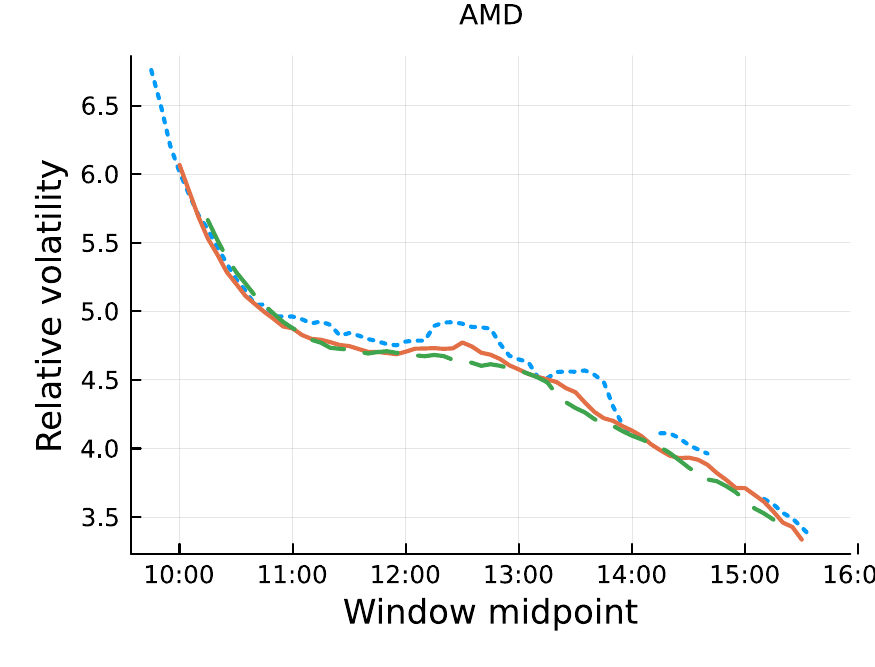}\includegraphics[width=0.24\textwidth]{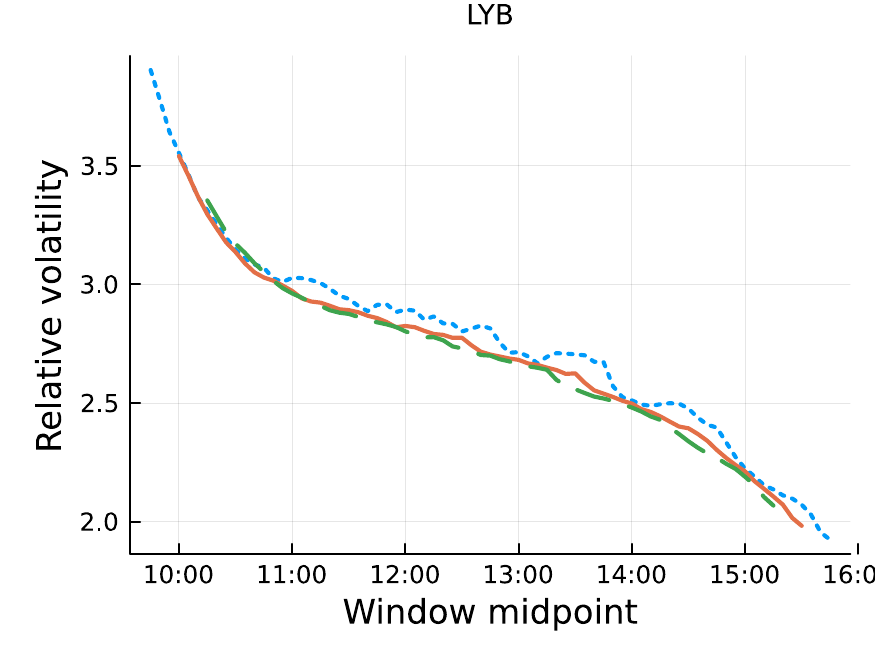}\includegraphics[width=0.24\textwidth]{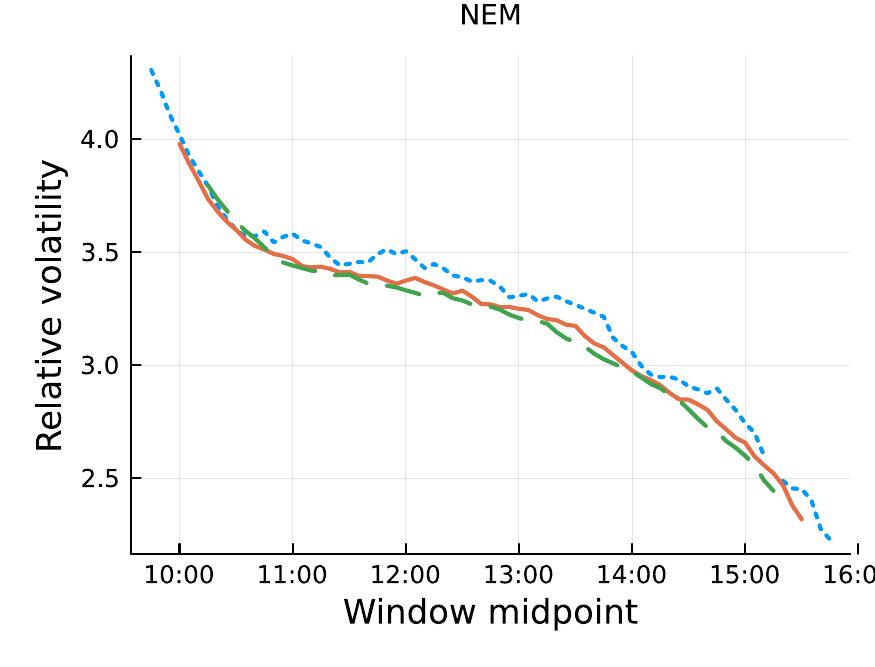}
\par\end{centering}
}
\par\end{centering}
\begin{centering}
\subfloat[Market beta]{\begin{centering}
\includegraphics[width=0.24\textwidth]{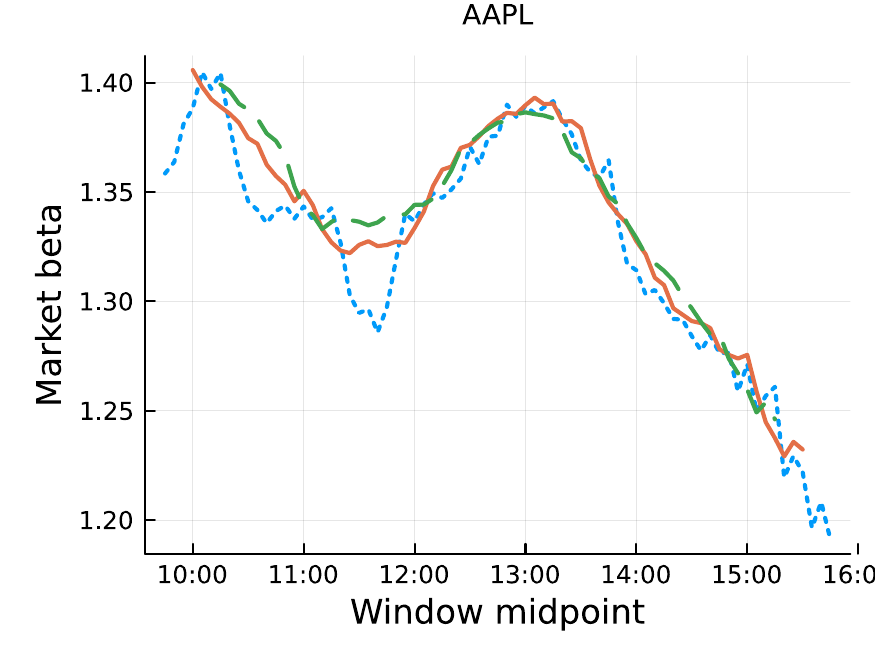}\includegraphics[width=0.24\textwidth]{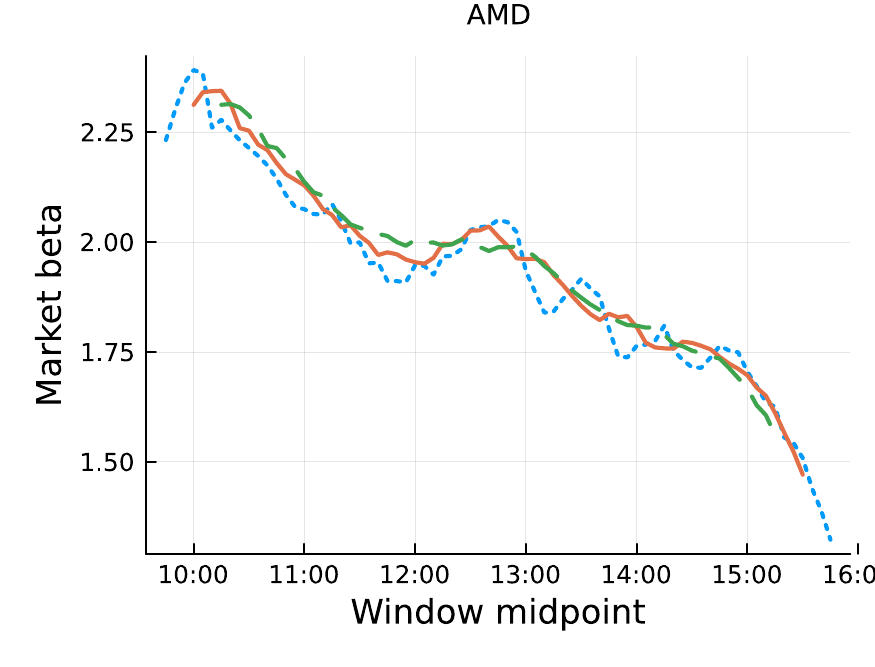}\includegraphics[width=0.24\textwidth]{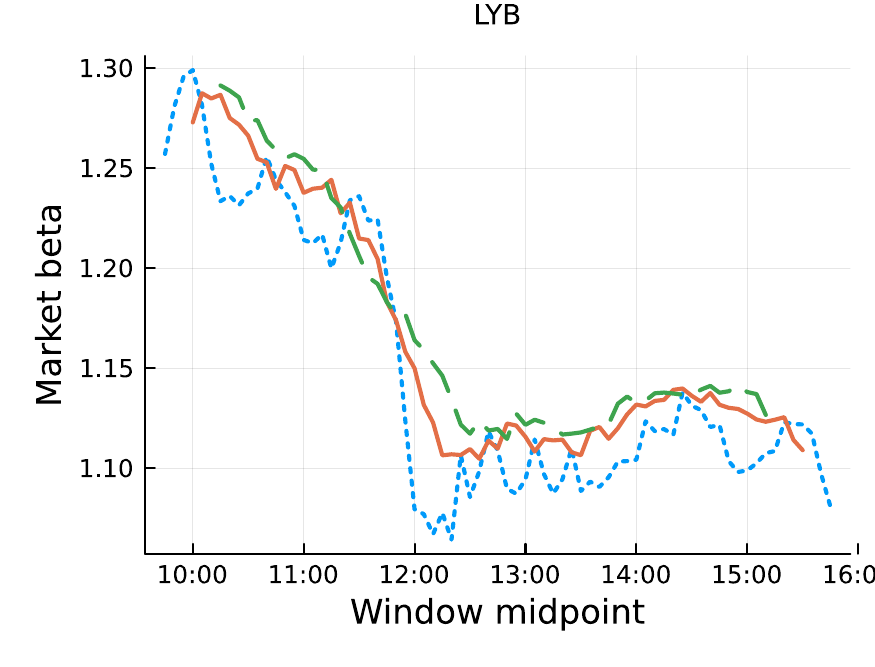}\includegraphics[width=0.24\textwidth]{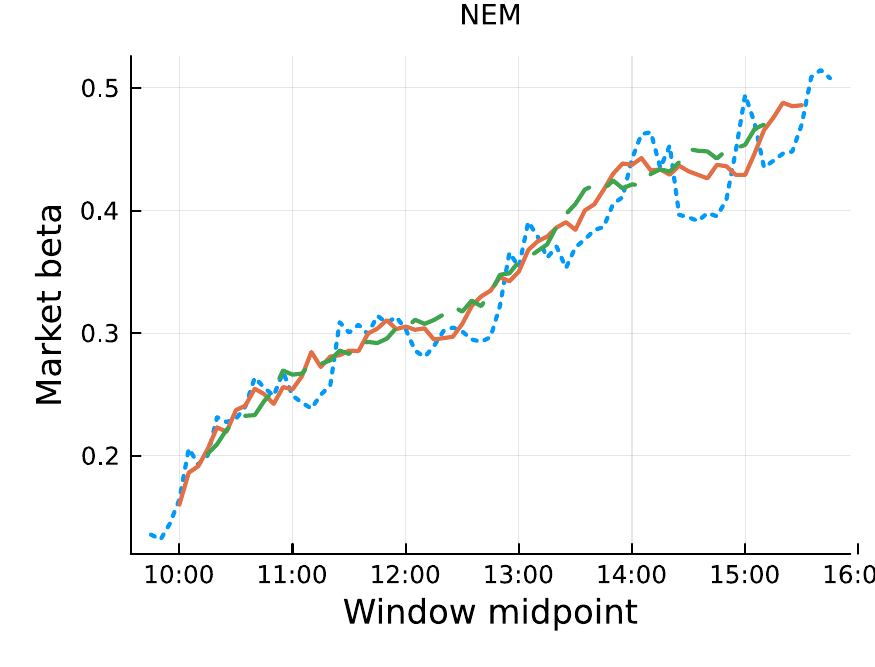}
\par\end{centering}
}
\par\end{centering}
\caption{Intraday correlation, relative volatility and market beta with respect
to SPY, for AAPL, AMD, LYB, and NEM (left to right). Estimates are
averages over the sample period, based on $3$-minute returns and the
subsampled quadrant estimator, computed with centered rolling windows
of $30$ minutes (dotted), $60$ minutes (solid), and $90$ minutes (dashed).
Each series is plotted against the midpoint of its window, so the
shorter windows extend closer to the open and the close: the $30$-minute
estimates span 9:45 to 15:45 and the $90$-minute estimates only 10:15
to 15:15.\label{fig:WindowSensitivityEmpirical}}
\end{figure}

\subsection{Signature Plots for the Small Universe}

We presented correlation signature plots for selected pairs of
assets in Figure \ref{fig:CorrSignaturePlots}. Here we present the
complete set of signature plots for all assets in the Small Universe,
see Figure \ref{fig:SignaturePlotsSmallUniverse}.
\begin{figure}[H]
\begin{centering}
\subfloat[$\mathrm{corr}$(SPY,D), $\mathrm{corr}$(SPY,DUK), $\mathrm{corr}$(D,DUK)]{\begin{centering}
\includegraphics[width=0.16\textwidth]{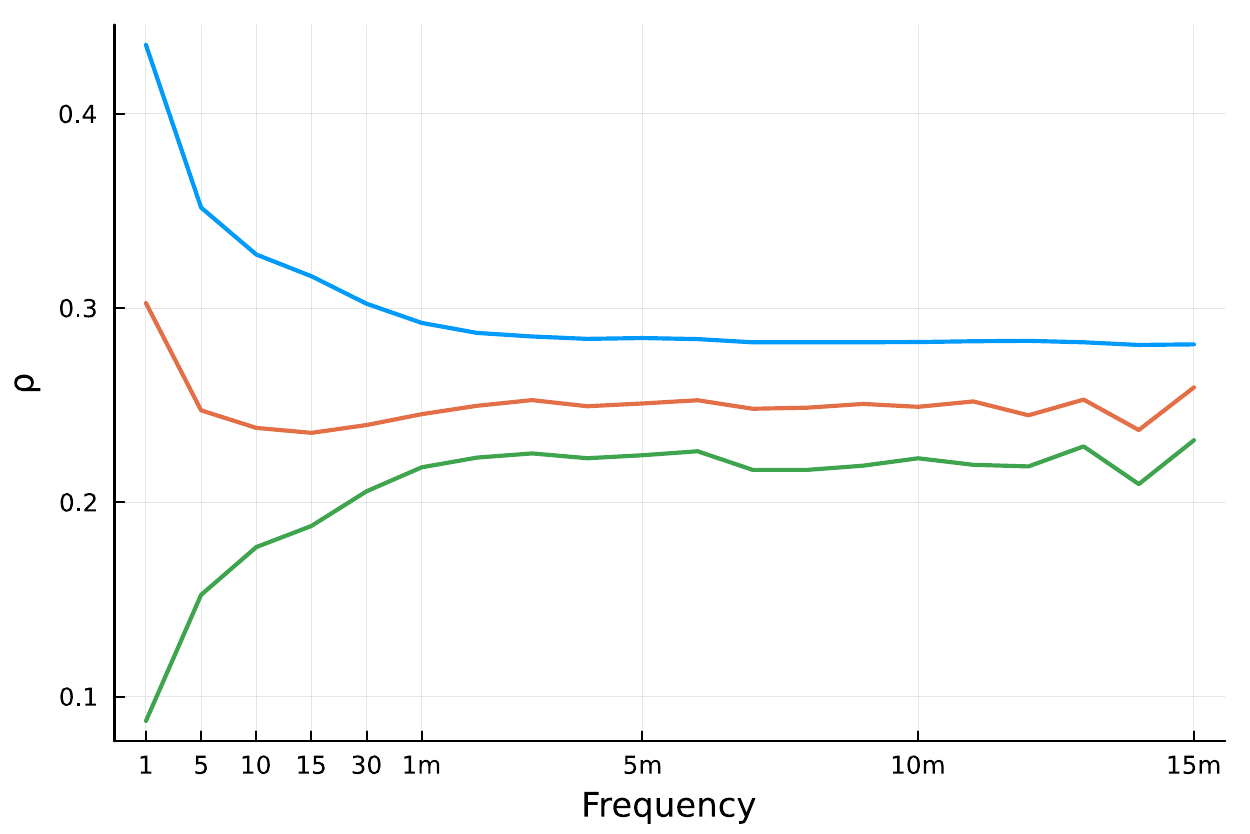}\includegraphics[width=0.16\textwidth]{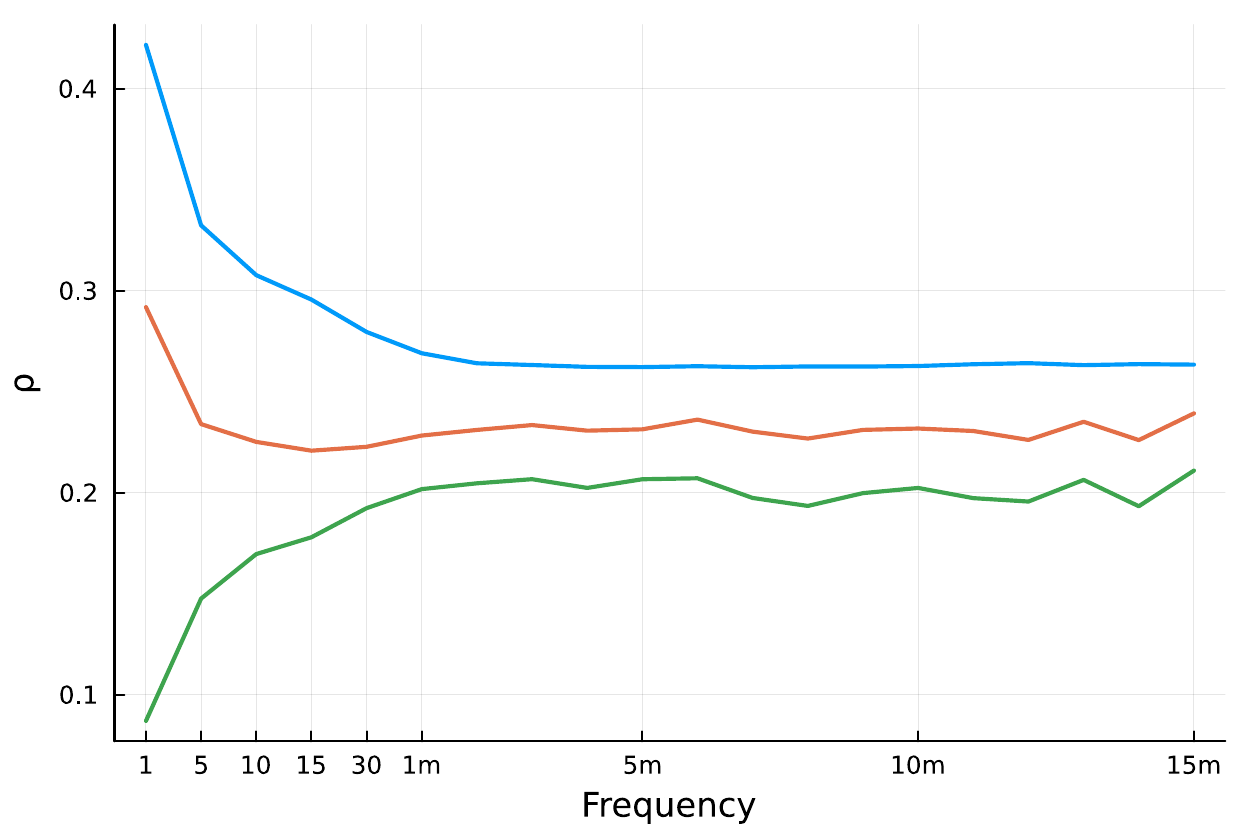}\includegraphics[width=0.16\textwidth]{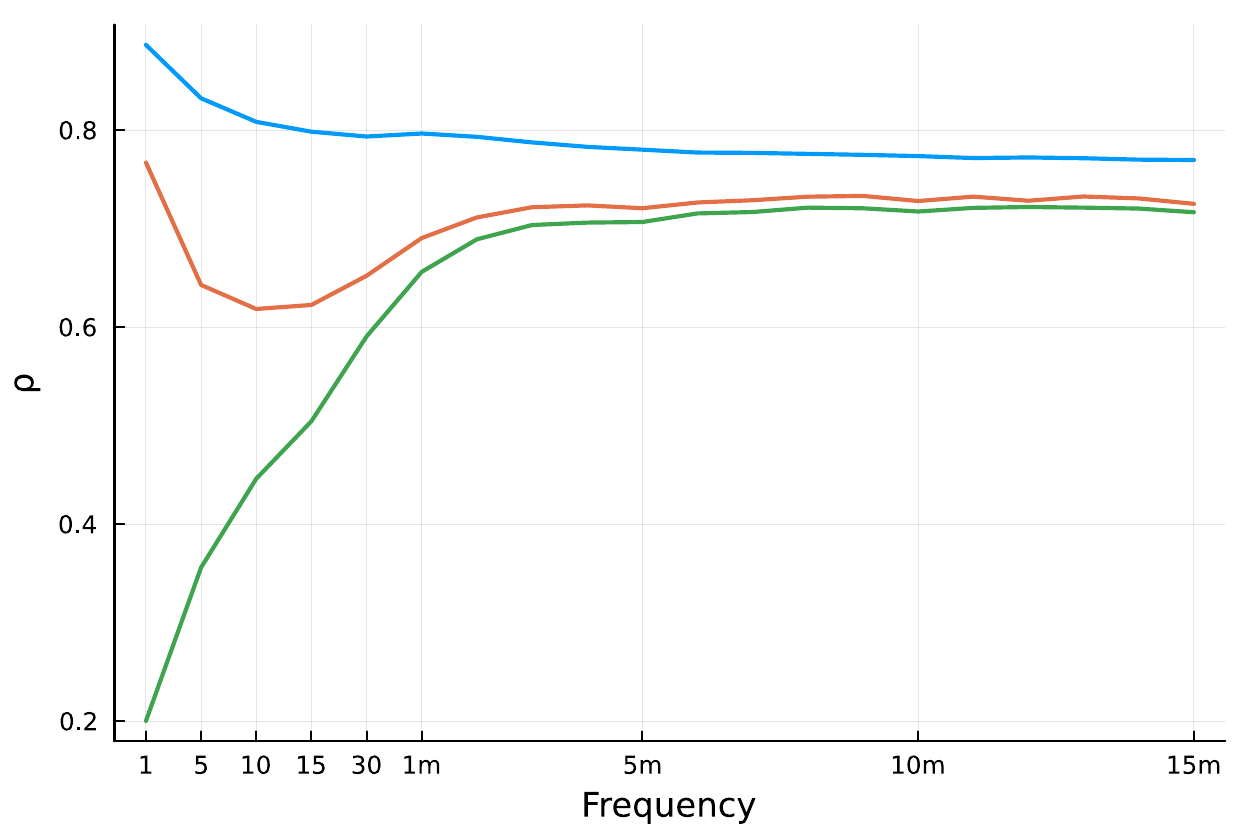}
\par\end{centering}
}\subfloat[$\mathrm{corr}$(SPY,AMT), $\mathrm{corr}$(SPY,PLD), $\mathrm{corr}$(AMT,PLD)]{\begin{centering}
\includegraphics[width=0.16\textwidth]{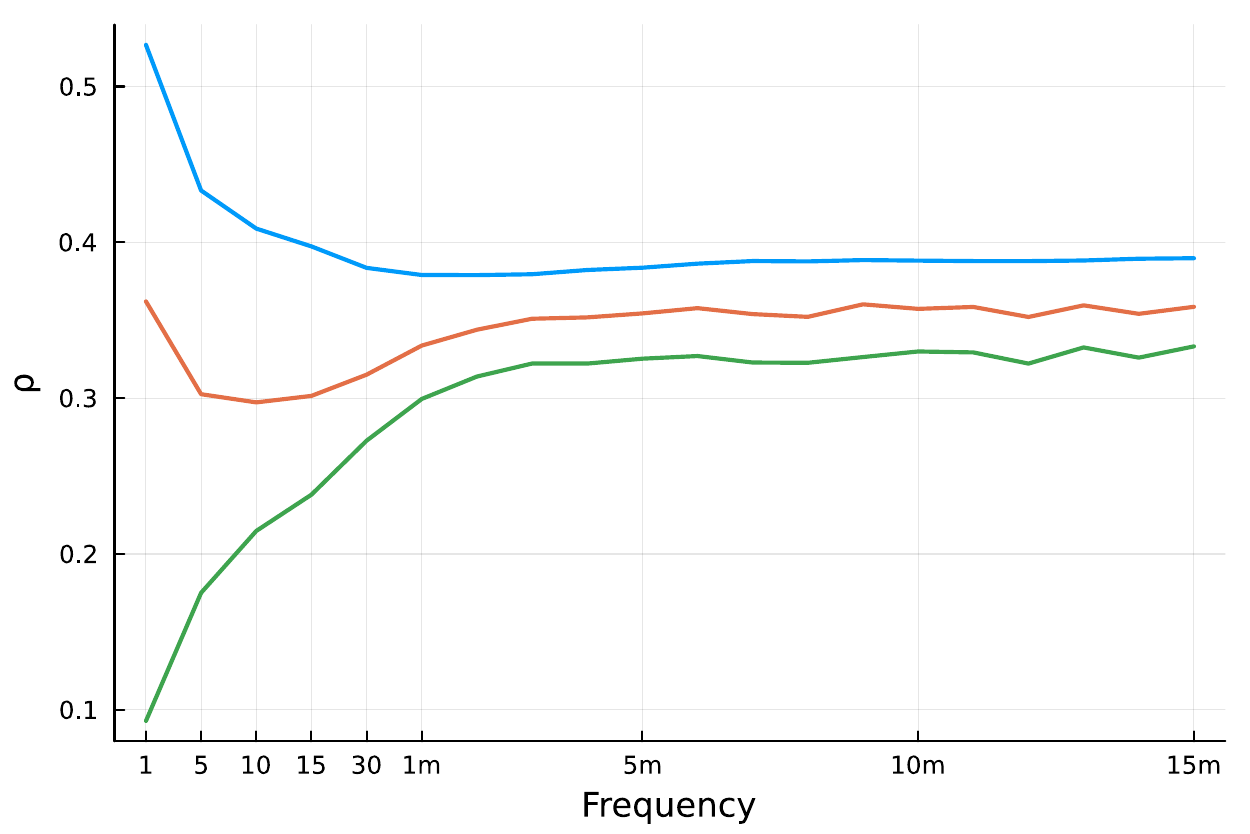}\includegraphics[width=0.16\textwidth]{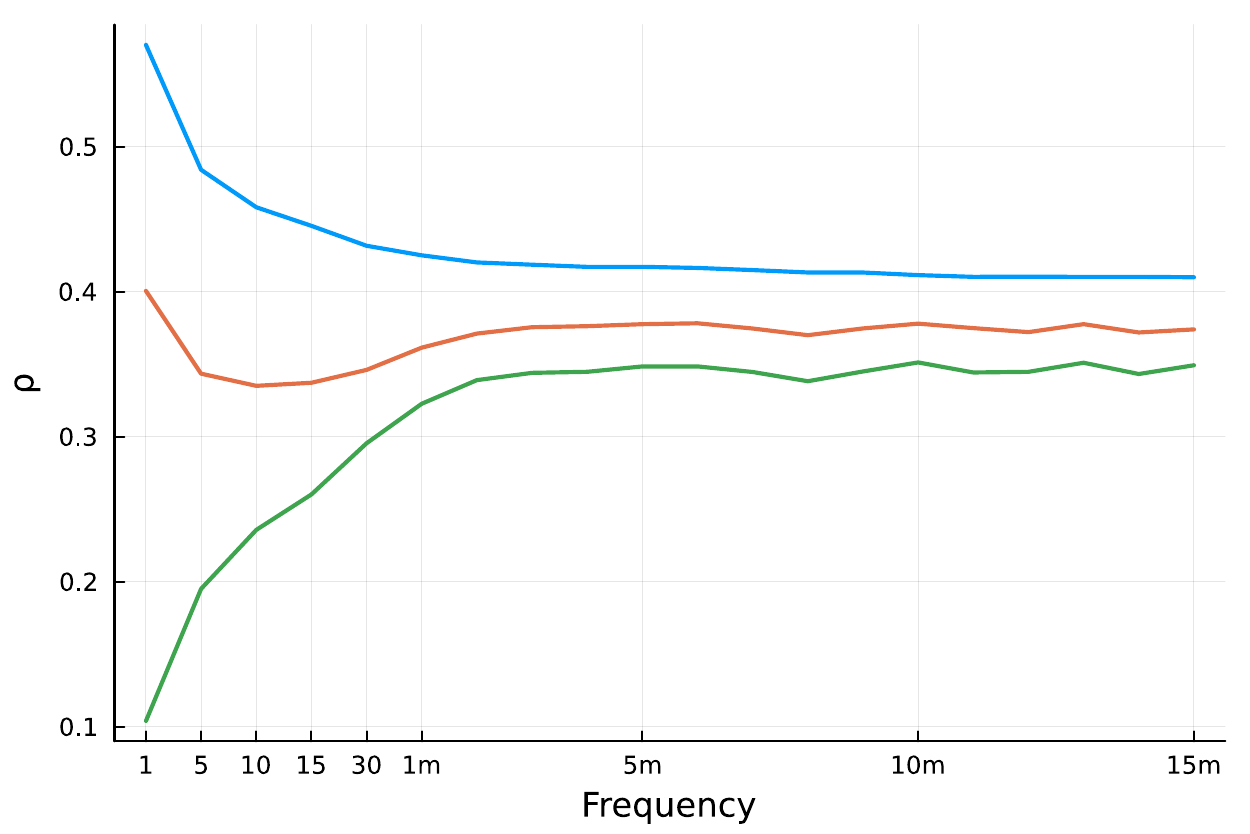}\includegraphics[width=0.16\textwidth]{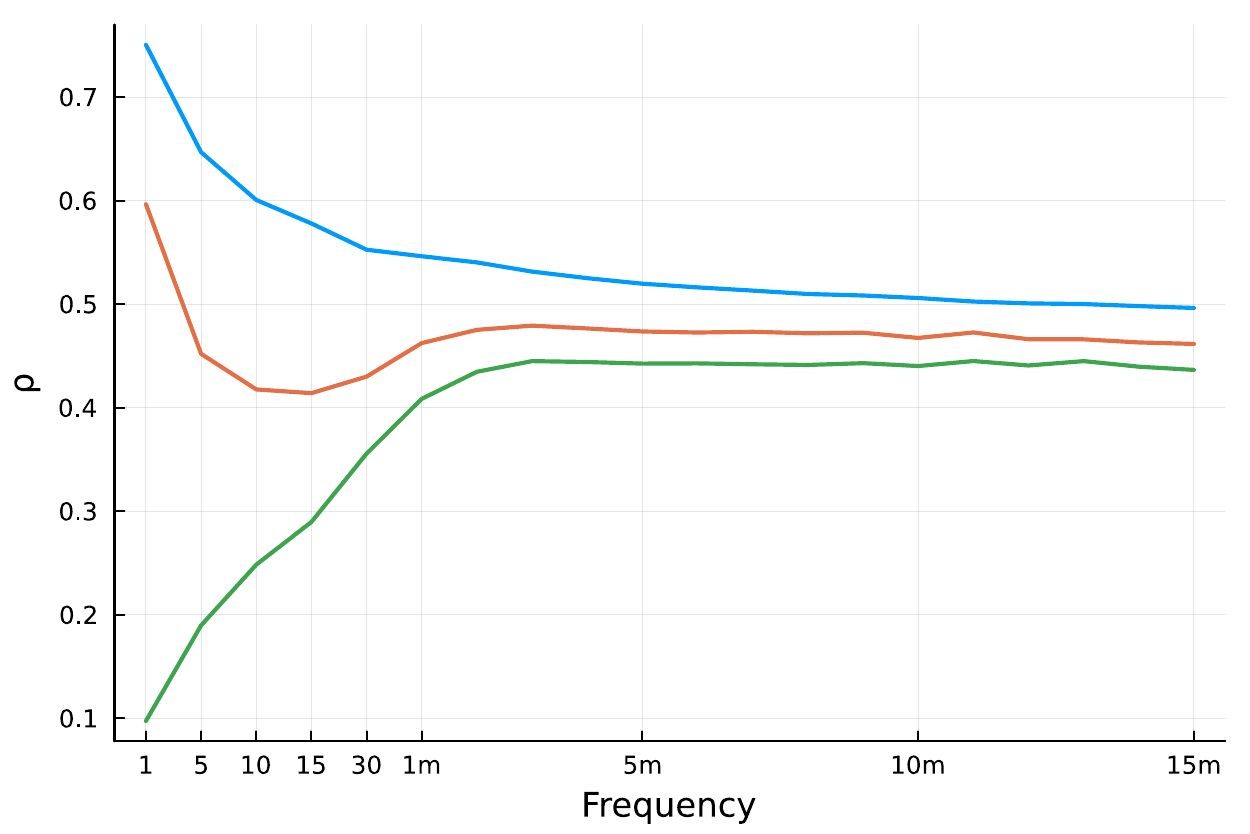}
\par\end{centering}
}
\par\end{centering}
\begin{centering}
\subfloat[$\mathrm{corr}$(SPY,LYB), $\mathrm{corr}$(SPY,NEM), $\mathrm{corr}$(LYB,NEM)]{\begin{centering}
\includegraphics[width=0.16\textwidth]{figures/SPY_LYB_cts.pdf}\includegraphics[width=0.16\textwidth]{figures/SPY_NEM_cts.pdf}\includegraphics[width=0.16\textwidth]{figures/LYB_NEM_cts.pdf}
\par\end{centering}
}\subfloat[$\mathrm{corr}$(SPY,AAPL), $\mathrm{corr}$(SPY,AMD), $\mathrm{corr}$(AAPL,AMD)]{\begin{centering}
\includegraphics[width=0.16\textwidth]{figures/SPY_AAPL_cts.pdf}\includegraphics[width=0.16\textwidth]{figures/SPY_AMD_cts.pdf}\includegraphics[width=0.16\textwidth]{figures/AAPL_AMD_cts.pdf}
\par\end{centering}
}
\par\end{centering}
\begin{centering}
\subfloat[$\mathrm{corr}$(SPY,AAL), $\mathrm{corr}$(SPY,UNP), $\mathrm{corr}$(AAL,UNP)]{\begin{centering}
\includegraphics[width=0.16\textwidth]{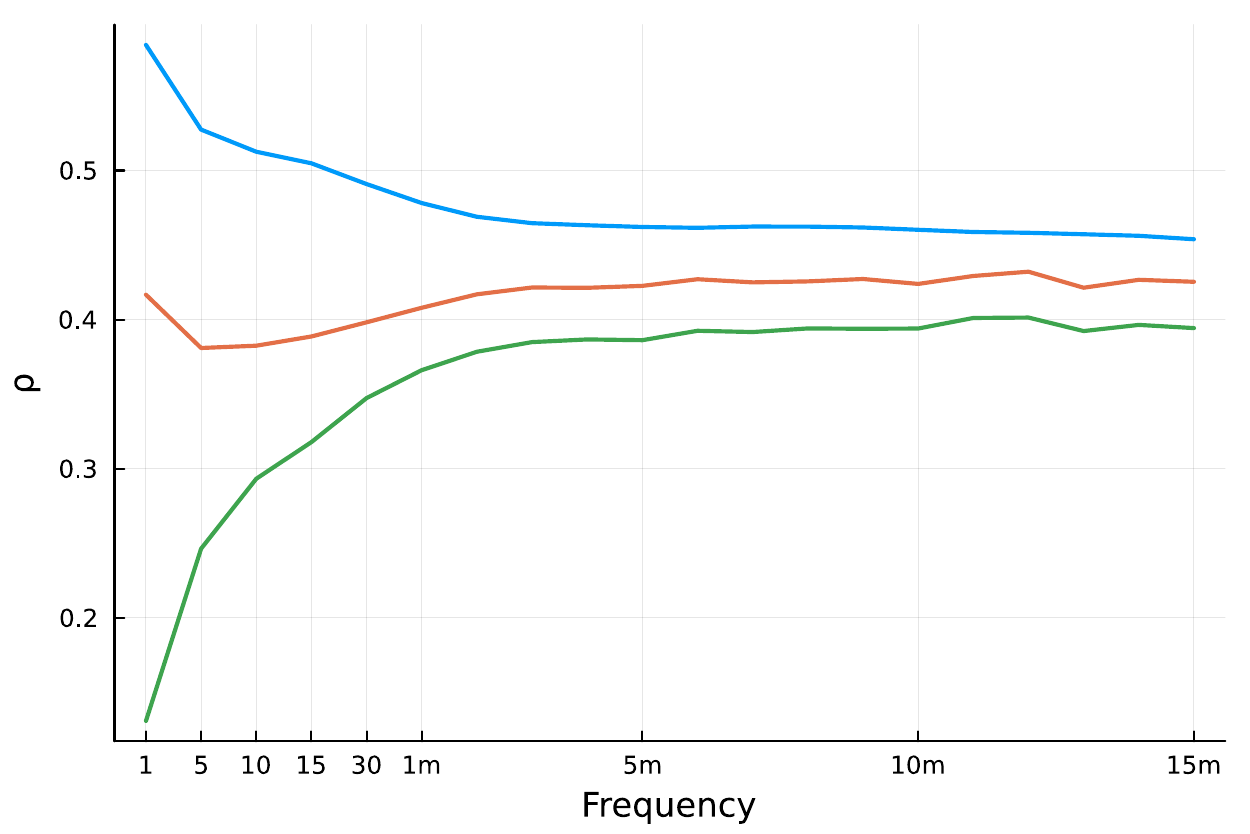}\includegraphics[width=0.16\textwidth]{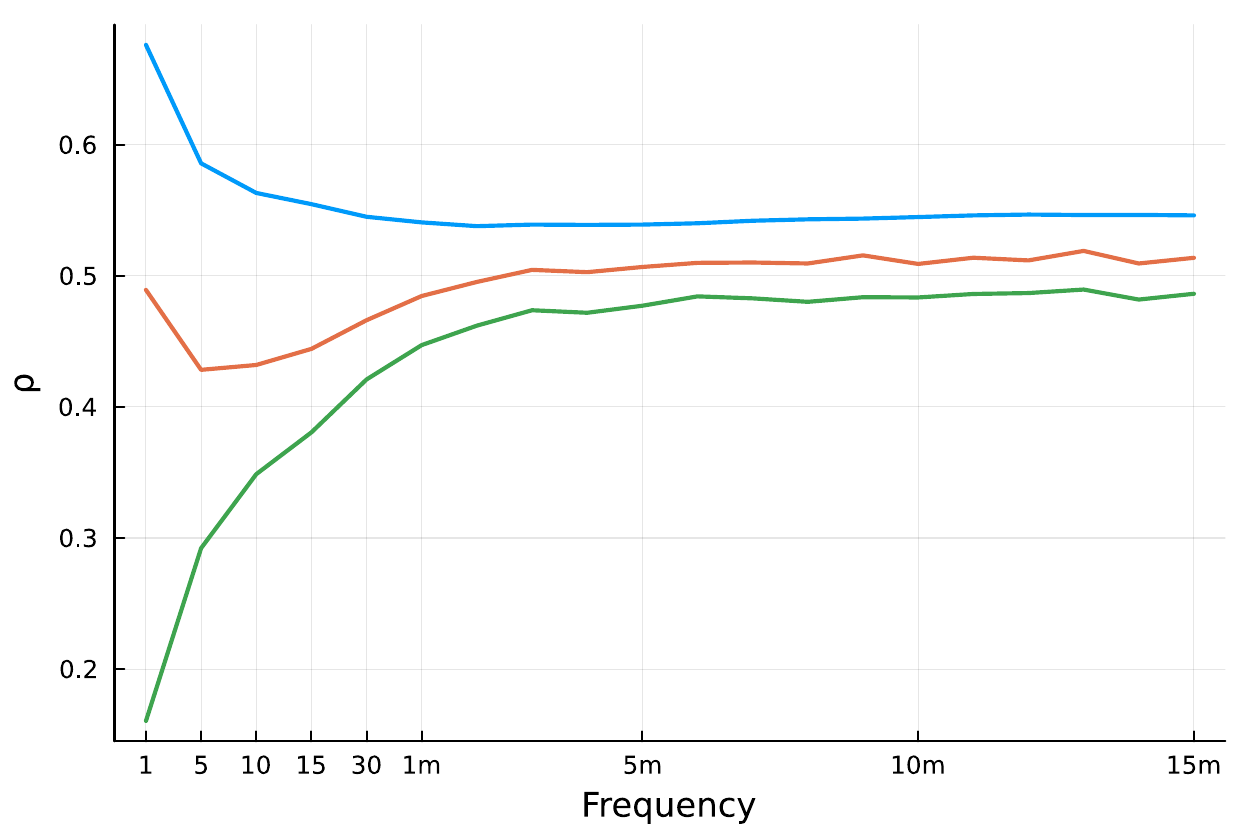}\includegraphics[width=0.16\textwidth]{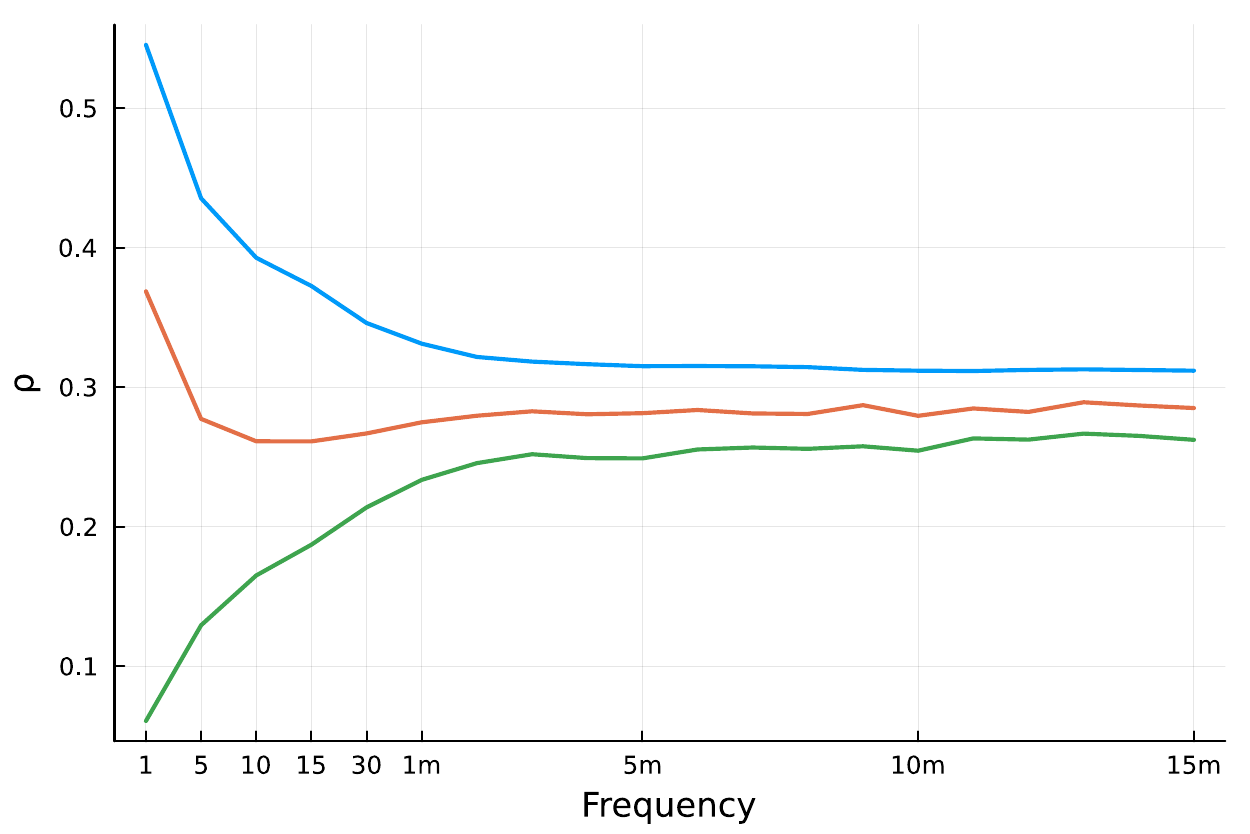}
\par\end{centering}
}\subfloat[$\mathrm{corr}$(SPY,JNJ), $\mathrm{corr}$(SPY,MRK), $\mathrm{corr}$(JNJ,MRK)]{\begin{centering}
\includegraphics[width=0.16\textwidth]{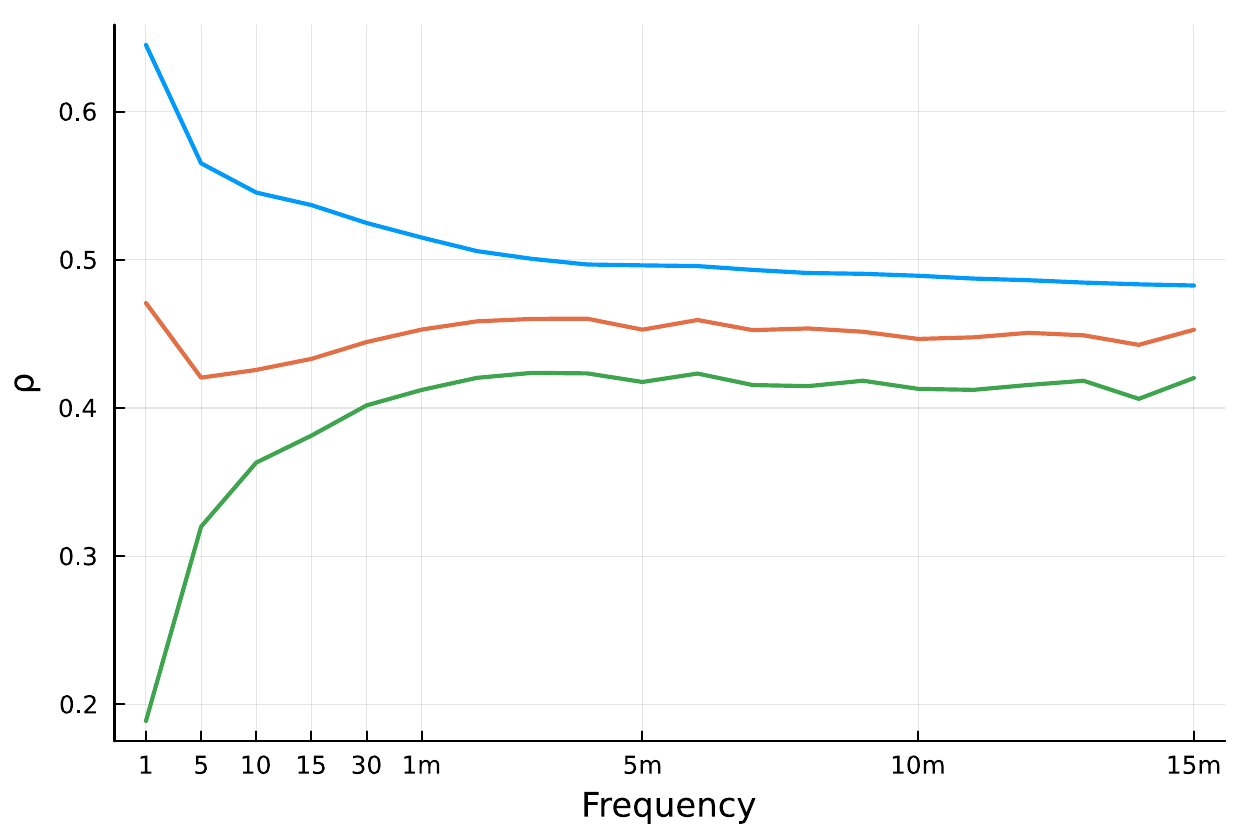}\includegraphics[width=0.16\textwidth]{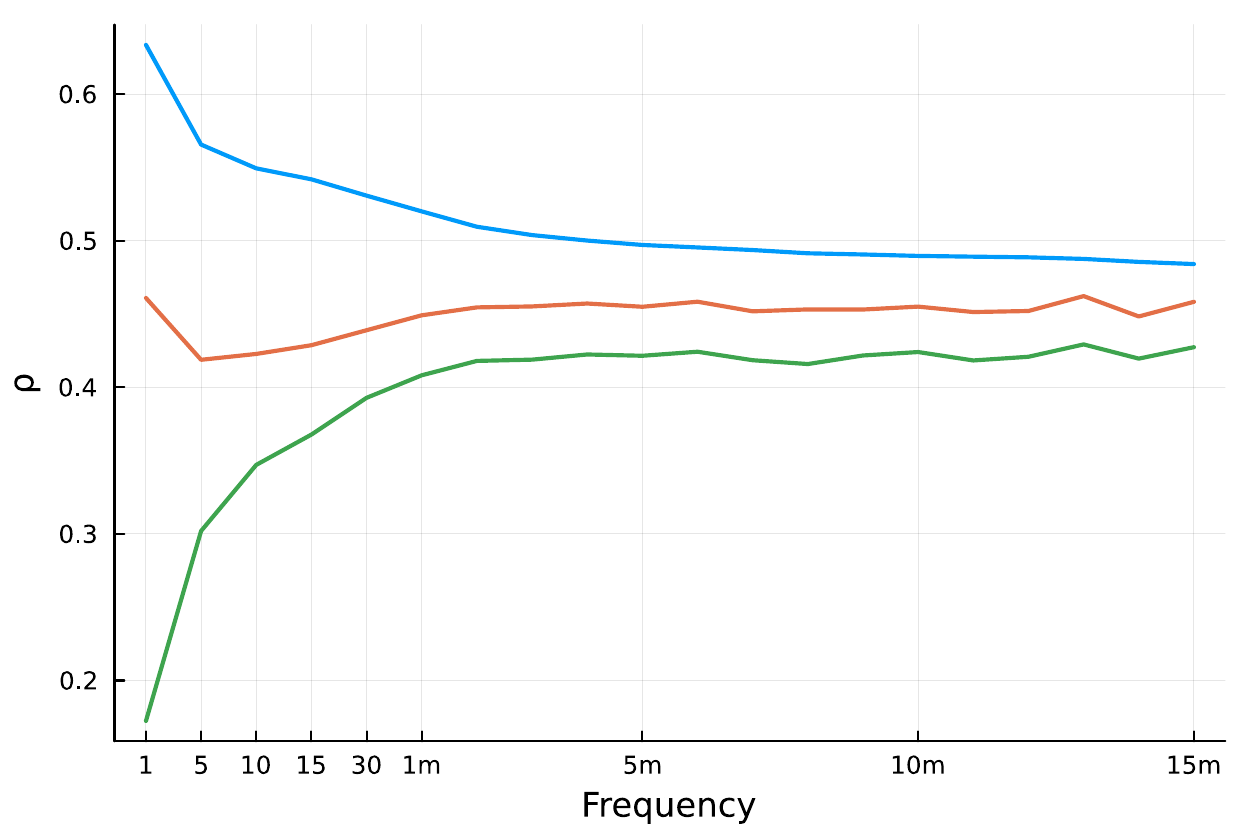}\includegraphics[width=0.16\textwidth]{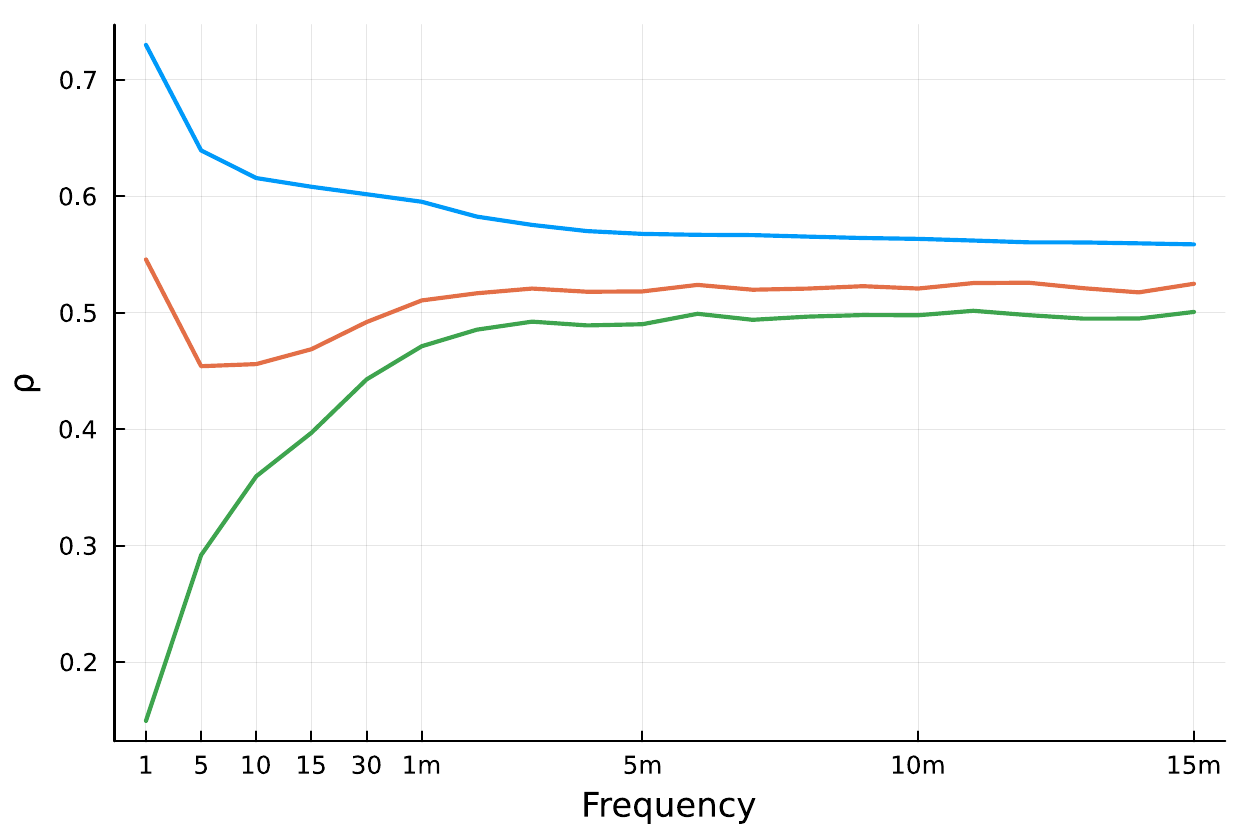}
\par\end{centering}
}
\par\end{centering}
\begin{centering}
\subfloat[$\mathrm{corr}$(SPY,JPM), $\mathrm{corr}$(SPY,WFC), $\mathrm{corr}$(JPM,WFC)]{\begin{centering}
\includegraphics[width=0.16\textwidth]{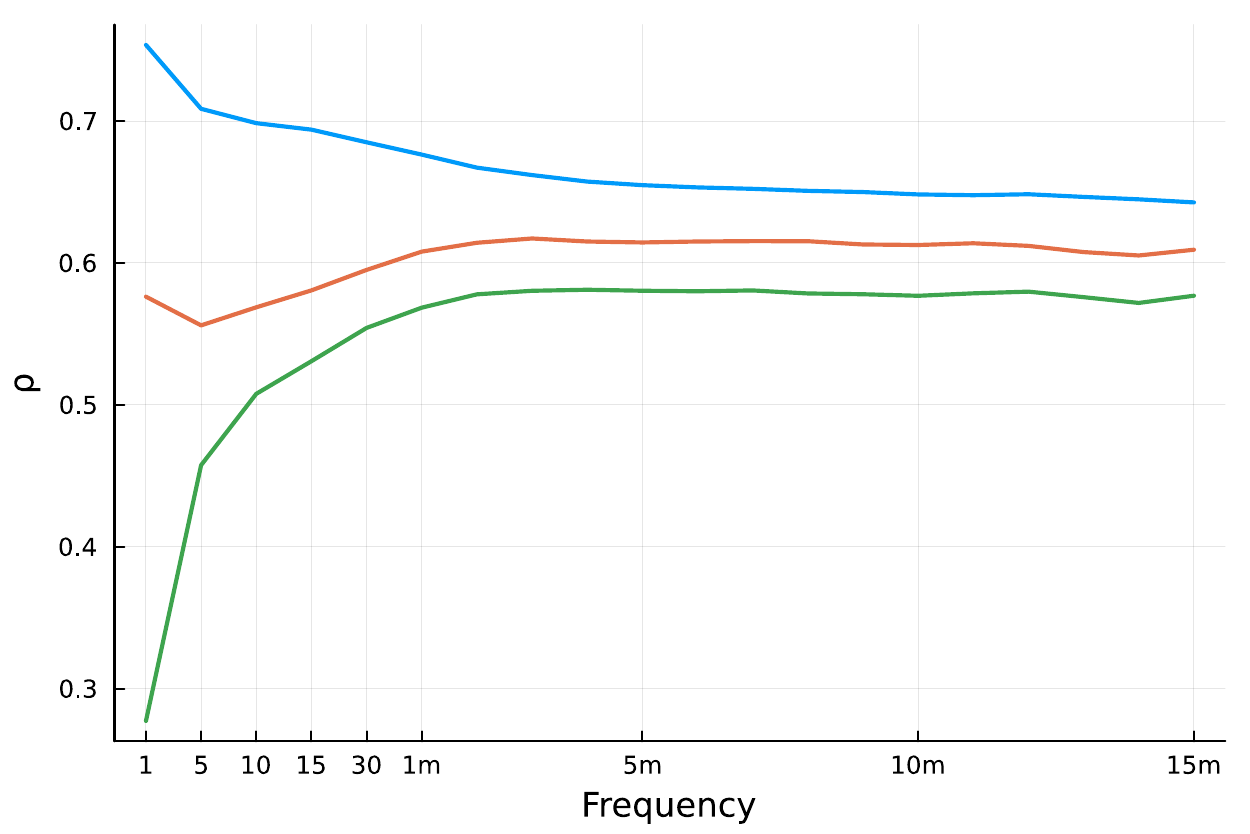}\includegraphics[width=0.16\textwidth]{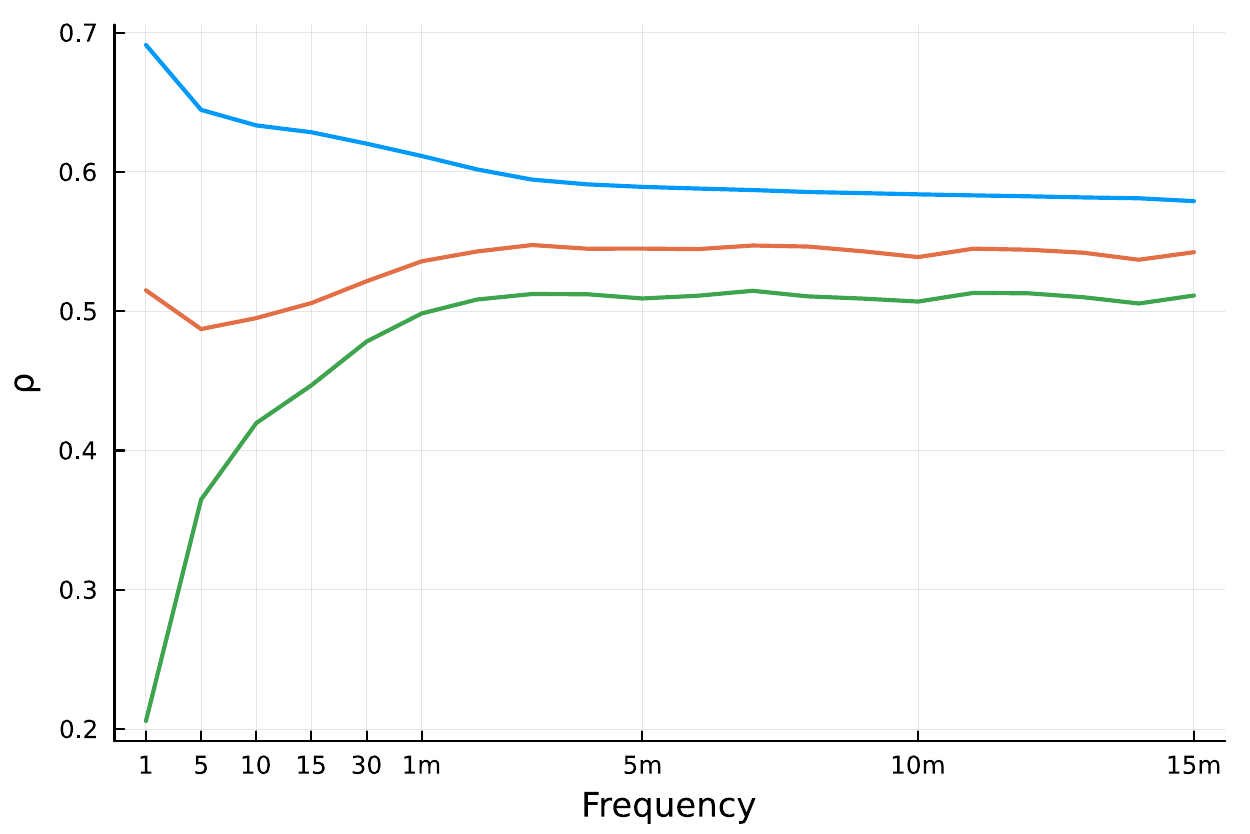}\includegraphics[width=0.16\textwidth]{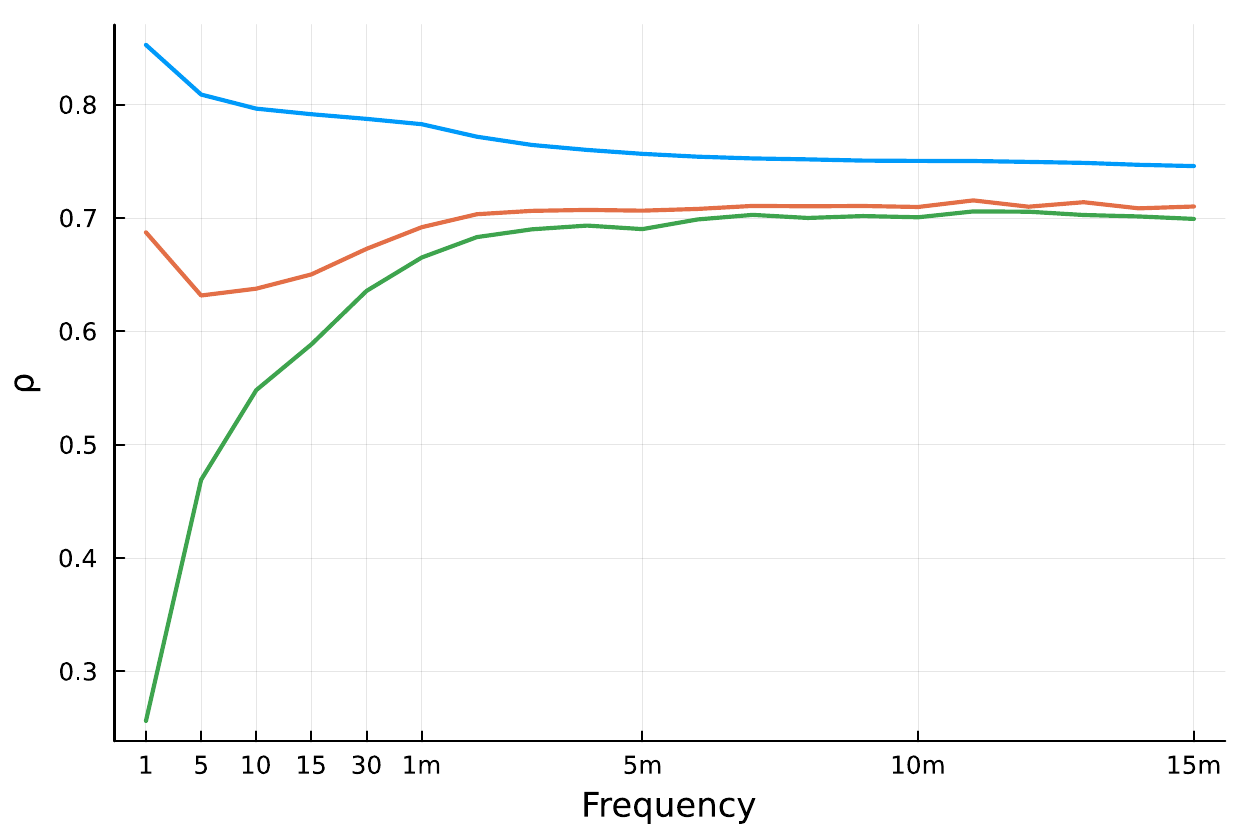}
\par\end{centering}
}\subfloat[$\mathrm{corr}$(SPY,HAL), $\mathrm{corr}$(SPY,XOM), $\mathrm{corr}$(HAL,XOM)]{\begin{centering}
\includegraphics[width=0.16\textwidth]{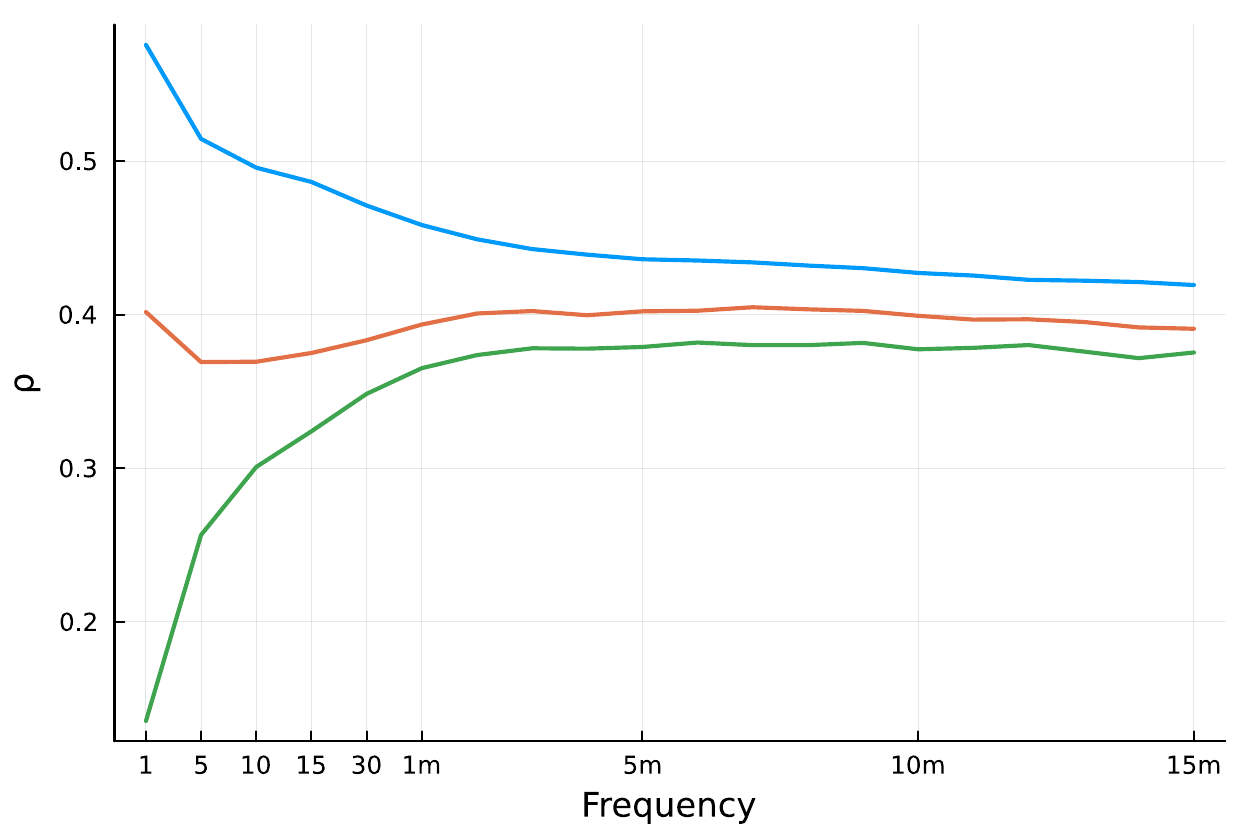}\includegraphics[width=0.16\textwidth]{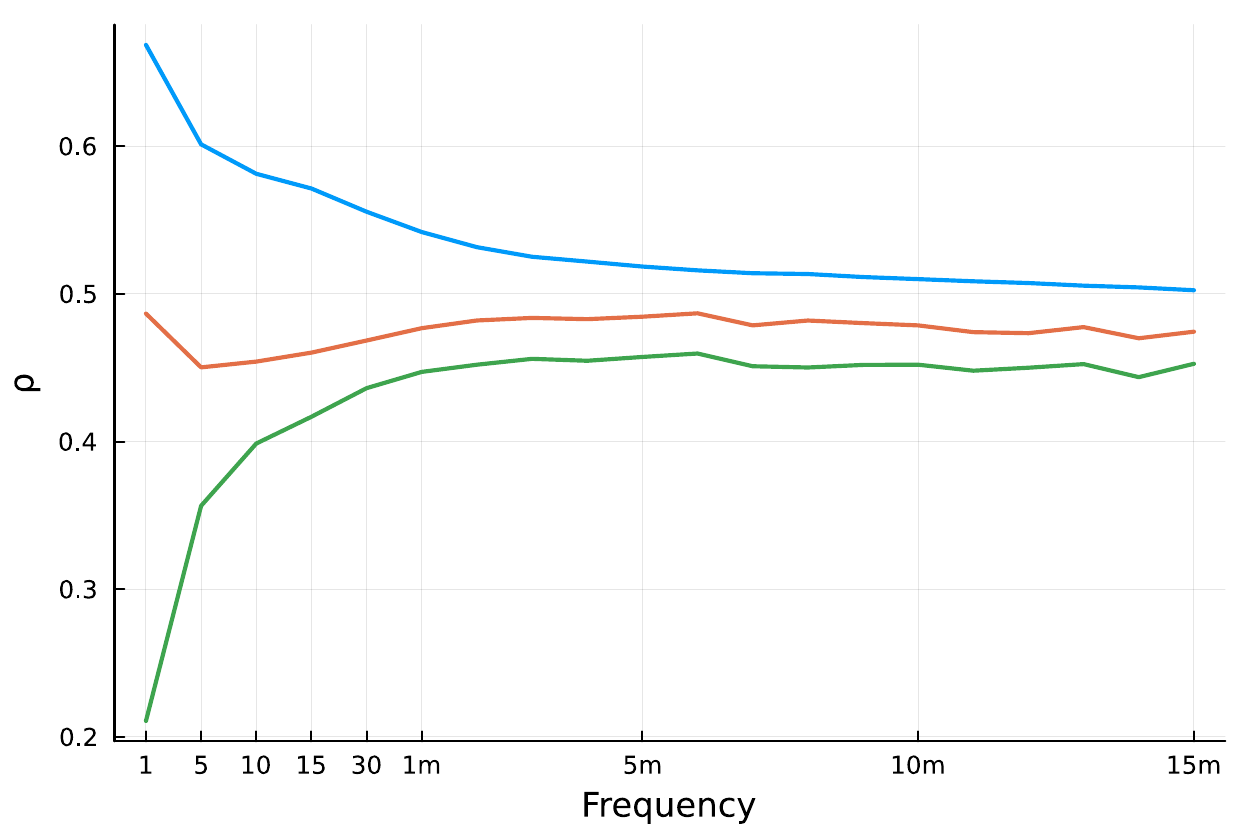}\includegraphics[width=0.16\textwidth]{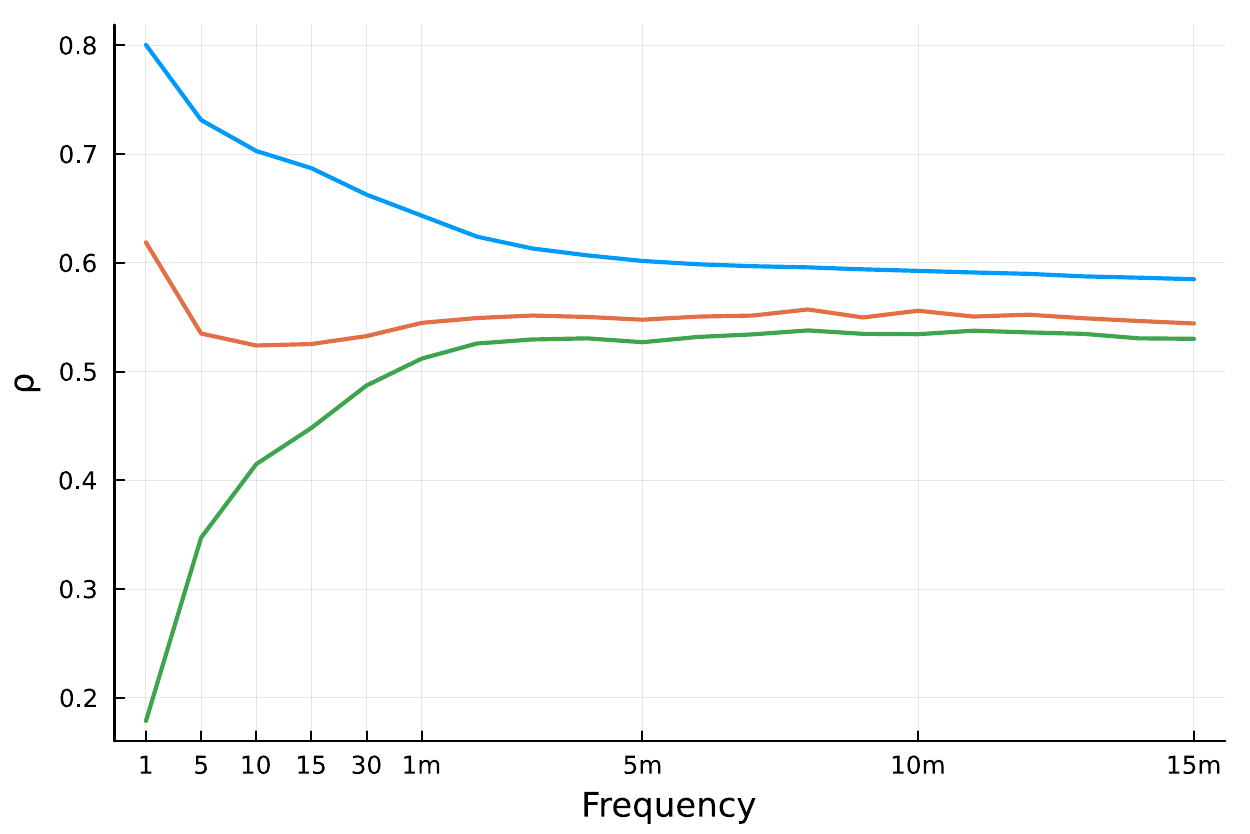}
\par\end{centering}
}
\par\end{centering}
\begin{centering}
\subfloat[$\mathrm{corr}$(SPY,PG), $\mathrm{corr}$(SPY,WMT), $\mathrm{corr}$(PG,WMT)]{\begin{centering}
\includegraphics[width=0.16\textwidth]{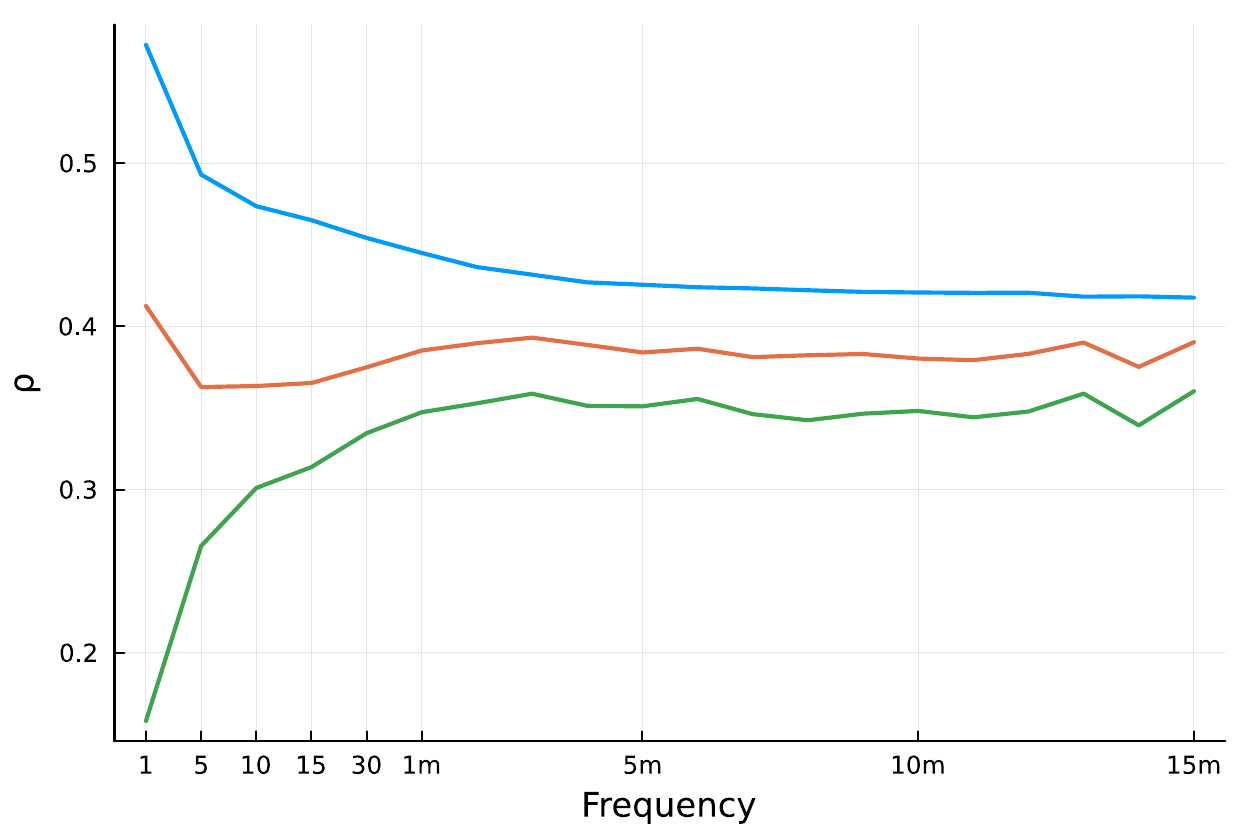}\includegraphics[width=0.16\textwidth]{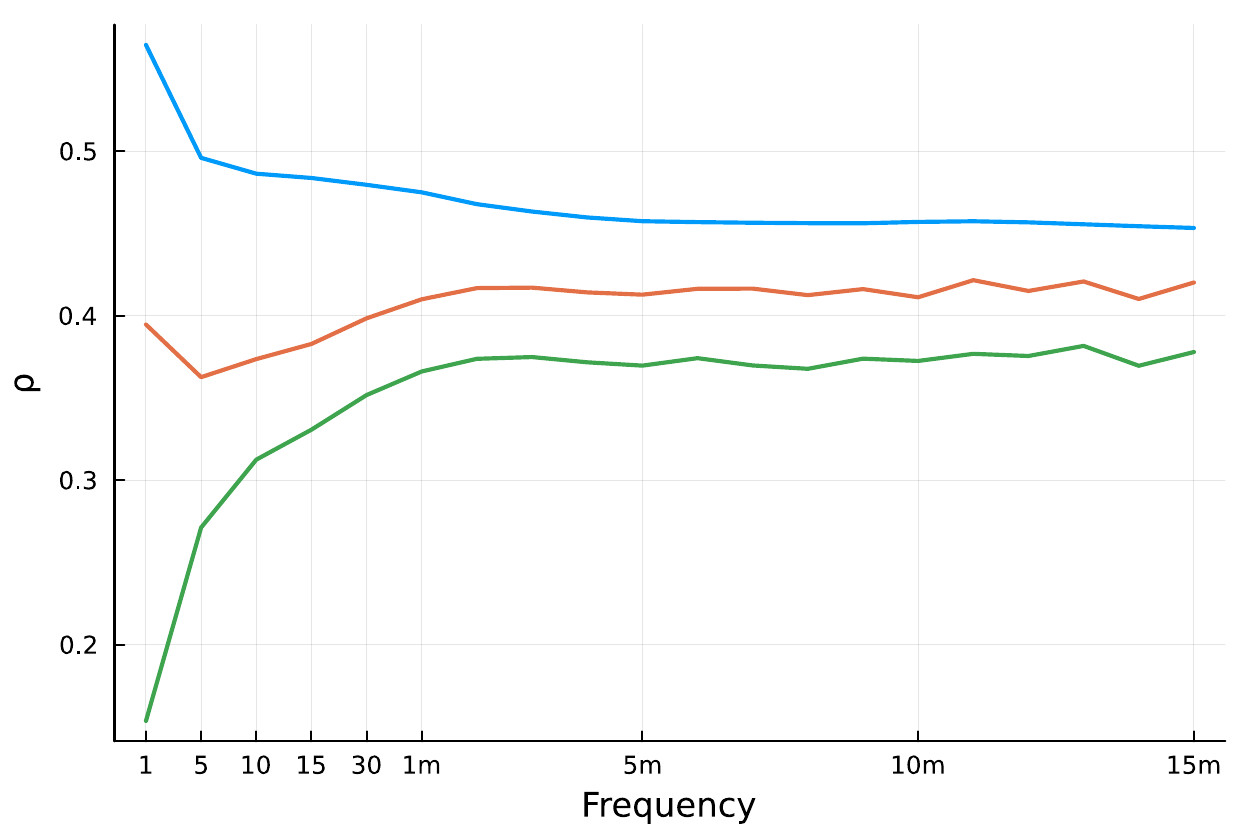}\includegraphics[width=0.16\textwidth]{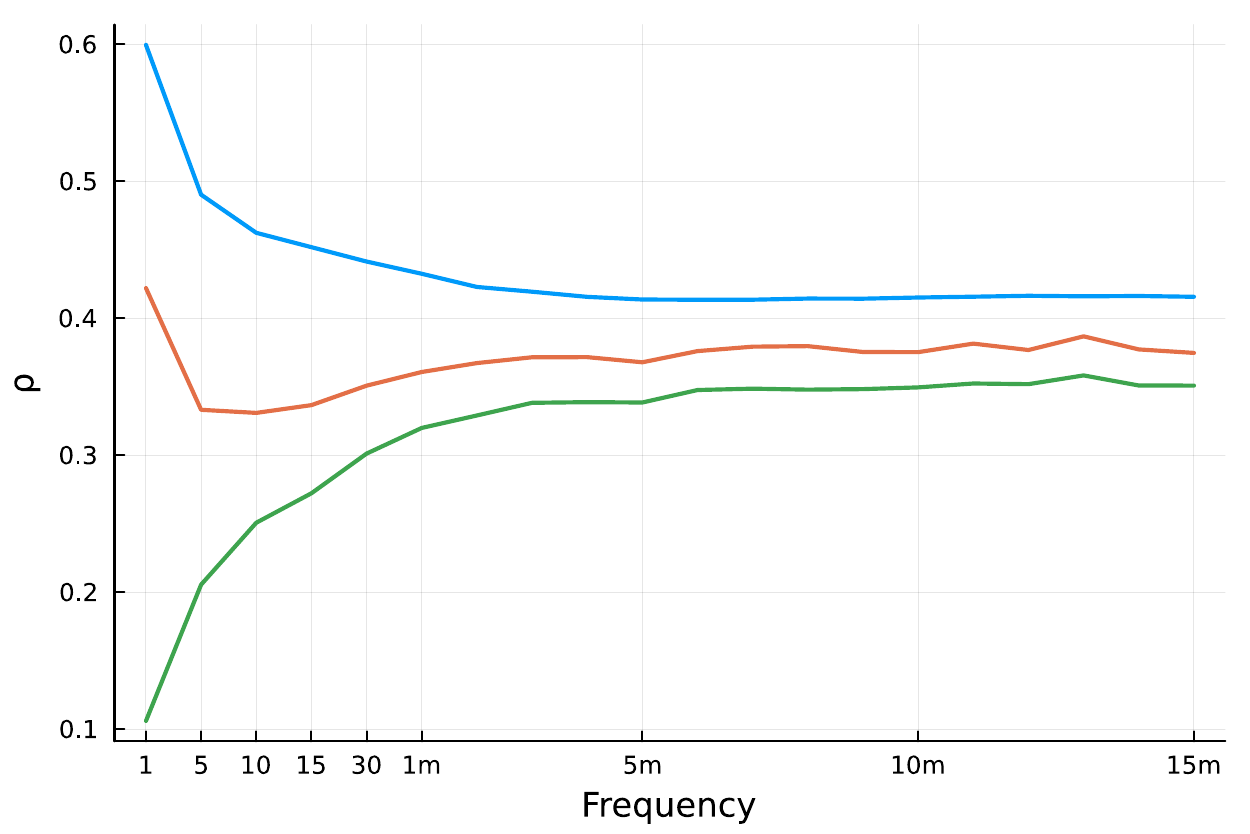}
\par\end{centering}
}\subfloat[$\mathrm{corr}$(SPY,TSLA), $\mathrm{corr}$(SPY,AMZN), $\mathrm{corr}$(TSLA,AMZN)]{\begin{centering}
\includegraphics[width=0.16\textwidth]{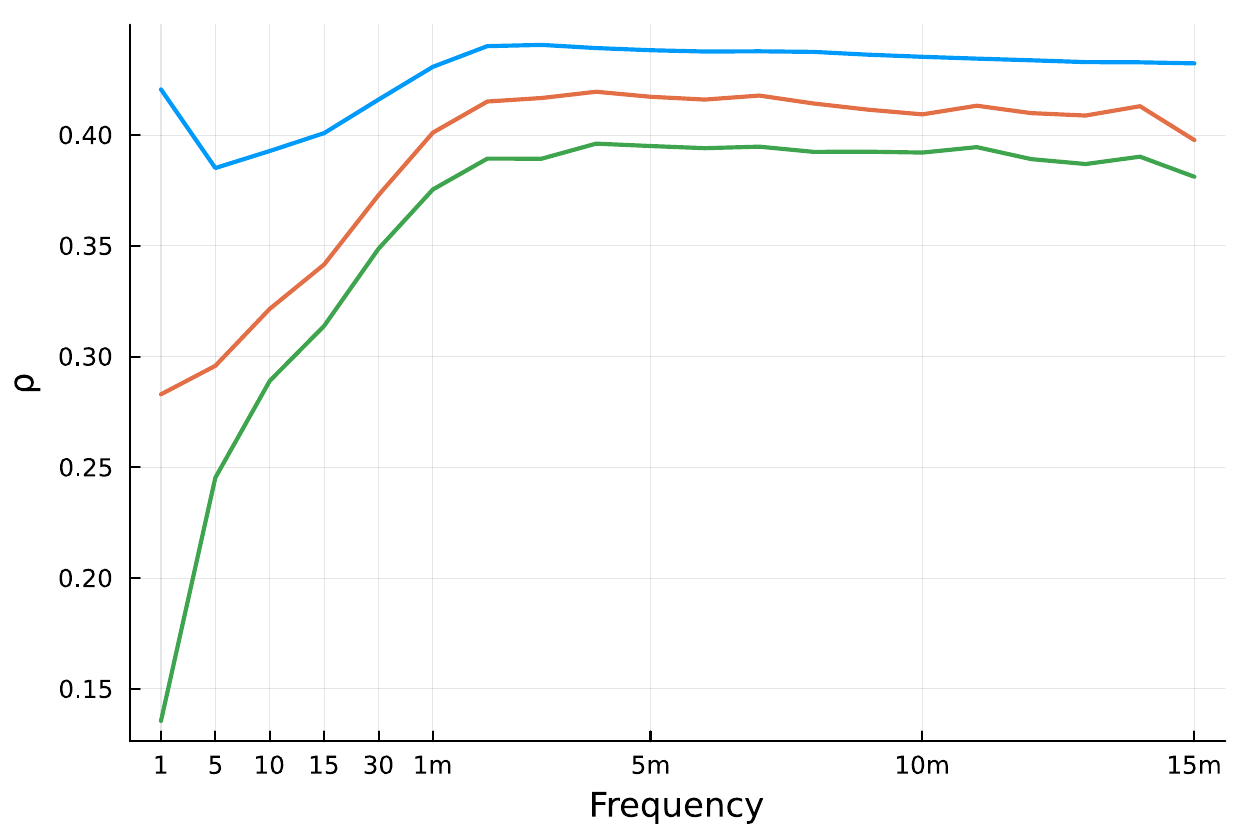}\includegraphics[width=0.16\textwidth]{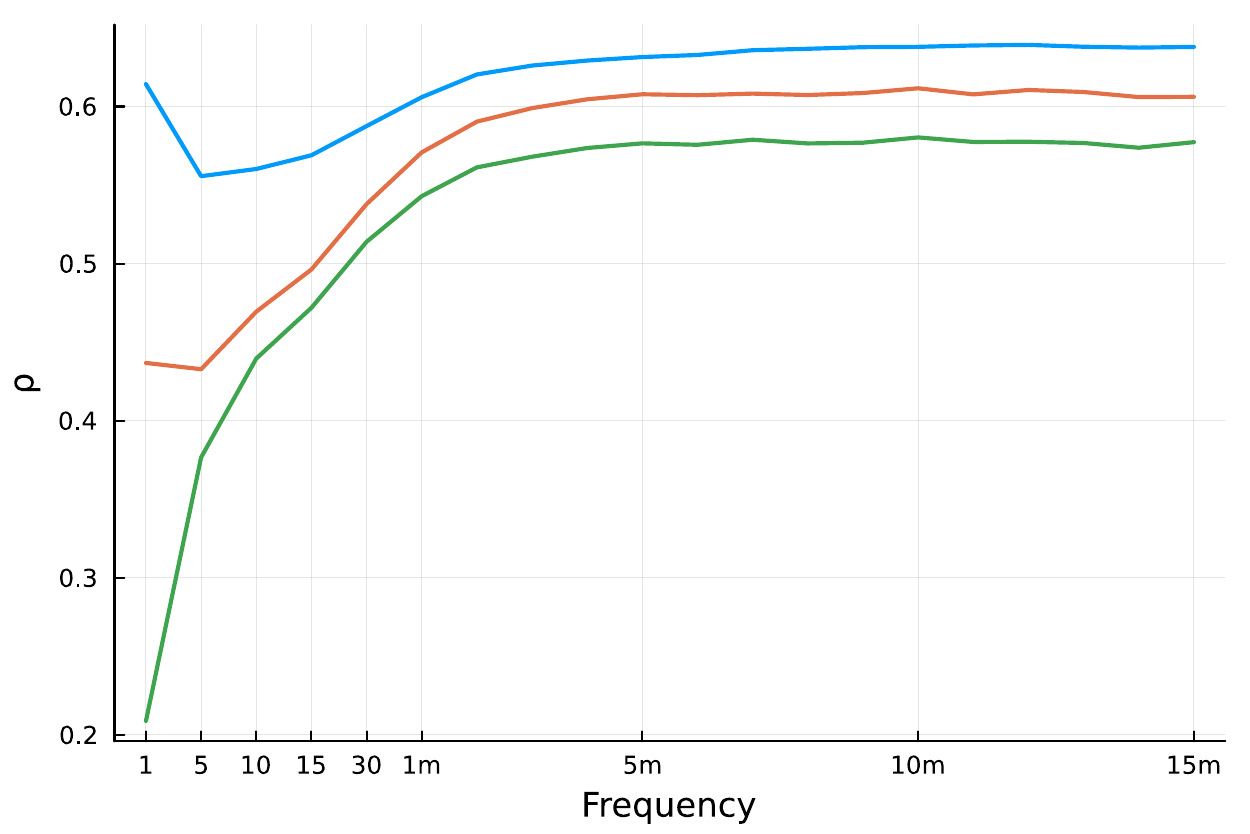}\includegraphics[width=0.16\textwidth]{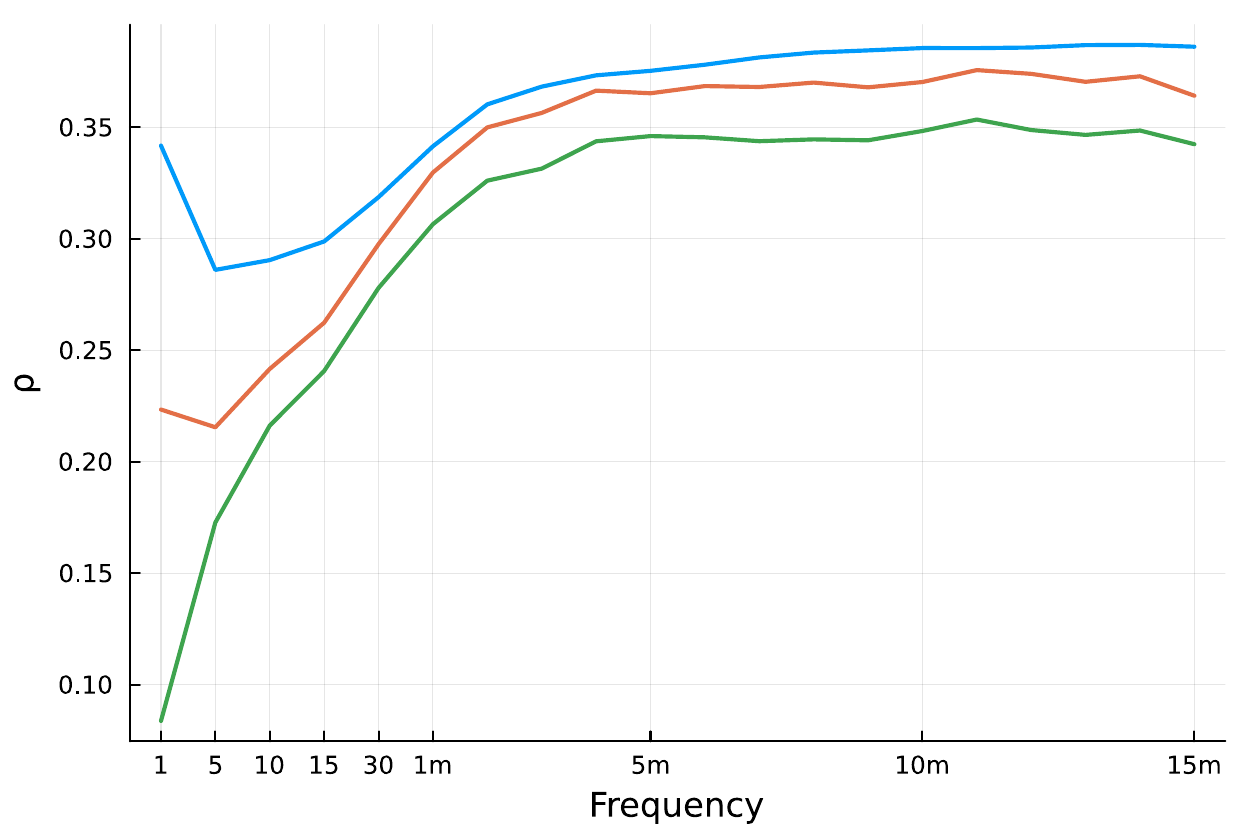}
\par\end{centering}
}
\par\end{centering}
\begin{centering}
\subfloat[$\mathrm{corr}$(SPY,DIS), $\mathrm{corr}$(SPY,FB), $\mathrm{corr}$(DIS,FB)]{\begin{centering}
\includegraphics[width=0.16\textwidth]{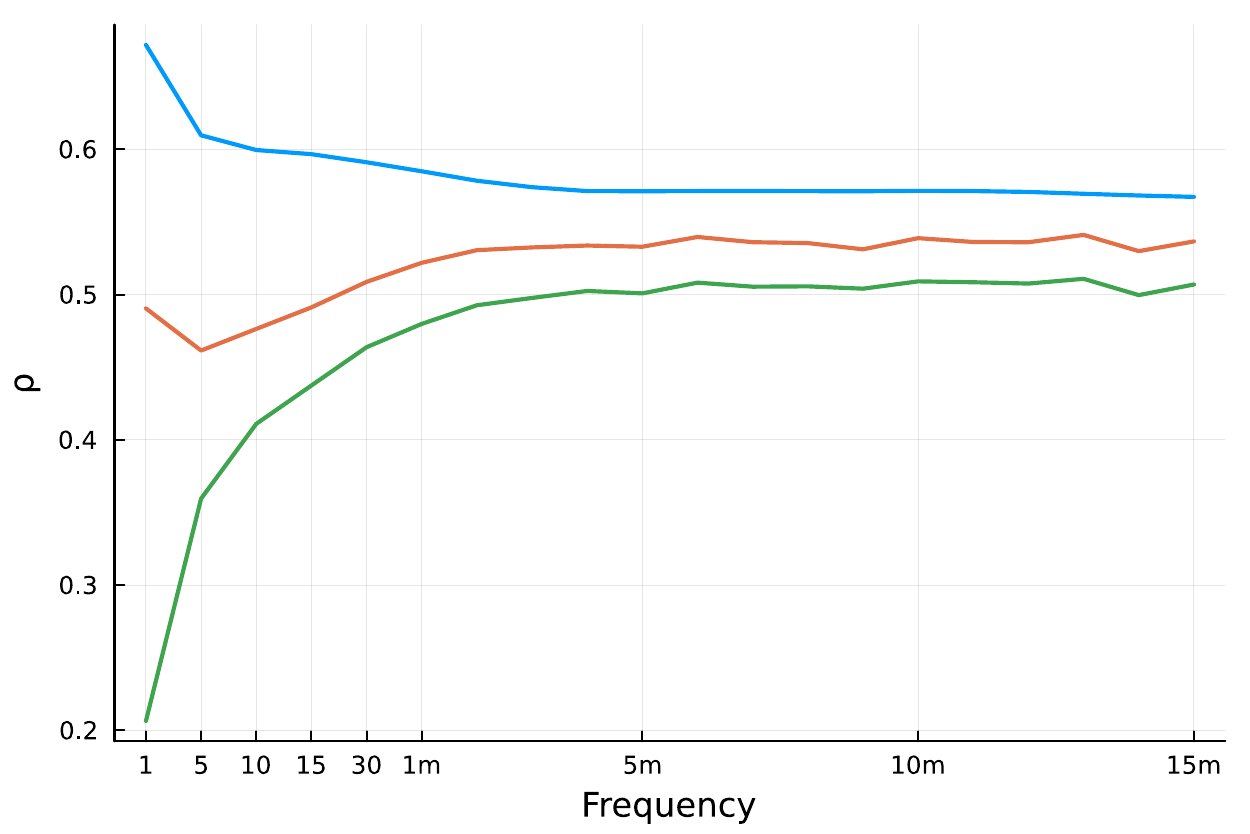}\includegraphics[width=0.16\textwidth]{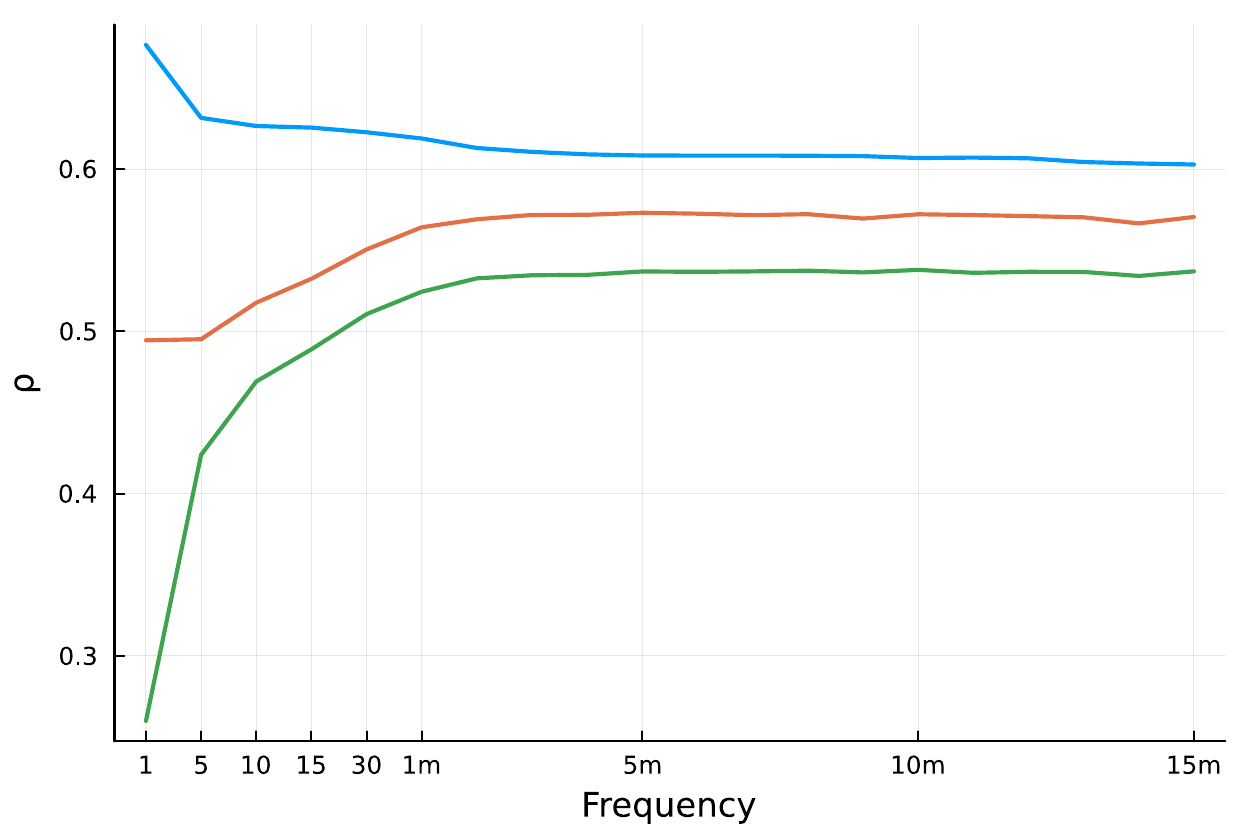}\includegraphics[width=0.16\textwidth]{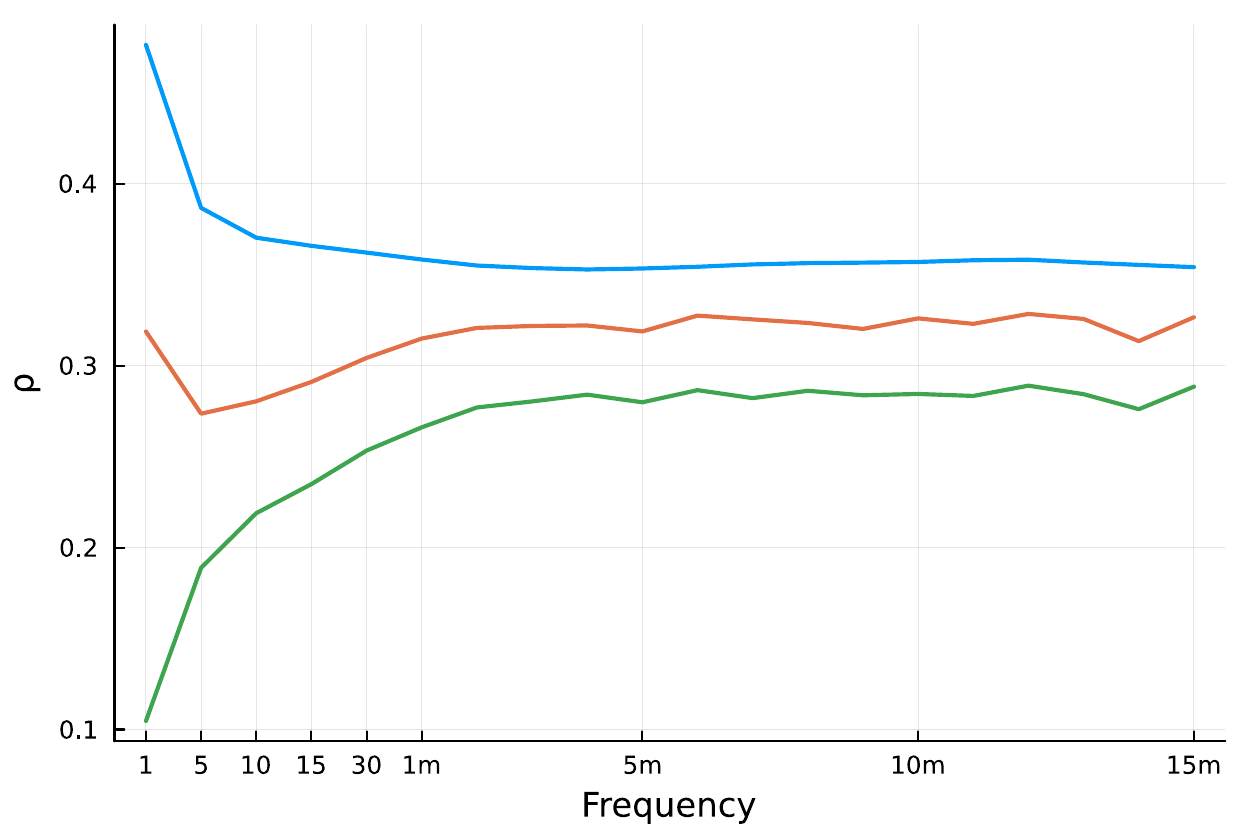}
\par\end{centering}
}
\par\end{centering}
\caption{Correlation signature plots for assets in Small Universe.\label{fig:SignaturePlotsSmallUniverse}}
\end{figure}

In Figure \ref{fig:Daily-Correlations-within-Sector-pair} we present
the range of correlation estimates for the pairs of assets within
the same sector. This figure is analogous to Figure \ref{fig:Daily-Correlations-with-SPY},
where we reported correlations between each asset and SPY. 
\begin{figure}[H]
\begin{centering}
\includegraphics[width=0.6\textwidth]{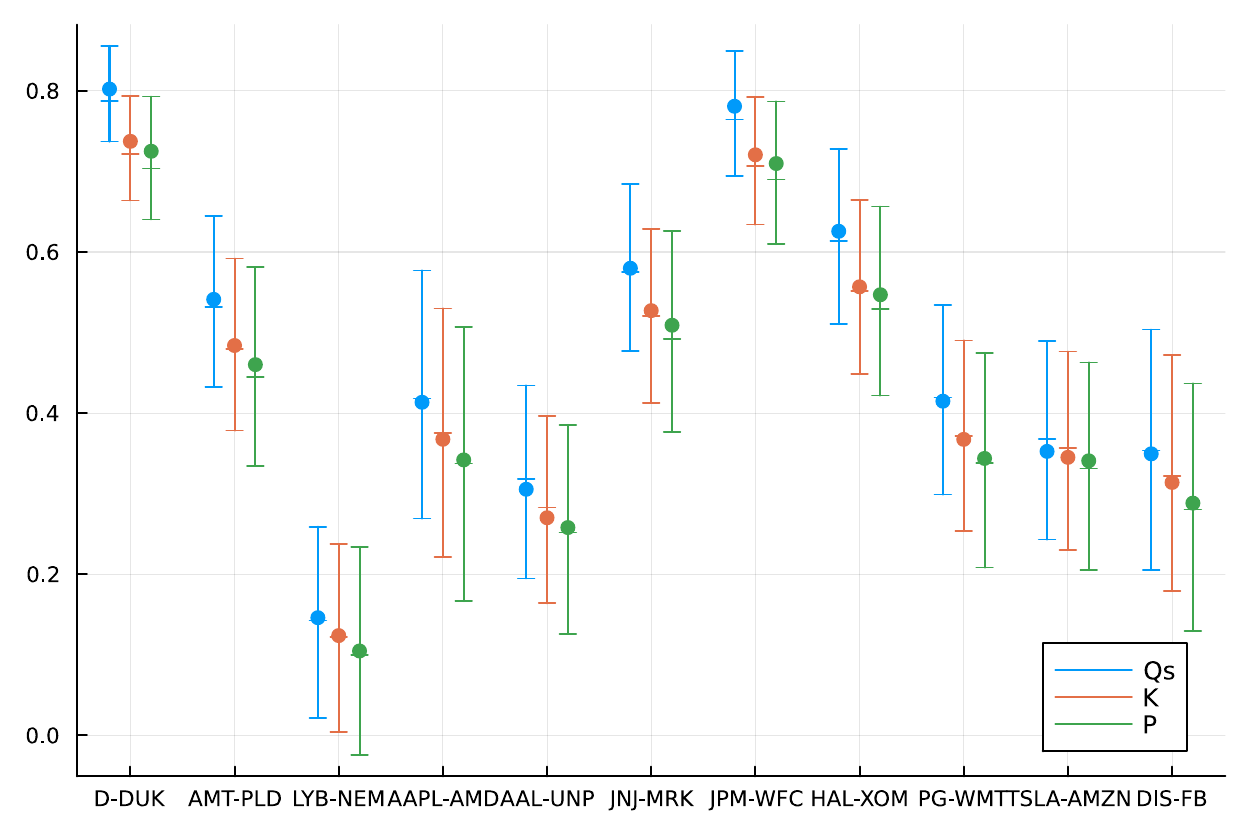}
\par\end{centering}
\caption{Daily correlations between pairs of assets in each of the 11 sectors.
The three estimators were applied to 1,763 trading days. The average
estimate (dash) and median estimate (bullet) are shown. The vertical
lines show the interquartile range over the 1,763 daily estimates.\label{fig:Daily-Correlations-within-Sector-pair}}
\end{figure}

Rolling-window estimates of correlations, relative volatilities, and
market betas are presented for all assets in the Small Universe in
Figure \ref{fig:RhoLambdaBetasSmallUniverse}.
\begin{figure}[H]
\begin{centering}
\subfloat[D]{\begin{centering}
\includegraphics[width=0.1\textwidth,height=0.06\textheight]{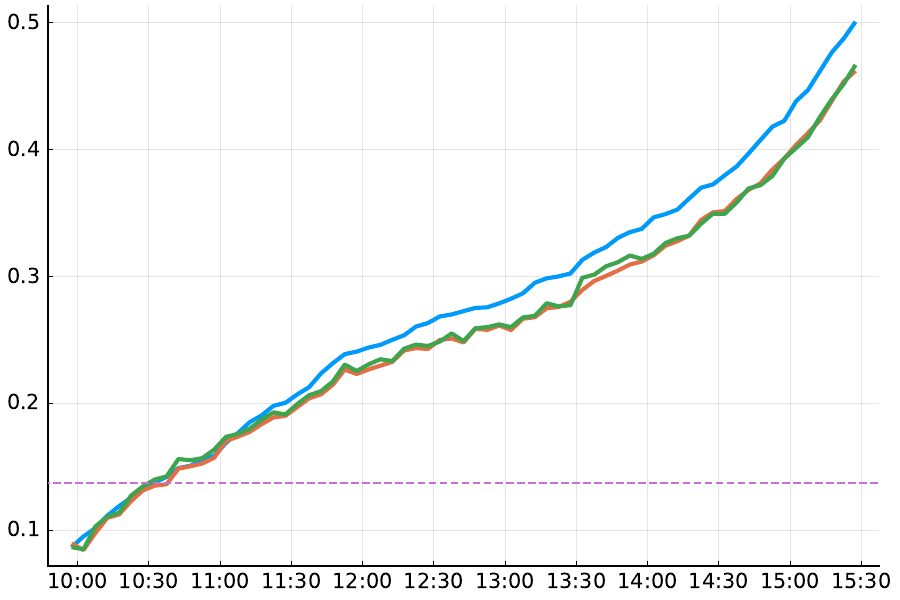}\includegraphics[width=0.1\textwidth,height=0.06\textheight]{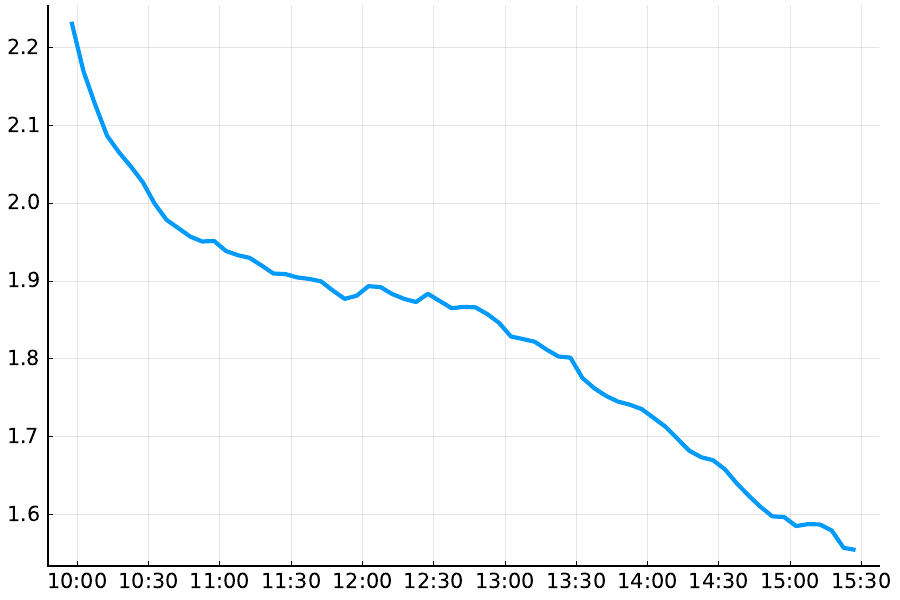}\includegraphics[width=0.1\textwidth,height=0.06\textheight]{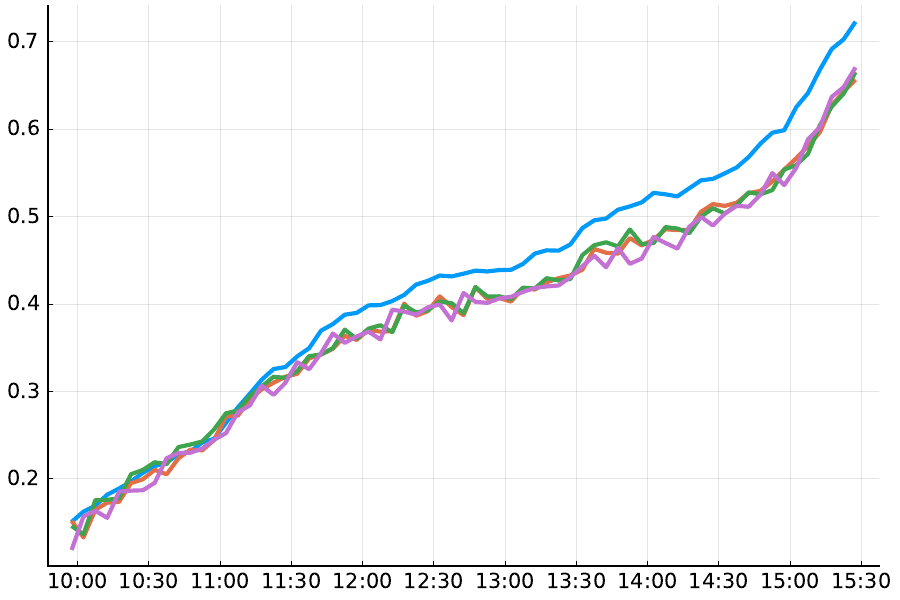}
\par\end{centering}
}\subfloat[DUK]{\begin{centering}
\includegraphics[width=0.1\textwidth,height=0.06\textheight]{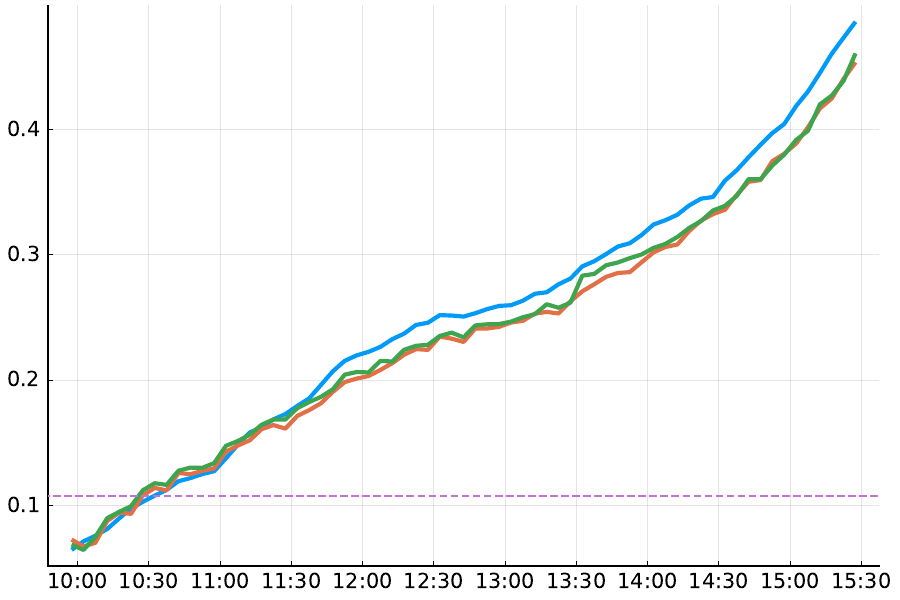}\includegraphics[width=0.1\textwidth,height=0.06\textheight]{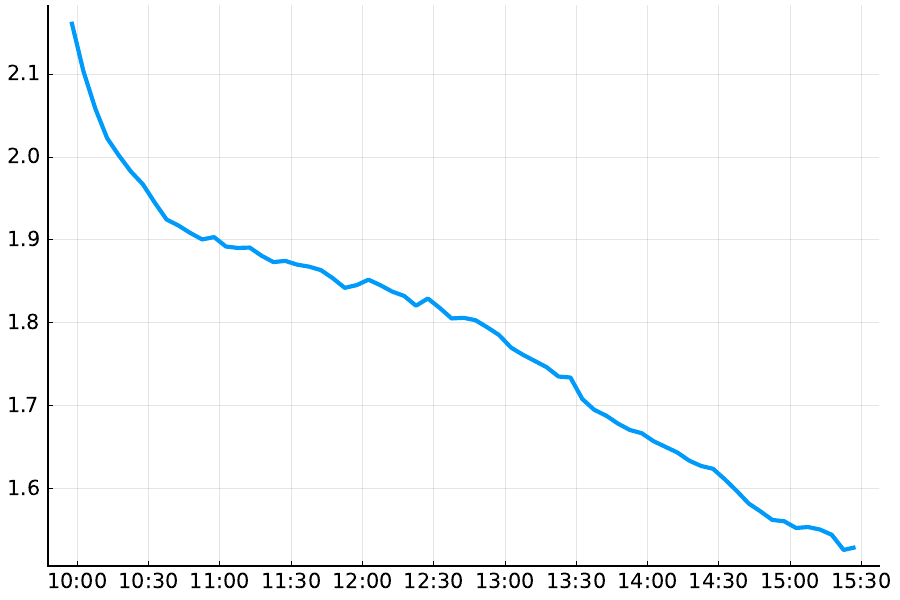}\includegraphics[width=0.1\textwidth,height=0.06\textheight]{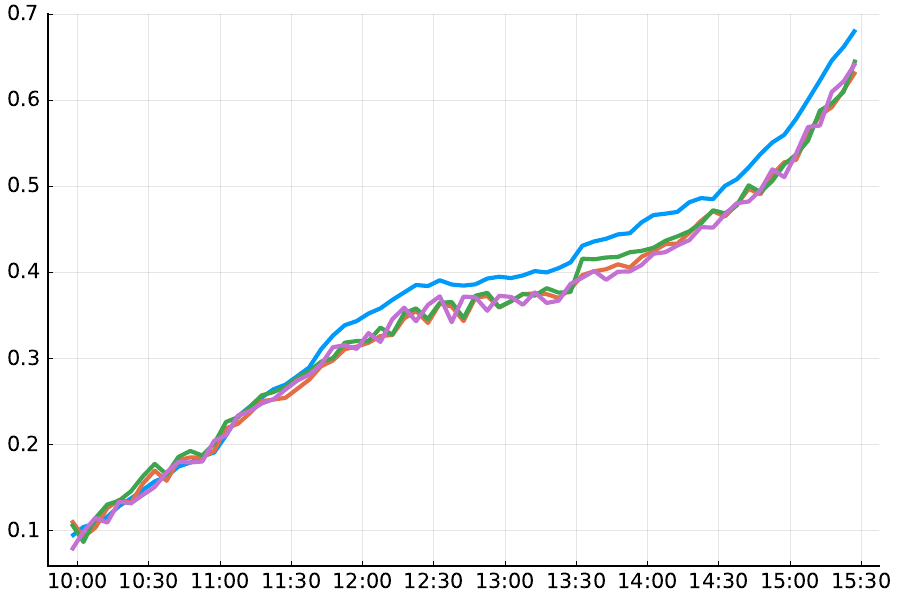}
\par\end{centering}
}\subfloat[AMT]{\begin{centering}
\includegraphics[width=0.1\textwidth,height=0.06\textheight]{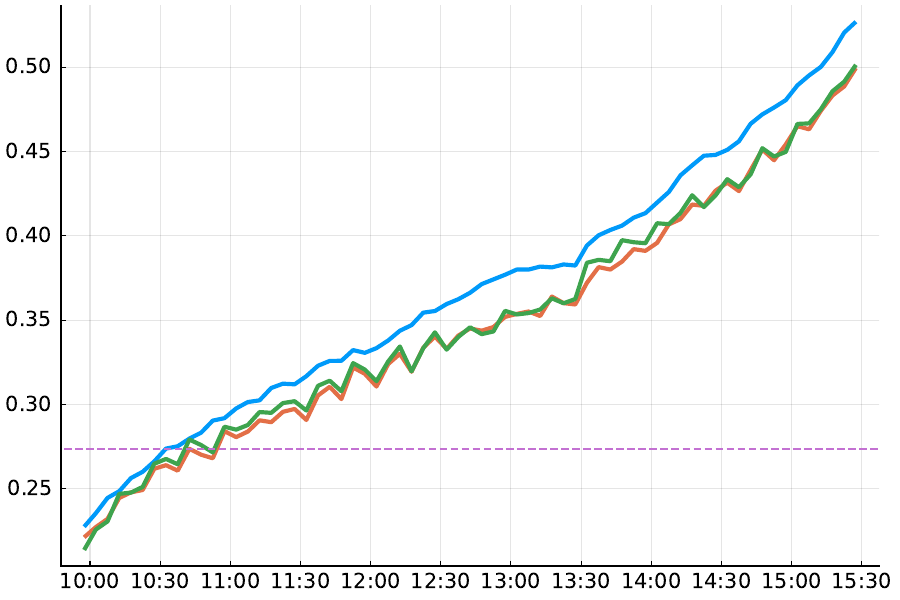}\includegraphics[width=0.1\textwidth,height=0.06\textheight]{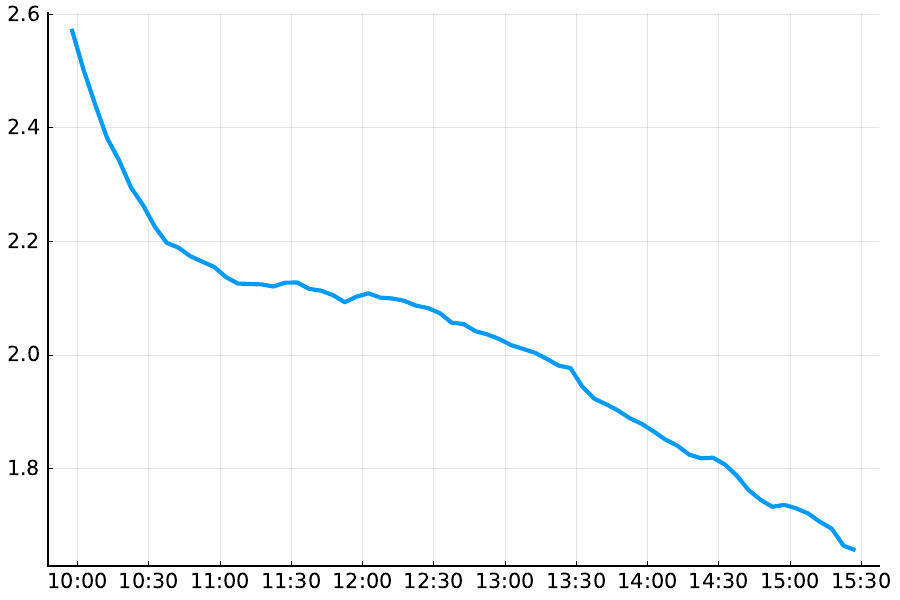}\includegraphics[width=0.1\textwidth,height=0.06\textheight]{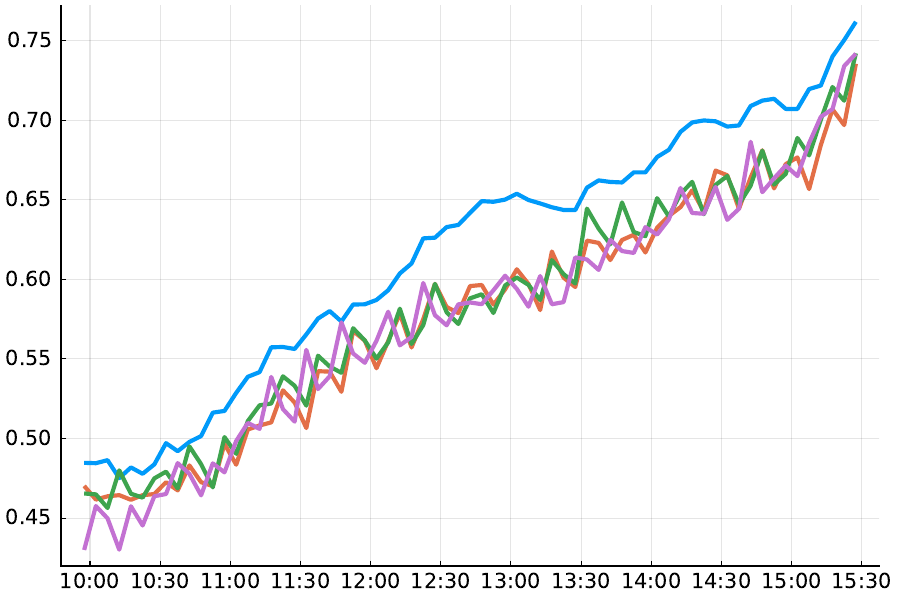}
\par\end{centering}
}
\par\end{centering}
\begin{centering}
\subfloat[PLD]{\begin{centering}
\includegraphics[width=0.1\textwidth,height=0.06\textheight]{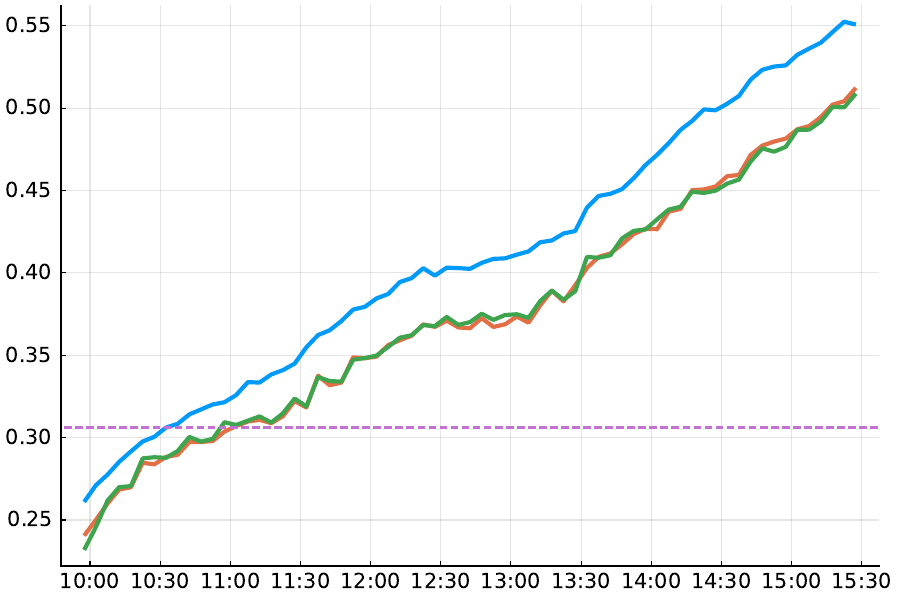}\includegraphics[width=0.1\textwidth,height=0.06\textheight]{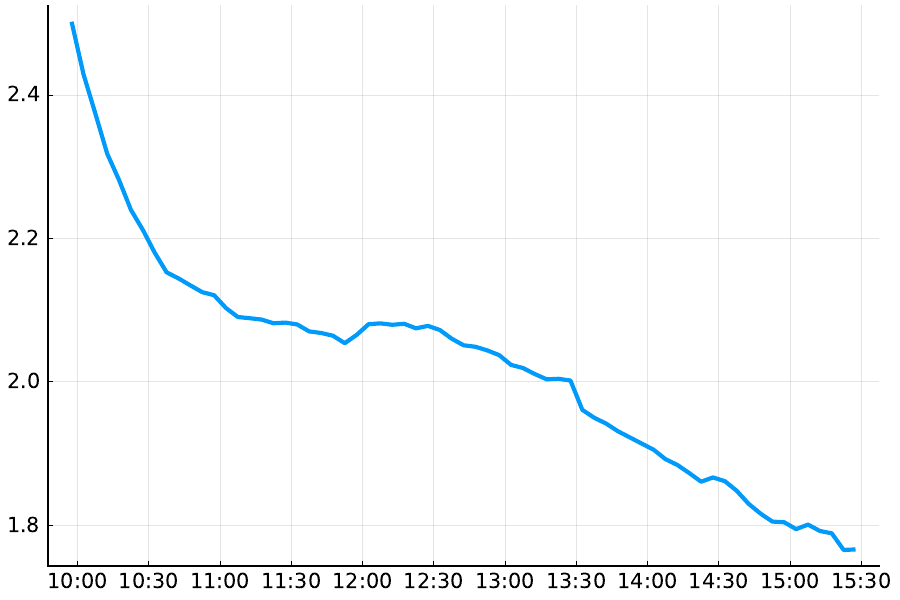}\includegraphics[width=0.1\textwidth,height=0.06\textheight]{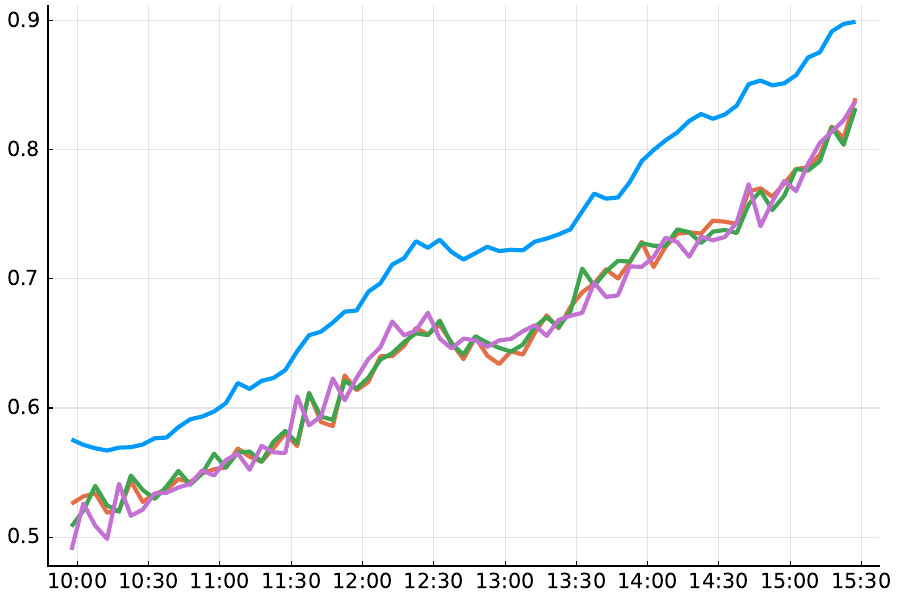}
\par\end{centering}
}\subfloat[AAL]{\begin{centering}
\includegraphics[width=0.1\textwidth,height=0.06\textheight]{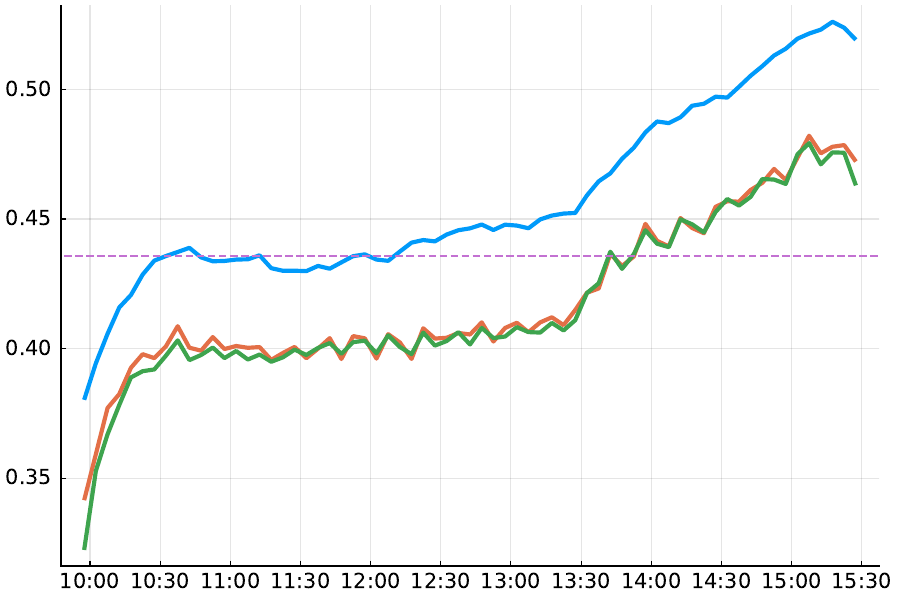}\includegraphics[width=0.1\textwidth,height=0.06\textheight]{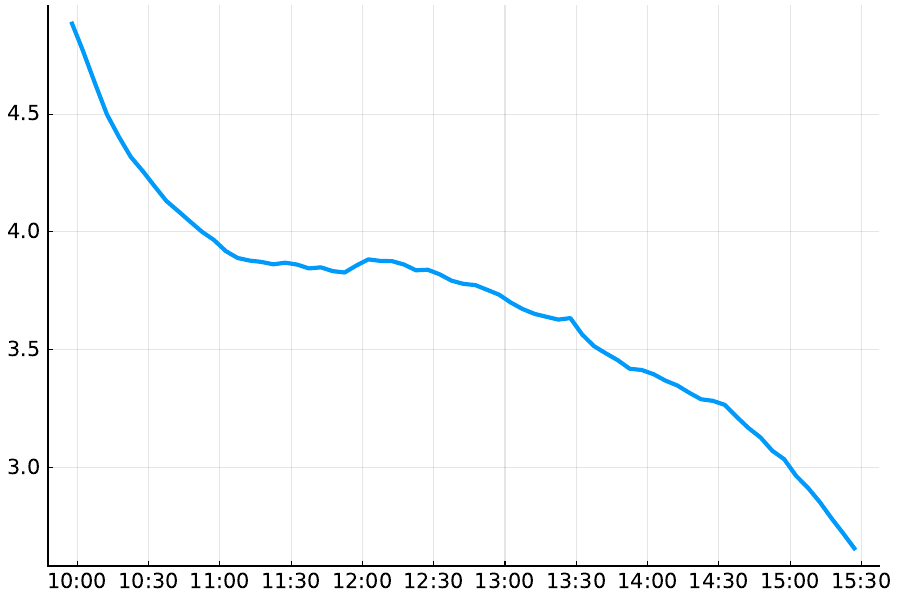}\includegraphics[width=0.1\textwidth,height=0.06\textheight]{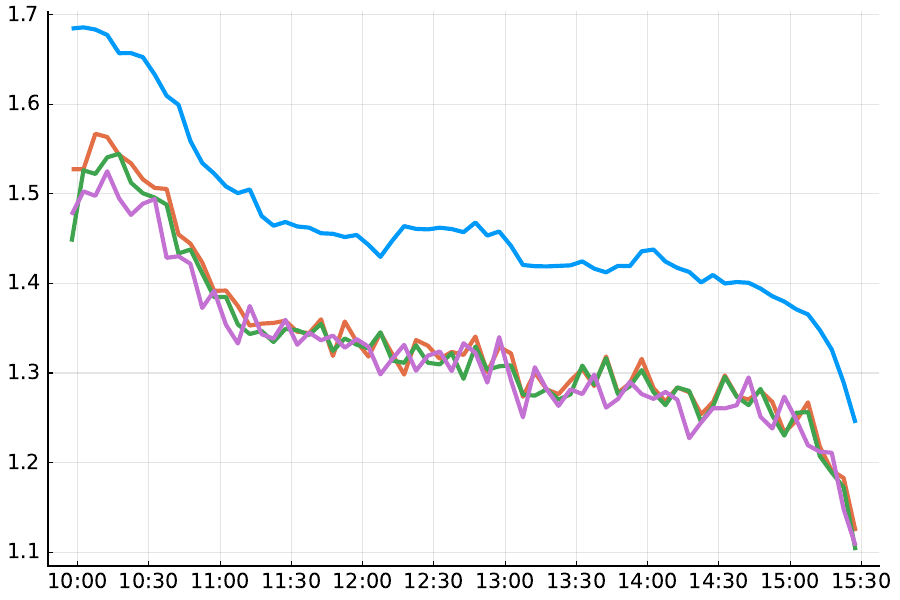}
\par\end{centering}
}\subfloat[UNP]{\begin{centering}
\includegraphics[width=0.1\textwidth,height=0.06\textheight]{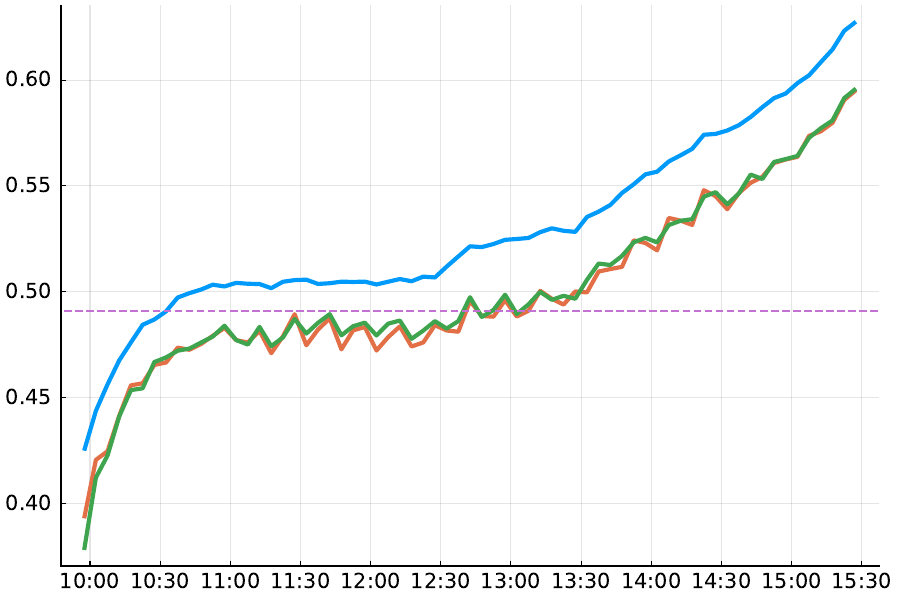}\includegraphics[width=0.1\textwidth,height=0.06\textheight]{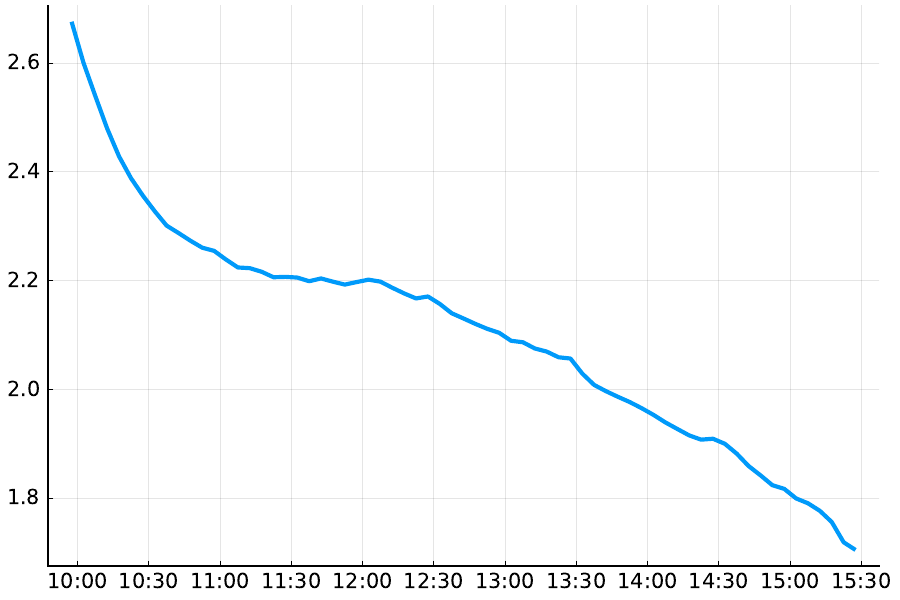}\includegraphics[width=0.1\textwidth,height=0.06\textheight]{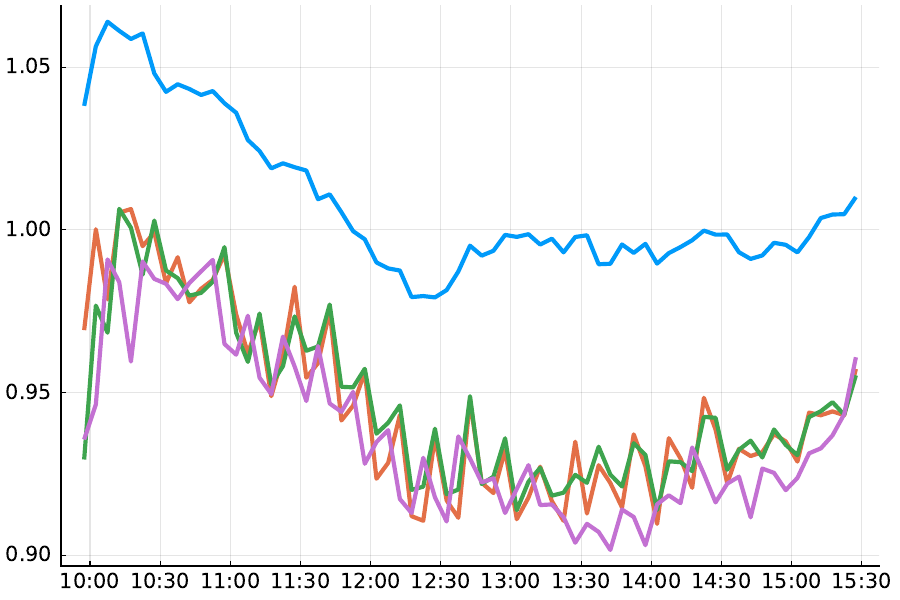}
\par\end{centering}
}
\par\end{centering}
\begin{centering}
\subfloat[JNJ]{\begin{centering}
\includegraphics[width=0.1\textwidth,height=0.06\textheight]{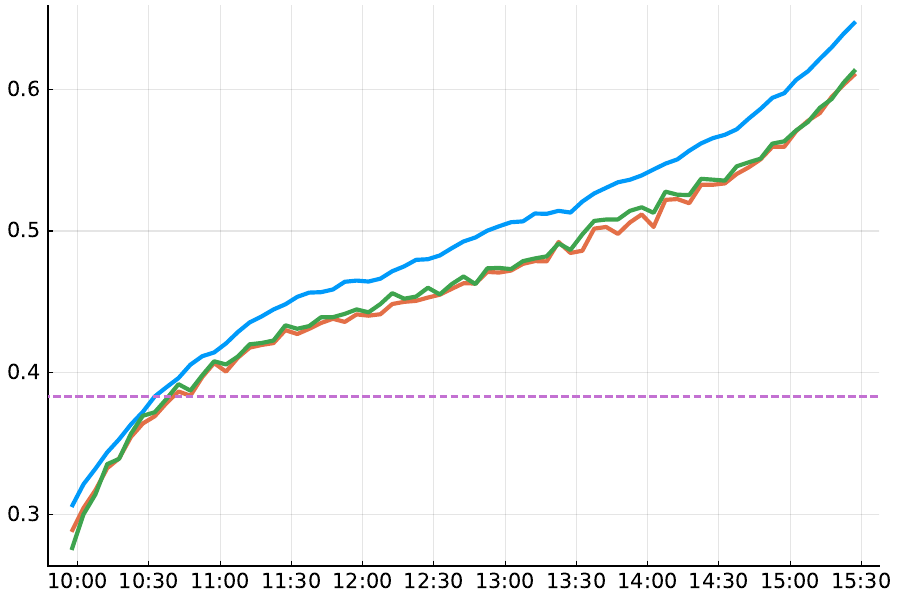}\includegraphics[width=0.1\textwidth,height=0.06\textheight]{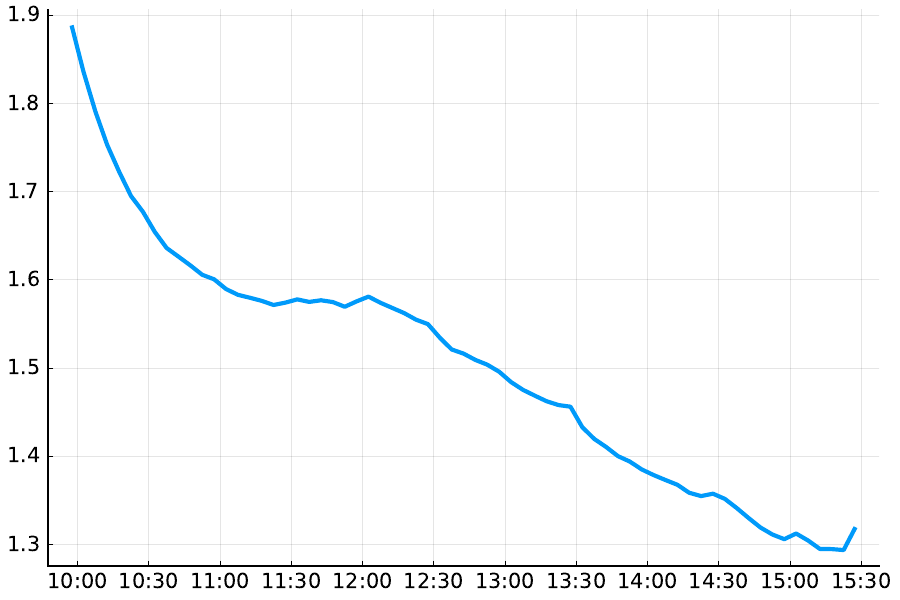}\includegraphics[width=0.1\textwidth,height=0.06\textheight]{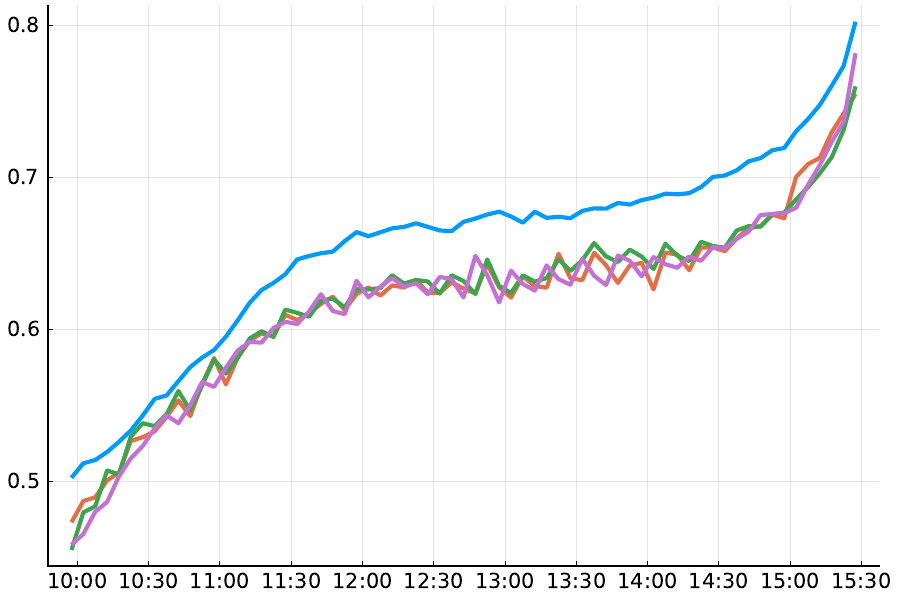}
\par\end{centering}
}\subfloat[MRK]{\begin{centering}
\includegraphics[width=0.1\textwidth,height=0.06\textheight]{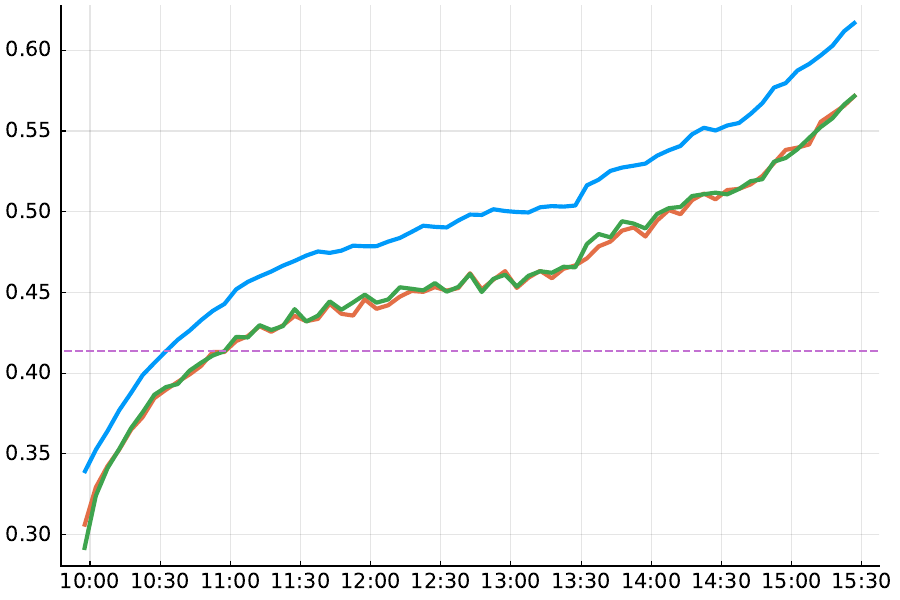}\includegraphics[width=0.1\textwidth,height=0.06\textheight]{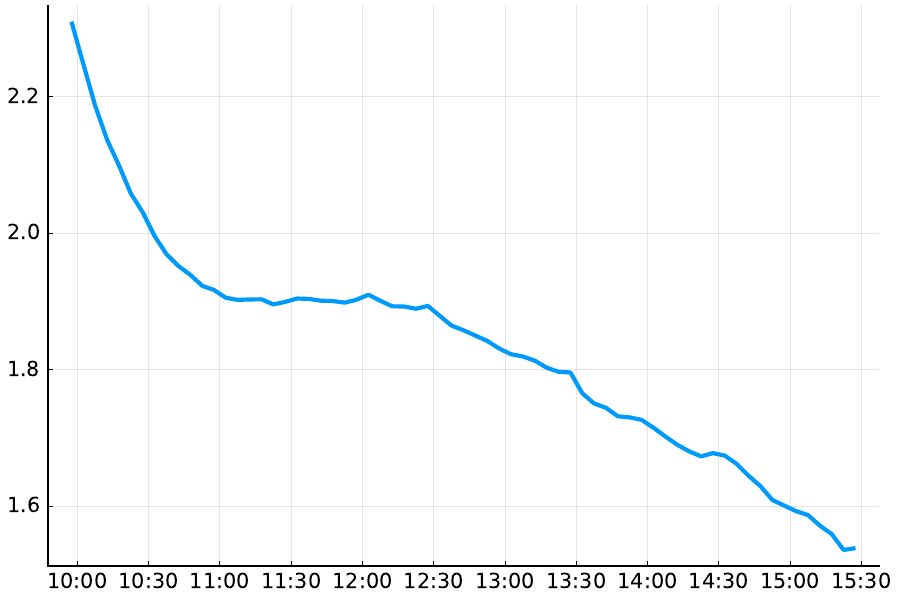}\includegraphics[width=0.1\textwidth,height=0.06\textheight]{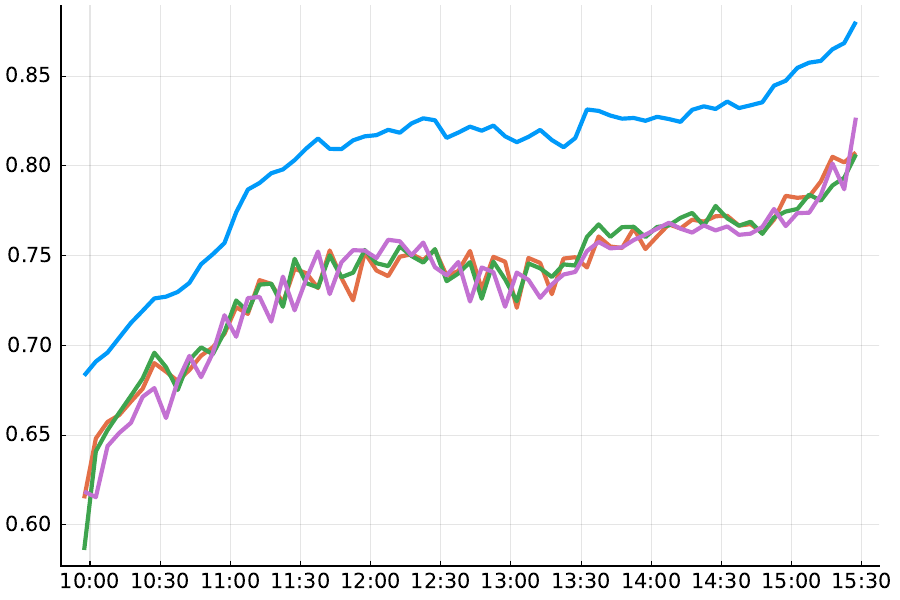}
\par\end{centering}
}\subfloat[JPM]{\begin{centering}
\includegraphics[width=0.1\textwidth,height=0.06\textheight]{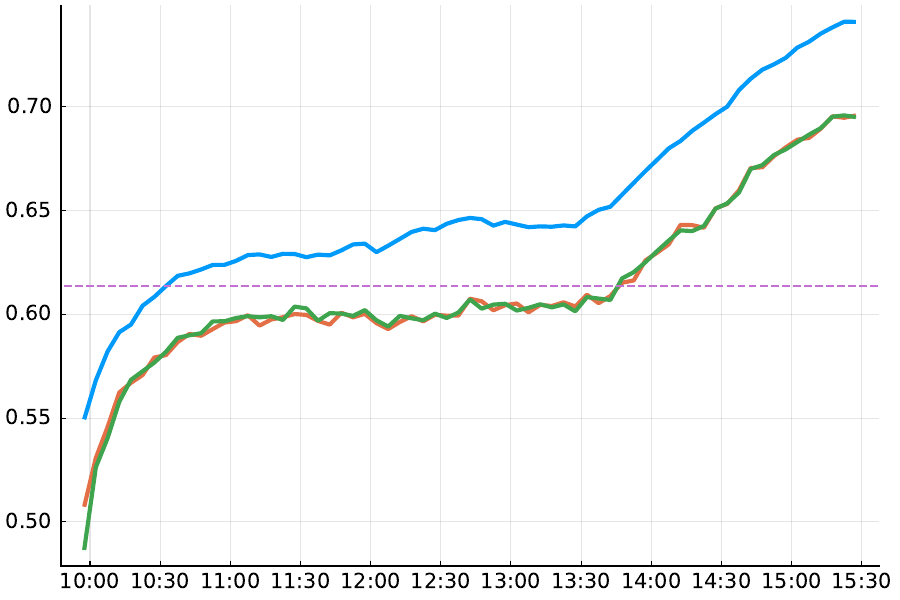}\includegraphics[width=0.1\textwidth,height=0.06\textheight]{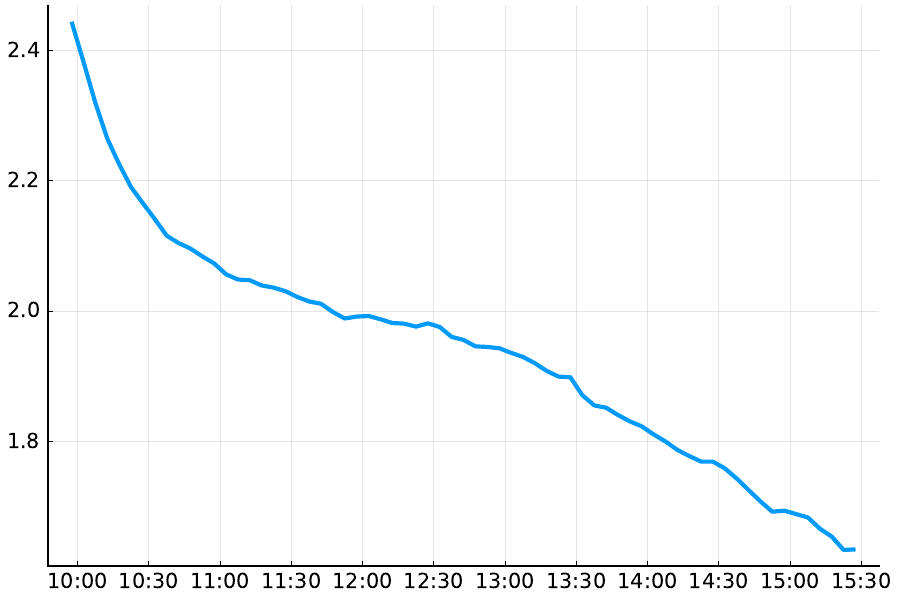}\includegraphics[width=0.1\textwidth,height=0.06\textheight]{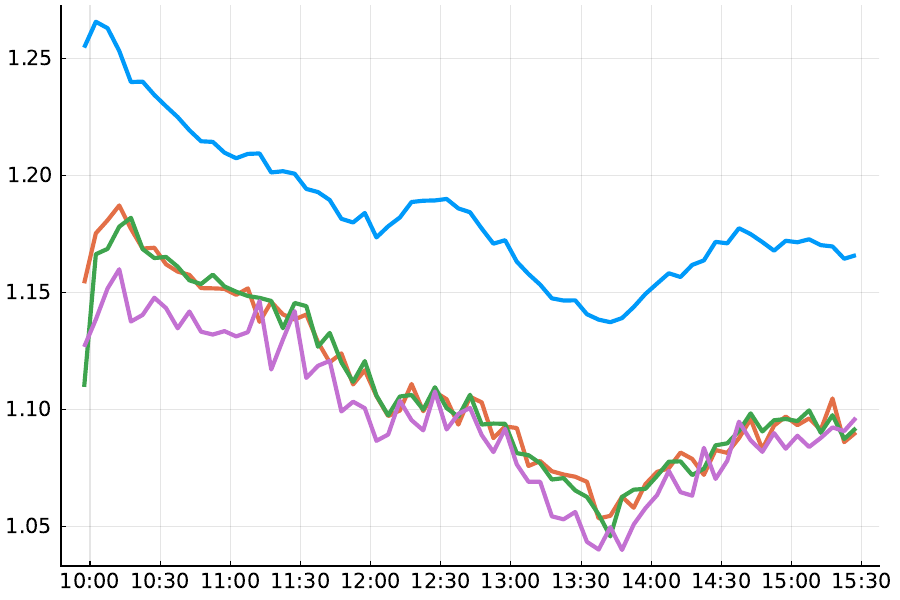}
\par\end{centering}
}
\par\end{centering}
\begin{centering}
\subfloat[WFC]{\begin{centering}
\includegraphics[width=0.1\textwidth,height=0.06\textheight]{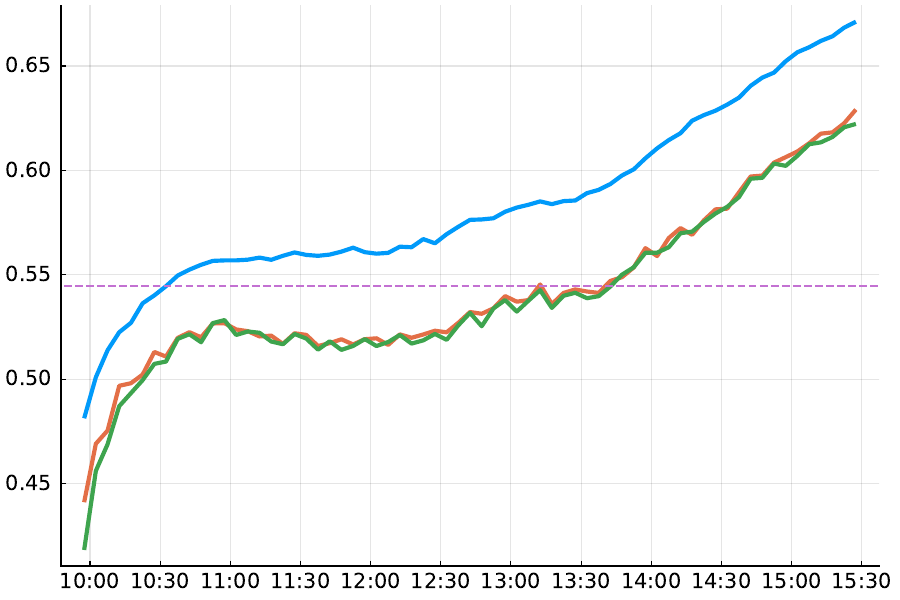}\includegraphics[width=0.1\textwidth,height=0.06\textheight]{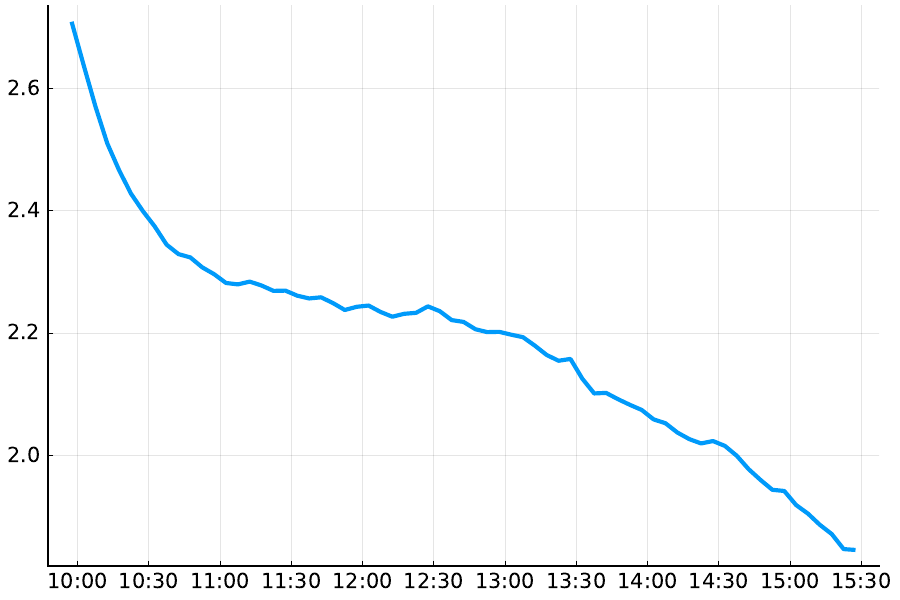}\includegraphics[width=0.1\textwidth,height=0.06\textheight]{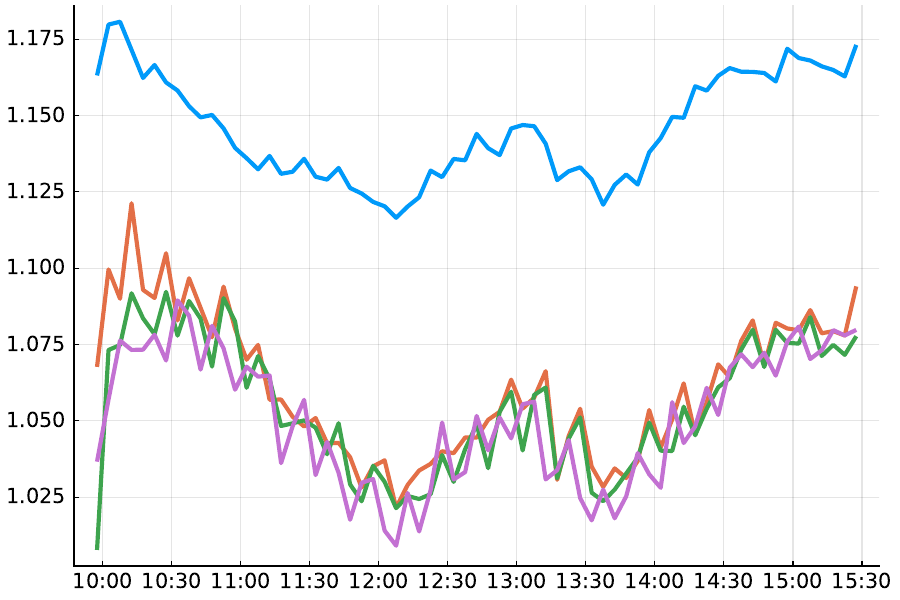}
\par\end{centering}
}\subfloat[HAL]{\begin{centering}
\includegraphics[width=0.1\textwidth,height=0.06\textheight]{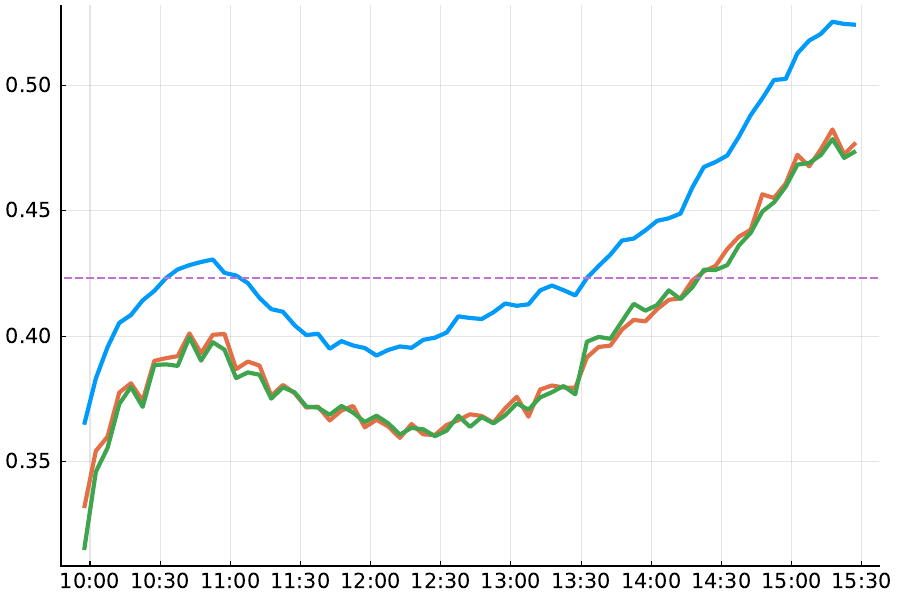}\includegraphics[width=0.1\textwidth,height=0.06\textheight]{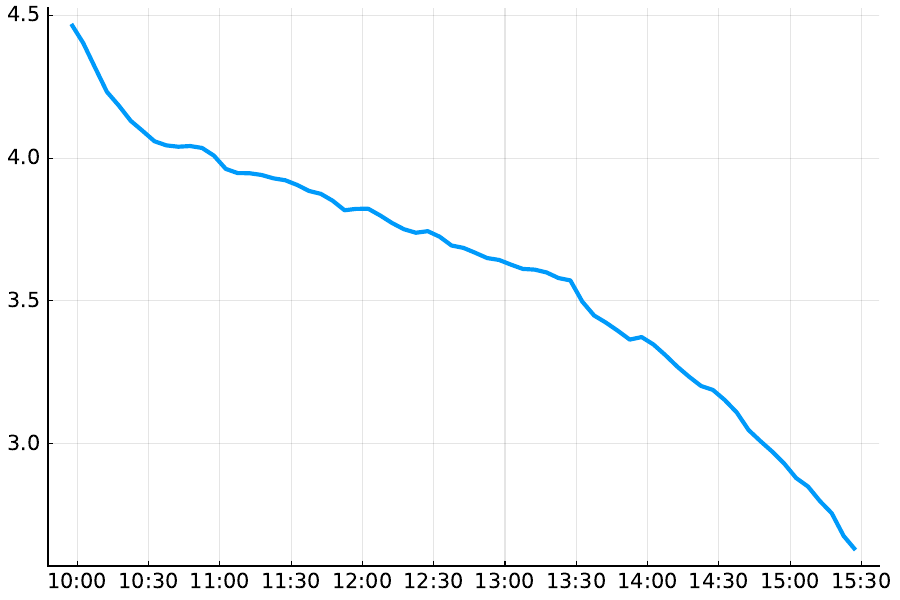}\includegraphics[width=0.1\textwidth,height=0.06\textheight]{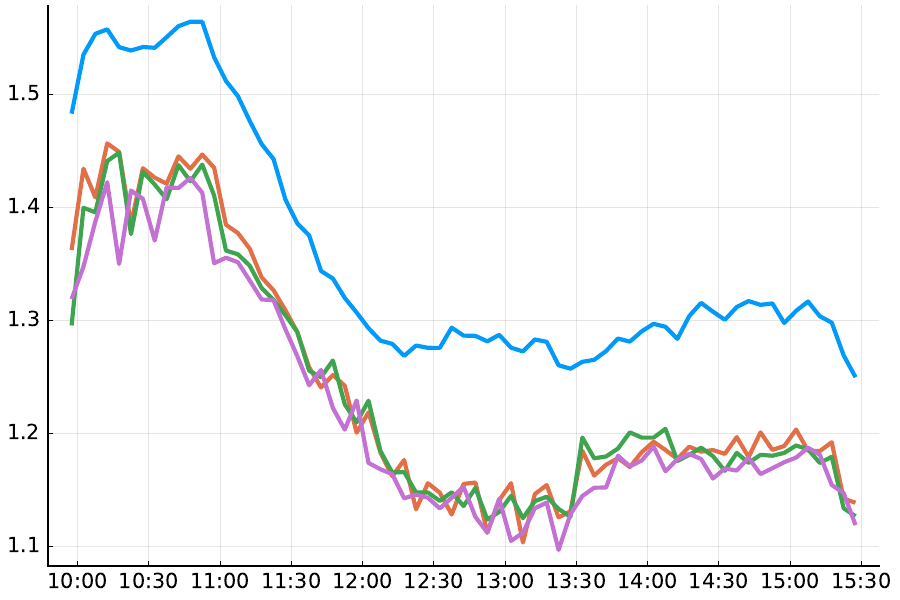}
\par\end{centering}
}\subfloat[XOM]{\begin{centering}
\includegraphics[width=0.1\textwidth,height=0.06\textheight]{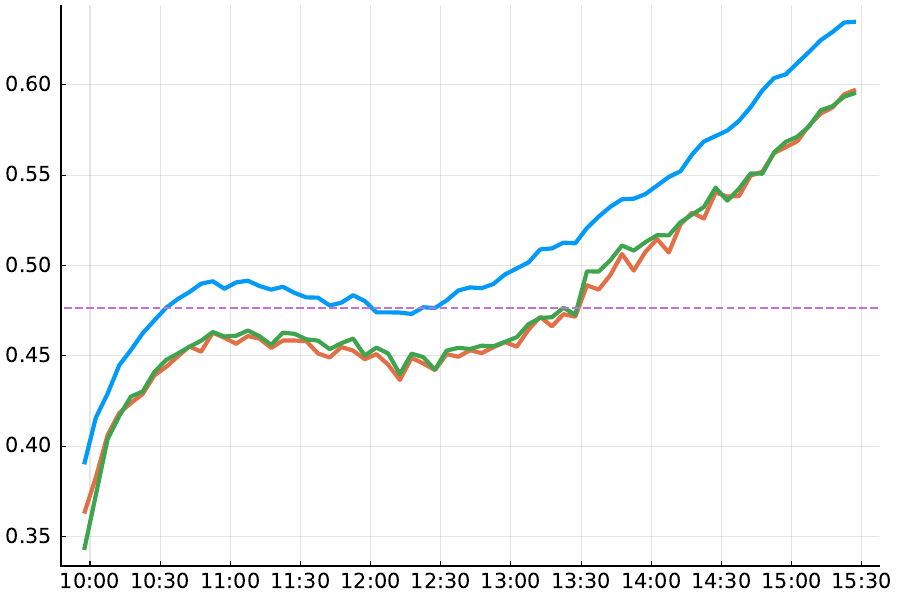}\includegraphics[width=0.1\textwidth,height=0.06\textheight]{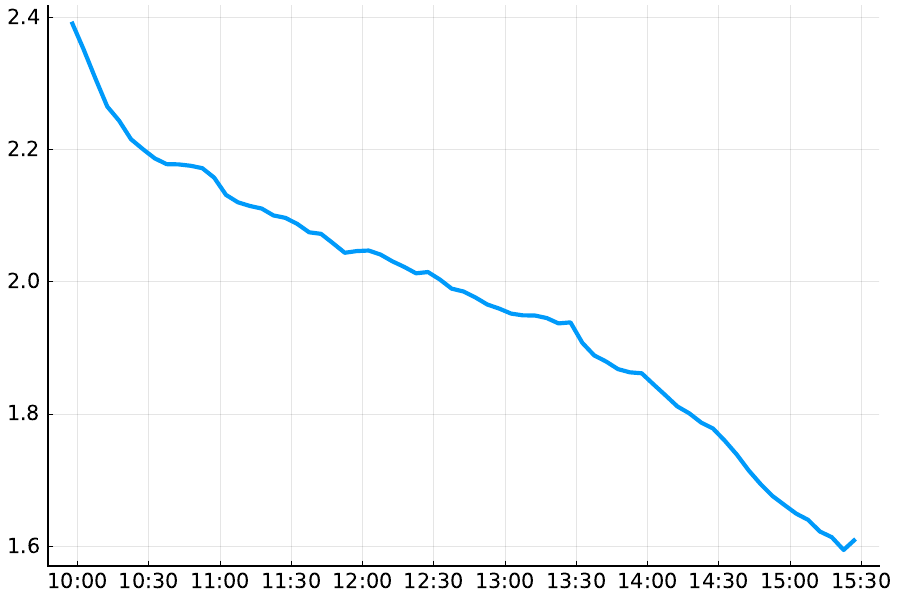}\includegraphics[width=0.1\textwidth,height=0.06\textheight]{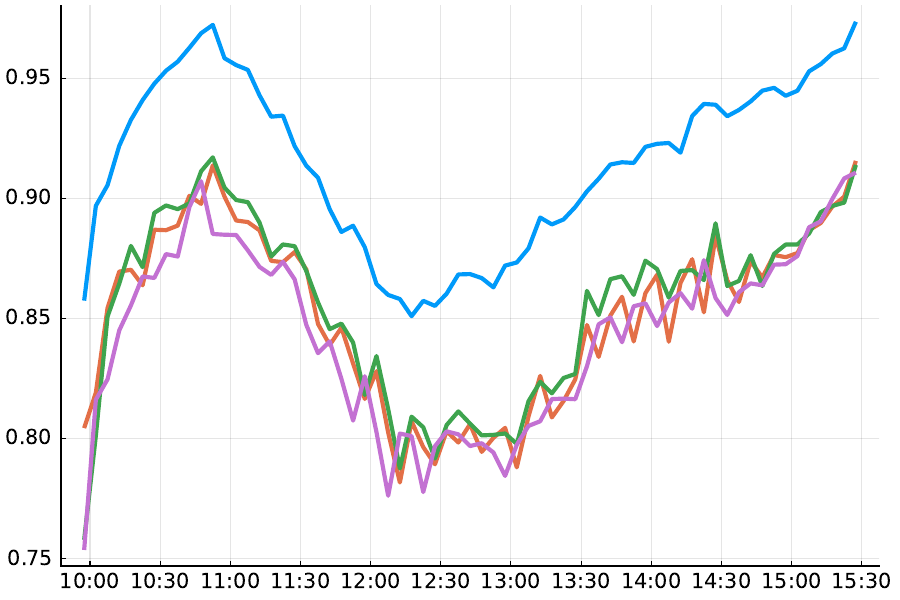}
\par\end{centering}
}
\par\end{centering}
\begin{centering}
\subfloat[PG]{\begin{centering}
\includegraphics[width=0.1\textwidth,height=0.06\textheight]{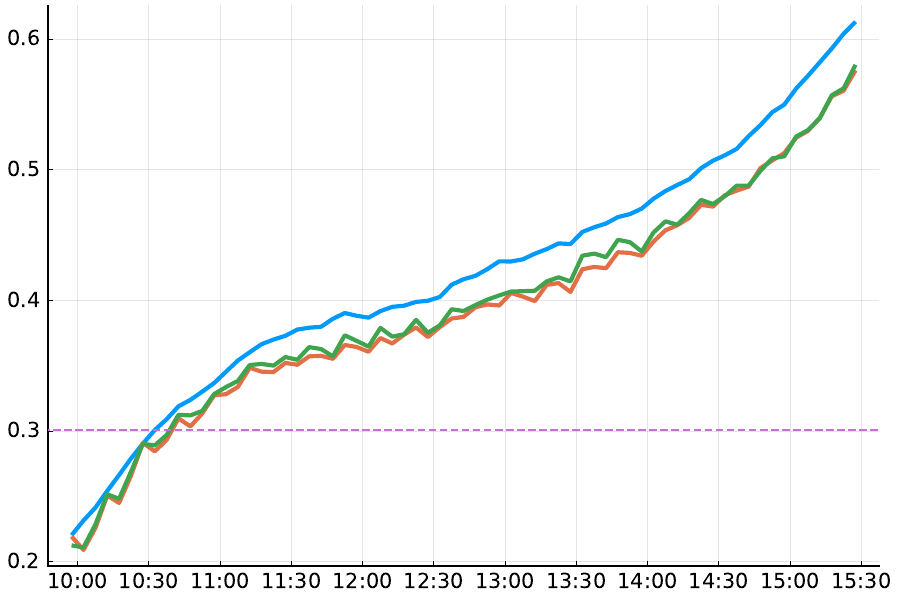}\includegraphics[width=0.1\textwidth,height=0.06\textheight]{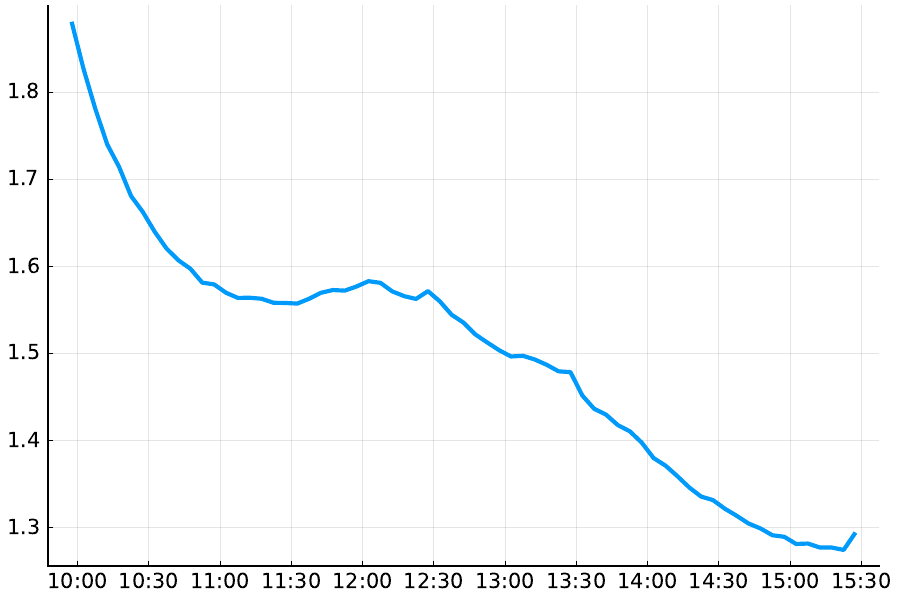}\includegraphics[width=0.1\textwidth,height=0.06\textheight]{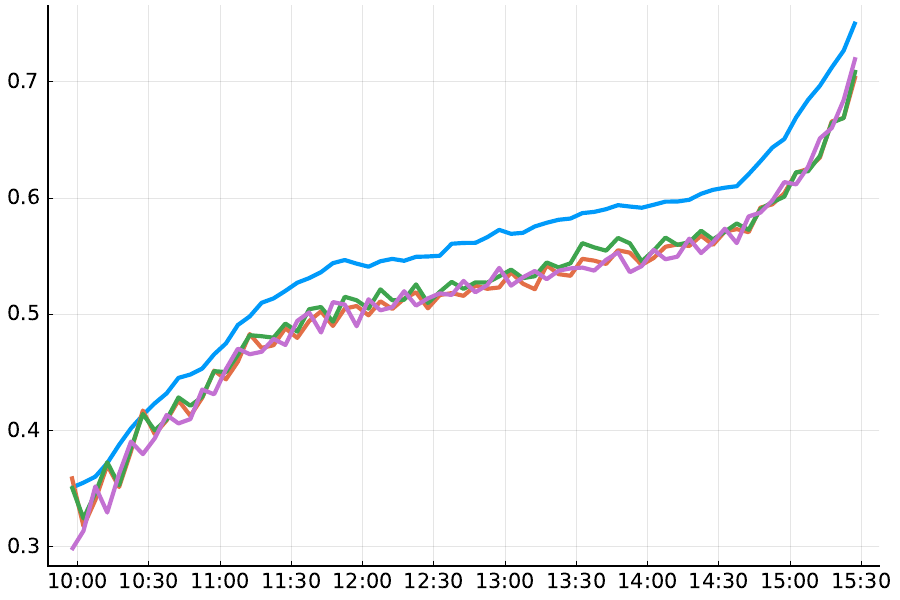}
\par\end{centering}
}\subfloat[WMT]{\begin{centering}
\includegraphics[width=0.1\textwidth,height=0.06\textheight]{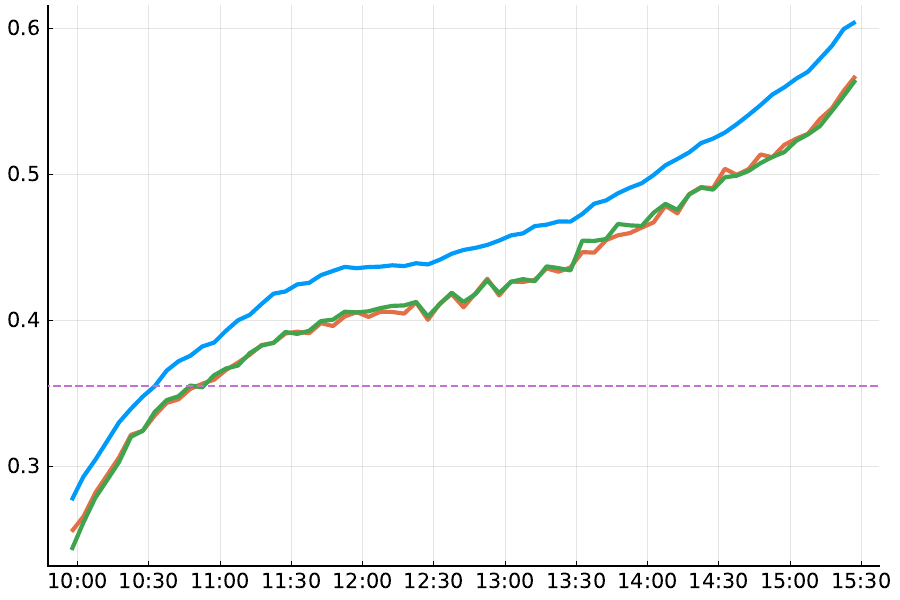}\includegraphics[width=0.1\textwidth,height=0.06\textheight]{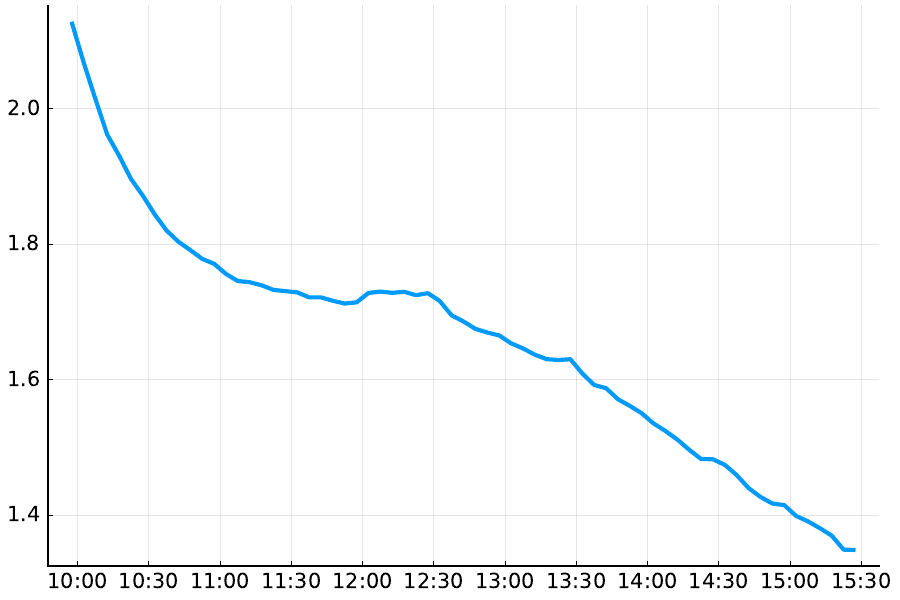}\includegraphics[width=0.1\textwidth,height=0.06\textheight]{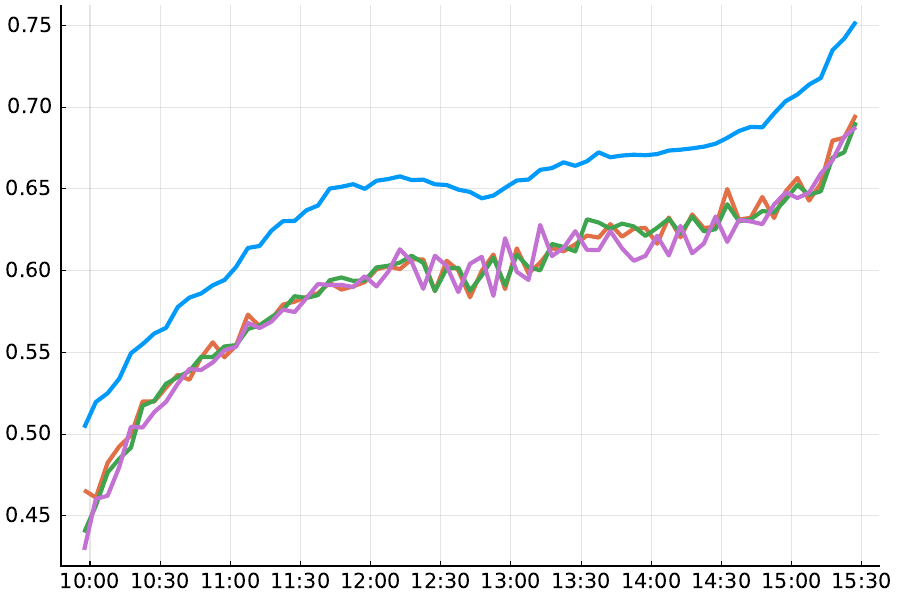}
\par\end{centering}
}\subfloat[TSLA]{\begin{centering}
\includegraphics[width=0.1\textwidth,height=0.06\textheight]{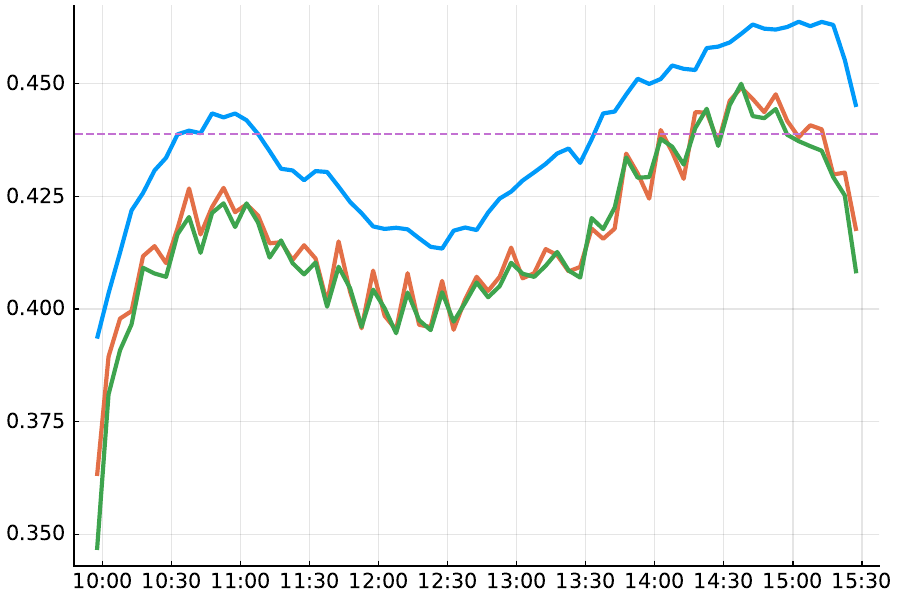}\includegraphics[width=0.1\textwidth,height=0.06\textheight]{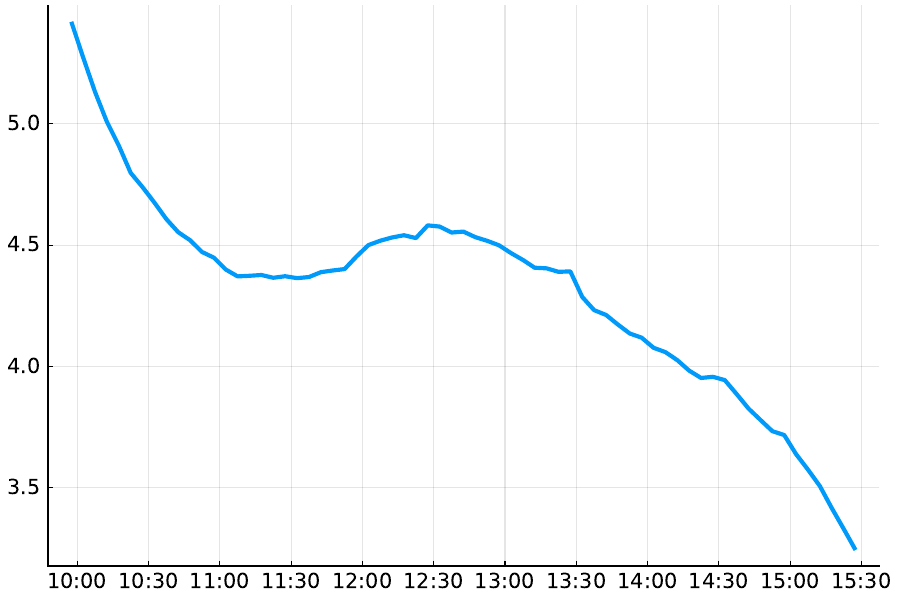}\includegraphics[width=0.1\textwidth,height=0.06\textheight]{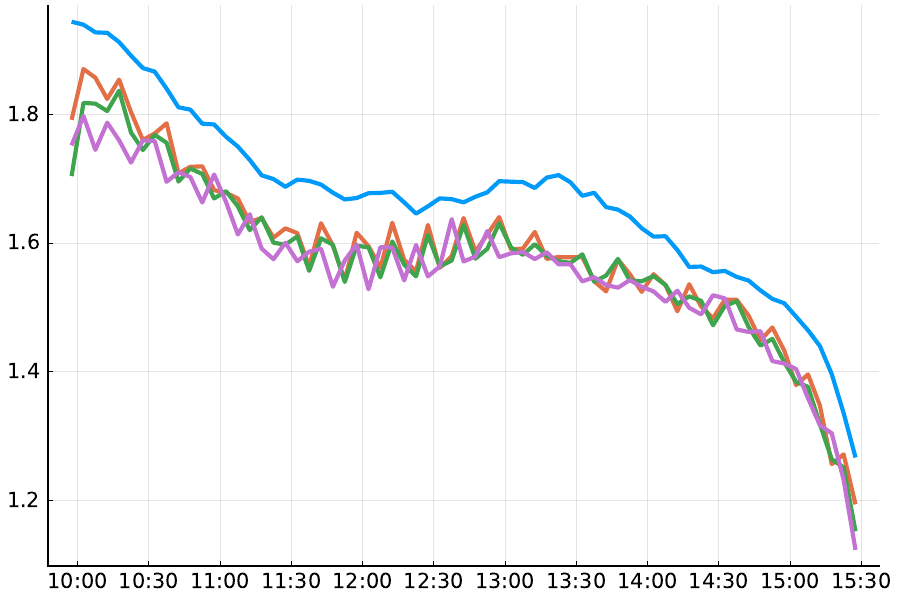}
\par\end{centering}
}
\par\end{centering}
\begin{centering}
\subfloat[AMZN]{\begin{centering}
\includegraphics[width=0.1\textwidth,height=0.06\textheight]{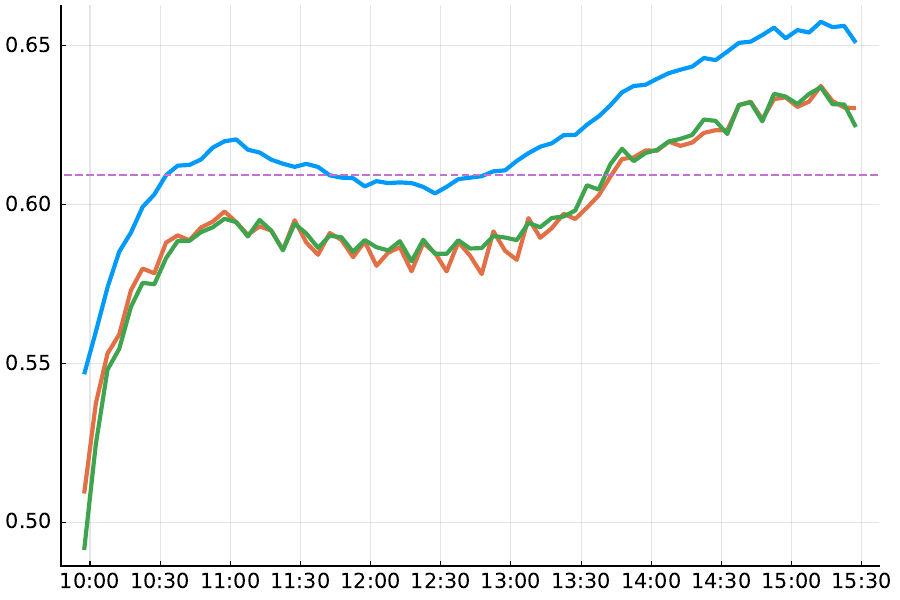}\includegraphics[width=0.1\textwidth,height=0.06\textheight]{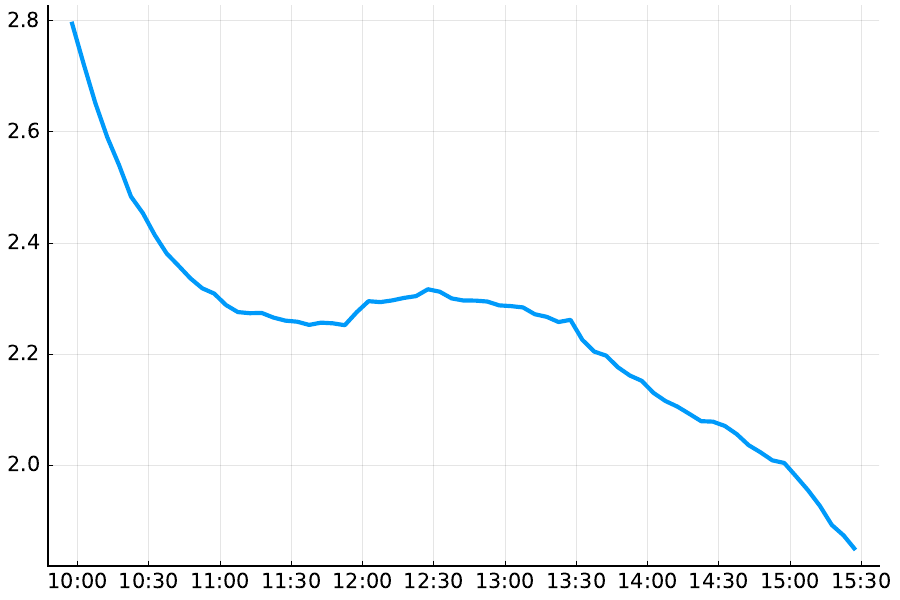}\includegraphics[width=0.1\textwidth,height=0.06\textheight]{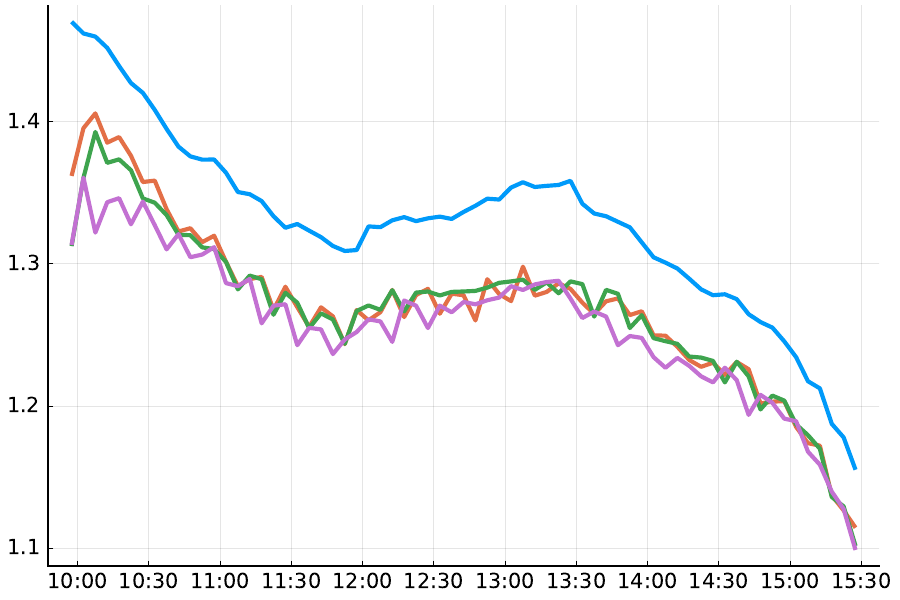}
\par\end{centering}
}\subfloat[DIS]{\begin{centering}
\includegraphics[width=0.1\textwidth,height=0.06\textheight]{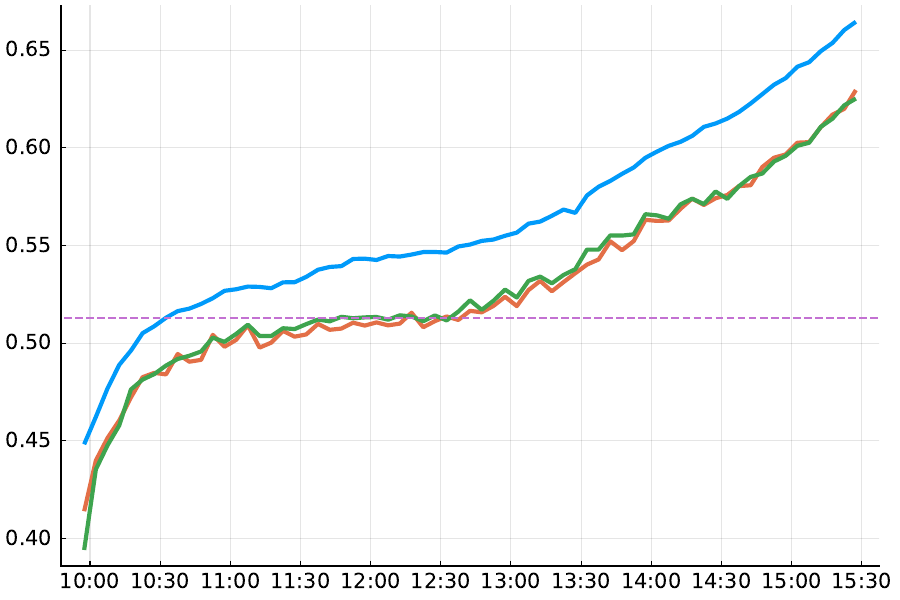}\includegraphics[width=0.1\textwidth,height=0.06\textheight]{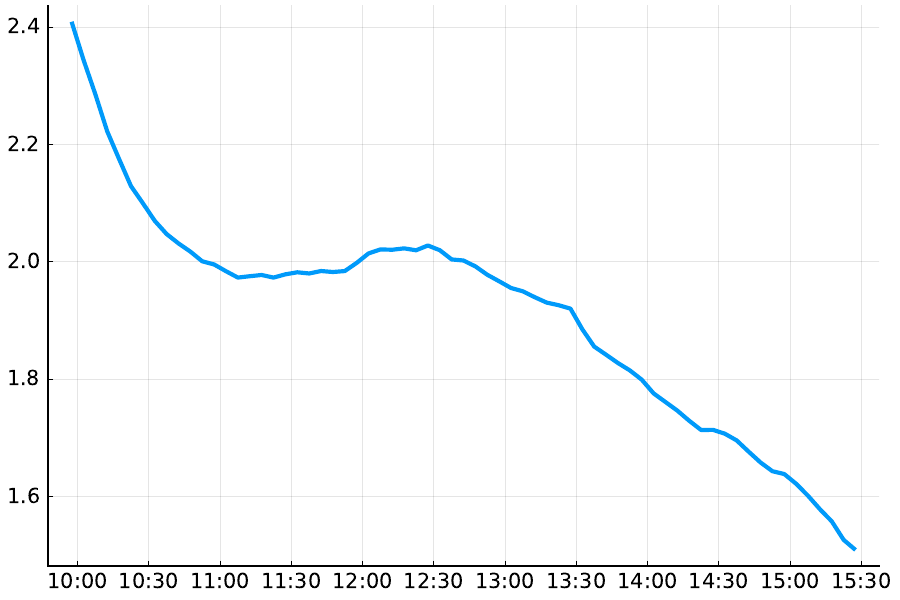}\includegraphics[width=0.1\textwidth,height=0.06\textheight]{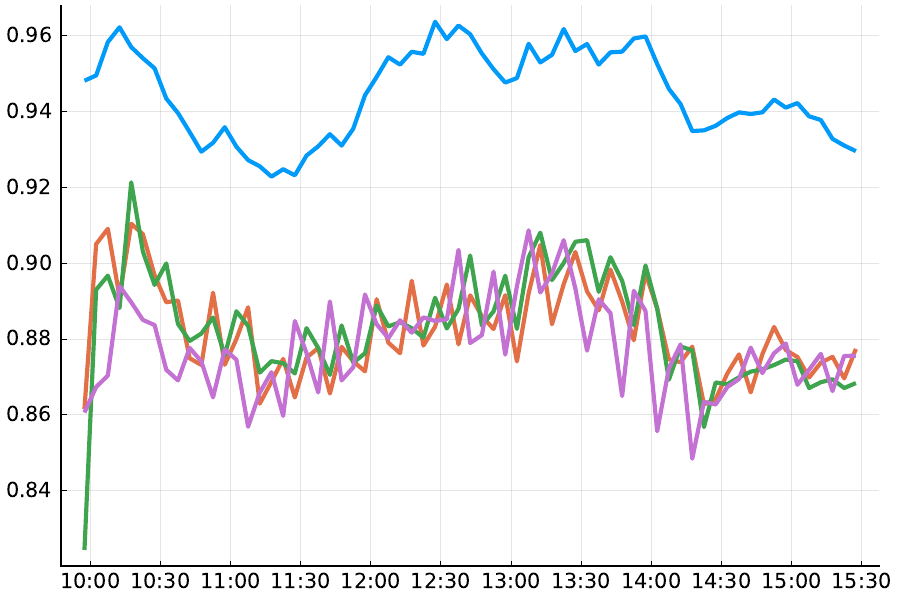}
\par\end{centering}
}\subfloat[FB]{\begin{centering}
\includegraphics[width=0.1\textwidth,height=0.06\textheight]{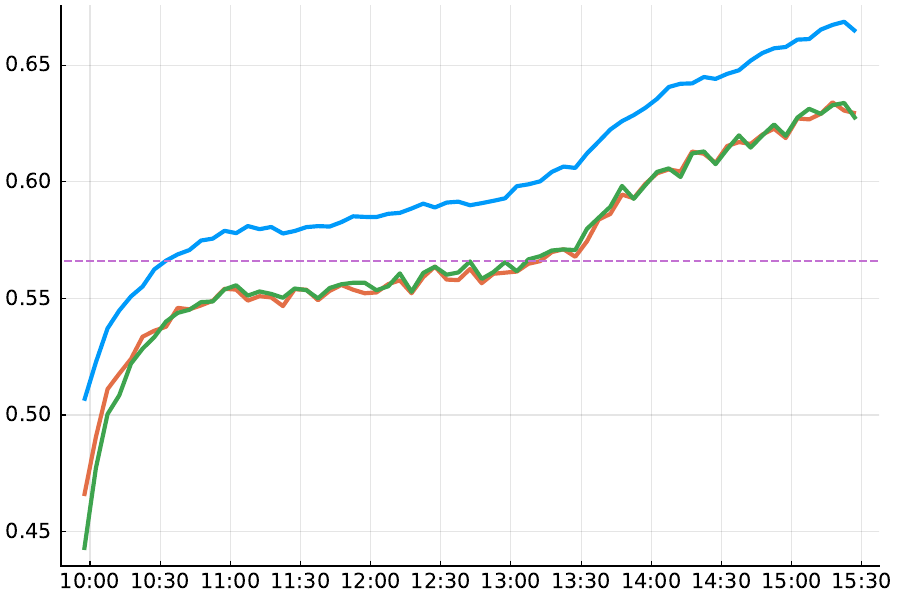}\includegraphics[width=0.1\textwidth,height=0.06\textheight]{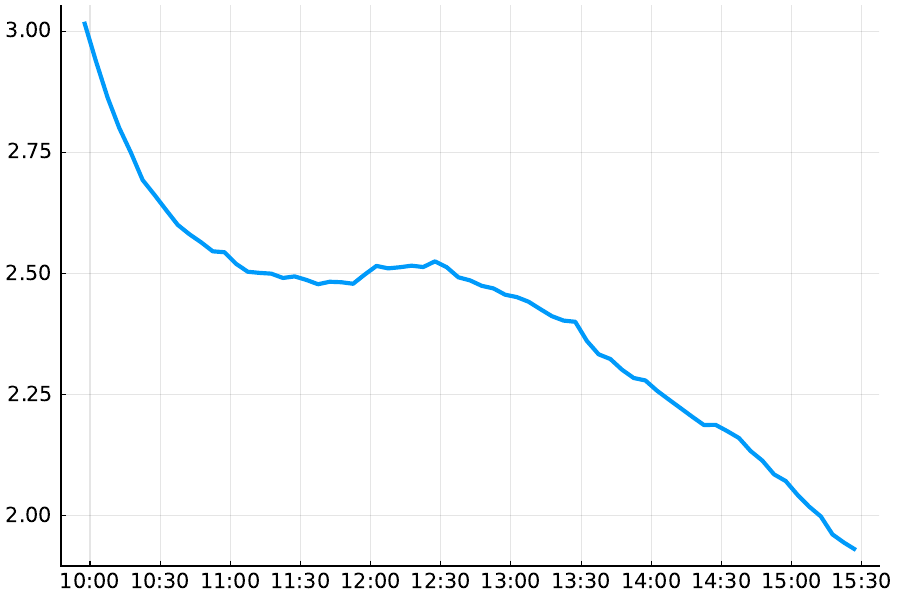}\includegraphics[width=0.1\textwidth,height=0.06\textheight]{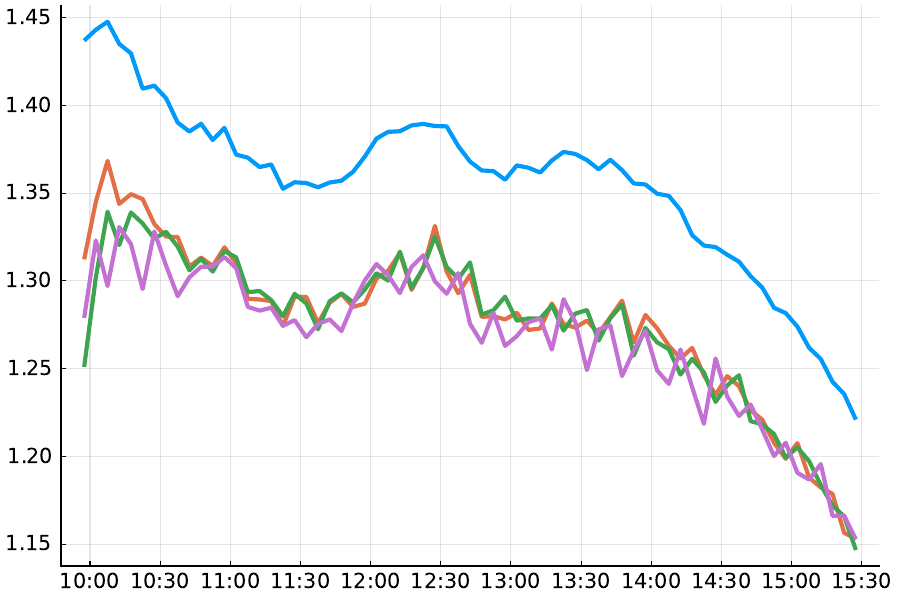}
\par\end{centering}
}
\par\end{centering}
\caption{Average intraday correlations (left), relative volatility (middle),
and market $\beta$ (right) for the remaining 18 stocks in the Small Universe. As in Figure \ref{fig:RhoLambdaBetaIntraday},
the rolling windows span $60$ minutes and each estimate is plotted
against the midpoint of its window.\label{fig:RhoLambdaBetasSmallUniverse}}
\end{figure}
 The correlation increases for all assets over the trading day, and
the relative volatility decreases. The effect on the market beta (the
product of the two) is therefore determined by which effect dominates. 

The correlation signature plots are based on averages over many trading
days. It is therefore important to explore if the findings are robust
to the choice of sample period. Correlation signature plots for each sub-period
are shown in Figure \ref{fig:SplitSignaturePlotsMarketCorr} for each
stock in the Small Universe and SPY, and for each pair of assets
within the same sector in Figure \ref{fig:SplitSignaturePlotsSector}.
\begin{figure}[H]
\begin{centering}
\includegraphics[width=0.2\textwidth]{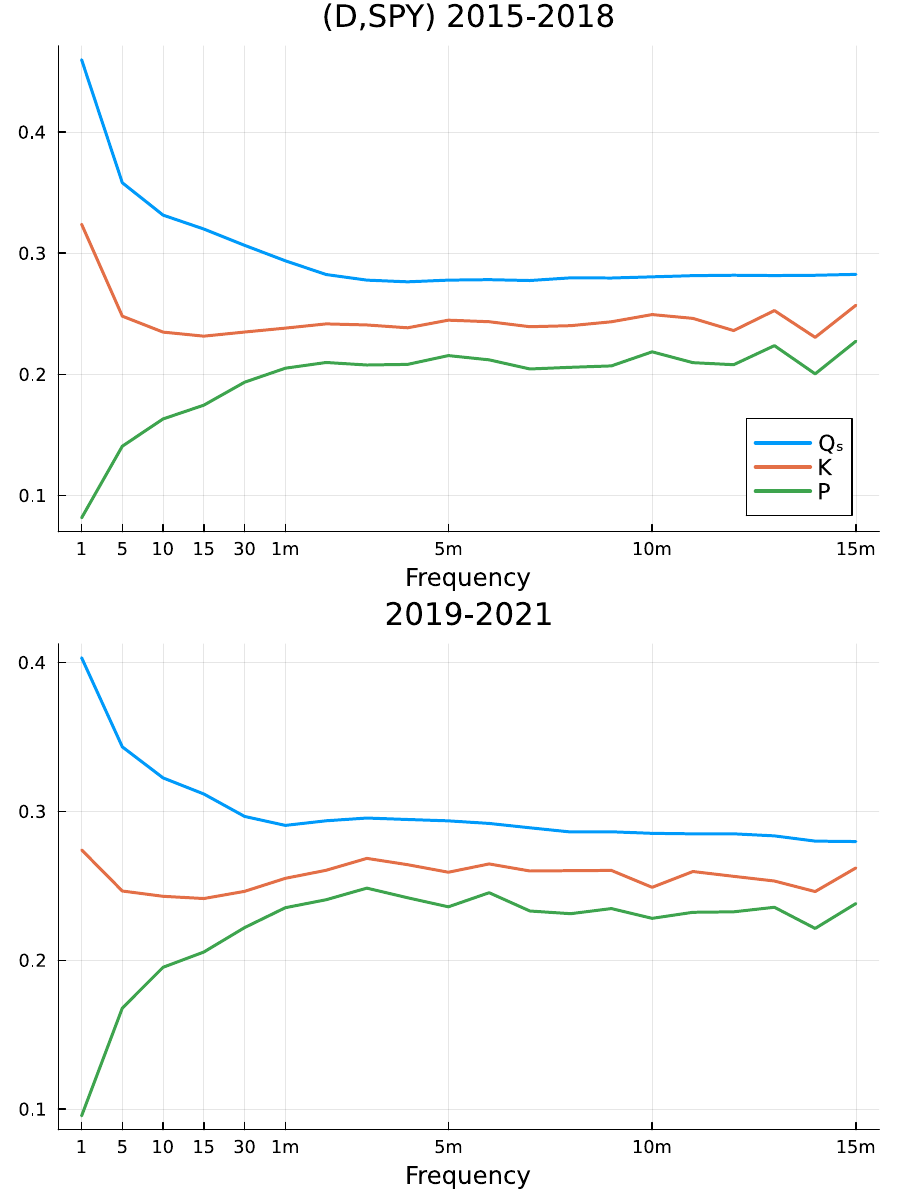}\includegraphics[width=0.2\textwidth]{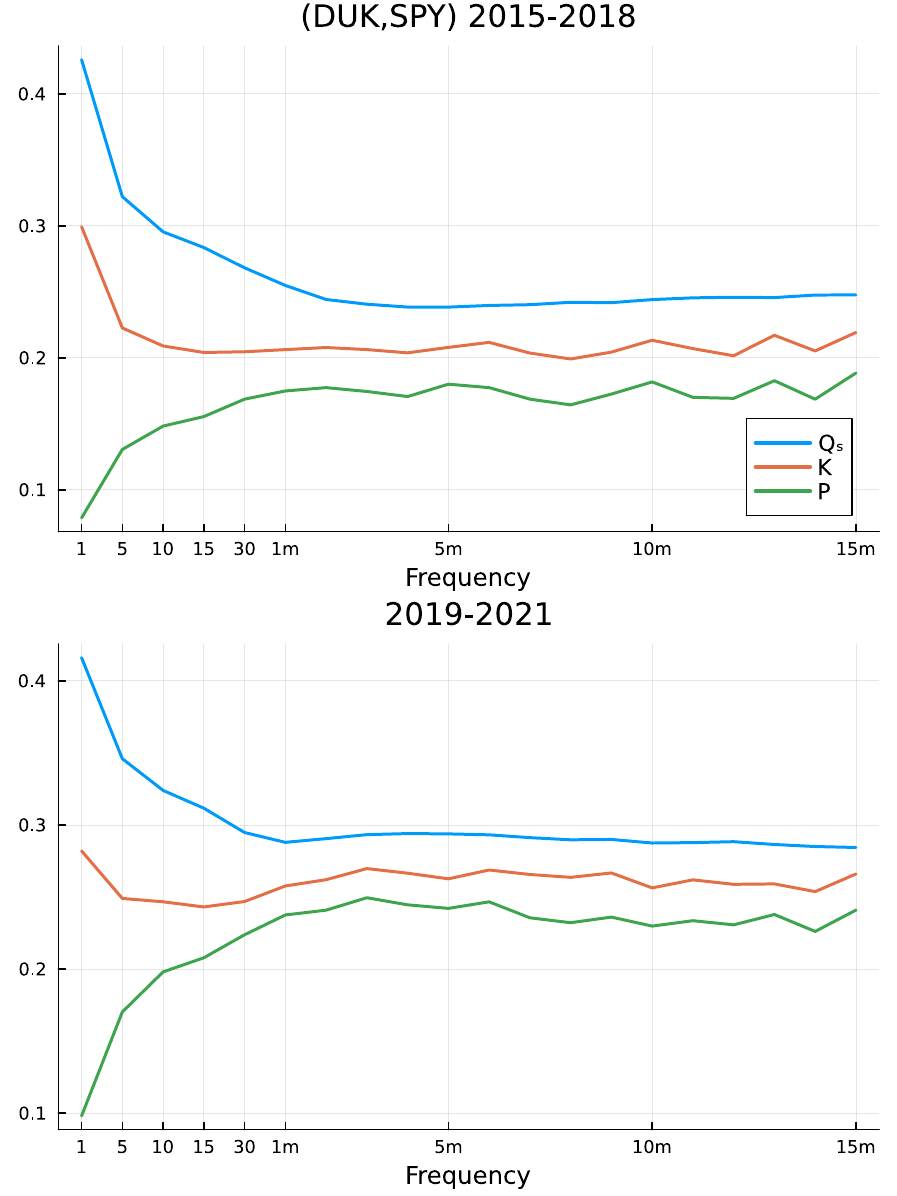}\includegraphics[width=0.2\textwidth]{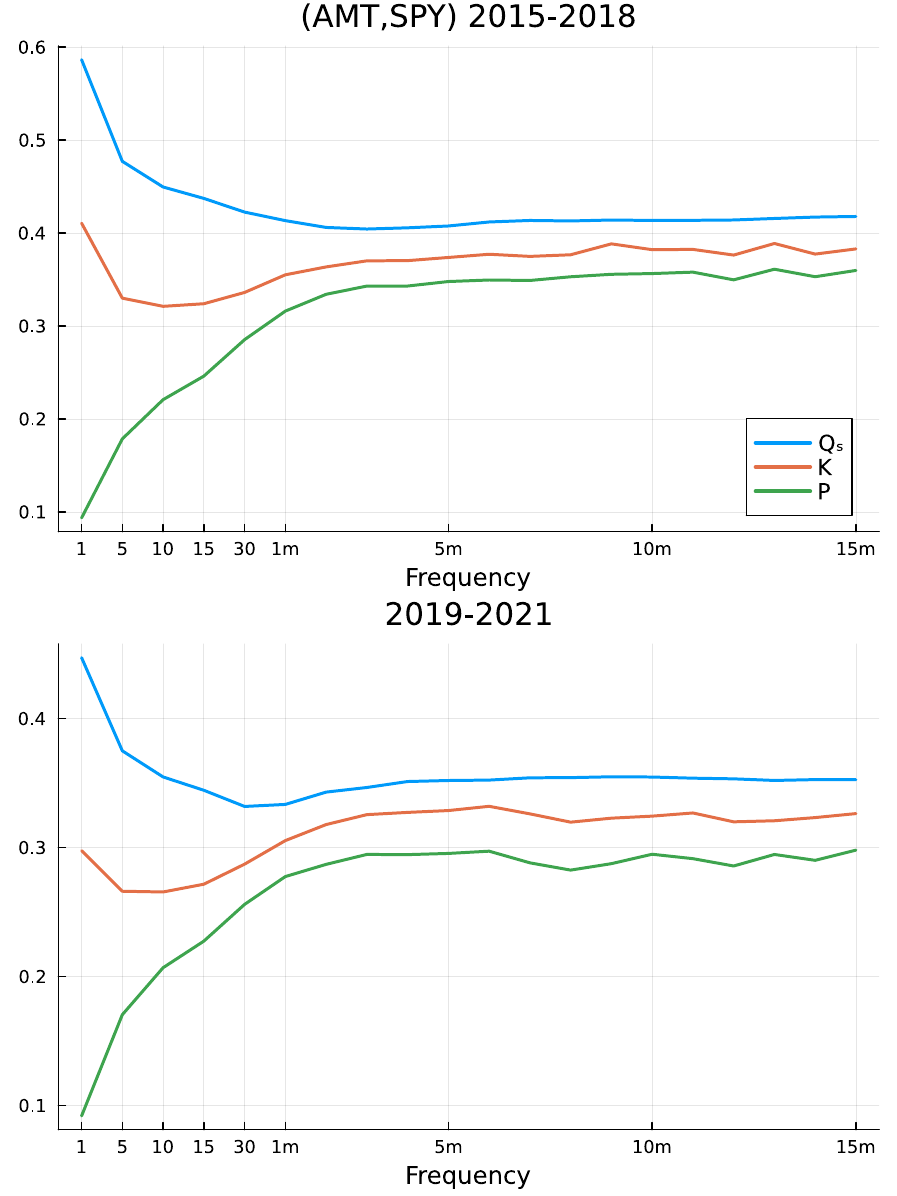}\includegraphics[width=0.2\textwidth]{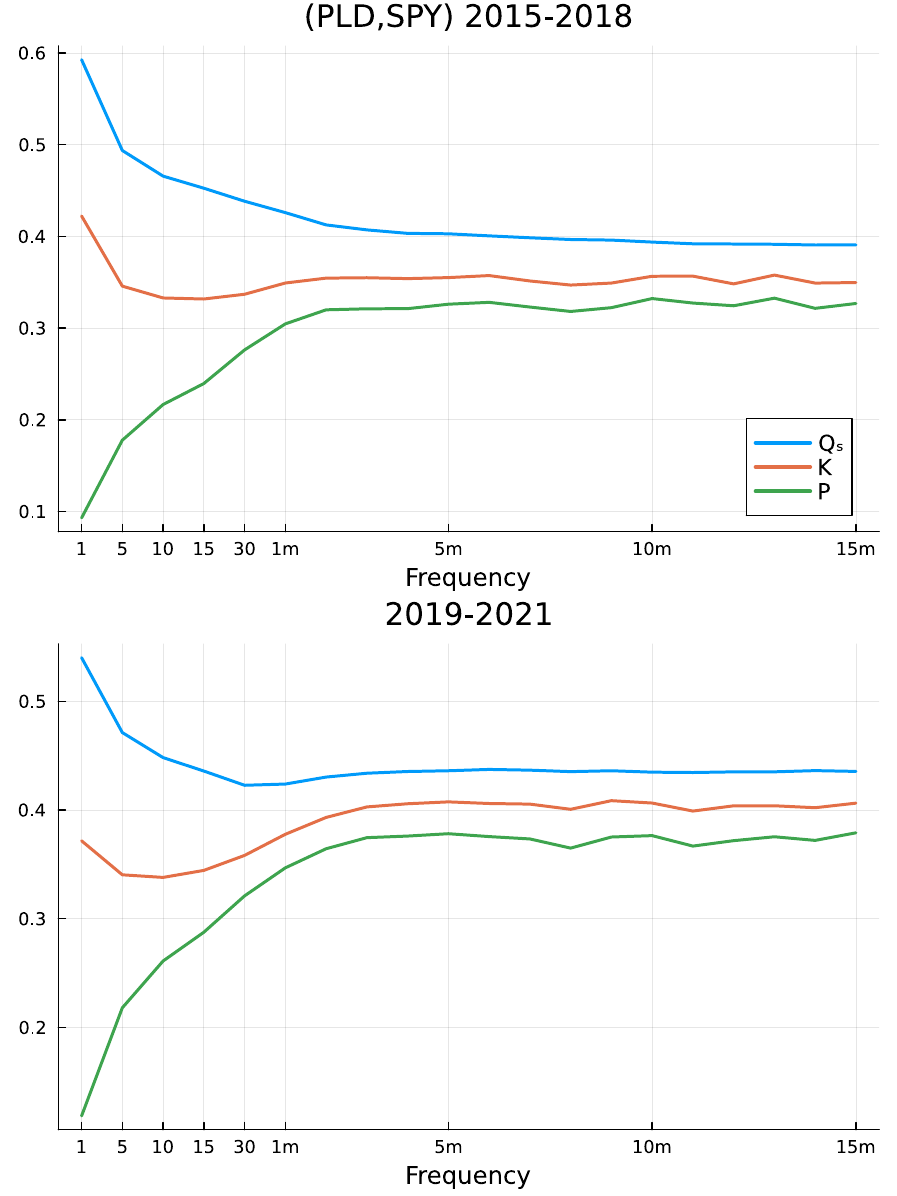}\includegraphics[width=0.2\textwidth]{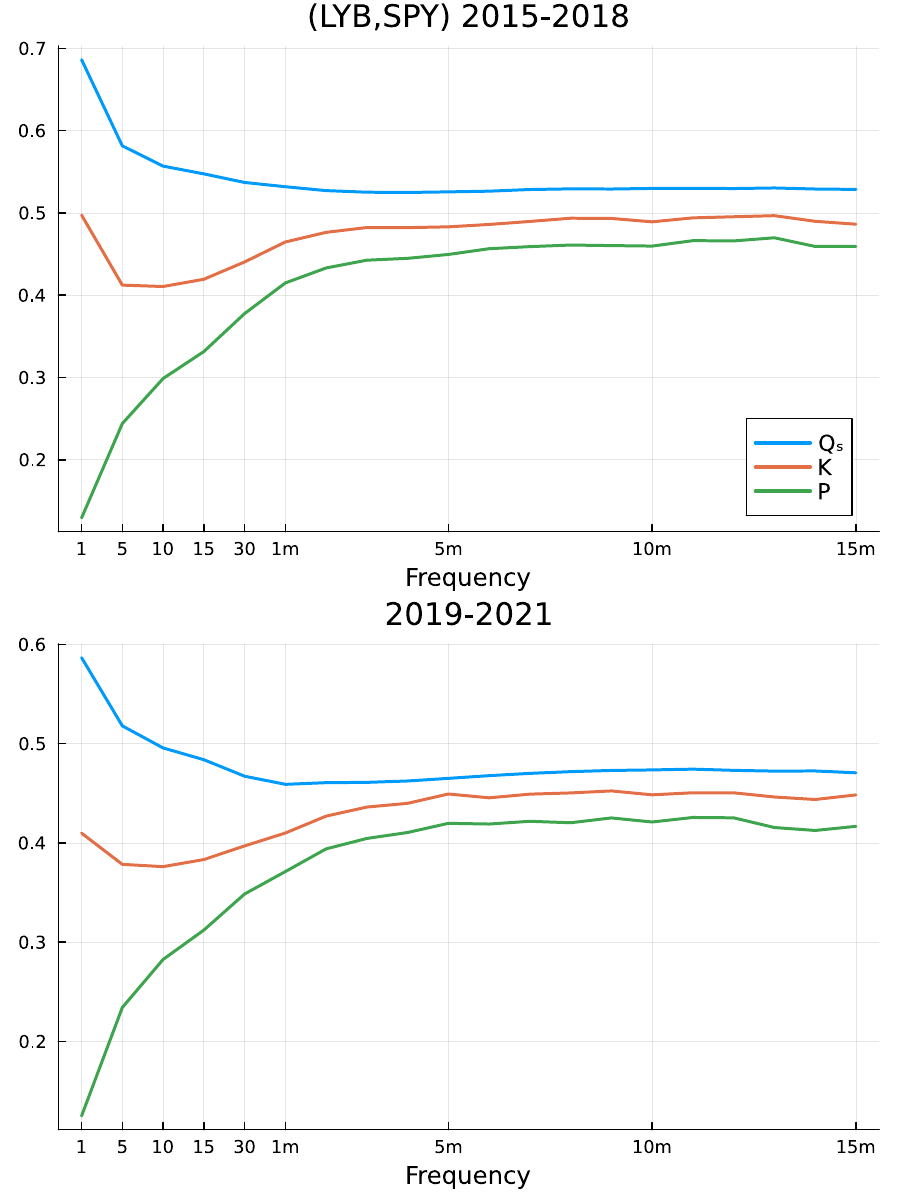}
\par\end{centering}
\begin{centering}
\includegraphics[width=0.2\textwidth]{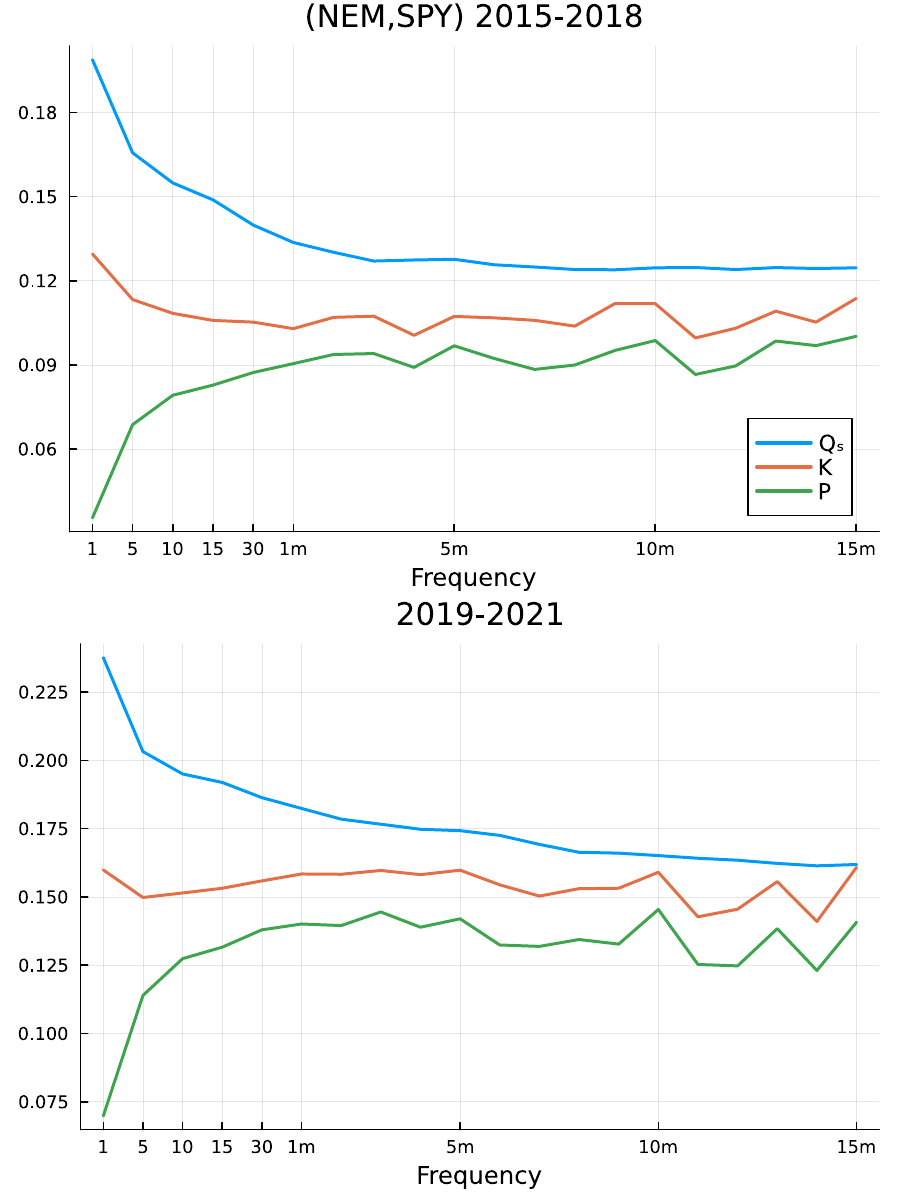}\includegraphics[width=0.2\textwidth]{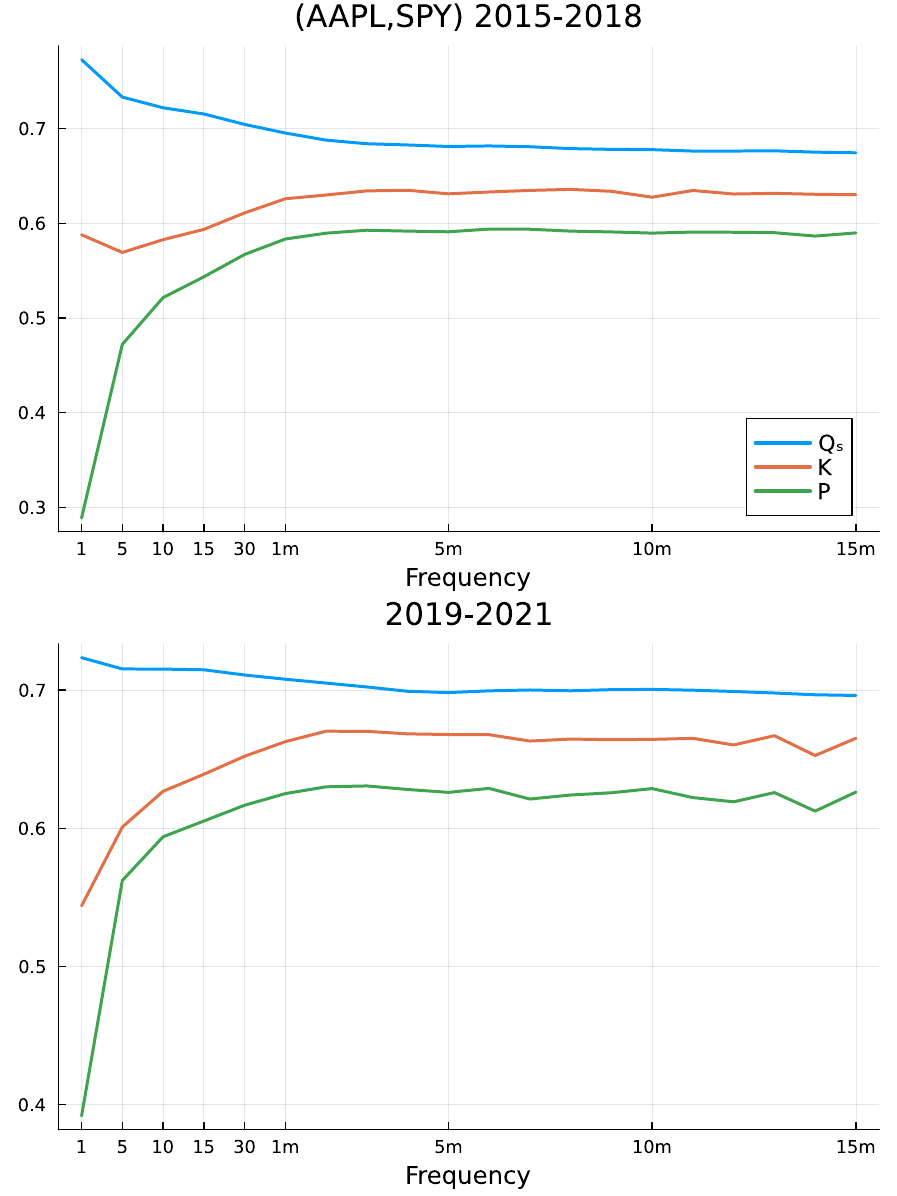}\includegraphics[width=0.2\textwidth]{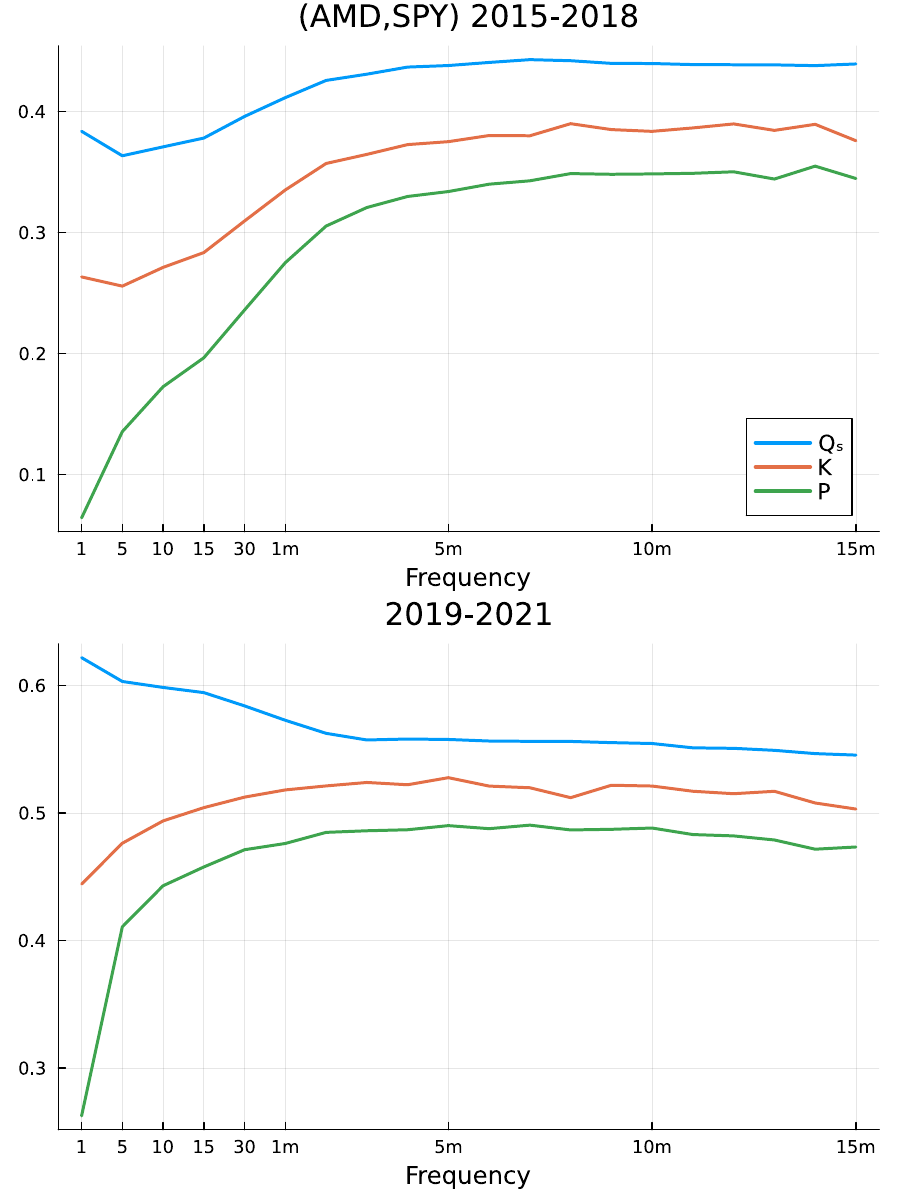}\includegraphics[width=0.2\textwidth]{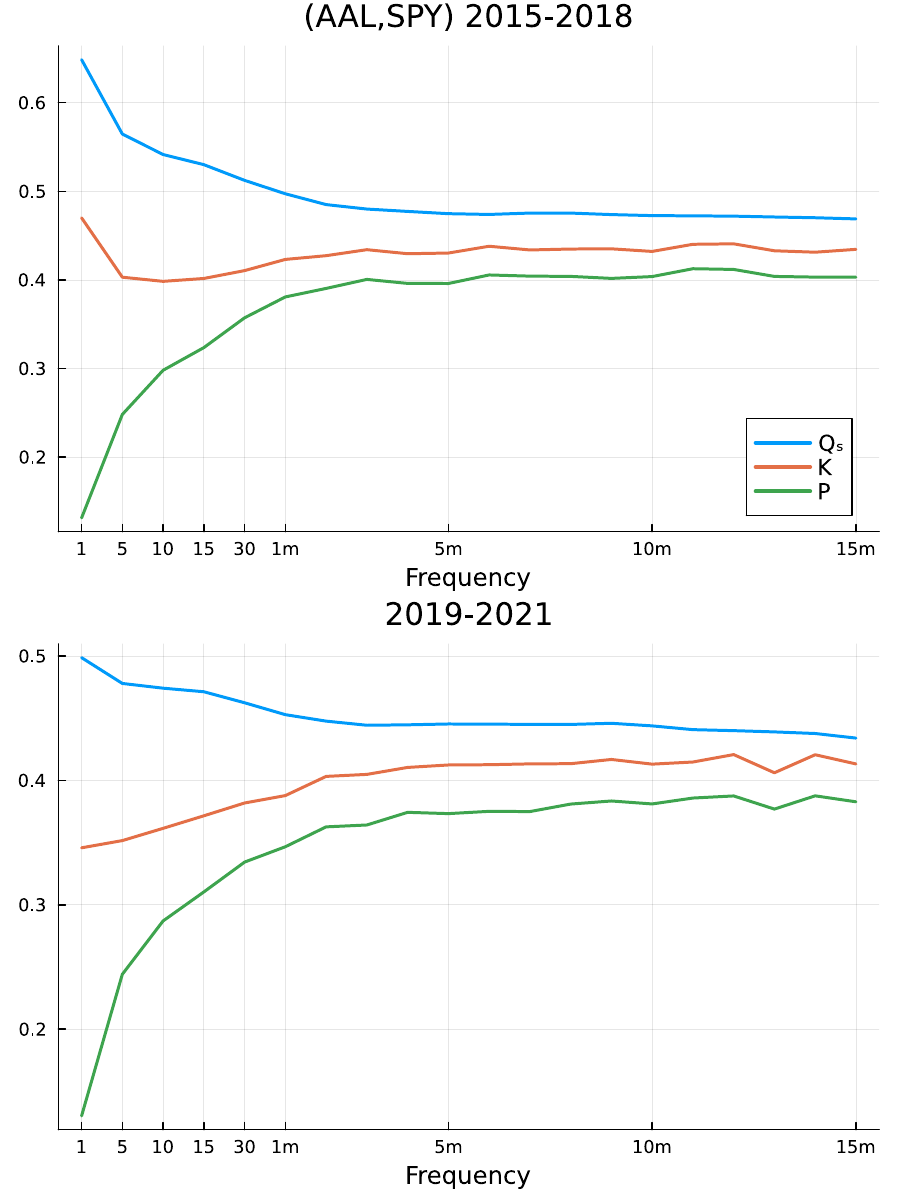}\includegraphics[width=0.2\textwidth]{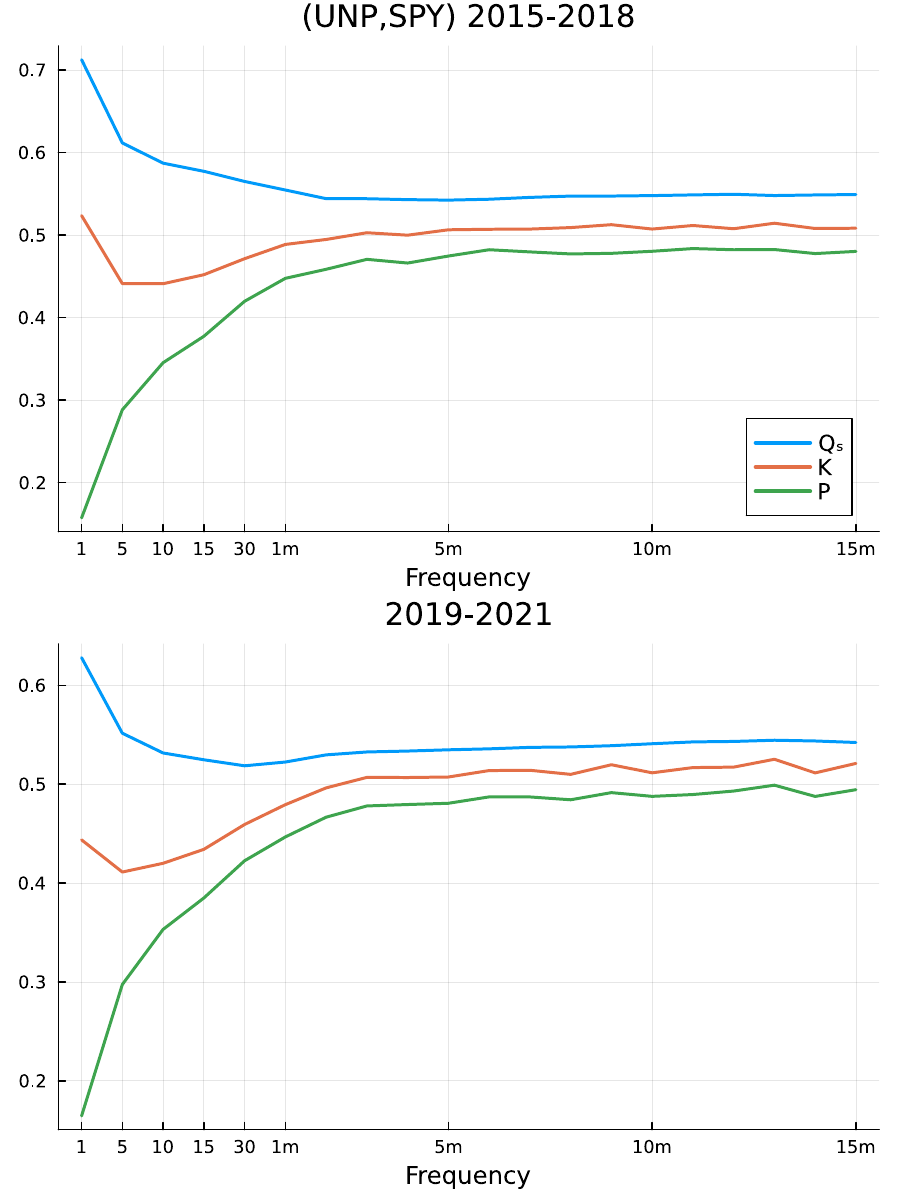}
\par\end{centering}
\begin{centering}
\includegraphics[width=0.2\textwidth]{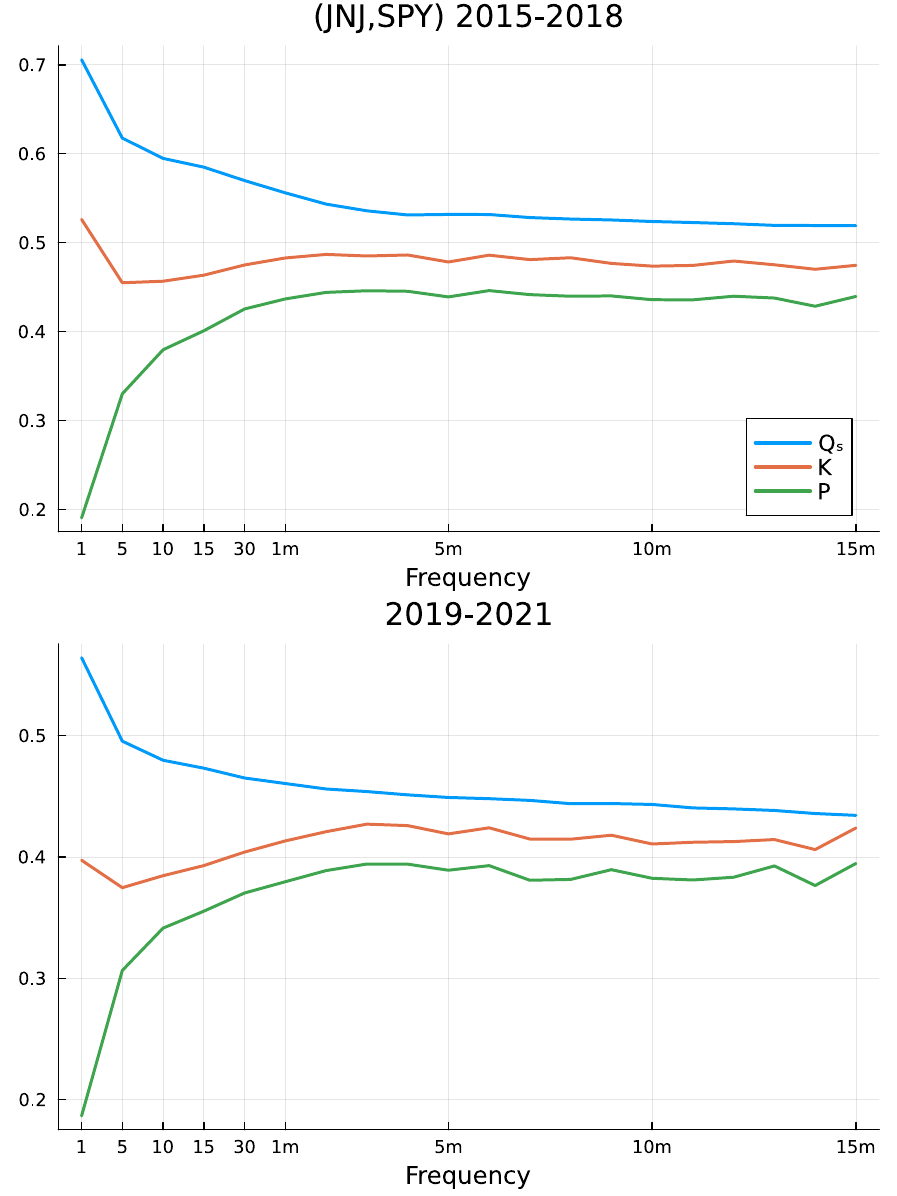}\includegraphics[width=0.2\textwidth]{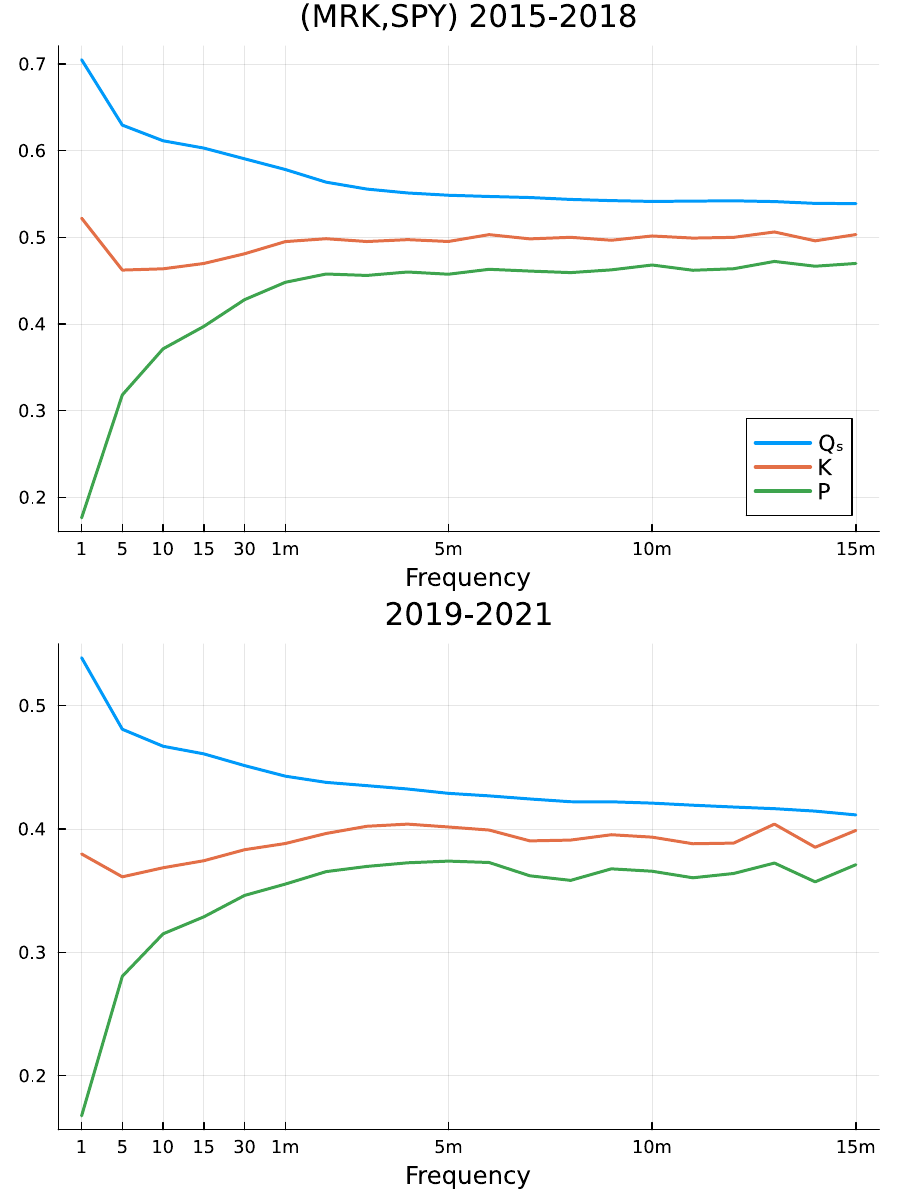}\includegraphics[width=0.2\textwidth]{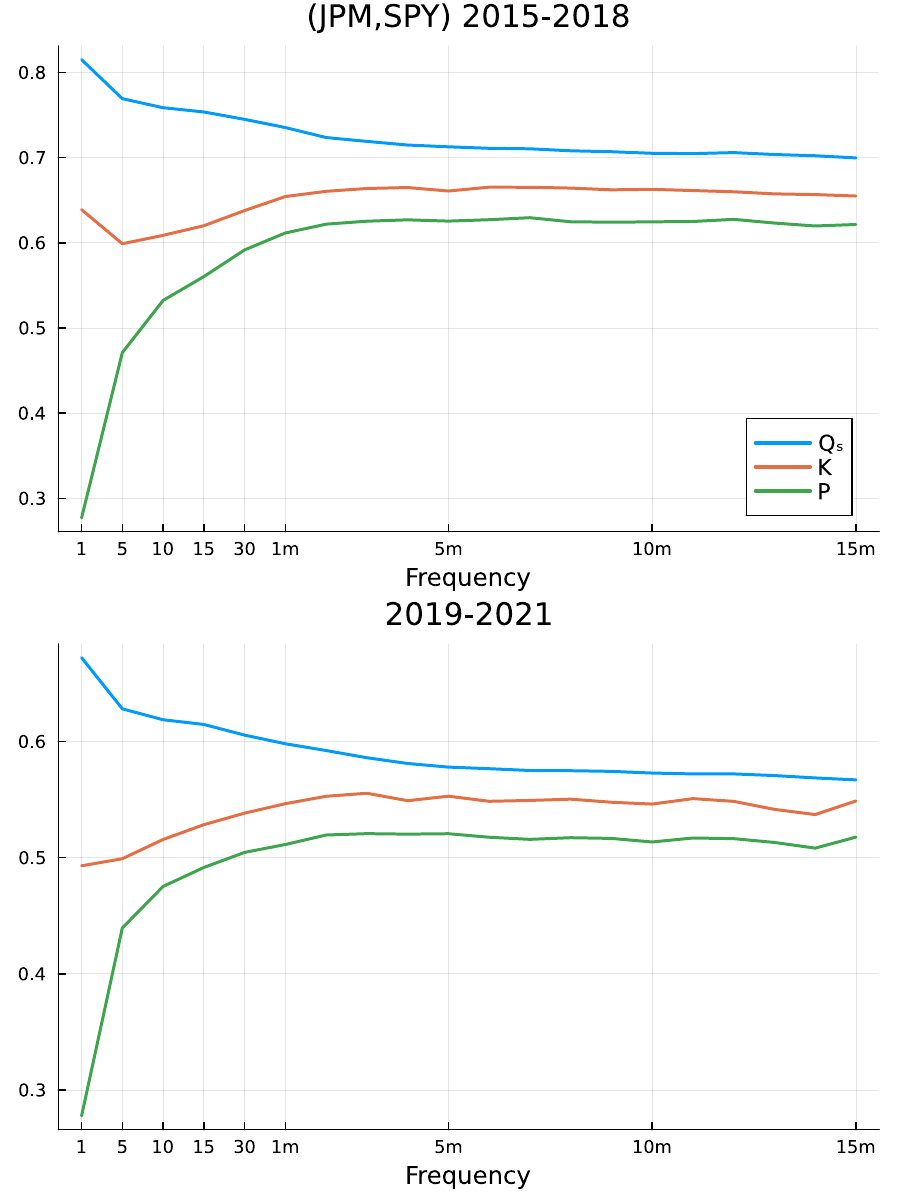}\includegraphics[width=0.2\textwidth]{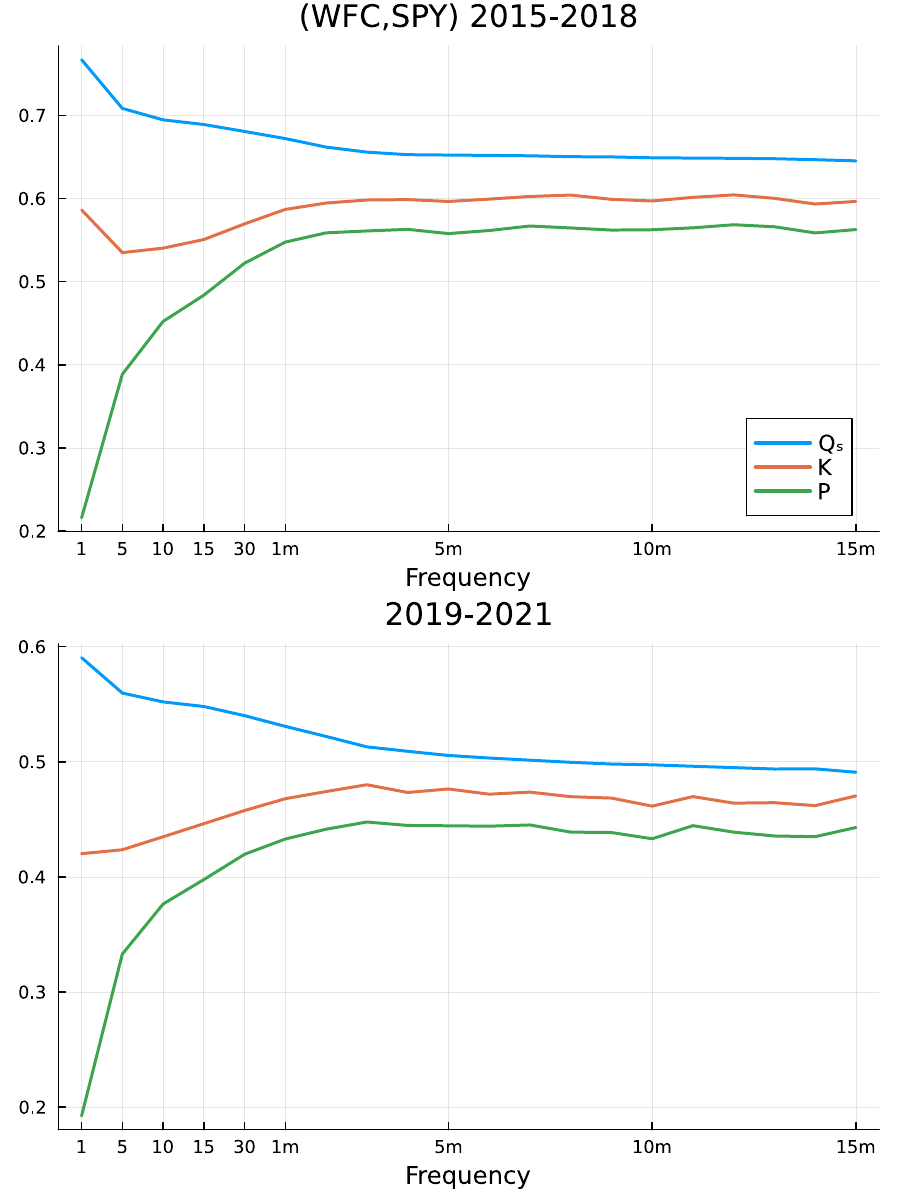}\includegraphics[width=0.2\textwidth]{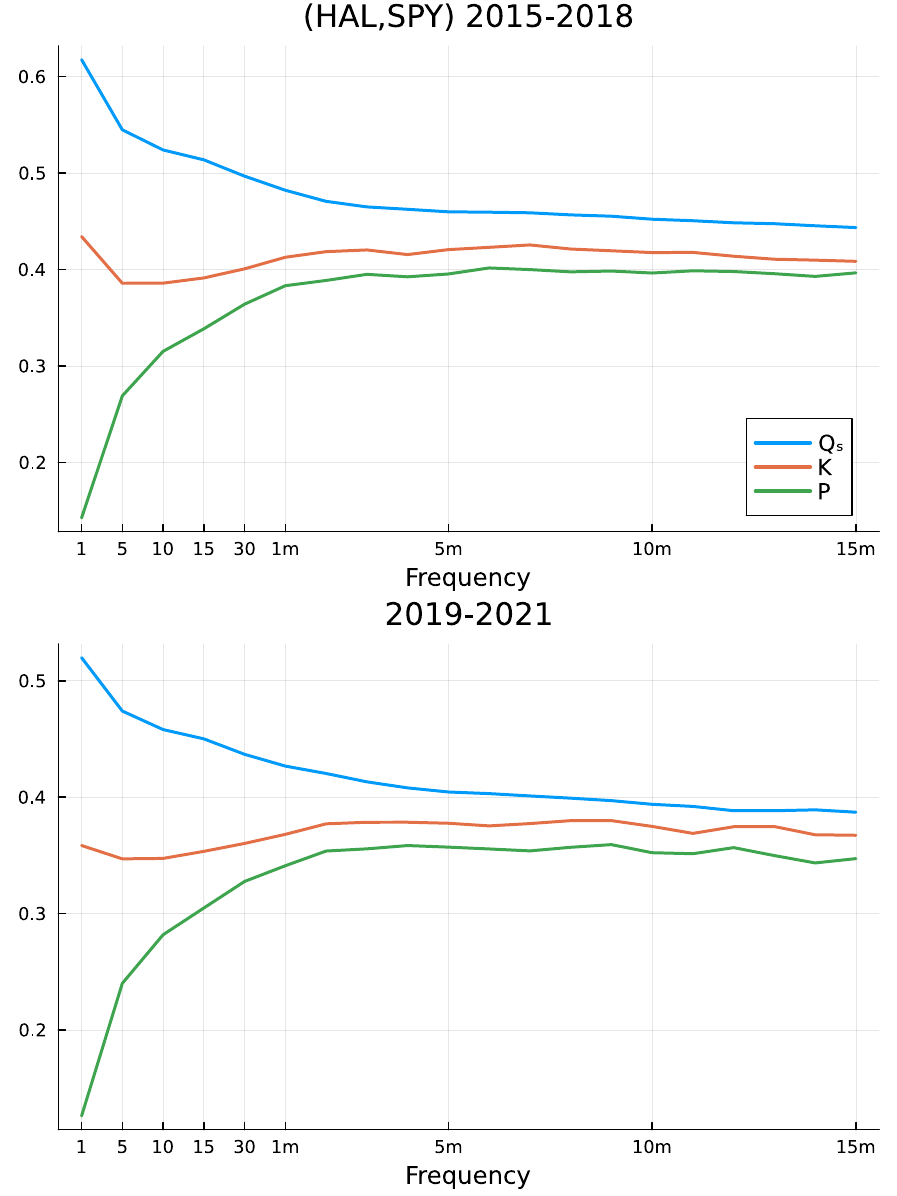}
\par\end{centering}
\begin{centering}
\includegraphics[width=0.2\textwidth]{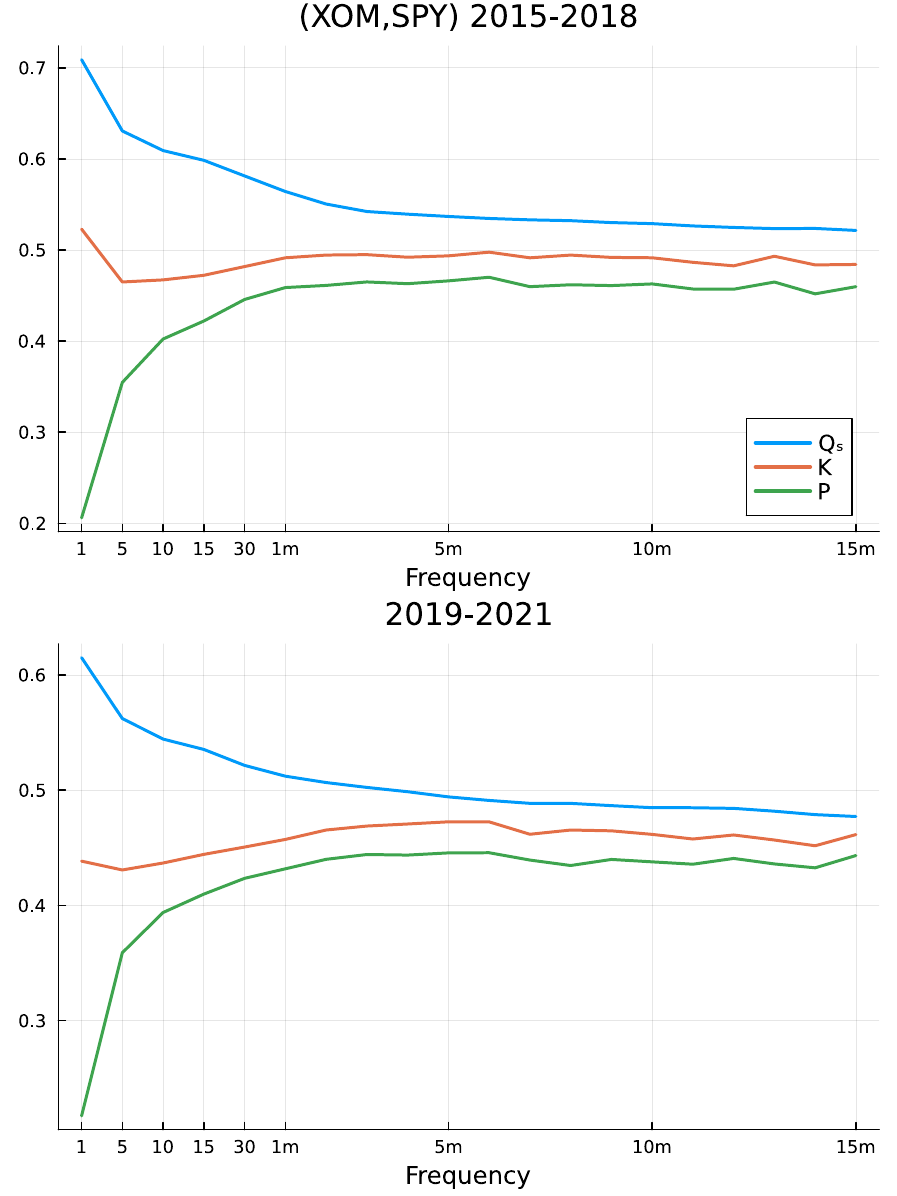}\includegraphics[width=0.2\textwidth]{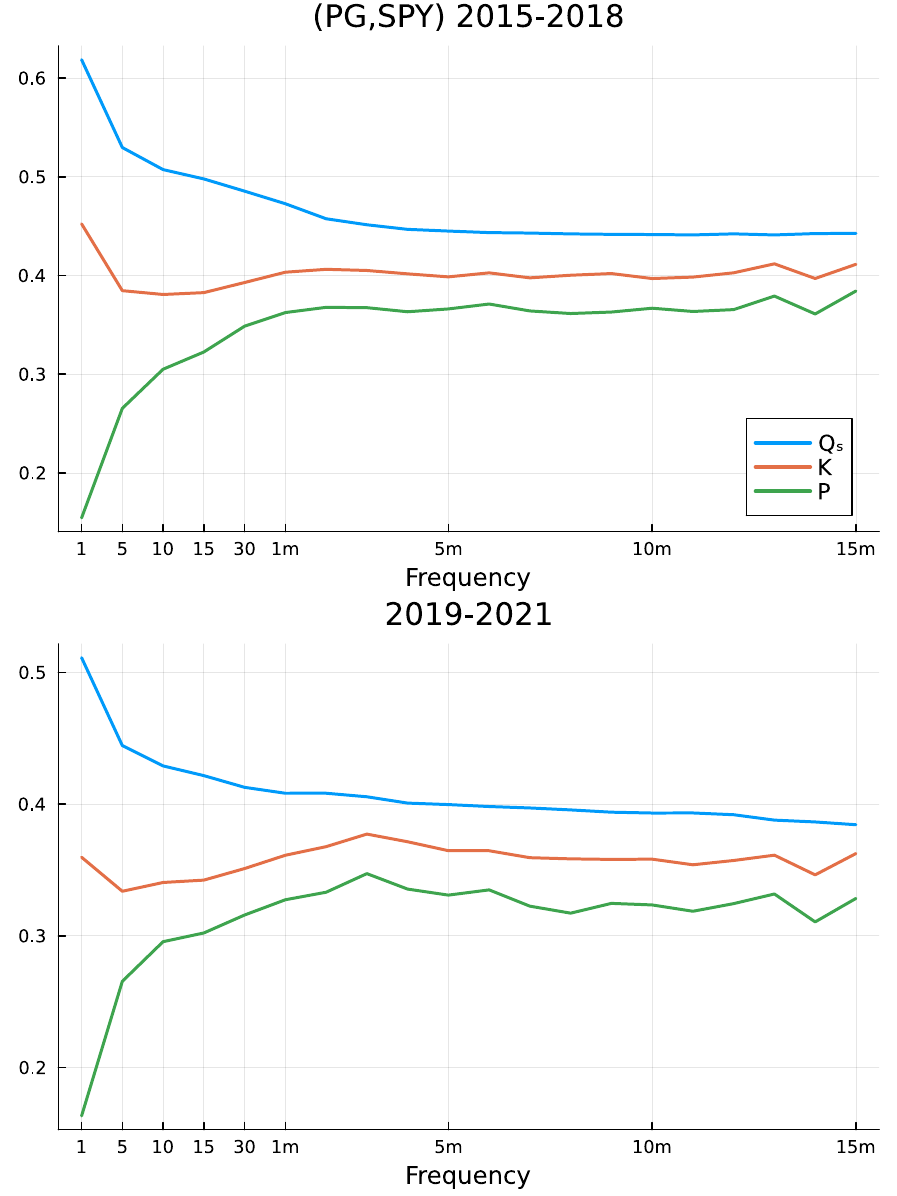}\includegraphics[width=0.2\textwidth]{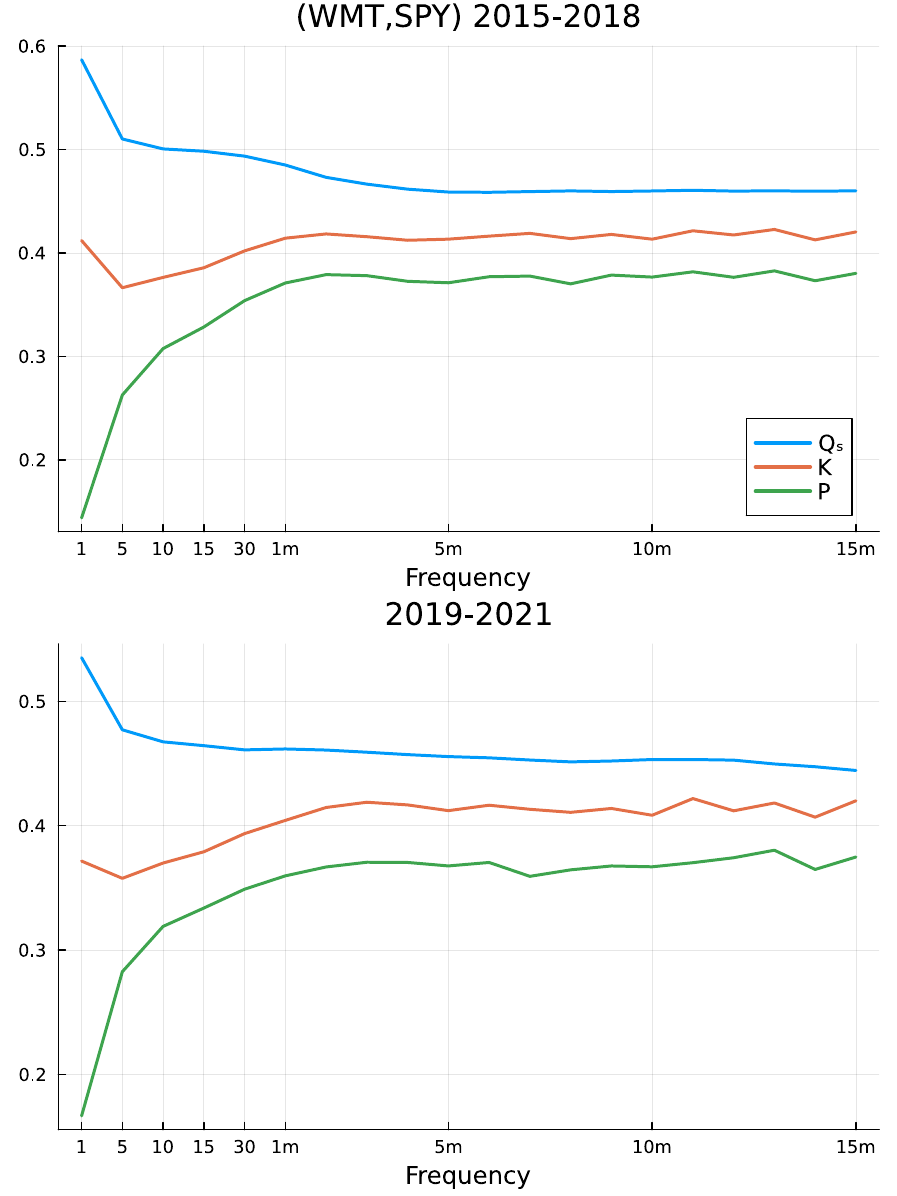}\includegraphics[width=0.2\textwidth]{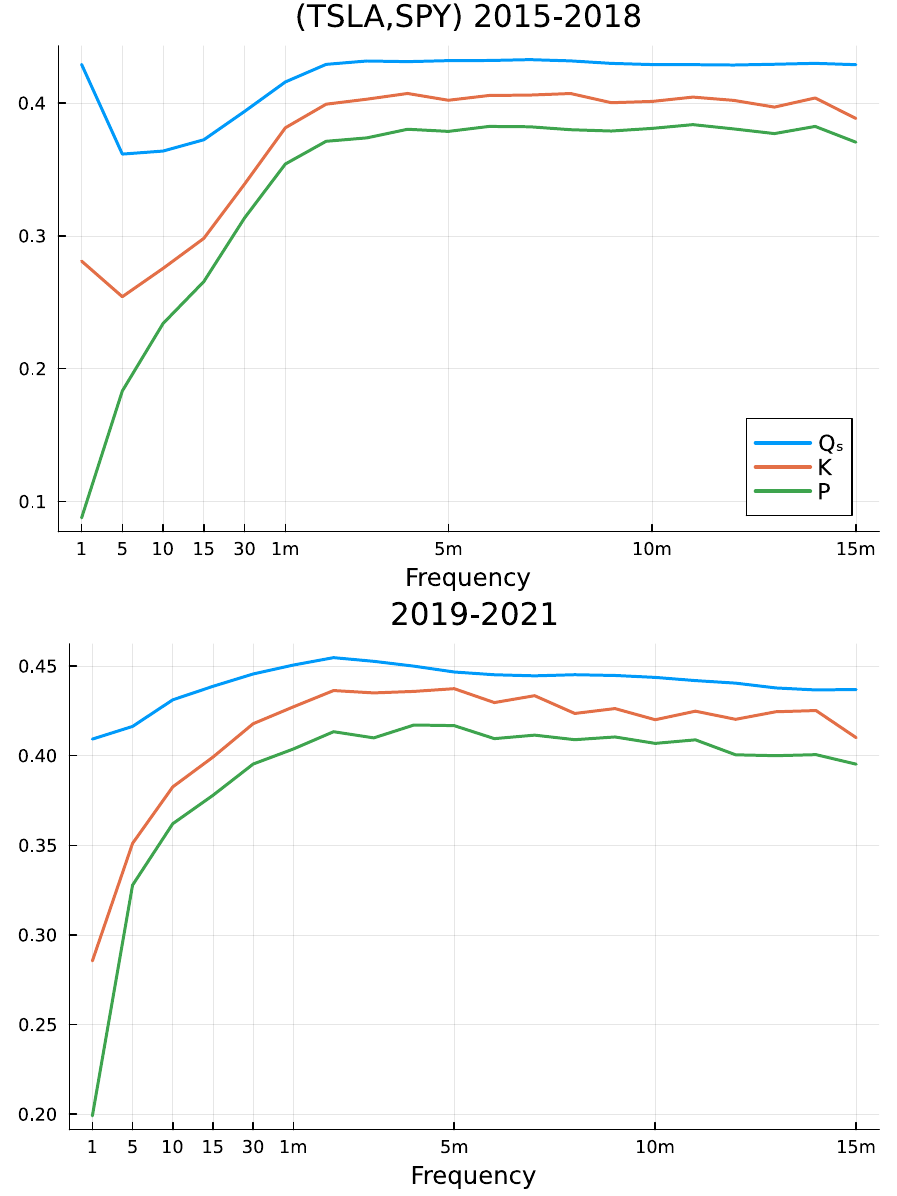}\includegraphics[width=0.2\textwidth]{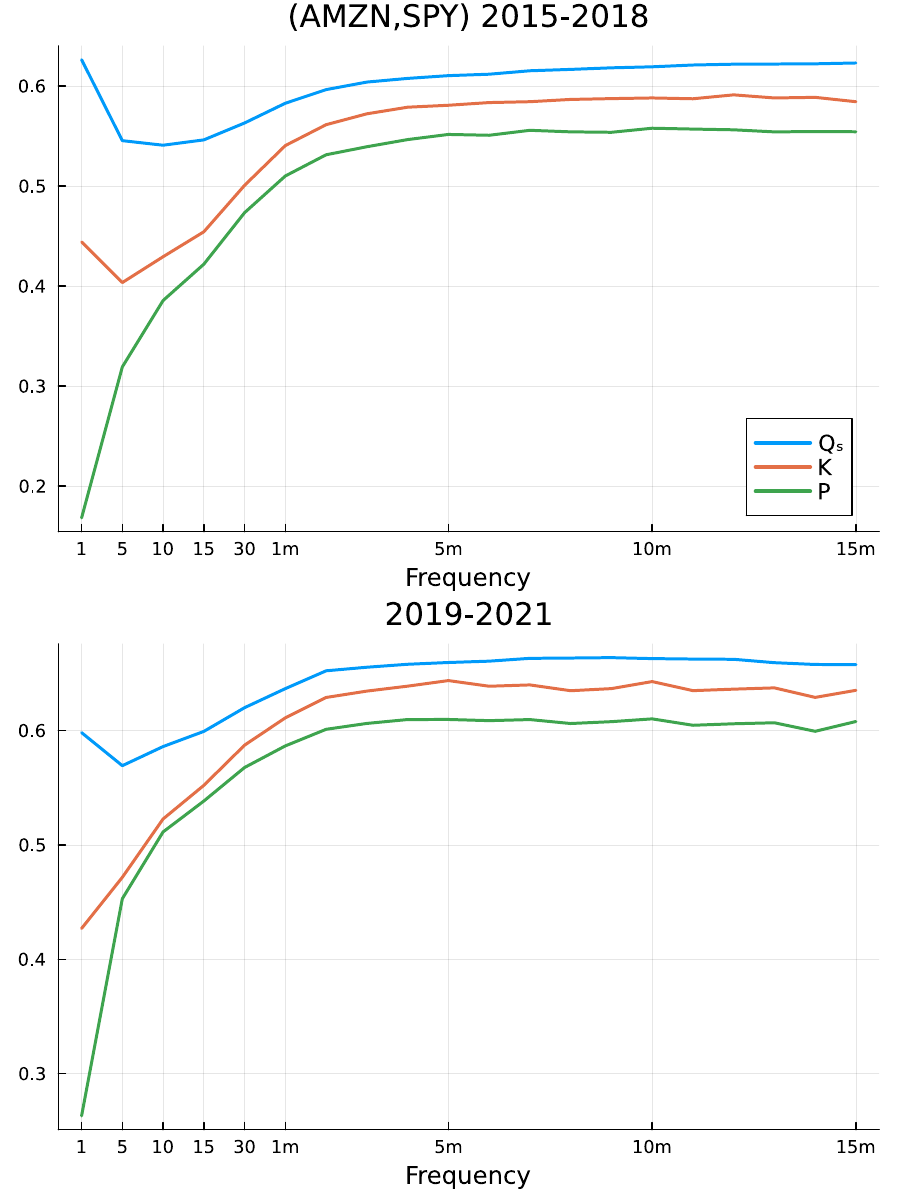}
\par\end{centering}
\begin{centering}
\includegraphics[width=0.2\textwidth]{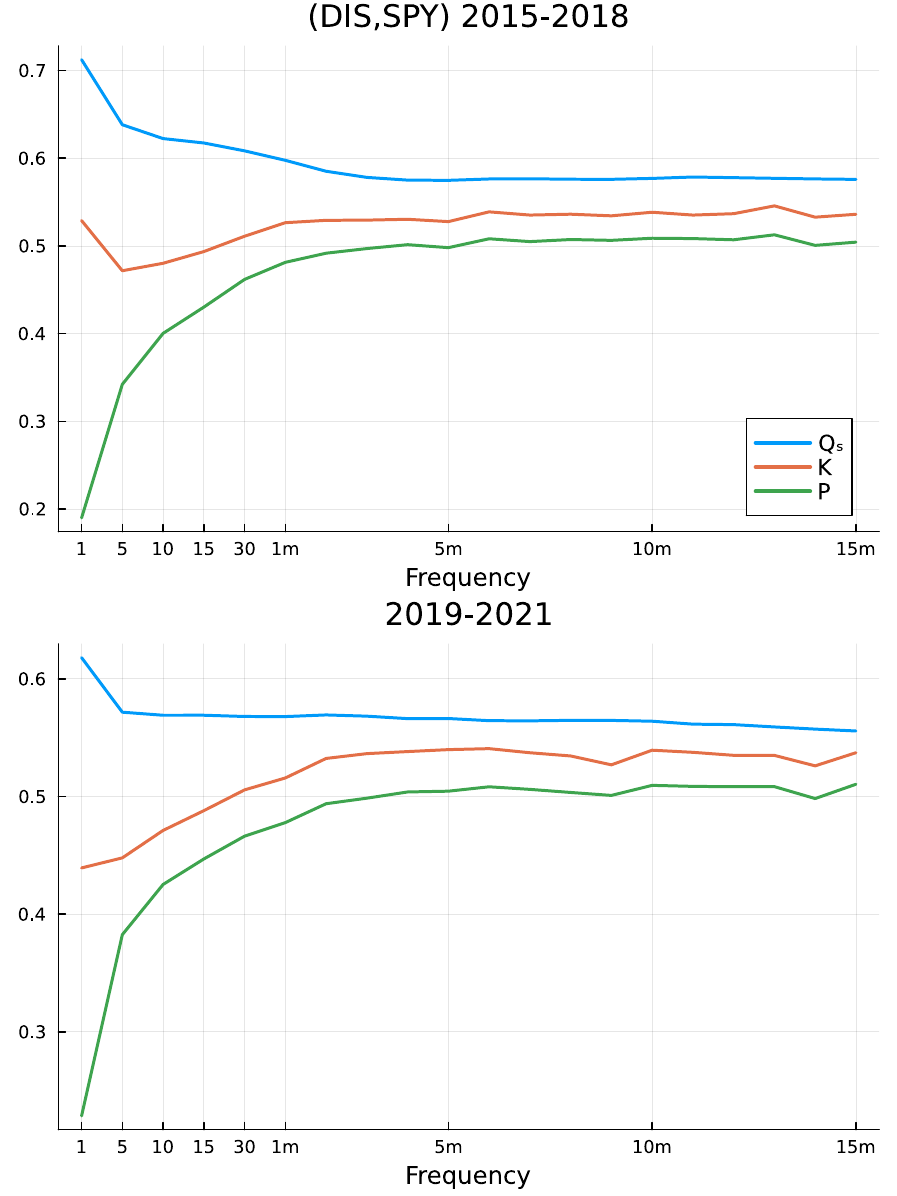}\includegraphics[width=0.2\textwidth]{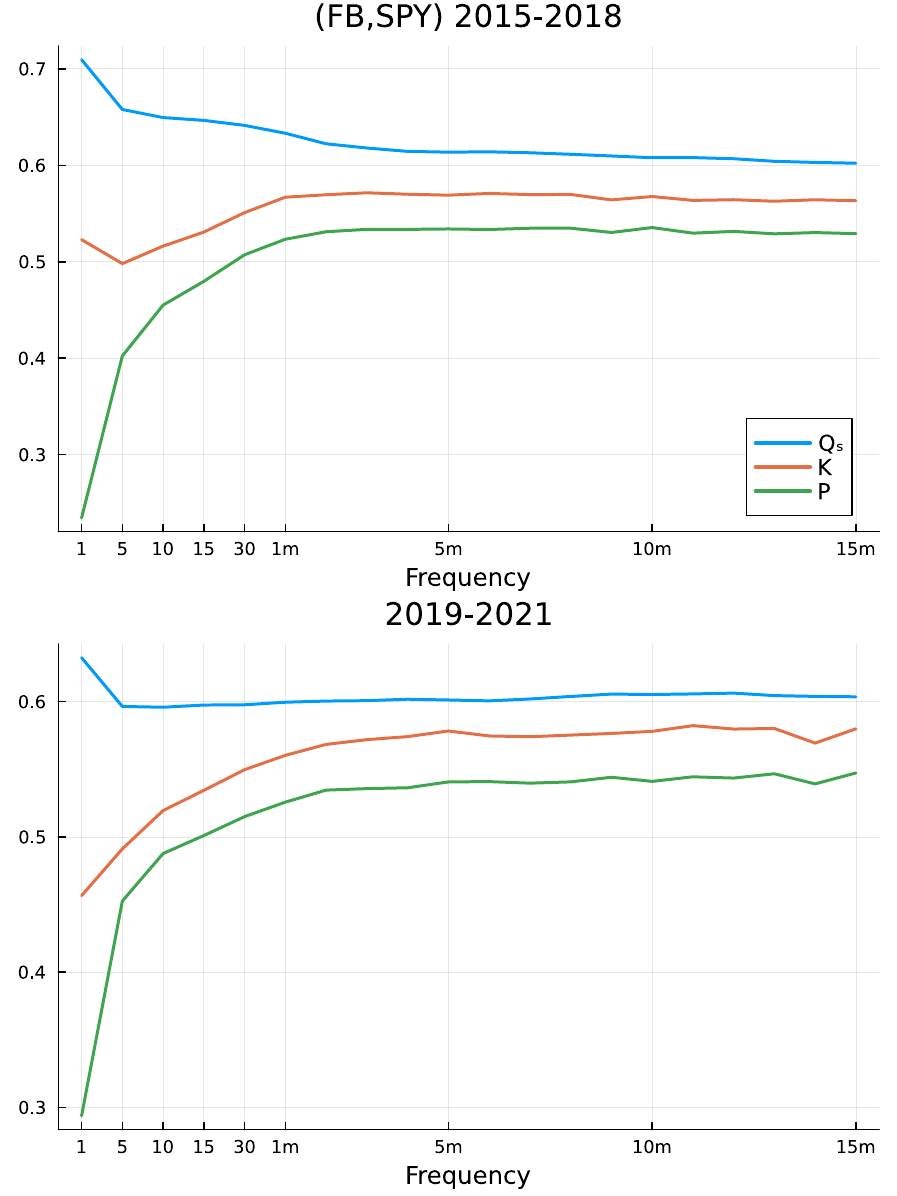}
\par\end{centering}
\caption{Correlation signature plots for 22 assets and SPY.\label{fig:SplitSignaturePlotsMarketCorr}}
\end{figure}
\begin{figure}[H]
\begin{centering}
\includegraphics[width=0.25\textwidth]{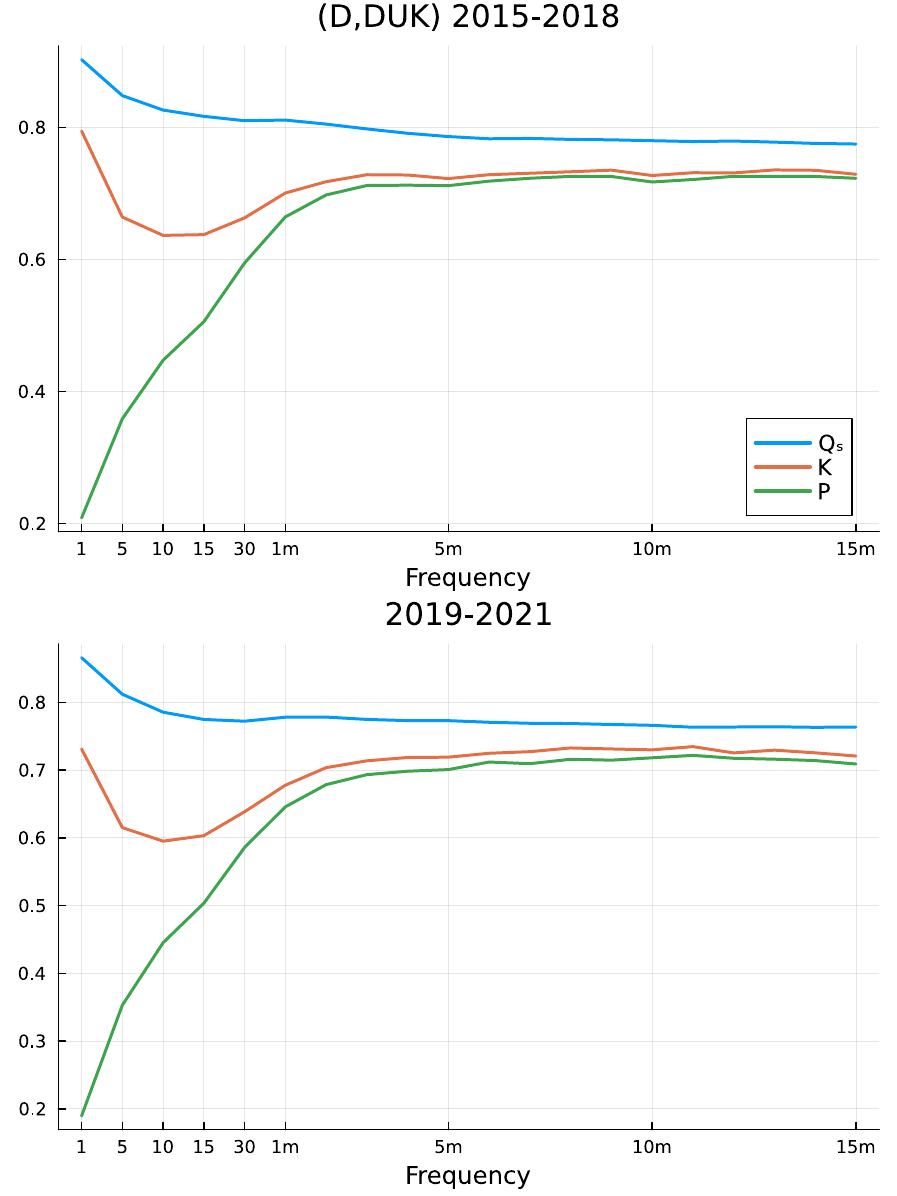}\includegraphics[width=0.25\textwidth]{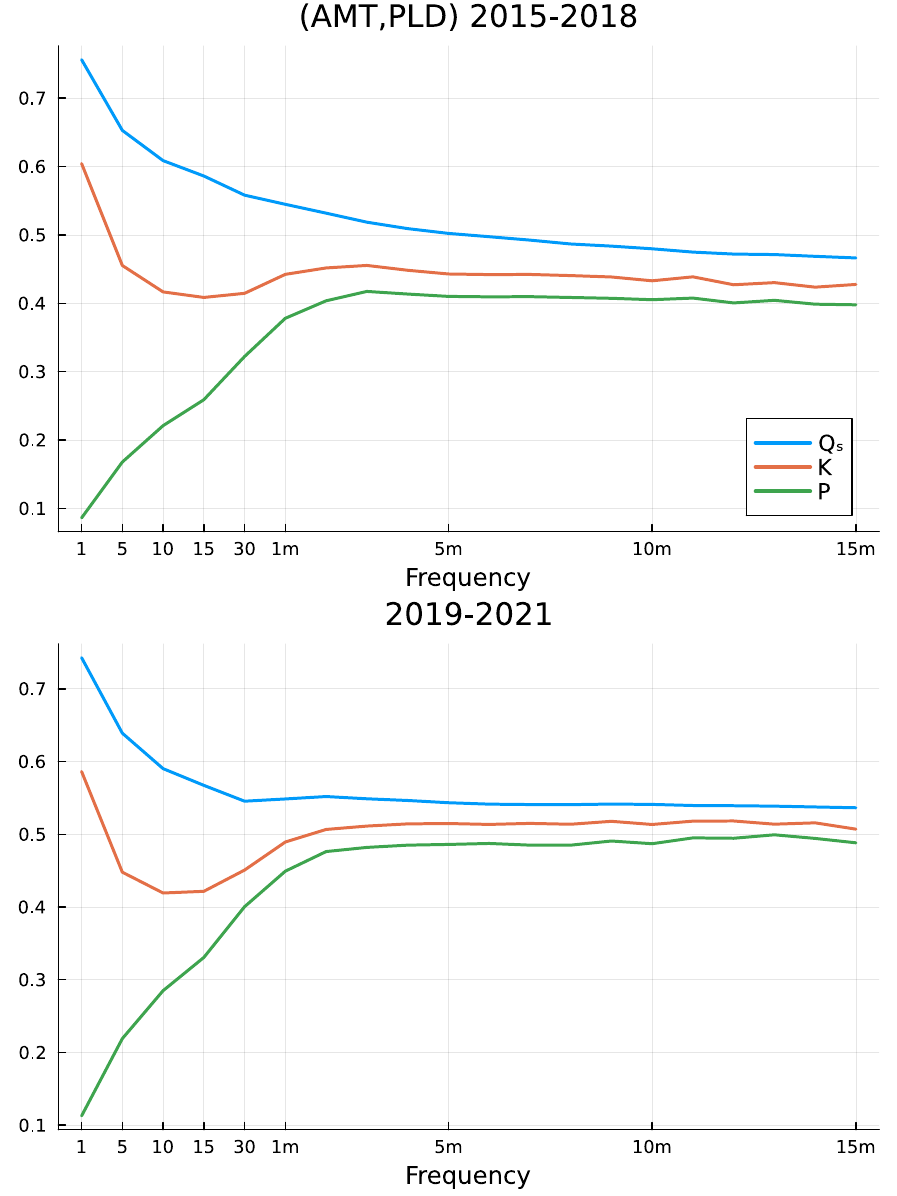}\includegraphics[width=0.25\textwidth]{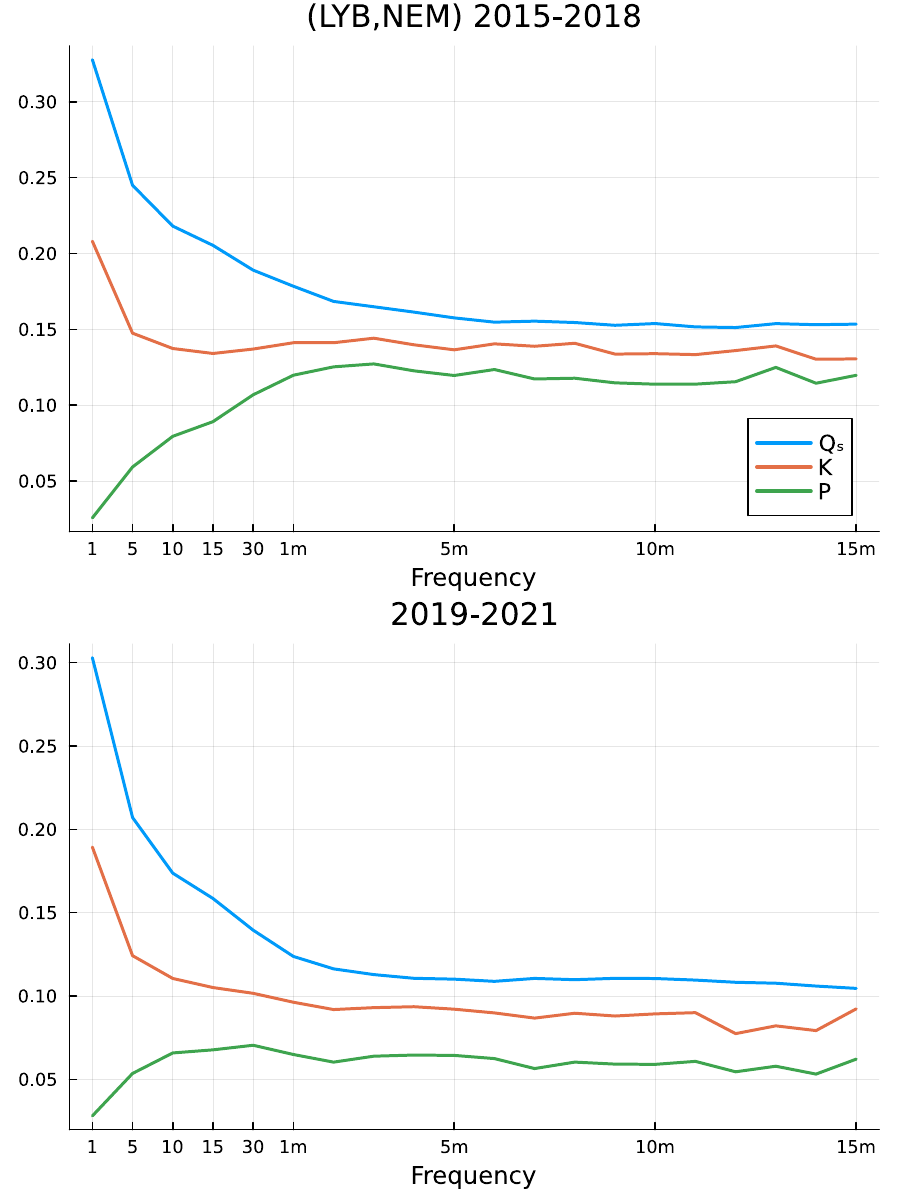}\includegraphics[width=0.25\textwidth]{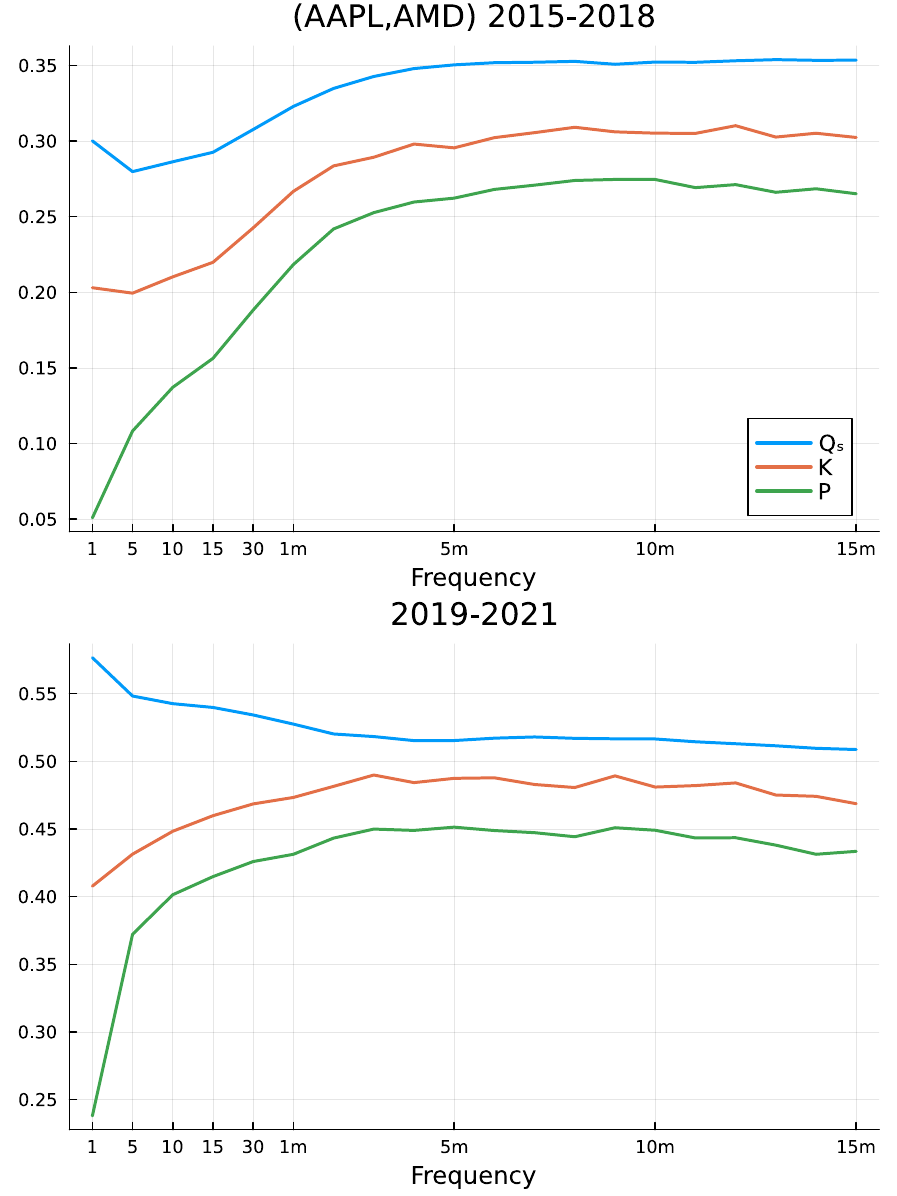}
\par\end{centering}
\begin{centering}
\includegraphics[width=0.25\textwidth]{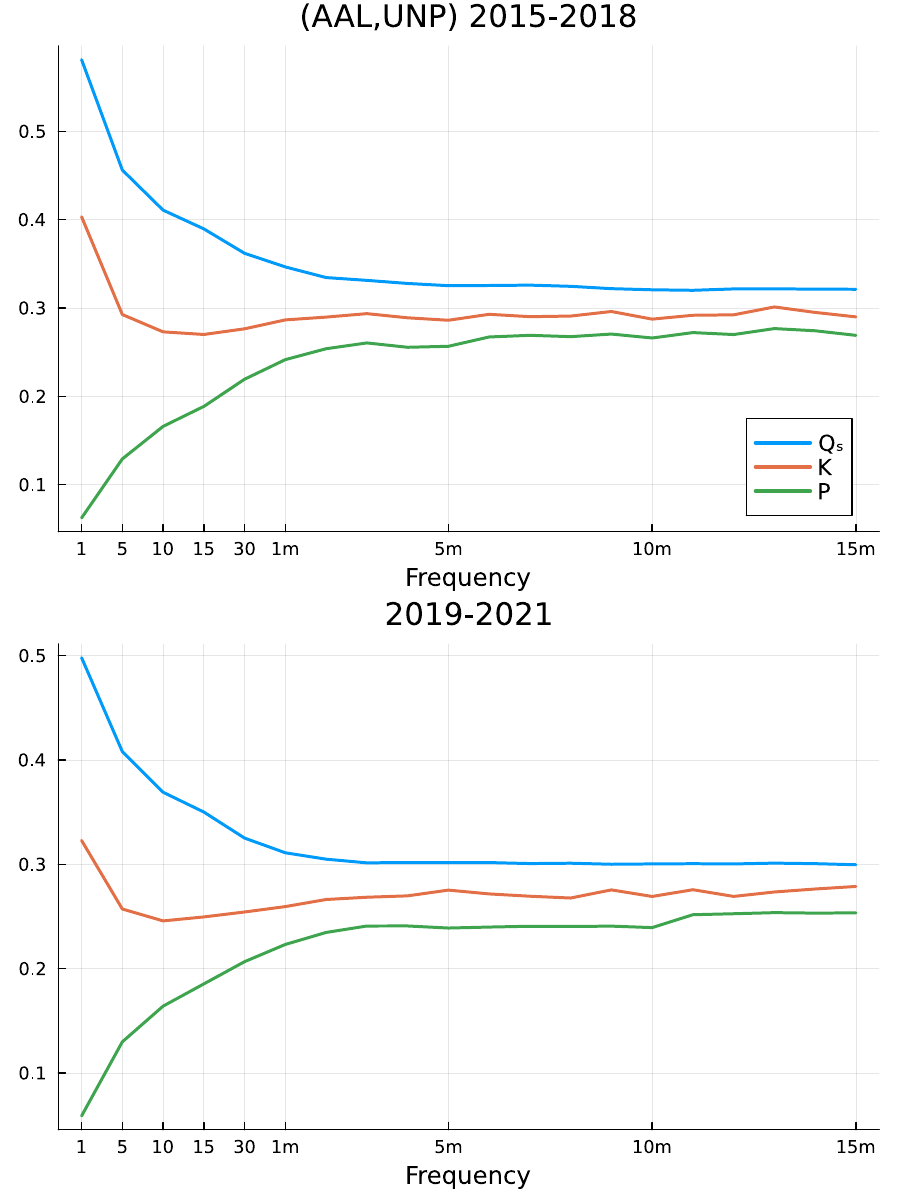}\includegraphics[width=0.25\textwidth]{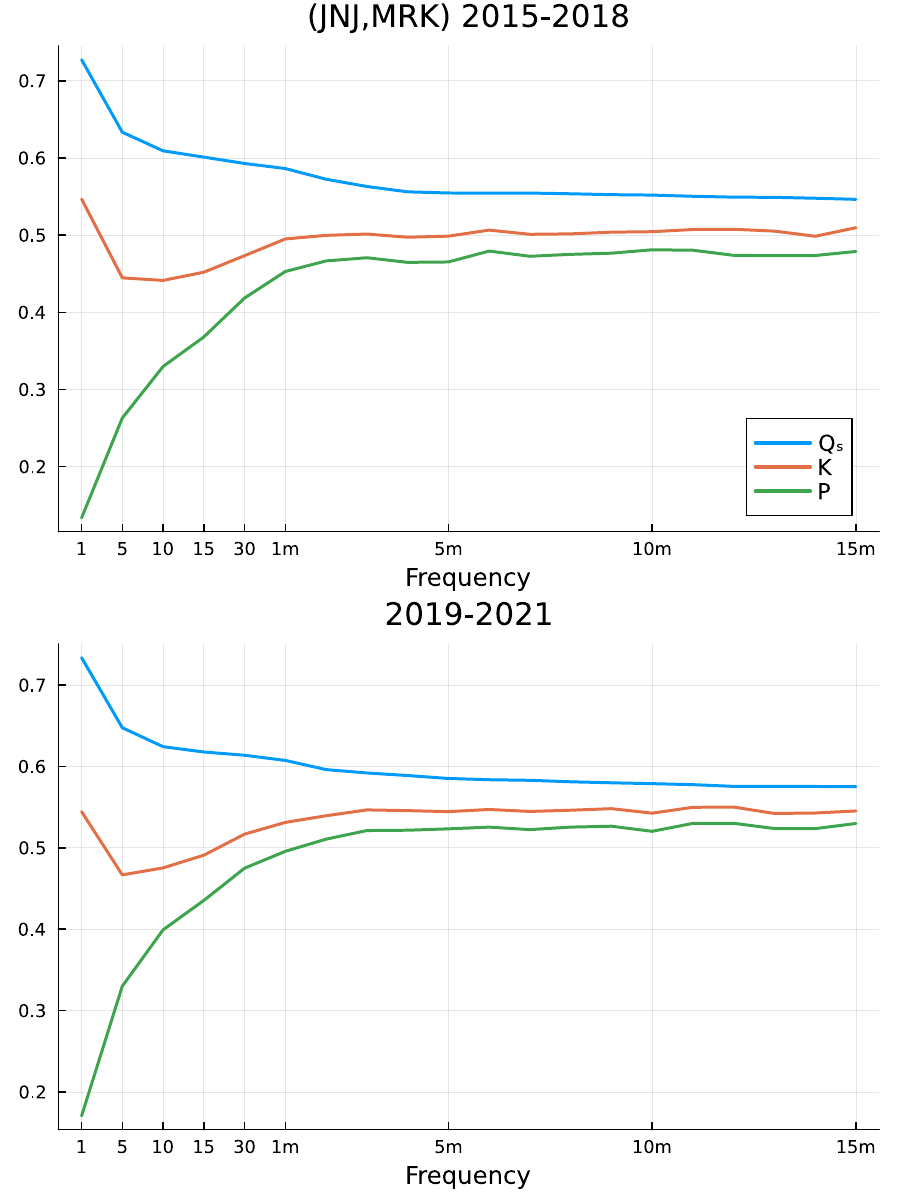}\includegraphics[width=0.25\textwidth]{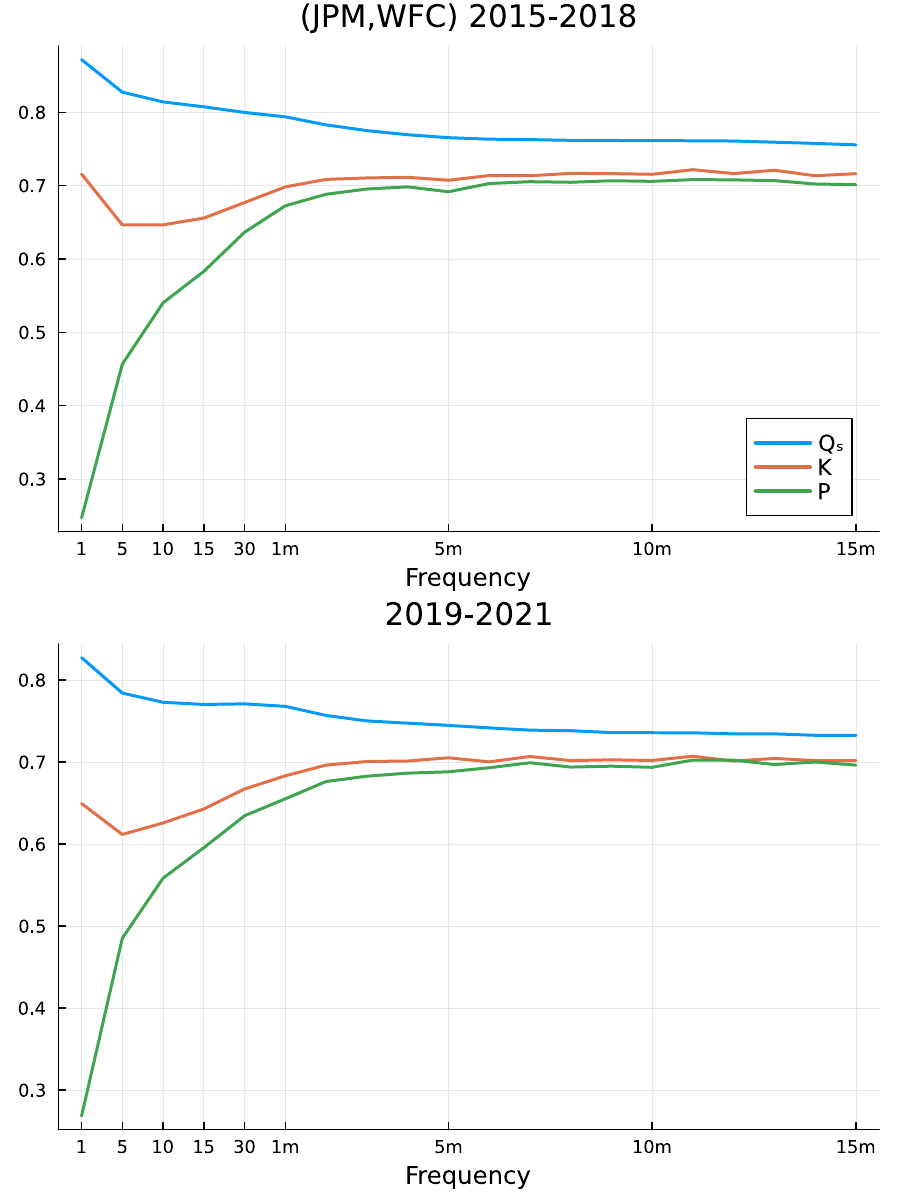}\includegraphics[width=0.25\textwidth]{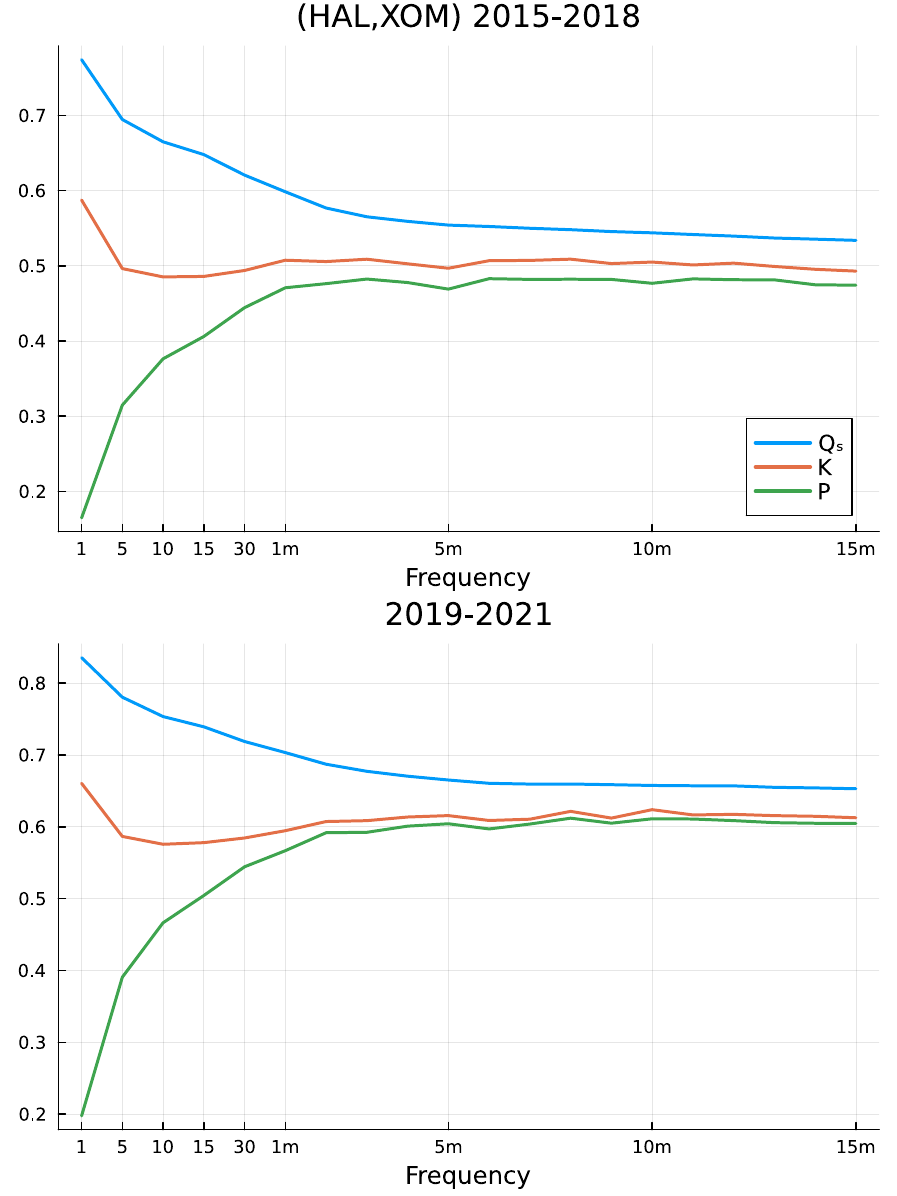}
\par\end{centering}
\begin{centering}
\includegraphics[width=0.25\textwidth]{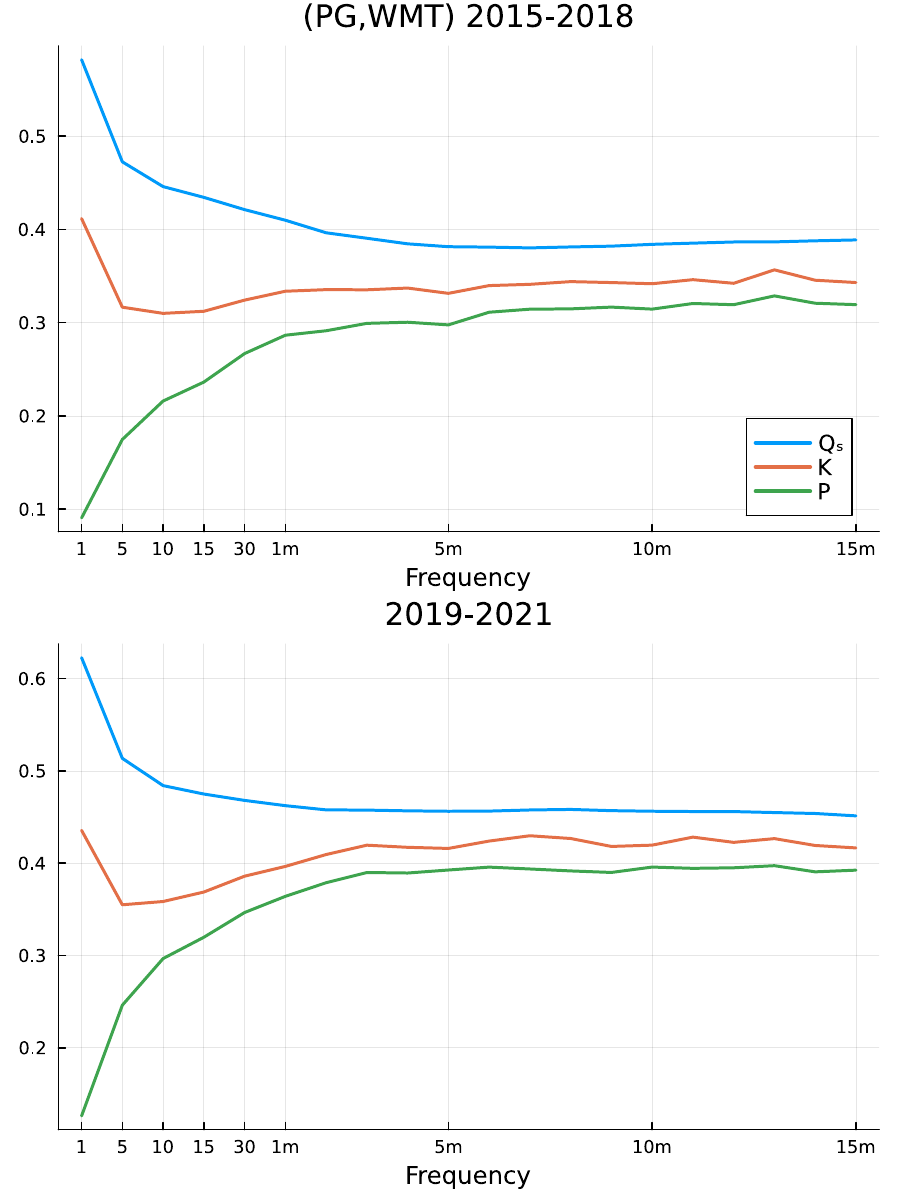}\includegraphics[width=0.25\textwidth]{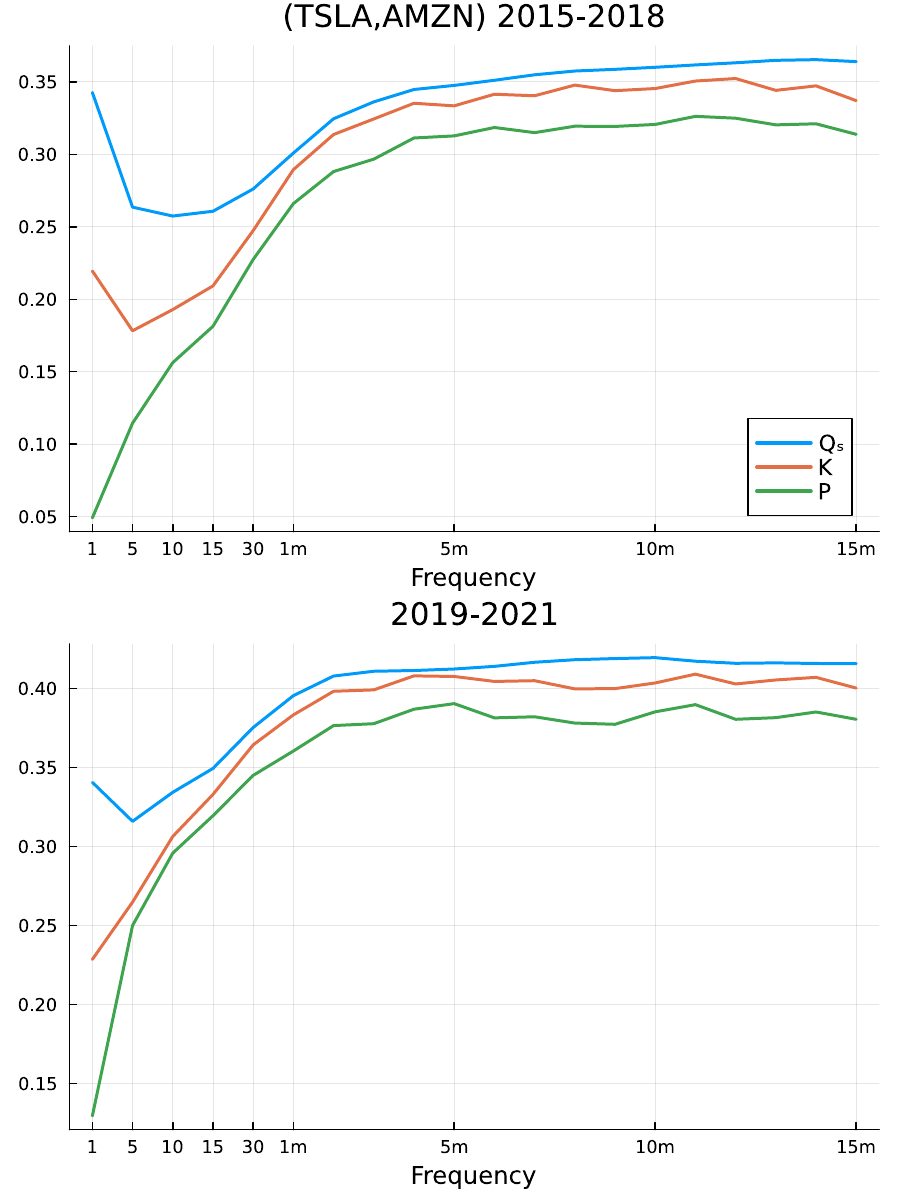}\includegraphics[width=0.25\textwidth]{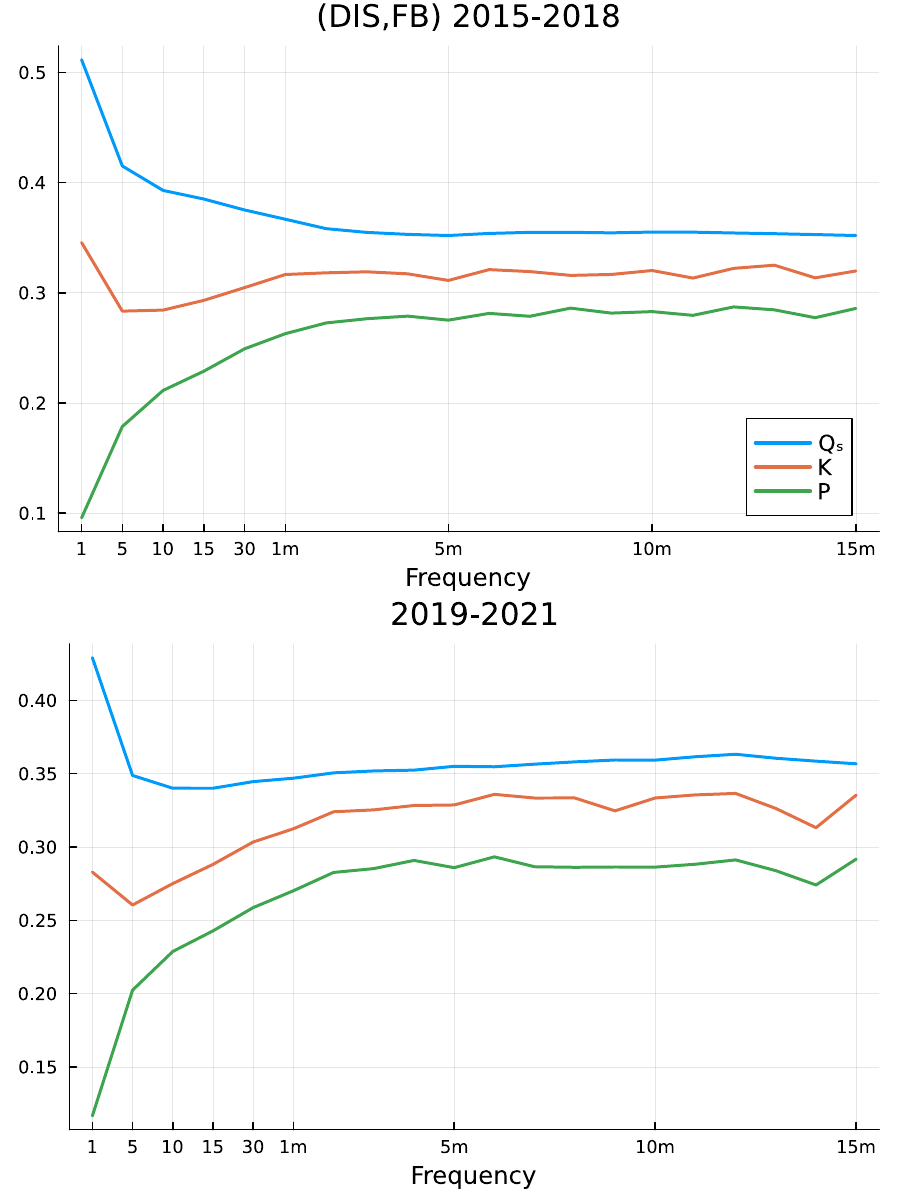}
\par\end{centering}
\caption{Correlation signature plots for 11 pairs of stocks.\label{fig:SplitSignaturePlotsSector}}
\end{figure}
 The average level of correlation can obviously be different in the
two sample periods, but the shapes of the signature plots are nevertheless
very similar for the two sub-periods for all pairs of assets,
see Figures \ref{fig:SplitSignaturePlotsMarketCorr} and \ref{fig:SplitSignaturePlotsSector}. 

We also split the sample in two to check whether the intraday patterns
in correlations, relative volatilities, and market betas depend on
the sample period, for each asset in the Small Universe.
\begin{figure}[H]
\begin{centering}
\includegraphics[width=0.33\textwidth]{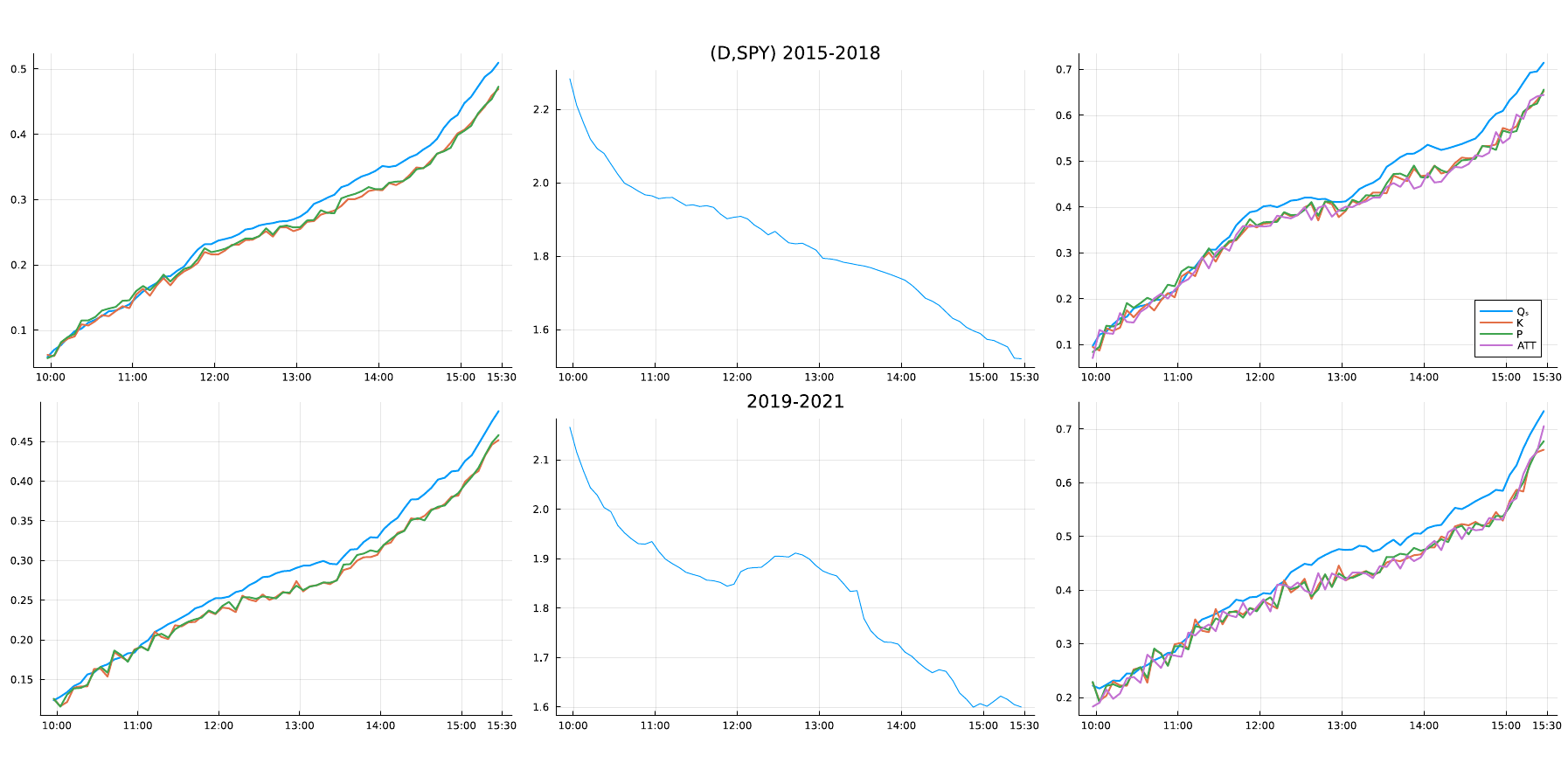}\includegraphics[width=0.33\textwidth]{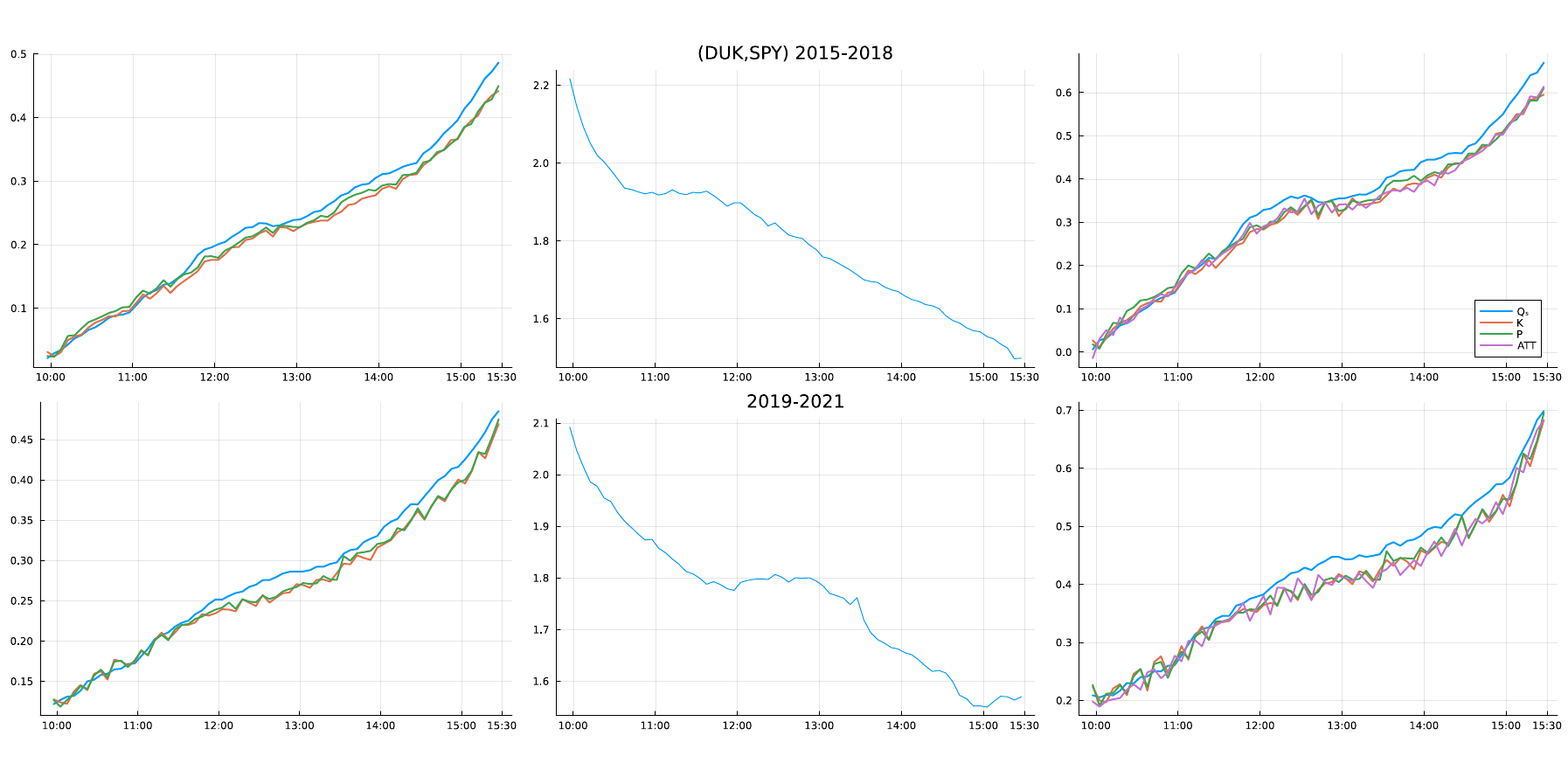}\includegraphics[width=0.33\textwidth]{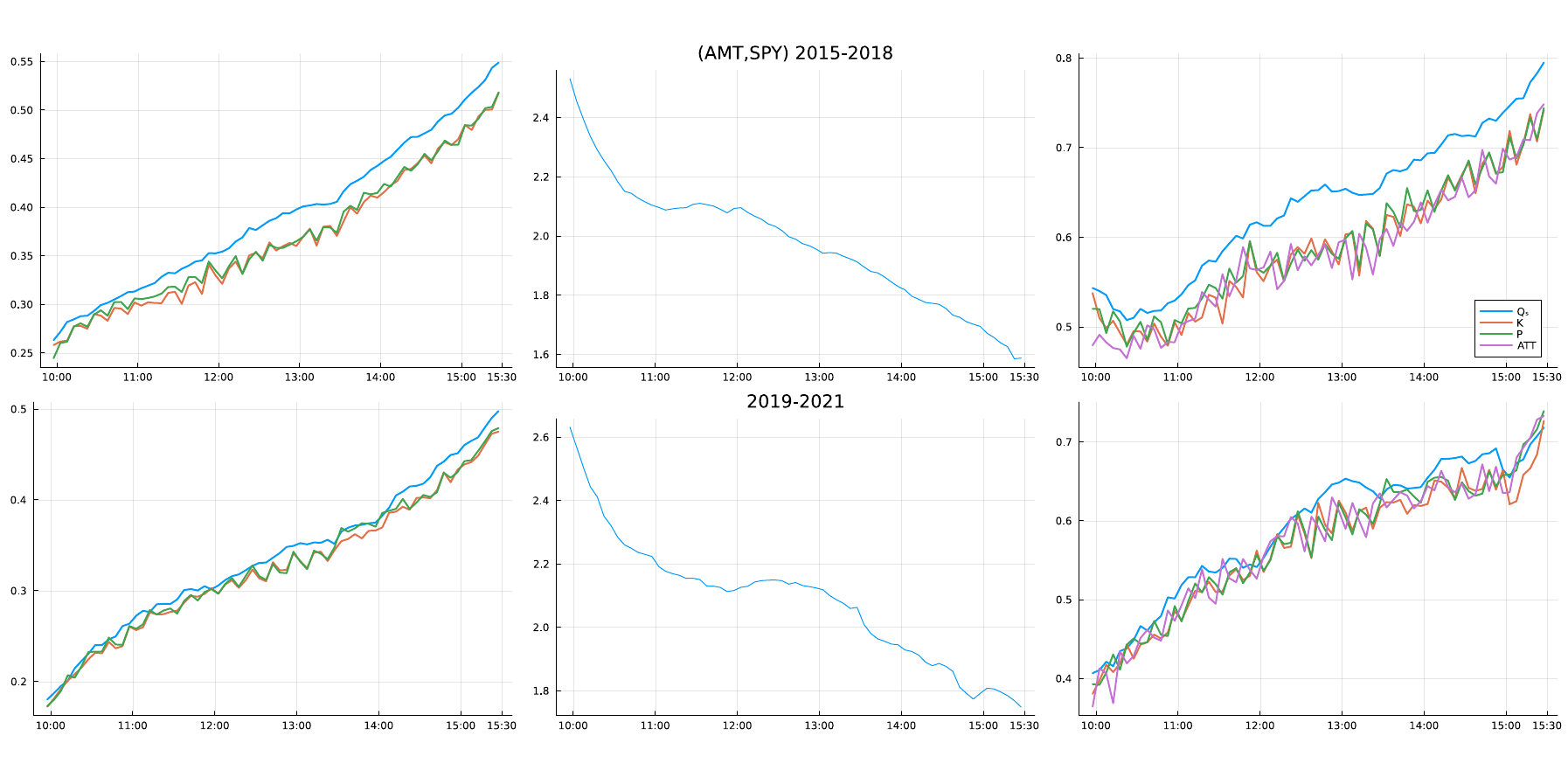}
\par\end{centering}
\begin{centering}
\includegraphics[width=0.33\textwidth]{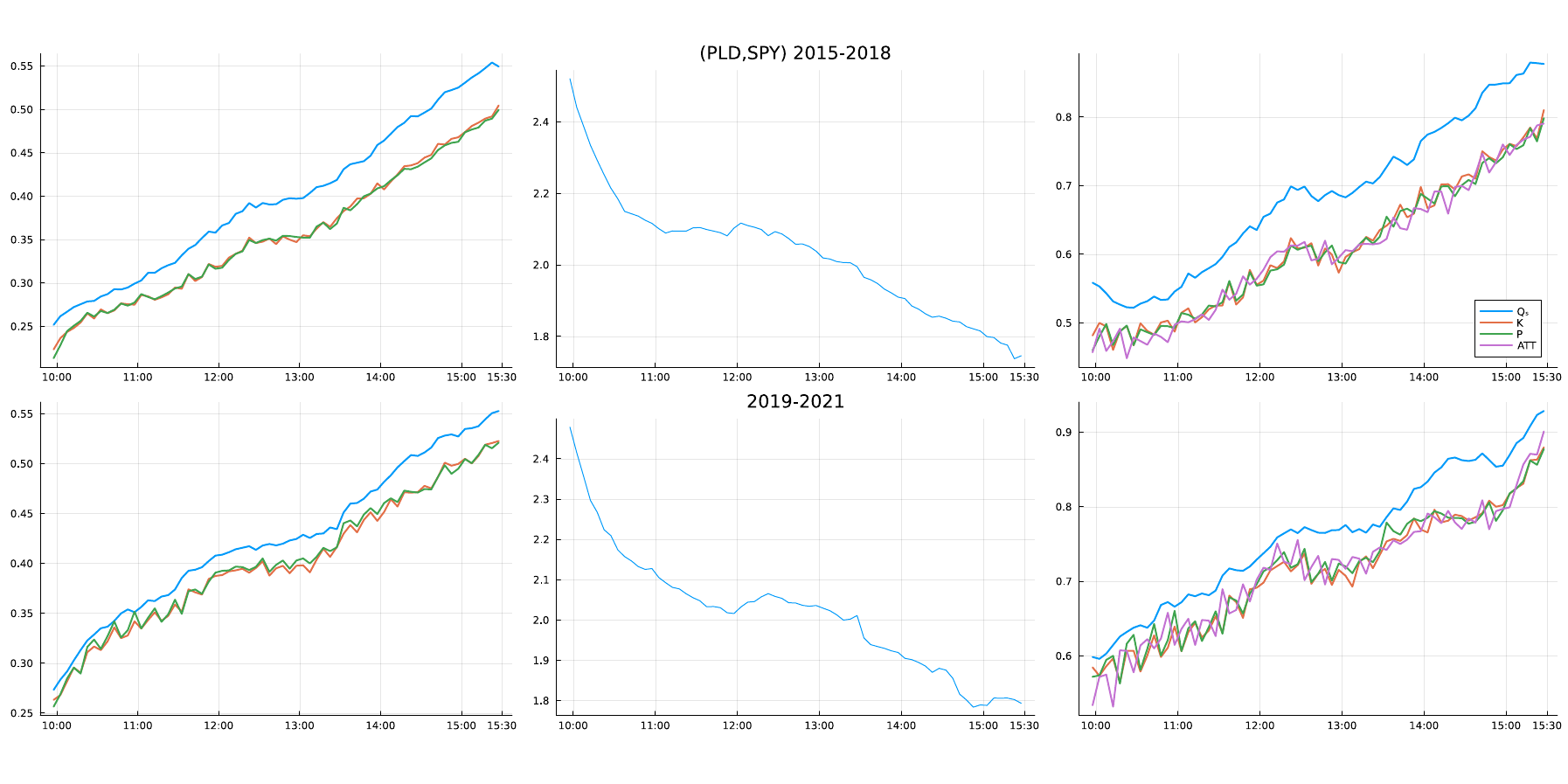}\includegraphics[width=0.33\textwidth]{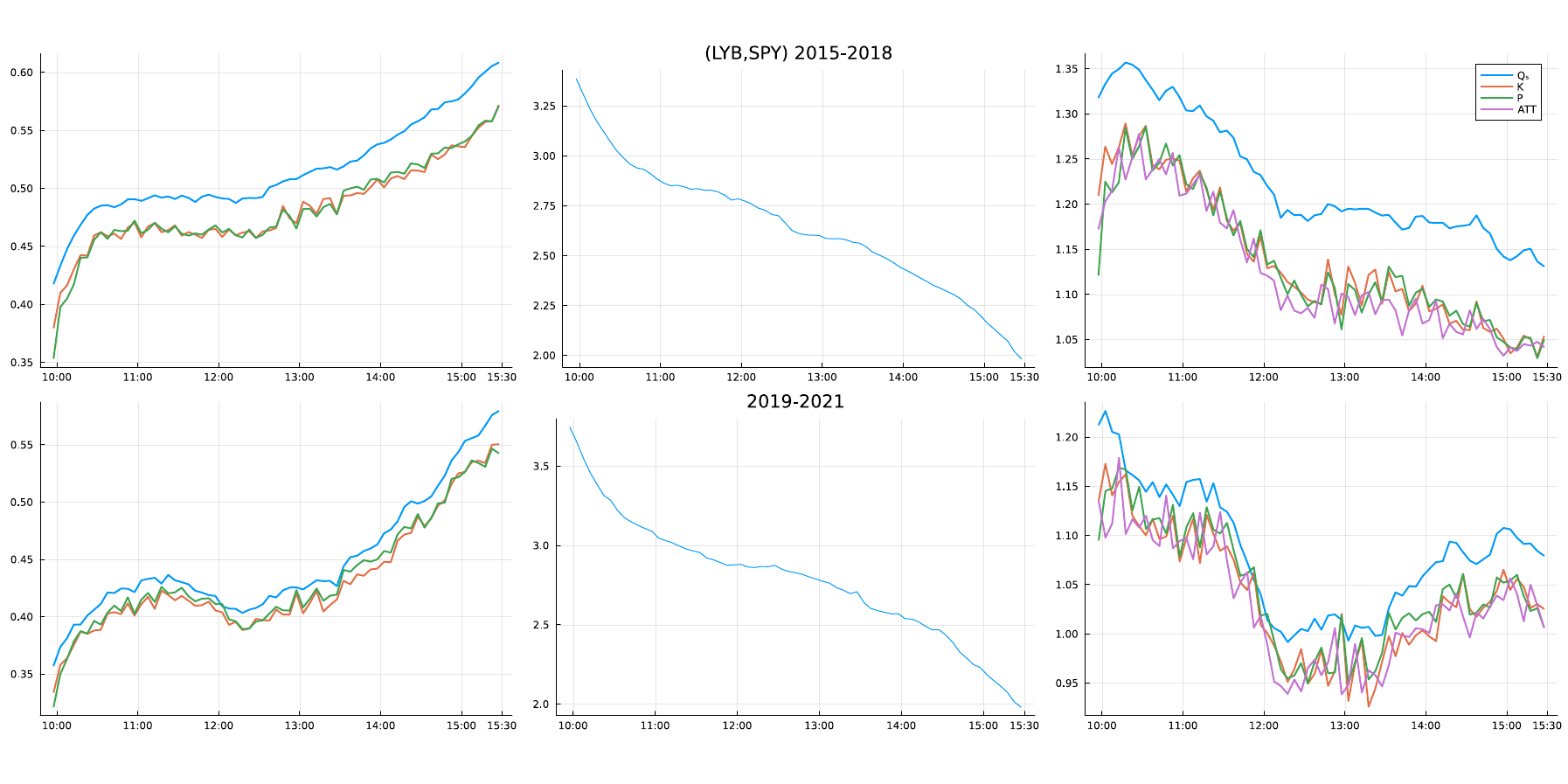}\includegraphics[width=0.33\textwidth]{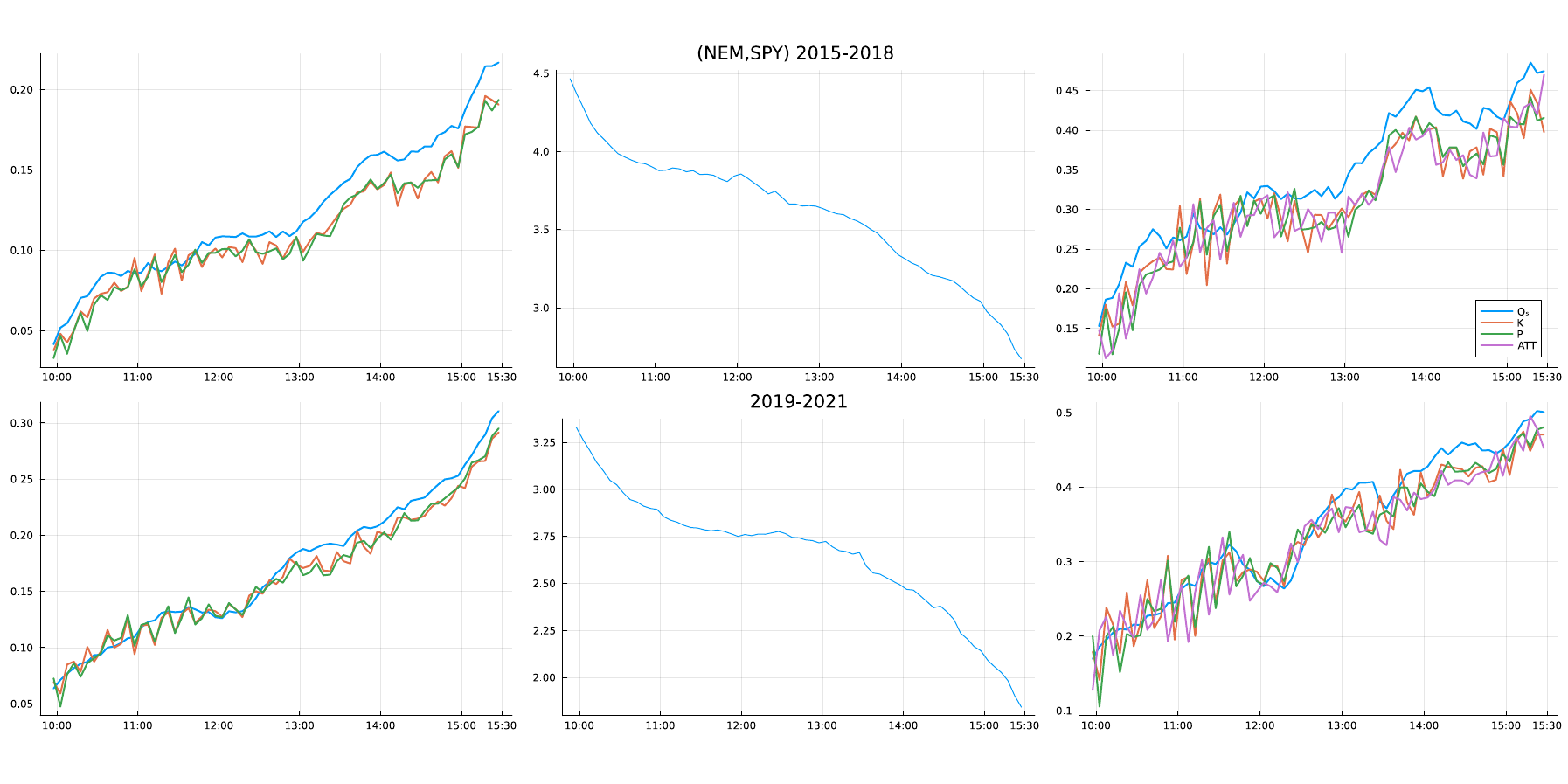}
\par\end{centering}
\begin{centering}
\includegraphics[width=0.33\textwidth]{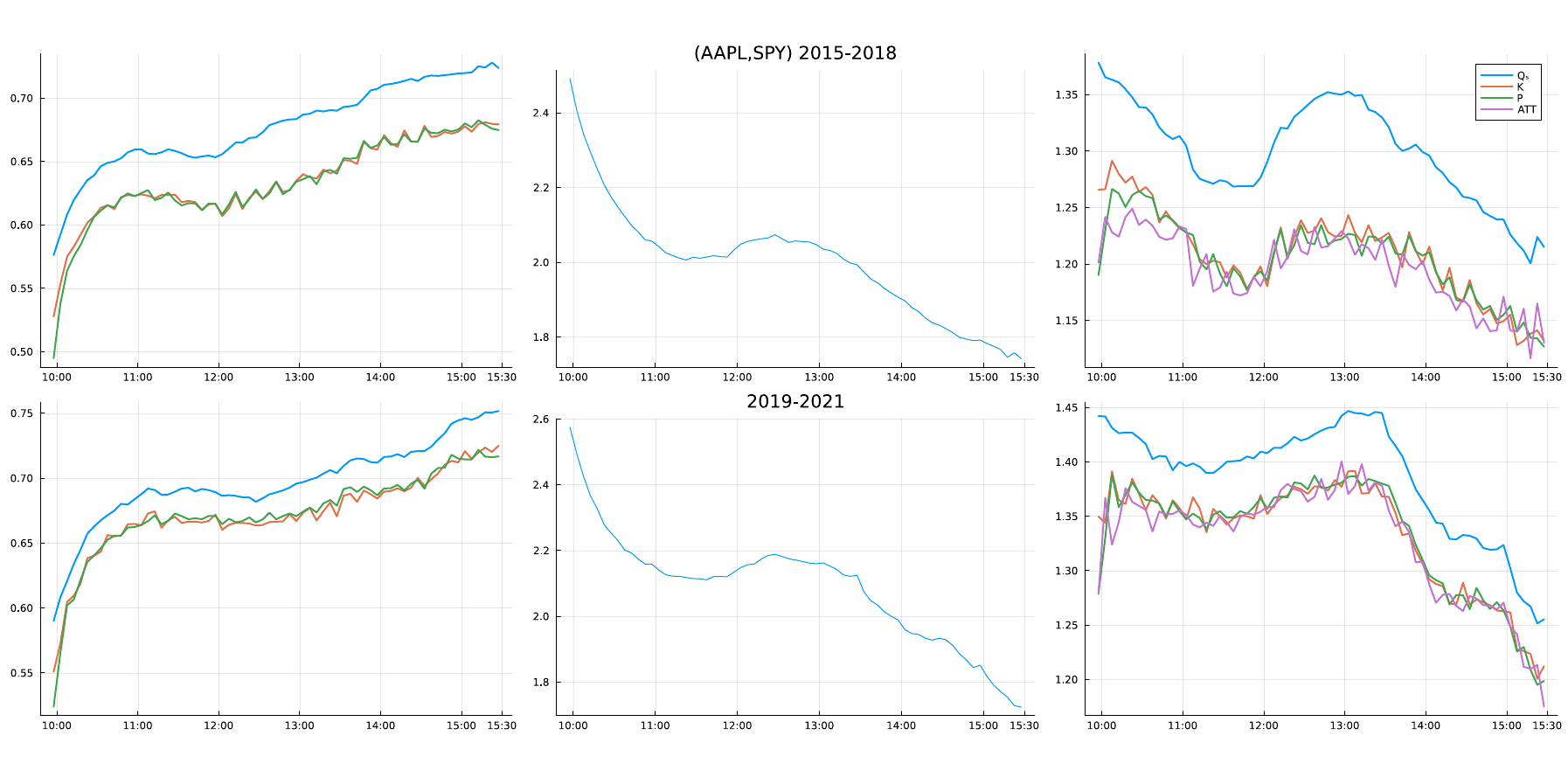}\includegraphics[width=0.33\textwidth]{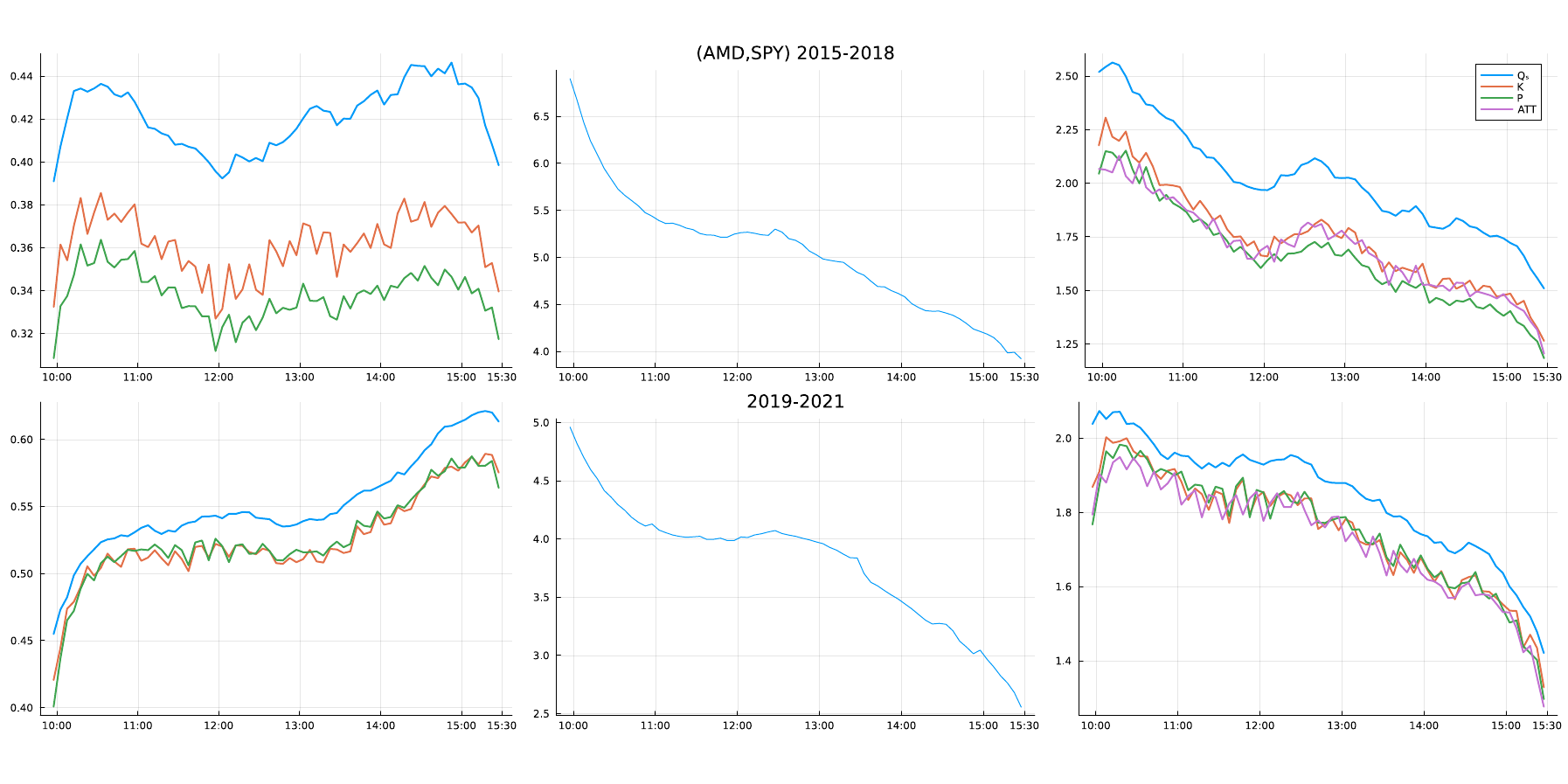}\includegraphics[width=0.33\textwidth]{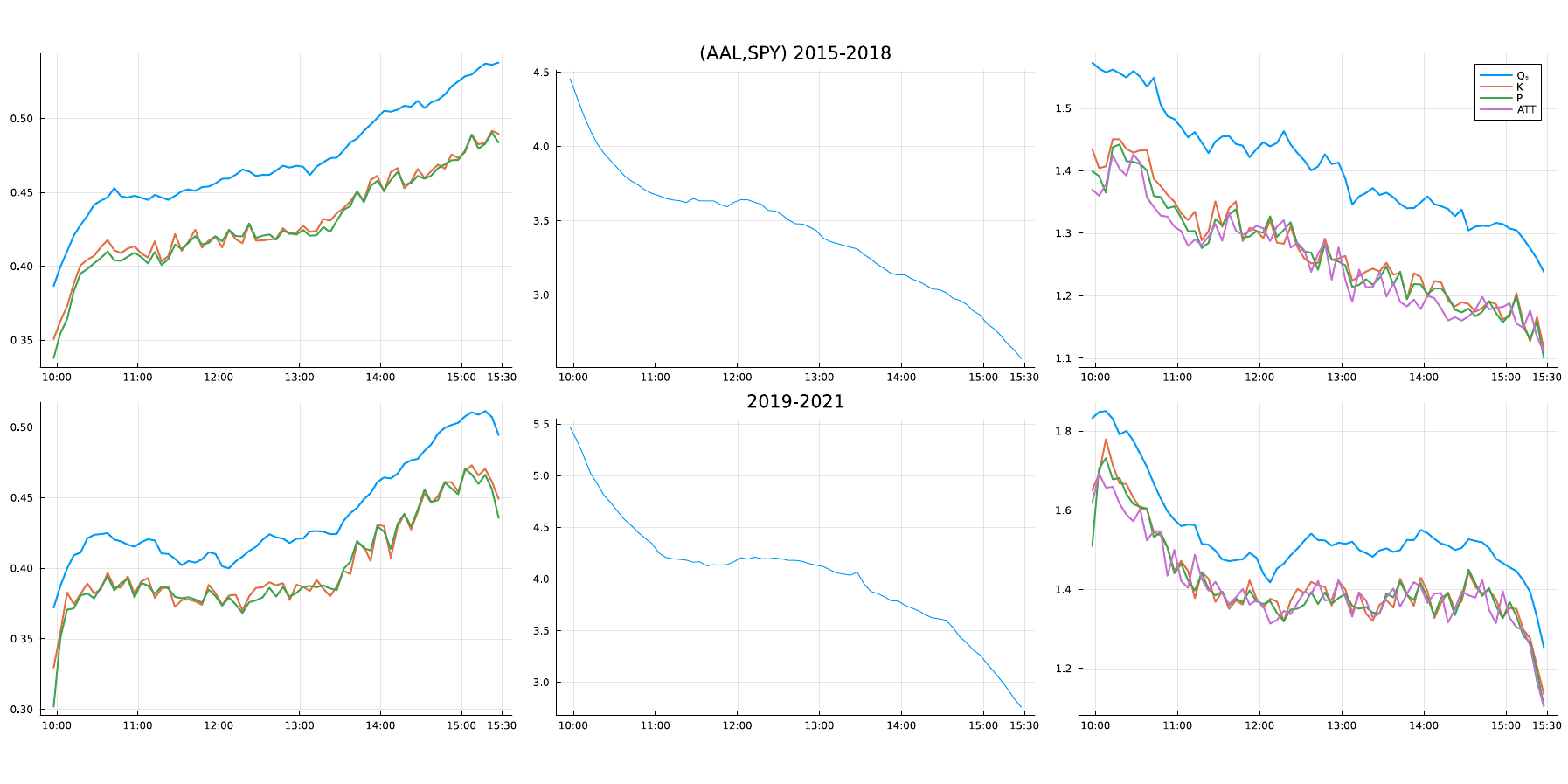}
\par\end{centering}
\begin{centering}
\includegraphics[width=0.33\textwidth]{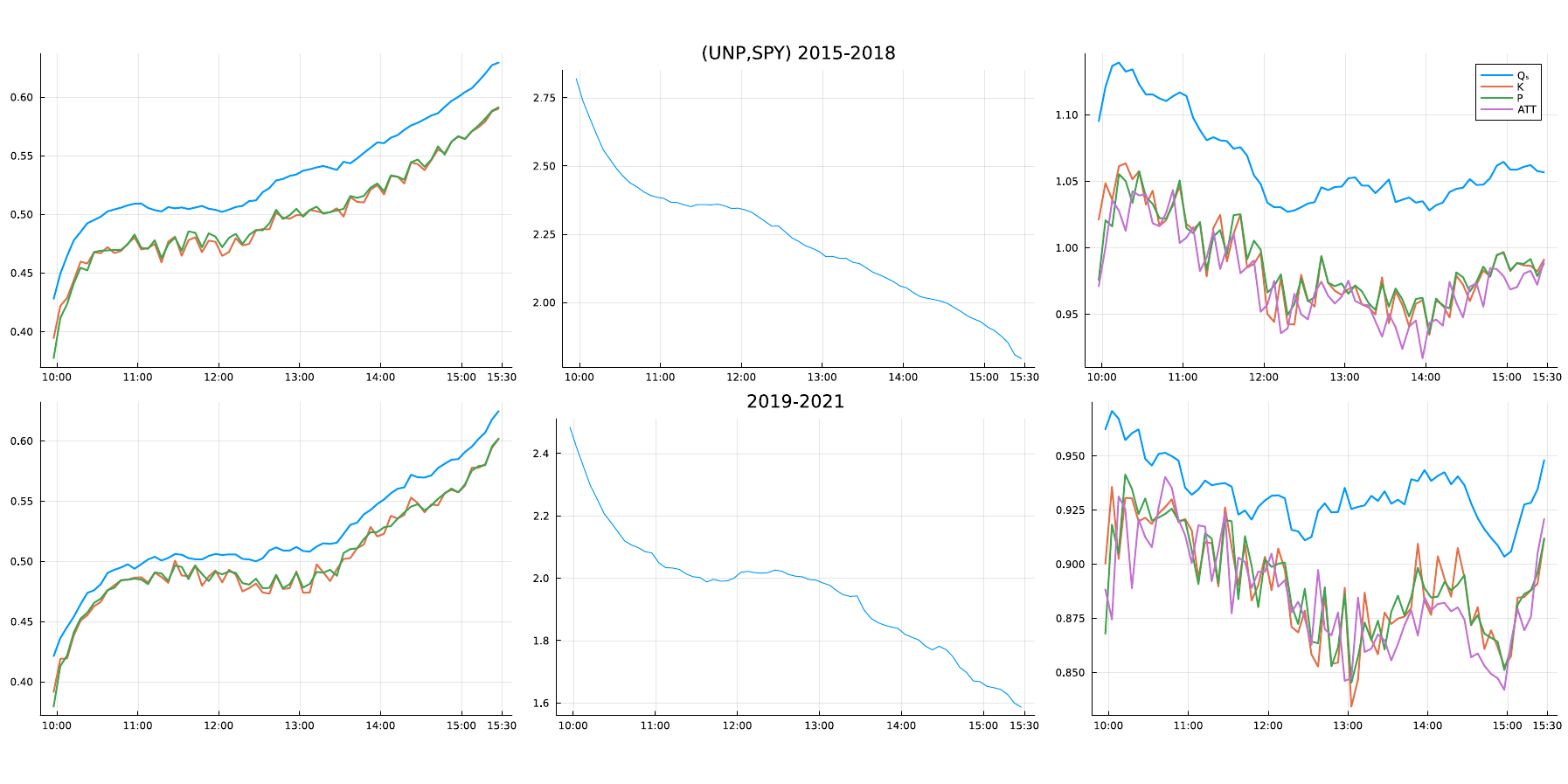}\includegraphics[width=0.33\textwidth]{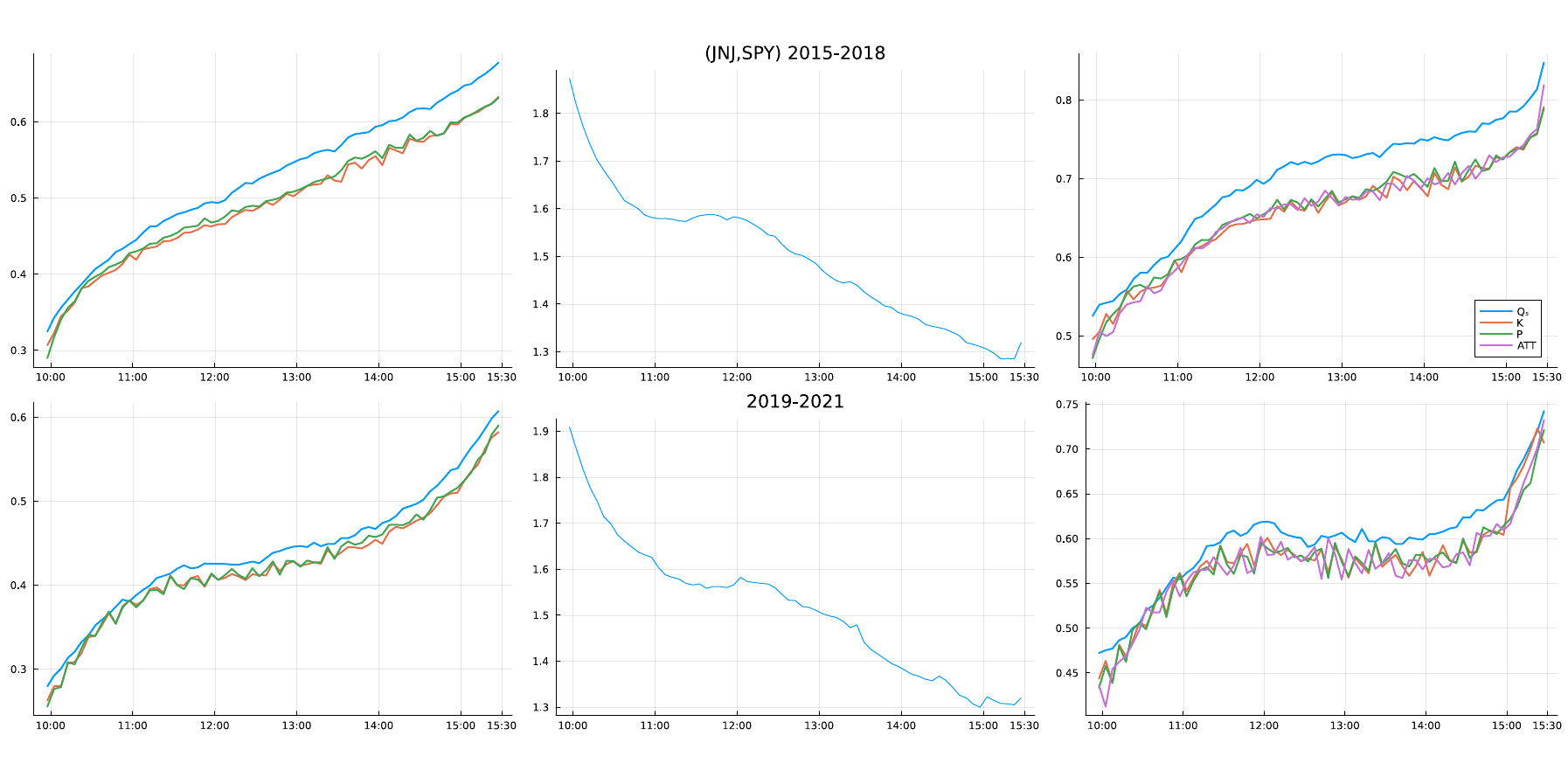}\includegraphics[width=0.33\textwidth]{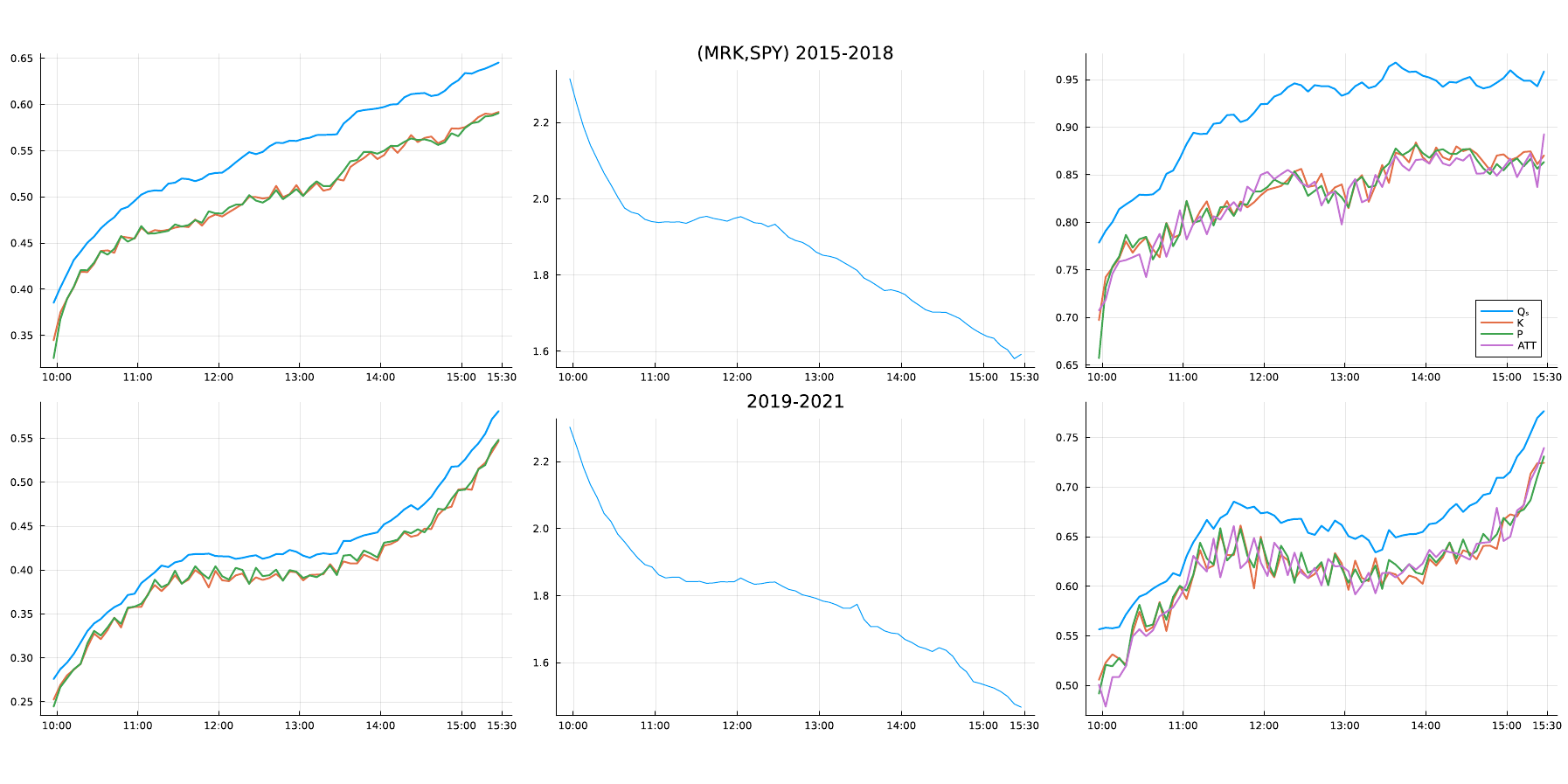}
\par\end{centering}
\begin{centering}
\includegraphics[width=0.33\textwidth]{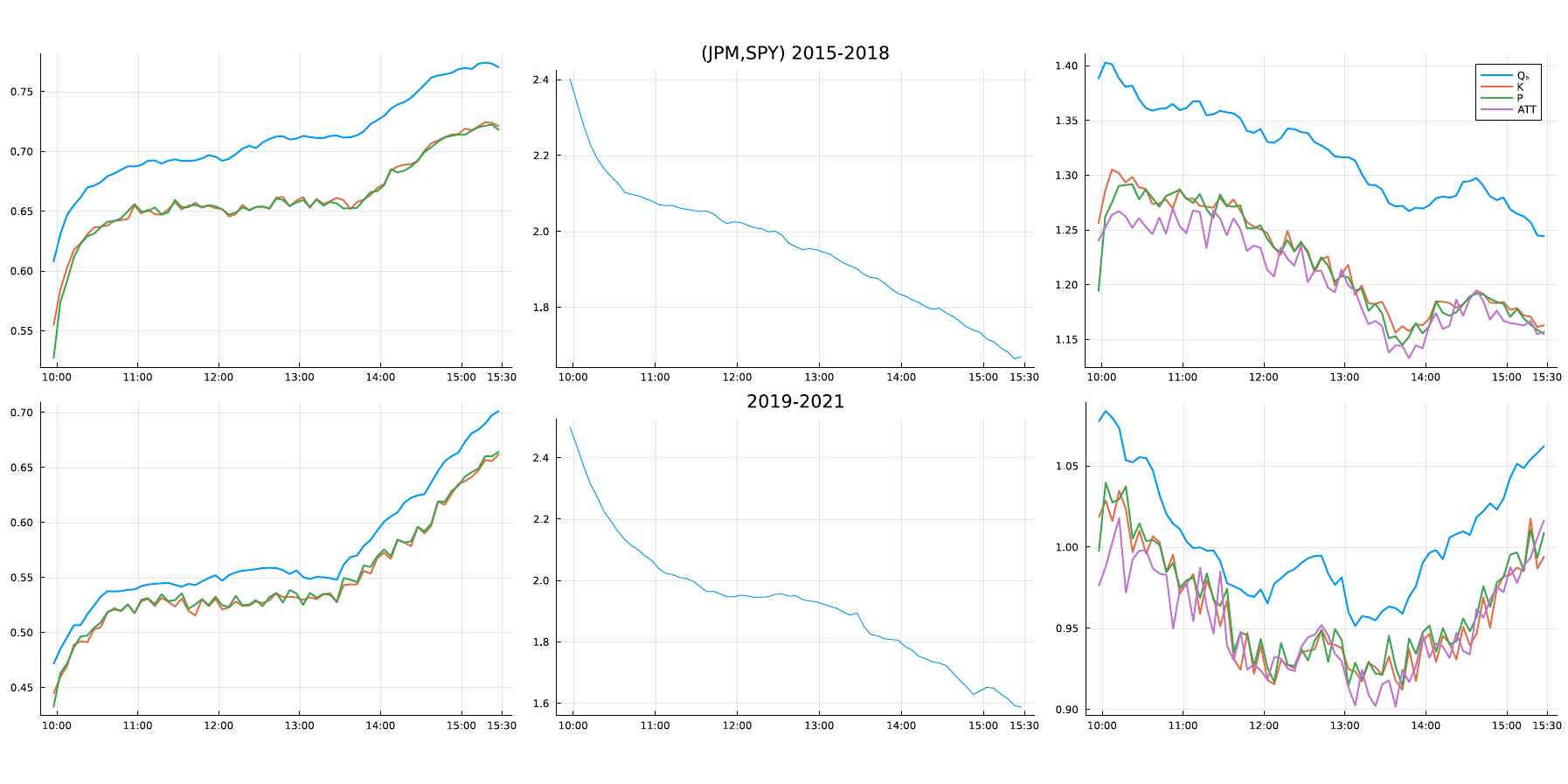}\includegraphics[width=0.33\textwidth]{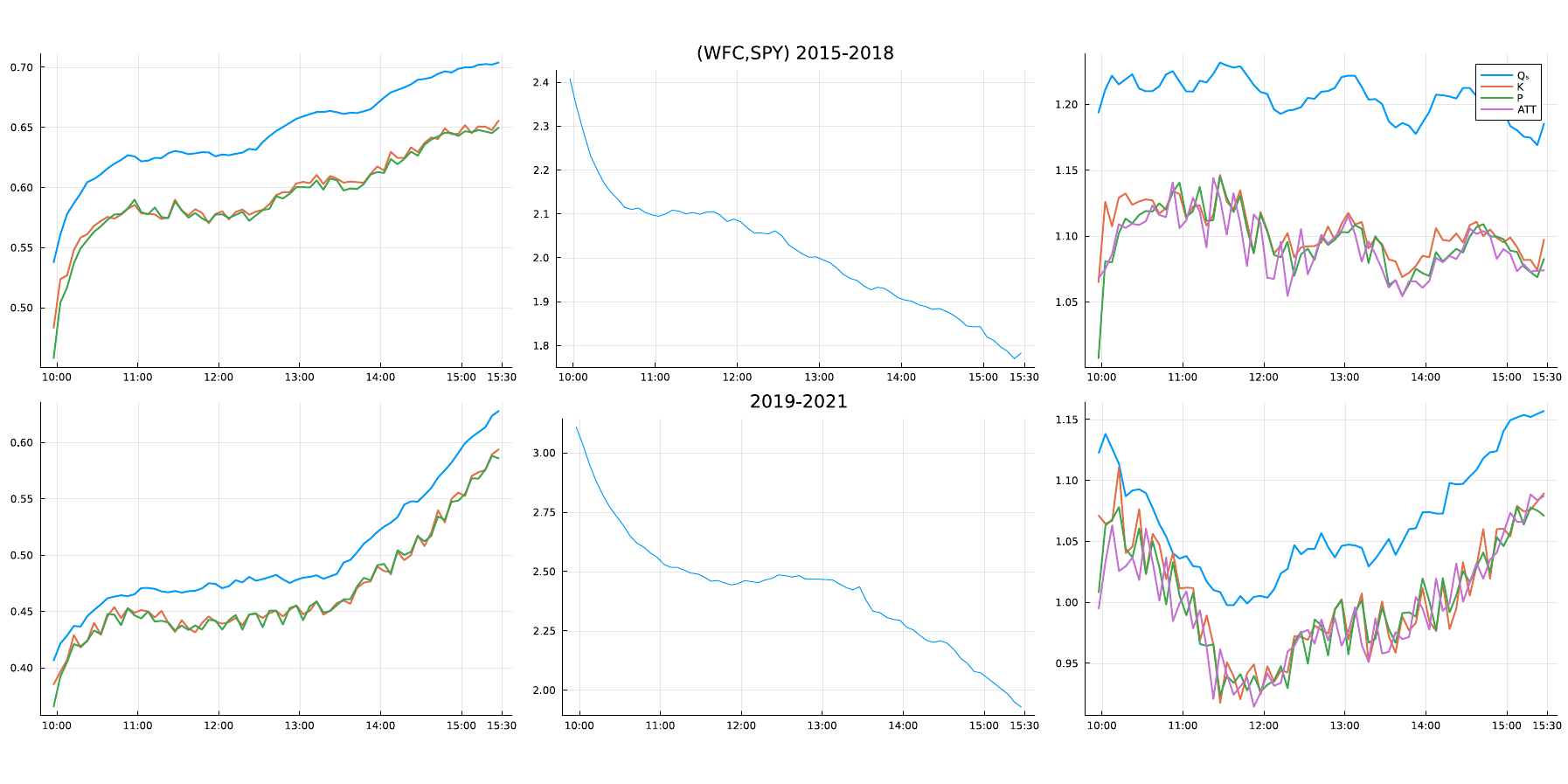}\includegraphics[width=0.33\textwidth]{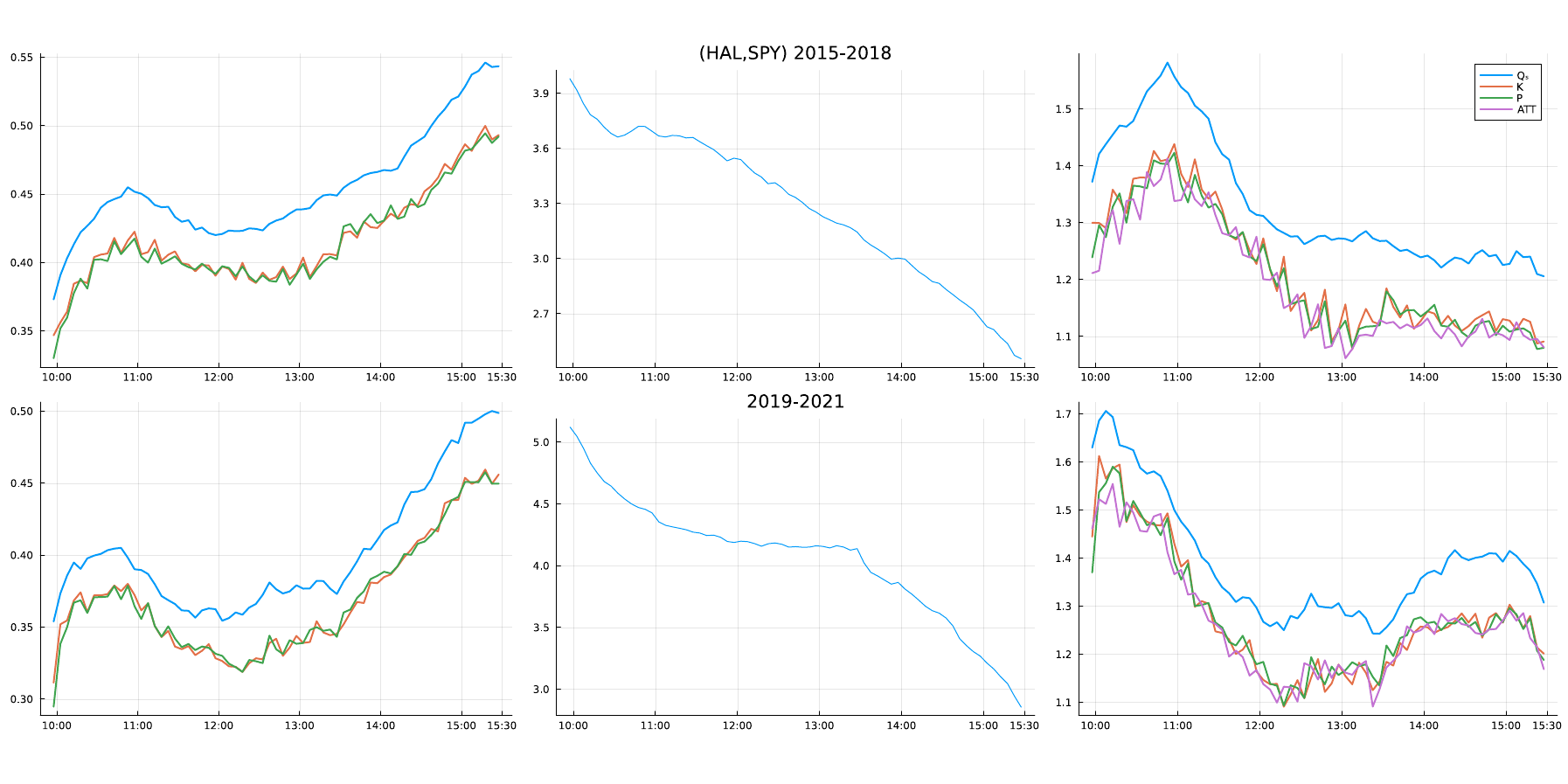}
\par\end{centering}
\begin{centering}
\includegraphics[width=0.33\textwidth]{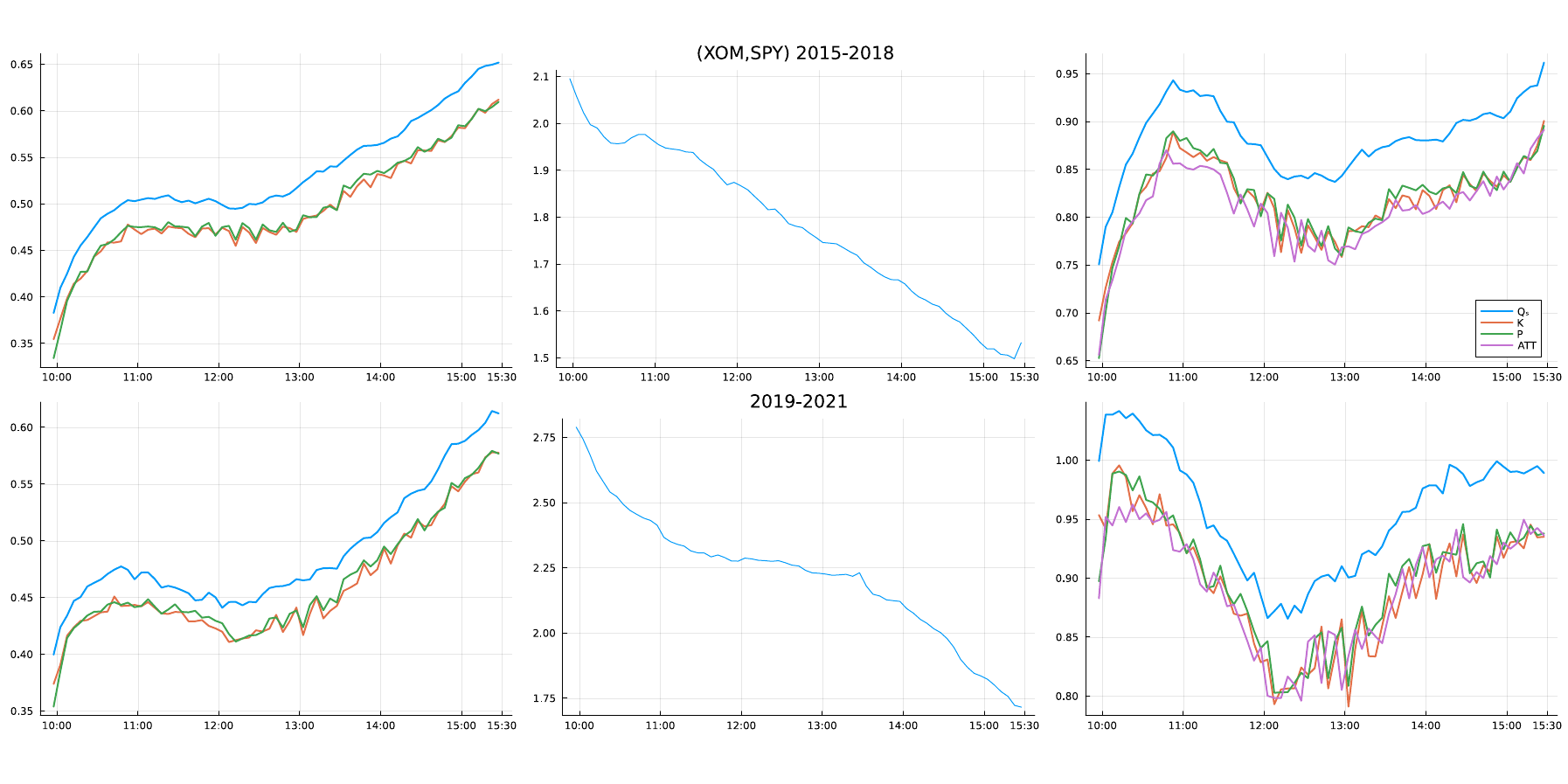}\includegraphics[width=0.33\textwidth]{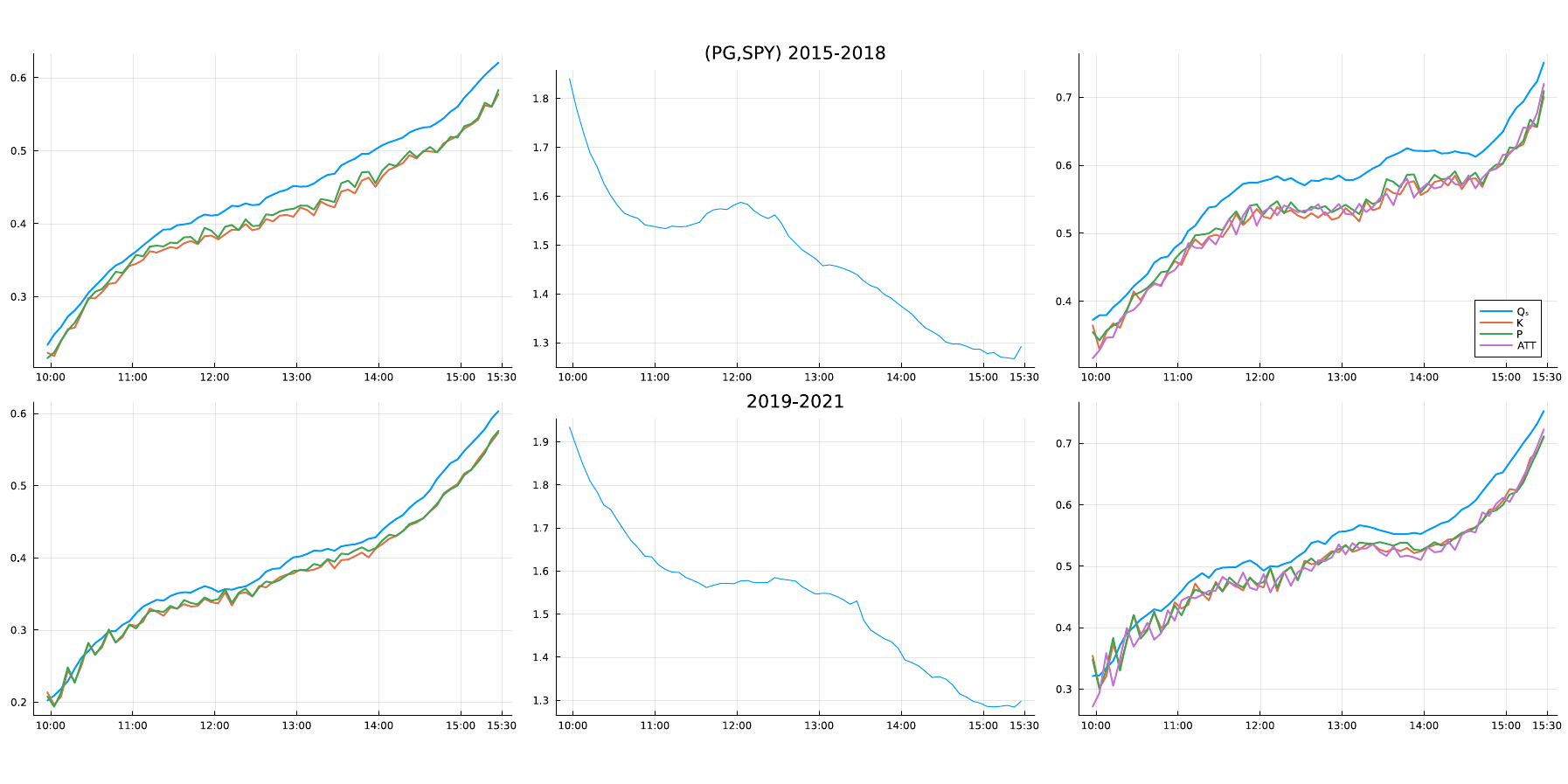}\includegraphics[width=0.33\textwidth]{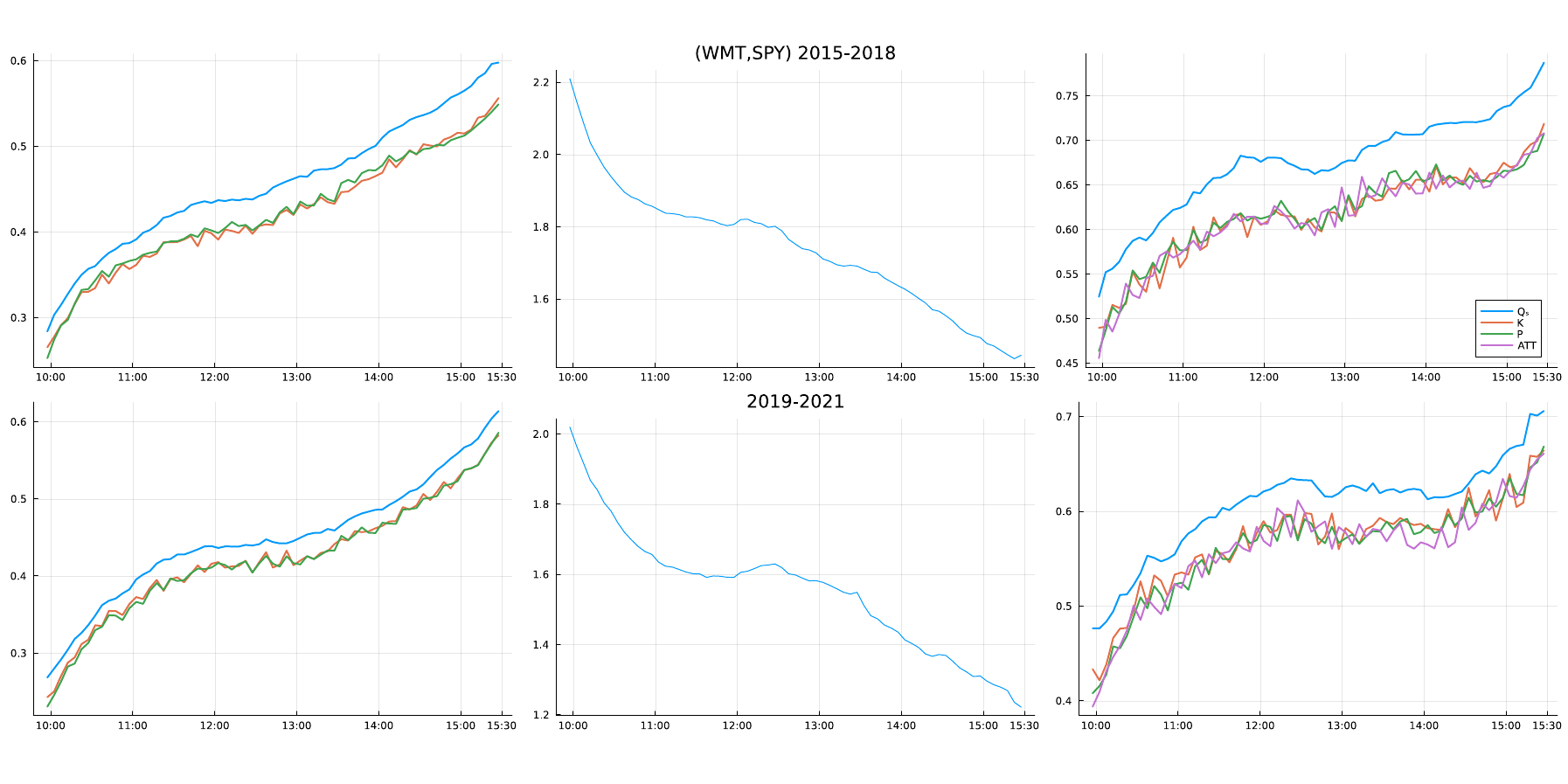}
\par\end{centering}
\begin{centering}
\includegraphics[width=0.33\textwidth]{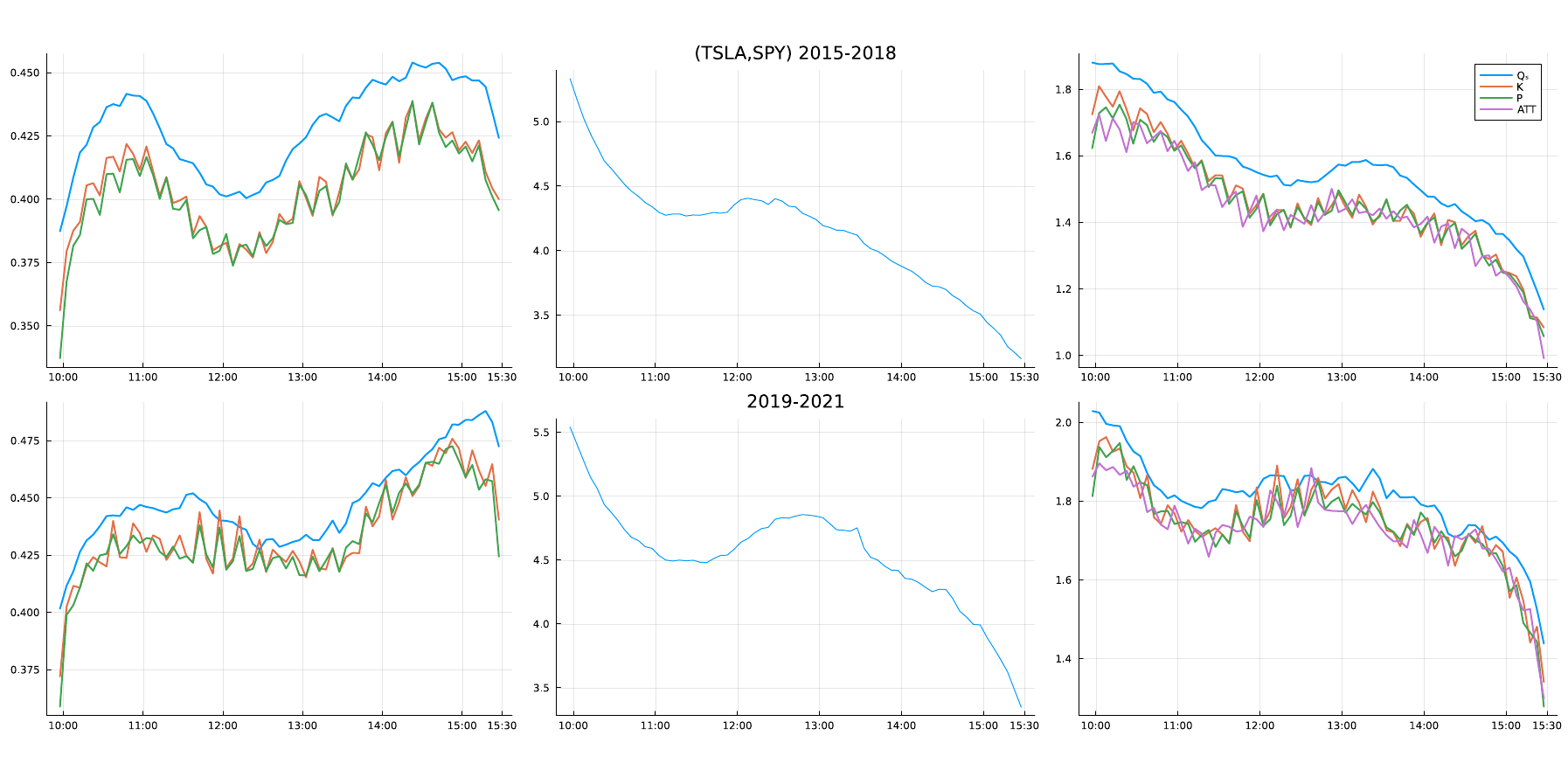}\includegraphics[width=0.33\textwidth]{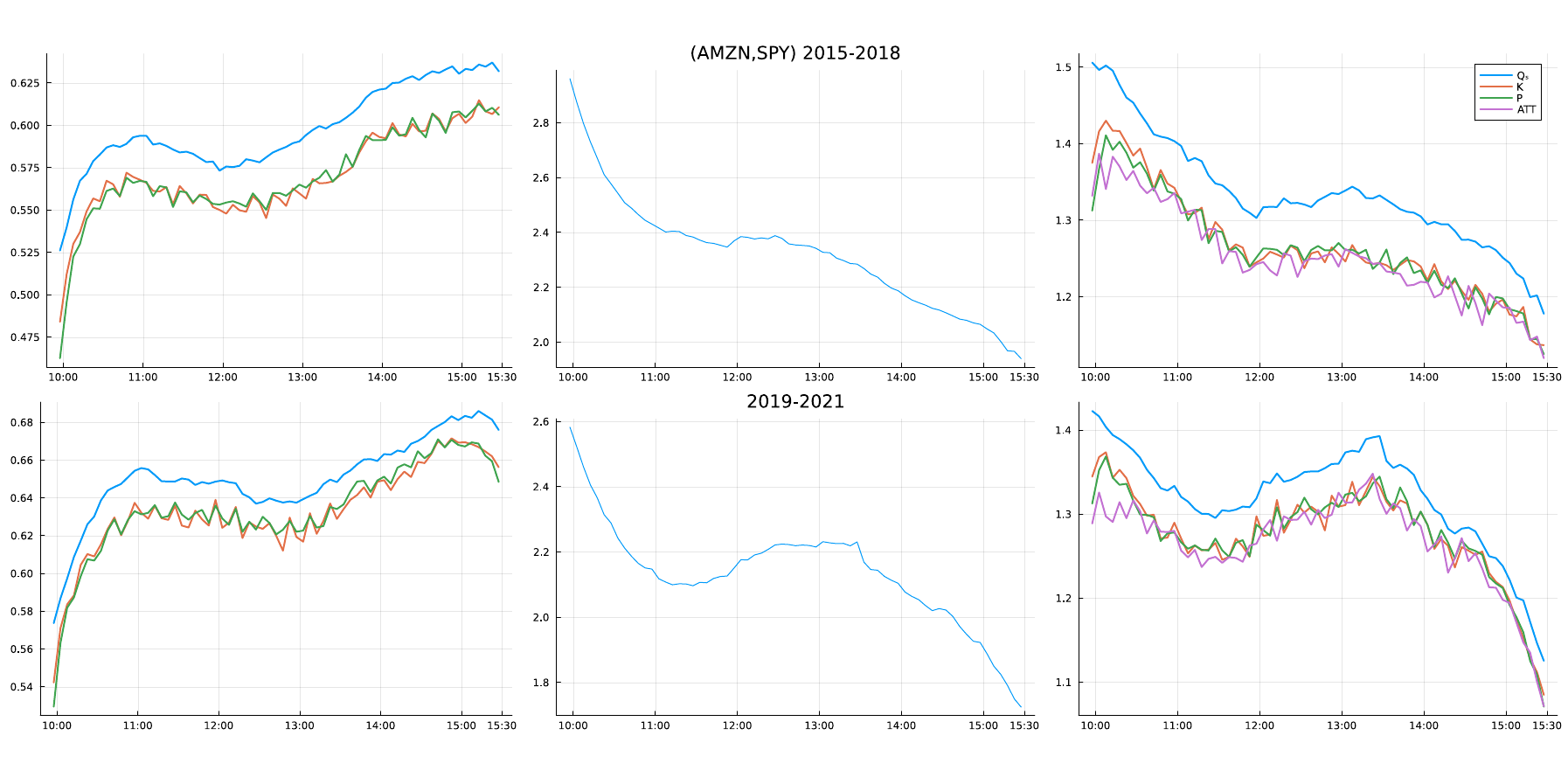}\includegraphics[width=0.33\textwidth]{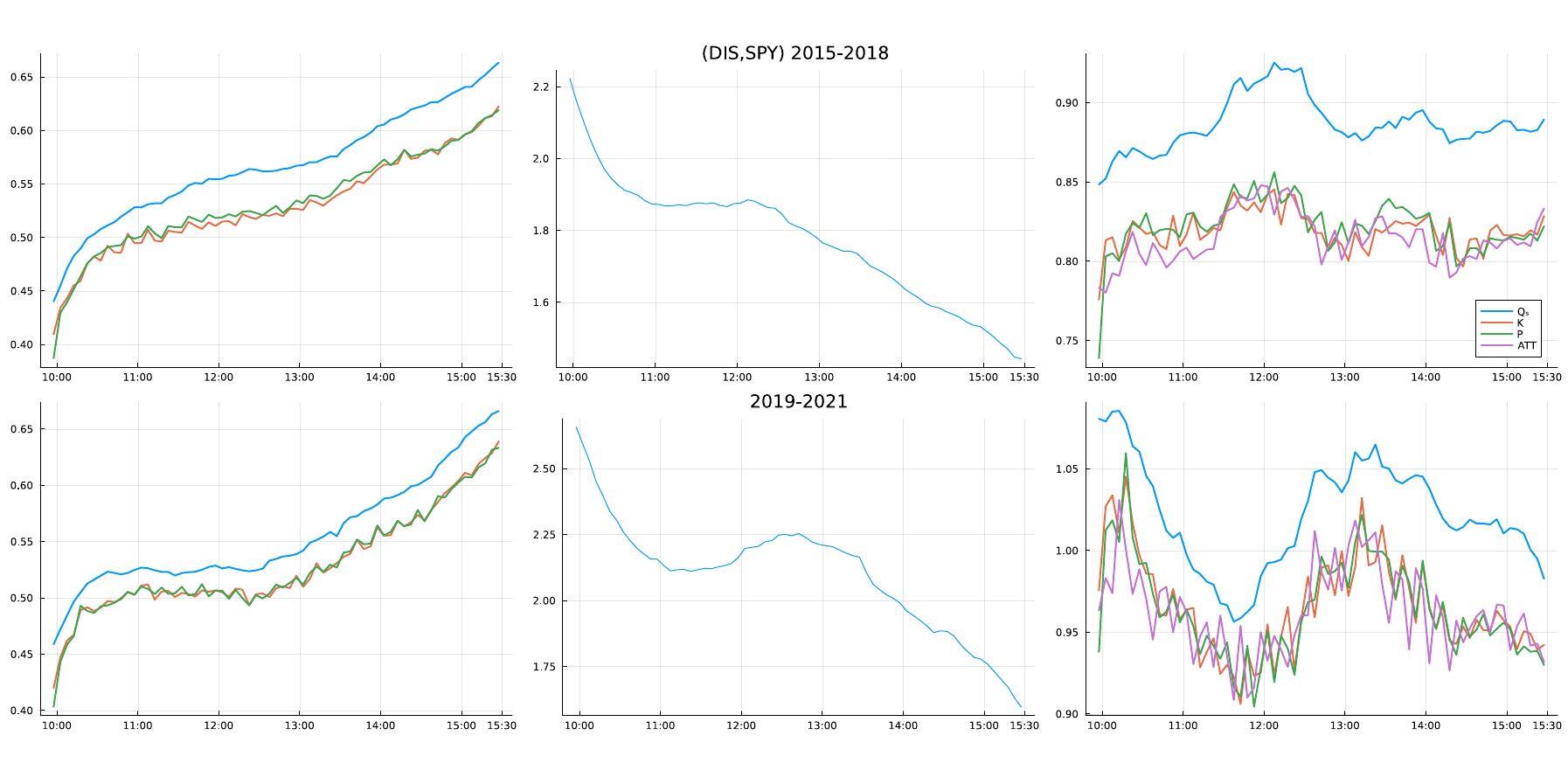}
\par\end{centering}
\begin{centering}
\includegraphics[width=0.33\textwidth]{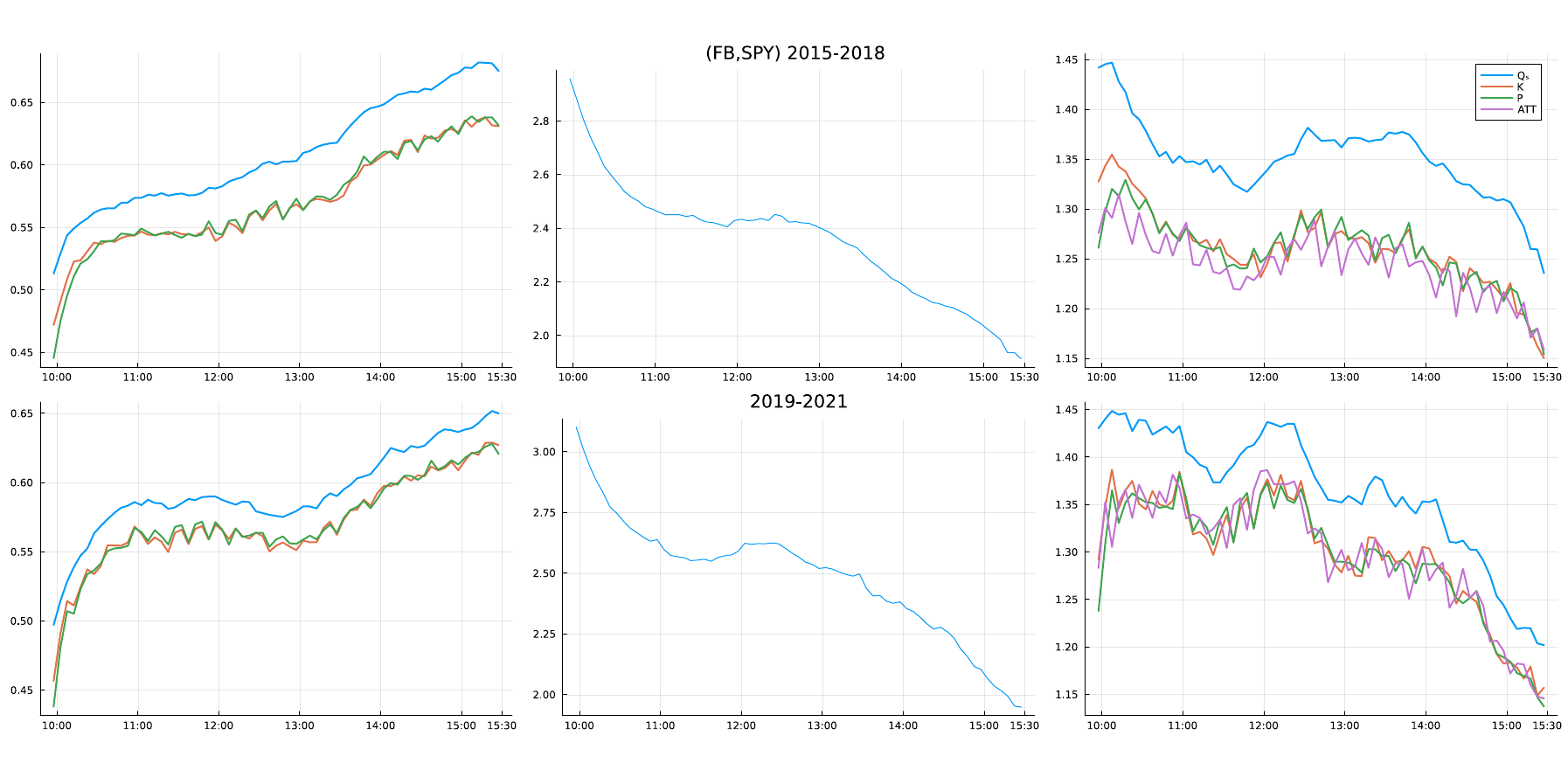}
\par\end{centering}
\caption{Intraday estimates for 22 assets for two sample periods. The rolling
windows span $60$ minutes and each estimate is plotted against the
midpoint of its window.\label{fig:SplitIntradayEstimates}}
\end{figure}
 Figure \ref{fig:SplitIntradayEstimates} shows that the results are
again very similar across the two subperiods.

\subsection{Zero Increment Products by Time of Day}

Table \ref{tab:ZeroProductsByTime} reports, for each stock in the
Small Universe, the percentage of three-minute increment products
with SPY that are exactly zero, over the full trading day and by hour.
Returns are calendar-time sampled from previous-tick prices on the
one-second grid, as in the rest of the empirical analysis, with the
first return ending at 9:33 and the last at 16:00. The frequencies
are averages over the 1,763 trading days. They are discussed in Section
\ref{sec:Empirical-intraday}.

\begin{table}[H]
\caption{Zero increment products with SPY by time of day\label{tab:ZeroProductsByTime}}

\begin{centering}
{\small\begin{tabular*}{\textwidth}{@{\extracolsep{\fill}}llS[table-format=2.1]S[table-format=2.1]S[table-format=2.1]S[table-format=2.1]S[table-format=2.1]S[table-format=2.1]S[table-format=2.1]S[table-format=2.1]}
\hline
\ \\[-6pt]
Sector & Ticker & \multicolumn{1}{c}{All} & \multicolumn{1}{c}{9:30--} & \multicolumn{1}{c}{10:30--} & \multicolumn{1}{c}{11:30--} & \multicolumn{1}{c}{12:30--} & \multicolumn{1}{c}{13:30--} & \multicolumn{1}{c}{14:30--} & \multicolumn{1}{c}{15:30--}\\
 &  &  & \multicolumn{1}{c}{10:30} & \multicolumn{1}{c}{11:30} & \multicolumn{1}{c}{12:30} & \multicolumn{1}{c}{13:30} & \multicolumn{1}{c}{14:30} & \multicolumn{1}{c}{15:30} & \multicolumn{1}{c}{16:00}\\[5pt]
\hline
\ \\[-6pt]
Energy                  & HAL  &  11.6 &   7.3 &   9.4 &  11.7 &  14.1 &  14.5 &  13.5 &  10.0\\
(10)                    & XOM  &   9.9 &   6.2 &   8.0 &  10.1 &  12.2 &  12.0 &  11.5 &   8.5\\[3pt]
Materials               & LYB  &   8.2 &   4.6 &   6.4 &   8.2 &  10.0 &  10.1 &  10.0 &   8.1\\
(15)                    & NEM  &  12.1 &   7.2 &  10.0 &  12.4 &  14.5 &  14.9 &  14.1 &  10.9\\[3pt]
Industrials             & AAL  &  11.7 &   6.5 &   9.5 &  12.0 &  14.1 &  14.7 &  13.9 &  10.9\\
(20)                    & UNP  &   7.7 &   4.4 &   6.2 &   7.8 &   9.2 &   9.6 &   9.2 &   7.0\\[3pt]
Consumer Discretionary  & TSLA &   4.7 &   3.0 &   3.8 &   4.9 &   5.5 &   5.7 &   5.5 &   4.4\\
(25)                    & AMZN &   4.4 &   2.8 &   3.6 &   4.6 &   5.1 &   5.3 &   5.0 &   3.8\\[3pt]
Consumer Staples        & PG   &  10.1 &   6.1 &   8.7 &  10.5 &  12.5 &  12.4 &  11.5 &   8.4\\
(30)                    & WMT  &   9.9 &   5.8 &   8.3 &  10.2 &  12.3 &  12.3 &  11.3 &   8.3\\[3pt]
Health Care             & JNJ  &   9.1 &   5.2 &   7.3 &   9.4 &  11.3 &  11.4 &  10.8 &   7.7\\
(35)                    & MRK  &  10.9 &   6.6 &   9.3 &  11.6 &  13.3 &  13.3 &  12.2 &   8.7\\[3pt]
Financials              & JPM  &   8.7 &   5.2 &   7.0 &   9.1 &  10.7 &  10.8 &  10.1 &   7.5\\
(40)                    & WFC  &  11.0 &   7.1 &   9.4 &  11.6 &  13.3 &  13.4 &  12.2 &   9.0\\[3pt]
Information Technology  & AAPL &   7.0 &   4.1 &   5.7 &   7.4 &   8.3 &   8.5 &   8.1 &   6.3\\
(45)                    & AMD  &  21.6 &  13.6 &  18.8 &  22.5 &  24.9 &  25.3 &  25.1 &  19.9\\[3pt]
Communication Services  & DIS  &   8.5 &   5.0 &   6.8 &   8.7 &  10.3 &  10.6 &   9.9 &   7.6\\
(50)                    & FB   &   6.6 &   3.9 &   5.3 &   6.8 &   7.9 &   8.3 &   7.8 &   5.8\\[3pt]
Utilities               & D    &  11.1 &   6.6 &   9.2 &  11.7 &  13.4 &  14.1 &  12.8 &   9.3\\
(55)                    & DUK  &  10.5 &   6.0 &   8.5 &  11.1 &  12.7 &  13.0 &  12.2 &   9.3\\[3pt]
Real Estate             & AMT  &   7.5 &   4.4 &   6.0 &   7.5 &   9.1 &   9.3 &   9.0 &   7.1\\
(60)                    & PLD  &  11.7 &   6.9 &   9.7 &  12.3 &  14.3 &  14.6 &  13.4 &   9.8\\[3pt]
\hline
\end{tabular*}}
\par\end{centering}
Note: Percentage of three-minute increment products $\Delta X\Delta Y$
that are exactly zero, where $X$ is SPY and $Y$ is the stock, for
the full trading day (All) and for each hour of the trading day, the
last column covering the final 30 minutes. An increment product is
zero when either the stock or SPY has a zero three-minute return.
The sample period is January 1, 2015, through December 31, 2021.
\end{table}

Additional results for the Large Universe of assets are reported in
Figures \ref{fig:DeltaRhoLambdaBeta-vs-BetaSizeBTM} and \ref{fig:FirstLastRhoBetavsBetaSizeBTM}.
We consider intraday changes in market correlations, relative volatility,
and intraday betas, defined as logarithmic differences between the
last-hour and first-hour estimates,
\[
\Delta\log\bar{z}_{i}=\log\bar{z}_{i,\mathrm{last}}-\log\bar{z}_{i,\mathrm{first}},\qquad z\in\{\rho,\lambda,\beta\},
\]
where the estimates are averages over the days in the sample period,
as in Section \ref{subsec:decomposition}. The changes are plotted
against low-frequency market betas, size, and book-to-market. 
\begin{figure}[H]
\begin{centering}
\subfloat[]{\begin{centering}
\includegraphics[width=0.32\textwidth]{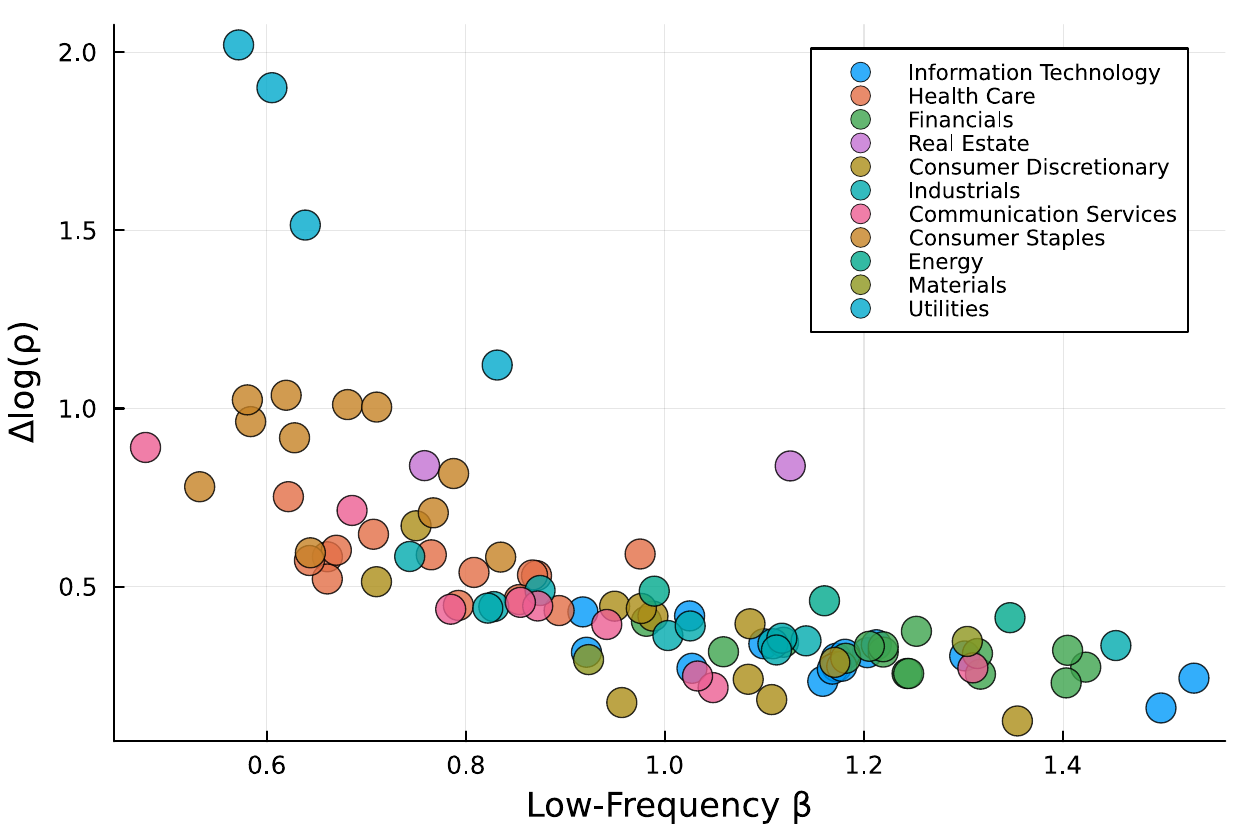}
\par\end{centering}
}\negthinspace{}\subfloat[]{\begin{centering}
\includegraphics[width=0.32\textwidth]{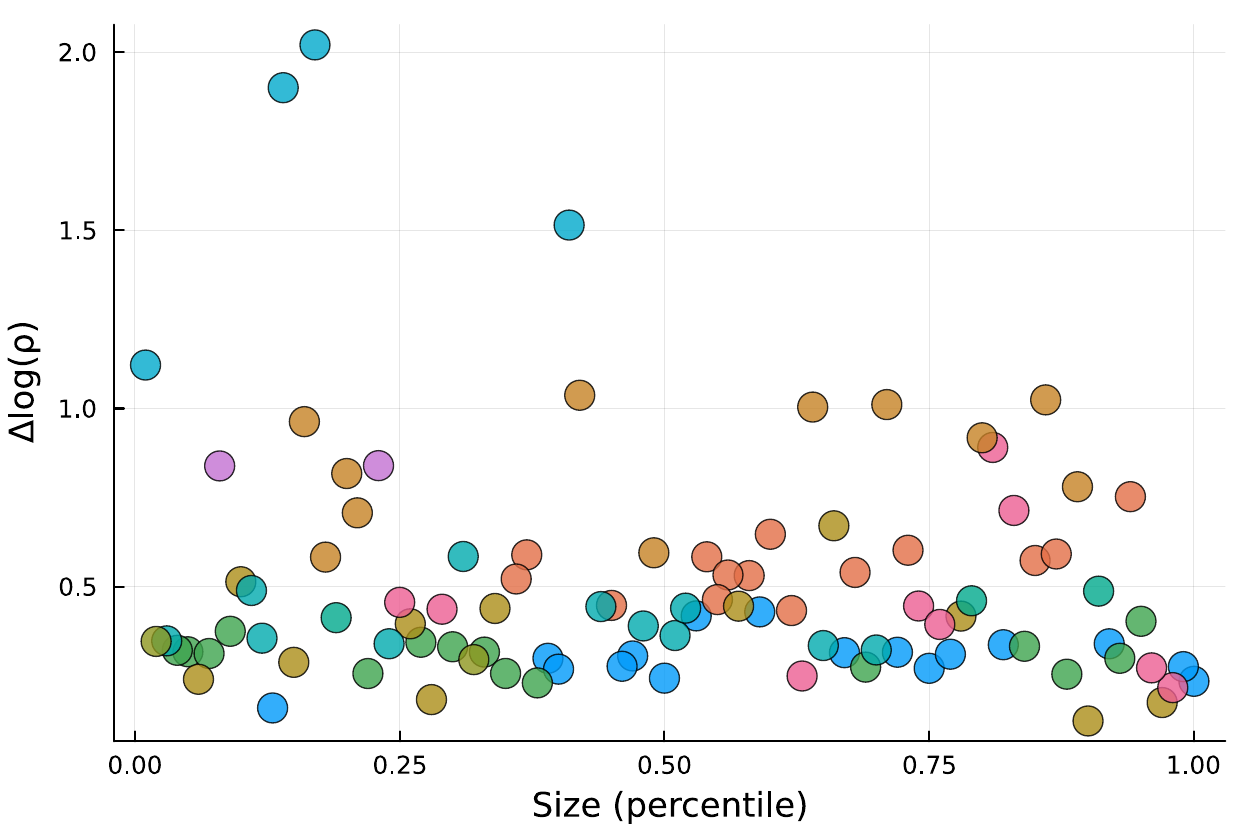}
\par\end{centering}
}\negthinspace{}\subfloat[]{\begin{centering}
\includegraphics[width=0.32\textwidth]{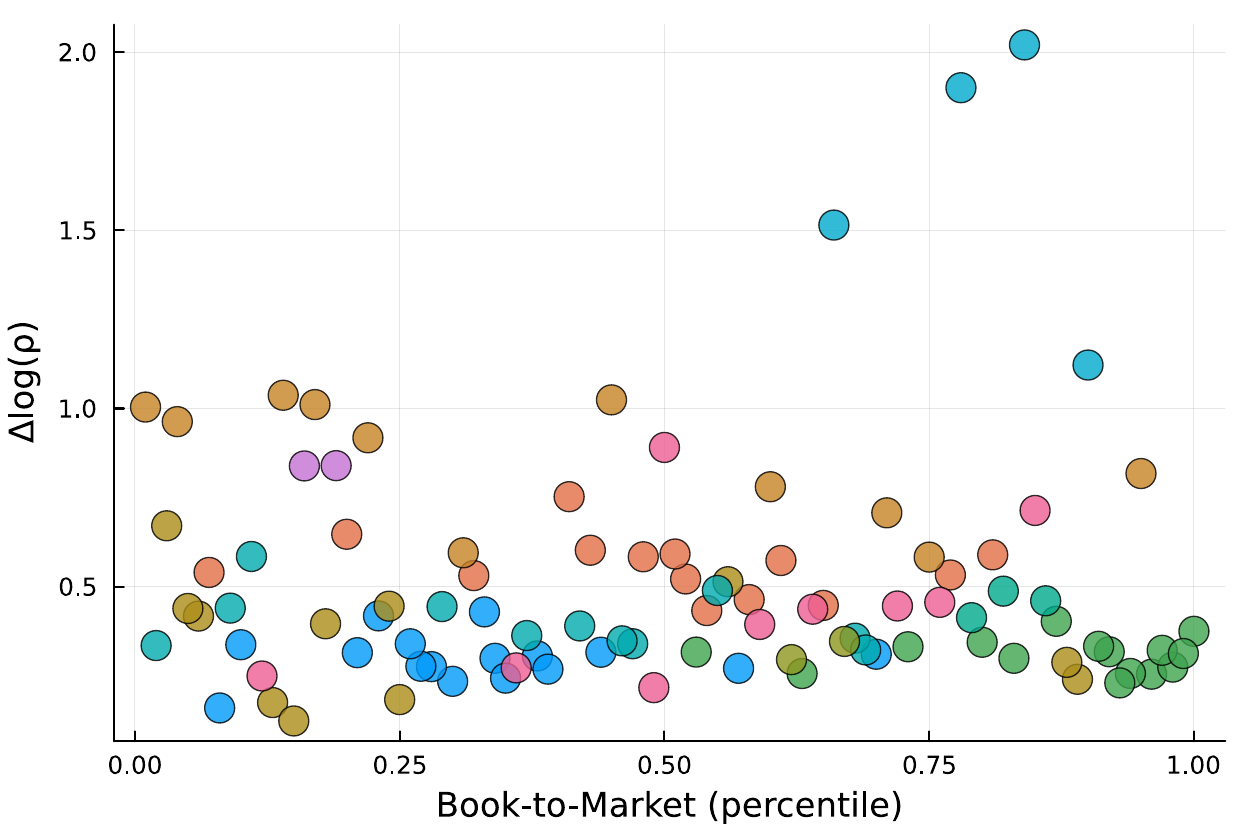}
\par\end{centering}
}
\par\end{centering}
\begin{centering}
\subfloat[]{\begin{centering}
\includegraphics[width=0.32\textwidth]{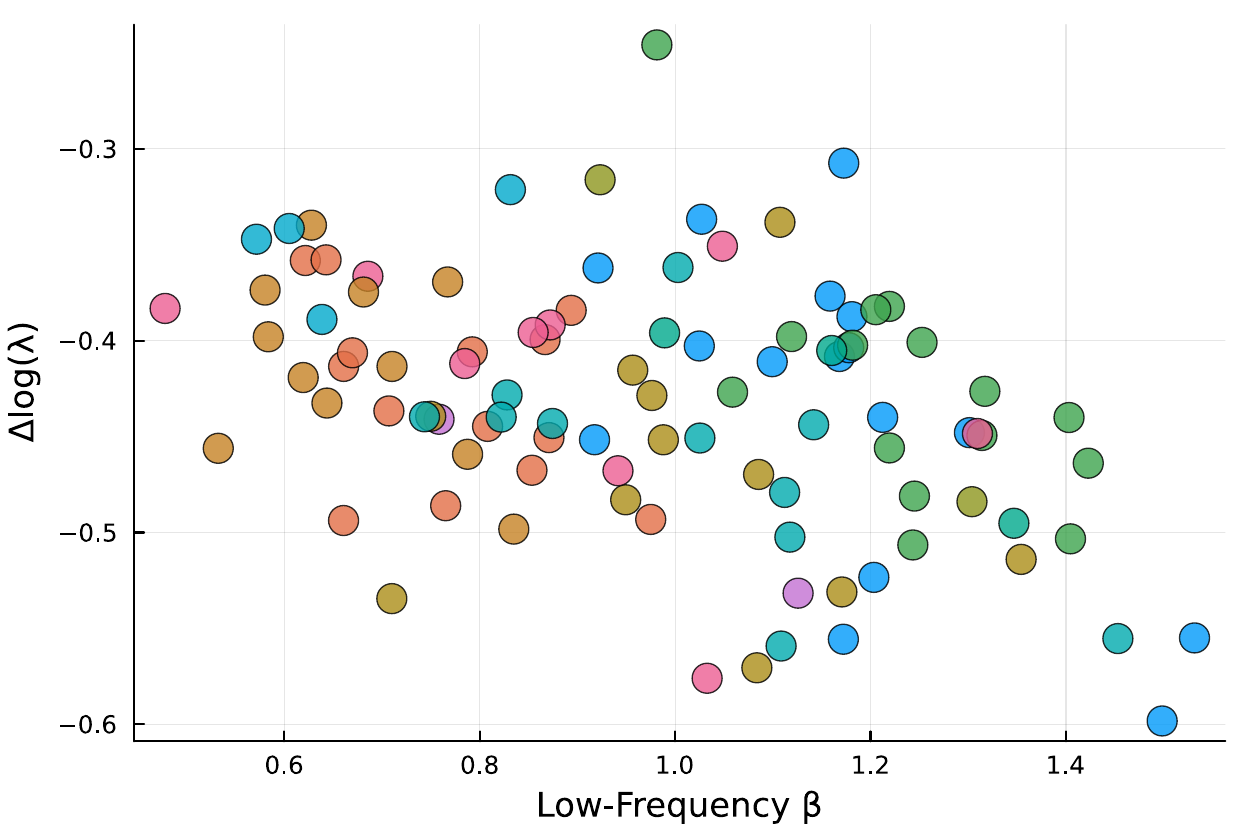}
\par\end{centering}
}\negthinspace{}\subfloat[]{\begin{centering}
\includegraphics[width=0.32\textwidth]{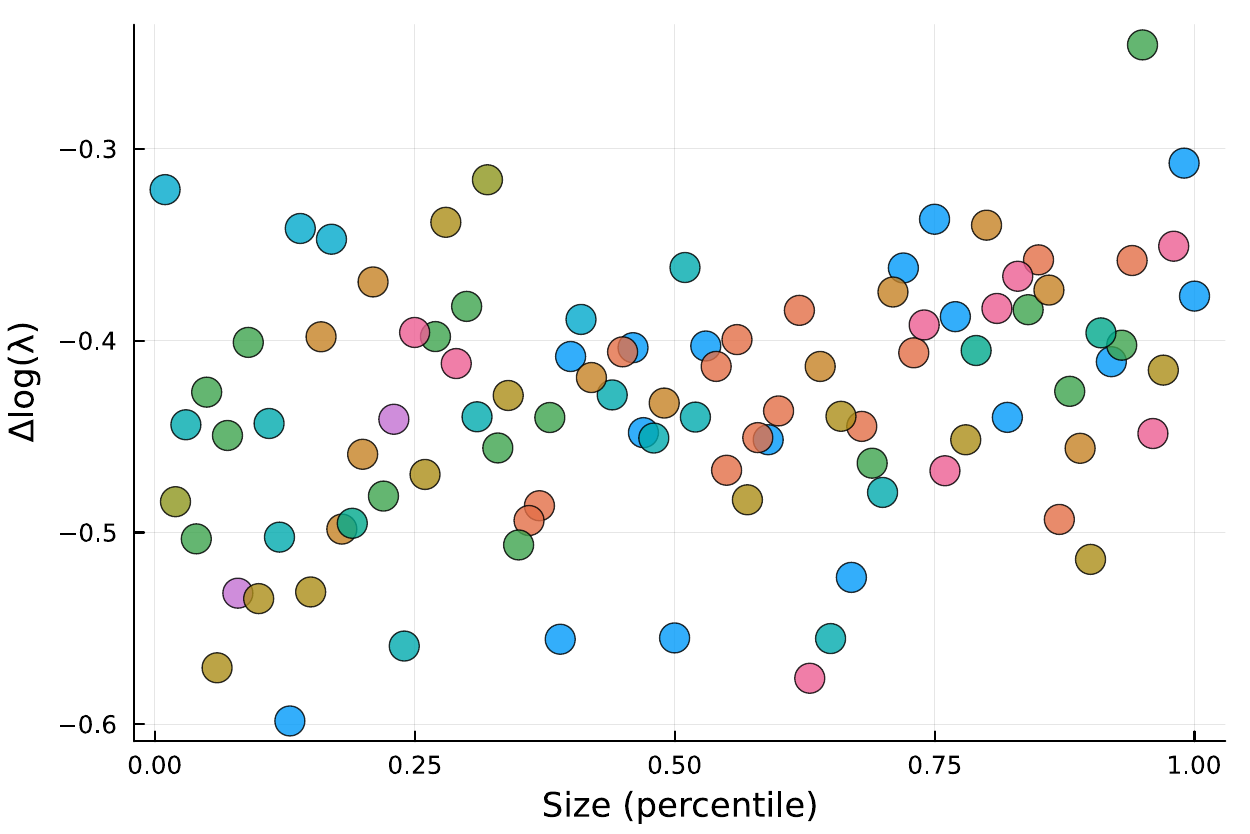}
\par\end{centering}
}\negthinspace{}\subfloat[]{\begin{centering}
\includegraphics[width=0.32\textwidth]{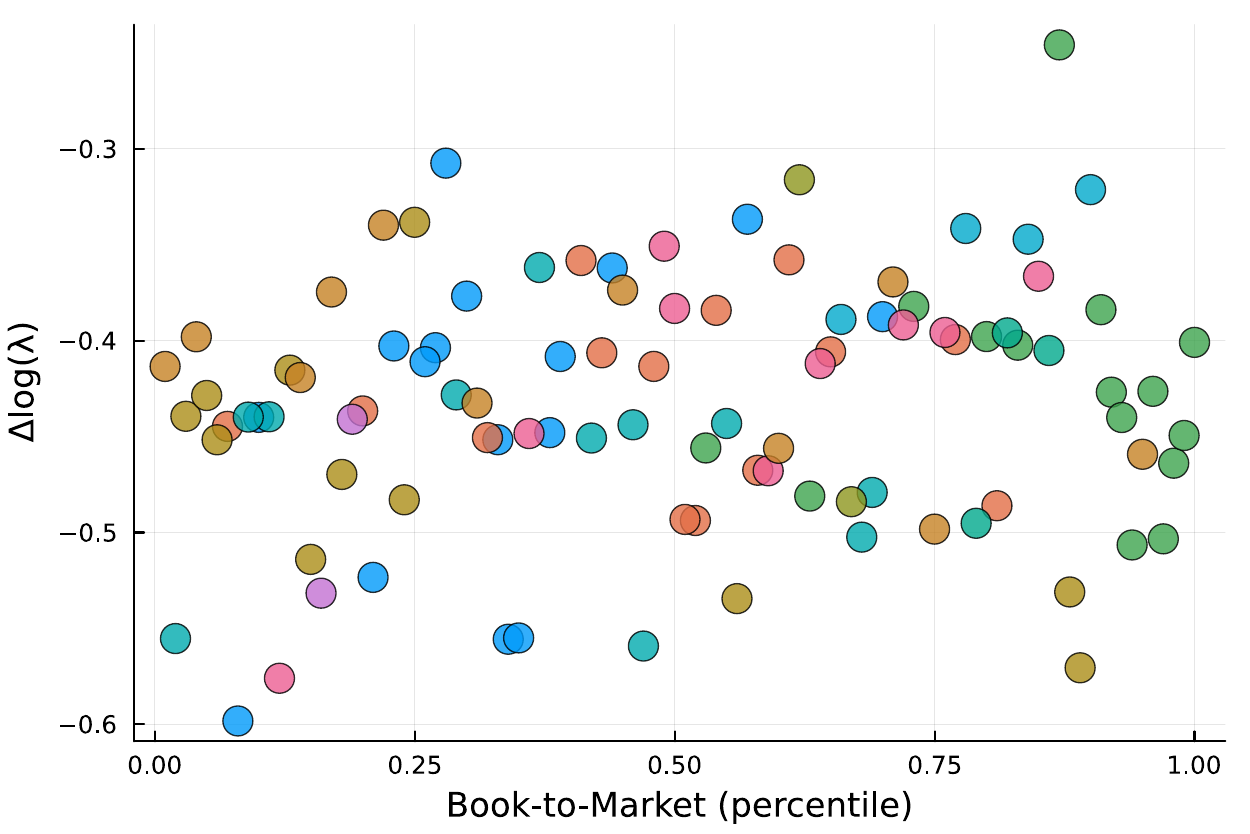}
\par\end{centering}
}
\par\end{centering}
\begin{centering}
\subfloat[]{\begin{centering}
\includegraphics[width=0.32\textwidth]{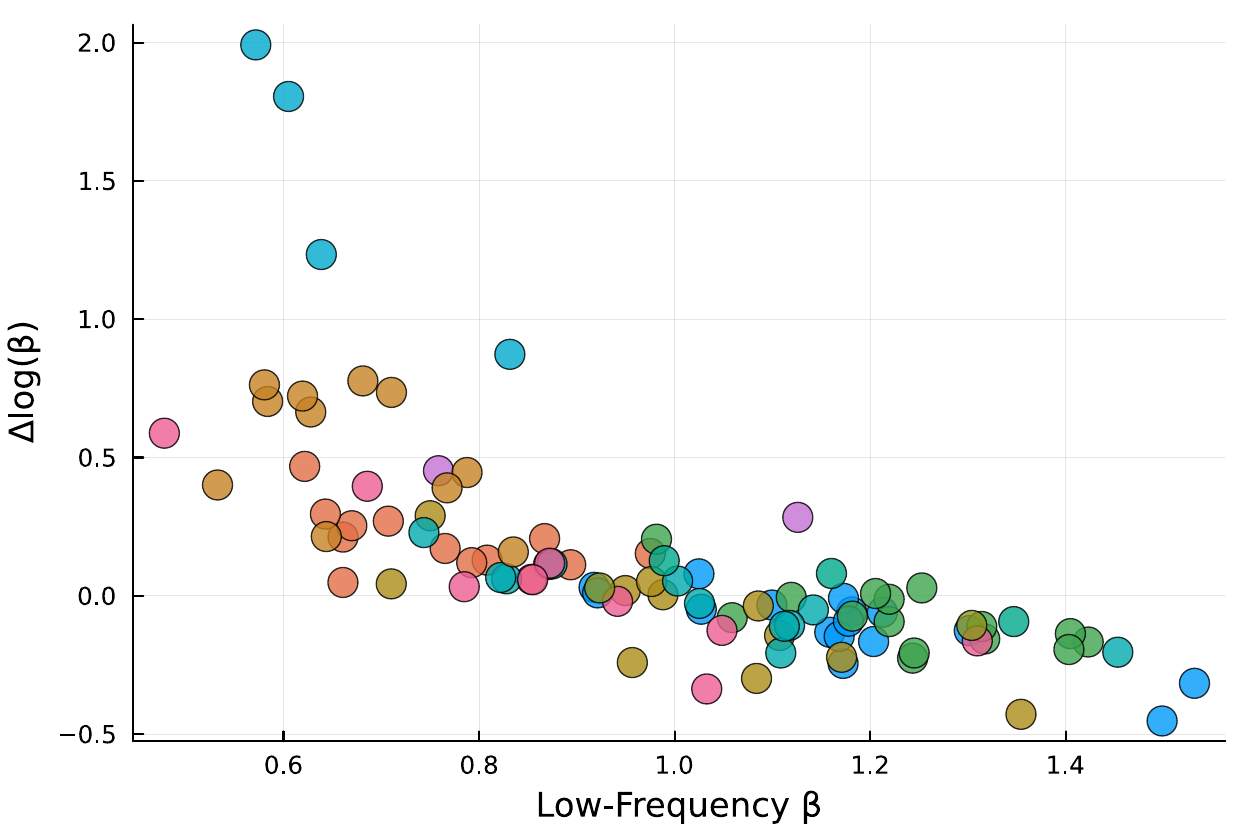}
\par\end{centering}
}\negthinspace{}\subfloat[]{\begin{centering}
\includegraphics[width=0.32\textwidth]{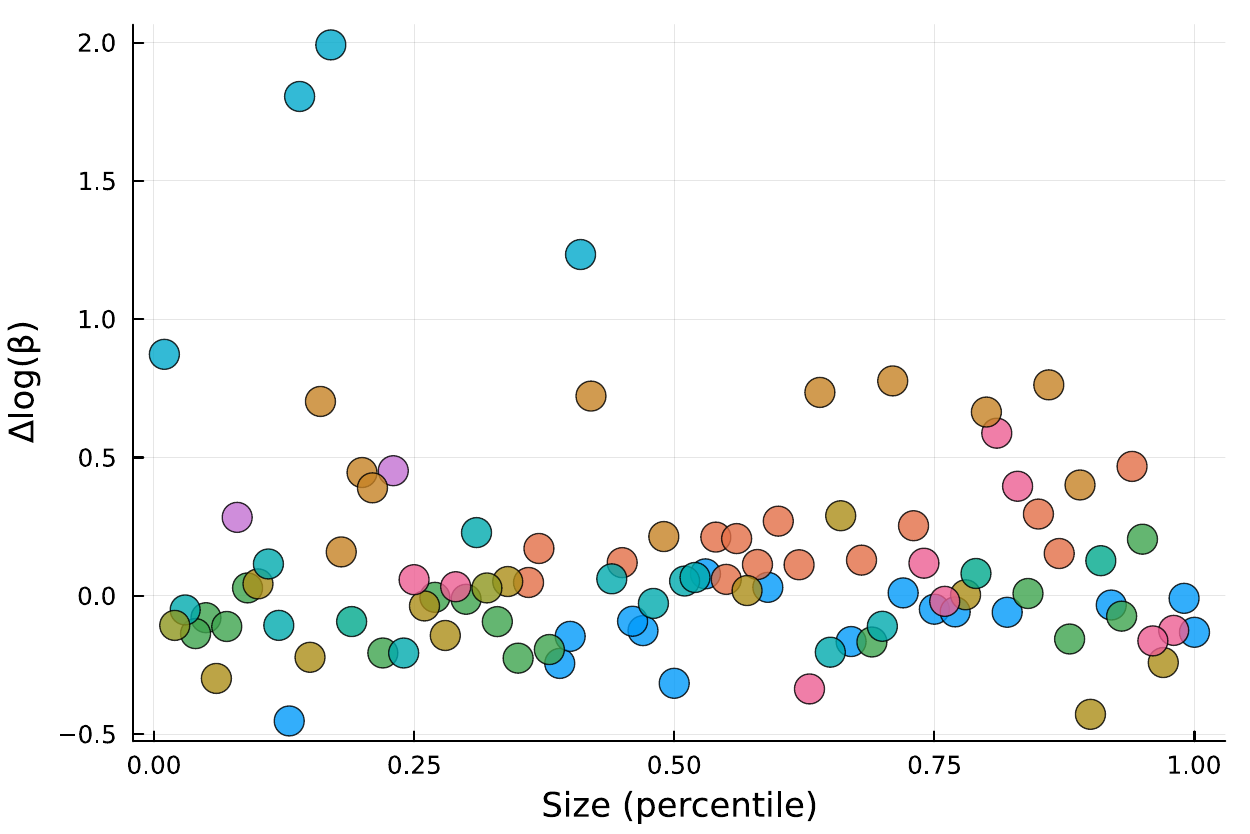}
\par\end{centering}
}\negthinspace{}\subfloat[]{\begin{centering}
\includegraphics[width=0.32\textwidth]{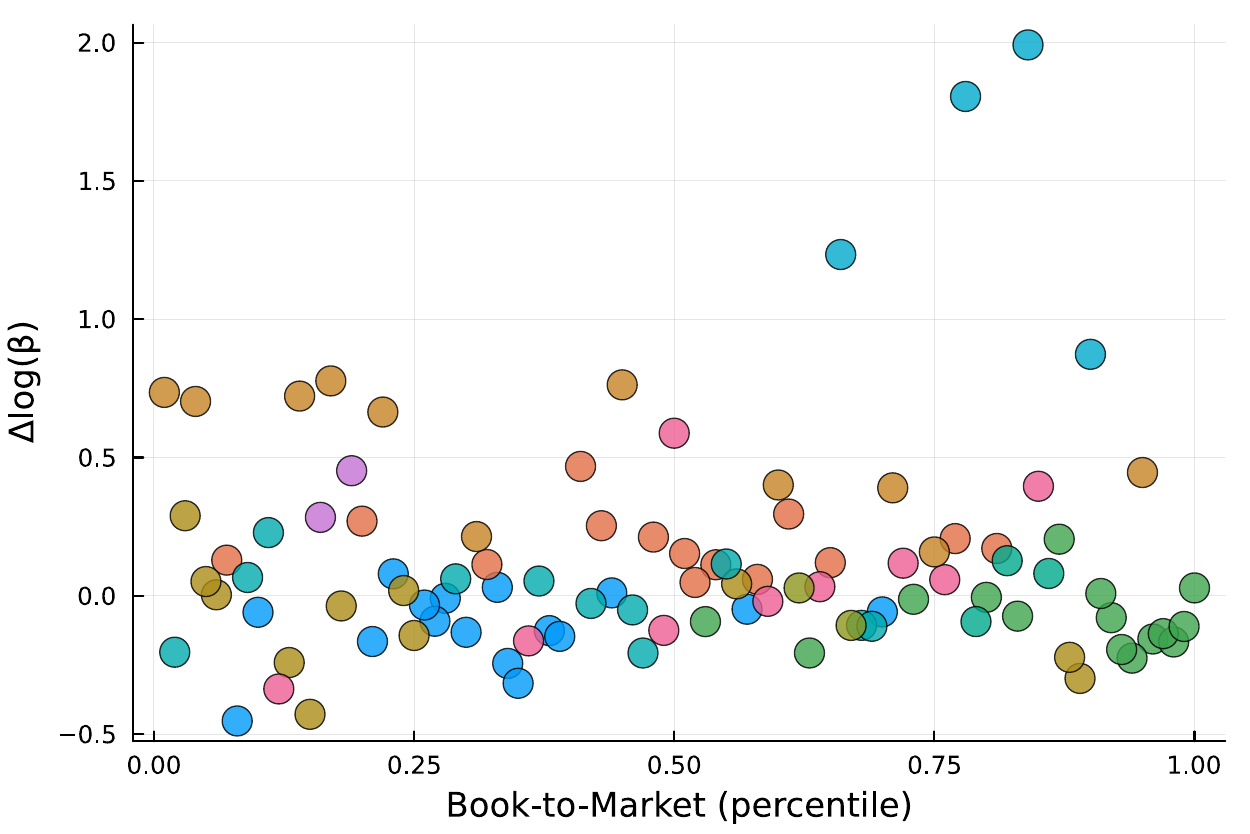}
\par\end{centering}
}
\par\end{centering}
\caption{Logarithmic changes in market correlations (top row), relative
volatility (middle row), and market betas (bottom row), from the first
hour to the last hour of the trading day, plotted against three variables. In the left panels the quantities are plotted
against the (low-frequency) market beta, which is estimated with daily
returns. In the middle panels they are plotted against (percentiles
of) market capitalization, and in the right panels they are plotted
against the (percentiles of) book-to-market values.\label{fig:DeltaRhoLambdaBeta-vs-BetaSizeBTM}}
\end{figure}
 The low-frequency market beta is the conventional estimate, which
is based on daily returns in our sample period. We have also sorted
assets by ``Size'' and ``Book-to-Market,'' where the former is the company's
market capitalization (market cap) and book-to-market is defined by
the company's book value relative to its market value. 

Figure \ref{fig:FirstLastRhoBetavsBetaSizeBTM} plots the intraday
correlations and market betas, estimated for the first and last hour
of the trading day, against the low-frequency beta and the percentiles
for Size and Book-to-Market. The association is clearly strongest
with the low-frequency betas in the left panels.
\begin{figure}[H]
\begin{centering}
\includegraphics[width=0.32\textwidth]{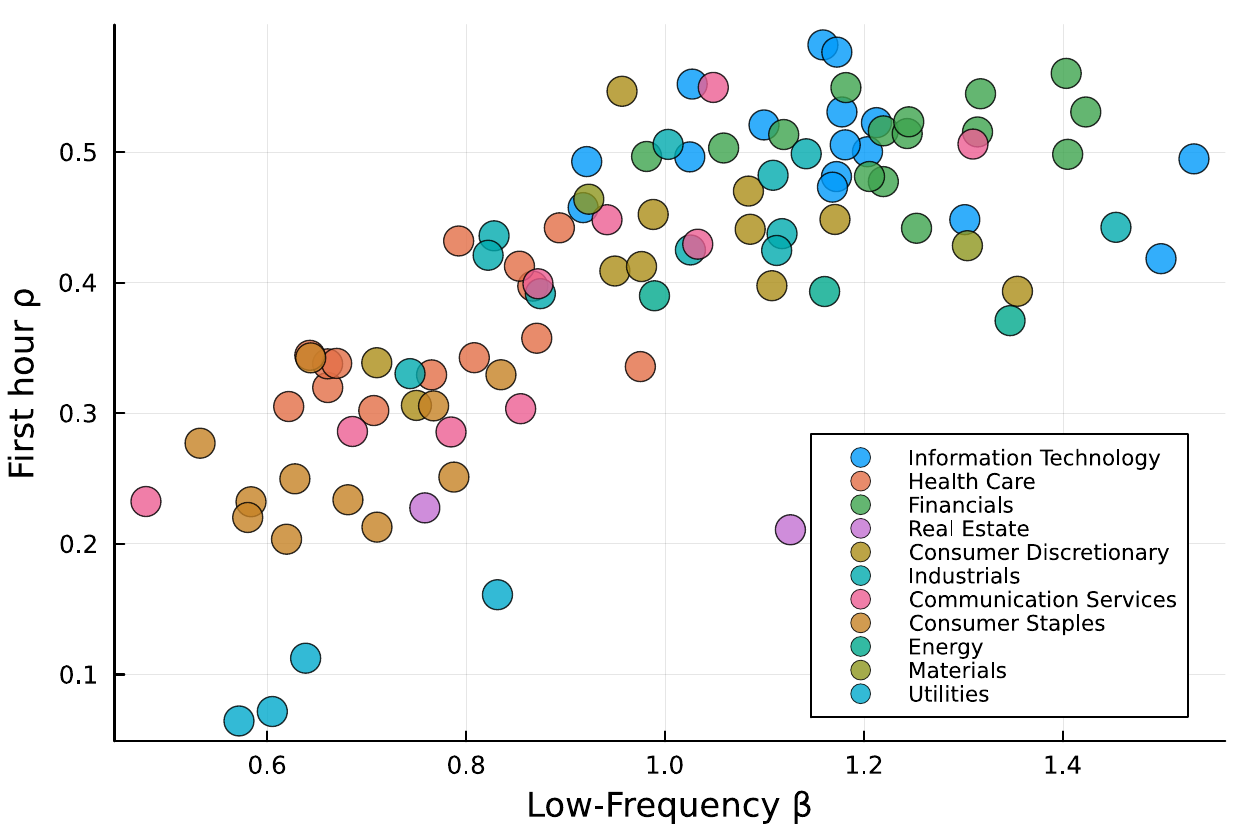}\includegraphics[width=0.32\textwidth]{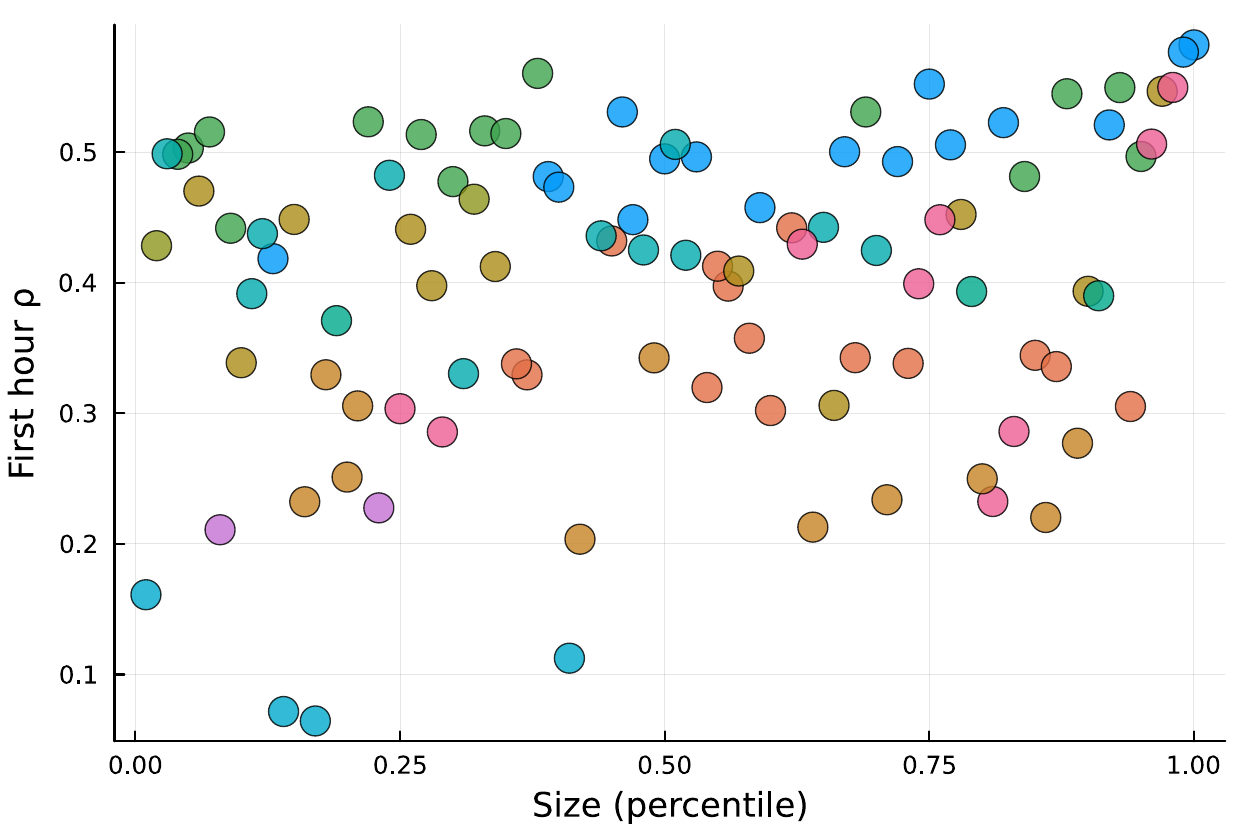}\includegraphics[width=0.32\textwidth]{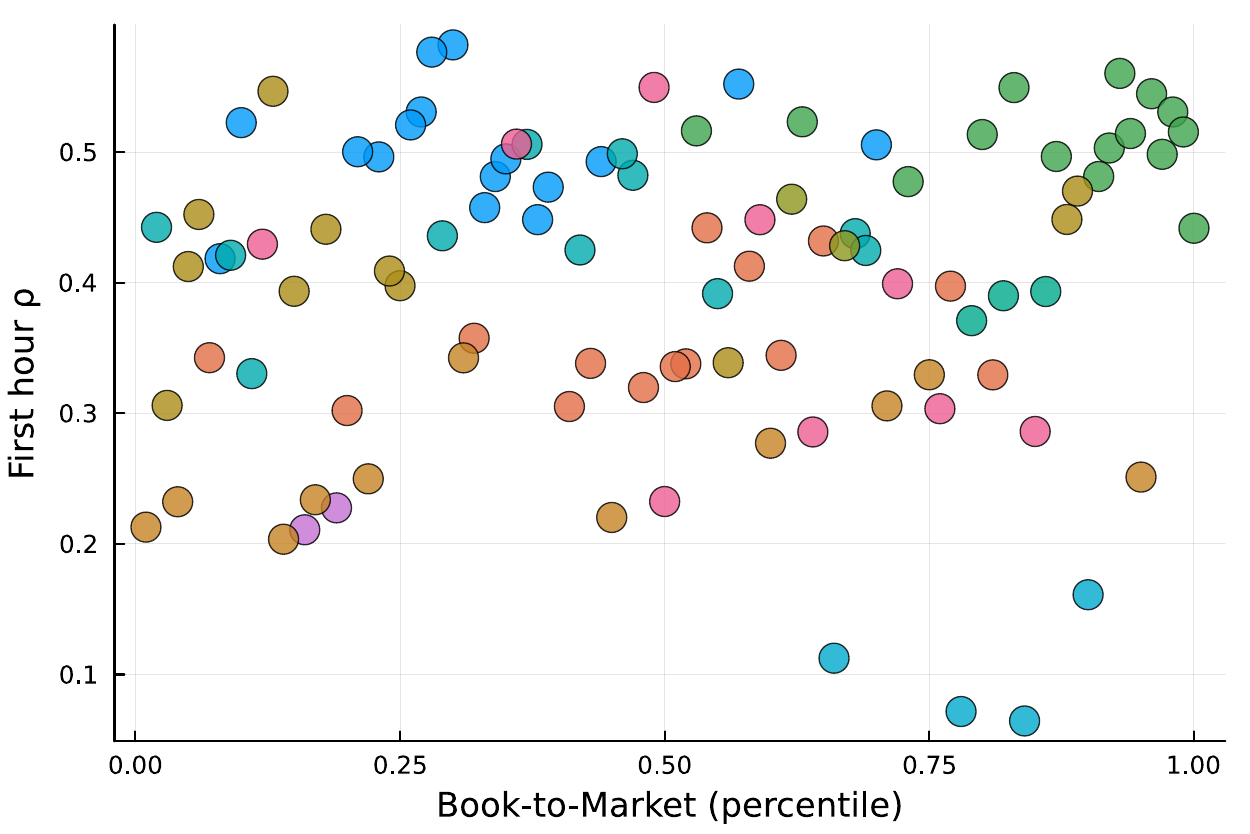}
\par\end{centering}
\begin{centering}
\includegraphics[width=0.32\textwidth]{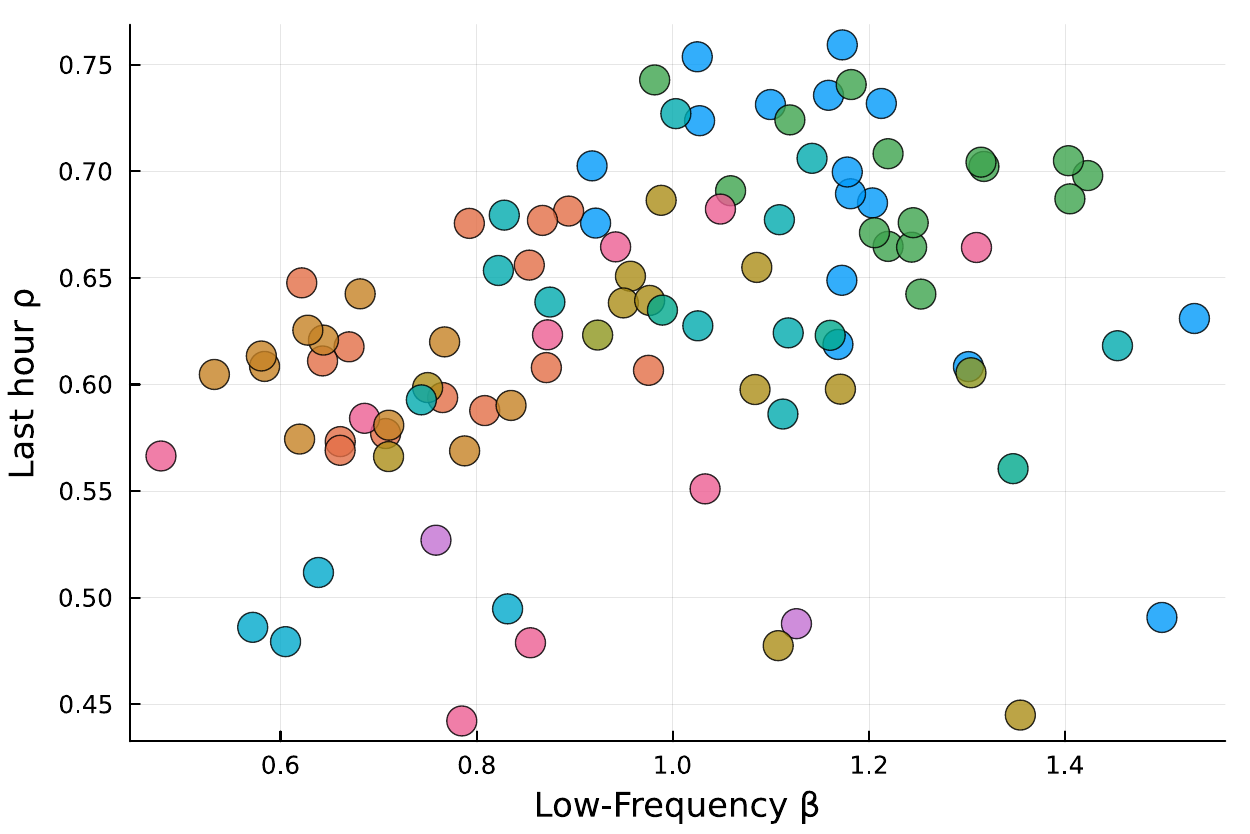}\includegraphics[width=0.32\textwidth]{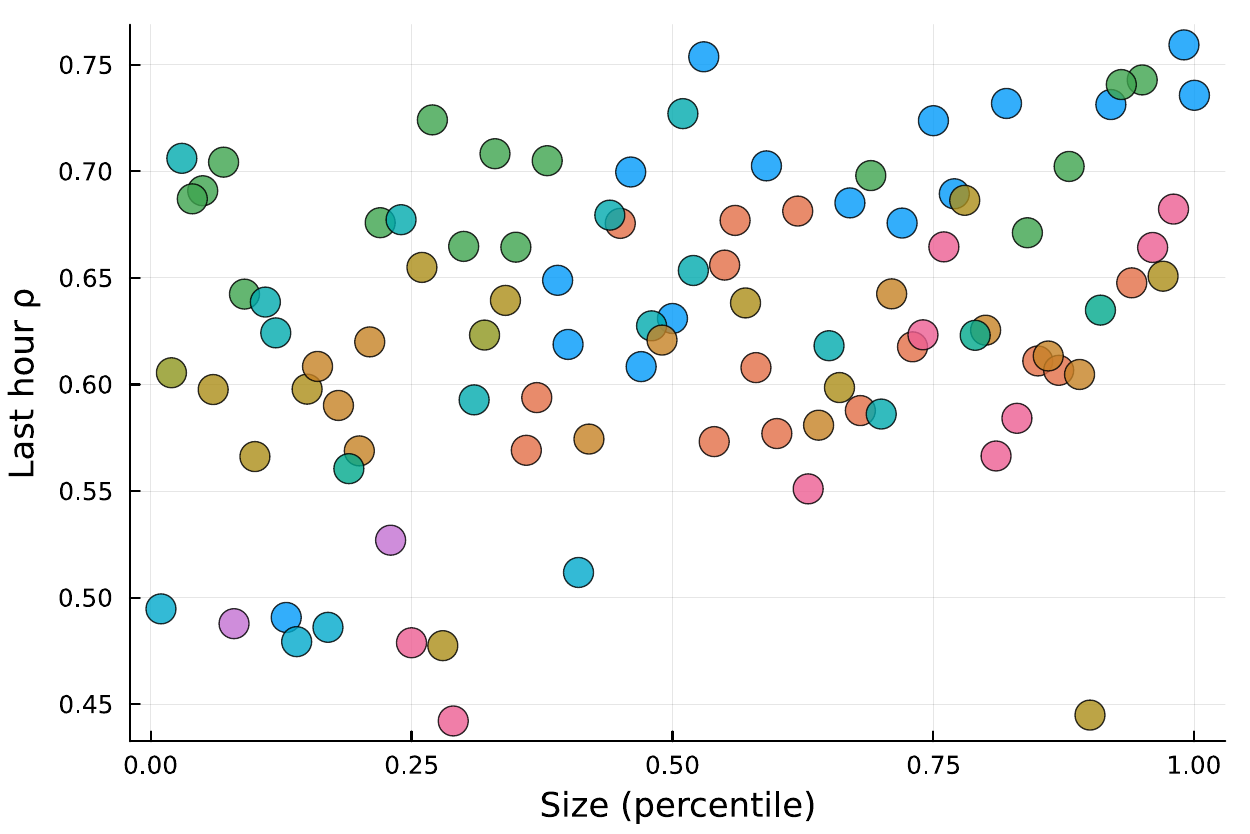}\includegraphics[width=0.32\textwidth]{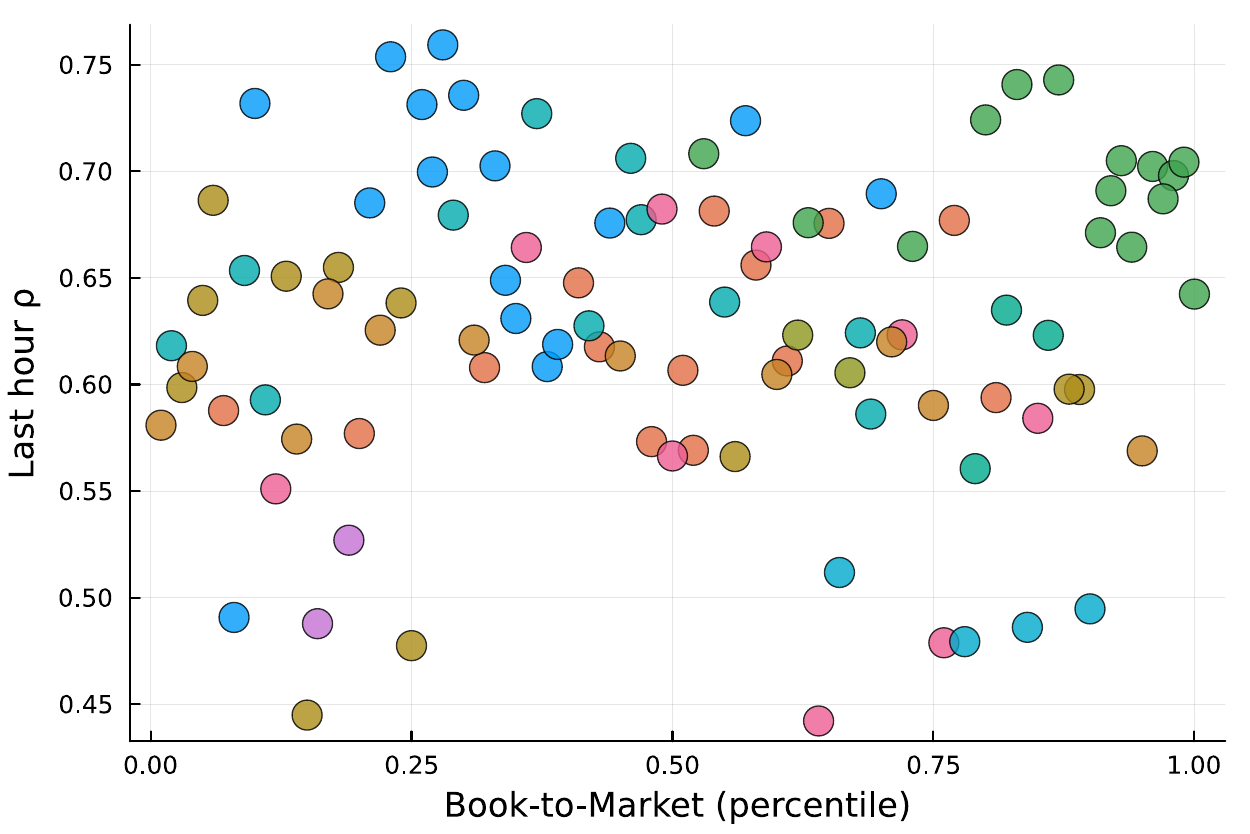}
\par\end{centering}
\begin{centering}
\includegraphics[width=0.32\textwidth]{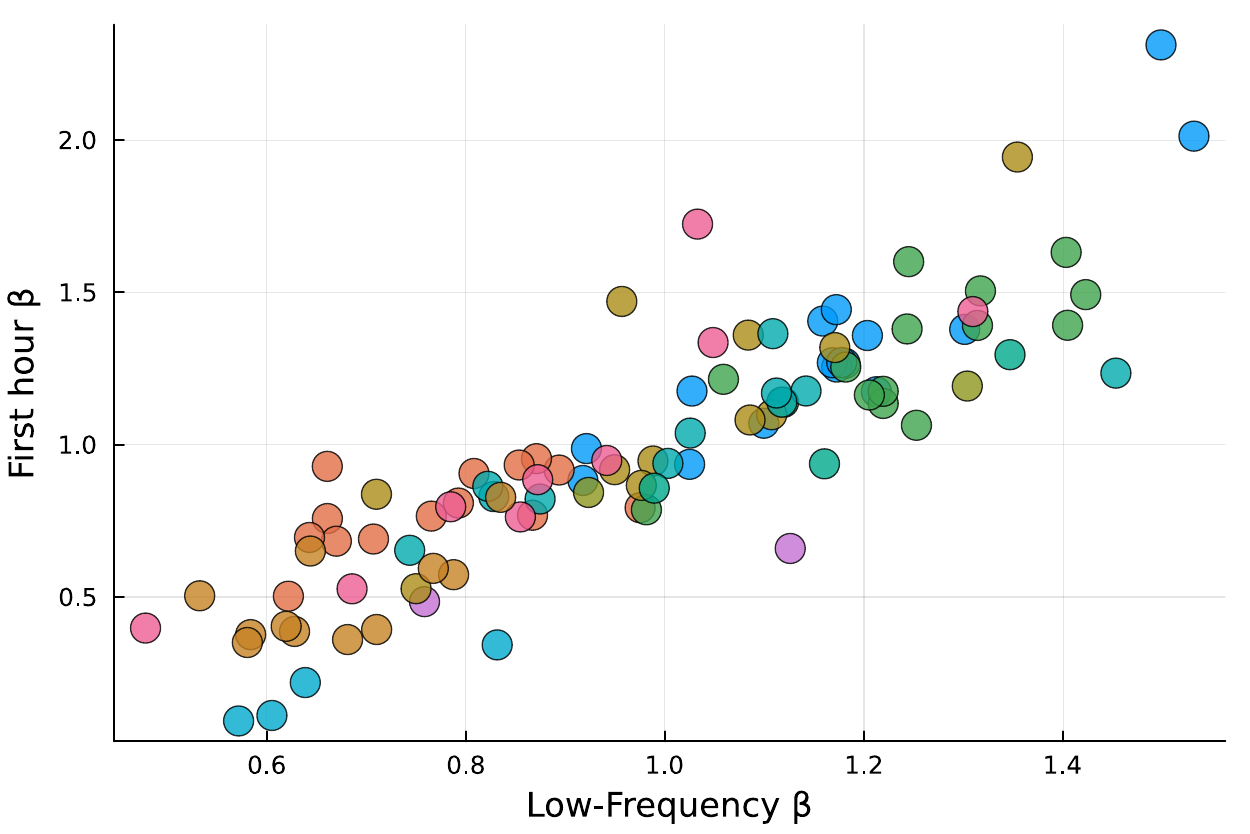}\includegraphics[width=0.32\textwidth]{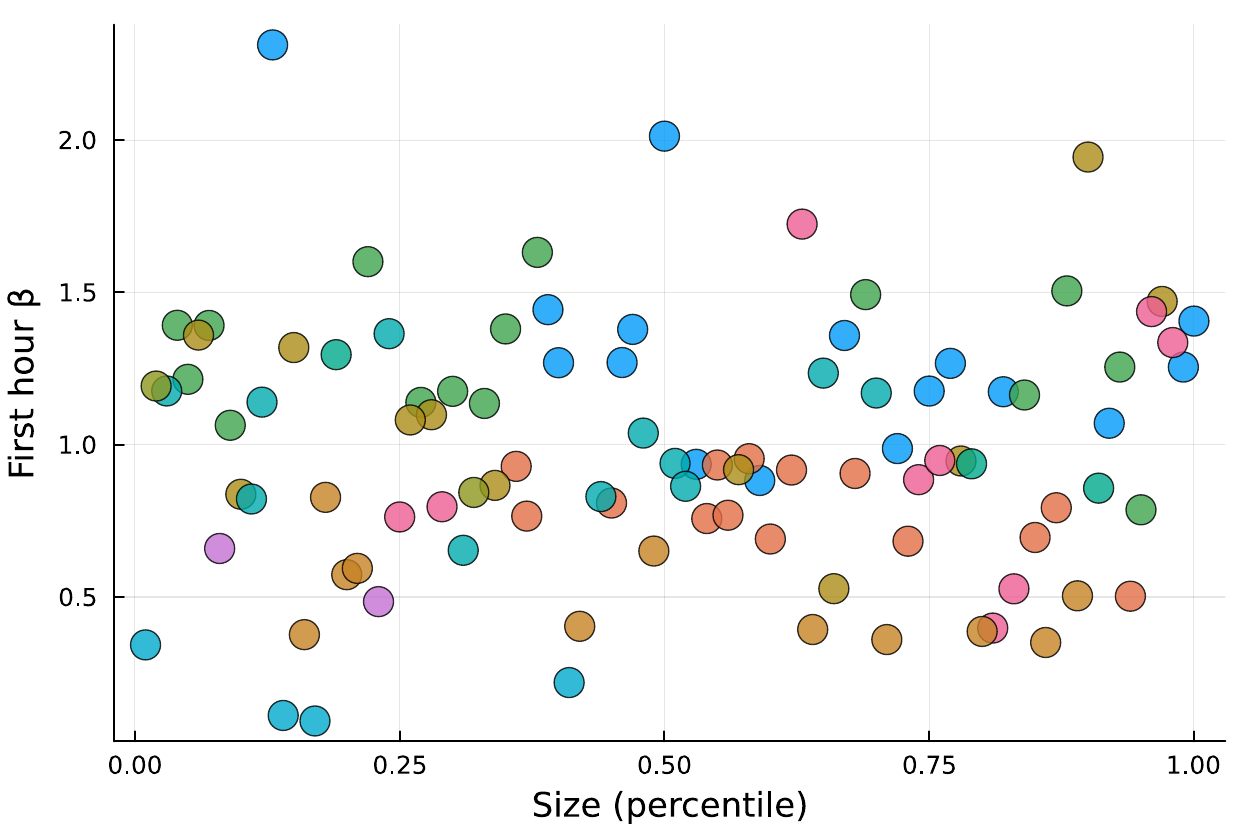}\includegraphics[width=0.32\textwidth]{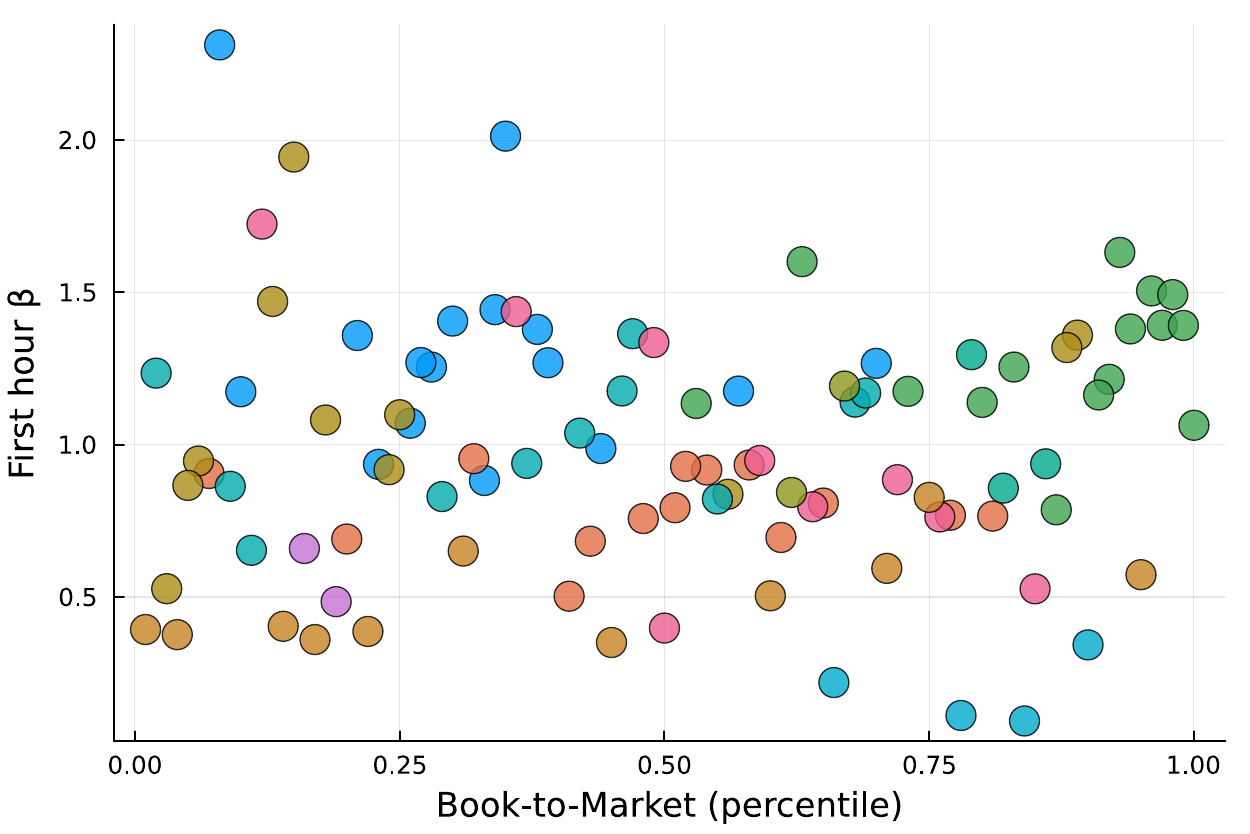}
\par\end{centering}
\begin{centering}
\includegraphics[width=0.32\textwidth]{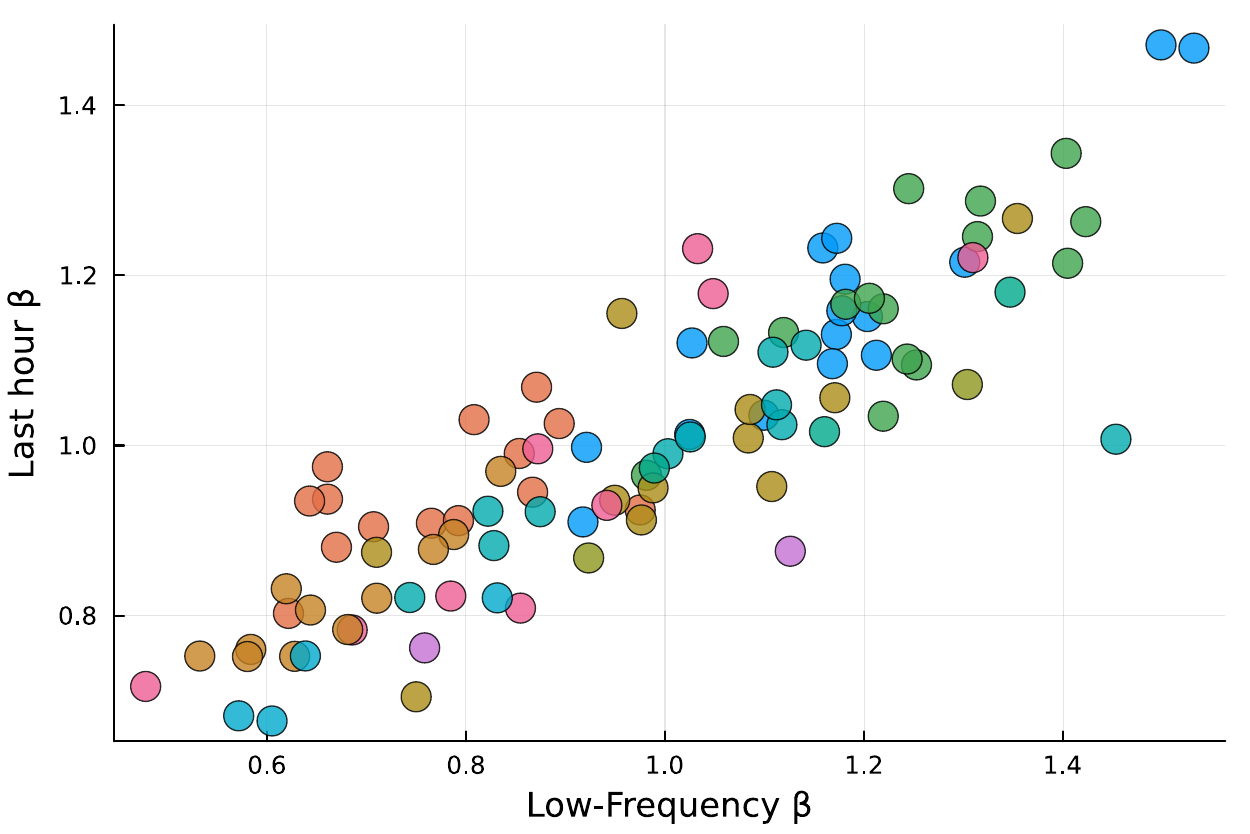}\includegraphics[width=0.32\textwidth]{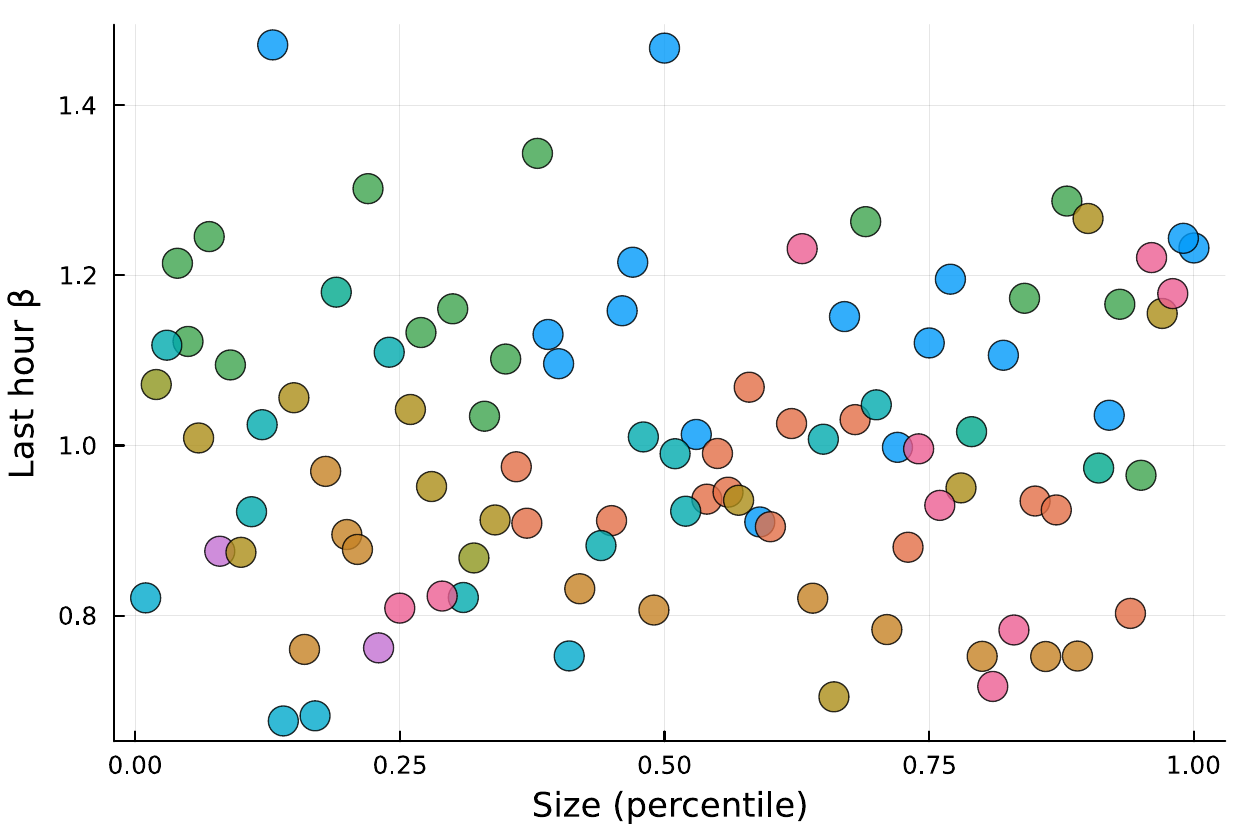}\includegraphics[width=0.32\textwidth]{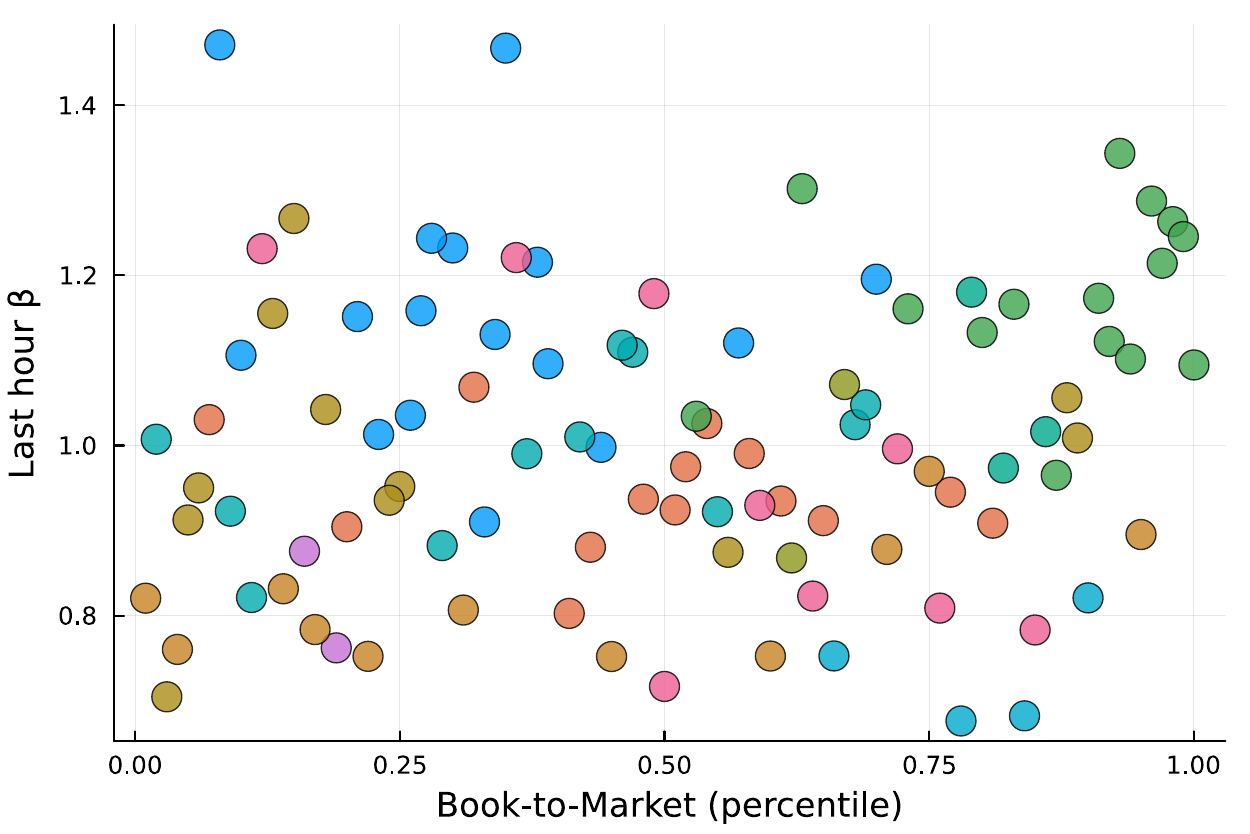}
\par\end{centering}
\caption{The average market correlations and market beta for the first and
last hours of the trading day are plotted against three variables.
In the left panels the quantities are plotted against the (low-frequency)
market beta, which is estimated with daily returns. In the middle
panels they are plotted against (percentiles of) market capitalization
size, and in the right panels they are plotted against the (percentiles
of) book-to-market values.\label{fig:FirstLastRhoBetavsBetaSizeBTM}}
\end{figure}

\section{Supplementary Theoretical Results}

For high-frequency financial data, estimators are typically applied
to sparsely sampled prices, and we consider the influence functions
of the estimators in this context. The contamination is placed on the
underlying high-frequency increments, $(X_{j},Y_{j})\sim(1-\varepsilon)\Phi_{\rho}+\varepsilon\mathbf{\Delta}_{(x_{0},y_{0})}$
independently across $j$, and the estimators are computed from the
sparse returns $\tilde{X}_{i}=\sum^{i+S-1}_{j=i}X_{j}$ and $\tilde{Y}_{i}=\sum^{i+S-1}_{j=i}Y_{j}$,
each of which aggregates $S$ increments. In this section $\Phi(\bullet,\bullet)$
denotes the joint distribution function of $\Phi_{\rho}$, $\Phi(\bullet)$
the standard normal distribution function, and $q=\Pr[XY>0]=\tfrac{1}{2}+\tfrac{1}{\pi}\arcsin\rho$
the probability of a concordant pair under $\Phi_{\rho}$, so that
$\tau=2q-1$. Throughout we take $S\geq2$; for $S=1$ the quadrant and
Kendall estimators reduce to $Q$ and $K$, whose influence functions
are given in Section \ref{subsec:influence}.
\begin{prop}
\label{thm:influence}The influence function of the Pearson estimator
at $\Phi_{\rho}$ is given by
\[
\mathrm{IF}((x_{0},y_{0}),R_{P},\Phi_{\rho})=x_{0}y_{0}-(\frac{x^{2}_{0}+y^{2}_{0}}{2})\rho.
\]
The influence function of the Kendall estimator at $\Phi_{\rho}$ is
given by
\[
\mathrm{IF}((x_{0},y_{0}),R_{K},\Phi_{\rho})=2\pi\sqrt{1-\rho^{2}}S\Big[2\Phi(\tfrac{x_{0}}{\sqrt{2S-1}},\tfrac{y_{0}}{\sqrt{2S-1}})-\Phi(\tfrac{x_{0}}{\sqrt{2S-1}})-\Phi(\tfrac{y_{0}}{\sqrt{2S-1}})+1-q\Big].
\]
The influence function of the (subsampled) quadrant estimator at $\Phi_{\rho}$
is given, for $S\geq2$, by
\[
\mathrm{IF}((x_{0},y_{0}),R_{Q},\Phi_{\rho})=\pi\sqrt{1-\rho^{2}}S\Big[2\Phi(\tfrac{x_{0}}{\sqrt{S-1}},\tfrac{y_{0}}{\sqrt{S-1}})-\Phi(\tfrac{x_{0}}{\sqrt{S-1}})-\Phi(\tfrac{y_{0}}{\sqrt{S-1}})+1-q\Big].
\]
\end{prop}
As a check, at $\rho=0$ with contamination at $(x_{0},y_{0})=(0,0)$
the bracket equals $2\Phi(0,0)-\Phi(0)-\Phi(0)+1-q=\tfrac{1}{2}-\tfrac{1}{2}=0$,
as required, since replacing one increment by zero preserves the independence
and the sign symmetry of the sparse returns.
\noindent\textbf{Proof of Proposition \ref{thm:influence}.} Under
$\Phi_{\rho}$ the sparse returns $(\tilde{X}_{i},\tilde{Y}_{i})$ are
bivariate normal with variances $S$ and covariance $S\rho$. Under the
contaminated distribution, the number $h$ of the $S$ increments in a
sparse return that are equal to $(x_{0},y_{0})$ is binomial, $h\sim\mathrm{Bin}(S,\varepsilon)$,
and conditional on $h$ the remaining $S-h$ increments are $\Phi_{\rho}$-distributed.
For a function $g$ of $h$ define
\[
G(\varepsilon,g)=\sum^{S}_{h=0}(1-\varepsilon)^{S-h}\varepsilon^{h}\binom{S}{h}g(h),\qquad G_{\varepsilon}(0,g)=\frac{\partial G}{\partial\varepsilon}\bigg|_{\varepsilon=0}=S[g(1)-g(0)].
\]
We repeatedly use the identity $\Pr[(A-a)(B-b)>0]=2\Phi(\tfrac{a}{v},\tfrac{b}{v})-\Phi(\tfrac{a}{v})-\Phi(\tfrac{b}{v})+1$
for $(A,B)$ bivariate normal with mean zero, variances $v^{2}$ and
correlation $\rho$, and note that the right hand side is invariant
under $(a,b)\mapsto(-a,-b)$, since $\Phi(-a,-b)=1-\Phi(a)-\Phi(b)+\Phi(a,b)$.
The estimators are functions of $\mathbb{E}[\mathrm{sgn}(Z)]=2\Pr(Z>0)-1+\Pr(Z=0)$
for the relevant product $Z$, and we work with $\Pr(Z>0)$ below.
Ties, $Z=0$, occur only in the fully contaminated cases: for the quadrant
statistic when all $S$ increments are contaminated and $x_{0}y_{0}=0$,
with probability $\varepsilon^{S}$, and for the Kendall statistic
when both sparse returns are fully contaminated, with probability
$\varepsilon^{2S}$. For $S\geq2$ these are $o(\varepsilon)$ and
do not affect the derivative at $\varepsilon=0$.
\begin{enumerate}
\item The Pearson estimator converges to
\[
R_{p}=\frac{\mathbb{E}[\tilde{X}\tilde{Y}]-\mathbb{E}[\tilde{X}]\mathbb{E}[\tilde{Y}]}{\sqrt{(\mathbb{E}[\tilde{X}^{2}]-\mathbb{E}[\tilde{X}]^{2})(\mathbb{E}[\tilde{Y}^{2}]-\mathbb{E}[\tilde{Y}]^{2})}},
\]
in probability. Under the contaminated distribution of $(\tilde{X},\tilde{Y})$,
\[
\mathbb{E}[\tilde{X}\tilde{Y}]=G(\varepsilon,U),\;\mathbb{E}[\tilde{X}]=G(\varepsilon,V_{x}),\;\text{and}\;\mathbb{E}[\tilde{X}^{2}]=G(\varepsilon,W_{x}),
\]
with
\[
\begin{aligned}U(h) & =(S-h)\rho+h^{2}x_{0}y_{0},\\
V_{x}(h) & =hx_{0},\\
W_{x}(h) & =S-h+h^{2}x^{2}_{0},
\end{aligned}
\]
and $V_{y}$, $W_{y}$ defined analogously. The terms in $h^{2}$ arise
because the $h$ contaminated increments are all equal to $(x_{0},y_{0})$,
so their sum is $h(x_{0},y_{0})$. Then
\[
\begin{array}{ccc}
G(0,U)=S\rho & G(0,V_{x})=0 & G(0,W_{x})=S\\
G_{\varepsilon}(0,U)=S(x_{0}y_{0}-\rho) & G_{\varepsilon}(0,V_{x})=Sx_{0} & G_{\varepsilon}(0,W_{x})=S(x^{2}_{0}-1),
\end{array}
\]
and since $G(0,V_{x})=G(0,V_{y})=0$ the products $\mathbb{E}[\tilde{X}]\mathbb{E}[\tilde{Y}]$,
$\mathbb{E}[\tilde{X}]^{2}$ and $\mathbb{E}[\tilde{Y}]^{2}$ have zero
derivative at $\varepsilon=0$. The influence function of the Pearson
estimator under sparse sampling is therefore
\[
\mathrm{IF}((x_{0},y_{0}),R_{P},\Phi_{\rho})=\frac{S(x_{0}y_{0}-\rho)S-S\rho\frac{1}{2}S(x^{2}_{0}+y^{2}_{0}-2)}{S^{2}}=x_{0}y_{0}-\rho\big(\frac{x^{2}_{0}+y^{2}_{0}}{2}\big).
\]
\item The Kendall estimator is associated with the concordance probability
of two independent sparse returns,
\[
\tilde{R}_{K}=\Pr[(\tilde{X}_{1}-\tilde{X}_{2})(\tilde{Y}_{1}-\tilde{Y}_{2})>0]=\tilde{G}(\varepsilon,F_{K}),
\]
where $h_{1}$ and $h_{2}$ are the numbers of contaminated increments
in the two sparse returns,
\[
\tilde{G}(\varepsilon,F_{K})=\sum^{S}_{h_{1}=0}\sum^{S}_{h_{2}=0}(1-\varepsilon)^{2S-h_{1}-h_{2}}\varepsilon^{h_{1}+h_{2}}\binom{S}{h_{1}}\binom{S}{h_{2}}F_{K}(h_{1},h_{2}),
\]
and, writing the $2S-h_{1}-h_{2}$ uncontaminated increments as $(X_{i},Y_{i})$,
\[
\begin{aligned}F_{K}(h_{1},h_{2}) & =\Pr\Big[\Big(\sum^{S-h_{1}}_{i=1}X_{i}-\sum^{S-h_{2}}_{i=1}X_{S+i}+(h_{1}-h_{2})x_{0}\Big)\Big(\sum^{S-h_{1}}_{i=1}Y_{i}-\sum^{S-h_{2}}_{i=1}Y_{S+i}+(h_{1}-h_{2})y_{0}\Big)>0\Big]\\
 & =2\Phi(a_{h_{1}h_{2}},b_{h_{1}h_{2}})-\Phi(a_{h_{1}h_{2}})-\Phi(b_{h_{1}h_{2}})+1,
\end{aligned}
\]
where $a_{h_{1}h_{2}}=(h_{1}-h_{2})x_{0}/\sqrt{2S-h_{1}-h_{2}}$ and
$b_{h_{1}h_{2}}=(h_{1}-h_{2})y_{0}/\sqrt{2S-h_{1}-h_{2}}$, for
$1\leq h_{1}+h_{2}<2S$, with $F_{K}(0,0)=q$ and $F_{K}(S,S)=0$,
since in the latter case both sparse returns equal $S(x_{0},y_{0})$
and their difference is zero. The derivative of $\tilde{G}$
at $\varepsilon=0$ picks up the terms $(h_{1},h_{2})=(1,0)$ and $(0,1)$,
each with coefficient $S$, and $F_{K}(1,0)=F_{K}(0,1)$ by the sign
invariance noted above, so $\tilde{G}_{\varepsilon}(0,F_{K})=2S[F_{K}(1,0)-q]$.
The estimator is $R_{K}=\sin\{\tfrac{\pi}{2}(2\tilde{R}_{K}-1)\}$,
whose derivative with respect to $\tilde{R}_{K}$ at $\tilde{R}_{K}=q$
is $\pi\sqrt{1-\rho^{2}}$. Thus, the Kendall estimator's influence
function is
\[
\mathrm{IF}((x_{0},y_{0}),R_{K},\Phi_{\rho})=2\pi\sqrt{1-\rho^{2}}S\Big[2\Phi(\tfrac{x_{0}}{\sqrt{2S-1}},\tfrac{y_{0}}{\sqrt{2S-1}})-\Phi(\tfrac{x_{0}}{\sqrt{2S-1}})-\Phi(\tfrac{y_{0}}{\sqrt{2S-1}})+1-q\Big].
\]
\item The quadrant estimator is associated with the quadrant probability
of a sparse return, which at the contaminated distribution is
\[
\tilde{R}_{S}=G(\varepsilon,F_{Q}),
\]
with 
\[
\begin{aligned}F_{Q}(h) & =\Pr\Big[\Big(\sum^{S-h}_{i=1}X_{i}+hx_{0}\Big)\Big(\sum^{S-h}_{i=1}Y_{i}+hy_{0}\Big)>0\Big]\\
 & =2\Phi\Big(\frac{hx_{0}}{\sqrt{S-h}},\frac{hy_{0}}{\sqrt{S-h}}\Big)-\Phi\Big(\frac{hx_{0}}{\sqrt{S-h}}\Big)-\Phi\Big(\frac{hy_{0}}{\sqrt{S-h}}\Big)+1,
\end{aligned}
\]
for $1\leq h\leq S-1$, with $F_{Q}(0)=q$ and $F_{Q}(S)=1_{\{x_{0}y_{0}>0\}}$,
the fully contaminated case in which the sparse return equals $S(x_{0},y_{0})$.
Hence $G_{\varepsilon}(0,F_{Q})=S[F_{Q}(1)-q]$, which for $S\geq2$ involves
only the Gaussian expression,
and since $R_{Q}=\sin\{\tfrac{\pi}{2}(2\tilde{R}_{S}-1)\}$ has derivative
$\pi\sqrt{1-\rho^{2}}$ at $\tilde{R}_{S}=q$, the influence function
of the quadrant estimator with sparse returns of length $S$ is
\[
\mathrm{IF}((x_{0},y_{0}),R_{Q},\Phi_{\rho})=\pi\sqrt{1-\rho^{2}}S\Big[2\Phi(\tfrac{x_{0}}{\sqrt{S-1}},\tfrac{y_{0}}{\sqrt{S-1}})-\Phi(\tfrac{x_{0}}{\sqrt{S-1}})-\Phi(\tfrac{y_{0}}{\sqrt{S-1}})+1-q\Big].
\]
The subsampled estimator $Q_{S}$ pools the sign concordances of the
$S$ offset grids before applying the sine transformation, see Section
\ref{subsubsec:subsampled}. Each grid has the same population functional
$\tilde{R}_{S}$, so the pooled functional is $\tilde{R}_{S}$ as well,
and the influence function of $Q_{S}$ is the same.
\end{enumerate}
\hfill{}$\square$

The Pearson estimator's influence function is invariant to the sparse
sampling, while the sampling frequency determines the (subsampled)
Quadrant and Kendall estimators' influence functions. In contrast with Pearson, the bounded influence functions of $Q_{S}$
and $K$ limit the effect of data contamination.

\subsection{The Covariance Term in the Beta Decomposition\label{subsec:decomp-supp}}

For a fixed time of day $t$ and asset $i$, let $\hat{\rho}_{d}$ and
$\hat{\lambda}_{d}$ denote the daily estimates of correlation and
relative volatility, $d=1,\ldots,D$, and let $\hat{\beta}_{d}=\hat{\rho}_{d}\hat{\lambda}_{d}$.
With $\bar{\rho}=D^{-1}\sum_{d}\hat{\rho}_{d}$, $\bar{\lambda}=D^{-1}\sum_{d}\hat{\lambda}_{d}$
and $\bar{\beta}=D^{-1}\sum_{d}\hat{\beta}_{d}$, we have the identity
\[
\bar{\beta}=\bar{\rho}\bar{\lambda}+\mathrm{cov}_{d}(\hat{\rho},\hat{\lambda}),\qquad\mathrm{cov}_{d}(\hat{\rho},\hat{\lambda})=\frac{1}{D}\sum_{d=1}^{D}(\hat{\rho}_{d}-\bar{\rho})(\hat{\lambda}_{d}-\bar{\lambda}),
\]
which is exact when the covariance is normalized by $D$, and
\[
\bar{\beta}=\bar{\rho}\bar{\lambda}\bar{\gamma},\qquad\bar{\gamma}=1+\frac{\mathrm{cov}_{d}(\hat{\rho},\hat{\lambda})}{\bar{\rho}\bar{\lambda}}.
\]
Days on which either estimate is unavailable are dropped from all
three averages, so the identity holds for the reported quantities.
Since $\bar{\beta}$, $\bar{\rho}$ and $\bar{\lambda}$ are positive
in every case, so is $\bar{\gamma}$, and $\log\bar{\beta}=\log\bar{\rho}+\log\bar{\lambda}+\log\bar{\gamma}$,
and the intraday changes used in Section \ref{subsec:decomposition},
$\Delta\beta_{i}=\log\bar{\beta}_{i,\mathrm{last}}-\log\bar{\beta}_{i,\mathrm{first}}$
and similarly for $\rho$, $\lambda$ and $\gamma$, satisfy
$\Delta\beta_{i}=\Delta\rho_{i}+\Delta\lambda_{i}+\Delta\gamma_{i}$.

The cross-day correlation between $\hat{\rho}_{d}$ and $\hat{\lambda}_{d}$
is negative for every stock in the Large Universe, averaging $-0.35$
in the first hour and $-0.43$ in the last, so $\bar{\gamma}<1$ throughout.
Across the 98 stocks $\bar{\gamma}$ ranges from $0.67$ to $0.96$ in the
first hour and from $0.86$ to $0.97$ in the last. For the four stocks
in Figure \ref{fig:RhoLambdaBetaIntraday}, $\bar{\gamma}_{i,t}$ lies
between $0.95$ and $0.97$ for AAPL, $0.88$ and $0.93$ for AMD, $0.88$
and $0.93$ for LYB, and $0.74$ and $0.82$ for NEM, over the trading
day. The betas in Figures \ref{fig:RhoLambdaBetaIntraday}, \ref{fig:WindowSensitivityEmpirical}
and \ref{fig:RhoLambdaBetasSmallUniverse} are all averages of the
daily products, so they include the covariance term.

\section{Additional Simulation Results}

Figure \ref{fig:Plot-the-averages1-1} presents the results for the
Heston model without noise for two different levels of the correlation:
$\rho=0.5$ and $\rho=0.75$. The analogous results for $\rho=0.25$
and $\rho=2/3$ are in Figure \ref{fig:Plot-the-averages1}.

\begin{figure}[H]
\begin{centering}
\includegraphics[width=0.5\textwidth]{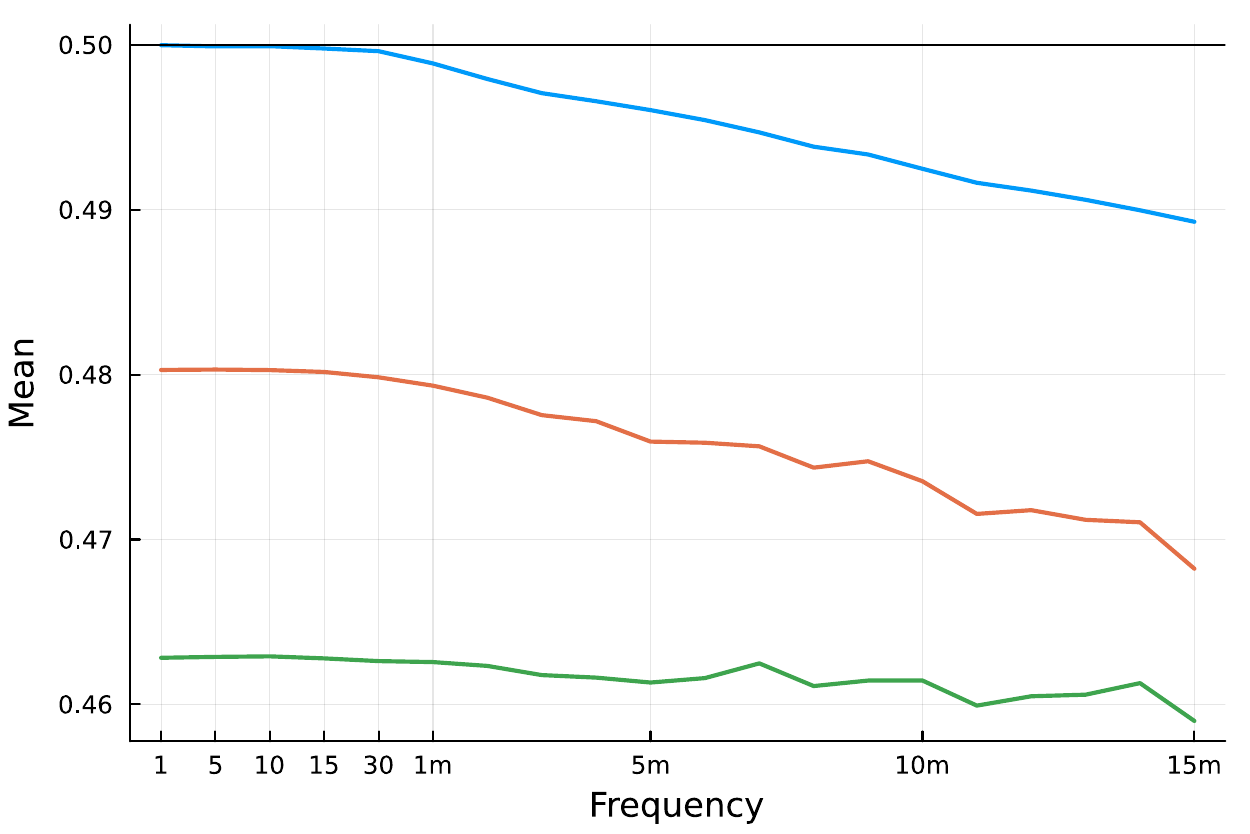}\includegraphics[width=0.5\textwidth]{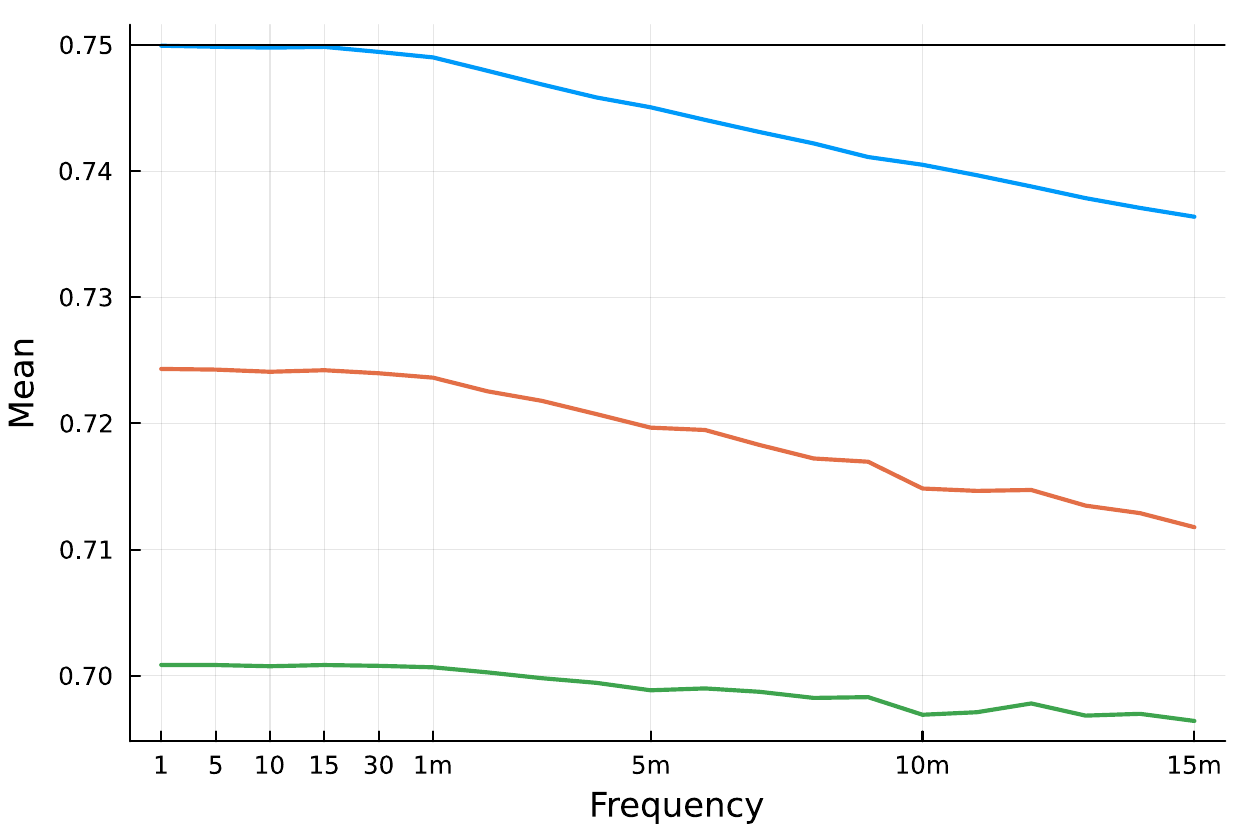}
\par\end{centering}
\centering{}\includegraphics[width=0.5\textwidth]{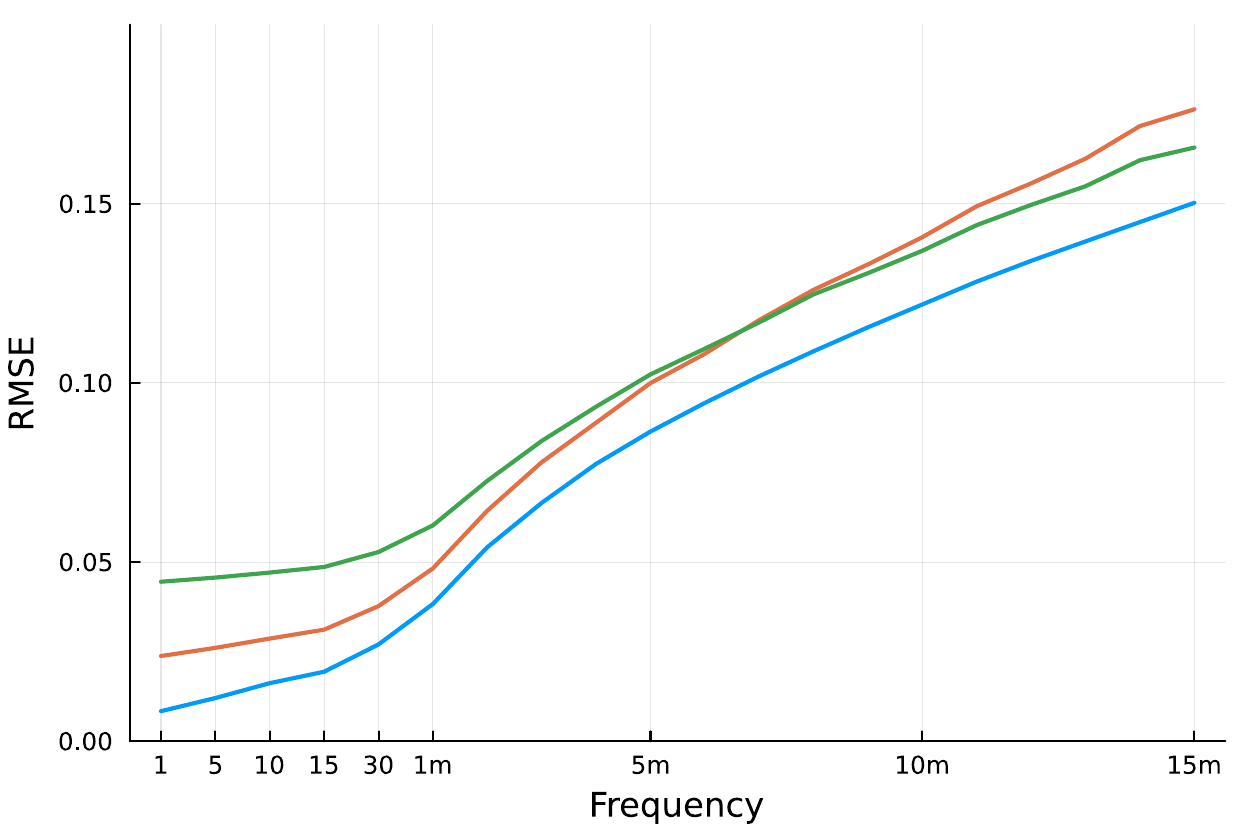}\includegraphics[width=0.5\textwidth]{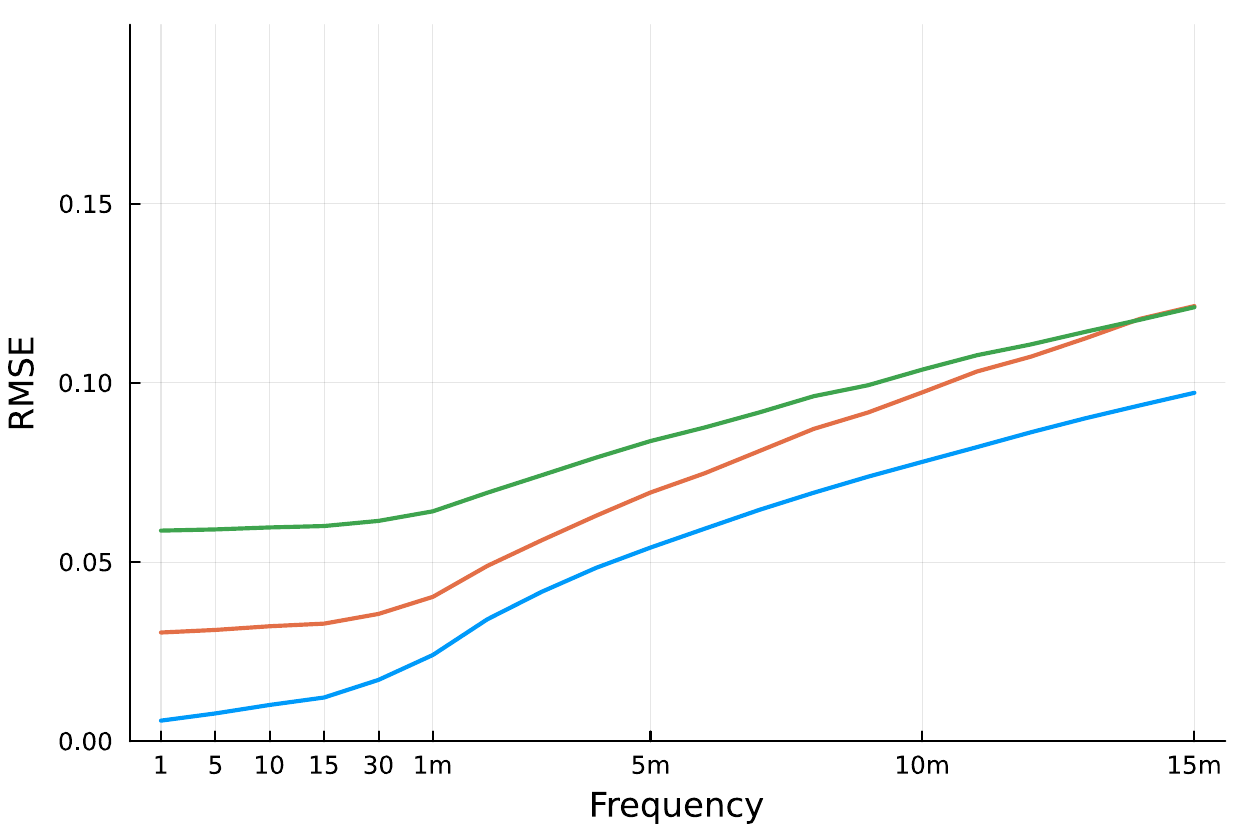}\caption{Means and RMSEs for the estimators as a function of sampling frequency.
Prices are generated by the Heston model with the true correlation
being $\rho=0.50$ (left panels) and $\rho=0.75$ (right panels).\label{fig:Plot-the-averages1-1}}
\end{figure}

\subsection{Noise versus Non-Synchronicity: Additional Results}

Figure \ref{fig:NoiseVsNonsyncRho25} repeats the four contamination
schemes of Figure \ref{fig:NoiseVsNonsync} with $\rho=1/4$. The
qualitative reading is the same: under independent noise the one-second
mean of $Q_{S}$ is close to zero, under correlated noise the signature
plot moves toward the noise correlation of $1/2$, here from below,
and under non-synchronicity $Q_{S}$ remains well above $P$ at the
highest frequencies. Figure \ref{fig:NoiseVsNonsyncZeros} reports
the fraction of zero increment products by sampling frequency. With
additive noise there are no exact zeros. Under previous-tick interpolation
the fraction is $0.97$ at one second, $0.64$ at five seconds, $0.36$
at ten seconds and $0.04$ at $30$ seconds, and it is essentially zero
from one minute onward, which is why the divergence between $Q_{S}$ and
$P$ in the non-synchronous design is confined to the highest frequencies.

\begin{figure}[H]
\begin{centering}
\includegraphics[width=\textwidth]{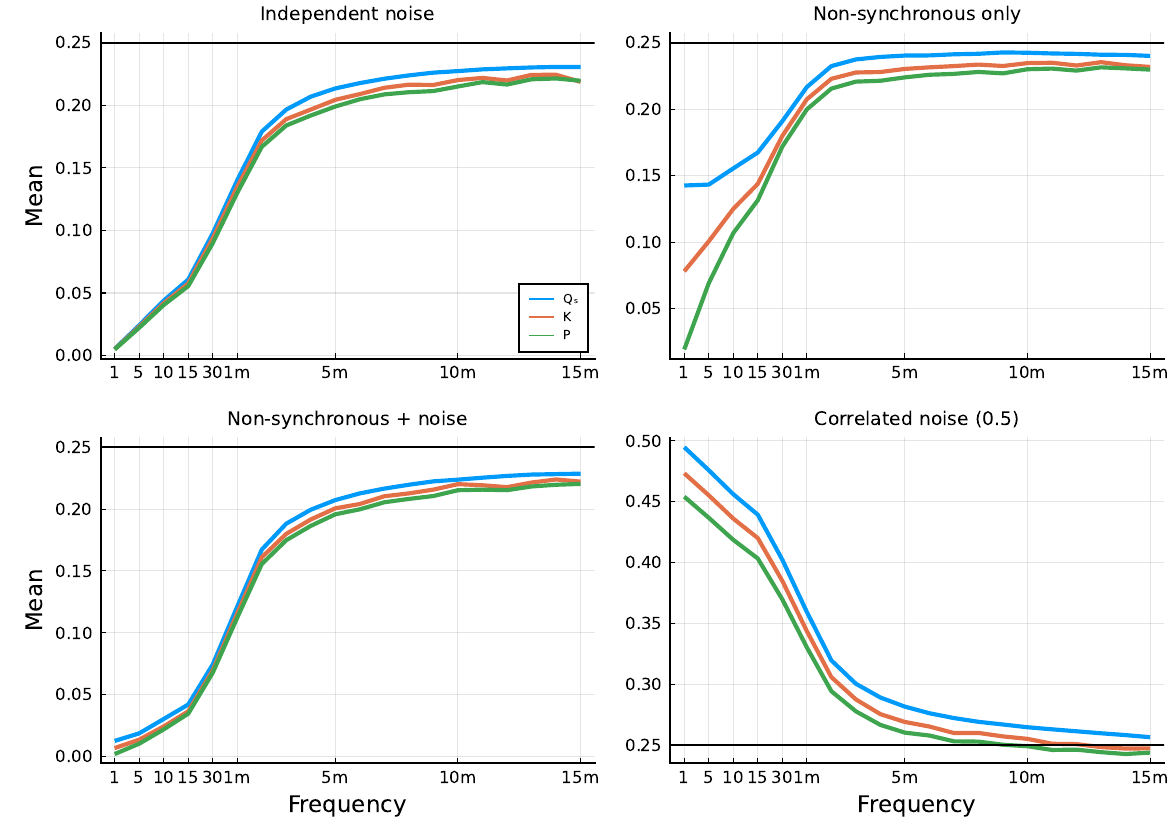}
\par\end{centering}
\caption{Correlation signature plots under the four contamination schemes
of Figure \ref{fig:NoiseVsNonsync}, with $\rho=1/4$. The horizontal
line is the true correlation.\label{fig:NoiseVsNonsyncRho25}}
\end{figure}

\begin{figure}[H]
\begin{centering}
\includegraphics[width=0.6\textwidth]{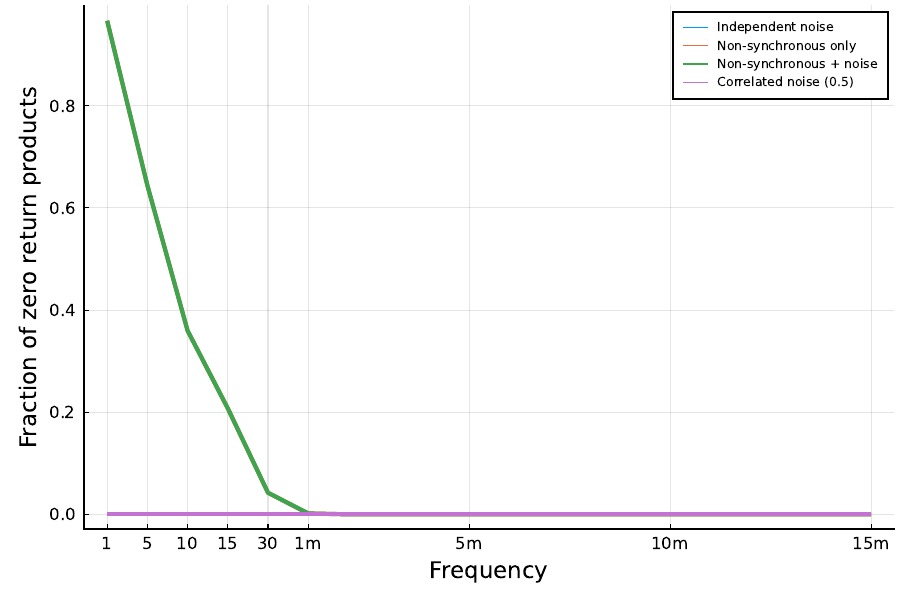}
\par\end{centering}
\caption{Fraction of zero increment products by sampling frequency in the
four contamination schemes of Figure \ref{fig:NoiseVsNonsync}, with
$\rho=2/3$. The two designs with additive noise produce no exact
zeros and their lines coincide with the horizontal axis.\label{fig:NoiseVsNonsyncZeros}}
\end{figure}

\subsection{Time-Varying Correlation: Additional Results}

Section \ref{subsec:WindowLength} reports the window-length comparison
for two designs with time-varying correlation. Design A has
$\mathrm{asin}\,\rho(u)$ linear in $u$ from $0.25$ to $0.60$. Design
B has $\rho(u)=\sin\{\tfrac{\pi}{2}\tanh(Z_{u})\}$, where
\[
\mathrm{d}Z_{u}=\kappa_{z}(\mu_{z}-Z_{u})\mathrm{d}u+\varsigma_{z}\mathrm{d}W^{z}_{u},\qquad\kappa_{z}=8,\quad\mu_{z}=0.5,\quad\varsigma_{z}=0.6,
\]
with $Z_{0}$ drawn from the stationary distribution and $W^{z}$ independent
of the price and volatility innovations. The transformation keeps $\rho(u)$
in $(-1,1)$. Under the stationary distribution $Z_{u}\sim N(0.5,0.15^{2})$,
so $\rho(u)$ has mean $0.64$ and standard deviation $0.14$, with
$90\%$ of its mass between $0.38$ and $0.84$, which provides realistic
intraday variation. Under Design B the targets are random and differ
across replications, so bias and root mean squared error are computed
against each replication's own realized $\rho_{\bullet}$ and $\rho(u)$.

Table \ref{tab:spotwindows} reports the underlying numbers for both
designs and both window placements, including the trailing placement
discussed in Section \ref{subsec:WindowLength}. Two features are worth
noting. Under Design A the centered window is unbiased at every window
length, as the linearity of $\mathrm{asin}\,\rho(u)$ implies, whereas
the trailing window carries a bias that is proportional to the window
length and reaches $-0.054$ at $120$ minutes. Under Design B the
centered window has a smoothing error that grows with the window length,
most of which is random rather than a mean bias, and it is the only
design in which the root mean squared error is not monotone in the
window length.

\begin{table}[H]
\centering
\caption{Window-length comparison for the spot $Q_{S}$\label{tab:spotwindows}}
\medskip{}

\begin{centering}
\begin{tabular*}{\textwidth}{@{\extracolsep{\fill}}llrrrrrrr}
\hline\hline
 & & & \multicolumn{3}{c}{Mean bias} & \multicolumn{3}{c}{RMSE}\\
\cline{4-6}\cline{7-9}
Design & Placement & Window & $m=1s$ & $m=30s$ & $m=3m$ & $m=1s$ & $m=30s$ & $m=3m$\\
\hline
A & centered & 30 min & -0.0005 & -0.0029 & -0.0184 & 0.0319 & 0.1054 & 0.2594\\
A & centered & 60 min & -0.0003 & -0.0014 & -0.0097 & 0.0225 & 0.0746 & 0.1814\\
A & centered & 90 min & -0.0003 & -0.0010 & -0.0069 & 0.0184 & 0.0607 & 0.1477\\
A & centered & 120 min & -0.0003 & -0.0008 & -0.0056 & 0.0159 & 0.0526 & 0.1279\\
A & trailing & 30 min & -0.0139 & -0.0163 & -0.0318 & 0.0348 & 0.1067 & 0.2607\\
A & trailing & 60 min & -0.0271 & -0.0283 & -0.0365 & 0.0353 & 0.0797 & 0.1848\\
A & trailing & 90 min & -0.0404 & -0.0411 & -0.0470 & 0.0444 & 0.0733 & 0.1548\\
A & trailing & 120 min & -0.0535 & -0.0540 & -0.0588 & 0.0558 & 0.0754 & 0.1406\\
\hline
B & centered & 30 min & 0.0016 & -0.0018 & -0.0206 & 0.0522 & 0.0913 & 0.2061\\
B & centered & 60 min & 0.0032 & 0.0012 & -0.0092 & 0.0662 & 0.0844 & 0.1521\\
B & centered & 90 min & 0.0043 & 0.0029 & -0.0047 & 0.0778 & 0.0886 & 0.1349\\
B & centered & 120 min & 0.0052 & 0.0040 & -0.0022 & 0.0868 & 0.0942 & 0.1285\\
B & trailing & 30 min & 0.0017 & -0.0018 & -0.0205 & 0.0860 & 0.1142 & 0.2181\\
B & trailing & 60 min & 0.0032 & 0.0012 & -0.0092 & 0.1061 & 0.1185 & 0.1742\\
B & trailing & 90 min & 0.0043 & 0.0028 & -0.0048 & 0.1176 & 0.1251 & 0.1620\\
B & trailing & 120 min & 0.0052 & 0.0040 & -0.0022 & 0.1246 & 0.1299 & 0.1571\\
\hline\hline
\end{tabular*}
\par\end{centering}
Note: Mean bias and root mean squared error of the rolling-window
$Q_{S}$ as an estimator of the spot correlation $\rho(u)$, averaged
over the window midpoints for which all four window lengths fit within
the trading day, for sampling intervals $m$ of one second, 30 seconds
and three minutes. ``Centered'' labels a window by its midpoint,
``trailing'' by its right endpoint. Design A has a deterministic
correlation path, Design B a stochastic one. Based on $5{,}000$
replications.
\end{table}

\subsection{Jumps}

Jumps are prevalent in high-frequency prices. We consider the following
jump processes,

\[
X_{t}=X^{\ast}_{t}+\sum_{s\leq t}J^{x}_{s},\qquad Y_{t}=Y^{\ast}_{t}+\sum_{s\leq t}J^{y}_{s},
\]
where $J^{x}_{t}$ and $J^{y}_{t}$ are Poisson jump processes with
intensities $\lambda_{x}$ and $\lambda_{y}$, respectively, with
jump sizes uniformly distributed on $[-\tfrac{2}{\sqrt{2\lambda}},-\tfrac{1}{\sqrt{2\lambda}}]\bigcup[\tfrac{1}{\sqrt{2\lambda}},\tfrac{2}{\sqrt{2\lambda}}]$. 

We focus on two types of jumps: independent jumps, $\sum J^{x}_{s}\bot\negthickspace\bot\sum J^{y}_{s}$,
and co-jumps, $\sum J^{x}_{s}=\sum J^{y}_{s}$. In all cases we use
the intensities $\lambda_{x}=\lambda_{y}=1$.

The bias properties of the estimators with independent jumps are in
the upper panels of Figure \ref{fig:HestonJumps} and the analogous
results for the case with co-jumps are in the lower panels. We present
results for $\rho=0.25$ in the left panels and $\rho=2/3$ in the
right panels. Both the $Q_{S}$ and $K$ estimators are largely unaffected
by jumps. This is to be expected, given their construction and influence
functions. The situation is quite different for the Pearson estimator.
Independent jumps induce a strong bias toward zero in $P$, while
co-jumps induce a strong positive bias toward one. That $P$ is very
sensitive to jumps is an implication of the influence function for
$P$. This sensitivity to jumps motivates the commonly used truncation
methods for computing realized variances, see \citet[2009]{Mancini:2001}
and \citet{AndersenDobrevSchaumburg2012},\nocite{Mancini:2009} and
truncation methods are clearly also very useful if $P$ is to be used.

\begin{figure}[H]
\begin{centering}
\subfloat[Heston with independent jumps]{\begin{centering}
\includegraphics[width=0.5\textwidth]{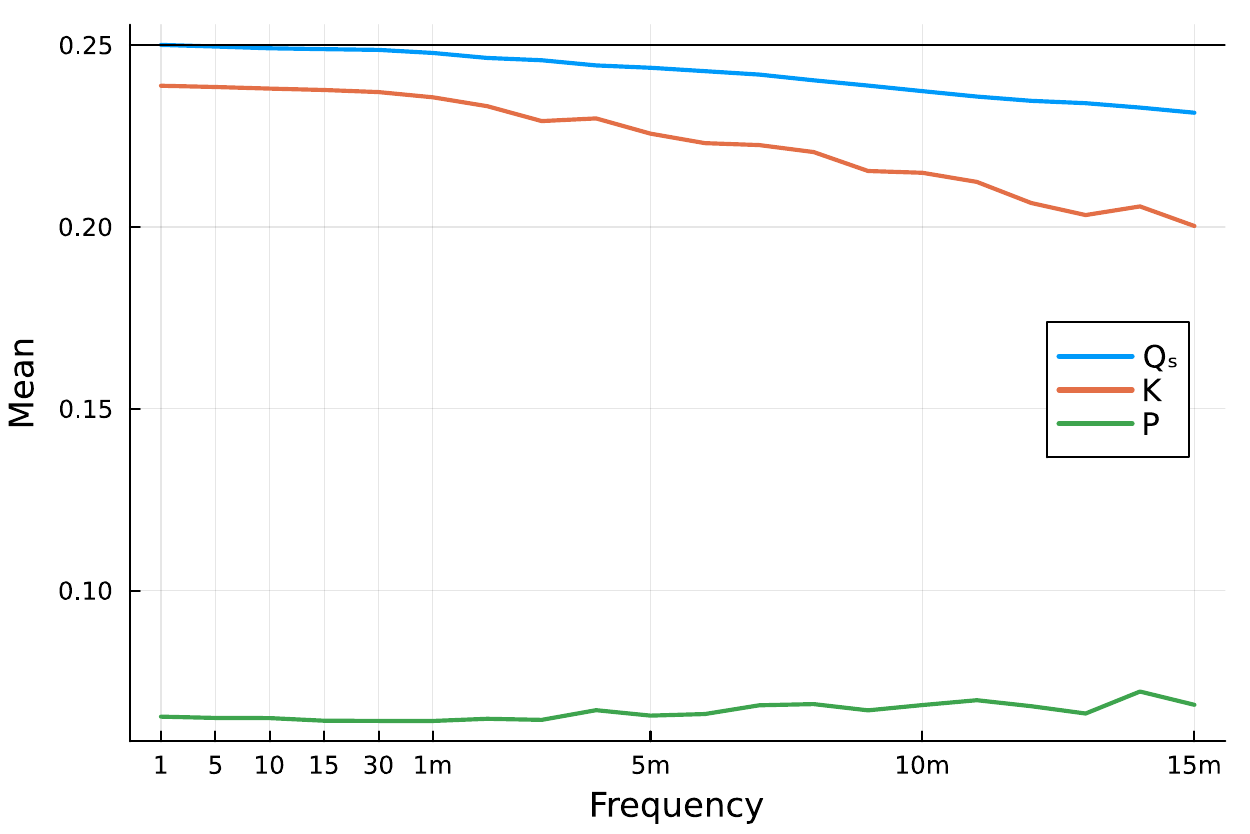}\includegraphics[width=0.5\textwidth]{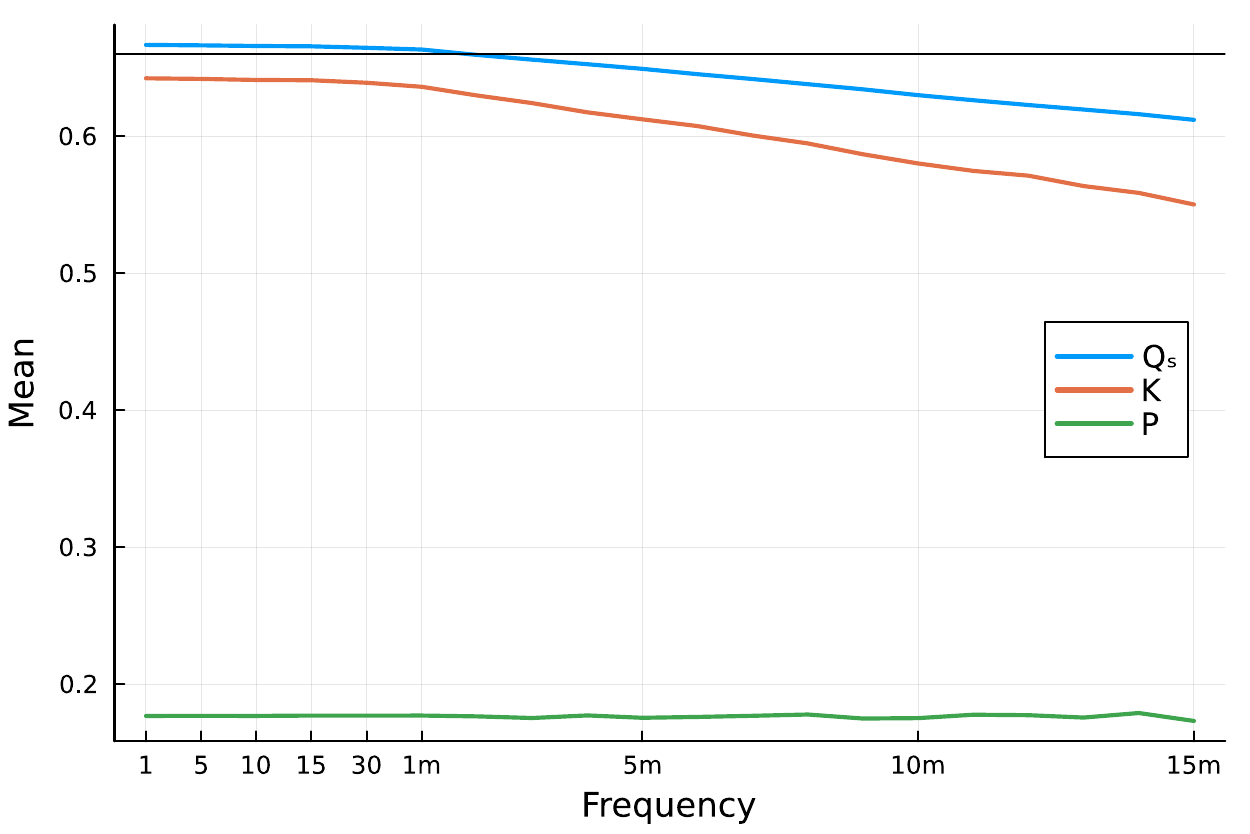}
\par\end{centering}
}
\par\end{centering}
\begin{centering}
\subfloat[Heston with co-jumps]{\begin{centering}
\includegraphics[width=0.5\textwidth]{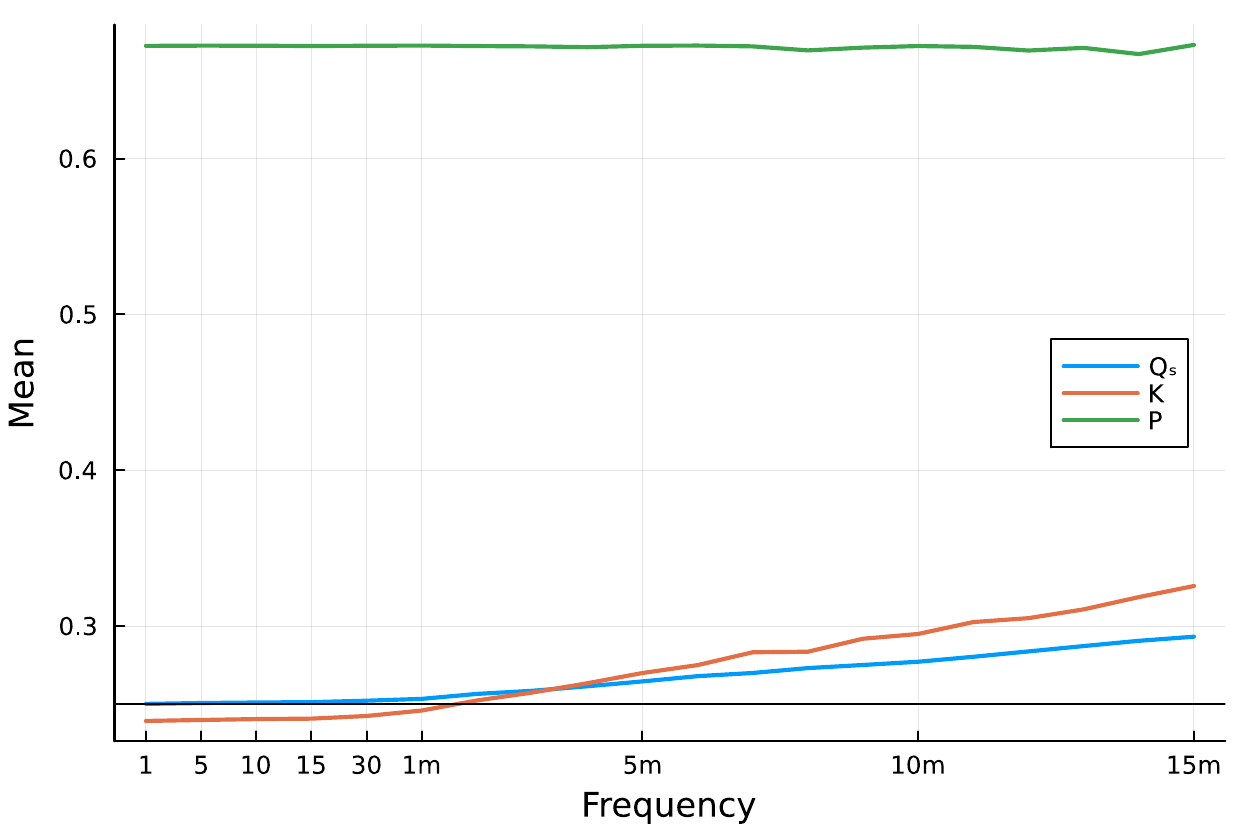}\includegraphics[width=0.5\textwidth]{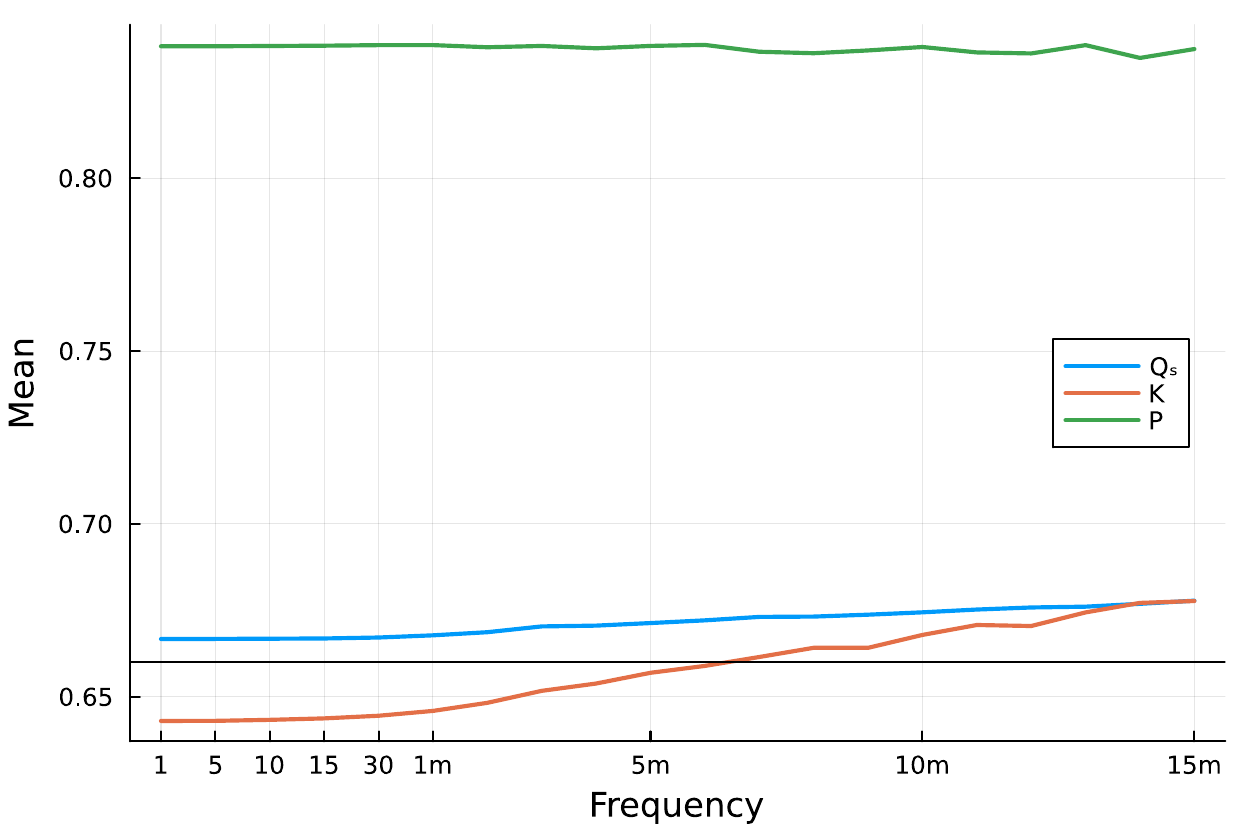}
\par\end{centering}
}
\par\end{centering}
\caption{Average estimates from $P$,
$K$, and $Q_{S}$ are plotted against sampling frequency. The underlying
model is the Heston model with independent jumps in the upper panels
and co-jumps in the lower panels. The correlation coefficient is $\rho=0.25$
in the left panels and $\rho=2/3$ in the right panels.\label{fig:HestonJumps}}
\end{figure}

\end{document}